\documentclass[12pt,letterpaper]{report}

\usepackage[utf8]{inputenc}

\usepackage{geometry}
\usepackage{fancyhdr}
\usepackage{afterpage}
\usepackage{graphicx}
\usepackage{amsmath,amssymb,amsbsy,physics,amsthm,amsfonts}
\usepackage{dcolumn,array}
\usepackage{tocloft}
\usepackage{asudis}
\usepackage{bm}
\usepackage[bibstyle=nature,refsection=chapter,backref=true]{biblatex}

\usepackage[T1]{fontenc} 
\usepackage{titlesec} 
\usepackage{comment}

\usepackage{subcaption}
\usepackage{overpic}
\usepackage{bbold}
\usepackage{mathrsfs}
\usepackage{mathbbol}
\usepackage{bbding}
\usepackage{esvect}
\usepackage{commath}
\usepackage{enumitem}
\usepackage{siunitx}
\usepackage{braket}
\usepackage{phaistos}
\usepackage{multicol}
\usepackage{bbm}
\usepackage{multirow}
\usepackage{tikz,float,pgfplots}
\makeatletter
\newcommand\notsotiny{\@setfontsize\notsotiny{7}{8}}
\makeatother

\usepackage{comment}
\usepackage[colorlinks=true
  ,urlcolor=black
  ,anchorcolor=black
  ,citecolor=black
  ,filecolor=black
  ,linkcolor=black
  ,menucolor=black
  ,pagecolor=black
  ,linktocpage=true
  ,pdfproducer=medialab
  ,pdfa=true
]{hyperref}
\usepackage{xcolor}
\usepackage[capitalise]{cleveref}
\definecolor{myurlcolor}{rgb}{0,0,0.7}
\definecolor{myrefcolor}{rgb}{0.1,0,0.9}

\usepackage[most]{tcolorbox}
\definecolor{light-gray}{gray}{0.95}
\tcbset{textmarker/.style={%
        enhanced,
        parbox=false,boxrule=0mm,boxsep=0mm,arc=0mm,
        outer arc=0mm,left=6mm,right=3mm,top=7pt,bottom=7pt,
        toptitle=1mm,bottomtitle=1mm,oversize}}
    \newtcolorbox{inputBox}{textmarker,
        borderline west={6pt}{0pt}{black},
        colback=black!2!white}
    \newtcolorbox{outputBox}{textmarker,
        borderline west={6pt}{0pt}{black},
        colback=black!8!light-gray}

\usepackage[linesnumbered, ruled,vlined]{algorithm2e}

\graphicspath{{./images/},{./imagesAppendix/}}

\renewcommand{\eqref}[1]{Eq.~(\ref{#1})} 

\def\app#1#2{%
  \mathrel{%
    \setbox0=\hbox{$#1\sim$}%
    \setbox2=\hbox{%
      \rlap{\hbox{$#1\propto$}}%
      \lower1.1\ht0\box0%
    }%
    \raise0.25\ht2\box2%
  }%
}

\newcommand{\inBox}[1]{\begin{inputBox} \textbf{Input:} #1 \end{inputBox}}
\newcommand{\outBox}[1]{\begin{outputBox} \textbf{Output:} #1 \end{outputBox}}

\newcommand\ReferencesList{}
\newcommand\AddtoRefsList[1]{\xdef\ReferencesList{\ReferencesList,#1}}
\AtEveryCitekey{\AddtoRefsList{\thefield{entrykey}}}
\defbibenvironment{nolabelbib}
  {\list
     {}
     {\setlength{\leftmargin}{\bibhang}%
      \setlength{\itemindent}{-\leftmargin}%
      \setlength{\itemsep}{\bibitemsep}%
      \setlength{\parsep}{\bibparsep}}}
  {\endlist}
  {\item}
\usepackage{etoolbox}
\makeatletter
\patchcmd{\blx@addbackref@i}{\c@refsection}{\c@savedrefsection}{}{}
\newcounter{savedrefsection}
\newcommand\saverefsection{%
  \protected@write\@mainaux{}{\string\setcounter{savedrefsection}{\the\c@refsection}}%
}
\makeatother

\newcommand{\Hil}{\mathcal{H}}
\newcommand{\Mod}[1]{\,\mathrm{mod}\,#1}
\newcommand{\Stab}[1]{\,\textnormal{Stab}\,#1}
\newcommand{\HC}{(HC)_{1,2}}
\def\ket#1{{|{#1}\rangle}} 

\usepackage{environ}
\NewEnviron{myequation}{%
\begin{equation*}
\scalebox{1.7}{$\BODY$}
\end{equation*}
}

\newtheorem{theorem}{Theorem}
\newtheorem{corollary}[theorem]{Corollary}
\newtheorem{lemma}[theorem]{Lemma}
\newtheorem{proposition}{Proposition}
\newtheorem{conjecture}{Conjecture}
\newtheorem{observation}{Observation}

\newtheorem{remark}{Remark}

\usepackage[linesnumbered, ruled,vlined]{algorithm2e}

\graphicspath{{./images/},{./imagesAppendix/}}

\renewcommand{\eqref}[1]{Eq.~(\ref{#1})} 

\def\app#1#2{%
  \mathrel{%
    \setbox0=\hbox{$#1\sim$}%
    \setbox2=\hbox{%
      \rlap{\hbox{$#1\propto$}}%
      \lower1.1\ht0\box0%
    }%
    \raise0.25\ht2\box2%
  }%
}

\ifx\proof\undefined
\newenvironment{proof}[1][\protect\proofname]{\par
	\normalfont\topsep6\p@\@plus6\p@\relax
	\ivlist
	\itemindent\parindent
	\item[\hskip\labelsep\scshape #1]\ignorespaces
}{%
	\endtrivlist\@endpefalse
}
\providecommand{\proofname}{Proof}
\fi

\makeatother

\providecommand{\factname}{Fact}
\providecommand{\theoremname}{Theorem}
\providecommand{\claimname}{Claim}
\providecommand{\lemmaname}{Lemma}
\providecommand{\definitionname}{Definition}

\def\bbbone{{\mathchoice {\rm 1\mskip-4mu l} {\rm 1\mskip-4mu l}
{\rm 1\mskip-4.5mu l} {\rm 1\mskip-5mu l}}}

\definecolor{KB}{rgb}{0.4,0.3,0.9}

\definecolor{THc}{rgb}{0.9,0.3,0.2}

\newcommand{\be}{\begin{equation}}
\newcommand{\ee}{\end{equation}}
\newcommand{\ba}{\begin{eqnarray}}
\newcommand{\ea}{\end{eqnarray}}

\newcommand{\pur}{\operatorname{Pur}}

\newcommand{\st}[1]{\ketbra{#1}{#1}}
\newcommand{\stab}{\operatorname{STAB}}
\newcommand{\pstab}{\operatorname{PSTAB}}
\newcommand{\sma}[1]{M^{(NL,\epsilon)}_{RS}(#1)}

\newtheorem{definition}{\protect\definitionname}

\def\d{\mathrm{d}}

\usepackage{framed}
\definecolor{shadecolor}{rgb}{0.90,0.90,0.90}
\usepackage[normalem]{ulem}
\usepackage{mathtools}

\numberwithin{equation}{section}

\def\beq{\begin{eqnarray}}\def\eeq{\end{eqnarray}}
\def\be{\begin{equation}}\def\ee{\end{equation}}
\def\nn{\nonumber}

\def\s{\sigma}

\def\b{\beta}
\def\d{\delta}

\def\D{\Delta}

\def\bz{\bar{z}}

\def\mc{{\mathcal{C}}}

\def\tr{{\rm tr~}}

\def\l{\eta_0}
\newcommand{\bmh}{\bar{\mathcal{H}}}
\newcommand{\mA}{\mathcal{A}}

\begin{document}

\pagenumbering{roman}
\title{Bit by Bit: Gravity Through the Lens of Quantum Information}
\author{William Richard Munizzi}
\degreeName{Doctor of Philosophy}
\defensemonth{April}
\gradmonth{May}
\gradyear{2024}
\chair{Cynthia Keeler, Chair \\ Maulik Parikh\\ Matthew Baumgart\\ Kevin Schmidt}		
\maketitle
\doublespace
\begin{abstract}

Computable properties of quantum states are given a dual gravitational interpretation via the AdS/CFT correspondence. For holographic states, boundary entanglement entropy is dual to the area of bulk geodesics, known as Ryu-Takayanagi surfaces. Furthermore, the viability of states to admit a holographic dual at all is constrained by their entanglement structure. Entanglement therefore defines a coarse classification of states in the Hilbert space. Similarly, how a state transforms under a group of operators also provides a classification on the Hilbert space. Certain states, e.g. stabilizer states, are invariant under large sets of operations, and consequently can be simulated on a classical computer. Cayley graphs offer a useful representation for a group of operators, where vertices represent group elements and edges represent group generators. In this representation, the orbit of a state under action of the group can also be represented as a ``reachability graph'', defined as a quotient of the group Cayley graph. Reachability graphs can be dressed to encode entanglement information, making them a useful tool for studying entanglement dynamics under quantum operations. Further quotienting a reachability graph by group elements that fix a chosen state property, e.g. entanglement entropy, builds a ``contracted graph''. Contracted graphs provide explicit bounds on state parameter evolution under quantum circuits. In this work, an upper bound on entropy vector evolution under Clifford group action is presented. Another important property of quantum systems is magic, which quantifies the difficulty of classically simulating a quantum state. Magic and entanglement are intimately related, but the two are not equivalent measures of complexity. Nonetheless, entanglement and magic play complementary roles when describing emergent gravitational phenomena in AdS/CFT. This manuscript describes the interplay between entanglement and magic, and offers a holographic interpretation for magic as cosmic brane back-reaction.

\end{abstract}

\begin{acknowledgements}
I would like to first express my deepest gratitude to my advisor, Dr. Cindy Keeler, whose continuous support, guidance, and motivation have led me to this moment. Throughout each and every stage of my graduate tenure, Cindy consistently prioritized my development as a thinker and as a physicist. Cindy taught me how to refine and sharpen my logical arguments, how to improve my expository writing and presentation skills, and how to organize my excursions into each new idea. I am tremendously thankful for her years of advice and the wonderful physics we studied.

I wish to personally thank each member of my doctoral committee: Dr. Kevin Schmidt, Dr. Matt Baumgart, and Dr. Maulik Parikh. I especially want to thank Kevin for his instruction during my preliminary graduate coursework, as well as his profound insights into just about any physics topic, from which I have learned so much. Likewise I thank Matt for teaching me quantum field theory, and always providing thought-provoking questions during my seminars. I am grateful to Maulik for entertaining my many naive questions about general relativity during journal club, presentations, or ski trips.

I give special thanks to Dr. Bert de Jong for his collaboration and mentorship while I was at Berkeley Lab. Through Bert's guidance I began 
refining some of my lofty ideas for practical implementation, a skill which I feel has improved my performance as a theorist.

I am perpetually grateful to my parents Jim and Carol, for their unwavering support through the many years of my academic journey. I likewise thank my sister Jaime for her sustained encouragement.

For my partner Serena, who fastened enumerable moments of solace out of otherwise stressful times, I am so very thankful. 

During my doctoral studies I was privileged to work alongside a collection of wonderful collaborators, for whom I hold special regard. I am continually inspired by my collaborator Howard Schnitzer, whose lifelong passion for thought and dedication to science and mathematics I can only hope to emulate. I am particularly grateful to Jason Pollack, Charles Cao, Carlos Cardona, and Ning Bao, all of whom I look up to as physicists and each of whom was always willing to give me practical advice and assist me in navigating the academic world. I strongly value my current collaborators Claire Zukowski, Temple He, Huo Chen, Dawei Zhong, Gong Cheng, Lorenzo Leone, and Savatore Oliviero for stimulating scientific conversations and from whom I have learned so much.

To my dear friends and fellow physicists Chris Knill, Shep Bryan, Sean Tilton, Aditya Dhumuntarao, Sasha Sypkens, and Farzad Faramarzi, I am very grateful for the years of countless memories, the years of conversations both scientific and otherwise, and the continued support. 

I give special thanks to my friends and collaborators within the American Physical Society (APS), specifically Farah Dawood, Sarah Monk, Charlotte Selton, Stephanie Lough, Dave Austin, Eric Switzer, David Stillwell, Bhavay Tyagi, Sophia Turner, and Joe Moscoso, who I look forward to seeing at every APS event, and who have provided me with so many laughs and so much inspiration over the years.

Finally I wish to acknowledge additional physicists and mathematicians who were formative to my career, William Terrano, Michael Treacy, Christian Arenz, Alicia Magann, Alexander Jahn, David Austin, Kingshuk Majumdar, Richard Vallery, Julien Paupert, Aaron Szasz, Mads Bahrami and David Polleta.

\end{acknowledgements}

\renewcommand{\cftlabel}{CHAPTER}

\tableofcontents
\addtocontents{toc}{~\hfill Page\par}
\newpage

\addcontentsline{toc}{part}{LIST OF TABLES}
\renewcommand{\cftlabel}{Table}
\listoftables
\addtocontents{lot}{Table~\hfill Page \par}
\newpage

\addcontentsline{toc}{part}{LIST OF FIGURES}
\addtocontents{toc}{CHAPTER \par}
\renewcommand{\cftlabel}{Figure}
\listoffigures

\addtocontents{lof}{Figure~\hfill Page \par}
\newpage
\renewcommand{\cftlabel}{APPENDIX}

\doublespace
\pagenumbering{arabic}
\chapter{INTRODUCTION}

\textit{``We do not know what the rules of the game are; all we are allowed to do is to watch the playing. Of course, if we watch long enough, we may eventually catch on to a few of the rules. The rules of the game are what we mean by fundamental physics.''} -Richard P. Feynman\\

In recent years, many research efforts have centered around an emergent connection between gravity and quantum information. The academic works compiled herein contribute only a small piece to this ongoing story, applying techniques and understandings from quantum information to gain further insight into the nature of gravity. Along the way, a number of interesting results regarding entanglement, magic, and the structure of quantum systems were fortuitously acquired as well. Looking eagerly towards the future of quantum computation, it is my hope that some of these results may find application in both scientific and technological settings.

\section{Entanglement in Quantum Systems}

One of the quintessential (and perhaps most bizarre) characteristics of quantum systems is the concept of entanglement. Contrary to classical intuition, the mechanism of entanglement renders a collection of systems, even in the absence of classical interaction, independently indescribable. While a complete understanding of the entanglement phenomenon is presently lacking, we know such a mechanism exists and we may nonetheless exploit its many consequences. In particular, with the advent of quantum computing, entanglement offers an invaluable resource for extending computation and understanding beyond what is classically achievable. In the following sections we will give an overview describing the quantitative measures of entanglement, as well as its role in quantum gravity and quantum computation.

When working with quantum-mechanical systems it is often most natural to consider a pure state $\ket{\Psi}$, which exists as some vector in a Hilbert space $\Hil$. Since $\ket{\Psi}$ is pure, its ``quantum'' characteristics are limited to those which are classically-achievable. However, if we decompose $\ket{\Psi}$ into a collection of constituent subsystems we observe a far richer structure. Performing this decomposition requires that we assume $\Hil$ itself can be decomposed into the tensor product of Hilbert spaces $\Hil_i$
\begin{equation}\label{HilDecomp}
    \Hil \equiv \bigotimes_{i} \Hil_i.
\end{equation}
\eqref{HilDecomp} suitably describes the setting for a system of qubits, a lattice of spins, or any discrete set of interacting quantum systems.

\subsection{Entropies of Entanglement}

Assuming the bipartition $\Hil = \Hil_A \otimes \Hil_{\bar{A}}$ exists, we can ask what an observer confined to $\Hil_A$ can learn about a full state $\ket{\Psi}$. The quantitative measure of this limitation is known as entanglement entropy, and is computed according to the von Neumann entropy in \eqref{introEE}. To calculate the entanglement entropy of some state in a $\Hil$ subspace we define subsystems $\rho_A \in \Hil_A$ and $\rho_{\bar{A}} \in \Hil_{\bar{A}}$, such that $\rho_{A} \cup \rho_{\bar{A}}$ forms the pure state $\ket{\Psi} \in \Hil$. The entanglement entropy of $\rho_A$, with respect to its complement $\rho_{\bar{A}}$, is then
\begin{equation}\label{introEE}
    S_A \equiv -\Tr \rho_{A}\log_{\rho_A}.
\end{equation}
The object $\rho_A$ is the reduced density matrix on the subregion $A$, and is constructed via a partial trace of $\ket{\Psi}$ over everything in $\Hil_{\bar{A}}$.

The entanglement entropy in \eqref{introEE} is one instance of a family of entropies, known as Renyi entropies. The Renyi entropies, defined
\begin{equation}\label{introRE}
    S_{\alpha} \equiv \frac{1}{1-\alpha} \ln \Tr \rho^{\alpha}, \qquad \forall \alpha \in [0,\infty],
\end{equation}
provide a generalization of the classical Shannon entropy. \eqref{introEE} is recovered from \eqref{introRE} in the $\alpha \to 1$ limit. In fact, this method for computing entanglement entropy as a limit of the generalized Renyi entropy often proves simpler for field theories in particular \cite{Rangamani:2016dms}. Chapters \ref{Chapter3}--\ref{Chapter6} of this dissertation rely specifically on the properties of entanglement entropy, while \ref{Chapter7} extends a more general interpretation for properties of Renyi entropies in holography.

Given a composite state $\rho \in \Hil$, comprised of $n$ disjoint subsystems, there are $2^n-1$ unique entanglement entropies that can be computed using \eqref{introEE}. Organizing this set of subsystem entanglement entropies into a $2^n-1$ component tuple builds the entropy vector \cite{Bao2015} for $\rho$. For example, a tripartite state with disjoint subsystems indexed $A, \,B,$ and $C$ yields an entropy vector of the form
\begin{equation}\label{introEVec}
    \vec{S}_{\rho} = (S_A,\, S_B,\,S_C,\,S_{AB},\,S_{AC},\,S_{BC},\,S_{ABC}).
\end{equation}
In this way, the entropy vector $\vec{S}_{\rho}$ for a state provides a complete description of subsystem entanglement entropy in $\rho$.

\subsection{Entropy Inequalities}

Knowledge of a state's entropy vector is sufficient to establish a classification on states in $\Hil$. One way to determine this classification is by validating the satisfaction, saturation, or failure of certain entropy inequalities \cite{He:2023aif,Linden:2013kal,Schnitzer:2022exe,Munizzi:2023ihc,Schnitzer:2020lcr}. An entropy inequality constrains the structure of a state's entropy vector in a way that is consistent with some otherwise understood property. For example, it is well-known that all quantum states are strongly subadditive, therefore every quantum state must possess an entropy vector which satisfies
\begin{equation}\label{introSSA}
    S_{AB} + S_{BC} \geq S_B + S_{ABC}.
\end{equation}
While \eqref{introSSA} is trivially satisfied for arbitrary quantum states, additional entropy inequalities provide stricter constraints on a state's entanglement structure \cite{Bao2015,Bao:2020zgx,Bao:2020mqq}. 

The monogamy of mutual information (MMI), constitutes an entropy inequality which is satisfied by a strict subset of quantum states. States which obey MMI must have an entropy vector which satisfies
\begin{equation}\label{introMMI}
    S_{AB} + S_{AC} + S_{BC} \geq S_A + S_B + S_C + S_{ABC}.
\end{equation}
\eqref{introMMI} describes a critical property of holographic quantum states \cite{Hayden:2011ag}, those states which admit a smooth classical description in a dual gravity theory via AdS/CFT. We will review the AdS/CFT correspondence and discuss holographic states further in the next section.

\section{Holography and AdS/CFT}

In 1997 a monumental duality was discovered, revealing that gravitational degrees of freedom in certain $(d+1)$-dimensional spacetimes can be encoded into the degrees of freedom for a class of $d$-dimensional quantum field theories \cite{Maldacena:1997re}. This conjecture, known as the AdS/CFT correspondence, is an explicit realization of the holographic principle and provides a host of tools for rigorously probing quantum gravity. Perhaps most useful is the nature of the AdS/CFT correspondence, a strong/weak coupling duality, in which the parameters of strongly-correlated field theories are expressible as classical (weakly-coupled) gravitational objects. Accordingly, certain properties which are inherently difficult to compute in one theory may be ported over to their corresponding dual and evaluated using the complement theory. While various extensions and generalizations of the AdS/CFT formulation exist, we will focus solely on the initial prescription in this document.

\subsection{Anti de-Sitter Spacetime}

The gravitational theory in AdS/CFT takes place on an Anti-de Sitter (AdS) spacetime, a maximally-symmetric solution to Einstein's equations with constant negative curvature. More precisely we require only that our spacetime be asymptotically AdS, i.e. the manifold behaves like AdS at large length scales. In this large radial limit, the spacetime metric is given by \cite{Headrick:2019eth}
\begin{equation}\label{introAdS}
    ds^2 = \frac{\ell^2}{z^2}\left(ds^2_B + dz^2 \right),
\end{equation}
where $\ell$ is known as the AdS radius, and $ds^2_B$ provides the metric for the conformal boundary $B$. 

For simplicity, it is often preferred to compute quantum field theories on flat spacetimes. Fortunately, $d+1$ dimensional AdS admits a boundary at spatial infinity which is precisely $d$-dimensional Minkowski space%
\footnote{The boundary manifold my be thought of as ordinary $d$-dimensional Minkowski space with the addition of several points at infinity to resolve conformal mappings from $\mathcal{M}_d$ to $\infty$.} %
under a particular choice of compactification. This conformal compactification also permits operator action under $SO(d,2)$, thus allowing conformal field theories to exist on the boundary of AdS with the necessary symmetry structure that enables the AdS/CFT correspondence \cite{Witten:1998qj}.

When the AdS length scale grows large compared to the Planck scale, string fluctuations become negligible and the gravitational theory in the bulk AdS spacetime behaves classically. Since AdS/CFT is a strong/weak coupling duality, this weak string coupling in the gravity theory corresponds precisely to the limit when degrees of freedom in the boundary CFT grow large, the so-called large $N$ limit \cite{HOOFT1974461}. This relationship between bulk and boundary degrees of freedom (see \eqref{introCentralCharge}) determines the energy sector of the theory, with the UV found nearest to the boundary and the IR located deep in the AdS bulk.

Until now we have not specified the dimension of our AdS spacetime. Perhaps the strongest demonstration of the AdS/CFT correspondence equates Type IIB supergravity on $AdS_5 \times S^5$ with $4$-dimensional $\mathcal{N}=4$ super Yang-Mills theory \cite{Maldacena:1997re,Gubser:1998bc,Witten:1998qj}. In this paper however, for reasons which will soon become clear, we concern ourselves only with the relationship between $3$-dimensional gravity ($AdS_3$) and $2$-dimensional conformal field theories ($CFT_2$). We now give a brief introduction to conformal field theories and the unique properties that emerge in two dimensions.

\subsection{Conformal Field Theories}

The second component of the AdS/CFT correspondence centers around a special class of quantum field theories with particularly friendly structure. Conformal field theories (CFTs) are quantum field theories invariant under the set of conformal transformations \cite{di1996conformal}, i.e. metric transformations of the form
\begin{equation}\label{introConformalInvariance}
    g_{\mu \nu}'(\textbf{x}') = \Lambda(\textbf{x})g_{\mu \nu}(\textbf{x}).
\end{equation}
The set of all transformations in \eqref{introConformalInvariance} forms a group, known as the conformal group, which contains the Poincare group as well as the set of all special conformal transformations on $g_{\mu \nu}$. Explicitly stated, the conformal group describes the following set of coordinate transformations \cite{di1996conformal}
\begin{equation}\label{introConformalGroup}
\begin{split}
    \textnormal{Translation:}& \quad x^{'\mu} = x^{\mu} + a^{\mu},\\
     \textnormal{Dilation:}& \quad x^{'\mu} = \lambda x^{\mu},\\
      \textnormal{Rotation:}& \quad x^{'\mu} = A^{\mu}_{\nu} x^{\nu},\\
       \textnormal{Special Conformal Transformation:}& \quad x^{'\mu} = \frac{x^{\mu} - b^{\mu}\vec{x}^2}{1-2 \vec{b} \cdot \vec{x} + b^2\vec{x}^2}.\\
\end{split}
\end{equation}
Alternatively, the above set of transformations may be thought of as the set of operations which locally preserve the angles of intersecting curves at a point. 

In two dimensions there exists an infinite set of unique local conformal transformations, which constrains the form of $n$-point correlation functions, therefore enabling exact solutions for $2$-dimensional CFTs \cite{BELAVIN1984333}. More rigorously stated, the generators of $2$-dimensional conformal transformations, shown in \eqref{introVirasoro}, satisfy an infinite-dimensional algebra known as the Virasoro algebra \cite{Virasoro:1969zu,Fubini:1971ce}. 

Upon quantizing $2$-dimensional conformal fields, the arrival of a symmetry deformation gives rise to a central charge $c$ often called the conformal anomaly. This central charge results in the addition of a $c$-dependent term, such that the Virasoro algebra is defined
\begin{equation}\label{introVirasoro}
\begin{split}
    [L_n, L_m] & = (n-m)L_{n+m} + \frac{c}{12}n(n^2-1)\delta_{n+m,0},\\
    [\bar{L}_n, \bar{L}_m] & = (n-m)\bar{L}_{n+m} + \frac{c}{12}n(n^2-1)\delta_{n+m,0},\\
    [L_n,\bar{L}_n] & = 0.\\
\end{split}
\end{equation}
Generators without an overline in \eqref{introVirasoro} act only on the holomorphic sector of the theory, whereas those with a bar act on the antiholomorphic sector.

In the large $c$ limit, the Virasoro algebra reduces to the group $SO(d,1)$. Beginning with two identical copies of the Virasoro algebra in the limit of large central charge allows us to recover $SO(d,2)$, precisely the isometry group of $AdS_3$! This magnificent symmetry matching allows the central charge of a double-copied boundary CFT to be expressed using bulk gravitational parameters \cite{Brown:1986nw}
\begin{equation}\label{introCentralCharge}
    c = \frac{3\ell}{2G_N}.
\end{equation}
We observe in \eqref{introCentralCharge} one of the foundational equivalences provided by the AdS/CFT correspondence.

\subsection{Entanglement In Holography}

Now that we have established a connection between boundary quantum fields and bulk geometry, we can begin performing explicit calculations. One of the most famous computations performed in the AdS/CFT framework is the explicit interpretation of boundary entanglement as areas of bulk geodesics. Formulated in the celebrated Ryu-Takayanagi conjecture \cite{Ryu_2006}, this result enables the practical calculation entanglement in strongly-correlated CFTs using rudimentary techniques of general relativity.

Before stating the Ryu-Takayanagi conjecture, we first recall the origin of black hole entropy. Remarkable insight from Bekenstein and Hawking demonstrated that the thermal entropy of a black hole is proportional to its horizon area \cite{PhysRevD.7.2333}, specifically
\begin{equation}\label{introBH}
    S_{BH} = \frac{1}{4G_N}A[H],
\end{equation}
where $H$ represents the black hole horizon. \eqref{introBH} allows the black hole to be cast as a thermodynamic system to an observer located outside the event horizon. 

Placing a black hole in the center of AdS, allows us to inquire about corresponding thermal properties in the boundary field theory. Specifically, the existence of an AdS black hole implies that the boundary CFT is in a deconfined%
\footnote{Here this deconfinement phase transition in the field theory occurs when the thermodynamic parameters are of order $c$.}
phase \cite{Headrick:2019eth}. Assuming a bipartition of the boundary Hilbert space, as defined in \eqref{HilDecomp}, into subspaces $A$ and $\bar{A}$, we construct the thermal state
\begin{equation}\label{introBH}
    \ket{\psi} = \frac{1}{\sqrt{Z}} \sum_{n} e^{-\beta E_n/2} \ket{n}_A \otimes \ket{n}_{\bar{A}}.
\end{equation}
The state $\ket{\psi}$ is known as the thermofield double state, and is interpreted as a two-sided black hole in the dual AdS geometry \cite{Eberhardt:2019ywk,Zaffaroni:2000vh}.

The state in \eqref{introBH} possesses several properties which enable an interpretation of the Beckenstein-Hawking entropy as an entanglement measure. First, the entanglement structure of $\ket{\psi}$ enables local operators, acting on $A$ and $\bar{A}$ respectively, to have non-zero correlation. Additionally, tracing out $\bar{A}$ yields a thermal density matrix on $A$. The black hole entropy now corresponds to the entanglement entropy between subregions of the boundary state $\ket{\psi}$. The black hole horizon then generalizes to the bulk AdS surface of minimal area which separates $A$ from its complement $\bar{A}$. This generalization reveals the Ryu-Takayanagi conjecture, which states
\begin{equation}\label{introRT}
    S_A = \frac{1}{4G_N} A[\gamma_A],  
\end{equation}
where $\gamma_A$ now denotes the minimal area bulk surface homologous to boundary region $A$. \eqref{introRT} is especially useful for computing entanglement in field theories, where the log of a reduced density operator is non-trivially defined \cite{Rangamani:2016dms}.

In the following section we provide an introduction to quantum computation and its related role in investigating quantum gravity. Accordingly, we transition towards a discrete treatment of quantum systems, focusing on the group-theoretic properties of operators on $\Hil$. We use this mathematical framework to define the set of stabilizer states, a superset of the holographic states, which admit similar entanglement structure at low qubit number.

\section{Quantum Computing}

An intuitive introduction to quantum computing relies heavily on terminology inherited from its classical analog. Much like a classical storage device containing information about the state of our computer, we consider a quantum system which likewise possesses some computable parameters of interest. The system is free to evolve under a variety of operations which, in the simplest case, act as logic gates. A sequence, in time, of these quantum gates constitutes a quantum circuit. The system is measured at selected intervals along the evolution, and information is extracted accordingly.

A natural starting point for discussing quantum gates is the set of Pauli matrices, defined
\begin{equation}\label{introPauli}
    I=\begin{pmatrix}1&0\\0&1\end{pmatrix}, \,\, \sigma_X=\begin{pmatrix}0&1\\1&0\end{pmatrix}, \,\,
    \sigma_Y=\begin{pmatrix}0&-i\\i&0\end{pmatrix}, \,\,
    \sigma_Z=\begin{pmatrix}1&0\\0&-1\end{pmatrix}.
\end{equation}
As a matrix object, each Pauli is a Hermitian unitary with eigenvalue $\pm 1$. In the \{$\ket{0},\ket{1}\}$ basis the matrices in \eqref{introPauli} act as operators on the Hilbert space $\mathcal{H} \equiv \mathbb{C}^2$. As generators of a multiplicative group, the operators $\{I,\, \sigma_X,\, \sigma_Y,$ and $\sigma_Z\}$ build the single-qubit Pauli group $\Pi_1$.

We can generalize the action of the Pauli group to arbitrary qubit number by composing strings of Pauli operators. A Pauli string generalizes local Pauli action to an $n$-qubit setting by building each operator as the $n$-fold tensor product of Paulis. For example, $\sigma_Z$ acting on the $k^{th}$ qubit of an $n$-qubit system is constructed
\begin{equation}\label{introPauliString}
    I^1\otimes\ldots\otimes I^{k-1} \otimes \sigma_Z^k \otimes I^{k+1} \otimes \ldots \otimes I^n.
\end{equation}
We refer to the number of non-identity factors in \eqref{introPauliString} as the weight of the Pauli string. The full $n$-qubit Pauli group $\Pi_n$ can then be generated from all weight-1 Pauli strings.

Another useful set of quantum operations is the $n$-qubit Clifford group $\mathcal{C}_n$. The group $\mathcal{C}_n$ is a set of unitaries which normalizes the Pauli group, i.e. elements of $\mathcal{C}_n$ map Pauli operators to Pauli operators via conjugation. We alternatively define $\mathcal{C}_n$ as the multiplicative matrix group generated by the Hadamard, phase, and CNOT operations \cite{sylvester1867lx,hadamard1893resolution}
\begin{equation}\label{introCliffordGates}
    H\equiv\frac{1}{\sqrt{2}}\begin{pmatrix}1&1\\1&-1\end{pmatrix}, \qquad P\equiv\begin{pmatrix}1&0\\0&i\end{pmatrix}, \qquad CNOT_{1,2} \equiv \begin{pmatrix}
            1 & 0 & 0 & 0\\
            0 & 1 & 0 & 0\\
	    0 & 0 & 0 & 1\\
	    0 & 0 & 1 & 0
            \end{pmatrix}.
\end{equation}
The Hadamard and phase gates in \eqref{introCliffordGates} are local unitaries which act on a single qubit in some $n$-qubit system. Conversely, the multi-qubit gate $CNOT_{i,j}$ acts simultaneously on two qubits, a control bit $i$ and a target bit $j$. The CNOT gate interprets the state of the control bit, and flips the state of the target bit iff the control bit is ``on''.

We again may generalize the action of the gates in \eqref{introCliffordGates} to arbitrary $n$-qubit systems by building Clifford strings, exactly analogous to the Pauli string construction presented in \eqref{introPauliString}. Accordingly, the group $\mathcal{C}_n$ is generated by all weight-1 Clifford strings, and contains the group $\Pi_n$ as a subgroup. Note, for the case of a single qubit, the group $\mathcal{C}_1$ is generated using only Hadamard and phase. The group $\mathcal{C}_n$ holds special regard in quantum computing since Clifford circuits, i.e. all circuits composed of Clifford gate sequences, are efficiently-simulable on a classical computer \cite{Gottesman:1997zz,Gottesman:1998hu,}. In other words, with access to only Clifford gates and stabilizer measurements, we can perform any task of interest in polynomial time using classical computation techniques \cite{aaronson_improved_2004}.

\subsection{Stabilizer Formalism}

For any state $\ket{\psi} \in \Hil$, there exists%
\footnote{At the very least, there exists an identity operation on $\Hil$ which acts trivially on arbitrary $\ket{\psi}$.} %
a set of operations on $\Hil$ which leave $\ket{\psi}$ fixed. For some chosen $\ket{\psi}$, and group of operators $G$ acting on $\Hil$, the subgroup of operations which leave $\ket{\psi}$ invariant is referred to as the stabilizer group of $\ket{\psi}$, defined
\begin{equation}
    G_\ket{\Psi} \equiv \{g\in G\: | \:g\ket{\Psi}=\ket{\Psi}\}.
\end{equation}

When considering the action of $\Pi_1$ on $\Hil$, every state in the Hilbert space is trivially stabilized by $\sigma_{I}$. Certain states however, are stabilized by a larger subset of $\Pi_1$ elements, specifically
\begin{equation}\label{introOneQubitStabilizerStates}
    S_1 \equiv \left\{ \ket{0}, \ket{1}, \ket{\pm}\equiv\frac{1}{\sqrt{2}}(\ket{0}\pm \ket{1}), \ket{\pm i}\equiv\frac{1}{\sqrt{2}}(\ket{0}\pm i\ket{1})  \right\}.
\end{equation}
The $6$ states in \eqref{introOneQubitStabilizerStates} comprise the set of single-qubit stabilizer states. The definition of stabilizer states as the set of states in $\Hil$ which are invariant under the largest subset of $\Pi_n$ generalizes to all $n$.

An alternative way to construct all $n$-qubit stabilizer states is to begin with any state in the canonical measurement basis, e.g. $\ket{0}^{\otimes n}$, and act with $\mathcal{C}_n$ on that state \cite{aaronson_improved_2004}. Since $\mathcal{C}_n$ is finite the orbit of any state under $\mathcal{C}_n$ is likewise finite \cite{garcia2017geometry}, and the orbit of any measurement basis state is precisely the $n$-qubit stabilizer state set $S_n$. The order of both are known, and are given by \cite{Gross_2013}
\begin{equation}
\begin{split}
    \left|C_n\right| &= 2^{n^2+2n} \prod_{k=1}^{n} (4^k-1),\\
    \left|S_n\right| &= 2^n \prod_{k=0}^{n-1} (2^{n-k}+1).
\end{split}
\end{equation}

The Clifford group forms a set of measure zero in the space of linear operations on $\Hil$. In fact, to perform any process which could yield quantum supremacy we must move beyond the set $\mathcal{C}_n$. The property of quantum magic characterizes non-stabilizer behavior, and offers a valuable resource for near-term quantum computing \cite{Oliviero:2022euv}.

\subsection{Quantum Magic}

As previously stated $\mathcal{C}_n$ is not a universal quantum gate set, i.e. elements of $\mathcal{C}_n$ are not sufficient to approximate any unitary with arbitrary precision. However, the set $\mathcal{C}_n$ along with any $V \notin \mathcal{C}_n$ is a universal gate set for quantum computing. The operation $V$ can be implemented as a gate in the quantum computer, or more cleverly inserted via the coupling of a non-stabilizer state to the system. Such states which encode this non-Clifford action are known as magic states, e.g.
\begin{equation}
    \ket{T} \equiv \cos \beta \ket{0} + e^{i\frac{\pi}{4}}\sin \beta \ket{1}, \qquad \cos(2\beta) = \frac{1}{\sqrt{3}}.
\end{equation}
The state $\ket{T}$ cannot be reached by applying only Clifford operations to a single-qubit measurement state.

Currently there is no universal magic measure, however many quantities help characterize the magic in a quantum system. One such quantity, known as the trace distance of magic \cite{Cao:2024nrx}, is defined
\begin{equation}\label{introMagic}
    M_{dist.}(\rho) \equiv \min_{\sigma \in \textnormal{Stab$_0$}} \frac{1}{2} ||\rho-\sigma||_1,
\end{equation}
where $Stab_0$ is the set of all states with zero stabilizer Renyi entropy \cite{Leone:2021rzd}. \eqref{introMagic} offers an intuitive notion of magic, computed as the trace distance to the nearest stabilizer state. Perhaps less-intuitive is the fact that magic can also exist non-locally in entangled quantum states \cite{Bao_2022}, with certain entanglement structures yielding states with higher magic than others \cite{Leone:2021rzd}. The nuanced relationship between entanglement and magic is discussed in detail in Chapter \ref{Chapter7}, as well as the role of magic in holography.

The papers in this thesis follow the order of their production. 
\begin{itemize}
    \item In Chapter \ref{Chapter2} we discuss \cite{Cardona:2021kyh}, where the modular conformal bootstrap technique is applied to four-point scalar functions in the lightcone limit. The OPE spectral decomposition of the Virasoro vacuum is computed in the large-dimension limit. A kernel ansatz method is further presented to generalize calculations beyond this large dimension limit. The author's primary contributions to this work include calculations of OPE spectral coefficients, as well as functional analysis.
    \item In Chapter \ref{Chapter3} we discuss \cite{Keeler:2022ajf}, where entropy vectors of stabilizer states are computed and analyzed using a graph-theoretic framework. The nature of these reachability graphs is investigated, and important dimension-dependent characteristics are highlighted. We offer a protocol for constructing restricted graphs, where group action on the Hilbert space is constrained to a chosen subgroup of operators. We observe stabilizer states which violate holographic entropy conditions, and remark on their relation to other states under Clifford operations. The author's primary contributions include developing a graph model for state orbits, programming and package development, calculating stabilizer entropy vectors, and analysis of reachability graphs.
    \item In Chapter \ref{Chapter4} we discuss \cite{Keeler:2023xcx}, where the structure of the $n$-qubit Clifford group is studied in depth. Useful relations between different Clifford circuits are analyzed, and a state-independent method for constructing reachability graphs is presented. The quotient protocol derived herein is used to explore the Clifford orbit of non-stabilizer states. The author's primary contributions include determining a quotient structure on Cayley graphs to yield reachability graphs, programming and package development, and an analysis of graph and group theoretic structures.
    \item In Chapter \ref{Chapter5} we discuss \cite{Munizzi:2023ihc}, where a formula for computing entanglement entropies in Dicke states is derived. This calculation of entanglement entropies enables us to generate all Dicke state entropy vectors, and accordingly describe the Dicke state entropy cone. Dicke state stabilizers under Clifford and Pauli group action are give, and orbits of Dicke states under these two groups are shown. The author's primary contributions to this work include developing a generalized form of Dicke state entanglement entropy, programming and package development, and Dicke state orbit calculations.
    \item In Chapter \ref{Chapter6} we discuss \cite{Keeler:2023shl}, where a new graph quotient, known as a contracted graph, is presented. Contracted graphs enables a rigorous analysis of state parameter evolution under group action on a Hilbert space. Furthermore, the structure of contracted graphs provides a strict bound on state parameter evolution under circuits composed of the group operators. We use the contracted graph technique to give an upper-bound on achievable entropy vectors built using Clifford circuits. The author's primary contributions to this work include developing a Cayley graph quotient that reproduces parameter evolution, formalising group and graph-theoretic concepts, programming and package development, and exploring applications for contracted graph protocols.
    \item In Chapter \ref{Chapter7} we discuss \cite{Cao:2024nrx}, where various magic measures and their associated estimates are introduced and compared. The relation between entanglement and magic is explored, and we conjecture on the role of non-local magic in conformal field theories. Finally, a gravitational dual for non-local magic in holographic CFTs is proposed. The author's primary contributions to this work include calculations of non-local magic in CFTs, and some analyses of the relationship between entanglement entropy and magic.
    \item In Chapter \ref{Chapter8} we conclude this dissertation. We propose a collection of interesting open questions, and suggest extensions of the work herein that are worthy of future study.
\end{itemize}

\clearpage
\begin{singlespace}
\printbibliography[heading=subbibliography]
\end{singlespace}

\chapter{FOUR-POINT CORRELATION MODULAR BOOTSTRAP FOR OPE DENSITIES}\label{Chapter2}

\textit{The contents of this chapter were originally published in the Journal of High Energy Physics \cite{Cardona:2021kyh}.}

\textit{In this work we apply the lightcone bootstrap to a four-point function of scalars in two-dimensional conformal field theory. We include the entire Virasoro symmetry and consider non-rational theories with a gap in the spectrum from the vacuum and no conserved currents. For those theories, we compute the large dimension limit ($h/c\gg1$) of the OPE spectral decomposition of the Virasoro vacuum.  We then propose a kernel ansatz that generalizes the spectral decomposition beyond $h/c\gg1$. Finally, we estimate the corrections to the OPE spectral densities from the inclusion of the lightest operator in the spectrum.}

\section{Introduction and motivation}

The study of conformal field theories in dimensions greater than two has seen a rapidly-increasing interest in recent years due to successful  application of the bootstrap program to conformally symmetric correlation functions,  revived in \cite{Rattazzi:2008pe, ElShowk:2012ht}.  In particular, by taking the lightcone limit of the crossing equation and expanding in large spin, \cite{Fitzpatrick:2014vua, Komargodski:2012ek, Alday:2016njk, Alday:2015ewa} obtained analytic results for the OPE coefficients and the anomalous dimensions of the spectrum. In this regime, the resultant quantities reveal properties universal to all conformal field theories in dimension three and higher. 

In two dimensions, the conformal symmetry extends to an infinite-dimensional Virasoro algebra. In the semi-classical limit, where the central charge is very large, this algebra reduces to the same finite-dimensional group $SO(d,1)$ as in the higher-dimensional case.  Thus, all the successful  techniques applied in higher dimensions can, in principle, be used in the semiclassical limit of any (non-rational) two-dimensional CFT.  However, applying these higher-dimensional techniques to the two-dimensional case does not produce a 
straightforward result due to several features of the two-dimensional theory. First, the two-dimensional theory has the infinite-dimensional Virasoro algebra instead of just the global conformal algebra. Next, in more than two dimensions,  all non-vacuum operators have higher weight.  In two dimensions,  instead,  all the operators in the Virasoro vacuum Verma module, such as the stress tensor and its composites, have zero weight; therefore a larger family of operators contributes in the same multiplet.

Two-dimensional CFTs themselves can be further divided into  rational and non-rational theories.  Techniques including integrability,  vertex algebras \cite{kac1998vertex} \footnote{For a nice recent account on the use of vertex algebras in this context see \cite{Cheng:2020srs}.},  and the quantum groups approach \cite{Ponsot:2000mt, Ponsot:1999uf} have led to impressive progress towards classifying the space of rational field theories.  Unfortunately for generic non-rational conformal field theories in two dimensions, this analysis is far from complete even in more-controllable regions of the parameter space, such as the lightcone limit. 

In order to translate the large spin results achieved in higher dimensions to the two-dimensional case,  we require a decomposition that accounts for the entire extended Virasoro symmetry.   A powerful tool that allows for a clean  analysis of large spin quantities in higher-dimensional conformal field theories is the Lorentzian inversion formula \cite{Caron-Huot:2017vep,  Simmons-Duffin:2017nub}, which shows the analyticity in spin of OPE coefficients (see \cite{Liu:2018jhs, Cardona:2018dov,  Cardona:2018qrt} 
 for earlier applications of the inversion formula to compute anomalous dimensions).  Unfortunately, we cannot directly apply this inversion formula to the two-dimensional case as it does not incorporate the entire Virasoro symmetry.  Remarkably,  progress is still possible:  there is  an analogous inversion formula relating t-channel to s-channel data that has been known for two decades \cite{Ponsot:2000mt, Ponsot:1999uf, Teschner:2001rv}.  By leveraging this impressive tool,  \cite{Collier:2018exn,  Collier:2019weq} studied the universal OPE large spin asymptotics for non-rational CFTs.
 
In the present work, we aim to reproduce some of these results via more traditional bootstrap methods.  Even though explicit forms for the Virasoro blocks are only known in particular limits, e.g. ~the semi-classical limit, their form is still constrained enough to justify certain conclusions regarding the large spin OPE spectral densities.  Previous work using the same strategy was done in \cite{Das:2017cnv, Kusuki:2019avm,Kusuki:2018wcv},  but we are able to extend their results by proposing a kernel ansatz that goes beyond the strict semi-classical limit previously considered.  

Most of the previous work using the modular bootstrap focuses on the partition function; for example, \cite{Kraus:2016nwo,  Keller:2014xba} studied the asymptotic formula for the average value of light-heavy-heavy three-point coefficients,  in a generalization of Cardy’s formula for the high energy density of states, while \cite{Benjamin:2016fhe, Benjamin:2019stq,Cho:2017fzo,  Keller:2017iql,  Dyer:2017rul, Lin:2021udi} further generalized these results. For the four-point function, \cite{Das:2017cnv} follows a similar approach. Numerical work on the modular bootstrap has been undertaken in \cite{Collier:2016cls, Friedan:2013cba,  Hartman:2014oaa}, mostly focusing on the partition function. Our work here is highly motivated by these references.  

In this paper we study a slight modification of the lightcone limit, which we term the \emph{modular lightcone limit}.  In particular, we examine the crossing equation for the scalar four-point function, extracting the large spin OPE spectral densities using familiar bootstrap techniques that require the Virasoro conformal blocks.  We hope this paper will provide a step towards a generalization of the lightcone and the Euclidean bootstrap, both well understood in higher dimensions using conformal blocks, to the two-dimensional case where the Virasoro blocks are needed for computation. We begin by establishing notation and reviewing the Virasoro crossing equation in Section \ref{sec:Vir}.  In Section \ref{sec:Vac} we re-establish the results for the vacuum contribution to the OPE coefficient spectral density using a standard bootstrap approach, and in Section \ref{sec:BeyondVac} we extend these results to next order in the semi-classical limit by incorporating the first correction beyond the vacuum contribution.  We conclude and discuss the gravitational implications of our results in Section \ref{sec:Discussion}.

\section{Virasoro Crossing Equation}\label{sec:Vir}

We consider a scalar four-point function.  By conformal invariance,  this four-point function is given by
\be \label{4p}
\langle \prod_{i=1}^4{\cal O}_{\varphi}(x_i) \rangle = 
  \frac{1}{(x_{12}^2)^{\frac12(\Delta_1+\Delta_2)}(x_{34}^2)^{\frac12(\Delta_3+\Delta_4)}}
A(z,\bz),
\ee
where $A(z,\bz)$ characterizes the cross-ratio dependence.  The cross ratios themselves are given by
\be
z\bz={x_{12}^2x_{34}^2\over x_{13}^2x_{24}^2},\,\quad (1-z)(1-\bz)={x_{14}^2x_{23}^2\over x_{13}^2x_{24}^2}\,.
\ee

In two dimensions, fields factorize into holomorphic and anti-holomorphic components.  The total conformal dimension of an operator, $\D=h+\bar{h}$,  is the sum of its holomorphic  weight $h$ and its anti-holomorphic weight $\bar{h}$.  The operator's spin is $j=|h-\bar{h}|$, and  $c$ is the central charge of the CFT \footnote{Here we restrict to non-rational CFTs, therefore setting $c>1$ to avoid the minimal models.  When we do numerics, we further specialize to $c>25$. }. 

Here we study the OPE decomposition of the four-point correlation function  by a traditional bootstrap approach,  i.e.  by using Virasoro conformal blocks $\mathcal{F}$ and the crossing equation.  For the case where all external operators have the same conformal dimensions,  i.e. $\D_i=\D_0$ for  $i=1,\cdots, 4$, the crossing equation simplifies to
\begin{equation}
	\sum_{{\cal O}^s_i} C_{i}^2\, \mathcal{F}\left(h_i\middle|z\right)\bar{\mathcal{F}}\left(\bar{h}_j\middle|\bar{z}\right) = \sum_{{\cal O}^t_j} C_{j}^2\, \mathcal{F}\left(h_j\middle|1-z\right)\bar{\mathcal{F}}\left(\bar{h}_j\middle| 1-\bar{z}\right)\,.
\end{equation}
We introduce a spectral density of OPE coefficients defined as
\be\label{spectralC}
C(h)= \sum_{i}C_i\,\d (h-h_i),
\ee
where $h_i$ denotes the conformal dimension of the exchange operator ${\cal O}_i$, with an analogous definition for the barred coefficients.  In this way the crossing equation takes the form
\beq\label{crossing2}
	&&\int_{0}^{\infty}d h_s d\bar{h}_s \,C(h_s)^2\,\mathcal{F}\left(h_s\middle|z\right)\bar{\mathcal{F}}\left(\bar{h}_s\middle|\bar{z}\right)\nn\\
	 &&= 	\int_{0}^{\infty}d h_t d\bar{h}_t \,C(h_t)^2\, \mathcal{F}\left(h_t\middle|1-z\right)\bar{\mathcal{F}}\left(\bar{h}_t\middle| 1-\bar{z}\right).
\eeq

An important comment is in order here.   As will be demonstrated in subsequent sections,  we will obtain solutions to the spectral density which are approximated by smooth functions rather than linear combinations of distributions.  The proper statement would be that those solutions approximate the RHS of equation \eqref{spectralC} in a smeared sense,  i.e.  after integration against an appropriate test function.  The smearing mechanism required to make this statement rigorous has been recently explored through the use of Tauberian theorems, for example \cite{Mukhametzhanov:2019pzy,  Pal:2019zzr,  Das:2020uax}%
\footnote{ We give special thanks to Alex Maloney for providing a clear explanation of this issue.}.%

For the Virasoro blocks,  we will use the elliptic representation found by Zamolodchikov \cite{Zamolodchikov:426555}:
\be\label{Virasoros}
{\cal F}_{h_0}(h_s|z)=(1-z)^{\frac{c-1}{24}-2h_{0}}\left(z\right)^{\frac{c-1}{24}-2h_{0}}[\theta_{3}\left(q\right)]^{\frac{c-1}{2}-16h_{0}}\left(16q\right)^{h_s-\frac{c-1}{24}}H(h_s,q)\,,
\ee
where $q$  is known as the elliptic nome and $h_0$ is the conformal dimension of the external operators, while $h_s$ is the exchange dimension. The elliptic nome can be thought of as a conformal transformation:
\beq\label{nome}
q=e^{i\pi\tau(z)}\,, \qquad
\tau(z)=i{K(1-z)\over K(z)}\,,\qquad K(z)={1\over 2}\int_0^1{dt\over \sqrt{t(1-t)(1-zt)}}\,,
\eeq
so $K$ is an elliptic integral of the first kind.  In \eqref{Virasoros}, the function $\theta_{3}$ is the Jacobi theta function
\begin{equation}
    \theta_3(q) \equiv \sum_{n\in \mathcal{Z}}q^{n^2},
\end{equation}
and the function $H(h_s,q)$ in \eqref{Virasoros} is unknown in closed form, but can be computed recursively as a power expansion in $q$ to very high order \cite{Zamolodchikov:1987, Zamolodchikov:1985ie}.  In the semi-classical regime, where $h_s\gg c$, $H(h_s,q) \approx 1$; the overall prefactors  in \eqref{Virasoros} thus capture the semi-classical behavior.
Later we will also use the notation $\tilde{q}=q(-{1\over \tau})$,  corresponding to the modular S transformation.  

Given the elliptic representation for the Virasoro blocks \eqref{Virasoros},  we are mainly interested in two particular limits of the crossing equation.  For later convenience,  let us define  $\tau \equiv {i\beta\over \pi}$.   First we consider the limit $\b \to 0$, with $\bar{\b}$ fixed.   Then,  as we will show,  the limit $\bar{\b}\to 0$ (Euclidean) simply becomes a copy of $\b \to 0$, with all quantities replaced by their bar counterparts.  We then study the limit $\bar \b \to \infty$ while keeping  $\b\to 0$, which we term the modular lightcone limit due to its similarities with the better-known global lightcone limit.

\section{Spectral density OPE for vacuum}\label{sec:Vac}

In this section we study the OPE spectral density in the s-channel in the limit where the t-channel contribution is only due to the vacuum.  Specifically, we first take the $\b\to 0$ limit of  the four-point function crossing equation while leaving $\bar{\b}$ fixed, isolating the vacuum contribution in the t-channel; we will return to contributions beyond this vacuum limit in Section \ref{sec:BeyondVac}. In the $\b \to 0$ limit, we then study the s-channel contribution,  proposing  a kernel ansatze for the spectral density $C(h_s)$ in both the Euclidean limit $\bar \b \to 0$ and the modular lightcone limit $\bar \b \to \infty$. 
\subsection{Small $\b$ limit} \label{smallbeta}

In order to isolate the t-channel vacuum contribution, we need to consider the leading behavior of the conformal blocks when $\b\to 0$, as we now show. We begin by inserting the elliptic representation for the Virasoro blocks \eqref{Virasoros} into the crossing equation \eqref{crossing2}.  We note that the t-channel blocks depend on $(1-z)$, and thus they take the form \eqref{Virasoros} except with $z\leftrightarrow 1-z$ and $q\rightarrow \tilde q$, where $\tilde q = q(-1/\tau)$ as in \eqref{nome}.

To compare the terms independent of exchange dimension, we rewrite the $\theta_{3}\left(\tilde q \right)$ from the t-channel blocks using the modular transformation of the theta function
\be
 \theta_{3}\left(\tilde q\right)=\left({\beta\over \pi}\right)^{1/2}\theta_{3}\left(q\right),
\ee
where we have used $\tau={i\beta\over \pi}$ and $\tilde q = q(-1/\tau)$.
Accordingly, the terms independent of exchange dimension in \eqref{Virasoros} cancel when we plug into the crossing equation, up to powers of $\beta$ and $\bar\beta$.  Explicitly the crossing equation \eqref{crossing2} becomes
\beq\label{crossing1}
	&&\int_{0}^{\infty}d h_s d\bar{h}_s \,C(h_s)^2 \,\left(16\right)^{h_s+\bar{h}_s}\,
	e^{-\b(h_s-\frac{c-1}{24})}e^{-\bar{\b}(\bar{h}_s-\frac{c-1}{24})}H(h_s,e^{-\b})H(\bar{h}_s,e^{-\bar{\b}})\\
	 &&= \left(\frac{\b \bar \b}{\pi^2}\right)^{\frac{c-1}{4}- 8h_{0}}\int_{0}^{\infty}d h_t d\bar{h}_t \,C(h_t)^2 \left(16\right)^{h_t+\bar{h}_t}	e^{-{\pi^2}(h_t-\frac{c-1}{24})/\b}e^{-{\pi^2}(\bar{h}_t-\frac{c-1}{24})/ \bar{\b}}
	 H(h_t,e^{-\pi^2/\b})H(\bar{h}_t,e^{-\pi^2/\bar{\b}})\,.\nn
\eeq

In the limit $\b\to0$, the leading contribution in the holomorphic t-channel comes from the operator with smallest weight $h_t$, because any heavier operator is exponentially suppressed by the $e^{-1/\beta}$ terms. In general dimensions, the smallest weight operator would just be the vacuum; in two dimensions,  the vacuum Verma module captures contributions of all its descendants,  which includes the stress tensor.  We will nonetheless still use the term vacuum block to refer to it.


In order to have only vacuum block contributions at leading order in the small $\b$ limit of the t-channel, we also need to disallow contributions from representations with weight $h=0$ but $\bar h \neq 0$ (and vice versa). Stated in another way, we do not want to allow for conserved currents, because a CFT possessing a primary operator that is also a conserved current will have a vanishing gap between the vacuum block and the rest of the spectrum, as has been recently shown in \cite{Benjamin:2020swg}.       Instituting this restriction, the dominant contribution given by the  holomorphic vacuum block $h=0$ only couples to the corresponding anti-holomorphic vacuum block $\bar{h}=0$. The crossing equation thus simplifies to
\beq
	&&\int_{0}^{\infty}d h_s d\bar{h}_s \,C(h_s) C(\bar{h}_s)\left(16\right)^{h_s+\bar{h}_s}\,
	e^{-\b(h_s-\frac{c-1}{24})}e^{-\bar{\b}(\bar{h}_s-\frac{c-1}{24})}H(h_s,e^{-\b})H(\bar{h}_s,e^{-\bar{\b}})\nn\\
	 &&= \left(\b\bar \b /\pi^2\right)^{\frac{c-1}{4}-8h_{0}}e^{{\pi^2}(\frac{c-1}{24})/\b}e^{{\pi^2}(\frac{c-1}{24})/\bar\b}\,C(0)^2
	 H(0,e^{-\pi^2/\b})H(0,e^{-\pi^2 /\bar{\b}}).
\eeq
As we mentioned at the end of Section \ref{sec:Vir},  the function $H(h, q)$ can be computed as a series expansion in small $q$ recursively, and we have $H(h,q) \sim 1+\mathcal{O}(q)$. Since the limit $\b\to 0$ implies  $q\to 0$ in the t-channel blocks,  we can simplify the crossing equation even further, obtaining
\beq\label{crossing_total}
	&&\int_{0}^{\infty}d h_s d\bar{h}_s \,C(h_s)^2\,\left(16\right)^{h_s+\bar{h}_s}\,
	e^{-\b(h_s-\frac{c-1}{24})}e^{-\bar{\b}(\bar{h}_s-\frac{c-1}{24})}H(h_s,e^{-\b})H(\bar{h}_s,e^{-\bar{\b}})\nn\\
	 &&=(\b\bar \b/\pi^2)^{\frac{c-1}{4}-8h_{0}}e^{{\pi^2}(\frac{c-1}{24})/\b}e^{{\pi^2}(\frac{c-1}{24})/\bar\b}\,C(0)^2 H(0,e^{-\pi^2 /\bar{\b}}).
\eeq

\subsection{Saddle point}\label{holosaddle}
Just as in the lightcone limit in higher dimensions \cite{Komargodski:2012ek,  Fitzpatrick:2014vua},  we can see that the crossing equation \eqref{crossing_total} develops an essential singularity on the righthand side as $\b\to 0$.  On the lefthand side, all of the terms in the integrand can be expanded in small $\beta$.  The only way a Taylor series in $\beta$ can equal an essential singularity is if it has an infinite number of terms, so we should expect infinite contributions to the s-channel sum on the lefthand side.  Accordingly, this sum should be dominated by the tail, i.e. the s-channel expansion of the vacuum block must be dominated by large values of $h_s$.  We now verify this intuition with a saddle point analysis.

Before moving on to solving the crossing equation we have to address an important subtlety.  
In order to solve for the spectral densities,  we need an explicit form for the block $H(h_s,e^{-\b})$,  which is unfortunately not known in general.  However,  as we argued in the previous paragraph,  we mainly need its behavior for large values of $h_s$. Fortunately,  it turns out that the blocks simplify dramatically in this limit  \cite{Cardona:2020cfy}.  At leading order in inverse powers of  $h_s$,  the Virasoro blocks are approximated by%
\footnote{Beyond first order has been considered in \cite{Das:2020fhs}.}%
\be\label{virq}
H(h_s,\b)\equiv 1-{H_{-1}\over  h_s}\left({E_2(\b)-1\over 24}\right) +{\cal O}(1/h_s^2)\,,
\ee
where 
\be
H_{-1}={((c+1)-32 h_0) ((c+5)-32 h_0) \over 16}\,,
\ee 
and $E_2(\b)$ is an Eisenstein series of weight two. 

The needed limit $\b\to 0$ is not expected to commute with the large  $h_s$ approximation,  as has been recently argued numerically \cite{Das:2020uax}.  However, if we constrain ourselves to a region where the constant $H_{-1}$ in \eqref{virq} becomes small, i.e. either
$32 h_0\sim c+1$ or $32 h_0\sim c+5$, then the $\mathcal{O} (1/h_s)$ terms become small regardless of the value of $\b$.  If we constrain the relation between the external operator dimension $h_0$ and the central charge $c$ in this way, we can then take the limit $\b \to 0$ and then take $h_s$ large, at the cost of keeping $c$ finite.%
\footnote{See \cite{Das:2020uax} for details.}
Thus, at zeroth order in $1/h_s$,  with $h_0\sim {c+1\over 32 }$ and in the $\b\to 0$ limit,   the  crossing equation \eqref{crossing_total} becomes%
\footnote{The reader may be concerned that we have approximated our integrand at large $h_s$ without modifying the integral as a whole.  It turns out the saddle point expansion that follows will actually provide a justification for this approximation in hindsight. We can rewrite any integral $\int_0^{\infty} \mathcal{F}(\mathcal{H},e^{-\beta}) d\mathcal{H}$ as $\int_0^{\mathcal{H}_{\Lambda}} \mathcal{F}(\mathcal{H},e^{-\beta}) d\mathcal{H} + \int_{\mathcal{H}_{\Lambda}}^{\infty} \mathcal{F}(\mathcal{H},e^{-\beta}) d\mathcal{H}$, for some finite $\mathcal{H}_{\Lambda}$. $\int_0^{\mathcal{H}_{\Lambda}} \mathcal{F}(\mathcal{H},e^{-\beta}) d\mathcal{H} $ remains finite in the limit $\beta \rightarrow 0$; specifically it is bounded by the finite value $\int_0^{\mathcal{H}_{\Lambda}} \mathcal{F}(\mathcal{H},1) d\mathcal{H}$. For our case, $\int_{\mathcal{H}_{\Lambda}}^{\infty} \mathcal{F}(\mathcal{H},e^{-\beta}) d\mathcal{H}$ diverges as $\beta \rightarrow 0$, as long as we pick $\mathcal{H}_{\Lambda}$ below the saddle, as we can see from the divergence of the saddle itself.  Since an infinite contribution will always dominate over a finite one, the approximation within the integrand is justified.}%
\beq\label{smallb_largeh}
	&&\int_{0}^{\infty}d h_s \int_{0}^{\infty}d\bar{h}_s \,C(h_s)^2\, \left(16\right)^{h_s+\bar{h}_s}\,
e^{-\b(h_s-\frac{c-1}{24})}e^{-\bar{\b}(\bar{h}_s-\frac{c-1}{24})}H(\bar{h}_s,e^{-\bar{\b}})\nn\\
	 &&= (\b\bar \b/\pi^2)^{\frac{c-1}{4}-8h_{0}}e^{{\pi^2}(\frac{c-1}{24})/\b}e^{{\pi^2}(\frac{c-1}{24})/\bar\b}\,C(0)^2 H(0,e^{-\pi^2 /\bar{\b}})\, .
\eeq

We first examine the holomorphic part
\be
\int_{0}^{\infty}d h_s  \,C(h_s)\,\left(16\right)^{h_s}\,
e^{-\b(h_s-\frac{c-1}{24})}
= (\b/\pi)^{\frac{c-1}{4}-8h_{0}}e^{{\pi^2}(\frac{c-1}{24})/\b}\,C(0)\,.
\ee
Actually, we should be more precise; previous arguments  \cite{Collier:2016cls,  Benjamin:2019stq} have shown that any compact, unitary two-dimensional conformal field theory must have a gap smaller than $\mathcal{C}\equiv \frac{c-1}{24}$. Although it is a choice here, in the following section, the integration over $h_s$ `knows' about this maximal gap.  Accordingly, we will set the lower limit to be $\mathcal{C}$, not zero.  Rewriting in terms of $\mathcal{C}$, we have
\be\label{holomorphiccrossing}
\int_{\mathcal{C}}^{\infty}d h_s  \,C(h_s) \,\left(16\right)^{h_s}\,
e^{-\b(h_s-\mathcal{C})}
= (\pi/\b)^{\l}e^{{\pi^2}\mathcal{C}/\b}\,C(0)\,,
\ee
where we have additionally defined a shifted external operator weight $\l \equiv  8h_0 - \frac{c-1}{4}$ for future convenience.  The left hand side of this equation is now the Laplace transform of $C(h_s) 16^{h_s}$, with Laplace parameter $\beta$ and integration variable $h_s-\mathcal C$.  

Given the form of \eqref{holomorphiccrossing}, we can solve for the spectral OPE density $C(h)$ by applying the inverse Laplace transform. We find 
\be
C(h_s) 16^{h_s} = \frac{C(0)\pi^{\l}}{2\pi i} \int_{-i \infty}^{i\infty} \b^{-\l}\exp\left[\beta (h_s-\mathcal{C})+\frac{\pi^2 \mathcal{C}}{\b}\right]\, d\b.
\ee
The saddle point for this integral is given by
\begin{equation}\label{saddle1}
	\beta_s = \frac{\l}{(h_s - \mathcal{C})} \pm \frac{\sqrt{(-\l)^2 + 4\pi^2 \mathcal{C}(h_s - \mathcal{C})}}{2(h_s - \mathcal{C})}\sim  \frac{\l}{\mathcal{C}({h_s\over \mathcal{C}} - 1)} \pm \frac{\pi}{\sqrt{{h_s\over \mathcal{C}} - 1}}\,,
\end{equation}
%
We can see from this saddle point that the $\b\to 0$ limit is indeed dominated by large values of $h_s$.   As usual the integral at the saddle point approximation is a simple Gaussian that can be straightforwardly evaluated, giving the OPE spectral density in the large $h_s$ limit as
\be\label{spectral_saddle}
C(h_s) \, \sim {C(0)\over \left(16\right)^{h_s} \sqrt{\pi \mathcal{C}}}(\b_s)^{-\l-3/2}e^{\b_s(h_s-\mathcal{C})+ \frac{\pi^2}{\beta_s}\mathcal{C} }\,,
\ee
with $\b_s$ given by \eqref{saddle1}.  

In the limit  $h_s\gg\mathcal{C}$,  the second term from \eqref{saddle1} dominates. We can write the OPE spectral density in the semi-classical limit more explicitly as
\be\label{spectral_saddle_leading}
C(h_s) \, \sim {C(0)\over \left(16\right)^{h_s} \sqrt{\pi \mathcal{C}}}\left({h_s\over \mathcal{C}} - 1\right)^{-\l/2-3/4}e^{2\pi\mathcal{C}\sqrt{{h_s\over \mathcal{C}} - 1}}\,.
\ee
This result is in agreement with previous similar analyses done in \cite{Kraus:2016nwo,  Das:2017cnv, Kusuki:2018nms}.

\subsection{Kernel ansatz}
In section above,  we have computed an approximate solution for $C(h_s)$ by means of a saddle point analysis.  However,  we did need to assume that this saddle point method is valid for any sufficiently large value of $h_s$.  Examining the evaluated saddle point \eqref{saddle1},  we see the $\b\to 0$ limit is indeed reliable at this point when $h_s/\mathcal{C}\gg1$.  In this section, we propose a kernel ansatz that generalizes the result for $C(h_s)$ from the saddle point analysis and henceforth allows matching the t-channel vacuum via integration in the s-channel in a wider regime, down towards $h_s/\mathcal{C}\sim \mathcal{O}(1)$.  

We begin with the integral 
 \be\label{hologuess}
 \int_0^{\infty}d\mathcal{H}\,\cosh(2\pi  \sqrt{\mathcal{H}\mathcal{C}} )\,\mathcal{H}^{a^2}e^{- \b\mathcal{H}}
 =\beta ^{-a^2-1} \Gamma \left(a^2+1\right) \, _1F_1\left(a^2+1;\frac{1}{2};\frac{\mathcal{C} \pi ^2}{\beta }\right).
 \ee
The format of this integral matches the left hand side of the holomorphic crossing equation \eqref{holomorphiccrossing}, provided we 
introduce the convenient variable
\be
 \mathcal{H} \equiv h_s-\mathcal{C}\, .
\ee
We can match\footnote{The choice of integrand for the left hand side of \eqref{hologuess} is not unique, but rather belongs to a family of integrands that correctly reproduce the leading order terms on the right hand side after integration. When considering higher order terms, such as in \eqref{kernel_with_polynomial}, any appropriate selection of integrand yielding the correct matching order by order is sufficient for this calculation.} the right hand side of  \eqref{holomorphiccrossing} by taking the $\b \to 0$ limit of \eqref{hologuess}:
\be\label{hologuessLimit}
 \int_0^{\infty}d\mathcal{H}\,\cosh(2\pi  \sqrt{\mathcal{H}\mathcal{C}} )\,\mathcal{H}^{a^2}e^{- \b\mathcal{H}} 
 \approx {(\mathcal{C}\pi^2)^{a^2+1/2}\over \b^{2a^2+3/2}}e^{\pi^2\mathcal{C}\over \b}\,.
\ee
We set
\be
	a^2=\l/2-3/4=4 h_0-3 \mathcal{C}-3/4
\ee
by matching the $\b$ exponent in \eqref{holomorphiccrossing}.  We can then identify the spectral density reproducing the vacuum as
\begin{align}\label{kernel_vac}
C(h_s) &=  
 \frac{C(0) \pi^{1/2}}{16^{h_s} \mathcal{C}^{\eta_0/2 - 1/4}}\cosh\left[2\pi \sqrt{\mathcal{H}\mathcal{C}}\right](\mathcal{H})^{\eta_0/2-3/4}
\\
&=
\frac{C(0) \pi^{1/2}}{16^{h_s}\mathcal{C}^{1/2}}\cosh\left[2\pi \mathcal{C} \sqrt{\frac{h_s}{\mathcal{C}}-1}\right] \left(\frac{h_s}{\mathcal{C}}-1\right)^{\eta_0/2-3/4}.
\nn
\end{align}
As we anticipated in our saddle point analysis in \eqref{holomorphiccrossing},  the lower integration limit of \eqref{hologuess},    written in terms of $dh_s$, is $h_s= \mathcal{C}$.  Here we have a stronger justification: the integral `knows' about the maximal size of the gap via the kernel choice.  This lower bound is again in line with previous arguments about the size of the gap \cite{Collier:2016cls,  Benjamin:2019stq}. The vacuum in the t-channel is reproduced by the large $h_s$ tail of the integral, as expected from the saddle point analysis in Section \ref{holosaddle}. 

This result agrees with the saddle point analysis for $h_s\gg\mathcal{C}$. As we have explicitly shown, it also reproduces the vacuum block when $h_s\sim \mathcal{C}$.  Therefore, the kernel ansatz \eqref{kernel_vac} should be thought of as a generalization of the spectral density that is computed by an inverse Laplace transform followed by a saddle point analysis to fix the constants. 

As a final step,  we use the spectral densities to define an average value for the OPE coefficients \cite{Das:2017cnv}
\be\label{average0}
 \mathbb{C}(h_s) ={C(h_s)\over\mathcal{S}_0}\,,
\ee
where $\mathcal{S}_0$ is the asymptotic density of primary states.%
\footnote{For a more rigorous treatment of this average coefficients from Tauberian theorems of distributions, see \cite{Das:2020uax,  Mukhametzhanov:2019pzy} } %
 In other words, $\mathcal{S}_0$ is the fusion kernel for the vacuum character decomposition associated to the partition function  \cite{Zamolodchikov:2001ah,  Collier:2018exn}, given by
\be\label{partition_fusion}
\mathcal{S}_0=4\sqrt{2} \sinh\left(2\pi b \sqrt{\mathcal{H}}\right) \sinh\left(2\pi b^{-1}\sqrt{\mathcal{H}}\right),
\ee
where $b$ relates to the central charge as $c= 1+ 6(b+1/b)^2$.  This density of primaries is a refined version of Cardy's formula  \cite{Cardy:1986ie}.  In the limit $h_s/\mathcal{C}\gg1$ equation \eqref{average0} can be written
\be\label{average_asym}
 \mathbb{C}(\mathcal{H})\sim \mathcal{H}^{\l/2-3/4} e^{-2\pi \sqrt{\mathcal{HC}}}.
\ee
This result is compatible with previous results \cite{Das:2017cnv,  Collier:2018exn,  Collier:2019weq}. 
\subsection{Anti-holomorphic piece}
When evaluating the Euclidean limit of the anti-holomorphic piece, $\bar{\b}\to 0$, the analysis is exactly the same as was performed for its holomorphic counterpart. Therefore the result for the anti-holomorphic spectral density in the Euclidean limit is given by \eqref{kernel_vac} with $\bar{h}_s$ substituted for $h_s$.  The opposite limit,  when $\bar{\b}\to \infty$, is more interesting and more involved,  as we will demonstrate in this section. 

Consider the anti-holomorphic portion of our crossing equation \eqref{crossing_total}:
 \be
\int_{0}^{\infty}d\bar{h}_s \,C(\bar{h}_s)\left(16\right)^{\bar{h}_s}\,e^{-\bar{\b}(\bar{h}_s-\mathcal{C})}H(\bar{h}_s,e^{-\bar{\b}})= \left(\bar{\b}/\pi\right)^{6\mathcal{C}-8h_0}e^{\pi^2 \mathcal{C}/\bar\b} \,
\bar C(0)
H(0,e^{-\pi^2/\bar{\b}})\,.
\ee
As in the holomorphic case, we do not have a useful form for the Virasoro blocks at finite $q$ and finite $\bar{h}$.  However, taking the limit $\bar{\b}\to \infty$ we can approximate
 \be\label{AntiHoloCrossApprox}
\int_{0}^{\infty}d\bar{h}_s \,\bar{C}(\bar{h}_s)\left(16\right)^{\bar{h}_s}\,e^{-\bar{\b}(\bar{h}_s-\mathcal{C})}= (\bar{\b}/\pi)^{6\mathcal{C}-8h_0}e^{\pi^2\mathcal{C}/\bar\b} \,\bar{C}(0)H(0,\tilde{q}\to 1)\,.
\ee
Evaluating the last term in this limit is subtle, but we can estimate using numerics.  Restricting to the case of central charge $c=30$ and external dimensions $h_0=1$,  which is close to the point ${c+1\over 32}$,  we computed%
\footnote{We have computed up to order 50 for computational time convenience,  but we checked that above order $130$,  the expansion becomes asymptotic in the limit $\tilde{q}\to 1$.} %
 $H(0,\tilde{q})$ up to  $\tilde{q}^{50}$,  and have found that the following function provides a good fit for the Virasoro block:
\begin{figure}
\begin{center}
\includegraphics[scale=0.6]{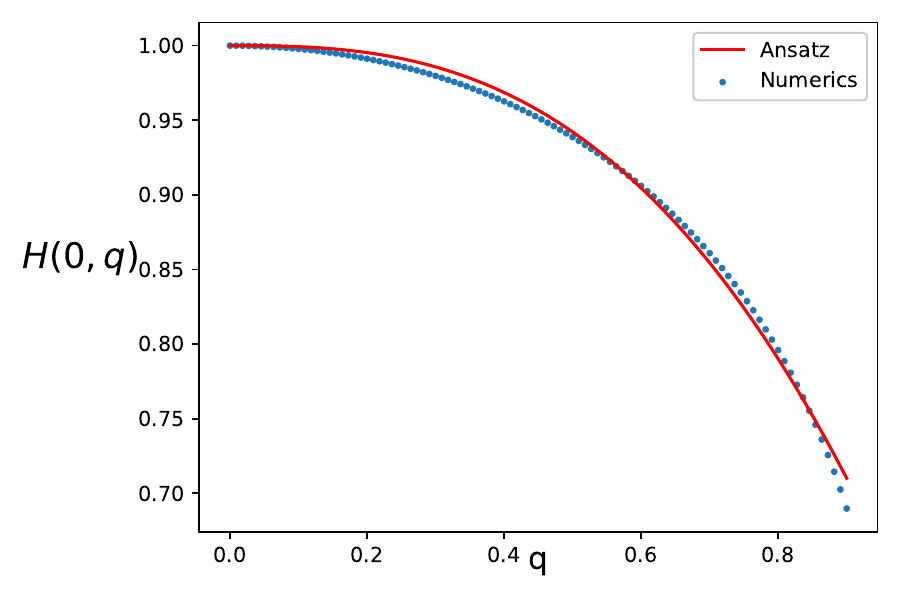}
\caption{ The dotted line corresponds to the vacuum Virasoro block computed numerically with the aid of Zamolodchikov recursion relations for the particular values $c=30$, $h_0=1$.  The solid red line is an approximation fit, explicitly given by equation \eqref{Hansatz}. }\label{H0vsq}
\end{center}
\end{figure}
\be\label{Hansatz}
H(0,\tilde{q})=1-a_1\,e^{-a_2\pi^2/\bar{\b}}\,.
\ee
Here we have introduced constants $a_1 = 0.8$ and $a_2 = 2.8$,  which are simply an approximation to the values given by the interpolating function fitting the numerical values for $H(0,\tilde{q})$.  If we instead did a numeric fit for the Virasoro block for different $c$ and $h_0$, we expect these values would change.  

%
\footnote{Here  we are not being too careful in finding the best fit to the numerical data,  since we just want to show how to deal with the anti-holomorphic piece once we interpolate the numerics.   We plan to  leverage the numerics in a more meaningful way in a future publication. } %
Inserting this approximation into the anti-holomorphic crossing equation \eqref{AntiHoloCrossApprox} we have 
 \be\label{cross_anti}
\int_{0}^{\infty}d\bar{h}_s \,\bar C(\bar{h}_s)\left(16\right)^{\bar{h}_s}\,e^{-\bar{\b}(\bar{h}_s-\mathcal{C})}= (\bar{\b}/\pi)^{6\mathcal{C}-8h_0}e^{\pi^2\mathcal{C}/\bar\b} \,\bar C(0)\left(1-a_1\,e^{-a_2\pi^2/\bar{\b}}\right)\,,
\ee
keeping in mind fixed parameters $h_0=1$ and $c=30$.  At leading order in $\bar{\b}\to \infty$ we can then equate
 \be\label{cross_anti_zero_order}
\int_{0}^{\infty}d\bar{h}_s \,\bar C(\bar{h}_s)\left(16\right)^{\bar{h}_s}\,e^{-\bar{\b}(\bar{h}_s-\mathcal{C})}= (\bar{\b}/\pi)^{6\mathcal{C}-8h_{0}}\,\bar C(0)\,.
\ee
We solve this equation using the same saddle point method as before,  employing a similar kernel ansatz as in the holomorphic analysis, with the important difference that we now take the  $\bar{\b}\to \infty$  limit of our result to compare with the expression above.  Explicitly, we consider the following integral
 \be\label{anitholoSimpleKernel}
\int_0^{\infty}d\bar{\mathcal{H}}\,\cosh(2\pi  \sqrt{\bar{\mathcal{H}}\mathcal{C}} )\,\bar{\mathcal{H}}^{\bar{a}^2}e^{- \bar{\b}\bar{\mathcal{H}}}= \bar{\b} ^{-\bar{a}^2-1} \Gamma \left(\bar{a}^2+1\right) \, _1F_1\left(\bar{a}^2+1;\frac{1}{2};\frac{\mathcal{C} \pi ^2}{ \bar{\b} }\right)\,,
\ee
where we have, as in the holomorphic sector, rewritten in terms of $\bar{\mathcal{H}}=\bar{h}_s-\mathcal{C}$, and set the lower bound at the maximal gap $\bar{h}_s=\mathcal{C}$.

In the limit $\bar \b \to \infty$, the argument of the hypergeometric function becomes $0$, so the lowest order term in $_1F_1$ is just one. Matching the lowest order power of $\bar \b$ with the right hand side of our crossing equation \eqref{cross_anti_zero_order},  we find 
\be\label{baraResult}
	\bar{a}^2=8 h_0-6\mathcal{C}-1=\l -1,
\ee
 where we have used our earlier definition of the shifted external weight $\l = 8h_0 - 6\mathcal{C}$. Solving for the OPE coefficient spectral density,  we find 
\begin{align}
\bar{C}(\bar{h}_s) &=\frac{\bar{C}(0)\pi^{\l}}{ 16^{\bar{h}_s} \Gamma\left(\l\right)}\bar{\mathcal{H}}^{\l-1}\cosh\left(2\pi  \sqrt{\bar{\mathcal{H}}\mathcal{C}} \right)
\nn\\
\label{kernel_vac_anti}&=  \frac{\bar{C}(0)\pi^{\l}}{ 16^{\bar{h}_s} \Gamma\left(\l\right)}\left(\bar{h}_s-\mathcal{C}\right)^{\l-1}\cosh\left(2\pi  \sqrt{(\bar{h}_s-\mathcal{C})\mathcal{C}} \right)\,.
\end{align}

This game can be continued to higher inverse-power orders in  $ \bar{\b}$. We wish to build a kernel ansatz for $\bar{C}({\bar{h}_s})$ that recovers the anti-holomorphic t-channel contribution on the righthand side of \eqref{cross_anti}. To simplify the algebra, we first expand this expression in powers of $x\equiv \pi^2\mathcal{C}/\bar{\beta}$: 
\begin{equation}
\label{Power_expansion_antiholo}
	\bar{C}(0) \left(\frac{\pi \mathcal{C}}{x} \right)^{-\l} \left(\sum_{k=0}^\infty \frac{x^k}{k!} \right) \left(1-a_1 \sum_{n=0}^\infty \frac{(-1)^n}{n!} \left(\frac{a_2 x}{\mathcal{C}} \right)^n \right).
\end{equation}

In order to match this t-channel expression, we need to upgrade the kernel \eqref{anitholoSimpleKernel}, to include a polynomial series in $\bmh$: 
\begin{align}
\label{kernel_with_polynomial}
	\int_0^\infty d\bmh \, \tilde{\mA}& \cosh(2\pi \sqrt{\bmh\mc}) \sum_{n=0}^{\infty} \mA_n \bmh^{\bar{a}^2+n} e^{-\frac{\pi^2}{x}\mc\bmh}
	\\\nn
	&=\tilde{\mA} \left( \frac{\pi^2 \mc}{x}\right)^{-\bar{a}^2-1} \sum_{n=0}^\infty \mA_n \left(\frac{x}{\pi^2 \mc}\right)^n \Gamma(\bar{a}^2+1+n) \, {}_1F_1\left(\bar{a}^2+1+n;\frac{1}{2}; x \right)\, .
\end{align}
We will fix the constants $\tilde{\mA}$ and $\mA_n$ by expanding in powers of $x$ and matching the coefficients to \eqref{Power_expansion_antiholo}:
\begin{align}\label{crossing_with_sums}
\tilde{\mA} \left( \frac{\pi^2 \mc}{x}\right)^{-\l} \sum_{n=0}^\infty \mA_n &\left(\frac{x}{\pi^2 \mc}\right)^n 
\sum_{k=0}^\infty\frac{ \Gamma(\l+n+k)\,\Gamma\left(\frac{1}{2}\right)}{\Gamma\left(\frac{1}{2}+k\right)}\frac{x^k}{k!}
\\\nn
&=\bar{C}(0) \left(\frac{\pi \mc}{x} \right)^{-\l} \left(\sum_{k=0}^\infty \frac{x^k}{k!} \right) \left(1-a_1 \sum_{n=0}^\infty \frac{(-1)^n}{n!} \left(\frac{a_2 x}{\mc} \right)^n \right).
\end{align}
Here we have already matched the lowest power of $x$, finding $\bar{a}^2+1=\l$ as in \eqref{baraResult}.
Matching the coefficient of the $x^{-\l}$ terms sets
\begin{equation}
	\tilde{\mathcal{A}} =\frac{\bar{C}(0) (1-a_1)}{\Gamma(\l)}\pi^{\l},
\end{equation}
where we have also chosen $\mA_0=1$.
Matching coefficients of the higher powers of $x$ in (\ref{crossing_with_sums}) gives a recursion relation for the coefficients $\mA_n$  for $n>0$:
\be
\mA_n=\frac{(\pi^2 \mc)^n \Gamma\left(\l\right)}{n! \,\Gamma \left(\l+n\right)}\left(\frac{1-a_1\left(1-\frac{a_2}{\mc}\right)^n}{1-a_1}\right)-\sum_{k=0}^{n-1} \frac{\left(\pi^2\mc\right)^{n-k}}{(n-k)!}\frac{\Gamma\left(\frac{1}{2}\right)}{\Gamma \left(\frac{1}{2}+n-k\right)}\mA_k\,.
\ee 
Solving this recursion relation for $\mathcal{A}_n$, we find the closed form solution
\be
\mA_n = \frac{\pi^{2n}\mc^n \Gamma(\l)}{1-a_1}\sum_{k=1}^n \left(\frac{1-a_1\left(1-\frac{a_2}{\mc}\right)^k}{k! \,\Gamma(\l +k)}-\frac{(1-a_1)\Gamma\left(\frac{1}{2}\right)}{k! \,\Gamma(\l) \,\Gamma\left(\frac{1}{2}+k\right)}\right)\sum_{\{\lambda_{n-k}\}}\prod_{i=1}^j \left(\frac{-\Gamma\left(\frac{1}{2}\right)}{\lambda_i ! \Gamma \left(\frac{1}{2}+\lambda_i\right)}\right)\,,
\ee
where the sum over $\left\{\lambda_{n-k}\right\}$ is taken over all sets of positive integers $\{\lambda_1, \lambda_2,\, \ldots ,\lambda_j\}$ such that $\sum_{i=1}^j \lambda_i=n-k$.  If $n-k=0$, this factor becomes 1.

Finally,  we can identify the anti-holomorphic OPE density as
\be
C(\bar{h}_s)= \tilde{\mA} \cosh\left(2\pi \sqrt{\bmh\mc}\right) \sum_{n=0}^{\infty} \mA_n \bmh^{\eta_0+n-1}\,.
\ee
Even though this result depends on the numerical fitting \eqref{Hansatz} with particular values $a_1=0.8,\,a_2=2.8$,  we have observed that for a wide enough range of values $\{h_0, c\}$,  the numerical  anti-holomorphic vacuum block can be approximated by the same function  \eqref{Hansatz}  at different values of the fitting constants $a_1$ and $a_2$.  The fact that we have found a closed form expression for the coefficients $\mA_n$, as functions of $a_1$ and $a_2$,  indicates a fruitful path towards numerically exploring the OPE density vacuum contribution in the modular lightcone limit.

\section{Spectral density OPE beyond vacuum contribution}\label{sec:BeyondVac}

So far we have considered OPE spectral densities that reproduce the vacuum block in the crossed channel, provided a gap separates this block from the rest of the spectrum.  We now go beyond the vacuum contribution in the t-channel, allowing a contribution from the Virasoro block associated to the primary operator closest to the vacuum.  We refer to this next-heaviest block as the lightest operator block, and we now compute the correction to the spectral density arising (in the s-channel) when we consider this lightest operator block in the t-channel of the crossing equation.

We first use a saddle point approximation to solve for the leading order correction for the spectral densities at $h_s/\mathcal{C}\gg1$ due to the presence of the  lightest operator block.  We then propose a generalization beyond $h_s/\mathcal{C}\gg1$ by introducing a kernel ansatz for the densities.
\subsection{Saddle point}\label{NonVacuumSaddleSection}
We start the discussion by writing the crossing equation explicitly in this case.  In Section \ref{smallbeta}, the $\b\to0$ limit, at leading order in small $\b$ and large exchange dimension  \eqref{smallb_largeh}, gave only the vacuum contribution in the t-channel.  If we want to include the next order corrections to the crossing equation \eqref{crossing_total} in small $\beta$,  a similar argument leads us to%
\footnote{In this section we are going to consider the holomorphic sector only,   mainly due to the fact that for the anti-holomorphic correction we don't have much information and  would be forced to rely on numerics alone. }%
\beq\label{hol_crossing_correction}
&&\int_{0}^{\infty}d h_s  \,\left(C^{0}(h_s)+\delta C^{0}(h_s)\right)\, \left(16\right)^{h_s}\,
e^{-\b(h_s-\frac{c-1}{24})}\nn\\
&&= (\b/\pi)^{\frac{c-1}{4}-8h_{0}}e^{{\pi^2}(\frac{c-1}{24})/\b}\,C(0) 
+\,(\b/\pi)^{\frac{c-1}{4}-8h_{0}}16^{h_{min}}C(h_{min})	e^{-{\pi^2\over\b}(h_{min}-\frac{c-1}{24})}.
\eeq
The new term $\delta C^0 (h_s)$ in the s-channel represents the corrections to the spectral OPE density produced from the lightest operator in the t-channel,  whose block contribution is given by
the second term in the second line.%
\footnote{ In the second line of \eqref{hol_crossing_correction} we have neglected higher order terms in small $\b$ from expanding $H\sim 1+\mathcal{O}(q)$, both when multiplied by the vacuum and the lightest operator.   It is possible that the $\mathcal{O}(q)$ terms in $C(0)\, H$ dominate over the leading term $C(h_{min})$ for the lightest operator.  However, here and in previous sections we have considered $h_{min}<{\cal C}$  with ${\cal C}$ to be of order one or less, while the $\mathcal{O}(q)$ term in $H$ has a more negative exponent set by a number greater than one, e.g. $a_2$ in \eqref{Hansatz}.  Thus, the second term in the second line of \eqref{hol_crossing_correction} is the dominant correction.  } %
The vacuum is of course already solved by \eqref{kernel_vac},  or in other words,  the first integration term at the first line of \eqref{hol_crossing_correction} equals the first term in the second line and therefore the correction terms satisfy
\be\label{crossing_beyond_vac}
\int_{0}^{\infty}d h_s  \,\delta C(h_s) \left(16\right)^{h_s}\,
e^{-\b(h_s-\frac{c-1}{24})}=(\b/\pi)^{\frac{c-1}{4}-8h_{0}}16^{h_{min}}C(h_{min})	e^{-{\pi^2\over\b}(h_{min}-\frac{c-1}{24})}\,.
\ee
By using again an inverse Laplace transform to solve for $\delta C$,   we find the saddle point of the resulting integral over $\b$ to be at
\be 
\b_s=\frac{4
  h_0-3 \mathcal{C}}{(h_s-\mathcal{C})}+\frac{\sqrt{(8 h_0-6 \mathcal{C})^2-4\pi ^2 (h_s-\mathcal{C}) \left(  h_{min}-\mathcal{C}\right)}}{2 (h_s-\mathcal{C})}\sim\frac{4h_0-3 \mathcal{C}}{(h_s-\mathcal{C})}+\pi\sqrt{  ( \mathcal{C}-h_{min})\over (h_s-\mathcal{C})}\, .
\ee
Notice that here we need $h_{min}<\mathcal{C}$ for the saddle to be real,  which is nevertheless automatically satisfied by the minimal gap.

The saddle point computation  leads to,
\be 
\d C(h_s) \sim {C(h_{min})16^{h_{min}-h_s }\over  \sqrt{\pi(\mathcal{C}-h_{min})}}\b_s^{6\mathcal{C}-8 h_0+3/2}\,e^{\b_s(h_s-\mathcal{C})-{\pi^2\over \b_s}(h_{min}-\mathcal{C})}
\ee
Taking the dominant second term in the saddle point, we can write this as
\be\label{saddle_beyond_vac}
\d C(h_s) \sim {C(h_{min})16^{h_{min}}\over  \sqrt{\pi(\mathcal{C}-h_{min})}}\left(\sqrt{  ( \mathcal{C}-h_{min})\over (\mathcal{C}-h_s)}\right)^{6\mathcal{C}-8 h_0+3/2}\,e^{2\pi\sqrt{  ( \mathcal{C}-h_{min})(h_s-\mathcal{C})}}\,.
\ee
As in the vacuum case,  the saddle point result is reliable as long as $h_s/\mathcal{C}\gg1$.  We would like to  generalize the spectral density OPE correction using the same technique as for the vacuum block contribution by proposing a similar spectral density ansatz.
\subsection{Kernel  ansatz}
Based on the saddle point analysis of Section \ref{NonVacuumSaddleSection} and the results from \cite{Collier:2019weq,  Collier:2018exn},  we propose the following ansatz as a  generalization for the correction from the lightest operator to the spectral density:
\be 
\d C(h_s) \sim C(h_{min})16^{h_{min}-h_s }\left((h_{min}- \mathcal{C})^{3\mathcal{C}-4 h_0+1/4}\over (\mathcal{C}-h_s)^{3\mathcal{C}-4 h_0+3/4}\right)\,
\cosh\left({2\pi\sqrt{ h_s-\mathcal{C}}}\right)^{\sqrt{  \mathcal{C}-h_{min}}}\, .
\ee
This ansatz reproduces the saddle point result  \eqref{saddle_beyond_vac} in the limit ${h_s\over \mathcal{C}}\gg1$.  However,  a stronger  check would be to prove that the integration over $h_s$ reproduces the lightest $h_{min}$ block in the t-channel.   Specifically, we want to perform the integral
\be\label{kernel_int_bv}
{C(h_{min})16^{h_{min} }\over (h_{min}- \mathcal{C})^{-3\mathcal{C}+4 h_0-1/4}}\int_{0}^{\infty} \,dh_s\,\left(1\over (\mathcal{C}-h_s)^{3\mathcal{C}-4 h_0+3/4}\right)\,
\cosh\left({2\pi\sqrt{ h_s-\mathcal{C}}}\right)^{\sqrt{  \mathcal{C}-h_{min}}}e^{-\b (h_s-\mathcal{C})}\,.
\ee
\begin{figure}
	\centering
		\begin{overpic}[scale=0.8]{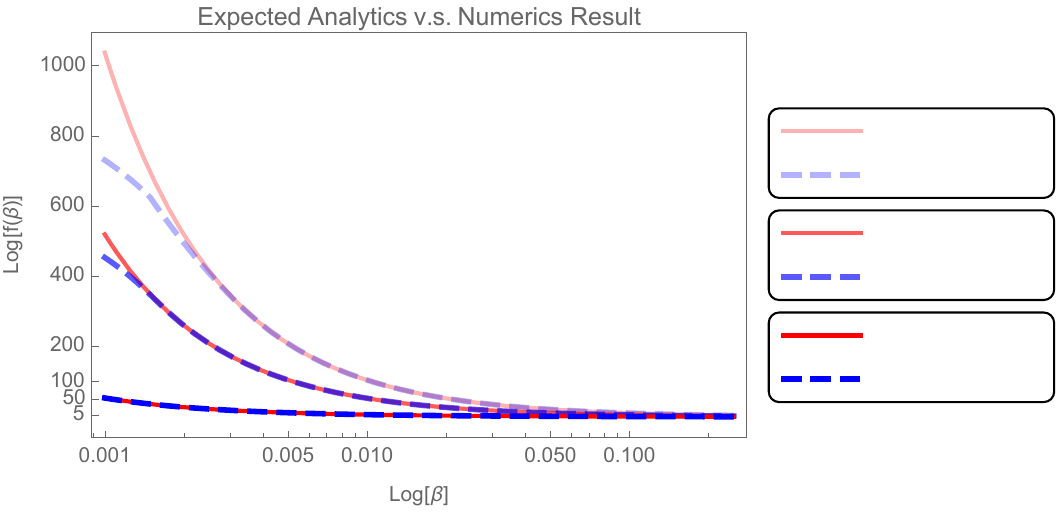}
			\put (82,33.45) {$h_{min} = 0.8 \mathcal{C}$}
			\put (82,23.75) {$h_{min} = 0.9 \mathcal{C}$}
			\put (82,14) {$h_{min} = 0.99 \mathcal{C}$}
		\end{overpic}
		\caption{Comparison of analytic t-channel crossing with results of numerical integration over kernel ansatz for relevant values of $h_{min}$. The solid red lines display the functional form expected from the t-channel crossing, given by the righthand side of \eqref{crossing_beyond_vac}. The dashed blue lines exhibit successive results of numerically integrating over the kernel \eqref{kernel_int_bv}. }
		\label{fig:NumericsFittingWithHmin}
\end{figure}
While we did not find a clever way to perform this integration, we have numerically evaluated the integral for several values of the difference $h_{min} - \mathcal{C}$ in the range $(0,1)$.  In figure \ref{fig:NumericsFittingWithHmin} we display three such cases for central charge $c=30$ and external dimension $h_0 =1$,  observing convincing agreement between the numerical result and the analytic t-channel of the crossing equation \eqref{crossing_beyond_vac}.  As the value of $h_{min}$ approaches the value of $\mathcal{C}$, the analytic expectation becomes closer to the numerical integration. This behavior is expected, since the smaller  $\mathcal{C}-h_{min}$ becomes,  the less suppressed is the contribution from the lightest operator;  that is,  for smaller $\mathcal{C}-h_{min}$, the ansatz is a better approximation in the range of smaller  $\b'$s,  which is our initial limiting condition.  When the difference between $\mathcal{C}$ and $h_{min}$ is increased,  deviation between the plots occurs for larger values of $\beta$,  which is evident from the saddle point  \eqref{saddle_beyond_vac}.

Beyond the numerical check,  we can, one last time, resort to a saddle point analysis. This time the saddle point analysis is on the integration \eqref{kernel_int_bv} over $h_s$ in the asymptotic limit, so we find 

\beq
I(\b)&=&\int_{0}^{\infty} \,dh_s\,\left( \mathcal{C}-h_s\right)^{-3\mathcal{C}+4 h_0-3/4}\,
e^{{2\pi\sqrt{ (h_s-\mathcal{C})(\mathcal{C}-h_{min})}}}e^{-\b (h_s-\mathcal{C})}\nn\\
&\sim &{2\pi(\mathcal{C}-h_{min})^{1/2}\over \b^{3/2} } e^{\pi^2\left({\mathcal{C}-h_{min}\over \b}\right)} \left({\pi^2(h_{min}-\mathcal{C})\over\b^2}\right)^{-3\mathcal{C}+4 h_0-3/4}\nn\\
&=& e^{\pi^2\left({\mathcal{C}-h_{min}\over \b}\right)} \left({(h_{min}-\mathcal{C})^{-3\mathcal{C}+4 h_0-1/4}\over\b^{(-6\mathcal{C}+8h_0)}}\right)\,.
\eeq
From this result, we obtain the right hand side of  \eqref{crossing_beyond_vac} after multiplying by the overall factor in front of the integral \eqref{kernel_int_bv}.  

Just as in the definition \eqref{average0},  we now define an average correction by dividing by the generalized Cardy formula \eqref{partition_fusion},  namely,
\be 
\d \mathbb{C}(h_s)\equiv{\d C(h_s)\over \mathcal{S}_0}\,.
\ee
Taking the limit ${h_s\over \mathcal{C}}\gg1$, 
\be
\mathcal{S}_0=4\sqrt{2} \sinh\left(2\pi b \sqrt{(h_s-\mathcal{C})}\right) \sinh\left(2\pi b^{-1}\sqrt{(h_s-\mathcal{C})}\right) \to e^{\left(4\pi \sqrt{\mathcal{C}(h_s-\mathcal{C})}\right) },
\ee
so we find 
\be
\d \mathbb{C}(h_s) \sim {C(h_{min})16^{h_{min}-h_s }\over  \sqrt{\pi(\mathcal{C}-h_{min})}}\left(\sqrt{  ( \mathcal{C}-h_{min})\over (\mathcal{C}-h_s)}\right)^{6\mathcal{C}-8 h_0+3/2}\ e^{-4\pi\sqrt{\mathcal{C}(h_s-\mathcal{C})}\left(1-{1\over 2}\sqrt{1-{h_{min}\over \mathcal{C}}}\right) }.
\ee
This result supports the claim that corrections to the spectral densities from including  non-vacuum contributions  are suppressed with respect to the leading contribution from the vacuum \eqref{spectral_saddle_leading}.  In fact, these corrections are exponentially suppressed.  This result corresponds to the analogous result for the one point function  in \cite{Kraus:2016nwo}. 
\section{Discussion and conclusions}\label{sec:Discussion}

In this paper we studied OPE spectral densities for the four-point correlation function of scalars in the large exchange dimension limit.  Our technique was to solve the modular bootstrap in an appropriate limit (the so-called modular lightcone limit),  allowing us to decouple the Virasoro vacuum  from the rest of the conformal dimensions spectrum.  We further restricted ourselves to theories that do not have conserved currents, and insisted on a twist gap in the spectrum,  allowing us to identify 
the contributions to the spectral densities of OPE coefficients at large spin from the Virasoro vacuum.

First we solved the crossing equations by resorting to a saddle point computation that allowed us to capture the OPE spectral density in the limit where the  dimension of the operator being exchanged is large in units of ${\cal C}$,  i.e.  ${h\over \mathcal{C}}\gg1$.  We then use this result as a leverage point to generalize the spectral densities beyond that region and towards ${h\over \mathcal{C}}\sim1$ by proposing an ansatz  that solves the crossing equations in the modular lightcone limit.   Just as the Cardy formula measures the density of primary operators in the large dimension limit,  our results can be understood as an extension of the Cardy formula to the density of OPE primary coefficients.

An obstruction to the development of a full-fledged large spin  perturbation theory comes from the lack of a practical closed form for the Virasoro blocks.  However,  we have shown in this paper that despite the limited knowledge we have of the blocks,  it is sufficient to allow for a leading order analysis.  We were able to perform an analytical analysis in the holomorphic sector of the OPE expansion,  but needed some numerical aid for the counterpart in the anti-holomorphic sector of the lightcone limit.  The numerical data needed for the Virasoro blocks was been obtained by solving the Zamolodchikov recursion relations numerically.

Recently,  some results for the Virasoro blocks at large exchange dimension have surfaced \cite{Cardona:2020cfy,  Das:2020fhs, Das:2020uax, Kashani-Poor:2013oza},  which offer hope of going beyond leading order in the large spin analysis, at least numerically.  We have in fact already used some of these results in the main body of the text.

Importantly,  we did not make any further assumptions on the theory under consideration,  beyond the existence of a gap separating the Virasoro vacuum from the rest of the spectrum and the absence of global conserved currents.  Henceforth,  we can think of our results as universal up to those assumptions.  Recently it was shown that this universality for the OPE coefficients of heavy operators is nicely captured by the DOZZ OPE of Liouville theory \cite{Collier:2019weq}.  In particular,  our result \eqref{average_asym} can be written in terms of the DOZZ coefficient.

Although we do not directly explore a gravitational interpretation of our work,  we expect it to provide a reliable semi-classical description of $AdS_3$ gravity even at finite values of $c$ \cite{Kraus:2016nwo,  Ghosh:2019rcj}.  Along the same lines,  it would be interesting to derive some of the results of this paper from a purely gravitational analysis,  in particular through the computation of the Witten diagrams corresponding to the four-point function of scalars considered here,  by using several of the methods developed recently  \cite{Costantino:2020vdu,Yuan:2017vgp,Carmi:2019ocp,Cardona:2017tsw,  Bertan:2018khc}.  Our results will be useful for studies of coarse-graining CFTs to produce gravitational duals, along the lines of \cite{Kraus:2016nwo,Das:2017cnv,Benjamin:2019stq,Miyaji:2021ktr, Benjamin:2021wzr}.  In those papers,  a coarse-grained average of the CFT result is compared to a black hole geometric result.  We specifically think it would be interesting to consider if a set of (non-minimal, as we consider) CFTs has a condition for successful coarse graining to a gravity theory,  where  the condition of $c=c_{crit}$ would be set by the asymptotic behavior of three-point spectral density instead of  the asymptotic density of states,  as proposed in \cite{Benjamin:2021wzr}.

\textit{We thank Alex Maloney and Eric Perlmutter for useful conversations.   This work was supported by the U.S. Department of Energy, Office of High Energy Physics, under Award No. DE-SC0019470 and C.C. by the National Science Foundation Award  No. PHY-2012195 at Arizona State University.}

\clearpage
\begin{singlespace}
\printbibliography[heading=subbibliography]
\end{singlespace}

\chapter{AN ENTROPIC LENS ON STABILIZER STATES}\label{Chapter3}

\textit{The contents of this chapter were originally published in Physical Review A \cite{Keeler:2022ajf}.}

\textit{The $n$-qubit stabilizer states are those left invariant by a $2^n$-element subset of the Pauli group. The Clifford group is the group of unitaries which take stabilizer states to stabilizer states; a physically--motivated generating set, the Hadamard, phase, and CNOT gates which comprise the Clifford gates, imposes a graph structure on the set of stabilizers. We explicitly construct these structures, the ``reachability graphs,'' at $n\le5$. When we consider only a subset of the Clifford gates, the reachability graphs separate into multiple, often complicated, connected components. Seeking an understanding of the entropic structure of the stabilizer states, which is ultimately built up by CNOT gate applications on two qubits, we are motivated to consider the restricted subgraphs built from the Hadamard and CNOT gates acting on only two of the $n$ qubits. We show how the two subgraphs already present at two qubits are embedded into more complicated subgraphs at three and four qubits. We argue that no additional types of subgraph appear beyond four qubits, but that the entropic structures within the subgraphs can grow progressively more complicated as the qubit number increases. Starting at four qubits, some of the stabilizer states have entropy vectors which are not allowed by holographic entropy inequalities. We comment on the nature of the transition between holographic and non-holographic states within the stabilizer reachability graphs.}

\section{Introduction}

Are all quantum states in a Hilbert space created equal? Viewed at the most abstract level, every pure quantum state can be rotated to any other by a unitary change of basis. But when the Hilbert space is endowed with some structure, different states play different roles with respect to that structure.  One natural structure is given by specifying a 
\emph{factorization} of the Hilbert space: an isomorphism between the abstract $\Hil$ and the tensor product Hilbert space $\bigotimes_{i=1}^N \Hil_i.$ For finite Hilbert spaces of composite dimension, the factorization associates to each pure state $\ket{\Psi}\in\Hil$ an \emph{entropy vector}, the collection of von Neumann entropies of the $2^N-1$ reduced density matrices formed by tracing out each possible tensor product formed from factors $\Hil_i$.

The entropy vectors provide a classification of the states in a tensor product Hilbert space, but because every state has an associated entropy vector they do not, by themselves, pick out any states as special. One way to accomplish this is by fixing one or more preferred operators acting on the Hilbert space. As a familiar example, choosing a particular Hermitian operator acting on a Hilbert space to be the Hamiltonian, the generator of time translations, specifies a basis of energy eigenstates, and every state in the Hilbert space can then be expanded in the energy basis. Most states are superpositions of more than one energy eigenstate, but not all: fixing a Hamiltonian picks out the basis of eigenstates of the Hamiltonian, the energy eigenstates themselves, which are the states where the Hamiltonian acts trivially, as scalar multiplication. The particular scalar is just given by the eigenvalue of the energy eigenstate, and if we just want to find the set of energy eigenstates this is unimportant: we could instead say that the energy eigenstates are the states where the projector onto the energy eigenspaces acts as the identity.

A similar procedure applies when instead of specifying a single Hermitian operator we pick a (multiplicative) group of Hermitian operators. Given the group $G\in L(\Hil)$, we can classify every state $\ket{\Psi}\in\Hil$ by the number of group elements that act trivially on this state: the dimension of the stabilizer subgroup $G_\ket{\Psi}$. Almost every state will have $\mathrm{dim}\:G_\ket{\Psi}=1$: the only group element that acts trivially is the identity operator. But some will have more, and the states that have the largest stabilizer subgroups relative to $G$ are called the \emph{stabilizer states} (which we will define more precisely below). For example, if the Hilbert space is that of a qubit, the stabilizer states of the Pauli group generated by $\langle X,Y,Z \rangle$ are the six states stabilized (up to sign) by the identity and one additional Pauli operator. Stabilizer states, especially with respect to the Pauli group on $n$ qubits, play an important role in the theory of quantum error correction, and in the fundamentals of quantum computing \cite{Gottesman:1997zz,Gottesman:1998hu,aaronson2004improved,knill2004fault,bravyi2005universal}. But here they emerged directly as ``generalized eigenstates,'' the states which play nicely with a specified group of operators. 

This paper initiates a research program aimed at combining these two classifications of states in Hilbert space. We seek an understanding of the stabilizer states, picked out by their interaction with a specified group of operators, in terms of their entropic structure, given by the underlying factorization of the Hilbert space. In particular, we will focus on the stabilizer states with respect to the Pauli group acting on $n$ qubits. In this setting, any stabilizer state can be reached by starting with any other stabilizer state and applying a unitary quantum circuit comprised of the Clifford gates: two one-qubit gates, the Hadamard and phase gates, and a single two-qubit gate, the controlled NOT gate. Since unitary operations on a single tensor factor do not change the entropy vector, it is already clear that moving between entropy vectors can only be accomplished via a CNOT gate. But not all such gate applications alter the entropic structure, and the picture we will find will be far richer than simply counting CNOT gates.

Besides understanding the general entropic structure of stabilizer states, we are additionally motivated by the connections between stabilizer states and holography. Stabilizer error-correcting codes satisfy a complementary recovery property which implies, and is equivalent to, an operator-algebraic version of the Ryu-Takayanagi formula \cite{Ryu:2006bv} relating entropies of states in the boundary/physical Hilbert space to the expectation value of an ``area'' operator acting on the bulk/logical Hilbert space \cite{Harlow:2016vwg,Pollack:2021yij}. A holographic bulk geometry can be discretized into a graph \cite{Bao:2015bfa}; in the limit of large bond dimension a tensor network built from random stabilizer tensors saturates the RT formula with probability one \cite{Hayden:2016cfa,Nezami:2016zni}.

Hence all holographic entropy vectors can be represented by stabilizer states, but the converse is not true: the holographic \emph{entropy cone} consisting of space of all allowed holographic entropy vectors is contained within the stabilizer entropy cone, and this containment is strict starting at three regions. Because the holographic entropy cone is well--characterized---explicitly at up to five regions \cite{HernandezCuenca:2019wgh}, and implicitly at arbitrary finite region number by the methods pioneered in \cite{Bao:2015bfa}---but the stabilizer entropy cone, and even the larger quantum entropy cone of all quantum states, are poorly understood, it is of great interest to understand in more detail how the holographic states are embedded into the larger space of stabilizer states. The stabilizer graph constructions presented in this paper allow this question to be attacked.

\subsection{Summary of Results}

\begin{figure}[!htb]
     \centering
     \begin{subfigure}{0.43\textwidth}
         \centering
         \includegraphics[width=\textwidth]{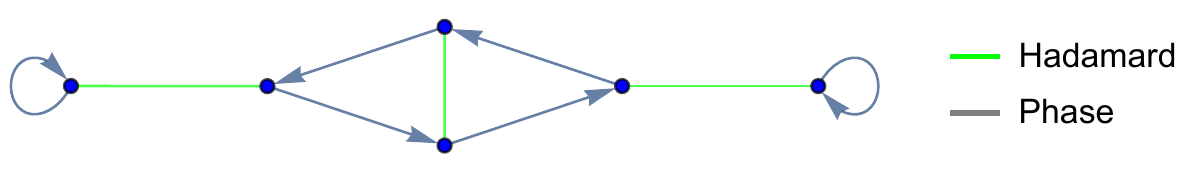}
         \caption{}
         \label{fig:OneQubitReachabilityGraph}
     \end{subfigure}
     \hfill
     \begin{subfigure}{0.47\textwidth}
         \centering
         \includegraphics[width=\textwidth]{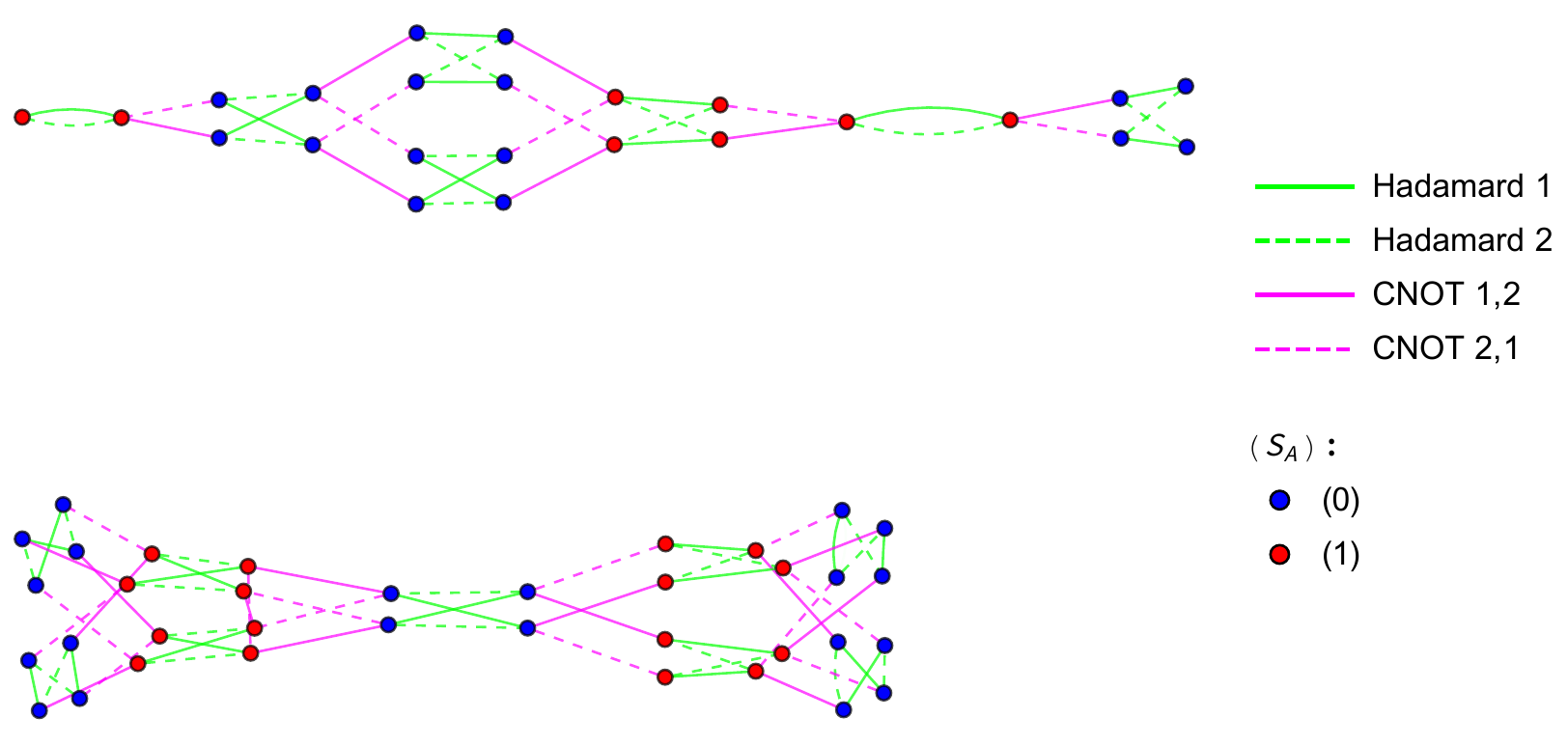}
         \caption{}
         \label{fig:TwoQubitG24G36}
     \end{subfigure}
     \hfill
     \begin{subfigure}{0.69\textwidth}
         \centering
         \includegraphics[width=\textwidth]{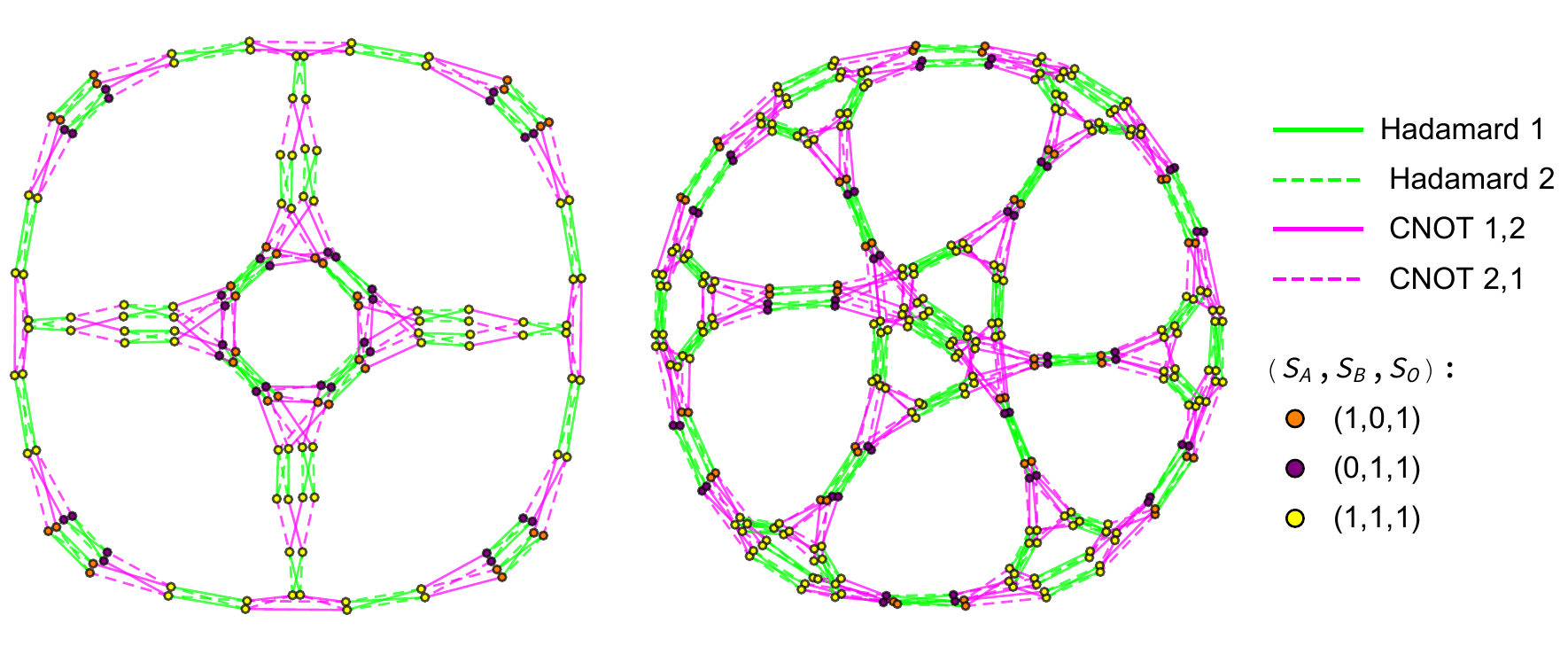}
         \caption{}
         \label{fig:ThreeQubitNewStructures}
     \end{subfigure}
     \hfill
     \begin{subfigure}{0.8\textwidth}
         \centering
         \includegraphics[width=\textwidth]{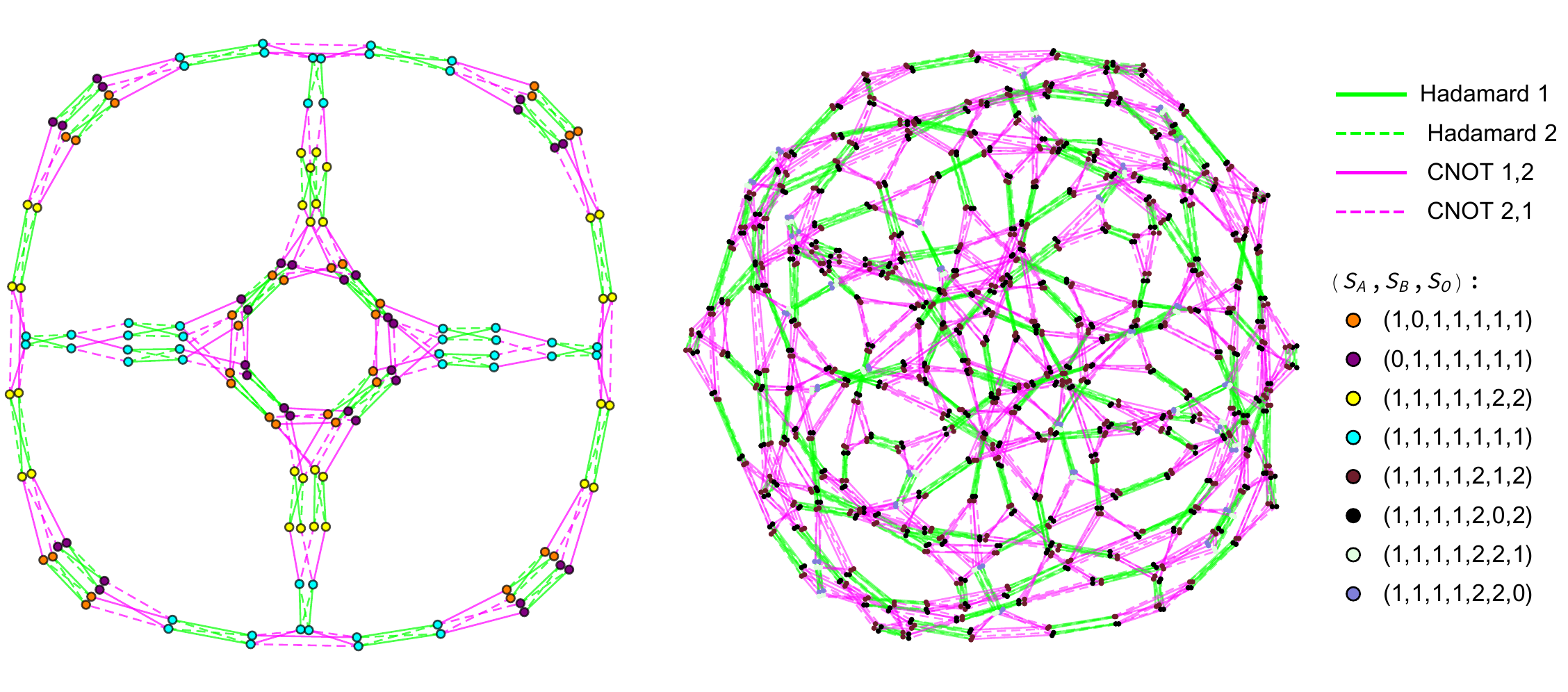}
         \caption{}
         \label{fig:4-Qubit}
     \end{subfigure}
        \caption{The collection of subgraph structures up through five qubits. (a) 1-Qubit reachability diagram displaying all single-qubit operations. (b) 2-Qubit restricted graph consisting of subgraph structures $g_{24}$ and $g_{36}$. (c) Two new subgraph structures occur at three qubits, $g_{144}$ and $g_{288}$. (d) At four qubits we witness the arrival of a new subgraph $g_{1152}$, as well as an increase of entropy vectors on $g_{144}$. No new subgraph structures emerge beyond $g_{1152}$, but the number of different entropy vectors present in each subgraphs continues to increase for higher qubit number. At four qubits we witness the first instance of non-holographic states, appearing on subgraph $g_{144}$.}
        \label{fig:AllStructures}
\end{figure}

Figure \ref{fig:AllStructures} presents a preview of our results. All of the objects in the figure are ``reachability graphs,'' which display the structure of a (subset of) the stabilizer states. Each vertex is a particular stabilizer state, colored by its entropy vector. The edges connecting them correspond to the action of a particular Clifford gate. The first panel presents the full reachability graph at one qubit, showing the action of the two one-qubit Clifford gates, Hadamard and phase, on the six one-qubit stabilizer states. The remaining panels show reachability graphs at higher qubits, constructed from a subset of the Clifford gates consisting of the Hadamard and CNOT gates applied to two particular qubits; we will argue below that this restricted set of gates is sufficient to capture the changes in entropic structure as we move between different stabilizer states. The second panel presents the two subgraphs found in the restricted two-qubit reachability graph: one with 24 vertices and one with 36. 

At three qubits, the restricted graph has 16 connected components, made up of copies of the 24- and 36-vertex subgraphs, as well as two additional structures, with 144 and 288 vertices, respectively, which are presented in the third panel. We will show how these more complicated structures, as well as an 1152-vertex subgraph that appears at four qubits, can be constructed in a simple way by certain ``lifts'' of states in the structures found at two qubits to stabilizer states in a larger Hilbert space. By understanding the lifts, we will also understand how to map entropy vectors at lower qubit numbers to entropy vectors at higher qubit numbers: for example, although only three distinct entropy vectors appear in the 144-vertex subgraphs at three qubits, at four qubits there are versions of this subgraph with four distinct entropy vectors, as seen in the final panel of the Figure.

At four qubits, the stabilizer states can have any of eighteen different entropy vectors (see Table \ref{tab:FourQubitEntropyVectors}), and one of these entropy vectors, shown in blue in the fourth panel of Figure \ref{fig:AllStructures}, is not holographic. The reachability graphs allow us to see how, at four and five qubits, applications of Clifford gates can move states out of, and back into, the holographic entropy cone. We see, for example, that moving from the inner octagonal structure to the outer ring cannot be accomplished without passing through non-holographic states.

\subsection{Structure of the Paper}

The remainder of this paper is organized as follows. Section \ref{sec:review} reviews the notions of stabilizer states and the entropy cone, which are of fundamental importance for the rest of the paper. Section \ref{sec:two} initiates the study of the intersection of these two concepts by presenting the two-qubit stabilizer graph colored by entropy vector. We illustrate that additional insight can be gained into the graph structure by considering the \emph{restricted graphs} generated by only a subset of the Clifford gates; in particular, we highlight the special role of the restricted graph generated by only Hadamard and CNOT gates.

Section \ref{sec:three} extends the discussion to the three-qubit stabilizer graph. We show that much of the three-qubit graph can be understood as the natural extension of the two-qubit graph, but that nontrivial new structures appear in addition. We argue that because the Clifford gates act on at most two qubits it is natural to consider the restricted graphs generated by gates that act on any two of the three qubits. In Section \ref{sec:general}, we pass to a discussion of the situation for higher qubit numbers, illuminated by some direct results at four and five qubits. We show that no fundamentally new objects appear as we go to higher qubit number, but the existing objects develop progressively more complicated entropic structures. At higher qubit numbers, entropy vectors appear which cannot be represented by holographic states, and we comment on how they fit into the stabilizer graph. Finally, in Section \ref{sec:discussion} we discuss and conclude. Additional results and graphs are presented in appendices.

\section{Reminder: Stabilizer States and the Entropy Cone}\label{sec:review}

\subsection{Review of Stabilizer States}\label{sec:StabilizerReview}

The Pauli matrices
\begin{equation}
    I=\begin{pmatrix}1&0\\0&1\end{pmatrix}, \,\, \sigma_X=\begin{pmatrix}0&1\\1&0\end{pmatrix}, \,\,
    \sigma_Y=\begin{pmatrix}0&-i\\i&0\end{pmatrix}, \,\,
    \sigma_Z=\begin{pmatrix}1&0\\0&-1\end{pmatrix},
\end{equation}
are a set of four Hermitian and unitary matrices with eigenvalues $\pm 1$, which, given a fixed basis \{$\ket{0},\ket{1}\}$, can be interpreted as operators acting on the Hilbert space $\mathbb{C}^2$ of a single qubit. The (nontrivial) Pauli operators $\{\sigma_X,\sigma_Y,\sigma_Z\}$ generate the full algebra of linear operators $L(\mathbb{C}^2)$. More importantly for our purposes, they generate a 16-element multiplicative matrix group, the Pauli group on one qubit 
\begin{equation}
    \Pi_1\equiv\langle\sigma_X,\sigma_Y,\sigma_Z\rangle
    = c \{I, \sigma_X, \sigma_Y, \sigma_Z\}, 
    c\in\{\pm 1,\pm i\}.
\end{equation} 

Recall that for a pure state $\ket{\Psi}\in \Hil$ and a group of operators $G\subset L(\Hil)$, the stabilizer group of $\ket{\Psi}$ is defined by
\begin{equation}
    G_\ket{\Psi} \equiv \{g\in G\: | \:g\ket{\Psi}=\ket{\Psi}\}.
\end{equation}
That is, the stabilizer group of $\ket{\Psi}$ consists of the operators in $G$ for which $\ket{\Psi}$ is an eigenvector with eigenvalue one.  By inspection, the Pauli group on one qubit $\Pi_1$ has one element, $I$, with two unit eigenvalues, which therefore stabilizes all states in $\mathbb{C}^2$; one element with two negative eigenvalues, $-I$, which stabilizes no states; eight elements with purely imaginary eigenvalues, which also stabilize no states; and six elements with one unit eigenvalue, $\{\pm \sigma_X,\pm \sigma_Y, \pm \sigma_Z\}$, which therefore each stabilize one state in $\mathbb{C}^2$. Hence there are six states in the Hilbert space, the \emph{one-qubit stabilizer states}, which are stabilized by a two-element subgroup:
\begin{equation}\label{OneQubitStabilizerStates}
    S_1 \equiv \left\{ \ket{0}, \ket{1}, \ket{\pm}\equiv\frac{1}{\sqrt{2}}(\ket{0}\pm \ket{1}), \ket{\pm i}\equiv\frac{1}{\sqrt{2}}(\ket{0}\pm i\ket{1})  \right\},
\end{equation}
and all other states in the Hilbert space are stabilized by only the identity. For example, $\ket{+}$ is stabilized by $I$ and $\sigma_X$, while $\ket{1}$ is stabilized by $I$ and $-\sigma_Z$.

If we fix a state $\ket{\Psi}\in S_1$, then there are a limited number of operations we can do that will map $\ket{\Psi}$ to some other $\ket{\Psi^\prime}\in S_1$. In particular, we can ask what unitary operators $U\in L(\mathbb{C}^2)$ are guaranteed to map \emph{any} state in $S_1$ back into $S_1$. Since we defined the stabilizer states by their property of having the largest stabilizer group with respect to $\Pi_1$, we can equivalently ask which unitaries $U$ \emph{normalize} the Pauli group, i.e.\ take group elements to group elements under conjugation by U. Clearly every element of the Pauli group itself has this property, but in general there is a larger group of unitaries which do as well. The group of unitaries which normalize the Pauli group is called the Clifford group,
\begin{equation}
    C_1=\left\{ U\in L(\mathbb{C}^2) \: | \: 
    UgU^\dagger\: \forall g \in \Pi_1\right\}.
\end{equation}
It suffices to check that a unitary takes $\sigma_X$ and $\sigma_Z$ to elements of $C_1$. For example, the Hadamard \cite{sylvester1867lx,hadamard1893resolution} and phase gates,
\begin{equation}
    H\equiv\frac{1}{\sqrt{2}}\begin{pmatrix}1&1\\1&-1\end{pmatrix}, \qquad P\equiv\begin{pmatrix}1&0\\0&e^{\frac{i \pi}{2}}\end{pmatrix}=\begin{pmatrix}1&0\\0&i\end{pmatrix},
\end{equation}
are both elements of the Clifford group, since $H\sigma_XH^\dagger=\sigma_Z$, $H\sigma_ZH^\dagger=\sigma_X$, and $P\sigma_XP^\dagger=\sigma_Y$, $P\sigma_ZP^\dagger=\sigma_Z$. In fact, these two gates suffice to generate the Clifford group, $C_1=\langle H, P \rangle$. We see, in particular, that because $PP=\sigma_Z$, we can easily construct the Paulis themselves out of $H$ and $P$.

	\begin{figure}[t]
		\begin{center}
		\begin{overpic}[width=0.9\textwidth]{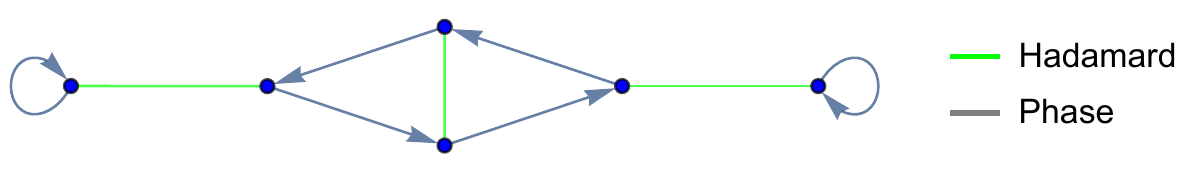}
		\put (5,2.5) {$\ket{0}$}
		\put (66,2.5) {$\ket{1}$}
		\put (20.5,3.5) {$\ket{+}$}
		\put (50.5,9.5) {$\ket{-}$}
		\put (38.8,0) {$\ket{i}$}
		\put (32.8,14.2) {$\ket{-i}$}
        \end{overpic}
		\caption{Complete one-qubit reachability diagram. Because the Hadamard gate is its own inverse we have depicted edges corresponding to Hadamard gate applications as undirected, but gates corresponding to phase gates as directed.}
		\label{OneQubitReachabilityDiagram}
	\end{center}
	\end{figure}

Given the set of stabilizer states $S_1$ and a set of generators of the Clifford group $C_1$, which we call \emph{Clifford gates} \cite{Gottesman:1997zz}, we can arrange the states into a \emph{reachability graph}, with each vertex labeled by a stabilizer state and each edge between vertices labeled by the Clifford gate which maps one vertex to the other. (The Hadamard gate is its own inverse, so it can be represented by an undirected edge, but the phase gate is not, so it is represented by a directed edge.) The reachability graph for $S_1$ is shown in Figure \ref{OneQubitReachabilityDiagram}. Note that the reachability graph depends on a particular choice of generators of $C_1$, e.g. the Clifford gates. The significance of choosing $H$ and $P$, in particular, as our generators is that both operators have physical significance and can be implemented experimentally (relatively) easily.

Note that the reachability graph consists of a single component: as implied by the definition of the Clifford group, we can use the Clifford gates to get (in some number of steps) from any initial stabilizer state, e.g. $\ket{0}$, to any other stabilizer state. Furthermore, the graph contains many cycles: trivial ones, where a Clifford gate acts as the identity on a particular stabilizer state, but also longer ones. For example, because $H^2=P^4=I$, every edge corresponding to a Hadamard application is part of a cycle of length two, and every edge corresponding to a nontrivial phase application is part of a cycle of length four, such as the diamond representing a nontrivial cycle consisting of four phase gates we see at the center of Figure \ref{OneQubitReachabilityDiagram}. More interestingly, we have cycles consisting of more than one type of gate, such as the triangles with two phase gates and one Hadamard.

To better understand these cycles, we can apply two basic group-theoretic results. First, Lagrange's theorem says that the order of any subgroup $H$ of a finite group $G$ gives an integer partition of that group \cite{Alperin1995}: explicitly,
	\begin{equation}\label{LagrangeTheorem}
		|G| = [G:H]\cdot |H|, \hspace{.5cm} \forall H \leq G,
	\end{equation}
with $[G:H]$ the index of $H$ in $G$. Second, the orbit-stabilizer theorem says that, when $H$ is a stabilizer subgroup of $G$ with respect to some state $x\in X$, i.e.\ $H=G_x$, there exists a bipartition between the orbit of x in G, $|G \cdot x|$, and the set of cosets of the stabilizer subgroup in $G$, $G/H$. Hence these two objects have the same dimension:
\begin{equation}
    |G \cdot x|=\frac{|G|}{|H|} = [G:H]\label{eq:orbit_stabilizer},
\end{equation}
where in the last equality we have used \eqref{LagrangeTheorem}.

The Clifford group $C_1$ is a group of operators acting on the one-qubit Hilbert space $\mathbb{C}^2$. As we have discussed, for each $\ket{\Psi} \in S_n \subset \mathbb{C}^2$, there exists a subgroup $G_\ket{\Psi}$. Hence we can apply this group-theoretic machinery to our case of interest. Substituting $\ket{\Psi}\in S_1$ for $x\in X$, $C_1$ for $G$, and $G_{\ket{\Psi}}$ for $H$ in \eqref{eq:orbit_stabilizer} gives
	\begin{equation}\label{Burnside+Lagrange}
		|C_1 \cdot \ket{\Psi}| = \frac{|C_1|}{|G_{\ket{\Psi}}|},
	\end{equation}
where $|C_1 \cdot \ket{\Psi}|$ denotes the length of the orbit of $\ket{\Psi}$. When we represent the action a of group element by a graph, the orbit length is the largest number of vertices in any connected component of the graph.

Explicitly, consider the set of one-qubit stabilizer states $S_1$ defined in \ref{OneQubitStabilizerStates}. The one-qubit Clifford group $C_1$ is constructed from the generating set $\langle H_1,P_1 \rangle$. The orbit of each $\ket{\Psi} \in S_1$ can be computed directly using Equation \eqref{Burnside+Lagrange}, with results shown in Table \ref{tab:OneQubitFullOrbits}. One can easily verify these results by comparing with the reachability diagram in Figure \ref{OneQubitReachabilityDiagram}.
\begin{table}
\begin{center}
\begin{tabular}{|c||c|c|c|}
	\hline
	$\ket{\Psi}$ & $|C_1|$ & $|C_{1_{\ket{\Psi}}}|$ & $|C_1 \cdot \ket{\Psi}|$ \\ 
	\hline
	\hline
	$\ket{0}$ & 192 & 32 &6\\
	\hline
	$\ket{1}$ & 192 & 32 &6\\
	\hline
	$\ket{+}$ & 192 & 32 &6\\
	\hline
	$\ket{-}$ & 192 & 32 &6\\
	\hline
	$\ket{i}$ & 192 & 32 &6\\
	\hline
	$\ket{-i}$ & 192 & 32 &6\\
	\hline
\end{tabular}
    \caption{Orbit lengths for each single-qubit stabilizer state under the one-qubit Clifford group $C_1$.}
    \label{tab:OneQubitFullOrbits}
    \end{center}
\end{table}

We can also use this machinery to consider the orbits of states $\ket{\Psi} \in S_1$ under subgroups of $C_1$, see Table \ref{tab:OneQubitSubgroupOrbits}. Let $G_P < C_1$ denote the subgroup generated by only the phase gate. This group contains the $4$ unique elements $P_1, P_1^2, P_1^3,$ and $P_1^4$ (recall $P_1^4 = I$). Each element acts trivially on $\ket{0}$ and $\ket{1}$, and thus these two states are stabilized by all elements of $G_P$. The remaining states $(\ket{+},\ket{-},\ket{i},\ket{-i})$ form a cycle of length $4$ under operation of $P_1$, each stabilized only by the identity $P_1^4$. These orbits manifest as subgraphs of the reachability graph as seen in Figure \ref{OneQubitH1OnlyP1Only}.

\begin{table}
\begin{center}
\begin{tabular}{|c||c|c|c|}
	\hline
	$\ket{\Psi}$ & $|G_P|$ & $|G_{P_{\ket{\Psi}}}|$ & $|G_P \cdot \ket{\Psi}|$ \\ 
	\hline
	\hline
	$\ket{0}$ & 4 & 4 & 1\\
	\hline
	$\ket{1}$ & 4 & 4 & 1\\
	\hline
	$\ket{+}$ & 4 & 1 & 4\\
	\hline
	$\ket{-}$ & 4 & 1 & 4\\
	\hline
	$\ket{i}$ & 4 & 1 & 4\\
	\hline
	$\ket{-i}$ & 4 & 1 & 4\\
	\hline
\end{tabular}
\hspace{.2cm}
\begin{tabular}{|c||c|c|c|}
	\hline
	$\ket{\Psi}$ & $|G_H|$ & $|G_{H_{\ket{\Psi}}}|$ & $|G_H \cdot \ket{\Psi}|$ \\ 
	\hline
	\hline
	$\ket{0}$ & 2 & 1 & 2\\
	\hline
	$\ket{1}$ & 2 & 1 & 2\\
	\hline
	$\ket{+}$ & 2 & 1 & 2\\
	\hline
	$\ket{-}$ & 2 & 1 & 2\\
	\hline
	$\ket{i}$ & 2 & 1 & 2\\
	\hline
	$\ket{-i}$ & 2 & 1 & 2\\
	\hline
\end{tabular}
    \caption{Orbit lengths for each single-qubit stabilizer state under subgroups $G_P<C_1$ and $G_H<C_1$, generated by only the phase gate and Hadamard gate respectively.}
    \label{tab:OneQubitSubgroupOrbits}
    \end{center}
\end{table}

	\begin{figure}
		\begin{center}
		\begin{overpic}[width=0.9\textwidth]{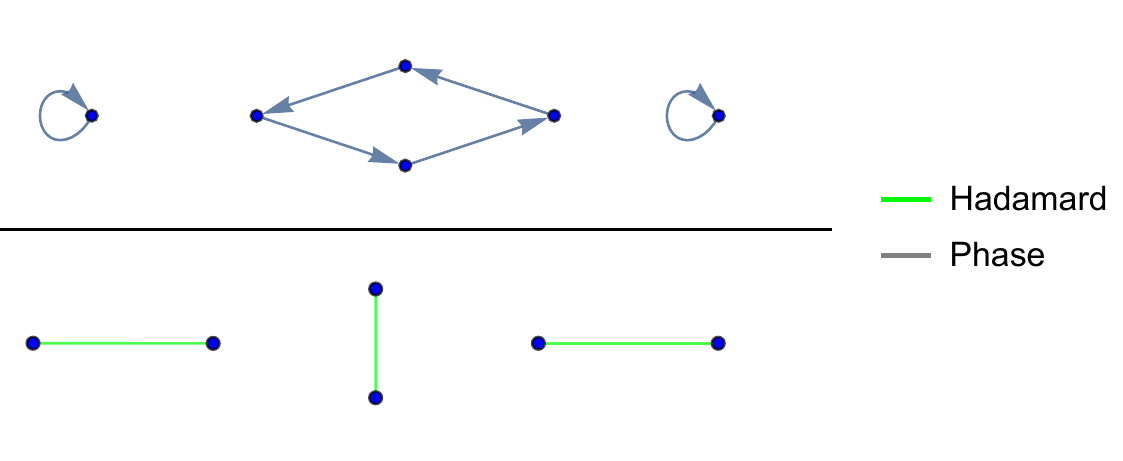}
		\put (8,26.5) {$\ket{0}$}
		\put (63.5,26.5) {$\ket{1}$}
		\put (32.5,36) {$\ket{i}$}
		\put (1.7,6) {$\ket{0}$}
		\put (62.8,6.2) {$\ket{1}$}
		\put (29,15) {$\ket{i}$}
        \end{overpic}
	\caption{The one-qubit reachability diagram restricted to only phase (Top) and only Hadamard (Bottom) operations reveals disconnected orbits of varying length for states $\ket{\Psi} \in S_1$.}
	\label{OneQubitH1OnlyP1Only}
	\end{center}
	\end{figure}
Similarly, consider the subgroup $G_H<C_1$, generated by only the Hadamard gate. This group only has $2$ elements since $H_1 = H_1^{-1}$ and $H_1^2 = I$. Distinctly, the Hadamard gate stabilizes no element of $S_n$, and thus each state in $S_1$ is stabilized by only the identity $(H_1^2)$. Each $\ket{\Psi}\in S_1$ subsequently has an orbit of length $2$, as shown in Table \ref{tab:OneQubitSubgroupOrbits}, which results in the decomposition of the reachability graph into pairs of states (Figure \ref{OneQubitH1OnlyP1Only}), equivalent up to a change of basis. 

In the remainder of the paper, we will often consider splitting the reachability graph into subgraphs constructed from a subset of the Clifford gates. Ultimately, the underlying structure behind this graphical representation is precisely the partitioning of the stabilizer set by orbit length.

As the alert reader will have realized from our notation, we can extend the definitions of the Pauli group, stabilizer states, and Clifford gates to more than one qubit. The Pauli group $\Pi_n$ on $n$ qubits consists of ``Pauli strings'' acting on each of the qubits, and is generated by\footnote{Note that writing a Pauli string requires not just a factorization of the $2^n$ dimensional Hilbert space into tensor factors representing qubits (each of which has a specified basis $\{0,1\}$), but a particular \emph{ordering} of the qubits from $1$ to $n$: we write $\ket{a_1\ldots a_n}\equiv \ket{a_1}_1\otimes\ldots\ket{a_n}_n$. It should be clear that the set of length-one Pauli strings as a whole are the independent of choice of ordering, and hence so is the Pauli group $\Pi_n$ they generate. This will also be the case for the Clifford group $C_n$ and the set of all Clifford gates. However, individual gates will of course depend on the choice of ordering. We will often consider a subset of the Clifford gates which act only on the first two qubits, which again depends on the choice of ordering, but there is an equivalent subset which acts on any two specified qubits, so although the position of a given state within the graphs we will generate depends on a choice of ordering, the overall graph structures themselves will not.} the length-one Pauli strings like $I^1\otimes\ldots\otimes I^{k-1} \otimes \sigma_Z^k \otimes I^{k+1} \otimes \ldots \otimes I^n$. The stabilizer states $S_n$ are those states $\ket{\Psi}\in (\mathbb{C}^2)^{\otimes n}$ with stabilizer groups of maximal size, $\mathrm{dim}\: G_\ket{\Psi} = 2^n$.\cite{aaronson2004improved,garcia2017geometry} 
\begin{equation}
    \left|S_n\right| = 2^n \prod_{k=0}^{n-1} (2^{n-k}+1)
    = 2 (2^n +1) \left|S_{n-1}\right| \approx 2^{(.5+o(1))n^2}.
\end{equation}
We have given a closed-form expression, recursion relation (with base case $|S_1|=6$, as constructed explicitly above), and asymptotic expression.

The Clifford group $C_n$ is again the group of unitaries which normalize $\Pi_n$, which contains the Hadamard and phase gates acting on each individual qubit. However, these gates no longer suffice to generate the full Clifford group: because the gates act only on a single qubit, they cannot change the entanglement structure of a state. Yet not every stabilizer state has the same entanglement structure. One subset of the two-qubit stabilizer states consists of the tensor product of a one-qubit stabilizer state on the first qubit and another one-qubit stabilizer state on the second qubit: these states are, of course, product states. But, for example, the Bell state $\frac{1}{\sqrt{2}}(\ket{00}+\ket{11})$ is a two-qubit stabilizer state, stabilized by $\{I^1I^2,\sigma_X^1\sigma_X^2,-\sigma_Y^1\sigma_Y^2,\sigma_Z^1\sigma_Z^2\}$. Hence any set of operators generating $C_n$ must contain operators which map product states to entangled states and vice versa.

One convenient gate which accomplishes this task is the $CNOT_{i,j}$ gate, which performs a controlled $NOT$ operation on the $j$th qubit depending on the state of the $i$th qubit:
\begin{equation}
		CNOT_{1,2} \equiv \begin{pmatrix}
            1 & 0 & 0 & 0\\
            0 & 1 & 0 & 0\\
	    0 & 0 & 0 & 1\\
	    0 & 0 & 1 & 0
            \end{pmatrix} \in L(\mathbb{C}^4).
\end{equation}
To check that $CNOT_{i,j}\in C_n$, it suffices to check its action on the length-one Pauli strings consisting of $\{\sigma_X^i,\sigma_Z^i,\sigma_X^j,\sigma_Z^j\}$ tensored with identities on every other qubit. A standard calculation shows that conjugation by $CNOT_{i,j}$ maps these four strings to $\{\sigma_X^i\otimes\sigma_X^j,\sigma_Z^i,\sigma_X^j,\sigma_Z^i\otimes\sigma_Z^j\}$ tensored with identities, respectively. Hence the CNOT gates are indeed members of the Clifford group; together with the Hadamards and phase gates, they generate\footnote{In fact, we only need half of the CNOT gates, the $n(n-1)/2$ gates $CNOT_{i,j}$ with $i<j$, because we have the relation $CNOT_{j,i}=H_i H_j CNOT_{i,j} H_j H_i$. Note that this expression is not unique, because $H_i$ and $H_j$ acting on distinct qubits commute. Following convention, we will nevertheless take the Clifford gates to include all $n(n-1)$ CNOT gates; e.g.\ our two-qubit reachability graph will show the actions of both $CNOT_{1,2}$ and $CNOT_{2,1}$.\label{fn:cnot_dependence}} $C_n$, so the \emph{$n$-qubit Clifford gates} are taken to be the $n$ Hadamard gates $H_n$, the $n$ phase gates $P_n$, and the $n(n-1)$ gates $CNOT_{i,j}$ (for $i\ne j$).

We can thus construct, for the Hilbert space of $n$ qubits $\mathbb{C}^{2n}$, a reachability graph for the stabilizer states $S_n$ using the Clifford gates which generate $C_n$. We will devote the rest of the paper to studying this object, and the subgraphs formed from it by restricting to a subset of the Clifford gates, at various qubit numbers $n>1$. Before we begin this study, however, we will first categorize the various possible allowed entropic structures of the stabilizer states $S_n$.
    
\subsection{Review of the Entropy Cone}\label{sec:EntropyCone}

All stabilizer states themselves are pure states; that is, for a given stabilizer state $\ket{\psi}$, the density matrix $\rho_{\psi}$ is idempotent and thus has zero total entropy:
\begin{equation}
    \rho_{\psi}=\ket{\psi}\bra{\psi}, \quad \rho_{\psi}^2=\rho_{\psi}, \quad S(\psi)=-\tr \rho_\psi \log_2 \rho_\psi=0.
\end{equation}
For this paper, we measure entropy in \emph{bits}: every $\log$ should be interpreted as $\log_2$ throughout.  As we will see shortly, this convention results in positive integer entries for every element in the entropy vector for all stabilizer states.

Non-trivial entropic structure arises when we consider how one subset of qubits relates to its complement.  Suppose we pick a $p$-qubit subset $I$ of the $n$ qubits in a full stabilizer state.  Then, the entanglement entropy between the $p$-qubit subset $I$ and its $(n-p)$-qubit complement $\bar{I}$ is given by
\begin{equation}
    \rho_{I}=\tr_{I} \ket{\psi}\bra{\psi}, \quad S_I=-\tr \rho_I \log_2 \rho_I.
\end{equation}
Here the trace $\tr_{I}$ is taken over only the $p$ qubits in the subset $I$. The density matrix $\rho_I$ is called a reduced density matrix, and since stabilizer states are pure only their reduced density matrices have nonzero entropy.    Since the entropy for the full state is zero, we also have $S_I=S_{\bar{I}}$.

For an $n$-qubit stabilizer state, there are thus $2^{n-1}-1$ entropies.  Listing all of these entropies produces the \emph{entropy vector} for a given state. As an example, 2-qubit stabilizer states have full entropy vector $\vec{S}=(S_A,S_O,S_{AO})$.  However, since the state is pure, $S_{AO}=0$ and $S_A=S_O$. We thus write the 2-qubit entropy vector as just $\vec{S}=(S_A)$. Here, we have labelled our last qubit with $O$ to indicate it acts as a purifier for the other qubits.

Similarly, for a $3$-qubit state, we write $\vec{S}=(S_A,S_B,S_O)$, or equivalently, $\vec{S}=(S_A,S_B,S_{AB})$. For $4$ qubits we have $\vec{S}=(S_A,S_B,S_C,S_O,S_{AB},S_{AC},S_{AO})$.  We could have written $S_O=S_{ABC}$ and $S_{AO}=S_{BC}$ instead; some sources choose a different ordering for the entropy vector accordingly.

As reviewed in Section \ref{sec:StabilizerReview}, only CNOT gates can create or destroy entanglement entropy.  We can now refine this statement: the $CNOT_{i,j}$ gate can only alter entropies $S_I$ where qubit $i\in I$ but qubit $j\in \bar{I}$, or vice versa.

In addition to the equation $S_I=S_{\bar{I}}$, which holds for any pure state, entropies for subsets of qubits also obey entropy inequalities.  The full set of entropy inequalities obeyed by a given set of states defines the \emph{entropy cone} \cite{10.1109/18.641561,1193790}. The quantum entropy cone is the largest region we will discuss; any quantum state obeys the inequalities that define its boundaries. The Araki-Lieb inequality \cite{cmp/1103842506} $S_{IJ}+S_I\geq S_J$ and subadditivity $S_I+S_J\geq S_{IJ}$, where $I,J$ are disjoint sets of qubits, are both examples of inequalities
obeyed by all quantum states. The full quantum cone at arbitrary qubit number is not
known, but many classes of inequalities are \cite{681320,4215134,Linden:2004ebt,Schnitzer:2022exe}. 

Instead, we will be interested in two smaller cones: the stabilizer cone, and the holographic entropy cone. The stabilizer cone \cite{Linden:2013kal,doi:10.1063/1.4818950,HernandezCuenca:2019wgh,Bao:2020zgx,Bao:2020mqq} is defined as the smallest convex cone which contains all stabilizer states.  Since all states we study are stabilizer states, they will all lie within the stabilizer cone. Although we will not study this cone in further detail, we will use the fact that it is larger than our next cone: the holographic entropy cone.

As defined in \cite{Bao:2015bfa}, the holographic entropy cone is the smallest convex cone in entropy space which contains all quantum states that have a dual representation as a classical gravity state.  The Ryu-Takayanagi formula relates the entanglement entropies for subregions of field theoretic states to areas of extremal surfaces in their dual holographic geometries.  The geometry of these extremal surfaces constrains the allowed entropy vectors.  The first such constraint was the monogamy of mutual information%
\footnote{As discussed in \cite{Bao:2015bfa,HernandezCuenca:2019wgh}, at 6 qubits (5 regions), further entropy inequalities arise not described here.  Since we limit our detailed discussion to 5 qubits or fewer, the Araki-Lieb, subadditivity, and monogamy inequalities are sufficient to test if a state lies within the holographic entropy cone. }%
\ \cite{Hayden:2011ag}, 
\begin{equation}\label{eq:MonogamyofMutualInformation}
    S_{IJ}+S_{IK}+S_{JK}\geq S_{IJK}+S_I+S_J+S_K.
\end{equation}
Here again $I,\, J,\, K$ are disjoint sets of qubits. This inequality is not obeyed by all quantum states, nor by all stabilizer states, but it is obeyed by all states which have a dual smooth classical geometry.  As a consequence, it sets the first boundary between the stabilizer and holographic cones.  As we will review below, beginning at four qubits (or, in the holographic dual language, three regions plus a purifier), some stabilizer states cannot have a smooth holographic dual, because they do not lie within the holographic cone. One of our interests in studying the reachability diagrams is to understand what gate actions on a given state can move it from within the holographic cone to outside of it.  These gate actions then describe how to create a state whose geometry is definitely nonclassical.

\section{The Two-Qubit Stabilizer Graph}\label{sec:two}

At two qubits, the reachability graph contains $60$ vertices, representing the full set of 2-qubit stabilizer states. These states are connected by the six Clifford gates: $H_1,\, H_2,\, P_1,\, P_2,\, CNOT_{1,2},$ and $CNOT_{2,1}$. The full graph is visible in Figure \ref{TwoQubitCompleteGraph}. 

As discussed in section \ref{sec:EntropyCone}, 2-qubit stabilizer states have a one-component reduced entropy vector $S_A$. For these states, $S_A$ is either zero or one,%
\footnote{Since we are working with qubits, we measure entropies using $\log_2$. That is, the reduced density matrix of one qubit in a maximally-entangled pair has von Neumann entropy 1.} %
so states on the reachability graph (represented by vertices) are either unentangled (blue) or form an maximally entangled pair (red). Since only a CNOT gate can alter the entropy vector, the graph has two subgraphs which are connected only by CNOT gates (pink lines).  One subgraph has all of the entangled stabilizer states, while the other has all of the unentangled ones.  At any number of qubits, removing all of the CNOT gates breaks the full graph into subcomponents. Each subcomponent has the same entropy vector throughout.

The graph here depicts every gate acting on each vertex. Since some of these gate actions act trivially on particular states, the graph contains loops. These loops thus represent gate actions which stabilize the state represented by the vertex attached to the loop. This graph also contains degenerate gate action, i.e. multiple edges that map one vertex to another as can be seen in the bottom-rightmost pair connected by $H_1$ and $H_2$. Beginning in section \ref{2qubitHCNOT}, we will suppress trivial loops since we are most interested in understanding the gates that move us between states. 

\begin{figure}[h]
\begin{center}\includegraphics[scale=0.45]{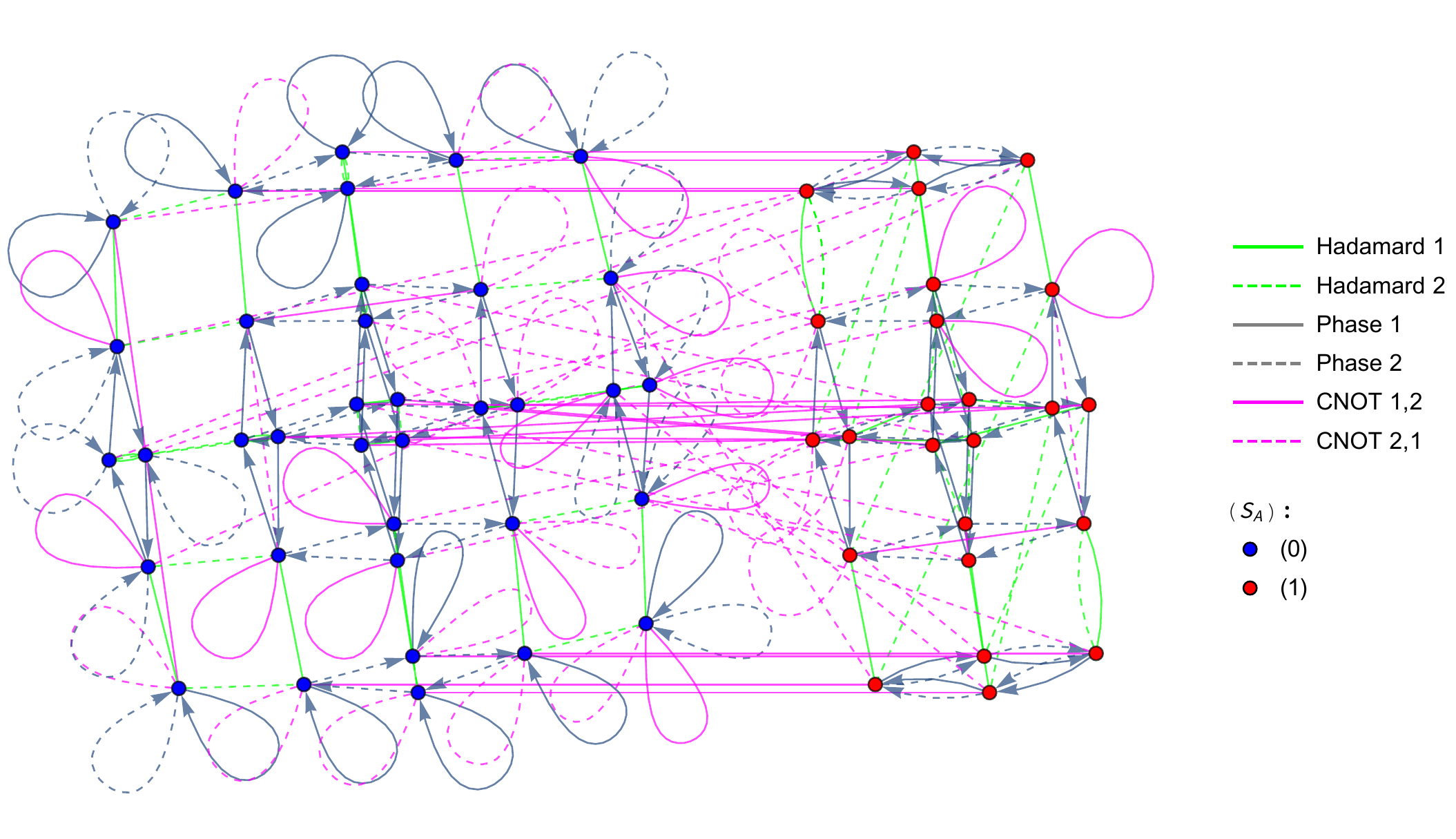}
\caption{This complete 2-qubit reachability graph depicts the full map between all stabilizer states under action of the Clifford group. Edges depicting $H_i$ and $CNOT_{i,j}$ are undirected since they are each their own inverse.  $P$ is not its own inverse since $P^4 = I$;  consequently the phase gates are represented by directed edges.  The color of the edge indicates the type of gate, and the line texture (solid vs. dashed) indicates the qubits which the gate acts on.}
\label{TwoQubitCompleteGraph}
\end{center}
\end{figure}

While Figure \ref{TwoQubitCompleteGraph} completely describes the connections between 2-qubit stabilizer states via the Clifford gates, the graph's complexity obscures some of the important features.  At higher qubit numbers, the full graphs quickly increase in complexity; we thus relegate their complete graphs to Appendix \ref{Two-Qubit Graphs}.  In order to further explore the structure of the 2-qubit graph, and to extend our understanding to higher qubits, we will now explore restricted graphs which only depict the action of various subsets of the Clifford gates.  

\subsection{Restricted Graphs}

Beginning at three qubits, we will further restrict our focus to the restricted graphs composed of the Hadamard and CNOT gates on only the first two qubits. To help motivate why we concentrate on this gate subset, we begin by constructing several different restricted graphs for two qubits.

We construct a restricted graph by considering only select operations of the full Clifford group. This restriction corresponds to removing edges, representing eliminated gate operations, from the complete reachability graph. Each restricted graph reveals different details about the 
connectivity and physics of the stabilizer states.

\subsubsection{Two-Qubit Hadamard}

We begin by considering only the two Hadamard gates on two qubits, $H_1$ and $H_2$, as in  Figure \ref{TwoQubitH1H2}. The Hadamard gate $H_i$ enacts a basis change on the $i$th qubit.  Consequently, Hadamards on different qubits commute.  Since each Hadamard gate also satisfies $H_i^2=\mathbbm{1}$, each subgraph can have at most four different states.  Indeed, the majority of the 2-qubit states organize themselves into squares, consisting of a starting state $\ket{\psi}$ and the states $H_1\ket{\psi}$, $H_2\ket{\psi}$, and $H_1H_2\ket{\psi}=H_2H_1\ket{\psi}$. Additionally, since Hadamard gates cannot change entanglement, each square has the same entropy vector throughout.%
\begin{figure}[h]
\begin{center}
\includegraphics[scale=0.8]{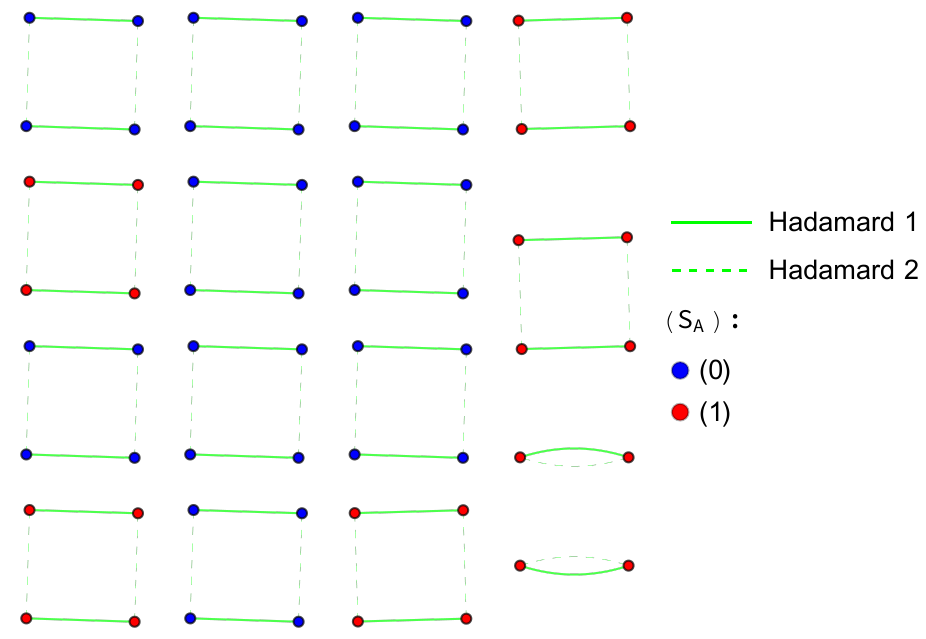}
\caption{The $14$ squares and $2$ connected pair subgraphs in the graph restricted to $H_1$ and $H_2$ at two qubits. Since $[H_1,H_2]=0$ and $H_1^2=H_2^2=\mathbbm{1}$, these subgraphs are the only allowed shapes. The four states which connect in pairs rather than squares are given in Equation \eqref{2qubitHadamardPairs}.}
\label{TwoQubitH1H2}
\end{center}
\end{figure}%

The four states arranged in pairs in the lower right of Figure \ref{TwoQubitH1H2} are worth further note.  They are
\begin{align}\nonumber
    \ket{00}+\ket{11}, \qquad \ket{00}+\ket{01}+\ket{10}-\ket{11}&=H_1\left(\ket{00}+\ket{11}\right)=H_2\left(\ket{00}+\ket{11}\right);
    \\\label{2qubitHadamardPairs}
    \ket{01}-\ket{10}, \qquad \ket{00}-\ket{01}-\ket{10}-\ket{11}&=H_1\left(\ket{01}-\ket{10}\right)=H_2\left(\ket{01}-\ket{10}\right).
\end{align}
Since $H_1$ and $H_2$ produce the same action on each of these states, their subgraphs thus form degenerate pairs instead of full squares. That is, the four states given in \eqref{2qubitHadamardPairs} are eigenstates of $H_1\otimes H_2$; the two states given in the first line have eigenvalue $+1$, while the two states on the second line have eigenvalue $-1$.

\subsubsection{Two-Qubit Phase and Hadamard}\label{TwoQubitPandHSection}

Considering both the Hadamard and phase gates yields the restricted graph depicted in Figure \ref{TwoQubitH1H2P1P2}. The right subgraph of Figure \ref{TwoQubitH1H2P1P2} contains all entangled states. The left subgraph of this figure consists entirely of unentangled states. This bisection of the restricted graph is required by the removal of CNOT edges since neither Hadamard nor phase can alter the entanglement entropies.
\begin{figure}[h]
\begin{center}
\includegraphics[scale=0.42]{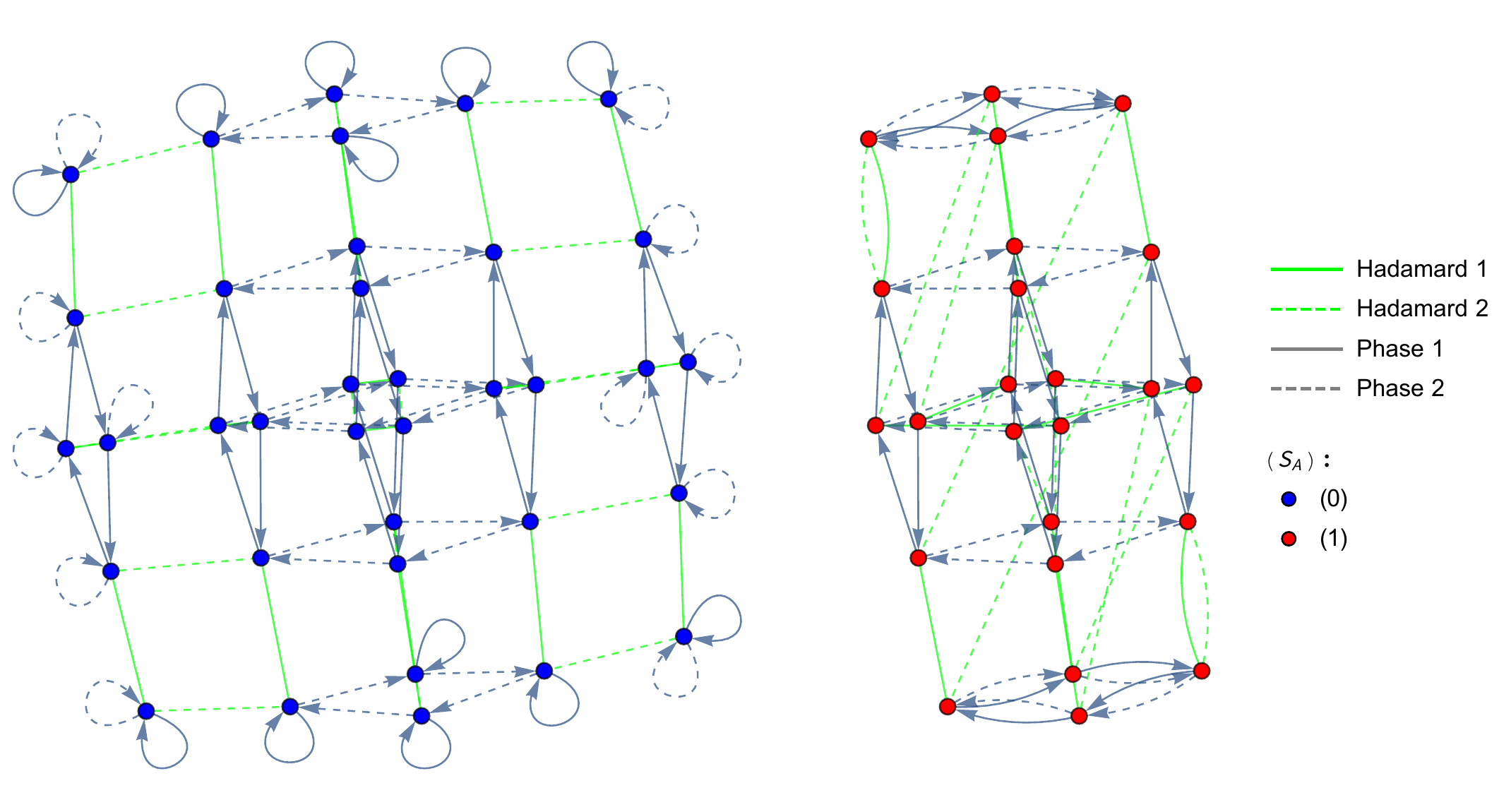}
\caption{Removing the set of CNOT operations from the full reachability graph Figure \ref{TwoQubitCompleteGraph} reveals two disconnected subgraphs.  Since only the CNOT gates can change the entropy, each subgraph has the same entropy vector for all of its states.}
\label{TwoQubitH1H2P1P2}
\end{center}
\end{figure}%

The Hadamard boxes from Figure \ref{TwoQubitH1H2} are clearly visible in the unentangled subgraph.  For the entangled state subgraph, the two degenerate pairs are present at the upper left and lower right, while the Hadamard boxes are still present but slightly harder to visualize.  Removing the phase gates of course reproduces the Hadamard-only Figure \ref{TwoQubitH1H2}, while removing both Hadamard operations yields the phase-only restricted graph (Figure \ref{TwoQubitP1P2} in Appendix \ref{Two-Qubit Graphs}).

The four basis states  $(\ket{00},\ket{01},\ket{10},\ket{11})$, are located at the corners of the unentangled subgraph.%
\footnote{Qubits added to a system are appended to the right of the quantum register, as described in section \ref{sec:StabilizerReview}. Thus, an $n$-qubit product state is represented by $\ket{a_1\ldots a_n}\equiv \ket{a_1}_1\otimes\ldots\ket{a_n}_n$.}
Since both phase gates act trivially on basis states, each corner has two attached directed loops.  For the other sixteen states on the outside of the unentangled subgraph (four on each side), one of the phase gates is trivial while the other makes a square (since $P_i^4=\mathbbm{1}$).  As an example, the state $\ket{+0}=H_1\ket{00}$ satisfies
\begin{equation}
    P_2\ket{+0}=\ket{+0}, \quad P_1\ket{+0}=\ket{i0}, \quad P_1^4\ket{+0}=\ket{+0}.
\end{equation}
The remaining 16 states in the center of the unentangled graph are connected by 4 $P_1$ squares and 4 $P_2$ squares, arising because $P_1^4=P_2^4=\mathbbm{1}$.

In the entangled subgraph, again the 16 states in the center are connected by 4 $P_1$ squares and 4 $P_2$ squares.  The remaining 8 entangled states, 4 at the top and 4 at the bottom of the entangled subgraph in Figure \ref{TwoQubitH1H2P1P2}, are in degenerate pairs.  For these states, either $P_1\ket{\psi}=P_2^3\ket{\psi}$ or $P_1\ket{\psi}=P_2\ket{\psi}$. More precisely, $\ket{01}+\ket{10}$ is one of the states at the top of the figure, so it satisfies
\begin{equation}
    P_1(\ket{01}+\ket{10})=P_2^3(\ket{01}+\ket{10})
\end{equation}
while $\ket{00}+\ket{11}$ is at the bottom of the figure and satisfies
\begin{equation}
    P_1(\ket{00}+\ket{11})=P_2(\ket{00}+\ket{11}.
\end{equation}
Consequently, the phase squares degenerate to connected pairs on all eight of these states.

We can also see the single-qubit reachability diagram reflected here.  In the unentangled graph, removing $H_2$ results in 6 copies of the one-qubit diagram in Figure \ref{OneQubitReachabilityDiagram}, arranged vertically in Figure \ref{TwoQubitH1H2P1P2}. The leftmost copy results from tensoring the six stabilizer states on the first qubit $\ket{0},\, \ket{1}, \, \ket{\pm},\,\ket{\pm i}$ with the state $\ket{0}$ on the second qubit.  Similarly the rightmost copy includes the states
$\ket{01},\, \ket{11},\, \ket{\pm,1},\,\ket{\pm i,1}$.  The four copies in the middle, still connected by the phase gate $P_2$, are constructed similarly except with $\ket{\pm},\, \ket{\pm i}$ for the second qubit.  This tensoring is our first example of a lift, in this case from a one-qubit structure to a two-qubit structure.

In general, a lift of a $k$-qubit state $\ket{\psi}$ to $n$ qubits is a quantum channel which maps $\ket{\psi}$ into $\mathbb{C}^{2n}$ by first tensoring on a $(n-k)$-qubit state $\ket{\phi}$ and then applying an operator $\mathcal{O}$:
\begin{equation}
    \ket{\psi} \rightarrow \mathcal{O}\left(\ket{\Psi}\otimes\ket{\phi} \right): \mathcal{O}\in L(\mathbb{C}^{2n}),\:\ket{\phi}\in\mathbb{C}^{2(n-k)}.
\end{equation}
 We will always have in mind lifts which take stabilizer states to stabilizer states, so we will take $\ket{\phi}$ to be an $(n-k)$-qubit stabilizer state and $\mathcal{O}$ to be a product of $n$-qubit Clifford gates. Given these restrictions, all lifts from $\mathbb{C}^{2k}$ to $\mathbb{C}^{2n}$ are on the same footing, and indeed lifts of $\ket{\psi}$ also successfully lift any other $k$-qubit stabilizer state. What makes a particular lift useful is that it preserves some of the structure seen at $k$ qubits when we go to $n$ qubits. In the example above, the lift given by $\ket{\phi}=\ket{0}$, $\mathcal{O}=I$ mapped all six one-qubit stabilizer states to six two-qubit stabilizer states, preserving their arrangement in the one-qubit reachability graph.

The entangled subgraph is not composed of tensor products of one-qubit stabilizer states, so it is unsurprising that the one-qubit reachability diagram has become more complicated.  We still see four almost-copies of the one-qubit diagram, except the Hadamard gate which connects $\ket{\pm i}$ in Figure \ref{OneQubitReachabilityDiagram} instead connects the four almost-copies to each other. Last, removing $H_1$ instead of $H_2$ results in exactly the same structures; we have just chosen to arrange the entangled subgraph to prioritize the $H_1$ structure.

\subsubsection{Two-Qubit Phase and CNOT}

Considering only operations of the subgroup generated by phase and CNOT, we observe a reduced graph composed of five disconnected substructures (Figure \ref{TwoQubitP1P2CNOT12CNOT21}). Each substructure is inherited from a combination of phase-only and CNOT-only restricted graphs (Figures \ref{TwoQubitP1P2} and \ref{TwoQubitCNOT12CNOT21} in Appendix \ref{Two-Qubit Graphs}). 

\begin{figure}[h]
\begin{center}
\includegraphics[scale=0.5]{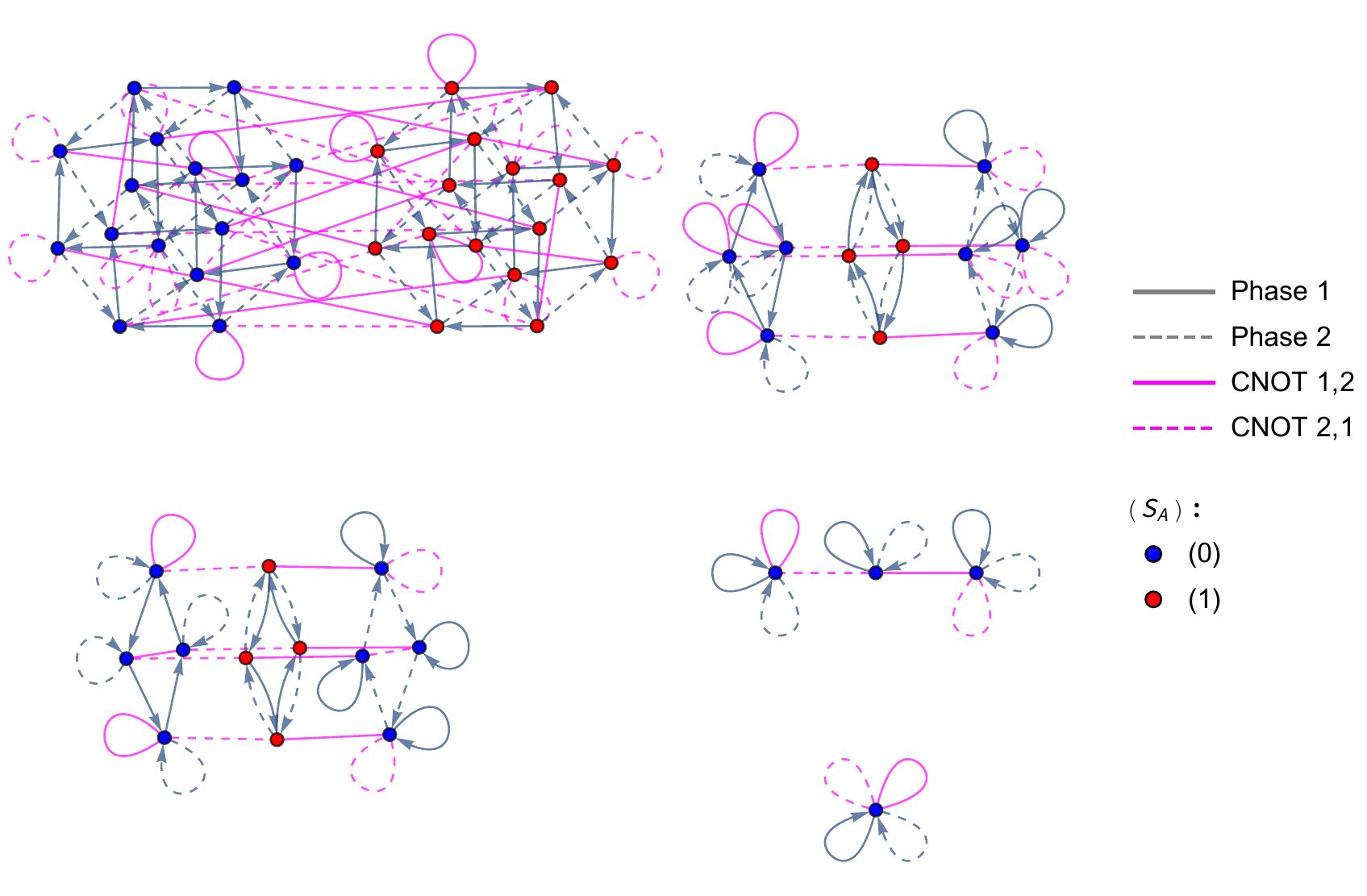}
\caption{Subgraph of 2-qubit complete reachability graph restricted the subgroup generated by CNOT and phase operations. Included structures are various attachments of the simpler structures found in the only-phase and only-CNOT restricted graphs in Figures \ref{TwoQubitP1P2} and \ref{TwoQubitCNOT12CNOT21} in Appendix \ref{Two-Qubit Graphs}.}
\label{TwoQubitP1P2CNOT12CNOT21}
\end{center}
\end{figure}

Again, only CNOT gates can alter the entropy, so states with different entanglement are connected only via CNOT gates.  However, only some CNOT gates modify the entropy.  For example,
\begin{equation}
    CNOT_{1,2}\ket{00}=CNOT_{2,1}\ket{00}=\ket{00}.
\end{equation}
Since $P_1$ and $P_2$ also stabilize $\ket{00}$, this state is represented by the isolated vertex in Figure \ref{TwoQubitP1P2CNOT12CNOT21}. The CNOT gates permute the remaining basis states
$\ket{01},\, \ket{10},\, \ket{11}$
among each other since CNOT can only flip a bit, not introduce a superposition. As with $\ket{00}$, phase acts trivially on all the remaining basis states.

The upper right subgraph in Figure \ref{TwoQubitP1P2CNOT12CNOT21} consists of $8$ unentangled states and $4$ entangled states. The unentangled states arrange into a $P_1^4$ cycle to the left, and a $P_2^4$ cycle to the right. For the four central entangled states, $P_1\ket{\psi}=P_2\ket{\psi}$, so they are linked in one phase cycle. The entangled states and unentangled states are necessarily connected by CNOT gates.  In this subgraph, all CNOT gates either act trivially or move between entangled and unentangled states.  The CNOT gates alone connect the states into 4 lines of 3 states each.

The bottom left subgraph is similar, except the central four entangled states satisfy
$P_1\ket{\psi}=P_2^{-1}\ket{\psi}$ instead. The three top states, and the three bottom states, both have either trivial CNOT action or CNOT moves between an entangled and unentangled states. For the middle states, however, both CNOT gates act nontrivially on each state, producing a hexagon.  Explicitly, we have the cycle
\begin{equation}
    \ket{1i}=CNOT_{2,1}CNOT_{1,2}CNOT_{2,1}CNOT_{1,2}CNOT_{2,1}CNOT_{1,2}\ket{1i},
\end{equation}
where only two of the states along the way are entangled.

Finally the largest structure, located in the upper left of Figure \ref{TwoQubitP1P2CNOT12CNOT21}, contains only states stabilized by neither $P_1$ nor $P_2$. Again the CNOT gates are the only connections between the entangled and unentangled states. Considering only these CNOT gates, this subgraph contains the remaining $2$ hexagonal cycles from the CNOT-only graph (Appendix \ref{Two-Qubit Graphs}, Figure \ref{TwoQubitCNOT12CNOT21}), as well as $6$ 3-state lines of paired CNOT edges, and $2$ additional states invariant under both CNOT gates.

\subsection{Two-Qubit Hadamard and CNOT}\label{2qubitHCNOT}

For the remainder of this work, we will focus on restricted graphs generated by $H_1,\, H_2,$  $CNOT_{1,2},$ and $CNOT_{2,1}$ (see Figure \ref{TwoQubitH1H2CNOT12CNOT21} for the two-qubit version). In order to understand changes in entanglement, we need to consider a restricted graph which includes CNOT gates since they are the only gates which alter entropy.  We specifically choose to include the Hadamard gates (and exclude the phase gates) because in the Hadamard-CNOT restricted graphs, each subgraph contains states with different entropy vectors.  Additionally, at two qubits, the Hadamard-CNOT graph will have only two subgraphs; we will see echoes of these structures repeated at higher qubit number.  When we go to higher qubit number, we will continue to use only the gate set  $H_1,\,H_2,\,CNOT_{1,2},CNOT_{2,1}$ because all entropic arrangements found in stabilizer states can be built from successive bipartite entanglements.%
\begin{figure}[h]
\begin{center}
    \begin{overpic}[width=0.8\textwidth]{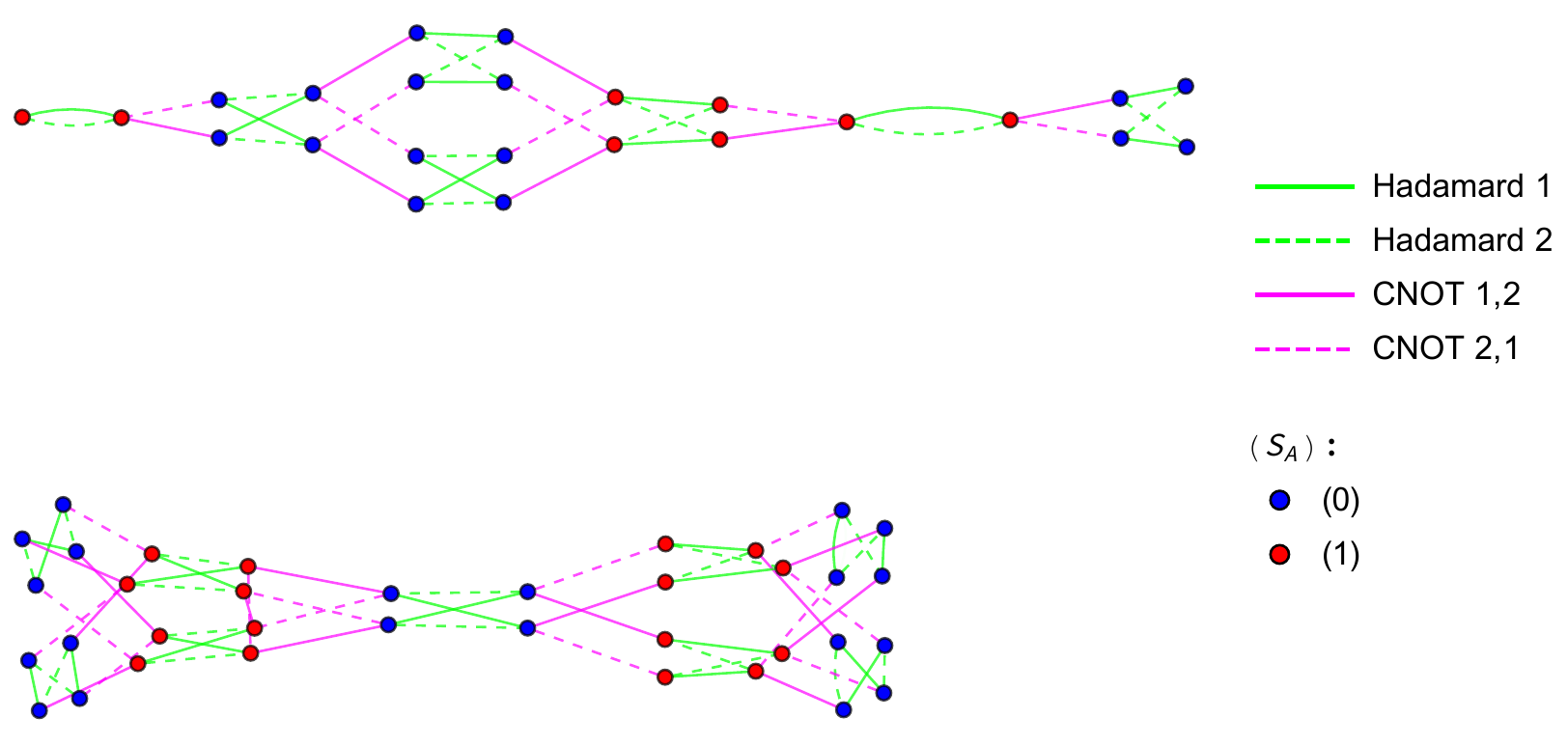}
		\put (37,33) {$g_{24}$}
		\put (27,13) {$g_{36}$}
		\put (76.6,41.4) {\footnotesize{$\leftarrow \ket{00}$}}
		\put (54.1,42) {\footnotesize{$\ket{GHZ} \downarrow$}}
		\end{overpic}
\caption{The 2-qubit subgraph restricted to $H_1,H_2,CNOT_{1,2},$ and $CNOT_{2,1}$ has two subgraphs which are connected only via phase gates. Trivial loops have been removed in this representation of the 2-qubit $H_1,H_2,CNOT_{1,2},CNOT_{2,1}$ restricted graph.}
\label{TwoQubitH1H2CNOT12CNOT21}
\end{center}
\end{figure}%

Beginning with this graph, and continuing below, we omit any gate whose action is the identity; thus no further trivial loops will appear. Because we have four possible gates that can act on each state, a vertex with valency $4-k$ has $k$ trivial loops. We also label the subgraph structures by the number of vertices they contain, so e.g. we use $g_{24}$ for the $24$-vertex substructure in Figure \ref{TwoQubitH1H2CNOT12CNOT21}. As noted in Footnote \ref{fn:cnot_dependence} above, we have the relation
\begin{equation}
    CNOT_{2,1}=H_1 H_2 CNOT_{1,2} H_2 H_1,
\end{equation}
which can be checked explicitly for $2$-qubit states using the figure; hence the presence of the $CNOT_{2,1}$ edges is completely fixed by the structure of the other three gates.

The subgraph $g_{24}$ contains all 2-qubit stabilizer states connected to the basis states $(\ket{00},\ket{01},\ket{10},\ket{11})$ via only Hadamard and CNOT operations. As we can see from the graph, acting a Hadamard and then a CNOT on $\ket{00}$ produces the GHZ state, which is entangled:
\begin{equation}
    CNOT_{1,2}H_1\ket{00}=\ket{GHZ}.
\end{equation}

Because the phase gate is the only Clifford gate with imaginary matrix elements, all states in $g_{24}$ can be written as superpositions of the basis states with purely real coefficients. States located in $g_{36}$, on the other hand, have relative phases $\pm i$ between different basis components. Accordingly, a phase gate is required to move between the two subgraphs.  In fact, every phase gate either acts trivially, or connects states in $g_{24}$ with states in $g_{36}$, since the CNOT and Hadamard gates are both Hermitian, but phase is not. Thus, the only way a product of CNOTs and Hadamards can have the same action as a non-Hermitian operator $O$ on a state is when the state has support only on the eigenspaces of the operator $O$ with real eigenvalues. For $O=\mathrm{phase}$, the only eigenspace with real eigenvalue has eigenvalue one. Hence the phase gate either acts as the identity or its action is non-Hermitian (and thus moves us between the $g_{24}$ and $g_{36}$ subgraphs).  We will see echoes of this structure when comparing subgraphs in the $H_1,\, H_2,\, CNOT_{1,2},\, CNOT_{2,1}$ restricted graphs at higher qubit number.

In our analysis at higher qubits, we will also rely on Hamiltonian paths and Hamiltonian cycles.  Hamiltonian paths visit every vertex in a graph only once (and thus do not self-intersect). Hamiltonian cycles are closed loops with the same property. The subgraph $g_{24}$ has no Hamiltonian paths, and therefore no Hamiltonian cycles either.
Subgraph $g_{36}$ does have Hamiltonian paths (although again no Hamiltonian cycles). One example is shown in Figure \ref{TwoQubitHamiltonianPath}.  The specific circuit depicted is
\begin{equation}\label{g36HamiltonianCircuit}
\begin{split}
	\mathcal{C} \equiv (H_2,&H_1,H_2,CNOT_{1,2},H_2,H_1,H_2,CNOT_{1,2},H_2,CNOT_{1,2},H_2,CNOT_{1,2},\\
&H_1,H_2,H_1,CNOT_{1,2},H_1,CNOT_{1,2},H_2,CNOT_{1,2},H_2,H_1,H_2,CNOT_{1,2},\\
&H_1,CNOT_{1,2},H_2,CNOT_{1,2},H_2,H_1,H_2,CNOT_{2,1},H_2,H_1,H_2).
\end{split}
\end{equation}
When applying this circuit to the graph, the leftmost gate acts first and the rightmost gate last. 
\begin{figure}[h]
\begin{center}
    \begin{overpic}[width=0.95\textwidth]{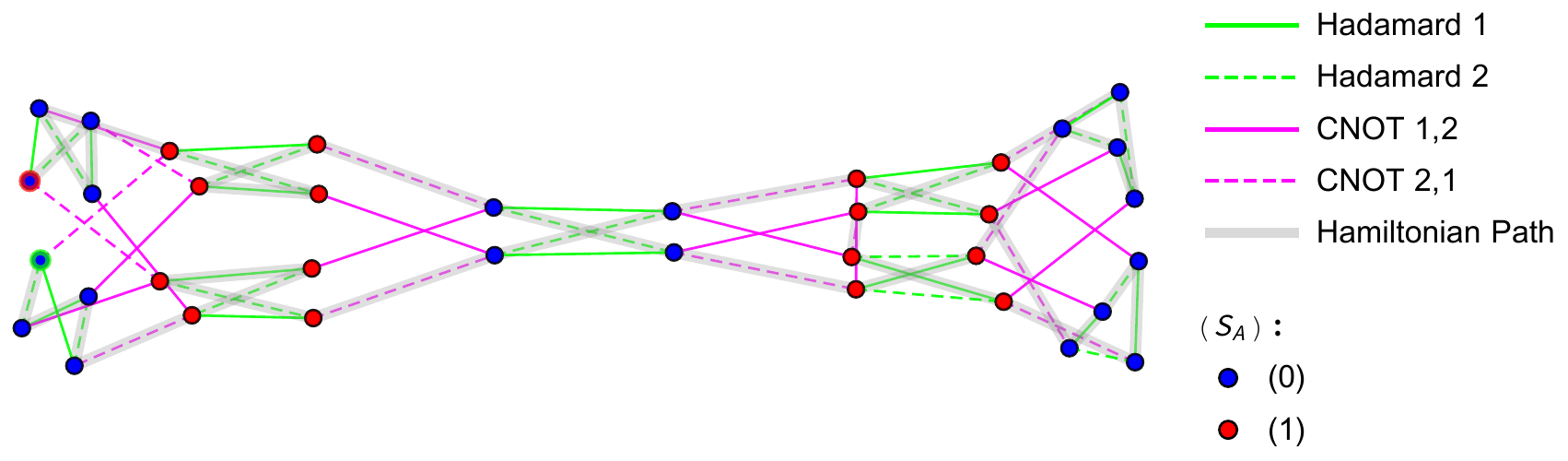}
		\put (-6.5,18) {\footnotesize{$\ket{-,-i}$}}
		\put (-3,12) {\footnotesize{$\ket{i,1}$}}
		\end{overpic}
\caption{Hamiltonian path $\mathcal{C}$ beginning on state $\ket{i,1}$ (encircled in green) and ending on state $\ket{-,-i}$ (encircled in red).}
\label{TwoQubitHamiltonianPath}
\end{center}
\end{figure}

Note that the path $\mathcal{C}$ starting on $\ket{i1}$ is not unique; other Hamiltonian paths exist on $g_{36}$.  First, $\mathcal{C}$ also traverses a Hamiltonian path starting on $\ket{-i,1}$ instead. Another example is $\mathcal{C^\top}$, which swaps qubits 1 and 2 in \eqref{g36HamiltonianCircuit}.  Two further possibilities are the inverse paths $\mathcal{C}^{-1}$ and $(\mathcal{C}^\top)^{-1}$, which simply apply the respective circuits in reverse order. In fact, there exists a Hamiltonian path on $g_{36}$ starting from $30$ (out of $36$) of its vertices.

For definiteness, we will use the Hamiltonian path $\mathcal{C}$, starting on $\ket{i1}$, in the lifting procedure introduced in Section \ref{liftTwoToThree}, where we turn to the three-qubit stabilizer graph and a detailed analysis of its reduced graphs. 

\section{The Three-Qubit Stabilizer Graph}\label{sec:three}

At three qubits there are now $1080$ stabilizer states, and accordingly the full reachability graph, which we defer to Figure \ref{ThreeQubitCompleteGraph} in Appendix \ref{Two-Qubit Graphs}, becomes unwieldy; we thus proceed in this section immediately to the restricted graphs. For three qubits, many of the reduced graphs show features similar to two qubits.  As an example, the Hadamard-only restricted graph at three qubits contains cubes and degenerate squares instead of squares and degenerate pairs.  As before, each subgraph has only one entropy type, since Hadamard gates cannot alter the entropy vector.  Figure \ref{ThreeQubitH1H2H3} shows the structures exhibited in the $H_1,\, H_2,\, H_3$ restricted three-qubit graph.  The phase graph at three qubits similarly extends, as exhibited by the full phase graph $P_1,\, P_2,\, P_3$, shown in Appendix \ref{Two-Qubit Graphs} in Figure \ref{ThreeQubitP1P2P3Subgraphs}.%
\begin{figure}[h]
\begin{center}
\includegraphics[scale=0.8]{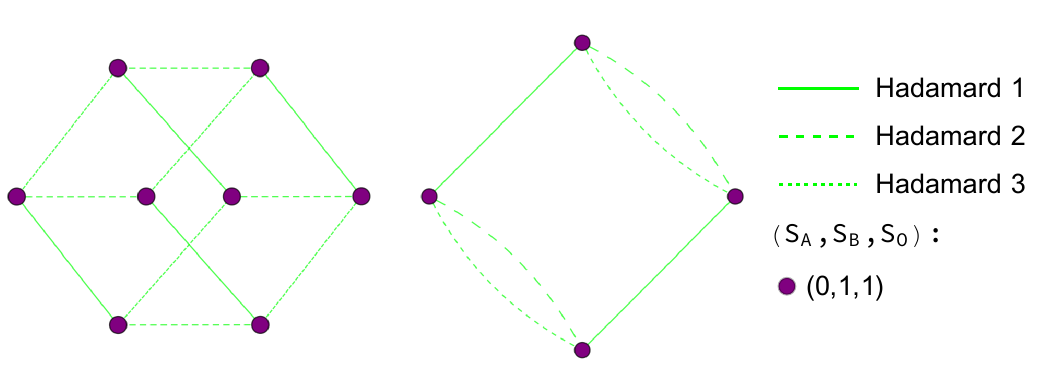}
\caption{Pictured are the two unique subgraph types that occur in the 3-qubit restricted graph of only $H_1,H_2,$ and $H_3$ operations. Degenerate pairs in the 2-qubit $H_1,H_2$ restricted graph (Figure \ref{TwoQubitH1H2}) are promoted to boxes at three qubits. Boxes from the 2-qubit $H_1,H_2$ graph are promoted to cubes by the addition of the $H_3$ gate.}
\label{ThreeQubitH1H2H3}
\end{center}
\end{figure}

In all of these graphs, the color of the vertex indicates the entropy vector of the associated state. As reviewed in section \ref{sec:EntropyCone} the entropy vector at three qubits is $\vec{S}=(S_A,S_B,S_{AB})$.  There are five different entropy vectors among the 3-qubit stabilizer states, as shown in Table \ref{tab:ThreeQubitEntropyVectors}.  All five entropy vectors lie within the holographic cone (that is, they satisfy the inequalities required for a holographic state, as discussed in Section \ref{sec:EntropyCone}).

Just as in the 2-qubit case, the entropy vector can only be changed by the action of a CNOT gate.  Accordingly the phase-- and Hadamard--restricted graphs do not allow us to study changes in entropy. Instead, as discussed in Section \ref{2qubitHCNOT}, we are most interested in the restricted graph which considers only $H_1,\, H_2,\, CNOT_{1,2},$ and $CNOT_{2,1}$, since it will allow us to understand the changes in entropy induced by the CNOT gates on a pair of qubits.

\subsection{Three-Qubit CNOT+Hadamard on 1 and 2 only}

We extend our analysis to three qubits, restricting the full stabilizer group to the subgroup generated by $H_1,\,H_2,\,CNOT_{1,2}$, and $CNOT_{2,1}$.  We concentrate on this gate set in order to focus on the entropic structure of qubits $1$ and $2$. The choice of qubits $1$ and $2$ is arbitrary; any pair would do. Because all stabilizer gates act on at most two qubits, and all entropic structure can be built from these bipartite interactions, our analysis of qubits $1$ and $2$ is sufficient to understand reachability for all $n$ qubits. 

The full restricted graph for the gate set $H_1,\,H_2,\,CNOT_{1,2}, \, CNOT_{2,1}$ is shown in Figure \ref{ThreeQubitH1H2CNOT12CNOT21} of Appendix \ref{Two-Qubit Graphs}. There are only four types of subgraph, shown in Figure \ref{ThreeQubitH1H2CNOT12CNOT21Subgraphs}, which arise in the full restricted graph. The full graph consists of $6$ copies of $g_{24}$ and $g_{36}$, $3$ copies of $g_{144}$, and a single copy of $g_{288}$, where as in the previous subsection the subscript denotes the number of vertices in the subgraph. 
\begin{figure}[h]
\begin{center}
        \begin{overpic}[width=0.94\textwidth]{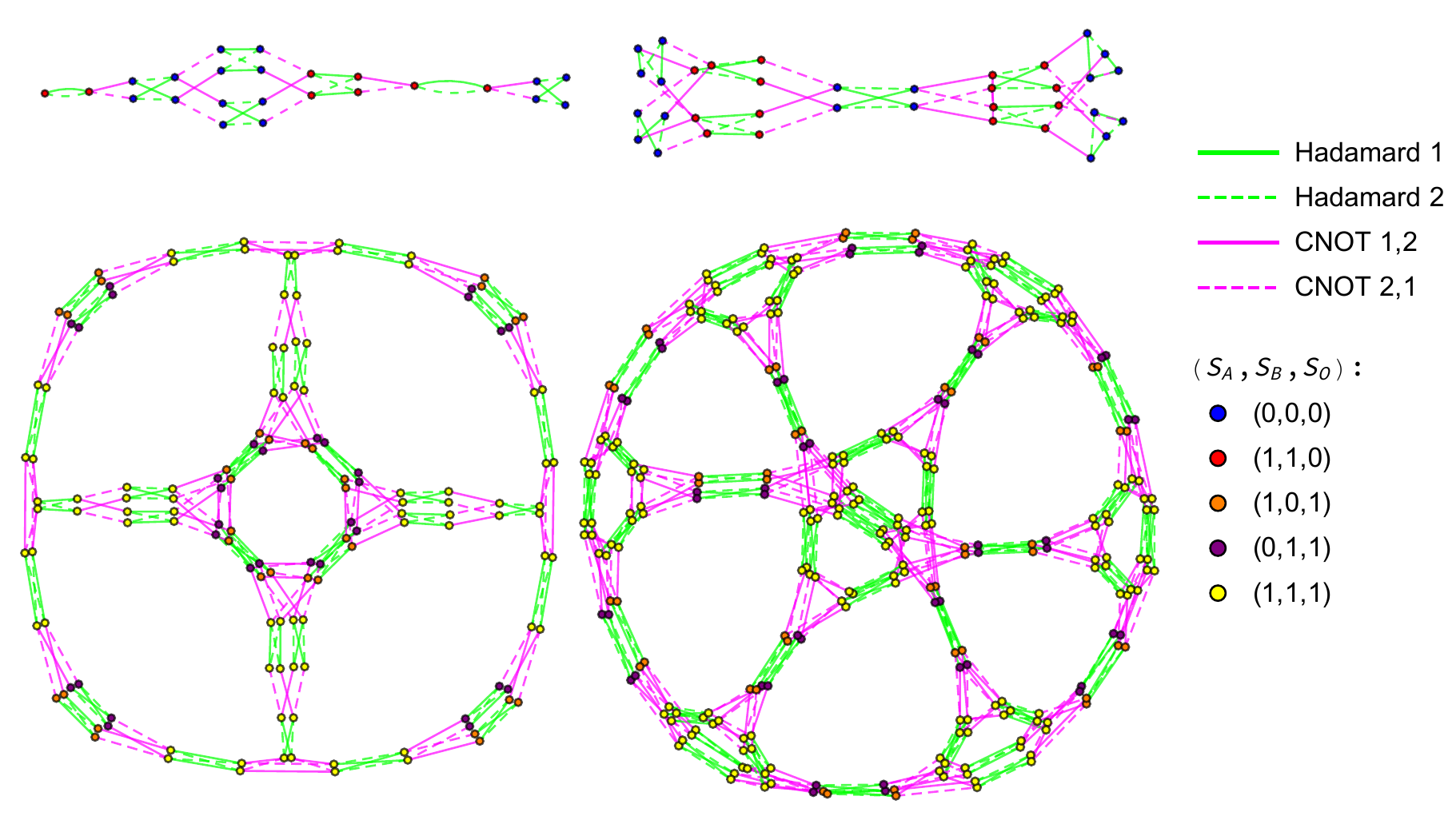}
		\put (20,54) {$g_{24}$}
		\put (58,53) {$g_{36}$}
		\put (17.1,42) {$g_{144}$}
		\put (56.6,43) {$g_{288}$}
		\end{overpic}
\caption{This figure depicts the four unique subgraphs which arise in the full 3-qubit $H_1,H_2,CNOT_{1,2},CNOT_{2,1}$ restricted graph, pictured in Figure \ref{ThreeQubitH1H2CNOT12CNOT21} of Appendix \ref{Two-Qubit Graphs}. Each subgraph is labeled by its corresponding vertex count. The full graph consists of $6$ copies of $g_{24}$ and $g_{36}$, $3$ copies of $g_{144}$, and a single copy of $g_{288}$.}
\label{ThreeQubitH1H2CNOT12CNOT21Subgraphs}
\end{center}
\end{figure}

As mentioned above, there are only five different entropy vectors at three qubits. In Table \ref{tab:ThreeQubitEntropyVectors}, we record the number of stabilizer states with each of these entropy vectors, and also record the subgraph types those states appear in.
\begin{table}[h]
    \begin{center}
    \begin{tabular}{|c||c|c|c|} 
\hline
Holographic & $\left(S_A, S_B, S_{AB}\right)$ & Number of States & Subgraph\\
\hline
\hline
 Yes & $\textcolor{blue}{\bullet} (0,0,0)$ & $216$ & $g_{24}, g_{36}$\\
 \hline
 Yes & $\textcolor{red}{\bullet}(1,1,0)$ & $144$ & $g_{24}, g_{36}$\\
 \hline
 \hline
 Yes & $\textcolor{violet}{\bullet}(0,1,1)$ & $144$ & $g_{144}, g_{288}$\\
 \hline
 Yes & $\textcolor{orange}{\bullet}(1,0,1)$ & $144$ & $g_{144}, g_{288}$\\
 \hline
 Yes & $\textcolor{yellow}{\bullet}(1,1,1)$ & $432$ & $g_{144}, g_{288}$\\
 \hline
\end{tabular}
\end{center}
\caption{Table of 3-qubit entropy vectors. The last column lists which subgraph types exhibit each entropy vector.}
\label{tab:ThreeQubitEntropyVectors}
\end{table}

As discussed in section \ref{2qubitHCNOT}, phase actions on states in the $H_1,\, H_2,\, CNOT_{1,2},\, CNOT_{2,1}$ restricted graph are always either trivial, or move you to a different subgraph.  However, unlike in the 2-qubit case, in addition to the phase gates we are also not representing $H_3$, nor any CNOT gate involving the third qubit.  Thus,
not every subgraph can be reached by phase actions alone: applying a phase gate does not change the entropy vector, and two of the five entropy vectors appear only in $g_{24}, g_{36}$ while the other three appear in only $g_{144}, g_{288}$. We are also not representing $H_3$, nor any CNOT gate involving the third qubit.

In particular, the subgraphs of type $g_{144}$ and $g_{288}$ contain the entropy vectors $(0,1,1)$, $(1,0,1)$, and $(1,1,1)$, but the subgraphs of type $g_{24}$ and $g_{36}$ only contain the entropy vectors $(0,0,0)$ and $(1,1,0)$.  This separation occurs because CNOT gates involving the third qubit are required to change between the two sets of entropy vectors.
As reviewed in section \ref{sec:EntropyCone}, the $CNOT_{i,j}$ gate can only alter entropies $S_I$ where qubit $i\in I$ but qubit $j\in \bar{I}$.  Our entropy vectors are listed as $(S_A,S_B,S_{AB})$, where $A$ refers to qubit 1 and $B$ refers to qubit 2.  So the $CNOT_{1,2}$ or $CNOT_{2,1}$ gates can only affect $S_A$ and $S_B$, not $S_{AB}$.  Thus, entropy vectors with different $S_{AB}$ can only show up on separate subgraphs in our restricted graph.

As we will see in the next section, we can understand both the number of copies of each subgraph, as well as the shape of the subgraphs themselves, by seeing how the 2-qubit states and subgraphs can be embedded into the 3-qubit structures.

\subsection{Lifting States from Two Qubits to Three Qubits}\label{liftTwoToThree}

We apply the lifting procedure introduced in Section \ref{TwoQubitPandHSection} to lift 2-qubit states to three qubits. States in the 2-qubit $g_{24}$ subgraph all lift to states in the 3-qubit $g_{24}$ subgraphs by tensoring on a third qubit. For example, starting with the state $\ket{00}$ on qubits 1 and 2,  we can tensor on $\ket{0}$ on qubit 3 to find
    \begin{equation}\label{g24Lift}
        \ket{000} = \ket{00} \otimes \ket{0}.
    \end{equation}
There are $6$ possible states of the third qubit:\ $\{\ket{0},\, \ket{1}, \, \ket{\pm},\,\ket{\pm i}\}$. Tensoring on all $6$ of these states, as in \eqref{g24Lift}, to the $24$ states in the 2-qubit $g_{24}$ subgraph generates $144$ of the 3-qubit stabilizer states. These $144$ states make up the $6$ copies of $g_{24}$ at three qubits (shown in Figure \ref{ThreeQubitH1H2CNOT12CNOT21} of Appendix \ref{Two-Qubit Graphs}).

The next simplest set of lifts begins with states in the 2-qubit $g_{36}$ subgraph, and then tensors on a third qubit. We start with the 2-qubit state $\ket{i1}$, since the circuit  $\mathcal{C}$ \eqref{g36HamiltonianCircuit} defines a Hamiltonian path on $g_{36}$ starting at this state (Figure \ref{TwoQubitHamiltonianPath}). The process
    \begin{equation}\label{g36Lift}
        \ket{i1+} =  \ket{i1} \otimes \ket{+},
    \end{equation}
lifts $\ket{i1}$ to one of the six 3-qubit $g_{36}$ subgraphs. Since there are again 6 1-qubit stabilizer states available for the third qubit, we generate a further 216 states via this lift.  These 216 states are all located on one of the $6$ copies of $g_{36}$ at three qubits, completely covering those subgraphs. 

When we tensor on a third qubit, we extend the entropy vector of the original 2-qubit state.  For a 2-qubit state with entropy vector $(S_A)$, we have $S_A=S_B$ since it is a pure state. When we tensor on the third qubit, these two values stay the same. Additionally, the third qubit is just tensored on, so it is not entangled; thus $S_O=S_{AB}=0$.  Accordingly, states with the entropy vector $(1)$ lift to states with $(1,1,0)$ and those with entropy vector $(0)$ lift to $(0,0,0)$ under the tensoring lift of Equations \eqref{g24Lift} and \eqref{g36Lift}.

We have accounted so far for all of the states in the $g_{24}$ and $g_{36}$ subgraphs at three qubits.  As shown in Table \ref{tab:ThreeQubitEntropyVectors}, these are the only states with entropy vector $(0,0,0)$ or $(1,1,0)$.   To modify entanglement and reach new 3-qubit entropy arrangements, we must act with CNOT gates involving the third qubit, i.e.\ $CNOT_{1,3},\,CNOT_{3,1},\,CNOT_{2,3},$ or $CNOT_{3,2}$. For example, acting with $CNOT_{3,1}$ on the lifted state $\ket{i1+}$ in Equation \eqref{g36Lift} gives
    \begin{equation}\label{CNOT31OnLiftedState}
        CNOT_{3,1}\ket{i1+} = \ket{010} + i\ket{011}+ i\ket{110}  + \ket{111}.
    \end{equation}
This procedure lifts $\ket{i1}$ from the 2-qubit $g_{36}$ subgraph to a state on one $g_{144}$ subgraph at three qubits. Acting with $CNOT_{3,1}$ on $\ket{i1+}$ entangles qubits $1$ and $3$, resulting in a state with entropy vector $\vec{S} = (1,0,1)$, indicated with an orange vertex in Figure \ref{ThreeQubitH1H2CNOT12CNOT21Subgraphs}. Similarly replacing the third qubit in $\ket{i1+}$ with $\ket{-}, \, \ket{i},$ or $\ket{-i}$ instead also results in states on the same copy of $g_{144}$.  $\ket{i10}$ and $\ket{i11}$, however, return to the same copies of $g_{36}$ they started on under the action of $CNOT_{3,1}$.

Acting a different gate on a lifted version of $\ket{i1}$ can result in a lifted state on a different copy of $g_{144}$. For example, the gate $CNOT_{1,3}$ on state $\ket{i10} \equiv \ket{i1} \otimes \ket{0}$ gives
    \begin{equation}\label{CNOT13LiftedState}
        CNOT_{1,3}\ket{i10} = \ket{010} + i\ket{111},
    \end{equation}
which resides on a different copy of $g_{144}$ than the state lifted in Equation \eqref{CNOT31OnLiftedState}. As before, the third ket can be replaced with $\ket{1},\, \ket{i}$ or $\ket{-i}$ resulting in three other states on the same copy of $g_{144}$, but using $\ket{\pm}$ results in $CNOT_{1,3}$ mapping back to the same copies of $g_{36}$.

The third copy of $g_{144}$ is phase-separated from the previous two. Lifting 2-qubit starting states to this subgraph requires a phase gate. One possible such lift to this subgraph is
    \begin{equation}\label{CNOT13LiftedState}
        P_3(CNOT_{1,3}\ket{i10}) = \ket{010} - \ket{111}.
    \end{equation}
Again, the third qubit can be replaced with $\ket{1},\, \ket{i}$ or $\ket{-i}$ giving three further states on the last copy of $g_{144}$.

There are two ways that we can lift the 2-qubit state $\ket{i1}$ to the 3-qubit $g_{288}$ subgraph. The first begins by tensoring on a third qubit, then applying the appropriate CNOT gate to move the lifted state from $g_{36}$ directly to $g_{288}$. For example,
    \begin{equation}\label{CNOT32LiftTo288}
        CNOT_{3,2}\ket{i1+}) = \ket{010} + i\ket{011} + \ket{100} + i\ket{101}
    \end{equation}
resides on $g_{288}$. It has entropy vector $\vec{S} = (0,1,1)$, and is represented by a purple vertex in $g_{288}$ in Figure \ref{ThreeQubitH1H2CNOT12CNOT21Subgraphs}.  The same procedure works when we replace the third qubit by $\ket{-}, \,\ket{i},$ or $\ket{-i}$.

The second method first lifts state $\ket{i1}$ to $g_{144}$ as in Equation \eqref{CNOT31OnLiftedState}, entangling qubits $1$ and $3$. Applying a second CNOT gate then entangles qubit $2$ with the other qubits.  Explicitly, we have e.g.
    \begin{equation}\label{CNOT32CNOT31LiftTo288}
        CNOT_{3,2}(CNOT_{3,1}\ket{i1+}) = \ket{010} + i\ket{011} + i\ket{100} + \ket{101}.
    \end{equation}
This process results in a final state on $g_{288}$. Its entropy vector  is $(1,1,1)$, represented by a yellow vertex in $g_{288}$.  Three further states arise by replacing the third qubit $\ket{+}$ with $\ket{-}, \,\ket{i},$ or $\ket{-i}$.

In the next subsection, we will use the lifted states described here, combined with the 2-qubit Hamiltonian path $\mathcal{C}$, to reach the remaining states as well as to understand the new $g_{144}$ and $g_{288}$ subgraph structures that arise at three qubits.

\subsection{Lifting Paths from Two Qubits to Three Qubits}

At two qubits, the Hamiltonian path $\mathcal{C}$ \eqref{g36HamiltonianCircuit} starting from the state $\ket{i1}$ covered the $g_{36}$ subgraph.  Similarly, $\mathcal{C}$ starting from the lifted states described in the previous section, listed in Table \ref{tab:OneFullCoveringOfLiftedStates}, cover every vertex exactly once on each of the $g_{36}$, $g_{144}$, and $g_{288}$ subgraphs of the $H_1,\, H_2,\, CNOT_{1,2},\, CNOT_{2,1}$ restricted graph at three qubits.
\begin{table}[h]
\centering
\begin{tabular}{|c|c|c|}
 \hline
 Lift & Starting State & Subgraph  \\
 \hline
  $\ket{i1}\otimes\ket{0}$ & $\ket{i10}$   & $g^1_{36}$ \\
  \hline
  $\ket{i1}\otimes\ket{1}$ & $\ket{i11}$   & $g^2_{36}$ \\
  \hline
  $\ket{i1}\otimes\ket{+}$ & $\ket{i1+}$   & $g^3_{36}$ \\
  \hline
  $\ket{i1}\otimes\ket{-}$ & $\ket{i1-}$   & $g^4_{36}$ \\
  \hline
  $\ket{i1}\otimes\ket{i}$ & $\ket{i1i}$   & $g^5_{36}$ \\
  \hline
  $\ket{i1}\otimes\ket{-i}$ & $\ket{i1-i}$   & $g^6_{36}$ \\
  \hline
  \hline
  $CNOT_{3,1}(\ket{i1}\otimes\ket{+})$ & $\ket{i10}+i\ket{-i,11}$   & \multirow{4}{*}{$g_{144}^1$} \\
  \cline{1-2}
  $CNOT_{3,1}(\ket{i1}\otimes\ket{-})$ & $\ket{i10}-i\ket{-i,11}$   & \\
  \cline{1-2}
  $CNOT_{3,1}(\ket{i1}\otimes\ket{i})$ & $\ket{i10}-\ket{-i,11}$   &  \\
  \cline{1-2}
  $CNOT_{3,1}(\ket{i1}\otimes\ket{-i})$ & $\ket{i10}+\ket{-i,11}$   & \\
  \hline
  \hline
    $CNOT_{1,3}(\ket{i1}\otimes\ket{0})$ & $\ket{010}+i\ket{111}$   & \multirow{4}{*}{$g_{144}^2$} \\
  \cline{1-2}
  $CNOT_{1,3}(\ket{i1}\otimes\ket{1})$ & $\ket{011}+i\ket{110}$   & \\
  \cline{1-2}
  $CNOT_{1,3}(\ket{i1}\otimes\ket{i})$ & $\ket{01i}-\ket{11,-i}$   &  \\
  \cline{1-2}
  $CNOT_{1,3}(\ket{i1}\otimes\ket{-i})$ & $\ket{01,-i}+\ket{11i}$   & \\
  \hline
  \hline
    $P_3CNOT_{1,3}(\ket{i1}\otimes\ket{0})$ & $\ket{010}-\ket{111}$   & \multirow{4}{*}{$g_{144}^3$} \\
  \cline{1-2}
  $P_3CNOT_{1,3}(\ket{i1}\otimes\ket{1})$ & $\ket{011}+\ket{110}$   & \\
  \cline{1-2}
  $P_3CNOT_{1,3}(\ket{i1}\otimes\ket{i})$ & $\ket{01-}-\ket{11+}$   &  \\
  \cline{1-2}
  $P_3CNOT_{1,3}(\ket{i1}\otimes\ket{-i})$ & $\ket{01+}+\ket{11-}$   & \\
  \hline
  \hline
    $CNOT_{3,2}(\ket{i1}\otimes\ket{+})$ & $\ket{i10}+\ket{i01}$   & \multirow{8}{*}{$g_{288}$} \\
  \cline{1-2}
  $CNOT_{3,2}(\ket{i1}\otimes\ket{-})$ & $\ket{i10}-\ket{i01}$   & \\
  \cline{1-2}
  $CNOT_{3,2}(\ket{i1}\otimes\ket{i})$ & $\ket{i10}+i\ket{i01}$   &  \\
  \cline{1-2}
  $CNOT_{3,2}(\ket{i1}\otimes\ket{-i})$ & $\ket{i10}-\ket{i01}$   & \\
  \cline{1-2}
  $CNOT_{3,2}CNOT_{3,1}(\ket{i1}\otimes\ket{+})$ & $\ket{i10}+i\ket{-i,01}$   & \\
  \cline{1-2}
  $CNOT_{3,2}CNOT_{3,1}(\ket{i1}\otimes\ket{-})$ & $\ket{i10}-i\ket{-i,01}$   &  \\
  \cline{1-2}
  $CNOT_{3,2}CNOT_{3,1}(\ket{i1}\otimes\ket{i})$ & $\ket{i10}-\ket{-i,01}$   & \\
  \cline{1-2}
  $CNOT_{3,2}CNOT_{3,1}(\ket{i1}\otimes\ket{-i})$ & $\ket{i10}+\ket{-i,01}$   & \\
  \hline
\end{tabular}
\caption{Using the 26 starting states in this table, and the path $\mathcal{C}$ \eqref{g36HamiltonianCircuit}, we cover every vertex on the $g_{36}$, $g_{144}$, and $g_{288}$ subgraphs of the $H_1,\, H_2,\, CNOT_{1,2},\, CNOT_{2,1}$ restricted graph at 3 qubits.  The $g_{24}$ subgraphs cannot be covered by a single path, but all states in them can be generated from states in the 2-qubit $g_{24}$ subgraph via the lift described in \eqref{g36Lift}.}
\label{tab:OneFullCoveringOfLiftedStates}
\end{table}

For the $g_{36}$ subgraphs, the lift only involved a simple tensor product with the third qubit. Additionally, the lift actually worked on every state in the $g_{36}$ subgraph, so the structure of $g_{36}$ in each of the $6$ copies is completely preserved. All lifted states from the 2-qubit $g_{36}$ subgraph lift to the same relative position in the 3-qubit $g_{36}$ subgraphs. Therefore, applying $\mathcal{C}$ on each lift of $\ket{i1}$ on a 3-qubit copy $g_{36}$ still builds a Hamiltonian path on each subgraph. Accordingly, lifting $\mathcal{C}$ from two qubits to three qubits by lifting the starting state gives a complete vertex covering of each 3-qubit copy of $g_{36}$, just as in Figure \ref{TwoQubitHamiltonianPath}. This covering of higher qubit $g_{36}$ subgraphs by the 2-qubit $g_{36}$ structure persists to arbitrary qubit number.

Lifting $\mathcal{C}$ to the three copies of $g_{144}$ and the single $g_{288}$ subgraph illustrates how the two-qubit subgraph $g_{36}$ is embedded into larger subgraphs at three qubits. Let us consider a specific example, beginning with the second $g_{144}$ subgraph, which we have termed $g_{144}^2$ in Table \ref{tab:OneFullCoveringOfLiftedStates}. As described in Equation \eqref{CNOT13LiftedState}, the lifted state $CNOT_{1,3}\ket{i10}$ is on this subgraph.  Applying the circuit $\mathcal{C}$ to this state produces a non-intersecting path on $g_{144}$, covering $1/4$ of its states as in Figure \ref{ThreeQubitLiftedHamiltonianPath}.
\begin{figure}[h]
\begin{center}
\begin{overpic}[width=0.8\textwidth]{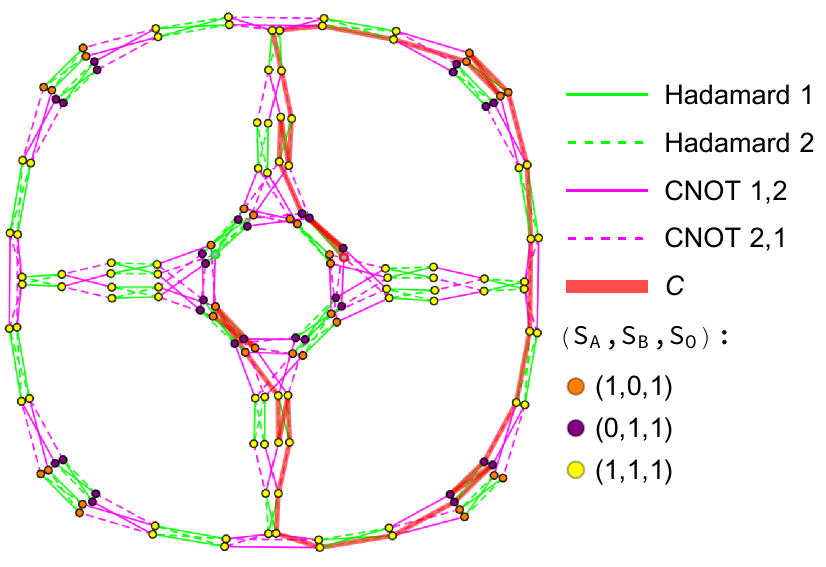}
		\put (27,37.3) {\footnotesize{$\leftarrow \ket{\psi}$}}
		\end{overpic}
\caption{Lift of Hamiltonian path $\mathcal{C}$ defined in Equation \eqref{g36HamiltonianCircuit}, from the 2-qubit $g_{36}$ subgraph to a copy of $g_{144}$ at three qubits. The path begins on $\ket{\psi} \equiv CNOT_{1,3}\ket{i10}$, the state lifted in Equation \eqref{CNOT13LiftedState}.}
\label{ThreeQubitLiftedHamiltonianPath}
\end{center}
\end{figure}

Each of the lifted starting states listed under $g_{144}^2$ in Table \ref{tab:OneFullCoveringOfLiftedStates} have entropy vector $\vec{S} = (1,0,1)$, and are located around the central octagonal structure of $g_{144}$. Applying $\mathcal{C}$ to each of the lifted states defines a separate path, covering $36$ vertices each, on $g_{144}$. The union of these $4$ lifts of $\mathcal{C}$ defines a vertex covering of $g_{144}$, displayed in Figure \ref{CoveringG144}. The other two $g_{144}$ subgraphs can be covered similarly, by four copies of $\mathcal{C}$ starting on each of the lifted starting states listed in the table.
\begin{figure}[h]
\begin{center}
\begin{overpic}[width=0.8\textwidth]{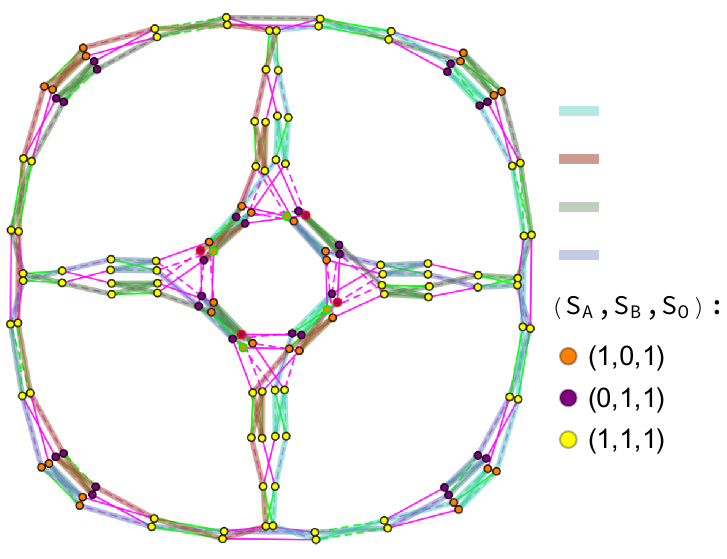}
		\put (84,60.5) {\footnotesize{$CNOT_{1,3}\ket{i,1,0}$}}
		\put (84,54) {\footnotesize{$CNOT_{1,3}\ket{i,1,1}$}}
		\put (84,47.5) {\footnotesize{$CNOT_{1,3}\ket{i,1,i}$}}
		\put (84,40.8) {\footnotesize{$CNOT_{1,3}\ket{i,1,-i}$}}
		\end{overpic}
\caption{Four copies of Hamiltonian path $\mathcal{C}$ that give a vertex cover of $g_{144}^2$. Each path is indicated by the lifted starting state where it begins.}
\label{CoveringG144}
\end{center}
\end{figure}

A covering of the $g_{288}$ subgraph can also be constructed from lifts of $\mathcal{C}$.  First, we lift the 2-qubit $g_{36}$ state $\ket{i1}$ via Equations \eqref{CNOT32LiftTo288} and  \eqref{CNOT32CNOT31LiftTo288}, as listed in Table \ref{tab:OneFullCoveringOfLiftedStates}. We then apply $\mathcal{C}$ on these $8$ lifted states, covering every vertex of $g_{288}$ once.

We have now reached all of the 3-qubit states, and we have covered the six subgraphs $g_{36}$, three subgraphs $g_{144}$, and single subgraph $g_{288}$ with lifted versions of the 2-qubit Hamiltonian path $g_{36}$.

Just as we described at the end of section \ref{2qubitHCNOT} for the 2-qubit $g_{36}$ subgraph, these coverings are not unique.  We could have lifted the 2-qubit state $\ket{-i,1}$ instead, since $\mathcal{C}$ is also a Hamiltonian path from that starting point.  Alternatively, we could have used the path $\mathcal{C}^\top$ or $\mathcal{C}^{-1}$ instead, or any other Hamiltonian path from the 2-qubit $g_{36}$ subgraph.

Sticking to just $\mathcal{C}$, $\mathcal{C}^\top$, $\mathcal{C}^{-1}$, and $(\mathcal{C}^\top)^{-1}$, each path has two possible starting states on the 2-qubit $g_{36}$ subgraph.  Lifting these eight states to four starting states each on $g_{144}^2$, we hit all 32 purple and orange states in the central octagon in Figure \ref{CoveringG144}.  The same statement works for the other $g_{144}$ subgraphs, so every central octagon state is a lifted starting state for some 2-qubit Hamiltonian path, which together with 3 other states, fully covers the $g_{144}$ subgraph. A similar situation arises for the $g_{288}$ subgraph.

Rather than exploring these other possible coverings, we will instead move on to further structures which first arise at three qubits: Hamiltonian cycles and Eulerian paths.

\subsection{Hamiltonian cycles and Eulerian paths}\label{sec:3-qubitHamiltonianEulerian}

At one qubit, the complete reachability diagram did not have a Hamiltonian path.  At two qubits, the $H_1,\, H_2, \, CNOT_{1,2},\, CNOT_{2,1}$ restricted subgraph $g_{24}$ also did not have a Hamiltonian path, but $g_{36}$ did.
Beginning at three qubits, we are able to construct Hamiltonian cycles as well, on subgraphs $g_{144}$ and $g_{288}$.  Subgraph $g_{288}$ will additionally have an Eulerian cycle, which traverses every path exactly once (instead of visiting every vertex once).

Hamiltonian cycles cannot be built on the $g_{24}$ and $g_{36}$ subgraphs.  For $g_{24}$ this check is simple: it contains a vertex whose removal disconnects the graph, known as a cut-vertex; the $GHZ$ state in Figure \ref{TwoQubitH1H2CNOT12CNOT21} is one example. Subgraph $g_{36}$ contains no cut-vertex; however, it can be verified to have no Hamiltonian cycle by exhaustive search.%
\footnote{Any graph containing a Hamiltonian cycle must satisfy all of the following conditions: each vertex of degree $2$ must be included in the Hamiltonian cycle, once two edges incident to a vertex are included in the Hamiltonian cycle all other edges incident to that vertex must be removed from consideration, and the Hamiltonian cycle must contain no proper subcycles. These criteria are only sufficient to eliminate a graph's candidacy for having a Hamiltonian cycle. In general, the ``Hamiltonian path problem'' of proving that a given graph does admit a Hamiltonian path or cycle is NP-complete.}%

For the $g_{144}$ subgraph, the four copies of $\mathcal{C}$ in each graph do not directly connect into a Hamiltonian cycle simply by connecting their ends.  However, other Hamiltonian paths exist on the 2-qubit $g_{36}$ subgraph, and some of these paths, when their starting states are lifted to four states on $g_{144}$, can be directly connected into one large loop that builds a Hamiltonian cycle.

For the $g_{288}$ subgraph, we have three coverings.  First, as in the previous section, eight copies of $\cal{C}$ (or any other Hamiltonian path from $g_{36}$ ) completely cover the graph.  Next, the Hamiltonian cycle on $g_{144}$ can be mapped to a pair of cycles which each cover half of $g_{288}$, as shown in Figure \ref{HamiltonianCycleMap}.  And last, a new Hamiltonian cycle exists on $g_{288}$ itself.
\begin{figure}[h]
\begin{center}
\includegraphics[scale=.9]{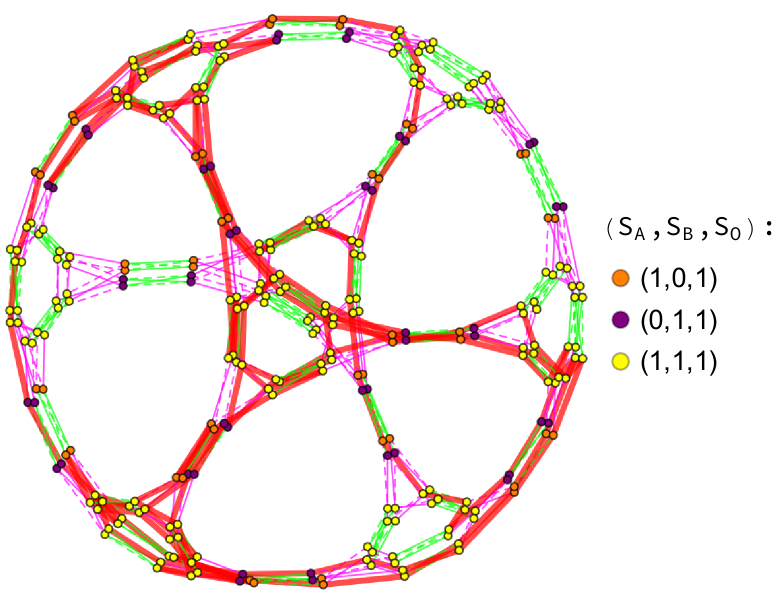}
\caption{The Hamiltonian cycle on $g_{144}$ can be mapped to a cycle on $g_{288}$ through a CNOT operation connecting the two subgraphs. Two copies of $g_{144}$ can be used to cover $g_{288}$ completely.}
\label{HamiltonianCycleMap}
\end{center}
\end{figure}

We also note that Eulerian cycles, which traverse every edge, first exist on $g_{288}$.  The cycle is of length 576, and it must exist because all vertices have an even number of attached edges (namely, four).  Stated another way, all four gates in the restricted set $H_1,\, H_2, \, CNOT_{1,2},\, CNOT_{2,1}$ act nontrivially and differently from each other, on every state in the subgraph. In particular, this means that $g_{288}$ has no trivial loops.

Hamiltonian cycles can be embedded into to higher-qubit subgraphs in almost exactly the same way as Hamiltonian paths; for cycles we have to pick an arbitrary starting state to lift. Each Hamiltonian cycle embeds into every subgraph of greater or equal vertex number. We will also see that embeddings of Hamiltonian paths and cycles will continue to cover subgraphs at higher qubit number.

\section{Towards the $n$-qubit Stabilizer Graph}\label{sec:general}

We will now generalize our discussion to the cases of four and five qubits.  Two particularly important features arise first at four qubits. First, the last new subgraph shape in the $H_1,\, H_2, \, CNOT_{1,2}, \, CNOT_{2,1}$ restricted graph appears.  Second, some entropy vectors of stabilizer states will disobey monogamy of mutual information \eqref{eq:MonogamyofMutualInformation}.  These states thus lie outside the holographic entropy cone; they cannot have a classical gravitational representation.

\subsection{Four Qubits}

In order to explore four qubits, we have explicitly generated all of the 36720 4-qubit stabilizer states. The explicit form of these states is available in the GitHub repository \cite{github} along with associated reachability diagrams, entropy vector data, and a Mathematica package for simulating stabilizer circuits, generating sets of states, and constructing reachability diagrams.

Examining the $H_1,\, H_2, \, CNOT_{1,2}, \, CNOT_{2,1}$ restricted graphs, we find copies of the $g_{24}$ and $g_{36}$ subgraphs we have seen before. We also see again the $g_{144}$ subgraph with 3 different entropy vectors hereafter referred to as $g_{144}(3)$, and the $g_{288}$ with 3 entropy vectors, now termed $g_{288}(3)$.  We specify the number of entropy vectors because we also find $g_{144}(4)$ and $g_{288}(4)$, as in Figures \ref{FourQubitG144} and \ref{FourQubitG288FourEntropies}.

\begin{figure}[h]
\begin{center}
\includegraphics[scale=0.75]{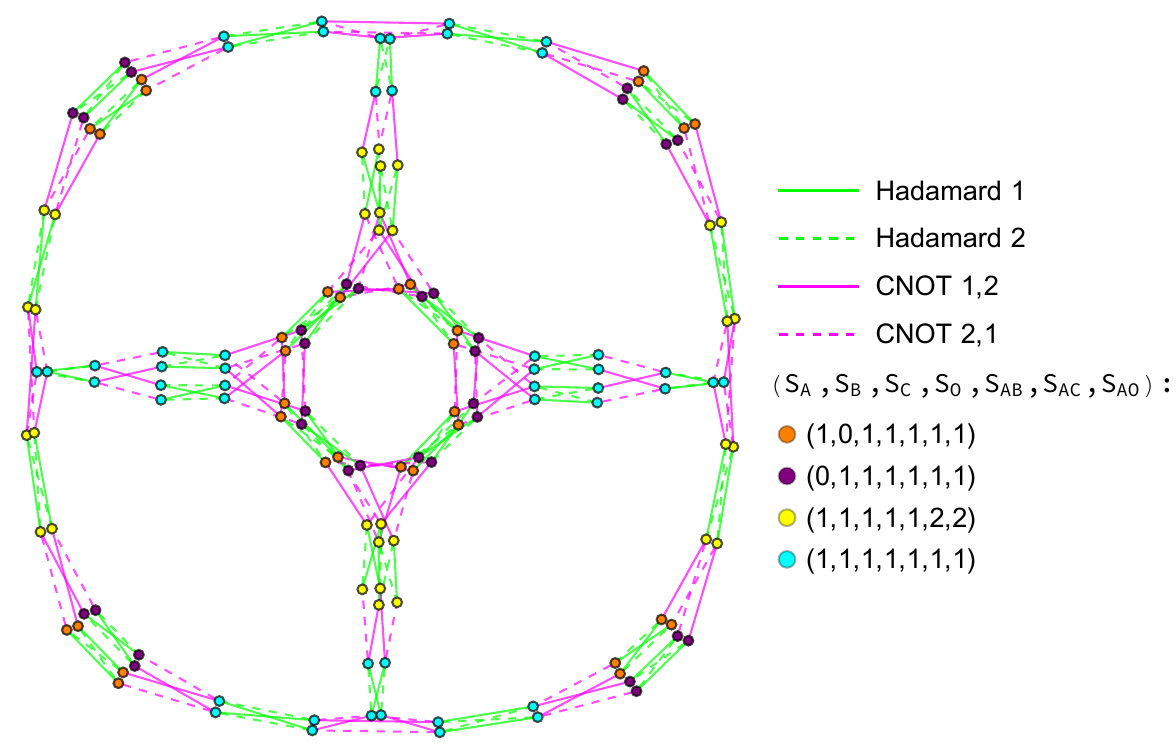}
\caption{At four qubits, the $g_{144}$ subgraph can have $4$ different entropy vectors among the vertices. The cyan vertices denote states with entropy vector $\vec{S} = (1,1,1,1,1,1,1)$, an entropic structure that violates monogamy of mutual information \eqref{eq:MonogamyofMutualInformation}. These 4-qubit stabilizer states, located on 4-qubit subgraphs $g_{144}$ and $g_{288}$, are the first instances of non-holographic states.}
\label{FourQubitG144}
\end{center}
\end{figure}

\begin{figure}[h]
\begin{center}
\includegraphics[scale=.9]{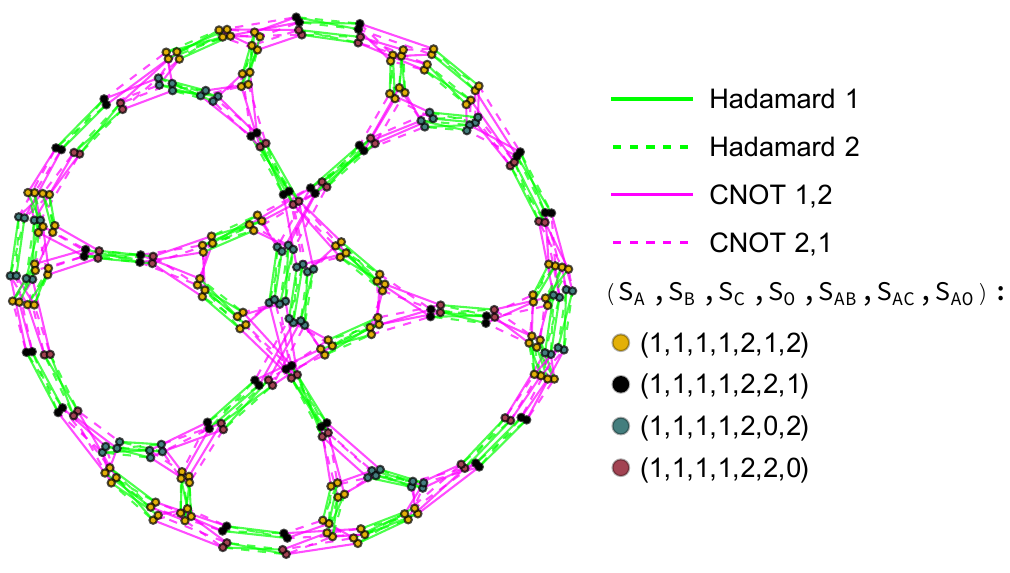}
\caption{At four qubits, we witness a variation to the $g_{288}$ subgraph first seen at three qubits. The structure is isomorphic to its 3-qubit counterpart; however, the number of entropy vectors present has increased to four, similar to what occurred in Figure \ref{FourQubitG144}.}
\label{FourQubitG288FourEntropies}
\end{center}
\end{figure}

We also obtain our last new structure: $g_{1152}$.  At four qubits we find this subgraph with either two or four different entropy vectors, as in Figures \ref{FourQubitG1152TwoEntropies} and \ref{FourQubitG1152FourEntropies}. 

\begin{figure}[h]
\begin{center}
\includegraphics[scale=.87]{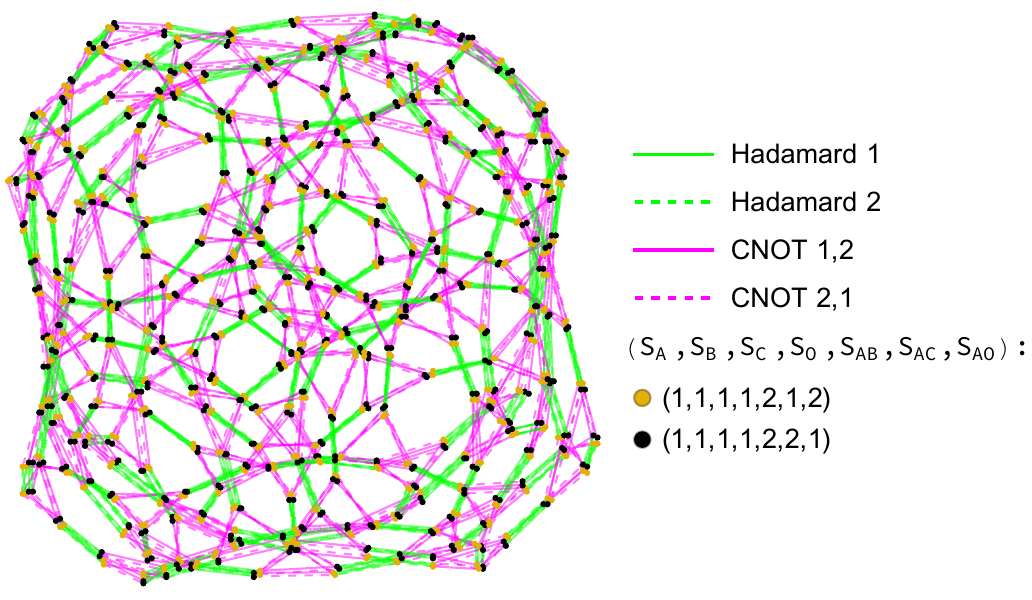}
\caption{One of the two $g_{1152}$ occurring at four qubits. This subgraph is isomorphic to the $g_{1152}$ subgraph in Figure \ref{FourQubitG1152FourEntropies}, but with only $2$ different entropy vectors among all the states in the subgraph.}
\label{FourQubitG1152TwoEntropies}
\end{center}
\end{figure}

\begin{figure}[h]
\begin{center}
\includegraphics[scale=0.9]{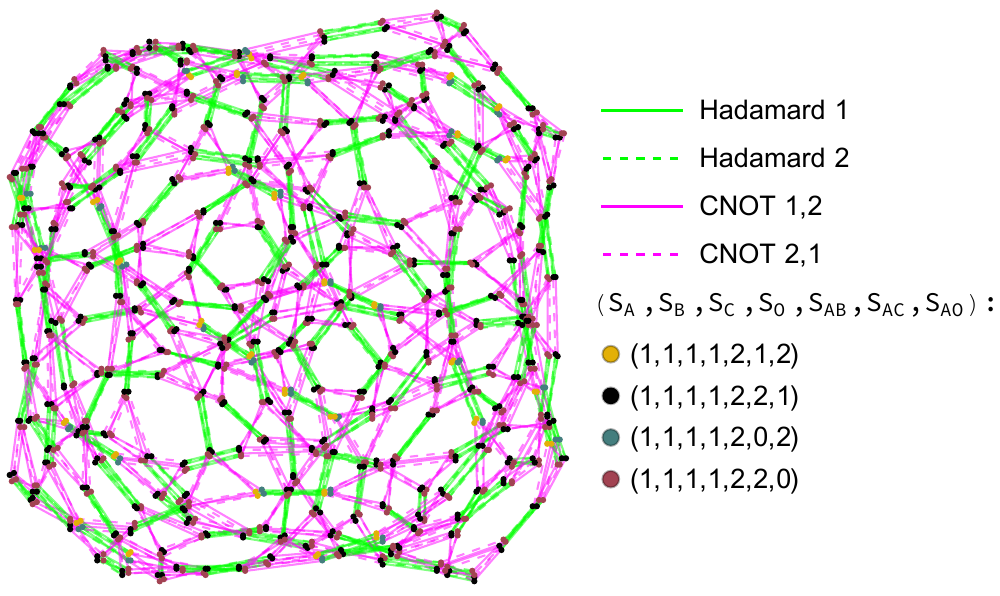}
\caption{Another copy of the $g_{1152}$ subgraph introduced at $4$ qubits. States located on this subgraph exhibit $4$ different entropy vectors, in contrast to the $g_{1152}$ subgraph in Figure \ref{FourQubitG1152TwoEntropies}.}
\label{FourQubitG1152FourEntropies}
\end{center}
\end{figure}

As in the 3-qubit case, lifting the $g_{24}$ and $g_{36}$ subgraphs from two qubits to four qubits proceeds in a simple manner. To lift the 2-qubit $g_{24}$ states to the 4-qubit $g_{24}$ subgraphs, we proceed as in \eqref{g24Lift}.  That is, we start with a state in the 2-qubit $g_{24}$ subgraph, and then tensor on any other 2-qubit state. Thus, both $\ket{0000}$ and $\ket{00i1}$ appear in copies of $g_{24}$ at four qubits. 

Similarly for $g_{36}$, we again tensor on 2-qubit states instead of the 1-qubit state added in \eqref{g36Lift}.  Thus $\ket{i100}$ and $\ket{i1i1}$ are in copies of $g_{24}$ at four qubits.  
Since there are 60 2-qubit stabilizer states, we get 60 copies each of $g_{24}$ and $g_{36}$.  For $g_{36}$ the Hamiltonian paths transfer just as before.

These Hamiltonian paths can be embedded in $g_{144}$ and $g_{288}$ by acting with the CNOT gates that involve qubits 3 and/or 4, or with phase gates when necessary.  The new structure, $g_{1152}$, behaves similarly, except 32 copies of $\mathcal{C}$ are needed for a full covering.

Not only the Hamiltonian paths lift; recall that Hamiltonian cycles exist for $g_{144}$ and $g_{288}$.  These cycles embed as well.  Choosing some starting state at 3 qubits, we tensor on a 1-qubit stabilizer state, act with a CNOT to entangle it and arrive on one of the ten copies of $g_{1152}$, and then enact the cycle in the same gate order as the Hamiltonian cycle at three qubits.  In order to cover one full $g_{1152}$ subgraph, four such cycles are needed.

Lastly, the $g_{1152}$ structure, like the $g_{288}$ subgraph, contains an Eulerian cycle of length 2304.  Just as for $g_{288}$, as discussed at the end of section \ref{sec:3-qubitHamiltonianEulerian}, this cycle exists because all four of the gates $H_1,\, H_2, \, CNOT_{1,2}, \, CNOT_{2,1}$ act nontrivially and differently from each other, on every state in the subgraph; no degeneracies or loops remain.

\subsubsection{Entropy Vector Analysis}

In Table \ref{tab:FourQubitEntropyVectors}, we list the 18 different entropy vectors which stabilizer states exhibit at four qubits.  In the table, the vectors are grouped together if they appear in the same subgraphs in the $H_1,\, H_2, \, CNOT_{1,2}, \, CNOT_{2,1}$ restricted graph.  These groupings arise because $S_C, \, S_O, $ and $S_{AB}$ cannot be altered by the Hadamards and CNOTs on qubits 1 and 2 only.
\begin{table}[h]
    \centering
    \small{
    \begin{tabular}{|c||c|c|c|} 
 \hline
Holographic & $\left(S_A,S_B,S_C,S_O,S_{AB},S_{AC},S_{AO}\right)$ & Number of States & Subgraph\\
\hline
\hline
 Yes & $(0,0,0,0,0,0,0)$ & $1296$ & $g_{24}$, $g_{36}$\\
 \hline
 Yes & $(1,1,0,0,0,1,1)$ & $864$ & $g_{24}$, $g_{36}$\\
 \hline
 \hline
 Yes & $(0,0,1,1,0,1,1)$ & $864$ & $g_{24}$, $g_{36}$\\
 \hline
 Yes & $(1,1,1,1,0,2,2)$ & $576$ & $g_{24}$, $g_{36}$\\
 \hline
 \hline
 Yes & $(0,1,1,0,1,1,0)$ & $864$ & $g_{144}, g_{288}$\\
 \hline
 Yes & $(1,0,1,0,1,0,1)$ & $864$ & $g_{144}, g_{288}$\\
 \hline
 Yes & $(1,1,1,0,1,1,1)$ & $2592$ & $g_{144}, g_{288}$\\
 \hline
 \hline
 Yes & $(0,1,0,1,1,0,1)$ & $864$ & $g_{144}, g_{288}$\\
 \hline
 Yes & $(1,0,0,1,1,1,0)$ & $864$ & $g_{144}, g_{288}$\\
 \hline
  Yes & $(1,1,0,1,1,1,1)$ & $2592$ & $g_{144}, g_{288}$\\
 \hline
 \hline
 Yes & $\textcolor{orange}{\bullet}(1,0,1,1,1,1,1)$ & $2592$ & $g_{144}, g_{288}$\\
 \hline
 Yes & $\textcolor{violet}{\bullet}(0,1,1,1,1,1,1)$ & $2592$ & $g_{144}, g_{288}$\\
 \hline
  Yes & $\textcolor{yellow}{\bullet}(1,1,1,1,1,2,2)$ & $5184$ & $g_{144}, g_{288}$\\
 \hline
 \color{red} No & $\textcolor{cyan}{\bullet}(1,1,1,1,1,1,1)$ & $2592$ & $g_{144}, g_{288}$\\
 \hline
 \hline
 Yes & $\textcolor{lime}{\bullet}(1,1,1,1,2,1,2)$ & $5184$ & $g_{1152}$ (2,4)\\
 \hline
 Yes & $\textcolor{black}{\bullet}(1,1,1,1,2,2,1)$ & $5184$ & $g_{1152}$ (2,4)\\
 \hline
 Yes & $\textcolor{teal}{\bullet}(1,1,1,1,2,0,2)$ & $576$ & $g_{1152}$ (4)\\
 \hline
 Yes & $\textcolor{purple}{\bullet}(1,1,1,1,2,2,0)$ & $576$ & $g_{1152}$ (4)\\
 \hline
\end{tabular}}
    \caption{Entropy vectors for the set of four-qubit stabilizer states. One entropy vector disobeys the holographic inequalities. Graph $g_{1152}$ comes in two varieties: one with only two different entropy vectors, and one with four.  The last two entropy vectors only appear on subgraphs with four different entropy vectors, while the previous two appear on subgraphs with either two or four entropy vectors. }
    \label{tab:FourQubitEntropyVectors}
\end{table}

Importantly, notice that our first non-holographic state, which disobeys the holographic entropy cone relation \eqref{eq:MonogamyofMutualInformation}, arises here.  CNOT gates are the only actions which could move a state from within the entropy cone to outside of it, so by studying the restricted graphs where this non-holographic entropy vector lives, we can define how to reach such states. In particular, notice that in $g_{144}(4)$, shown in Figure \ref{HamiltonianCycleMap}, the non-holographic states divide the subgraph into multiple parts; removing them entirely would separate the inner octagonal structure from the outside, and divide the outer ring into several pieces.

\subsection{Five Qubits and Beyond}

At five qubits, no new structures arise. The subgraph types may begin to have more entropy vectors on a given subgraph, and each subgraph has more copies, but the shapes themselves do not increase in size. Again, we have generated the full set of stabilizer states as well as the complete reachability diagram, and this data is fully accessible at \cite{github}. The full five qubit list of entropy vectors is in Table \ref{tab:5QubitEntropiesPart1} of the appendix.

Again, we note 12 entropy vectors arise which disobey the holographic entropy cone relations; as usual, examining the gate actions which move to and from those entropy vectors will help us to learn how states move in and out of the holographic entropy cone.

For now, we note that the $g_{24}$ and $g_{36}$ subgraphs each have 1080 copies, as expected since they are built by tensoring 2-qubit stabilizer states with 3-qubit stabilizer states; there are 1080 3-qubit stabilizer states.  We also note that the number of $g_{1152}$ copies scales faster than for the smaller index number graphs, as shown in Table \ref{tab:NumberOfSubgraphs}.
\begin{table}[h]
    \centering
    \begin{tabular}{|c||c|c|c|c|c|}
    \hline
    Qubit \# & $g_{24}$ & $g_{36}$ & $g_{144}$ & $g_{288}$ & $g_{1152}$\\
    \hline
    \hline
    Two & $1$ & $1$ & - & - & -\\
    \hline
    Three & $6$ & $6$ & $3$ & $1$ & -\\
    \hline
    Four & $60$ & $60$ & $90$ & $30$ & $10$\\
    \hline
    Five & $1080$ & $1080$ & $3780$ & $1260$ & $1260$\\
    \hline
    \end{tabular}
\caption{The number of occurrences of each subgraph structure at two, three, four, and five qubits.  The number of copies of $g_{1152}$ exhibits the most dramatic growth with qubit number.}
\label{tab:NumberOfSubgraphs}
\end{table}

We believe no further new structures arise at higher qubits beyond the $g_{1152}$ first found at four qubits.%
\footnote{\label{NoHigherStructures} This claim has since been verified explicitly by directly computing the full set of unique Clifford group elements generated by $H_1, H_2, CNOT_{1,2},$ and $CNOT_{2,1}$. A simple justification can be observed by allowing each leg of the $g_{1152}$ subgraph to represent a unique operation generated by $H_1, H_2, CNOT_{1,2}$ and $CNOT_{2,1}$. The proof of this result will be presented in forthcoming work.} %
However, individual subgraphs will start to contain larger numbers of different entropy vectors, just as we found when increasing to four and then five qubits.

\section{Discussion}\label{sec:discussion}

Our goal in the body of the paper has been to understand the $n$-qubit stabilizer states by first constructing their reachability graphs, then passing to the restricted graphs formed using only the Hadamards and CNOTs on two qubits. We were able to understand how the simpler structures at low qubit number ($n=1$ and $n=2$) assemble, via lifts, to build the more complicated structures observed at higher qubit number ($n=3$ and $n=4$). We have argued that the four graph structures shown already in the Introduction are ``all there is:'' passing to higher qubit number merely proliferates larger numbers of the four basic structures, perhaps with more complicated entropic arrangements. We noted the presence of non-holographic states beginning at four qubits, and their confinement to particular groups of subgraphs, as demanded by their having entropy vectors that violate holographic inequalities. In this Discussion, we collect some other ways of thinking about the reachability graphs we have constructed, along with potential generalizations and questions for further research.

We first note that, to our knowledge, already at two qubits this is the first time that the complete, explicit reachability graphs have been presented in the literature. We have constructed the full reachability graphs up through five qubits, and have made them available in our GitHub repository \cite{github}.

Having the explicit form of the graphs allows for direct computation of some interesting quantities. For example, the minimum distance on the graph from one state to another is, explicitly, the gate complexity of the stabilizer circuit which maps the first state to the second. The state with largest minimum distance relative to a reference state thus is the ``most complicated'' state by this measure. We have explicitly computed this minimum distance relative to the all-zero state for $n=1\ldots4$, and find that it is circuit distance $4,7,10,13$, respectively. Since the asymptotic complexity of $n$-qubit stabilizer states is known to go as $n^2/\mathrm{log}(n)$ \cite{aaronson2004improved}, this is perhaps a surprising result---it indicates that at small, finite qubit number the main contribution to the complexity is not the CNOT gate applications which lead to the asymptotic result. Of course, we note that, even with the explicit reachability graph, finding a state with maximal relative distance is a highly nontrivial problem requiring an exhaustive numerical search.

We recall also that the group-theoretic identity in \eqref{eq:orbit_stabilizer} equates the number of elements in a given subgraph to the orbit length of an element in that subgraph under the action of the subgroup of the Clifford group generated by the subset of the Clifford gates we are considering (see Table \ref{tab:OneQubitSubgroupOrbits}). Hence we can interpret $24, 36, 144, 288, 1152$ as the orbit lengths of various stabilizer states under the action of the subgroup generated by $\{H_1,H_2,CNOT_{1,2},CNOT_{2,1}\}$. It would be interesting to find these orbits directly using a group-theoretic approach, and confirm or disprove our conjecture that no additional structures appear beyond $g_{1152}$.

We emphasize that we chose to focus primarily on Clifford gates acting on the first two qubits because the stabilizer states exhibit only bipartite entanglement, constructed explicitly through CNOT gate applications. Hence our restricted graphs are sufficient to analyze how the entropy vector of a state changes as it is acted upon by various stabilizer circuits. Of course, adding additional gates to our restricted gate will result in more complicated connected subgraphs, culminating in the full reachability graph at each qubit number.

In the holographic literature, we typically think of each subsystem as some portion of a spatial boundary, with the exception of an additional ``purifier'' subsystem which does not necessarily have an interpretation as a spatial subsystem. The holographic interpretation of our results will therefore differ depending on whether we take the purifier to be one of the first two qubits, or instead some other qubit $>=$ 3. In the body of the paper we have taken the purifying system $O$ to be the last qubit.

In what ways can our reachability graph be generalized? First consider staying within the setting of $n$-qubit states. The existence of a graph structure is most useful when we have a gate set that can be used to pick out only a finite number of states, rather than a continuous subspace of $\mathbb{C}^{2n}$ or the entire Hilbert space itself. Hence we do not expect choosing a modification of the Clifford gates which yields a universal gate set to provide interesting results. We could instead consider the stabilizers of some other finite group of operators besides the Pauli group. It is an interesting question whether there are any such finite groups with a compelling physical interpretation. We could also consider allowing a small number of discrete applications of non-Clifford gates, i.e. allowing a small amount of ``magic'' \cite{bravyi2005universal,White:2020zoz} as a resource. Or, we could consider restricting ourselves to stabilizer states, but allowing some applications of operators which are in the Clifford group but not Clifford generators themselves, which could be used to ``fast-forward'' a circuit. Finally, it would be very interesting to understand whether there is a set of gates that produce a discrete set of states but allow, for example, for tripartite entanglement.

In more general Hilbert spaces which are not isomorphic to tensor products of qubits, we expect that a similar story should hold with a different group than the Pauli group. There are a number of approaches in the literature to generalizing Pauli matrices to larger discrete systems. One approach, well-suited to qudit systems, is to use Gell-Mann matrices and their generalizations. Another, quite different, approach better suited to approaching the continuum limit of a lattice system is to use the ``clock'' and ``shift'' matrices that generate a ``generalized Clifford algebra'' \cite{Jagannathan:2010sb,Pollack:2018yum}.

Finally, we recall our motivation in the Introduction, where we described both entanglement entropies in a factorized system, and stabilizer states relative to a preferred group of algebra, as two different ways of imposing structure on the set of states in Hilbert space. We can further generalize this picture by noting that the von Neumann entropy of a reduced density matrix of a subsystem can be identified with the \emph{algebraic} entropy associated with the algebra of operators which act as the identity everywhere outside that subsystem. Hence entropies are defined relative to an \emph{algebra} of operators while stabilizer states are defined relative to a \emph{group} of operators. In particular, recall that the algebra generated by the Pauli operators acting on a qubit is in fact the full algebra of linear operators on that qubit. 

When we consider more general von Neumann algebras, instead of the Hilbert space decomposing as a product of tensor factors, we get a more general \emph{Wedderburn decomposition} \cite{wedderburn1934lectures,Harlow:2016vwg,Kabernik:2019jko}, which decomposes the Hilbert space as a \emph{sum} of products of tensor factors, where each term in the sum represents a superselection sector which has its own associated entropy. Thus we can envision a very abstract version of our picture in which we fix a \emph{set} of operators, and then find stabilizer states by using this set to generate a \emph{group} of operators and entropies by using the same set to generate an \emph{algebra} of operators. It would be interesting to see how far this picture can be taken in generic cases. A first step would be to understand how to generate not just a single entropy for a given reduced density matrix on a state, but rather a larger entropy vector, which would pick out some \emph{sequence} of algebras, the equivalent of the operators acting on the first, second, and in general $n^{th}$ qubit.

\paragraph{Data Repository} The full set of states, reachability diagrams, and entropy vectors for stabilizer states at $n \leq 5$ qubits can be accessed via the GitHub repository \cite{github}. The repository additionally includes Mathematica notebooks used to generate the data for this paper, as well as a Stabilizer State package designed to generate reachability graphs and analyze stabilizer state structure.

\textit{The authors thank Scott Aaronson, Ning Bao, Charles Cao, and Sergio Hernandez-Cuenca for helpful discussions. J.P. is supported by the Simons Foundation through \emph{It from Qubit: Simons Collaboration on Quantum Fields, Gravity, and Information}. C.K. and W.M. are supported by the U.S. Department of Energy under grant number DE-SC0019470 and by the Heising-Simons Foundation “Observational Signatures
of Quantum Gravity” collaboration grant 2021-2818.}

\clearpage
\begin{singlespace}
\printbibliography[heading=subbibliography]
\end{singlespace}

\chapter{CLIFFORD ORBITS FROM CAYLEY GRAPH QUOTIENTS}\label{Chapter4}

\textit{The contents of this chapter were originally published in Physical Review A \cite{Keeler:2023xcx}.}

\textit{We describe the structure of the $n$-qubit Clifford group $\mathcal{C}_n$ via Cayley graphs, whose vertices represent group elements and edges represent generators. In order to obtain the action of Clifford gates on a given quantum state, we introduce a quotient procedure. Quotienting the Cayley graph by the stabilizer subgroup of a state gives a reduced graph which depicts the state's Clifford orbit. Using this protocol for $\mathcal{C}_2$, we reproduce and generalize the reachability graphs introduced in \cite{Keeler2022}. Since the procedure is state-independent, we extend our study to non-stabilizer states, including the W and Dicke states. Our new construction provides a more precise understanding of state evolution under Clifford circuit action.}

\section{Introduction}\label{Intro}

How can we track the evolution of information about a quantum state? 
In the general case, where we'd like to obtain the outcome of arbitrary measurements under continuous time evolution, we can do no better than to work with the state itself. 
For an $n$-qubit pure state $\ket{\psi}\in \Hil \cong \mathbb{C}^{2n}$, tracking the state requires knowing all of its overlaps $\braket{a_i|\psi}$ with a given orthonormal basis $\{\ket{a_i}\}$: that is, $4^n-2$ real parameters, accounting for normalization and the unobservability of global phase.
We can do better, however, by restricting which information we want to keep track of, or restricting how the state might evolve, or starting with a special initial state:
\begin{itemize}
    \item We might only care about some particular properties of the state. 
    If we only want to predict the outcome of measurements on $k<n$ of the qubits, we can trace out the remaining $n-k$ qubits and work with the reduced state $\rho_{\{n-k\}}$, which requires only $4^k-1$ real parameters. 
    If we want to understand the entanglement properties of the state, we can collect together the von Neumann entropies of each independent reduced density matrix: $2^{n-1}-1$ real parameters.
    \item We might want to evolve the state through a quantum circuit with a limited set of possible unitaries that can be applied, such as the Clifford gates: Hadamard, phase, and CNOT. 
    It might be that the multiplicative group spanned by the gate set is all unitaries acting on $\mathbb{C}^{2n}$, in which case the gate set is universal. 
    But we might instead find that the group acts on a smaller Hilbert space, for example if every gate in the gate set conserves some charge. 
    Or, as is the case for the Clifford gates \cite{Calderbank:1996hm,Gottesman:1997qd}, the generated group might be \emph{finite}, in which case, for a given initial state, there are only a discrete set of possible states which can be reached.
    \item We might have a special state that admits a reduced description. 
    If we know our state has decohered, we can write it as a superposition of pointer basis states, and classical observables are independent of the relative phase between branches, requiring only $2^n-1$ real parameters. 
    Or, if we know our state is an eigenstate of some specified observable, or a simultaneous eigenstate of a group of observables, we can obtain a compact description.
    For example, the Pauli group on $n$ qubits comprises the $2\cdot4^n$ (signed) Pauli strings. 
    All $n$-qubit states stabilize, i.e.\ are unit eigenvectors of, the identity operator $\mathbb{1}$ (and none stabilize $-\mathbb{1}$). 
    But only a discrete set of states stabilize any additional Pauli strings: a special set of states that can be specified by discrete rather than continuous information.
\end{itemize}

One of the most famous results in the theory of quantum computation concerns a computational setting where we allow ourselves such simplifications. Quantum circuits which take an initial stabilizer state\footnote{This paper, unfortunately, will have to deal with two meanings of the word `stabilizer': the group-theoretic meaning, where a state stabilizes a group element if the group element acts trivially on the state, and the quantum-information-theoretic, where it is standard to refer to those $n$-qubit states which stabilize $2^n$ elements of the $n$-qubit Pauli group simply as ``stabilizer states''. 
We will have cause in this paper to refer to \emph{both} the traditional stabilizer states \emph{and} states which stabilize elements of other groups, most notably the $n$-qubit Clifford group and its subgroups.}, a simultaneous unit eigenvector of $2^n$ Pauli strings, to any other stabilizer state can be represented as ``Clifford circuits'', which contain only Clifford gate applications. 
Unlike circuits made from a universal quantum gate set, Clifford circuits are efficiently classically simulable \cite{Gottesman:1997zz,Gottesman:1998hu,aaronson2004improved}. 

In this paper, we exploit the finiteness of the Clifford group to reinterpret Clifford circuits graphically. 
Our key tool is a group-theoretic notion: the Cayley graph \cite{10.2307/2369306}, which, given a choice of generators, graphically encodes the structure of the group. We're interested in studying what states can be reached if, instead of acting with \emph{arbitrary} Clifford circuits, we restrict to only a subset of the possible Clifford gates. 
In particular, following our previous paper \cite{Keeler2022}, we'd like to understand how the entanglement entropy evolves as we act on a state with a Clifford circuit. 
Because all of the entanglement created in this way is bipartite---the result of a $CNOT$ gate action---it suffices to consider entangling operations acting on only $2$ of the $n$ qubits. 

We are thus led to consider the actions of the $2$-qubit Clifford group on $n\ge 2$-qubit states, which might themselves be stabilizer states, or might be more general.
We previously considered a version of this problem in \cite{Keeler2022}, developing ``reachability graphs'' in which each vertex was a stabilizer state and each edge a Clifford gate, and ``restricted graphs'' in which only certain types of edges were allowed.
We found that the complicated graphs encoding the action of Clifford circuits on stabilizer states decomposed into highly structured subgraphs.

In this paper, using the technology of Cayley graphs, we are able to reproduce and generalize our previous results. 
Rather than working with the action of Clifford gates on states, we are instead led to work more directly with the abstract group elements themselves. We can then recover the action on states via a quotienting procedure. 
By working group-theoretically, we can easily understand the full diversity of subgraph structures that arise, as well as extend our results to the action of Clifford circuits on other states --- for example, states which stabilize a non-maximal number of Pauli group elements --- which allows for new structures to arise. 
Along the way, we will gain a better understanding of the Clifford group itself, deriving a formal presentation for the group, as well as data on its subgroup structure. 
While a presentation was previously found in the literature \cite{Selinger2013}, our reformulation gives insight into the circumstances in which seemingly entangling gate operations fail to ultimately produce entanglement. 

In forthcoming work \cite{Keeler2023b} we will use the relations of our presentation to examine and bound the dynamics of entanglement entropy. Because holographic states, which have a classical geometric description, live inside the stabilizer entropy cone \cite{Bao:2015bfa}, tracking the evolution of entanglement can give us insight into the operations which move quantum states into and out of the holographic cone. 

\subsection{Summary of Results}

We previously introduced, in \cite{Keeler2022}, \emph{reachability graphs}, in which each vertex is an $n$-qubit pure quantum state, typically a stabilizer state, and each (directed) edge is a Clifford gate taking the state at the initial vertex to the state at the final vertex, as well as \emph{restricted graphs}, in which only some subset of the Clifford gates, typically $\{H_1,\, H_2,\, C_{1,2},\, C_{2,1}\}$, is allowed. (Here and throughout we abbreviate $CNOT_{i,j}$ as $C_{i,j}$.)
Reachability graphs graphically encode the result of performing Clifford circuits on a given set of states; restricted graphs give a more refined picture which is often more useful for understanding entropic evolution.  

The main task accomplished in this paper is the reinterpretation of reachability graphs as certain quotients of a group-theoretic object, \emph{the Cayley graph}, a directed graph that encodes the structure of a group by identifying a vertex for every group element and a set of edges for each group generator.%
\footnote{Because the Cayley graph depends on a choice of generators, there are many different Cayley graphs which each correspond to a given group. Each of these graphs has an isomorphic set of vertices, namely one vertex for every group element, but in general inequivalent edges. For example, $\{H_1,\, H_2,\, C_{1,2},\, C_{2,1}\}$ and $\{H_1,\, H_2,\, C_{1,2}\}$ both generate the two-qubit Clifford group; as seen in Table \ref{tab:OrbitLengthCliffordSubgroupNoRelations}, the corresponding Cayley graphs have the same number of vertices, $2304$, but different properties such as graph diameter.} %
Since finite groups have finite Cayley graphs, we can use quotients of the group to construct quotient spaces on its Cayley graph.

The general protocol for constructing the quotient map that yields (graphs isomorphic to) reachability graphs from Cayley graphs is:
\begin{enumerate}
    \item For a group $G$ and chosen state $\ket{\Psi}$, we first identify the stabilizer subgroup of $\ket{\Psi}$ in $G$, denoted $\Stab_G(\ket{\Psi})$.
    \item Since $\Stab_G(\ket{\Psi})$ is a normal subgroup of $G$, all equivalence classes of the quotient group $G/\Stab_G(\ket{\Psi})$ can be generated by taking the left-cosets $h\cdot\Stab_G(\ket{\Psi})$ for all $h \in G$.
    \item Each equivalence class of $G/\Stab_G(\ket{\Psi})$ is assigned a vertex, and each vertex is connected by the generator which maps each element of one equivalence class to exactly one element of the other equivalence class.
\end{enumerate}

This procedure takes as input a choice of group and a choice of state.
To recover the reachability graphs, we do not take $G$ to be the Clifford group itself, because the group contains elements which act as an (unobservable) global phase. 
Instead, we first quotient the Clifford group (or $\HC$, the group generated by Hadamard and CNOT acting on the first two qubits) by such elements, resulting in a smaller group whose elements are isomorphic to equivalence classes of the original group; only then do we specify a state. 
Figure \ref{SummaryImagePaper} illustrates this procedure, starting with the one-qubit Clifford group and the chosen state $\ket{0}$ and producing a quotiented graph isomorphic to the one-qubit stabilizer reachability graph. 
We could have started with any of the six one-qubit stabilizer states and gotten the same result, but the precise mapping of initial group elements to vertices of the final Cayley graph is state-dependent.

%
\begin{figure}[h]
    \centering
    \includegraphics[width=14.5cm]{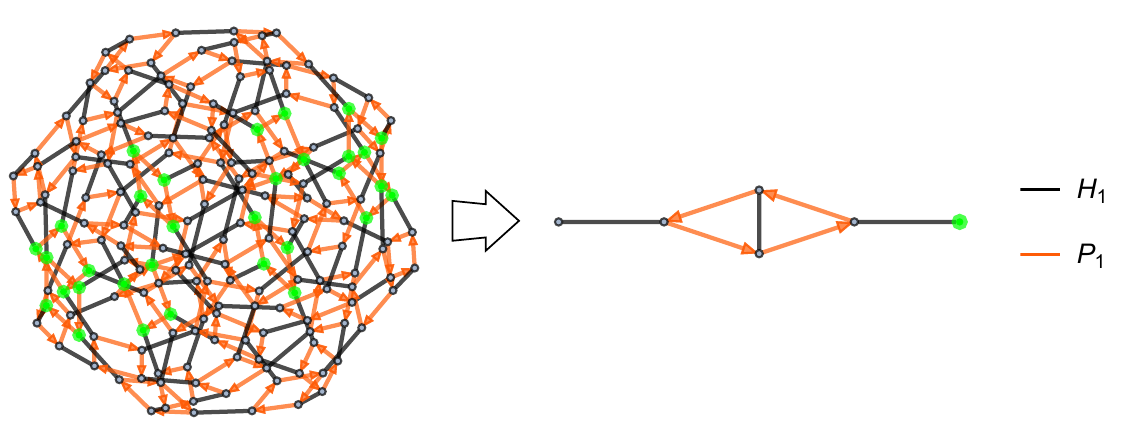}
    \caption{There are $32$ elements of the single-qubit Clifford group $\mathcal{C}_1$ which form the equivalence class of stabilizers for $\ket{0}$, denoted $\Stab_{\mathcal{C}_1}(\ket{0})$. We first build the quotient group $\bar{\mathcal{C}_1} \equiv \mathcal{C}_1/\langle \omega\rangle$ by modding out global phase elements, then quotient $\bar{\mathcal{C}_1}$ by $\Stab_{\bar{\mathcal{C}_1}}(\ket{0})$ to identify equivalent vertices in the $\mathcal{C}_1$ Cayley graph. This process yields a quotient space of the Cayley graph which is isomorphic to the one-qubit stabilizer reachability graph.}
    \label{SummaryImagePaper}
\end{figure}

In the remainder of the paper, we build an understanding of the Clifford group detailed enough to construct its Cayley graph and those of its subgroups. Then, with that task accomplished, we display the diversity of structures which ensue from applying the procedure to various quantum states. 
In Section \ref{sec:reminder}, we recall the Pauli group, Clifford group, stabilizer states, and the reachability graphs defined in our previous paper. 
In Section \ref{OneQubitSection}, we begin with the simple case of the one-qubit Clifford group $\mathcal{C}_1$, writing down its presentation and displaying its Cayley graph. 
We accomplish the same task for the two-qubit Clifford group $\mathcal{C}_2$ in Section \ref{TwoQubitSection}, where both of these steps are considerably more complicated. 
We also discuss those subgroups of $\mathcal{C}_2$ generated by a proper subset of the Clifford gates, presenting in Table \ref{tab:OrbitLengthCliffordSubgroupNoRelations} a comprehensive list of these subgroups and their properties, which might be of independent interest, as well as discussing a number of the groups in more detail. 
With the needed group-theoretic data obtained, we proceed in Section \ref{ReachabilityGraphsFromCayleyGraphs} to defining in detail the quotienting procedure summarized above and applying it to various groups and states of interest. 
We summarize and discuss further directions in Section \ref{sec:discussion}. 
Appendix \ref{ExtendedTableAppendix} provides further details of our derivations. All Mathematica data and packages are publicly available \cite{github}.

\section{Reminder: Clifford Group and Reachability Graphs}\label{sec:reminder}

We begin with a brief review of background material discussed throughout this paper. Much of this review was covered more extensively in Section 2 of \cite{Keeler2022}, and we invite the interested reader to consult it for additional details. Likewise, for a more pedagogical reference on the group-theoretic concepts we recommend e.g.\ \cite{Alperin1995}.

\subsection{The Clifford Group}

The Pauli matrices are a set of unitary and Hermitian matrices with $\pm 1$ eigenvalues, defined
\begin{equation}
    \mathbb{1}\equiv\begin{bmatrix}1&0\\0&1\end{bmatrix}, \,\, \sigma_X\equiv\begin{bmatrix}0&1\\1&0\end{bmatrix}, \,\,
    \sigma_Y\equiv\begin{bmatrix}0&-i\\i&0\end{bmatrix}, \,\,
    \sigma_Z\equiv\begin{bmatrix}1&0\\0&-1\end{bmatrix}.
\end{equation}
These matrices act as operators on a Hilbert space $\mathbb{C}^2$ in the fixed measurement basis $\{\ket{0}, \ket{1}\}$. The Pauli operators $\{\sigma_x, \sigma_y, \sigma_z \}$ generate the algebra of all linear operators on $\mathbb{C}^2$, and define a 16-element multiplicative matrix group, the one-qubit Pauli group:
\begin{equation}
    \Pi_1\equiv\langle\sigma_X,\sigma_Y,\sigma_Z\rangle,
\end{equation} 

The set of unitary matrices that normalize the Pauli group is known as the (one-qubit) Clifford group,
\begin{equation}
    \mathcal{C}_1\equiv\left\{ U\in L(\mathbb{C}^2) \: | \: 
    UgU^\dagger\: \forall g \in \Pi_1\right\}.
\end{equation}
Elements of $\mathcal{C}_1$ act as automorphisms on $\Pi_1$ via conjugation by $U$. The single-qubit Clifford group $\mathcal{C}_1$ is generated by the Hadamard and phase quantum gates, defined in a matrix representation as
\begin{equation}\label{HAndPMatrices}
    H \equiv \frac{1}{\sqrt{2}}\begin{bmatrix} 1 & 1 \\
    1 & -1\\
    \end{bmatrix}, \qquad P \equiv \begin{bmatrix} 1 & 0 \\
    0 & i\\
    \end{bmatrix}.
\end{equation}

We can extend the action of the Pauli group and Clifford groups to multiple qubits by composing strings of operators. These ``Pauli strings'' generalize local Pauli group action to a selected qubit in an $n$-qubit system, e.g. the operator which acts with $\sigma_Z$ on only the $k^{th}$ qubit can be written
\begin{equation}\label{PauliString}
    I^1\otimes\ldots\otimes I^{k-1} \otimes \sigma_Z^k \otimes I^{k+1} \otimes \ldots \otimes I^n.
\end{equation}
The weight of a Pauli string refers to the number of non-identity insertions in its tensor product representation. Eq. \eqref{PauliString} shows a Pauli string of weight one, and the set of all weight-one Pauli strings is sufficient to generate%
 \footnote{Constructing Pauli strings requires both a factorization of the $2^n$ dimensional Hilbert space in some fixed basis, as well as a chosen ordering of these factors $\{1,...,n\}$. We choose the ordering $\ket{a_1\ldots a_n}\equiv \ket{a_1}_1\otimes\ldots\ket{a_n}_n$. Elements of the $n$-qubit Pauli group are independent of any ordering choice, as are elements of the $n$-qubit Clifford group; however, the matrix representation of specific gates will depend on this order. We often consider groups which act on an $\ell$-qubit subsystem of an $n$-qubit state, fixed by a choice of ordered indices.} %
the $n$-qubit Pauli group $\Pi_n$.

The construction of the Clifford group can likewise be extended to $n>1$ qubits by adding CNOT gates to the generating set. The CNOT gate $C_{i,j}$ acts bi-locally on two qubits, performing a NOT operation on the $j^{th}$ qubit depending on the state of the $i^{th}$ qubit. In our matrix representation, we write $C_{i,j}$ as
\begin{equation}
    C_{i,j} = 
    \begin{bmatrix} 
    1 & 0 & 0 & 0 \\
    0 & 1 & 0 & 0 \\
    0 & 0 & 0 & 1 \\
    0 & 0 & 1 & 0 \\
    \end{bmatrix},
\end{equation}
where $i$ denotes the control bit and $j$ the target bit. We emphasize the fact that $C_{i,j} \neq C_{j,i}$. The group $\mathcal{C}_n$ is then
\begin{equation}
    \mathcal{C}_n \equiv \langle H_1,...,\,H_n,\,P_1,...,P_n,\,C_{1,2},\,C_{2,1},...,\,C_{n-1,n},\,C_{n,n-1}\rangle.
\end{equation}
We use a similar scheme when representing local gates, where the index denotes the qubit being acted on, e.g. $H_1$.

In this paper, we construct a presentation for the groups $\mathcal{C}_1$ and  $\mathcal{C}_2$, and analyze subgroups which are built by restricting the generating set. A presentation specifies a group by choosing a set of generators and fixing a set of relations among those generators. Elements of the group are then constructed by composing generators using the group operation, subject to the constraints set by the relations.

Every element of a multiplicative group can be written as a product of generators, known as a word. Words which independently equate to the same group element can be transformed into each other using relations in the presentation. In this way, unique words composed of Clifford group generators correspond to different constructible stabilizer circuits. We will present our set of relations as equalities between words built from Clifford generators.

\subsection{Stabilizer Formalism and Reachability Graphs}

For a group $G\subset L(\Hil)$ acting on a Hilbert space we define the stabilizer subgroup $\Stab_G(\ket{\Psi}) \leq G$, for some $\ket{\Psi} \in \Hil$, as the set of elements that leave $\ket{\Psi}$ unchanged,
\begin{equation}\label{StabilizerGroupDefinition}
    \Stab_G(\ket{\Psi}) \equiv \{g\in G\: | \:g\ket{\Psi}=\ket{\Psi}\}.
\end{equation}
That is, $\Stab_G(\ket{\Psi})$ contains only the elements of $G$ for which $\ket{\Psi}$ is an eigenvector with eigenvalue $+1$.

To make further use of this group-theoretic concept, we invoke two important theorems \cite{Alperin1995}. First, given a finite group $G$ and subgroup $H \leq G$, Lagrange's theorem states that the order of $H$ gives an integer partition of $G$, that is
\begin{equation}
    |G| = [G:H] \cdot |H|, \quad \forall \,H \leq G,
\end{equation}
where $[G:H]$ denotes the index of $H$ in $G$. Subsequently the Orbit-Stabilizer theorem says that, when considering the action of $G$ on a set $X$ and $H=\Stab_G(x)$, the orbit of $x \in X$ under $G$ has size
\begin{equation}\label{OrbitStabilizerTheorem}
    [G\cdot x] = [G:H] = \frac{|G|}{|H|}, \quad \forall \, x \in X.
\end{equation}

Considering the action of $\Pi_n$ on $\Hil$, it is clear that all states are trivially stabilized by $\mathbb{1}$. Certain states, however, are stabilized by additional elements of $\Pi_n$. The $n$-qubit ``stabilizer states'' are those which are stabilized by a subgroup of $\Pi_n$ of size $2^n$, the largest allowed size for an $n$-qubit state \cite{aaronson2004improved,garcia2017geometry}. In general, the set of $n$-qubit stabilizer states contains
\begin{equation}
    |S_n| = 2^n \prod_{k=0}^{n-1}(2^{n-k}+1)
\end{equation}
states \cite{doi:10.1063/1.4818950}. The set $S_n$ can be generated by starting with a state in the measurement basis, typically $\ket{0}^{\otimes n}$, and acting on that state with all elements of $\mathcal{C}_n$. In this way $S_n$ is the orbit $\left[\mathcal{C}_n \cdot \ket{0}^n\right]$. 

This method to generate $S_n$, by acting with $\mathcal{C}_n$ on $\ket{0}^{\otimes n}$, lends itself to a natural graph-theoretic description. By assigning a vertex to each state in the orbit $\left[\mathcal{C}_n \cdot \ket{0}^n\right]$, and an edge to every $\mathcal{C}_n$ generator, the evolution of $\ket{0}^{\otimes n}$ through $\Hil$ generates a discrete and finite graph. This structure, introduced in \cite{Keeler2022}, is known as a reachability graph, as we discussed further in Section \ref{Intro} above. When the action of a proper subgroup of $\mathcal{C}_n$ rather than the group itself is considered, we often use the term ``restricted graph'' to describe the reachability graph.

\section{$\mathcal{C}_1$ Presentation and Cayley Graph}\label{OneQubitSection}

In this section we give a presentation for the one-qubit Clifford group $\mathcal{C}_1$ and construct its Cayley graph. We will use this understanding of $\mathcal{C}_1$ to build a presentation for $\mathcal{C}_2$, as well as its subgroups, in Section \ref{TwoQubitSection}. We demonstrate that restricting the set of generators builds subgraphs of the $\mathcal{C}_1$ Cayley graph. We show that quotienting by a global phase reduces $\mathcal{C}_1$ to the symmetric group $S_4$.

 The one-qubit Clifford group $\mathcal{C}_1$ is generated by $\{H_i,P_i\}$, whose matrix representations are given in Eq. \eqref{HAndPMatrices}. Here $i \in \{1,n\}$ is the qubit being acted on in an $n$-qubit system. Relations \ref{HSquared}, \ref{PFourth}, and \ref{HPComm} give a presentation\footnote{All presentations in this paper were verified using the Magma computer algebra system \cite{MR1484478}. Additional details and code can be found in Appendix \ref{ExtendedTableAppendix}.} for $\mathcal{C}_1$,
\begin{align}
   \quad H_i^2 &= \mathbb{1}, \label{HSquared}\\
   \quad P_i^4 &= \mathbb{1}, \label{PFourth}\\ 
  (H_iP_i)^3 &= (P_iH_i)^3 = \omega \label{HPComm}, 
\end{align}
where $\omega^8 = \mathbb{1}$ acts as a global phase\footnote{There exist additional relations involving Clifford gates and $\omega$. Some notable ones which are used in Section \ref{TwoQubitSection} include $(H_iP_jC_{i,j})^6 = (H_iC_{i,j}P_j)^6 = \omega^6$.} for the group. Eqs. \eqref{HSquared}--\eqref{HPComm} can be directly verified by examining the matrix representations in Eq. \eqref{HAndPMatrices}. All elements of $\mathcal{C}_1$ act locally on qubits, and therefore cannot generate or modify entanglement in a physical system.

A \textit{Cayley graph} \cite{10.2307/2369306}, for a group $G$, is built by assigning a vertex to every element in $G$, and an edge for each generator of $G$. The structure of $\mathcal{C}_1$ can be visualized in the Cayley graph shown in Figure \ref{C1CayleyGraph}. Edges of this $\mathcal{C}_1$ Cayley graph represent $H_i$ and $P_i$, while vertices indicate the $192$ unique group elements. Since $H_i^2 = \mathbb{1}$, we use a single undirected edge to represent $H_i$. Directed edges are used to represent $P_i$, as $P_i^2 \neq \mathbb{1}$. In this Cayley graph representation, sequential products of group elements exist as graph paths. Different paths which start and end on the same pair of vertices represent products whose action on the initial element is identical. Loops in the Cayley graph correspond to a sequence of operations which act as the identity.
    \begin{figure}[h]
        \centering
        \includegraphics[width=8.4cm]{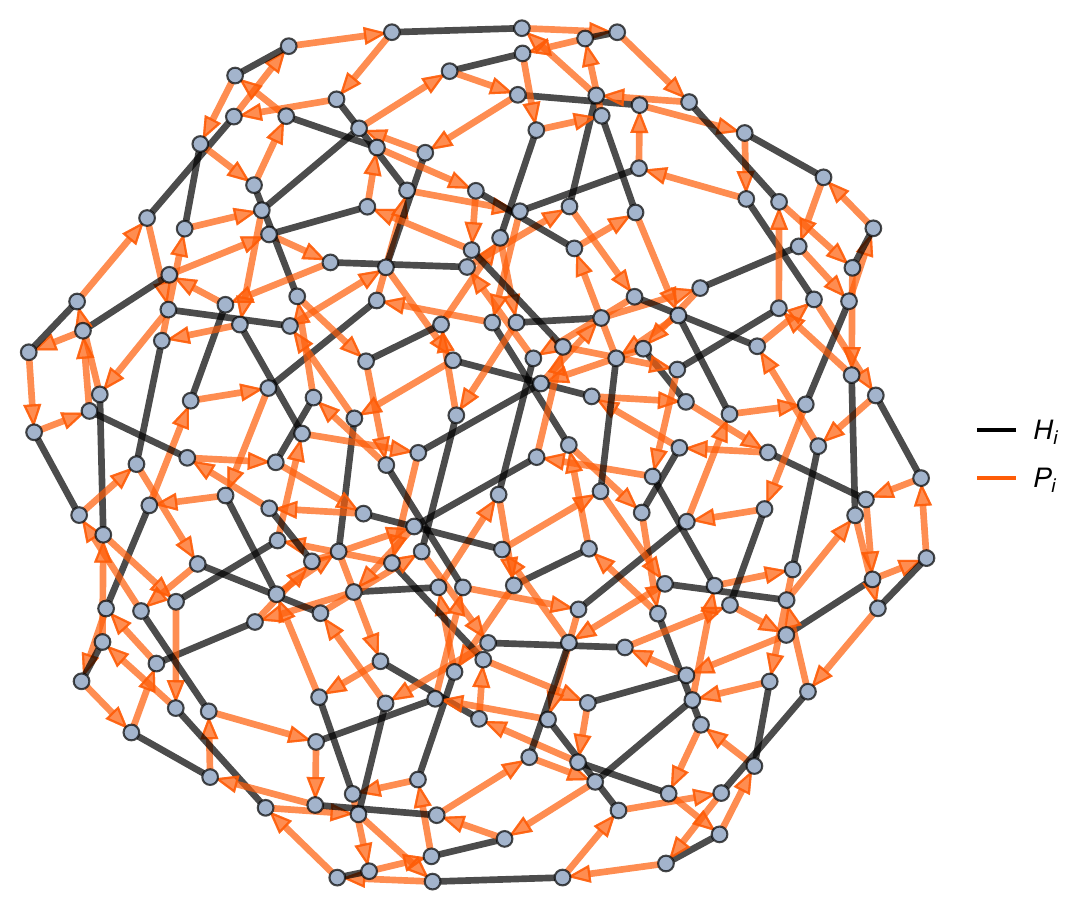}
        \caption{Cayley graph of $\mathcal{C}_1$, with vertices representing group elements and edges representing the generators $H_i$ and $P_i$. We use undirected edges for $H_i$ since $H_i^2 = \mathbb{1}$. The graph has $192$ vertices and $384$ edges, and completely encodes the $\mathcal{C}_1$ group structure.}
    \label{C1CayleyGraph}
    \end{figure}

When we ignore global phase, i.e.\ distinguish group elements only up to the factor $\omega$, $\mathcal{C}_1$ reduces to a quotient group with $24$ elements. This quotient group, is isomorphic to $S_4$, the symmetric group of degree $4$; we give its Cayley graph later in the paper, in Figure \ref{C1Quotient}. $S_4$ describes the rotational symmetries of an octahedron, like the well-known stabilizer octahedron shown in Figure \ref{StabilizerOctahedron}. 
	\begin{figure}[bh]
		\begin{center}
		\begin{overpic}[width=0.35\textwidth]{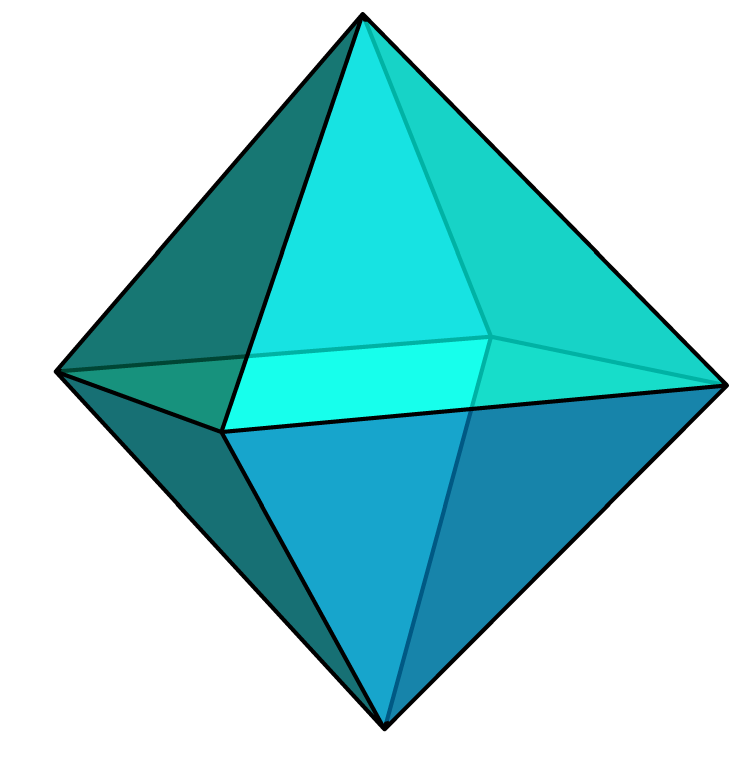}
		\put (44,101) {$\ket{0}$}
		\put (47,-3) {$\ket{1}$}
		\put (-4,50) {$\ket{+}$}
		\put (97,48) {$\ket{-}$}
		\put (23,36.5) {$\ket{i}$}
		\put (66,58) {$\ket{-i}$}
        \end{overpic}
        \caption{The stabilizer octahedron, with $6$ single-qubit stabilizer states at the corners, is often shown embedded in the Bloch sphere. The group $\mathcal{C}_1$, after quotienting by $\omega$, gives the $24$ orientation-preserving maps of this octahedron to itself.}
		\label{StabilizerOctahedron}
	\end{center}
	\end{figure}

In addition to modding by global phase, we can also construct subgroups of $\mathcal{C}_1$ by restricting our set of generators. The single-generator subgroups $\langle H_i \rangle$ and $\langle P_i \rangle$ are completely described by relations \ref{HSquared} and \ref{PFourth} respectively. The Cayley graphs of $\langle H_i \rangle$ and $\langle P_i \rangle$ are shown in Figure \ref{HadamardAndPhaseCayleyGraphs}.
    \begin{figure}[h]
        \centering
        \includegraphics[width=9cm]{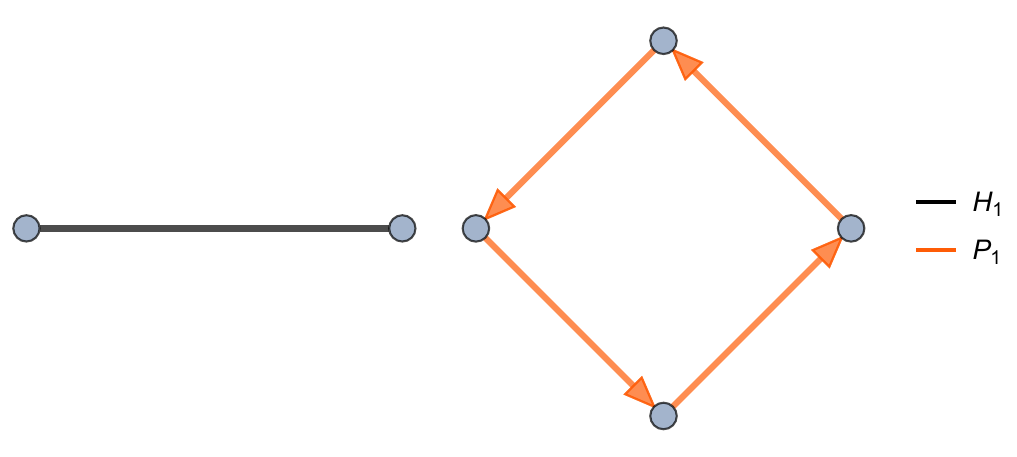}
        \caption{The Cayley graphs for $\mathcal{C}_1$ subgroups $\langle H_i \rangle$ and $\langle P_i \rangle$.}
    \label{HadamardAndPhaseCayleyGraphs}
    \end{figure}

We gave a presentation for $\mathcal{C}_1$ generated by $H$ and $P$, and introduced $\omega \equiv (H_iP_i)^3$ which acts as a global phase on $\mathcal{C}_1$. We introduced the concept of a Cayley graph, and constructed specific Cayley graphs for $\mathcal{C}_1$ and its single-generator subgroups $\langle H_i \rangle$ and $\langle P_i \rangle$. We described a quotient procedure for groups, and use it to quotient $\mathcal{C}_1$ by $\omega$ to recover $S_4$. Later, we will implement this quotient by $\omega$, as well as by stabilizer subgroups of quantum states, to generate reachability graphs from Cayley graphs. 

\section{$\mathcal{C}_2$ Presentation and Cayley Graphs}\label{TwoQubitSection}

In this section we give a presentation for the two-qubit Clifford group generated by $H, \,P,$ and $CNOT$. This presentation includes a set of operator-level relations, which serve as a set of state-independent constraints on Clifford circuits. We use this presentation to construct all subgroups of $\mathcal{C}_2$ which are generated by subsets of $\{H_i,H_j,P_i,P_j,C_{i,j},C_{j,i}\}$. We give the order of each $\mathcal{C}_2$ subgroup and show how each order is reduced after quotienting the group by $\omega$. For several examples we explicitly build up each element of a subgroup to demonstrate how our relations constrain combinations of Clifford operators.

Every group can be represented by a Cayley graph, which we build for all $\mathcal{C}_2$ subgroups. Since Cayley graphs are state-independent structures, we can use them to study Clifford orbits of arbitrary quantum states. We compute the graph diameter for each $\mathcal{C}_2$ subgroup Cayley graph, both before and after quotienting by $\omega$. We display the Cayley graphs for several example subgroups and highlight group relations that can be visualized as graph paths. In later sections, we will use the quotient procedure outlined here to construct reachability graphs as quotient spaces of Cayley graphs.

\subsection{$\mathcal{C}_2$ Presentation}\label{C2Presentation}

The two-qubit Clifford group $\mathcal{C}_2$ is generated by the set $\{H_i,H_j,P_i,P_j,C_{i,j},C_{j,i}\}$ which consists of the local Hadamard and phase gates, as well as the bi-local CNOT gate. Relations \ref{CSquared}-\ref{ChFourth}, in addition to the $\mathcal{C}_1$ relations \ref{HSquared}-\ref{HPComm}, give a presentation for $\mathcal{C}_2$:
\begin{align}
    C_{i,j}^2 &= \mathbb{1}, \label{CSquared}\\ 
    P_i^{-1}P_jP_i &= P_j, \label{PpComm}\\
    H_i^{-1}H_jH_i &= H_j, \label{HhComm}\\
    P_i^{-1}H_jP_i &= H_j, \label{HpComm}\\
    C_{i,j}H_jC_{i,j}P_jC_{i,j}P_j^3H_j &= P_i, \label{FourGenRelation}\\
    H_iH_jC_{j,i}H_iH_j &= C_{i,j}, \label{CNOTTransform}\\
    (C_{i,j}P_j)^4 &= P_i^2, \label{CpFourth}\\
   \quad C_{i,j}^{-1}C_{j,i}C_{i,j} &= C_{j,i}^{-1}C_{i,j}C_{j,i}, \label{CcComm}\\
    \quad  P_i^{3}C_{i,j}P_i &= C_{i,j}, \label{PCComm}\\
    \quad  (C_{i,j}H_j)^4 &= P_i^2.\label{ChFourth}\\ \nonumber
\end{align}
The relations \ref{CcComm}--\ref{ChFourth}, along with $H_i^2 = \mathbb{1}$ and $P_i^4 = \mathbb{1}$, can be removed to furnish a more minimal presentation%
\footnote{We additionally note that $\mathcal{C}_2$ can be minimally generated from the set $\{H_i, \,H_j, \,P_i, \,C_{j,i}\}$, as can be seen from relations \ref{FourGenRelation} and \ref{CNOTTransform}.} %
using only relations \ref{CSquared}--\ref{CpFourth}. We have nevertheless retained a number of non-minimal relations as they provide insight into the structure of $\mathcal{C}_2$ and will be useful for constructing subgroups in the following subsection. 

Every relation in Eqs. \eqref{CSquared}--\eqref{ChFourth} is a cycle in the Cayley graph of $\mathcal{C}_2$. We especially note relation \ref{ChFourth}, $(C_{i,j}H_j)^4 = P_i^2$, which allows us to build a phase operation using only Hadamard and CNOT. Since $P_i^2$ cannot modify entanglement, neither can the sequence $(C_{i,j}H_j)^4$. This relation is critical for demonstrating entropy bounds on reachability graphs in \cite{Keeler2023b}. Relation \eqref{ChFourth} is derived explicitly in \ref{ChFourthEqualsPSquared}. 

While our presentation for $\mathcal{C}_2$ does not depend on the choice of qubits $1$ and $2$, and describes the action of $\mathcal{C}_2$ on an $n$-qubit system, it is not a presentation for $\mathcal{C}_n$ when $n > 2$. A presentation for $\mathcal{C}_n$ requires additional generators for each increase in qubit number. One can, however, generalize our $\mathcal{C}_2$ presentation to a presentation for $\mathcal{C}_3$ by adding only four relations. Each of these four new relations pertain only to Hadamard and CNOT gates, and no new phase gate relations are needed. The additional relations can be found in \cite{Selinger2013}, where alternative presentations%
\footnote{The presentation in \cite{Selinger2013} is given with generators $H_i,\,P_i,$ and $CZ_{i,j}$, and offers a different set of relations. Additionally, $CZ_{i,j} = CZ_{j,i}$ while $C_{i,j} \neq C_{j,i}$.} %
for $\mathcal{C}_1, \,\mathcal{C}_2,$ and $\mathcal{C}_3$ are studied using Clifford circuit normal forms.

\subsection{$\mathcal{C}_2$ Subgroups}

We now give a complete description of all $\mathcal{C}_2$ subgroups built by restricting the generating set. First, we list all such subgroups, as well as their group and Cayley graph properties, in Table \ref{tab:OrbitLengthCliffordSubgroupNoRelations}. We directly construct several subgroups as examples that highlight how our relations constrain strings of Clifford gates at the operator level. We will use these state-independent relations in Section \ref{ReachabilityGraphsFromCayleyGraphs} to build reachability graphs for non-stabilizer quantum states, and further in \cite{Keeler2023b} to bound entanglement entropy. 

We can construct subgroups of $\mathcal{C}_2$ by restricting our set of generators to subsets of $\{H_i,H_j,P_i,P_j,C_{i,j},C_{j,i}\}$. One simple case is the subgroup $\mathcal{C}_1$, generated by only $\{H_i,P_i\}$ and discussed in Section \ref{OneQubitSection}. Table \ref{tab:OrbitLengthCliffordSubgroupNoRelations} gives a list of all subgroups constructed in this way. The Table gives the order of each subgroup, as well as the graph diameter (the maximum over all minimum distances between vertices) of the Cayley graph for each subgroup. Since each subgroup below is isomorphic under qubit exchange, we only list one example for each generating set. Subgroups in bold are explicitly constructed in the following text.
\begin{table}[h]
    \centering
    \begin{tabular}{|c||c|c||c|c|}
    \hline
    Generators & Order & Diam. (w/ phase) & Factor & (no phase)\\
    \hline
    \hline
    $\mathbf{\{H_1\}}$ & $2^\dagger$ & 1 & - & - \\
    \hline
    $\mathbf{\{C_{1,2}\}}$ & $2^\dagger$ & 1 & - & - \\
    \hline
    $\mathbf{\{P_1\}}$ & $4$ & 3 & - & - \\
    \hline
    $\mathbf{\{H_1,\,H_2\}}$ & $4$ & 2 & - & - \\
    \hline
    $\mathbf{\{C_{1,2},\,C_{2,1}\}}$ & $6$ & 3 & - & -\\
    \hline
    \{$H_1,\,P_2\}$ & $8^\dagger$ & 4 & - & - \\
    \hline
    \{$P_1,\,C_{1,2}\}$ & $8^\dagger$ & 4 & - & - \\
    \hline
    \{$P_1,\,P_2$\} & $16$ & 6 & - & - \\
    \hline
    \{$H_1,\,C_{2,1}\}$ & $16^\dagger$ & 8 & - & - \\
    \hline
    \{$H_1,\,C_{1,2}\}$ & $16^\dagger$ & 8 & - & - \\
    \hline
    $\mathbf{\{H_1,\,P_2,\,C_{2,1}\}}$ & $32$ & 6 & - & - \\
    \hline
    $\mathbf{\{P_2,\,C_{1,2}\}}$ & $32$ & 8 & - & - \\
    \hline
    \{$P_1,\,P_2,\,C_{2,1}$\} & $64$ & 7 & - & -\\
    \hline
    \{$P_1,\,C_{2,1},\,C_{1,2}$\} & $192$ & 11 & - & - \\
    \hline
    \{$H_1,\,P_1$\} & $192$ & 16 & 8 & 6 \\
    \hline
    \{$H_1,\,H_2,\,P_1$\} & $384$ & 17 & 8 & 7 \\
    \hline
    \{$P_1,\,P_2,\,H_1$\} & $768$ & 19 & 8 & 9 \\
    \hline
    \{$H_1,\,C_{2,1},\,C_{1,2}\}$ & $2304^*$ & 26 & 2 & 15 \\
    \hline
    \{$H_1,\,H_2,\,C_{1,2}\}$ & $2304^*$ & 27 & 2 & 17 \\
    \hline
    $\mathbf{\{H_1,\,H_2,\,C_{1,2},\,C_{2,1}\}}$ & $2304^*$ & 25 & 2 & 15 \\
    \hline
    \{$H_1,\,P_1,\,C_{2,1}\}$ & $3072^*$ & 19 & 8 & 9 \\
    \hline
    \{$H_1,\,P_1,\,C_{1,2}$\} & $3072$ & 19 & 8 & 11 \\
    \hline
    \{$H_1,\,P_1,\,P_2,\,C_{2,1}\}$ & $3072^*$ & 19 & 8 & 9 \\
    \hline
    \{$H_1,\,H_2,\,P_1,\,P_2$\} & $4608$ & 17 & 8 & 12 \\
    \hline
    $\mathbf{\{H_1,\,P_2,\,C_{1,2}\}}$ & $9216$ & 24 & 8 & 13 \\
    \hline
    \{$H_1,\,H_2,\,P_1,\,C_{2,1}\}$ & $92160^*$ & 21 & 8 & 13 \\
    \hline
    \{$H_1,\,H_2,\,P_1,\,C_{1,2}\}$ & $92160^*$ & 21 & 8 & 16 \\
    \hline
    \{$H_1,\,P_1,\,P_2,\,C_{1,2}\}$ & $92160^*$ & 21 & 8 & 14 \\
    \hline
    \{$H_1,\,H_2,\,P_1,\,P_2,\,C_{1,2},\,C_{2,1}\}$ & $92160^*$ & 19 & 8 & 11\\
    \hline
    \end{tabular}
\caption{Subgroups generated by generator subsets are shown in the leftmost column. We give the order of each subgroup and its Cayley graph diameter, both before and after modding by global phase. The third column gives the factor reduction by removing global phase. An asterisk indicates groups with the same elements, and a dagger indicates groups with isomorphic Cayley graphs. Bolded subgroups are explicitly constructed in the text.}
\label{tab:OrbitLengthCliffordSubgroupNoRelations}
\end{table}

Some subgroups have the same order and are isomorphic, e.g. $\langle H_1 \rangle \cong \langle C_{1,2}\rangle$ and $\langle H_1, \,P_1\rangle \cong\langle P_1, \,C_{1,2}\rangle$. Other subgroups have the same order, but are not isomorphic; for example, subgroups $\langle P_1, \,P_2 \rangle,\, \langle H_1, \,C_{1,2} \rangle,$ and $\langle H_1, \,C_{2,1} \rangle$ all have order $16$, but $\langle P_1, \,P_2 \rangle \ncong \langle H_1, \,C_{1,2} \rangle \cong \langle H_1, \,C_{2,1} \rangle$. Even when generated groups are isomorphic, as is the case for subgroups $\langle H_1, \,H_2, \,C_{1,2} \rangle, \,\langle H_1, \,C_{1,2}, \,C_{2,1} \rangle,$ and $\langle H_1, \,H_2, \,C_{1,2}, \,C_{2,1} \rangle$, we emphasize that they may not have isomorphic Cayley graphs, since the Cayley graph depends on not just the group but a choice of generators: here, none of the three descriptions do.
\newpage
Subgroup $\langle H_1, \,H_2, \,C_{1,2} \rangle$, which contains the same elements as $\langle H_1, \,C_{1,2}, \,C_{2,1}\rangle$ and $\langle H_1, \,H_2, \,C_{1,2}, \,C_{2,1}\rangle$, has the Cayley graph of largest diameter. Adding $C_{2,1}$ to the set $\langle H_1, \,H_2, \,C_{1,2} \rangle$ generates no new group elements, and instead lowers the graph diameter by introducing additional edges between the set of vertices. Adding $P_1$ to the set $\langle H_1, \,H_2, \,C_{1,2} \rangle$ does generate additional elements---in fact $\langle H_1, \,H_2, \,P_1, \,C_{1,2} \rangle$ generates all of $\mathcal{C}_2$---but also lowers the Cayley graph diameter by adding additional edges.

We now discuss in depth how several subgroups are constructed, offering an explanation for order of each group seen in Table \ref{tab:OrbitLengthCliffordSubgroupNoRelations}. An extended version of Table \ref{tab:OrbitLengthCliffordSubgroupNoRelations}, containing the relations needed to present each subgroup, is given in Appendix \ref{ExtendedTableAppendix}. Additional Cayley graph illustrations are given in Appendix \ref{ExtendedTableAppendix}.

\paragraph{Single-Generator Subgroups:} The $\mathcal{C}_2$ subgroups generated by a single Clifford element, i.e. $\langle H_i \rangle, \,\langle P_i \rangle$, and $\langle C_{i,j} \rangle$, are completely described by Eqs. \eqref{HSquared}, \eqref{PFourth}, and \eqref{CSquared} respectively. Groups $\langle H_i \rangle$ and $\langle P_i \rangle$ were discussed in Section \ref{OneQubitSection}, and their Cayley graphs shown in Figure \ref{HadamardAndPhaseCayleyGraphs}. At two qubits we have the possibility of bi-local gates, such as $C_{i,j}$. Since $C_{i,j}^2 = \mathbb{1}$, as shown by Eq. \eqref{CSquared}, the group $\langle C_{i,j} \rangle$ is isomorphic to $\langle H_i \rangle$. 

\paragraph{Subgroups $\langle H_i,H_j \rangle$ and $\langle C_{i,j},C_{j,i} \rangle$:}
The subgroup generated by $\{H_i,H_j\}$ is completely described by Eqs. \eqref{HSquared} and \eqref{HhComm}. Since $H_i$ and $H_j$ commute for $i\neq j$, the group $\langle H_i,H_j \rangle$ has only $4$ elements, and its structure can be easily understood by examining the left image of Figure \ref{HBoxCHexPaper}. Similarly, the subgroup $\langle C_{i,j},C_{j,i} \rangle$ is described by Eqs. \eqref{CSquared} and \eqref{CcComm}. The elements $C_{i,j}$ and $C_{j,i}$ do not commute, but instead form the hexagonal structure to the right of Figure \ref{HBoxCHexPaper}.
\begin{figure}[h]
\begin{center}
\includegraphics[width=10.5cm]{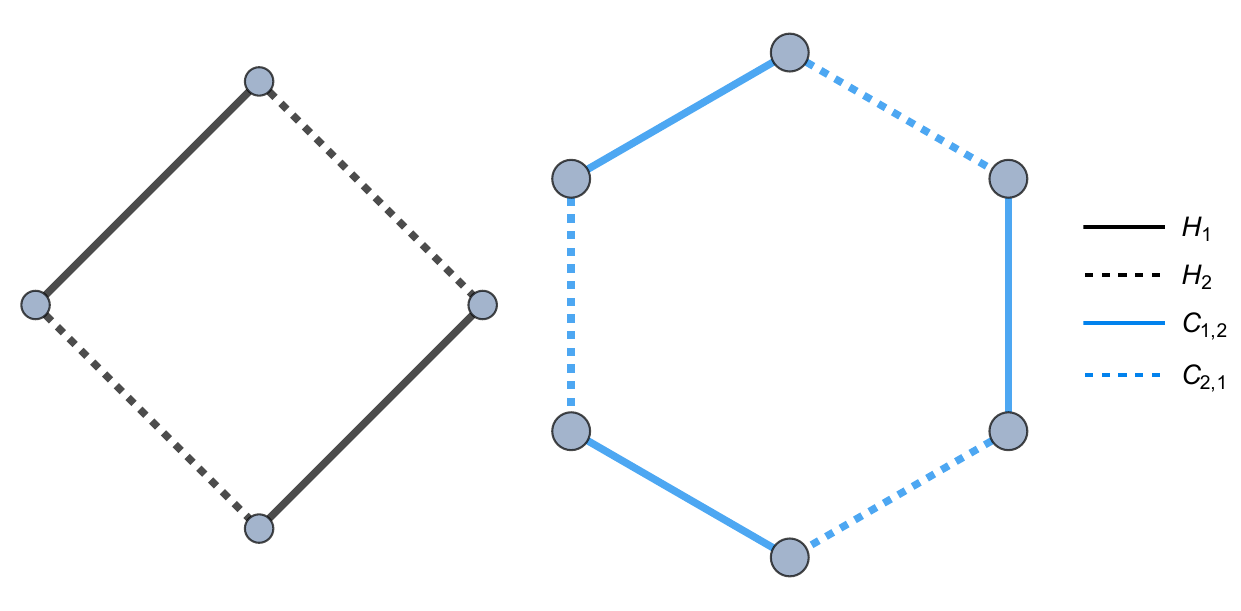}
\caption{Cayley graphs for subgroups $\langle H_1,H_2\rangle$ and $\langle C_{1,2},C_{2,1} \rangle$ arrange into square and hexagonal structures respectively. All edges in these figures are undirected since both $H$ and $C$ are their own inverse.}
\label{HBoxCHexPaper}
\end{center}
\end{figure}

\paragraph{Subgroup $\langle H_1, \,P_2, \,C_{2,1} \rangle$:} The subgroup generated by $\{H_1, \,P_2, \,C_{2,1}\}$ can be presented using Eqs. \eqref{HSquared}, \eqref{PFourth}, \eqref{HPComm}, \eqref{CSquared}, \eqref{HpComm}, \eqref{PCComm}, and \eqref{ChFourth}. Since $P_2$ commutes with both $H_1$ and $C_{2,1}$, all $P_2$ in a word can be pushed completely to one end, such that they occur either before or after all $H_1$ and $C_{2,1}$ operations. In this way, the elements of $\langle H_1, \,P_2, \,C_{2,1} \rangle$ can be constructed as products of some element from $\langle H_1, \,C_{2,1} \rangle$ with an element from the set $\{\mathbb{1},\,P_2,\,P_2^2,\,P_2^3\}$. Initially, this generates $16 \times 4 = 64$ words, however Eq. \eqref{ChFourth} demonstrates how $P_2^2$ can be built from $H_1$ and $C_{2,1}$. Therefore, all words containing $P_2^2$ or $P_2^3$ can be reduced to a shorter sequence, and the order $\langle H_1, \,P_2, \,C_{2,1} \rangle$ becomes $32$. 

\paragraph{Subgroup $\langle P_2,\,C_{1,2}\rangle$:} The subgroup generated by $\{ P_2,\,C_{1,2}\}$ can be built using Eqs. \eqref{PFourth}, \eqref{CSquared}, and \eqref{CpFourth}. Figure \ref{C12P2CayleyGraph} shows the Cayley graph for $\langle P_2,\,C_{1,2}\rangle$. We construct this subgroup by building words of alternating $C_{1,2}$ and $p$, where
\begin{equation}\label{PhaseGroup}
    p \in \{\mathbb{1},\,P_2,\,P_2^2,\,P_2^3\}.    
\end{equation}
 For clarity, we introduce the notation $\overline{p} \in \{P_2,\,P_2^2,\,P_2^3\}$. We generate all words containing up to $2$ CNOT operations, since the relation
 \begin{equation}\label{CpCpRelation}
     C_{i,j}P_jC_{i,j}P_j = P_jC_{i,j}P_jC_{i,j},
 \end{equation}
 derived in Eq. \eqref{CpCpRelationDerivation}, allows words with $3$ or more $C_{1,2}$ operations to be written as duplicates of words containing fewer $C_{1,2}$ operations.
    \begin{enumerate}
        \item For words containing $0$ $C_{1,2}$ operations, we have only the set $p$, containing $4$ unique elements.
        \item Words containing a $1$ $C_{1,2}$ operation have the form $pC_{1,2}p$, with full choice of $p$ on either side of $C_{1,2}$, giving $4 \times 4 = 16$ possible new elements.
        \item Words containing $2$ or more $C_{1,2}$ operations must alternate $C_{1,2}$ and $\overline{p}$ operations, or could otherwise be reduced by $(C_{1,2})^2 = \mathbb{1}$. Thus all $2$ $C_{1,2}$ words have the form $C_{1,2}\overline{p}C_{1,2}p$ (note that we never include $\mathbb{1}$ between $C_{1,2}$ operations as it could be carried through a $C_{1,2}$ to collapse the $C_{1,2}$ pair). We apply Eq. \eqref{CpCpRelation} to any word of the form $pC_{1,2}pC_{1,2}p$ to move all $p$ operations as far to the right as possible. In this way, we generate $3 \times 4 = 12$ new elements.
    \end{enumerate}
The above construction explicitly generates the $4+16+12 = 32$ elements of $\langle C_{1,2},\, P_2 \rangle$, each having one of the forms $\{p,\, pCp,\, C\overline{p}Cp\}$.

\begin{figure}[h]
\begin{center}
\includegraphics[width=9cm]{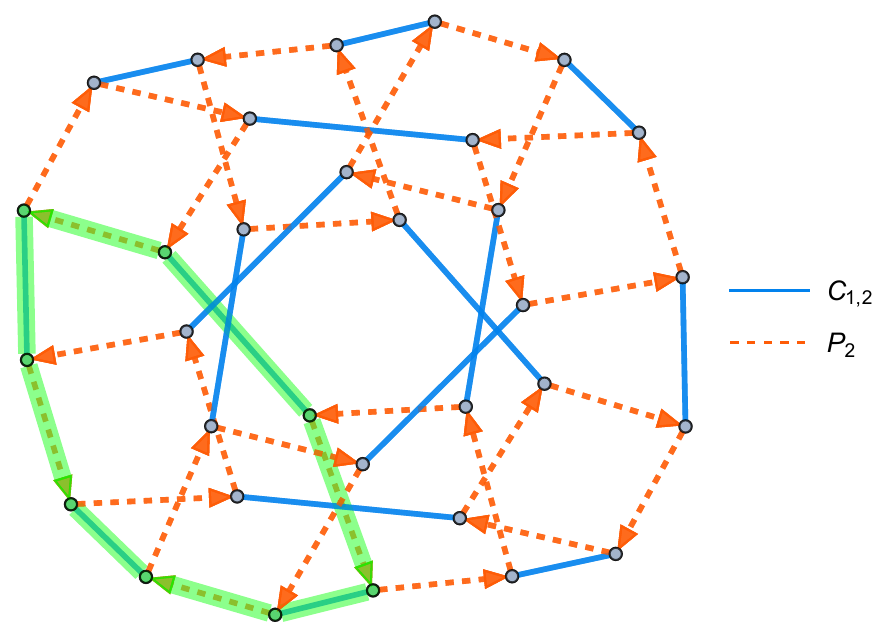}
\caption{Cayley graph of $\langle P_2,\,C_{1,2}\rangle$ subgroup, which is useful for visualizing relations such as $C_{i,j}P_jC_{i,j}P_j = P_jC_{i,j}P_jC_{i,j}$, highlighted in green.}
\label{C12P2CayleyGraph}
\end{center}
\end{figure}

\paragraph{Subgroup $\langle H_1,\, P_2,\, C_{1,2} \rangle$:} The subgroup generated by $\{H_1,\, P_2, \,C_{1,2}\}$ can be built using Eqs. \eqref{HSquared}--\eqref{HPComm}, \eqref{CSquared}, \eqref{PpComm}, \eqref{HpComm}, \eqref{FourGenRelation}, \eqref{CNOTTransform}, \eqref{CcComm}, \eqref{PCComm}, and \eqref{ChFourth}. We will also use
\begin{equation}\label{CHpSquared}
    (C_{1,2}H_1P_2^2)^2 = (P_2^2H_1C_{1,2})^2,
\end{equation}
which can be derived from the relations \eqref{CSquared}--\eqref{ChFourth}. 

As with $\langle P_2, \,C_{1,2} \rangle$, we construct $\langle H_1,\,P_2, \,C_{1,2} \rangle$ by building words of alternating $H_1$ and one element from $\{P_2, \,C_{1,2}\}$, since $(H_1)^2 = \mathbb{1}$. We only need to construct all words containing up to $5$ $H_1$ operations, since words with $6$ $H_1$ operations can be reduced using $(H_1C_{1,2})^8 = \mathbb{1}$ or $(H_1C_{1,2}P_2)^6 = (P_2H_1C_{1,2})^6 = \omega^6$. The full construction of $\langle H_1,\,P_2, \,C_{1,2} \rangle$ is given in Appendix \ref{ExtendedTableAppendix}. 

Appending $H_2$ to the generating set $\{H_1, \,P_2, \,C_{1,2}\}$ results in a factor of $10$ more elements, giving the full group $\mathcal{C}_2$. Adding more generators to $\{H_1, \,H_2, \,P_2, \,C_{1,2}\}$ does not add more group elements, and instead lowers the graph diameter. In fact, the set $\{H_1, \,H_2, \,P_2, \,C_{1,2}\}$ is a minimal generating set%
\footnote{There exist generating sets for $\mathcal{C}_2$ with fewer elements which involve composite Clifford operations e.g. the set $\{H_1P_1, \,H_2P_2, \,C_{1,2}\}$.} %
for $\mathcal{C}_2$ with the generators $\{H_1, \,H_2, \,P_1, \,P_2, \,C_{1,2}, \,C_{2,1}\}$. 

\paragraph{Subgroup $\langle H_1, \,H_2, \,C_{1,2}, \,C_{2,1} \rangle$:} The group generated by $\{H_1, \,H_2, \,C_{1,2},\,C_{2,1} \}$ can be understood from Eqs. \eqref{HSquared}, \eqref{PFourth}, \eqref{HPComm}, \eqref{CSquared}, \eqref{CcComm}, and \eqref{ChFourth}. We additionally make use of the identity,
\begin{equation}\label{HadamardTransform}
    C_{j,i}C_{i,j}C_{j,i}H_iC_{j,i}C_{i,j}C_{j,i} = H_j,
\end{equation}
which transforms a Hadamard using a sequence of CNOT operations. Initially, we might assume this group is the direct product%
\footnote{While $\langle H_1, \,H_2, \,C_{1,2}, \,C_{2,1} \rangle$ is not a direct product, one example which is a direct product is $\langle H_1, \,H_2, \,P_1,\,P_2 \rangle  = \langle H_1, \,P_1 \rangle \times \langle H_2,\,P_2 \rangle $, which has $192^2/8=4608$ elements.} %
of $\langle H_1, \,C_{1,2} \rangle$ and $\langle H_2, \,C_{2,1} \rangle$, thereby having $256$ elements. However, Eq. \eqref{ChFourth} importantly demonstrates how a sequence of $H$ and $C$ operations can create $P_i^2$. This generates a factor of $9$ more elements, for a total of $2304$. Furthermore, since Eq. \eqref{CNOTTransform} offers a way to construct $C_{2,1}$ as the product of $H_1, \,H_2,$ and $C_{1,2}$, the subgroup $\langle H_1, \,H_2, \,C_{1,2}, \,C_{2,1}\rangle$ can be minimally generated from sets $\{H_i,\,H_j,\,C_{i,j}\}$ or $\{H_i,\,C_{i,j},\,C_{j,i}\}$. Figure \ref{HhCcCayleyGraph} shows the Cayley graph for $\langle H_1, \,H_2, \,C_{1,2}, \,C_{2,1} \rangle$.
\begin{figure}[h]
\begin{center}
\includegraphics[width=10cm]{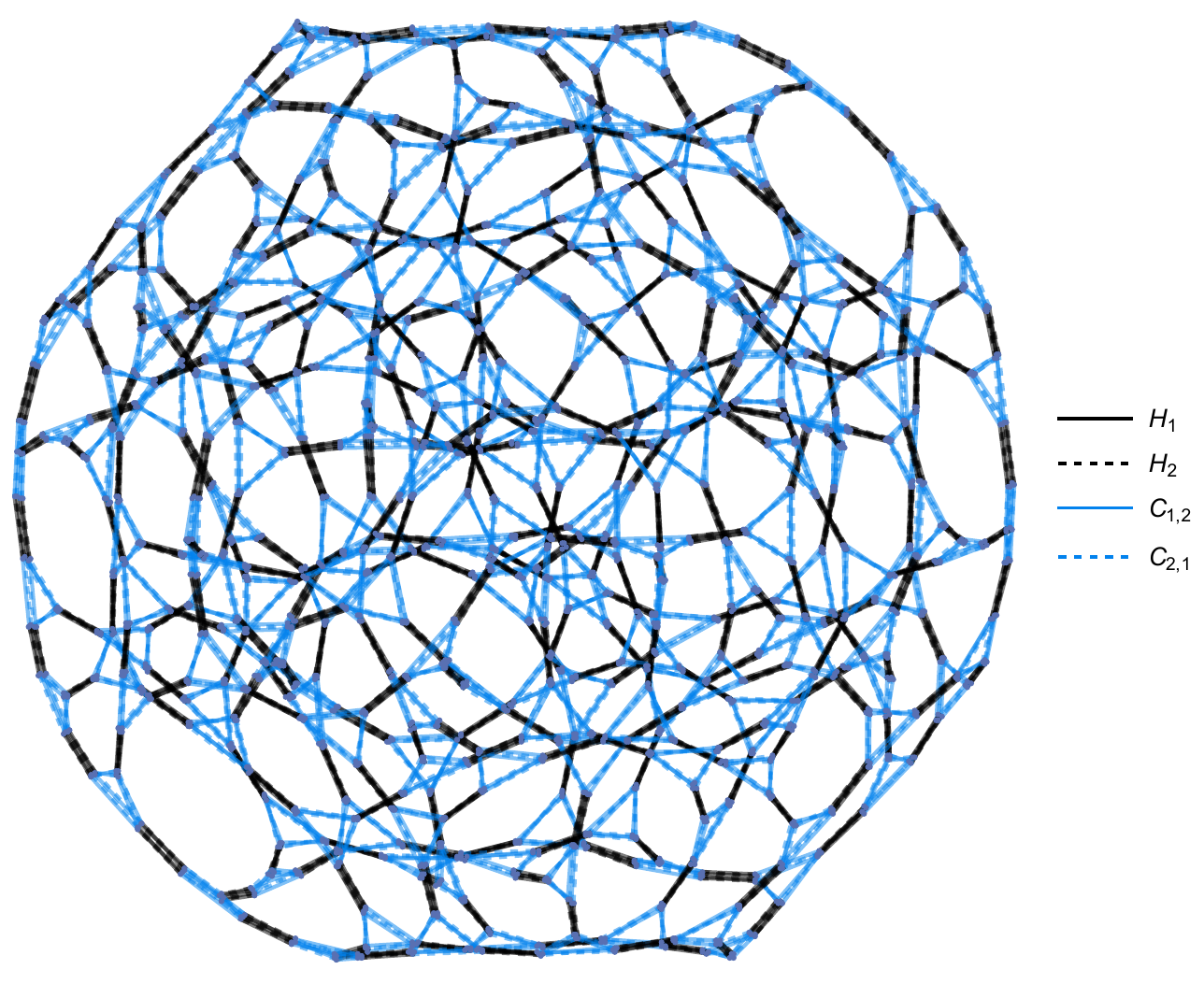}
\caption{The Cayley graph for $\mathcal{C}_2$ subgroup $\langle H_1,\,H_2,\,C_{1,2},\,C_{2,1}\rangle$. This graph has $2304$ vertices, and displays the orbit of an arbitrary quantum state under the action of $\langle H_1,\,H_2,\,C_{1,2},\,C_{2,1}\rangle$.}
\label{HhCcCayleyGraph}
\end{center}
\end{figure}

Through building this presentation and constructing subgroups in detail, we have developed a functional understanding of $\mathcal{C}_2$. We identified a collection of Clifford group relations which are independent of the state set being acted on. This state-independent description will allow us to extend an analysis beyond the set of stabilizer states, and to explore action of the Clifford group on arbitrary quantum states. By systematically constructing all words in a subgroup, we were able to highlight exactly how our relations transform Clifford strings. We found additional relations, such as $C_{i,j}P_jC_{i,j}P_j = P_jC_{i,j}P_jC_{i,j}$ and $C_{j,i}C_{i,j}C_{j,i}H_iC_{j,i}C_{i,j}C_{j,i} = H_j$, which are not included in our presentation, but can be derived from relations \eqref{CSquared}--\eqref{ChFourth}. These auxiliary relations proved useful for understanding why certain sequences of Clifford gates are non-trivially equivalent to others, as well as how entanglement entropy evolves through Clifford circuits.

Constructing the Cayley graph for each subgroup further illustrated the structure of $\mathcal{C}_2$ and its subgroups. These graphs enabled us to visualize the operator relations that were used to build each subgroup. Computing the Cayley graph diameter, and observing its change after adding generators or quotienting by $\omega$, offered additional intuition for $\mathcal{C}_2$ subgroup connectivity. These Cayley graphs will constitute a state-independent starting point for constructing reachability graphs in the following section. By considering quotient spaces of this purely group-theoretic structure, we are able to analyze the orbit of arbitrary quantum states under a selected gate set. Furthermore, by understanding how this quotient protocol modifies a Cayley graph we are able to build alternative graphs that can track and bound the evolution of certain system properties, such as entanglement entropy \cite{Keeler2023b}.

\section{Reachability Graphs as Cayley Graph Quotients}\label{ReachabilityGraphsFromCayleyGraphs}

We now generalize the notion of reachability graphs by constructing them as quotient spaces of Cayley graphs. We define equivalence classes on group elements by their congruent action on a chosen state. We demonstrate how identifying vertices in a Cayley graph collapses its structure to a state reachability graph. Starting with the state-independent Cayley graph, we are able to strictly bound the orbits for different states under select sets of gates.

In each example below, we first quotient the Cayley graph by global phase, then by the stabilizer subgroup for a chosen state. We explicitly compute quotients of $\mathcal{C}_1$ and $\mathcal{C}_2$ which yield familiar reachability graphs for one and two qubit stabilizer states. Restricting to the subgroup $\langle H_1,\,H_2,\,C_{1,2},\,C_{2,1} \rangle \lhd \mathcal{C}_2$, which we denote $\HC$ going forward, we recover the reachability graphs studied in \cite{Keeler2022}.

By adding $P_1$ and $P_2$ to the set $\{ H_1,\,H_2,\,C_{1,2},\,C_{2,1} \}$, we consider the full action of $\mathcal{C}_2$. We observe how the addition of these two phase gates ties disconnected $\HC$ subgraphs together. Finally we apply our generalized understanding of $\mathcal{C}_2$ operators to extend beyond the set of stabilizer states, and generate $\HC$ orbits for non-stabilizer states.

\subsection{Quotient by Global Phase}\label{QuotientByPhase}

In this section, we define a procedure to quotient%
\footnote{Formally we are building the map $\mathcal{Q}: G \rightarrow G/N$, which takes elements of a group $G$ into a set of equivalence classes $G/N$. The set of equivalence classes is fixed by choice of congruence relation, e.g. congruence up to action by $\omega^n$.} %
by elements which act as a phase on the group. When building reachability graphs from Cayley graphs we always quotient first by the group element $\omega = (H_1P_1)^3$, as in Eq. \eqref{HPComm}, since quantum states can only be operationally distinguished up to global phase. Accordingly, all quotient groups $G/H$ we construct going forward are quotients by the product $\langle \omega\rangle \times H$.

We begin by explicitly building the quotient of $\mathcal{C}_1/\langle \omega \rangle$. As discussed in Section \ref{OneQubitSection}, the group $\mathcal{C}_1$ is generated by $\{H_1,\,P_1\}$ and contains $192$ elements. When quotienting by $\omega$, we identify together all elements of $\mathcal{C}_1$ that are equivalent up to powers of $\omega$. For $g_1,\,g_2 \in \mathcal{C}_1$,
\begin{equation}
    g_1 \equiv g_2 \textnormal{ if } g_1 = \omega^{n\Mod{8}}g_2.
\end{equation}
This identification defines the normal subgroup $\langle \omega \rangle \lhd \mathcal{C}_1$, where
\begin{equation}
    \langle \omega \rangle \equiv \{\mathbb{1},\,\omega,\,\omega^2,\,\omega^3,\,\omega^4,\,\omega^5,\,\omega^6,\,\omega^7\},
\end{equation}
and allows us to construct the quotient group $\bar{\mathcal{C}_1} \equiv \mathcal{C}_1/\langle \omega \rangle$. 

The quotient $\bar{\mathcal{C}_1}$ consists of $24$ equivalence classes of $8$ elements each. This is a factor of $8$ reduction in group order, from $192$ to $24$, as shown in the $\{H_1, \,P_1\}$ row of Table \ref{tab:OrbitLengthCliffordSubgroupNoRelations}. All elements of each class are equivalent up to powers of $\omega$. The $24$ equivalence classes can be represented by elements of the form
\begin{equation}\label{NonTrivEqClasses}
    \begin{split}
        \{p,\, pH_1p,\, H_1P_1^2H_1p\},
    \end{split}
\end{equation}
where $p \in \{\mathbb{1},\, P_1, \, P_1^2,\, P_1^3\}$ as defined in Eq. \eqref{PhaseGroup}.

Quotienting $\mathcal{C}_1$ by $\langle \omega \rangle$ likewise modifies the $\mathcal{C}_1$ Cayley graph by gluing together all vertices that represent operators in the same equivalence class. Figure \ref{C1Quotient} shows the Cayley graph of $\mathcal{C}_1$ before and after modding out by $\omega$. Each vertex in the $\bar{\mathcal{C}_1}$ graph represents $8$ elements of $\mathcal{C}_1$, collapsing the $192$ vertices of the $\mathcal{C}_1$ Cayley graph down to $24$. Every $H_1$ edge in the contracted graph represents the $8$ operators $\omega^nH_1$, and similarly for $P_1$.
    \begin{figure}[h]
        \centering        
        \includegraphics[width=13cm]{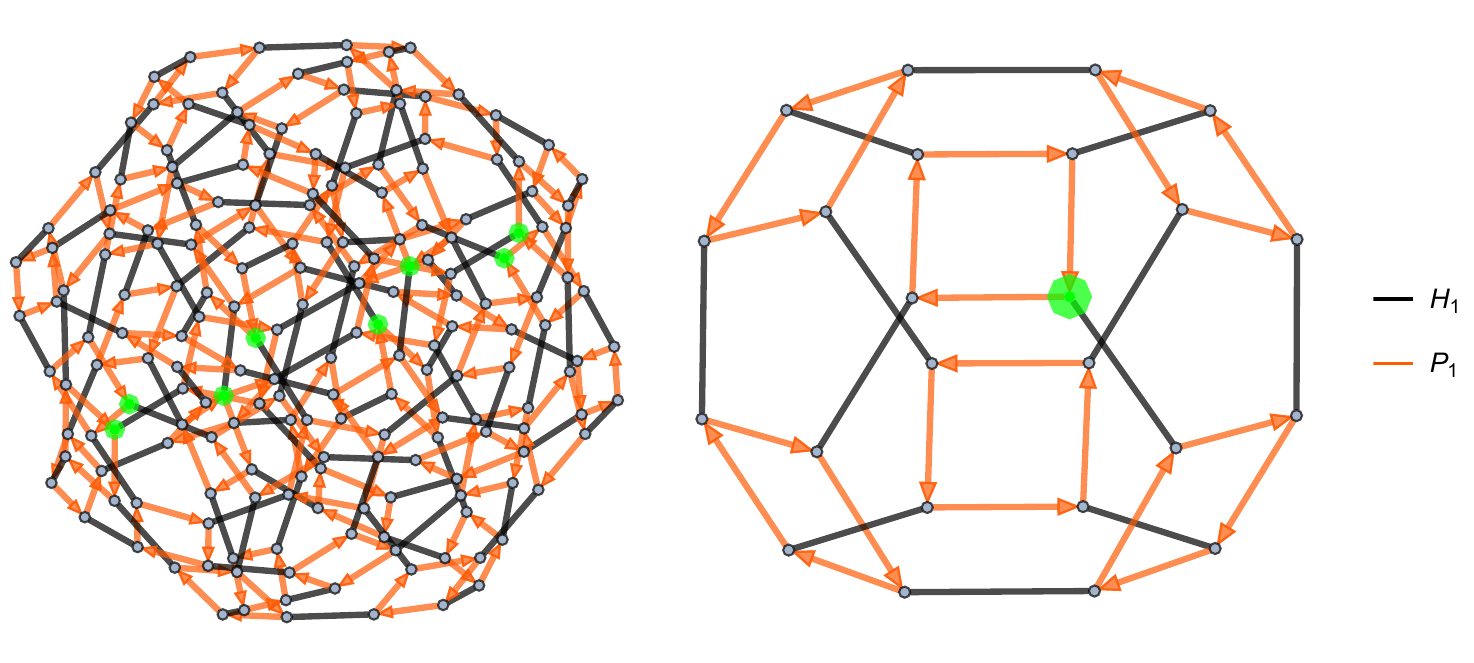}
        \caption{Cayley graph of $\mathcal{C}_1$  before and after quotienting by $\omega$. The $192$ vertices in the $\mathcal{C}_1$ Cayley graph collapse to $24$ vertices in the $\bar{\mathcal{C}_1}$ quotient graph. Every edge in the quotient graph represents the $8$ edges $\omega^nH_1$ and $\omega^nP_1$. One set of $8$ vertices, which are identified to a single vertex under this quotient, is highlighted in green.}
        \label{C1Quotient}
    \end{figure}

We have described a procedure for quotienting by $\omega$, which acts as a global phase. In the following sections, we will first quotient by $\omega$ when constructing reachability graphs. Quotienting a group by $\omega$ contracts the Cayley graph, by identifying vertices which represent elements equivalent up to $\omega$. As we highlight below, a similar graph contraction will yield reachability graphs as Cayley graph quotients, which we will highlight further in the following sections.

\subsection{Quotient by Stabilizer Subgroup}

In this section we show how to quotient a group by the stabilizer subgroup of a chosen state, and how to construct the state's reachability graph as a quotient space of the group Cayley graph. While Cayley graphs offer a state-independent description of a group, the orbit of a particular state under that group action is state-dependent. Our quotient procedure defines the collapse of a group Cayley graph into a subgraph which gives the reachability graph for a chosen state.

A state's reachability graph displays the evolution of that state under some chosen set of quantum gates. Since we are deriving each reachability graph from the Cayley graph of a group, this procedure can be applied to any\footnote{We do require that the state should be thought of as a state on $n$ qubits, i.e. with a fixed factorization in a fixed $2^n$-dimensional Hilbert space.} chosen quantum state on which the group acts.

The general procedure is to identify a group $G$ which acts on a Hilbert space $\Hil$, as well as a choice of generators for $G$. We first quotient $G$ by the global phase%
\footnote{For the general $G$, $\omega$ can be any element that acts as a root of unity times the identity operator.  For the Clifford subgroups we study here, $\omega$ will be an eighth root of unity.}
$\omega$, giving the quotient group $\bar{G} = G/\langle \omega \rangle$. Each element of $\bar{G}$ is isomorphic to the equivalence class $\omega^ng \in G$, for some $g \in G$. We then identify a state $\ket{\psi} \in \Hil$, which selects the stabilizer subgroup $\Stab_{\bar{G}}(\ket{\psi})$ that we will use to quotient $\bar{G}$. Since $\Stab_{\bar{G}}(\ket{\psi})$ is a normal subgroup, we can generate the quotient group $\bar{G}/\Stab_{\bar{G}}(\ket{\psi})$ by computing all cosets
\begin{equation}
    g\cdot\Stab_{\bar{G}}(\ket{\psi}) \quad \forall g \in \bar{G}.
\end{equation}

Constructing the quotient group above again generates a set of equivalence classes on $G$, with elements of each class congruent in their action on $\ket{\psi}$. To map elements between different equivalence classes we define the function $f: \Stab_{\bar{G}}(\ket{\psi}) \rightarrow \Stab_{\bar{G}}(\ket{\phi})$, where
\begin{equation}
    f(g) = hg^{-1}, \quad \forall \,g \in G_{\ket{\psi}},\, h \in G_{\ket{\phi}}.
\end{equation}
For example, to transform $P_1 \in \Stab_{\bar{\mathcal{C}_2}}(\ket{00})$ to $H_1P_1^2H_2P_2^2 \in \Stab_{\bar{\mathcal{C}_2}}(\ket{GHZ}_2)$ we apply the sequence $H_1P_1^2H_2P_2^2P_1^{-1}$. 

As an illustration of this procedure, we construct the quotient of $\mathcal{C}_1$ by $\langle \omega \rangle \times \Stab_{\mathcal{C}_1}(\ket{0})$. We first build the quotient group $\bar{\mathcal{C}}_1 = \mathcal{C}_1/\langle \omega \rangle$ as detailed in Section \ref{QuotientByPhase}. We then identify the stabilizer group $\Stab_{\mathcal{C}_1}(\ket{0}) \lhd \bar{\mathcal{C}}_1$, which comprises the $4$ elements that stabilize $\ket{0}$, i.e.
\begin{equation}
    \Stab_{\mathcal{C}_1}(\ket{0}) = \{\mathbb{1}, \,P_1, \,P_1^2, \,P_1^3\}.
\end{equation}
Quotienting $\bar{\mathcal{C}}_1$ by $\Stab_{\mathcal{C}_1}(\ket{0})$ then gives a set of $6$ equivalence classes, with a representative element from each class being
\begin{equation}\label{RepresentativesC1}
    \{\mathbb{1},\, pH_1,\, H_1P_1^2H_1\},
\end{equation}
with $p \in \{\mathbb{1},\, P_1, \, P_1^2,\, P_1^3\}$ as before. The elements of each equivalence class are identified by multiplying each representative in Eq. \eqref{RepresentativesC1} by the $4$ elements of $\Stab_{\mathcal{C}_1}(\ket{0})$. 

We build the graph corresponding to $\bar{\mathcal{C}}_1/\Stab_{\mathcal{C}_1}(\ket{0})$ by assigning a vertex to each of the $6$ equivalence classes. Figure \ref{VacuumQuotient} shows the $\mathcal{C}_1$ Cayley graph before and after modding by $\langle \omega \rangle \times \Stab_{\mathcal{C}_1}(\ket{0})$. The $192$-vertex $\mathcal{C}_1$ Cayley graph is reduced to a graph with $6$ vertices, which is isomorphic to the complete reachability graph for single-qubit stabilizer states.
\begin{figure}[h]
    \centering
		\begin{overpic}[width=12.5cm]{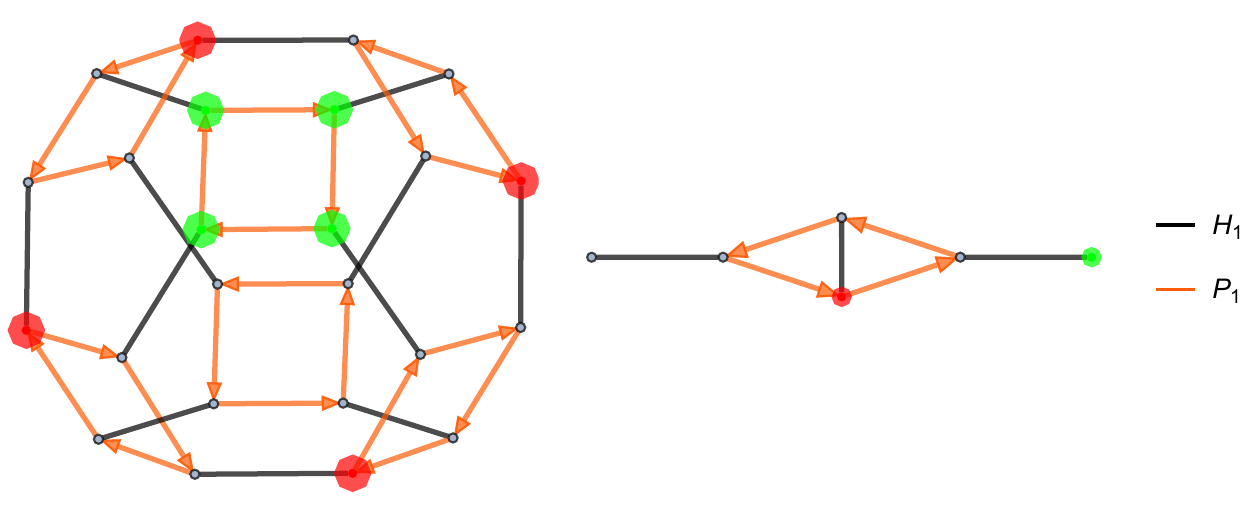}
		\put (79.5,22){\footnotesize{$\ket{0} \searrow$}}
		\put (67.5,14.5) {\footnotesize{$\nwarrow \ket{-i}$}}
        \end{overpic}
        \caption{Quotient of $\mathcal{C}_1$ by $\Stab_{\mathcal{C}_1}(\ket{0})$, the stabilizer subgroup of $\ket{0}$. The $\mathcal{C}_1$ Cayley graph on the left collapses to a $6$-vertex reachability graph on the right. Four green vertices identify to a single vertex representing the equivalence class of $\Stab_{\mathcal{C}_1}(\ket{0})$, while four red vertices likewise identify to a vertex for $\Stab_{\mathcal{C}_1}(\ket{-i})$.}
		\label{VacuumQuotient}
	\end{figure}

We have demonstrated a procedure for quotienting Clifford groups and subgroups by the stabilizer subgroup of a quantum state. We illustrated how Cayley graphs are contracted to state reachability graphs under this quotient. We will now use this protocol to explore subgroups of $\mathcal{C}_2$.

\subsection{Stabilizer Restricted Graphs from $\langle H_i,\,H_j,\,C_{i,j},\,C_{j,i} \rangle$ Quotients}\label{HhCcStabilizerQuotients}

In \cite{Keeler2022} we constructed and analyzed reachability graphs under the action of $\mathcal{C}_2$ subgroups. We termed these restricted graphs, and focused on the subgroup\\ $\HC \equiv \langle H_1,\,H_2,\,C_{1,2},\,C_{2,1}\rangle$. Since entanglement entropy among stabilizer states is modified by, at most, bi-local action, this subgroup gives useful insight into stabilizer entanglement. We now generalize the construction of restricted graphs in \cite{Keeler2022} by constructing the reachability graphs as quotient spaces of Cayley graphs. We specifically reproduce all $\HC$ restricted graphs that arise for stabilizer states, then use our model to explore the orbit of non-stabilizer states as well. 

The quotiented Cayley graphs we construct, in addition to representing a particular quotient group, are isomorphic to state reachability graphs. As defined in \cite{Keeler2022}, the vertices of reachability graphs represent states in a Hilbert space, while edges represent gates acting which transform these states. Vertices in the quotient space of a Cayley graph represent equivalence classes of group elements, defined by their orbit with respect to a chosen subgroup, while edges represent sets of generators. Going forward, we refer to Cayley graph quotients as reachability graphs and note the distinction when necessary.

All stabilizer states can reached by acting on $\ket{0}^{\otimes n}$ with $\mathcal{C}_n$. Acting on $\ket{0}^{\otimes n}$ with the subgroup $\HC$ generates $24$ stabilizer states, including all measurement states of the computational basis. Every state in the orbit of $\ket{0}^{\otimes n}$ is stabilized by $48$ elements of $\HC$.

Figure \ref{g24CayleyQuotient} shows the $\HC$ Cayley graph after quotienting by the stabilizer subgroup for $\ket{0}^{\otimes n}$, specifically for the $2$-qubit example $\ket{00}$. Since the stabilizer subgroup is preserved when tensoring on additional qubits to the system, this $24$-vertex graph likewise displays the orbit for any product of $\ket{00}$ with any $(n-2)$-qubit state. Each vertex in Figure \ref{g24CayleyQuotient} represents the stabilizer subgroup for one state in the orbit of $\ket{00}$ under $\HC$.
    \begin{figure}[h]
    \centering
        \begin{overpic}[width=12.5cm]{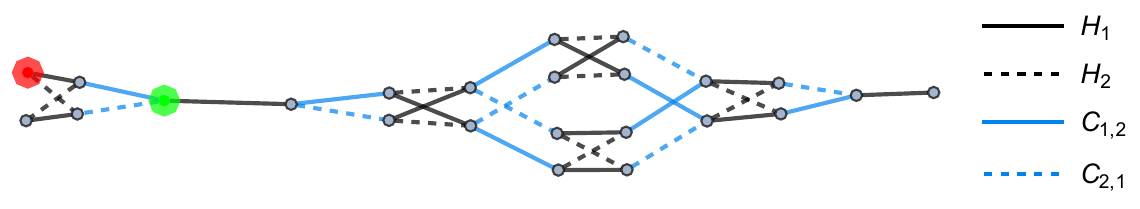}
		  \put (3.25,14) {\footnotesize{$\swarrow \ket{00}$}}
		  \put (15.2,6) {\footnotesize{$\nwarrow \ket{GHZ}_2$}}
        \end{overpic}
    \caption{Quotient space of $\HC$ Cayley graph after modding out by the stabilizer subgroup of $\ket{00}$. The equivalence class of $\Stab_{\HC}(\ket{00})$ is highlighted in red, while the equivalence class of $\Stab_{\HC}(\ket{GHZ}_2)$ is highlighted in green.}
    \label{g24CayleyQuotient}
    \end{figure}

In general any product state, as well as states with no entanglement between the first two qubits and the other $n-2$, will have either the $24$-vertex reachability graph in Figure \ref{g24CayleyQuotient} or the $36$-vertex reachability graph, also defined in \cite{Keeler2022}, that appears below in Figure \ref{HhCcPhaseOverlay2}. 

For additional entangled states which arise at higher qubit number, new reachability graph structures appear when acting with $\HC$. Figure \ref{g144CayleyQuotient} shows a quotient space of the $\HC$ Cayley graph after modding out by the stabilizer subgroup for $\ket{GHZ}_3 \equiv \ket{000} + \ket{111}$. The orbit of $\ket{GHZ}_3$ under $\HC$ reaches $144$ states, each of which is stabilized by $8$ elements of $\HC$.
\begin{figure}[h]
    \centering
		\begin{overpic}[width=11cm]{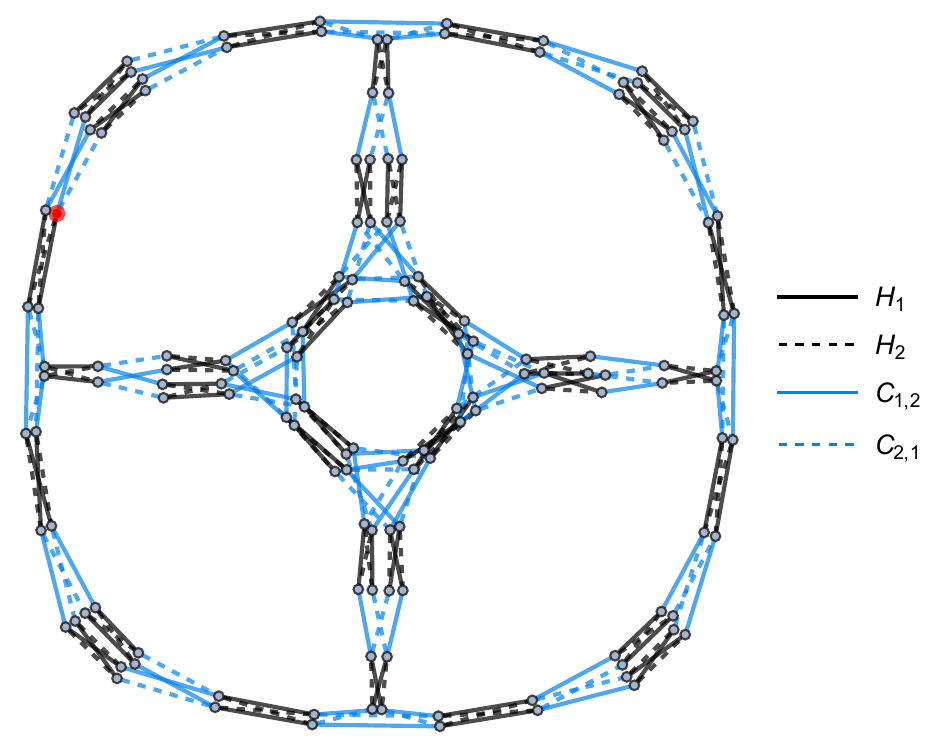}
		\put (7.5,56.6) {\footnotesize{$\leftarrow \ket{GHZ}_3$}}
        \end{overpic}
        \caption{Quotient of $\HC$ Cayley graph after modding by stabilizer subgroup for $\ket{GHZ}_3$. Red vertex gives the equivalence class of $\HC$ elements that stabilize $\ket{GHZ}_3$.}
    \label{g144CayleyQuotient}
	\end{figure}

Also at three qubits, there exist stabilizer states which are stabilized by only $4$ elements of $\HC$. In Section $4$ of \cite{Keeler2022}, we defined a lifting procedure which allows us to find an example state in each stabilizer reachability graph. To identify a state that is stabilized by $4$ elements of $\HC$, we act with $C_{3,2}$ on the product state $\ket{i} \otimes \ket{1} \otimes \ket{+}$. The resultant state $\ket{010} + i\ket{011} + \ket{100} + i\ket{101}$ is stabilized by the elements
\begin{equation}\label{g288StabGroup}
    \{\mathbb{1},\, H_2(C_{1,2}H_1)^4,\, (C_{1,2}H_1)^4H_2,\, \left((C_{1,2}H_1)^3C_{1,2}H_2\right)^2\}.
\end{equation}
Figure \ref{g288CayleyQuotient} shows the quotient space of $\HC$ after modding out by the stabilizer subgroup for $C_{3,2}\ket{i1+} \equiv \ket{010} + i\ket{011} + \ket{100} + i\ket{101}$.
\begin{figure}[h]
    \centering
		\begin{overpic}[width=11cm]{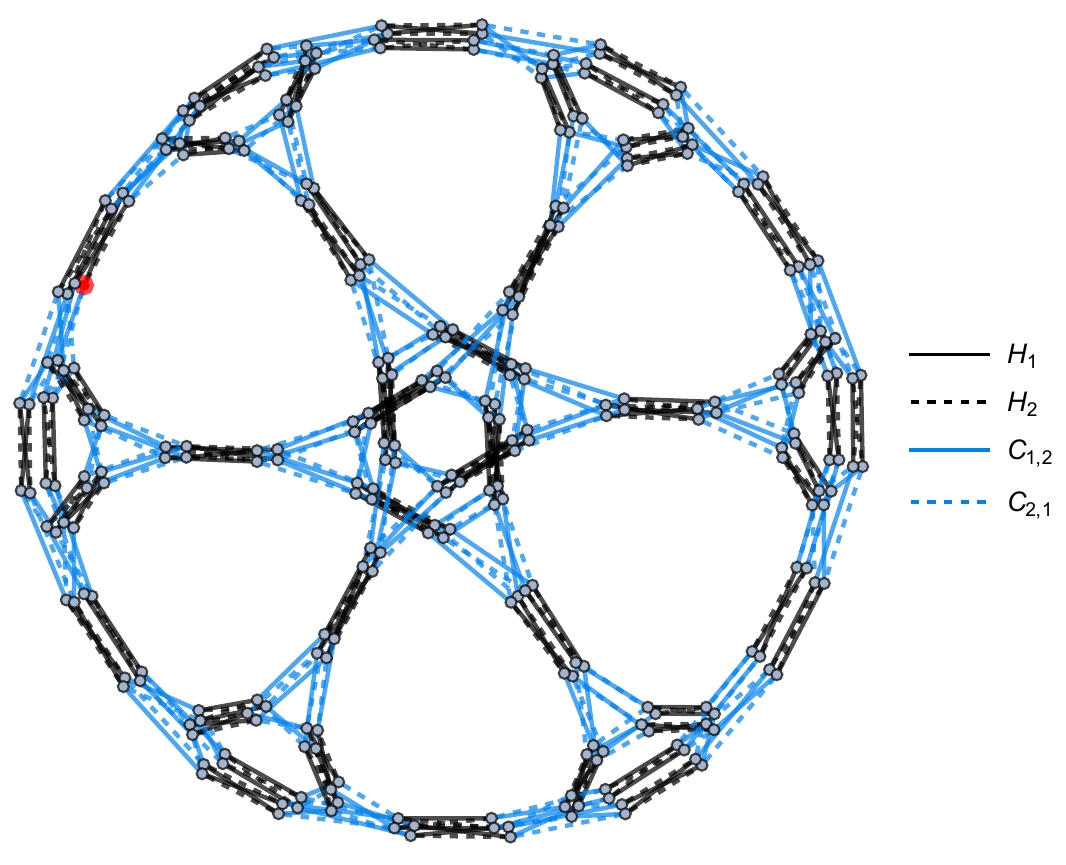}
		\put (9,53.5) {\footnotesize{$\leftarrow C_{3,2}\ket{i1+}$}}
        \end{overpic}
        \caption{Orbit of stabilizer states which are stabilized by $4$ elements of $\HC$. This graph has different topology from the $288$-vertex graph in Figure \ref{D31HhCcGraph}. Elements of $\HC$ that stabilize $C_{3,2}\ket{i1+} \equiv \ket{010} + i\ket{011} + \ket{100} + i\ket{101}$ are represented by the red vertex.}
    \label{g288CayleyQuotient}
	\end{figure}

Finally, there are stabilizer states which are only stabilized by $\mathbb{1}$ in $\HC$. Figure \ref{g1152CayleyQuotient} illustrates the $1152$-vertex reachability graph for such states. This graph represents the largest possible orbit of any quantum state under $\HC$, since all states are trivially stabilized by $\mathbb{1}$. Figure \ref{g1152CayleyQuotient} is first observed at four qubits. 
    \begin{figure}[h]
        \centering
        \includegraphics[width=11cm]{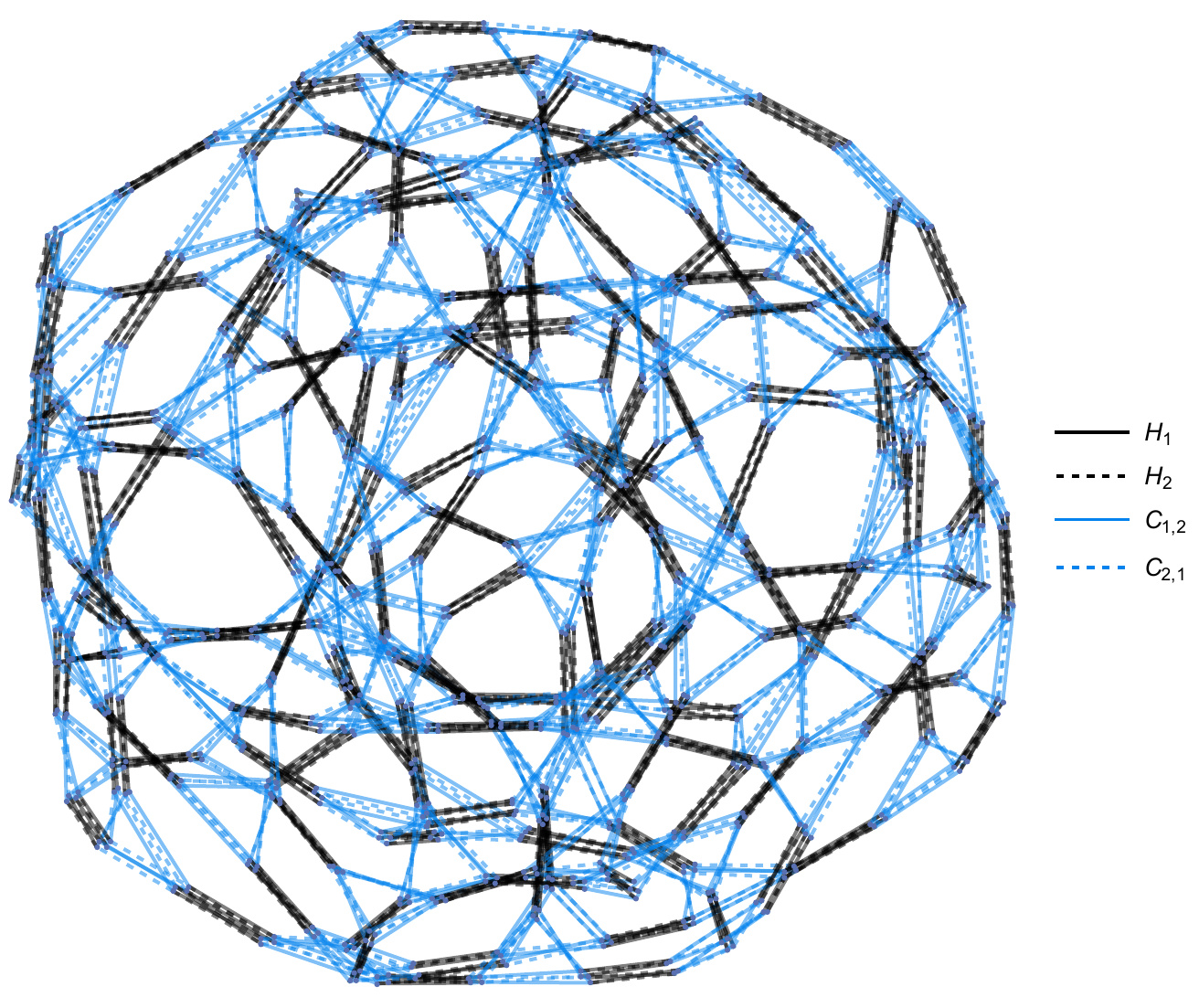}
        \caption{Orbit for states stabilized by only $\mathbb{1}$ in $\HC$. This $1152$-vertex graph gives the orbit of a generic quantum state under action of $\HC$.}
    \label{g1152CayleyQuotient}
    \end{figure}

By taking quotients of the $\HC$ Cayley graph, we have reproduced all stabilizer reachability graphs found in \cite{Keeler2022} under the action of this subgroup. We demonstrated that the largest such subgraph contains $1152$ vertices, by representing the orbit of states which are stabilized by only the identity in $\HC$, in agreement with \cite{Keeler2022}. In the following subsection we add $P_1$ and $P_2$ back into our generating set, and study the action of the full group $\mathcal{C}_2$. We will show how the addition of these two phase gates does not generate any additional graphs, but instead connects the existent structures studied above. We also discover new reachability graphs which arise from quotienting $\HC$ by stabilizer subgroups of non-stabilizer states.

\subsection{Full $\mathcal{C}_2$ Action}

Adding $P_1$ and $P_2$ to the set $\{H_1,\, H_2,\, C_{1,2},\, C_{2,1}\}$ generates the full group $\mathcal{C}_2$, which contains $92160$ elements. The Cayley graph for $\mathcal{C}_2$ after quotienting by $\langle \omega \rangle$ will accordingly have $11520$ vertices, as seen in the last four lines of Table \ref{tab:OrbitLengthCliffordSubgroupNoRelations}. Similarly the reachability graph for any $n$-qubit state stabilized by only $\mathbb{1}$ in $\mathcal{C}_2$ will have $11520$ vertices. By first considering the action of $\langle H_1,\, H_2,\, C_{1,2},\, C_{2,1}\rangle$ on a set of states, followed by the action of $P_1$ and $P_2$, we observe how the reachability graphs from Section \ref{HhCcStabilizerQuotients} are connected. 

To illustrate how phase gates tie $\langle H_1,\, H_2,\, C_{1,2},\, C_{2,1}\rangle$ reachability graphs together, we first consider the orbit of $\ket{0}^{\otimes n}$ shown in Figure \ref{g24CayleyQuotient}. Acting with $P_1$ and $P_2$ on all states in this orbit connects the $24$-vertex reachability graph to a $36$-vertex graph, as in Figure \ref{HhCcPhaseOverlay2}. These two graphs combine to give the orbit of any pure state under the action of $\mathcal{C}_2$, as well as any state%
\footnote{Figure \ref{HhCcPhaseOverlay2} actually shows the orbit of any $n$-qubit state with no entanglement between one pair of qubits and the remaining $n-2$ qubits, since qubits $1$ and $2$ can be exchanged, without loss of generality, with any qubits in both the state and the Clifford subgroup.}
with only entanglement among its first two qubits. In both Figure \ref{HhCcPhaseOverlay2} and Figure \ref{HhCcPhaseOverlay} we have removed all ``trivial loops'', that is all edges which map a vertex back to itself, as these loops represent a stabilizing action on the vertex.
    \begin{figure}[h]
        \centering
        \begin{overpic}[width=11cm]{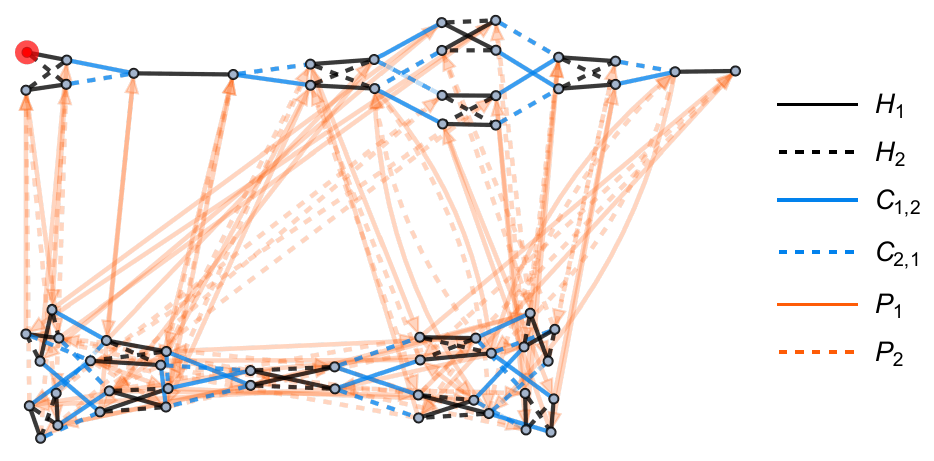}
            \put (4,45.5) {\footnotesize{$\swarrow \ket{00}^{\otimes n}$}}
        \end{overpic}
        \caption{Acting with $P_1$ and $P_2$ on all states in the $\HC$ orbit of $\ket{00}\otimes \ket{\psi}_{n-2}$ connects the $24$-vertex reachability graph from Figure \ref{g24CayleyQuotient} to a graph of $36$ vertices. Together these two graphs show the $\mathcal{C}_2$ orbit of any product state, and all states with no entanglement between the first two qubits and the remaining $n-2$ qubits.}
    \label{HhCcPhaseOverlay2}
    \end{figure}

Similarly at higher qubit number, phase gates on the first two qubits tie together the larger $\langle H_1,\, H_2,\, C_{1,2},\, C_{2,1}\rangle$ reachability graphs. Acting with $P_1$ and $P_2$ on states in the stabilizer $288$-vertex graph, seen in Figure \ref{g288CayleyQuotient}, will sometimes act trivially, sometimes map the state to another in the $288$-vertex graph, and sometimes map it to one of three $144$-vertex graphs. Figure \ref{HhCcPhaseOverlay} depicts how these four graphs are connected via phase operations, where again trivial loops have been removed. 
    \begin{figure}[h]
        \centering
        \includegraphics[width=12.5cm]{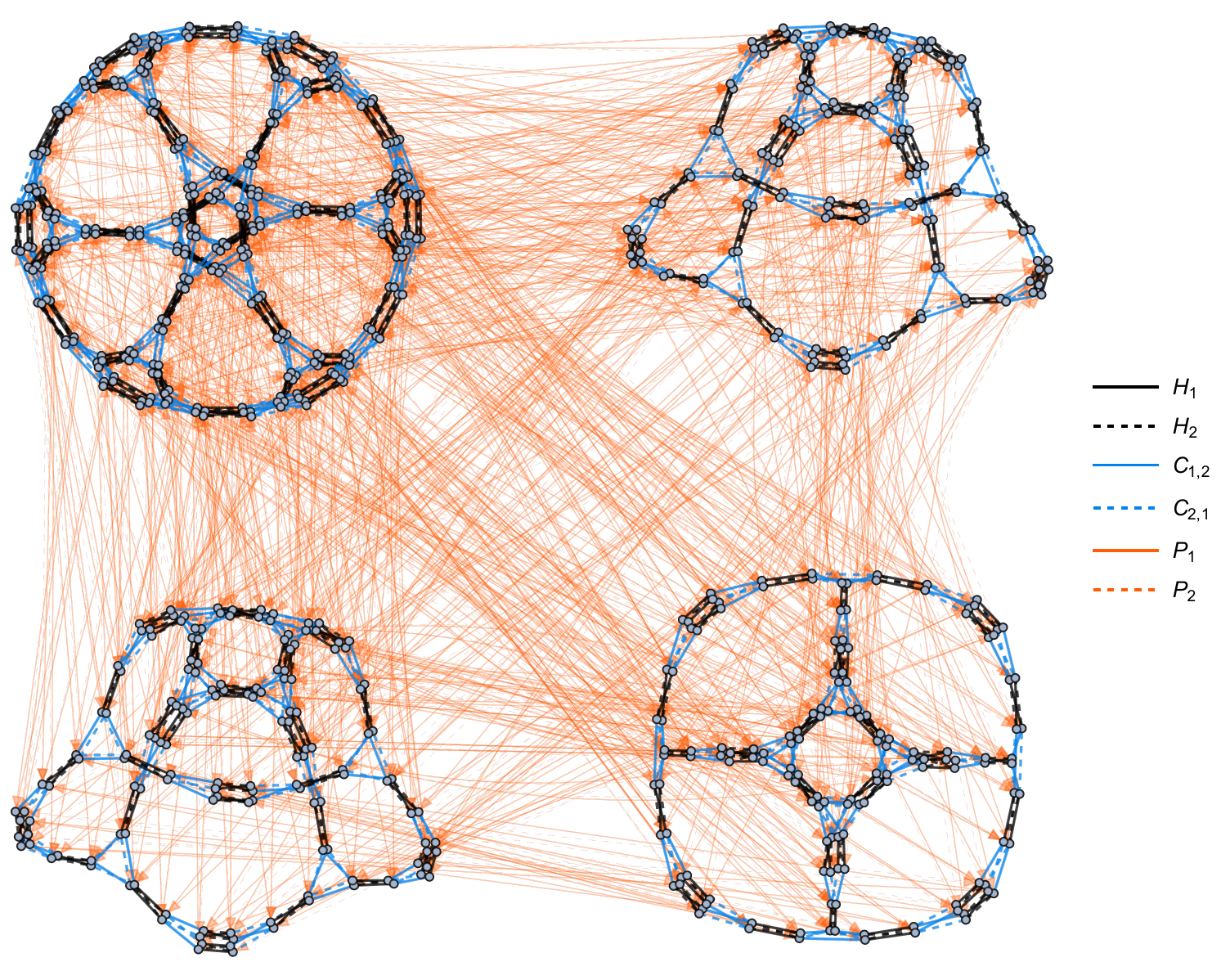}
        \caption{Three copies of the $144$-vertex reachability graph in Figure \ref{g144CayleyQuotient} connect to a single copy of the $288$-vertex graph in Figure \ref{g288CayleyQuotient}, after acting with $P_1$ and $P_2$.}
    \label{HhCcPhaseOverlay}
    \end{figure}

The largest reachability graph under the action of $\HC$ contains $1152$ vertices, and is depicted in Figure \ref{g1152CayleyQuotient}. This reachability graph, which we term $g_{1152}$, gives the orbit of states which are stabilized by only the identity in $\HC$. Acting with $P_1$ and $P_2$ on every state in $g_{1152}$ either connects the $1152$-vertex graph to itself, or maps to one of its $9$ isomorphic copies. Figure \ref{Connected1152} shows how these $10$ copies of $g_{1152}$ are symmetrically attached via phase operations. Upon acting with $P_1$ and $P_2$ the resulting structure forms a completely-connected graph of $10$ vertices, where each vertex actually represents a $g_{1152}$ graph. 
\begin{figure}[h]
    \centering
 \includegraphics[width=8cm]{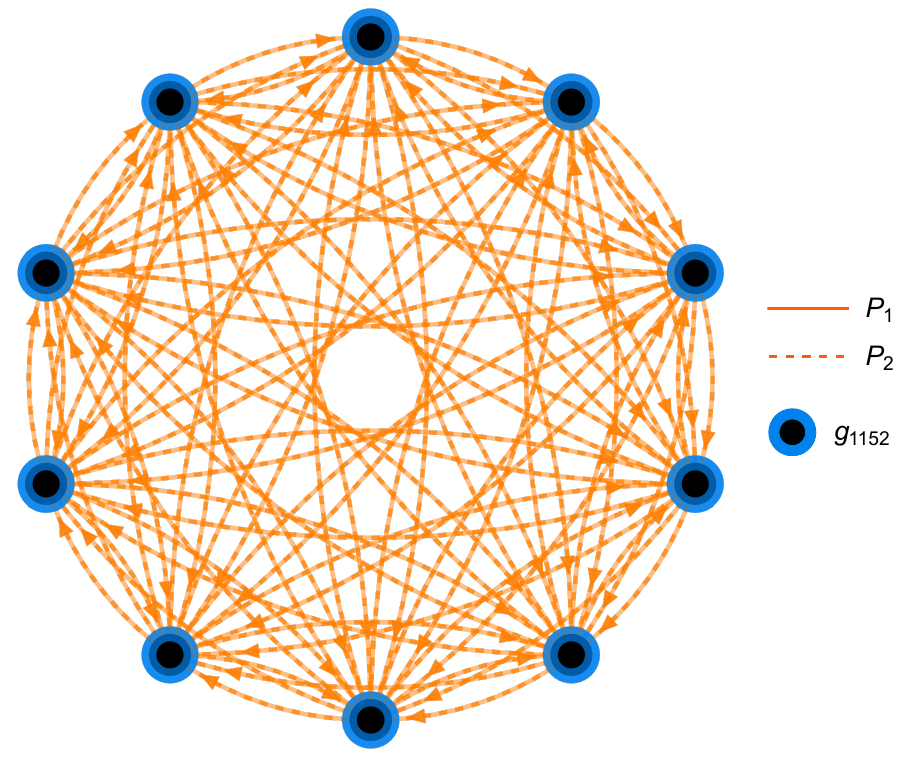}
    \caption{Acting $P_1$ and $P_2$ on states in $g_{1152}$ attaches the graph to $9$ isomorphic copies of itself, forming a completely-connected $10$-vertex graph. Each vertex represents one copy of $g_{1152}$ (Figure \ref{g1152CayleyQuotient}) and every edge is a set of phase gates which connect $g_{1152}$ graphs.}
    \label{Connected1152}
\end{figure}

We have examined the full action of $\mathcal{C}_2$ by acting with $P_1$ and $P_2$ on states in $\HC$ orbits. We demonstrated how the addition of these two phase gates ties together the $\HC$ reachability graphs shown in Section \ref{HhCcStabilizerQuotients}. Specifically we observed how the $24$-vertex reachability graph in Figure \ref{g24CayleyQuotient} connects to another graph of $36$ vertices. Meanwhile, three copies of the $144$-vertex graph in Figure \ref{g144CayleyQuotient} connect to a single copy of the $288$-vertex graph in Figure \ref{g288CayleyQuotient}. Finally, the largest $1152$-vertex reachability graph connects to $9$ isomorphic copies of itself under the action of $P_1$ and $P_2$. In the following section, we move beyond the set of stabilizer states and consider the action of $\langle H_1,\,H_2,\,C_{1,2},\,C_{2,1}\rangle$ on some notable non-stabilizer states.

\subsection{Non-Stabilizer Quotients}

Our state-independent description of the Clifford group allows us to examine the action of $\mathcal{C}_n$ on states which are not stabilizer states. For a quantum information theorist, the term ``stabilizer states'' typically refers to the set of $n$-qubit quantum states which are stabilized by a $2^n$-element subset of the Pauli group. There exist, however, states which are not stabilizer states, but are stabilized by additional Clifford group elements besides $\mathbb{1}$. These states likewise admit reachability graphs through our quotient procedure, and their graph properties reflect their distinction from the set of stabilizer states. Below we give a few examples of reachability graphs for notable non-stabilizer states, under the action of $\langle H_1,\,H_2,\,C_{1,2},\,C_{2,1}\rangle$, and contrast their structure with the stabilizer state graphs.

The $n$-qubit $W$-state holds particular interest as a highly-entangled, non-biseparable quantum state \cite{Dur2000,Schnitzer:2022exe}. Defined as
    \begin{equation}
        \ket{W}_n \equiv \left(\ket{100...00} + \ket{010...00} + ... + \ket{000...01} \right),
    \end{equation}
 $\ket{W}_n$ is famously not a stabilizer state when $n \geq 3$. However, $\ket{W}_n$ is stabilized by more than just $\mathbb{1}$ in $\mathcal{C}_n$. Even considering just the action of $\HC$, $\ket{W}_n$ is stabilized by the four elements
\begin{equation}\label{W3StabGroup}
    \{\mathbb{1},\, H_2C_{1,2}H_2,\, H_1C_{1,2}H_2C_{2,1},\, H_2C_{1,2}C_{2,1}C_{1,2}H_1\}.
\end{equation}

Figure \ref{D31HhCcGraph} shows the orbit of $\ket{W}_n$ under the action of $\HC$. The stabilizer subgroup of $\ket{W}_n$, in Eq. \eqref{W3StabGroup}, is isomorphic to all other stabilizer subgroups in the orbit seen in Figure \ref{D31HhCcGraph}. The stabilizer group of $\ket{W}_n$ is not, however, isomorphic to any subgroup in the orbit of the stabilizer state group in Eq. \eqref{g288StabGroup}. Consequently, while the reachability graph of $\ket{W}_n$ under $\HC$ contains $288$ vertices, its structure is distinctly different from the stabilizer state graph seen in Figure \ref{g288CayleyQuotient}.
    \begin{figure}[h]
        \centering    
        \begin{overpic}[width=11.5cm]{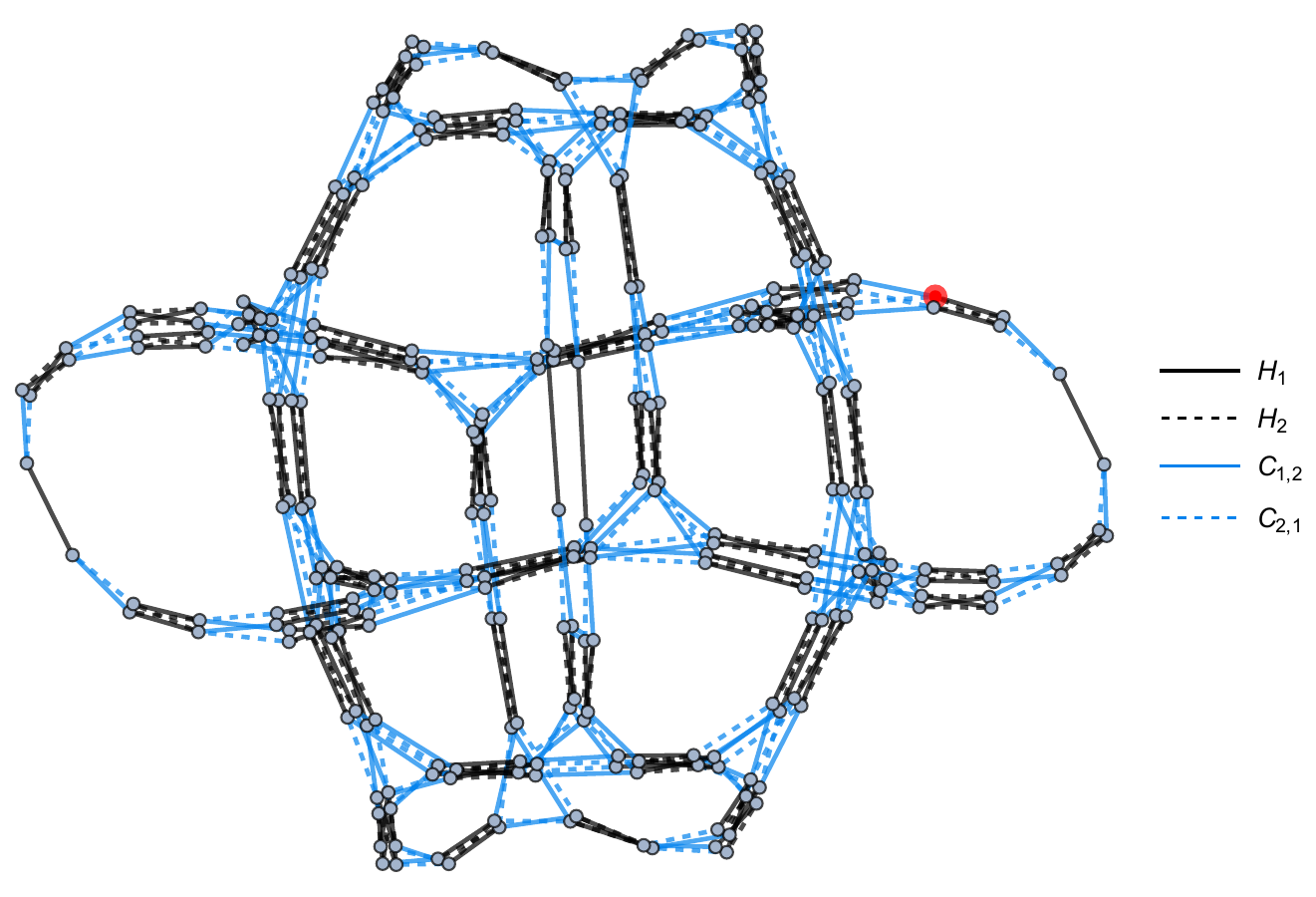}
            \put (72,47.5) {\footnotesize{$\swarrow \ket{W}_3$}}
        \end{overpic}
        \caption{Quotient space of $\HC$ Cayley graph after modding by $\ket{W}_n$ stabilizer subgroup. This reachability graph is not isomorphic to the $288$-vertex subgraph seen for stabilizer states in Figure \ref{g288CayleyQuotient}.}
    \label{D31HhCcGraph}
    \end{figure}

Another notable set of non-stabilizer states are the $n$-qubit Dicke states \cite{Dicke1954,Baertschi2019}. Dicke states are equal superpositions of $n$-qubit basis states with Hamming weight $k$, defined
\begin{equation}
    \ket{D^n_k} \equiv \binom{n}{k}^{-1/2} \sum_{b \in \{0,1\}^n, \hspace{.1cm} h(b) = k} \ket{b},
\end{equation}
where $h(b)$ is the standard Hamming weight for binary stings. 

While $\ket{D^n_k}$ is not a stabilizer state%
\footnote{The state $\ket{D^n_n} = \ket{1}^{\otimes n}$ is a stabilizer state and its reachability graph is given in Figure \ref{g24CayleyQuotient}. Additionally, states $\ket{D^n_1} = \ket{W}_n$ and $\ket{D^n_{n-1}}$ have reachability graphs as shown in Figure \ref{D31HhCcGraph}.}
for all $n \neq k$ and $n >2$, every $\ket{D^n_k}$ is stabilized by more than $\mathbb{1}$ in $\HC$. States $\ket{D^n_k}$ where $1 < k < n-1$ are stabilized by exactly two elements of $\HC$, namely
\begin{equation}\label{D42StabGroup}
    \{\mathbb{1},\, H_2C_{1,2}C_{2,1}C_{1,2}H_1\}.
\end{equation}

Figure \ref{D42HhCcGraph} shows an example reachability graph for the state $\ket{D^4_2}$. Since $\ket{D^4_2}$ is only stabilized by the two elements in Eq. \eqref{D42StabGroup}, its orbit under $\HC$ reaches $576$ states. Graphs with $576$ vertices are never observed among stabilizer states at any qubit number.
    \begin{figure}[h]
        \centering
        \begin{overpic}[width=10cm]{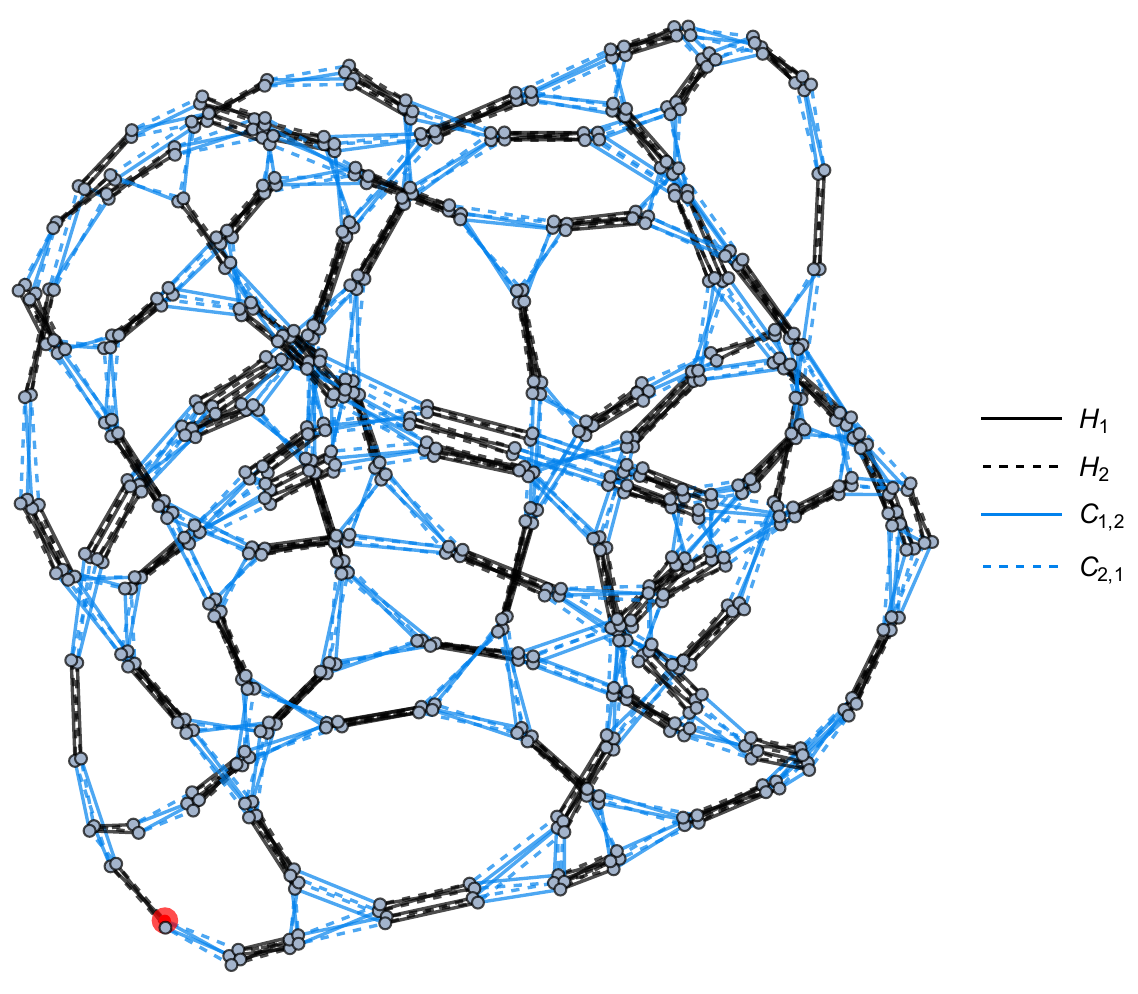}
            \put (1,5.7) {\footnotesize{$\ket{D^4_2}\rightarrow$}}
        \end{overpic}
        \caption{Quotient space of $\HC$ Cayley graph after modding by the stabilizer subgroup of $\ket{D^n_k}$, where $1 < k < n-1$. This reachability graph contains $576$ vertices, an orbit size not observed among the set of stabilizer states.}
    \label{D42HhCcGraph}
    \end{figure}

We have displayed the orbits of non-stabilizer states under the action of $\HC$, specifically for certain states which are stabilized by more than just the identity. We observed reachability graphs with vertex count not seen among the set of stabilizer states. Additionally, we identified graphs with vertex count shared with stabilizer states, but possessing a different topology. The orbits of Dicke states and their entanglement properties are studied in detail in \cite{Munizzi2023}.

\section{Discussion}\label{sec:discussion}

In this work we have presented a generalized construction for reachability graphs, defining them as quotient spaces of Cayley graphs. We began by constructing a presentation for $\mathcal{C}_1$ and $\mathcal{C}_2$, which we used to highlight the non-trivial relations among Clifford group elements. These relations allowed us to understand the structure of $\mathcal{C}_1$ and $\mathcal{C}_2$, and to explicitly construct all subgroups built from a restricted set of generators. Our operator-level, state-independent construction allowed us to obtain constraints on the evolution of any state through Clifford circuits. 

Extending our construction to higher qubit Clifford groups would require the definition of new relations  with each increase in qubit number. Intriguingly, extending our presentation for $\mathcal{C}_2$ to a presentation for $\mathcal{C}_3$ only requires the addition of $4$ relations \cite{Selinger2013}, none of which involve the phase gate.

After building their presentations, we studied the Cayley graphs for $\mathcal{C}_1$ and $\mathcal{C}_2$, as well as for all $\mathcal{C}_2$ subgroups generated by a subset of Hadamard, phase, and CNOT gates. Our protocol contracts the Cayley graph, yielding a quotient graph that is isomorphic to the state's reachability graph. Specifically, we quotient by the stabilizer subgroup for a particular state, ensuring only non-trivial group action remains. Using this procedure, we can analyze the evolution of a state through circuits comprised of the given gate set. Since we begin with the state-independent Cayley graph, this quotient protocol and analysis can be applied to any quantum state.

We emphasize that the techniques put forth in this paper are not limited to Clifford circuits. Any finite gate set can be represented by a discrete Cayley graph. The Cayley graph can then be made finite by imposing a cutoff on graph distance, which constrains the depth of a circuit. Accordingly, the program established in this paper could be used to study even universal gate sets in quantum computation, up to a fixed circuit depth. Our techniques could furthermore be straightforwardly extended to computation with qutrits or qudits.

Access to a graph-theoretic description of evolution through quantum circuits allows for direct calculation of some interesting circuit properties. For example, the gate complexity of a given circuit which transforms one state into another is precisely the minimum graph distance separating the two vertices that represent each state. The Cayley graph diameter hence immediately bounds the maximal change in complexity that can be observed under the constituent gate set. Additionally, given a fixed set of universal generators, one could compute complexity growth for circuits of varying depths. Conversely, one could fix a circuit depth and consider the growth of complexity under alternative sets of universal generators \cite{Munizzi_2022}. It would be interesting to relate the discrete picture of gate complexity obtained here to a more continuous picture, as in \cite{Balasubramanian:2019wgd}.

The graph analysis in this paper might be useful to better understand circuit architecture and reduce resource overhead in a quantum computation framework. The relations in our presentation often describe non-trivial, and sometimes unexpected, equivalences between sequences of quantum gates. In many cases, large-depth circuits containing strings of gates which are difficult to implement can be reduced to sequences of shorter depth, and simpler gate composition. One such example is $C_{i,j}H_jC_{i,j}P_jC_{i,j}P_j^3H_j = P_i$, where a circuit of $9$ gates, including numerous (and resource-expensive) $CNOT$ insertions, can be optimized to a single phase gate. Similarly, in the context of state preparation, an optimal circuit to transform an initial state into some desired final state can be identified by the appropriate extremal graph path. If computational or experimental constraints exist that limit the set of viable gates, corresponding edges in the graph can be modified or removed to accommodate this restriction. This analysis could be particularly interesting in the context of near-term quantum computing, where it is often easier to implement some specified set of gates than arbitrary two-body couplings.

In this work we focused on the group $\HC = \langle H_1,\,H_2,\,C_{1,2},\,C_{2,1} \rangle \subset \mathcal{C}_2$, which offered useful insight into the bipartite entanglement generated by Clifford circuits. Using our quotient procedure for the $\HC$ we were able to recover all stabilizer reachability graphs from our last paper, Figures \ref{g24CayleyQuotient}--\ref{g1152CayleyQuotient}, as well as reachability graphs for some non-stabilizer states. In particular, we showed that the $1152$-vertex graph is the largest reachability graph for any state, stabilizer or otherwise. We believe we have exhibited all reachability graphs involving $\mathcal{C}_2$ for stabilizer states, but proving so would require a deeper understanding of the relation between the Pauli stabilizer groups and Clifford stabilizer groups for a given state.

Instead of $\HC$, we could consider state orbits under alternative $\mathcal{C}_2$ subgroups. Orbits under different $\mathcal{C}_2$ subgroups should also decompose the full $\mathcal{C}_2$ reachability graph into disconnected pieces, similarly to the situation for $\HC$.

Additional stabilizer subgroups exist which quotient $\mathcal{C}_n$ and $\HC$, distinct from the stabilizer subgroups of individual stabilizer states. A $2$-element subgroup of $\HC$ stabilizes all Dicke states with a certain structure, building the reachability graph seen in Figure \ref{D42HhCcGraph}. Furthermore, all states we examined from which magic can be fault-tolerantly distilled \cite{bravyi2005universal} are stabilized by more than just $\mathbb{1}$ in $\mathcal{C}_n$. We conjecture that some measure of ``stabilizerness'', similar to stabilizer rank or mana \cite{White:2020zoz}, can be defined using the order of state's stabilizer subgroup in $\mathcal{C}_n$. 

We initiated this study to explore the evolution of entanglement entropy through Clifford circuits. Since entanglement in Clifford circuits can only be modified through $CNOT$ action, the number of $C_{i,j}$ edges in a reachability graph, which we term the ``CNOT diameter'', weakly bounds the number of times the entropy vector can change. However, our deeper exploration of the Clifford group revealed that not every $C_{i,j}$ gate modifies entropic structure, since relations like $(H_1C_{2,1})^4 = P_1^2$ demonstrate that some circuits with CNOT gates nonetheless never modify entanglement. In forthcoming work \cite{Keeler2023b}, we will build ``contracted graphs'' which exhibit how entropy vectors can change within a given reachability graph.

\textit{The authors thank ChunJun Cao, Temple He, William Kretschmer, Daniel Liang, Julien Paupert, David Polleta, and Claire Zukowski for useful discussions.  CAK and WM are supported by the U.S. Department of Energy under grant number DE-SC0019470 and by the Heising-Simons Foundation ``Observational Signatures of Quantum Gravity'' collaboration grant 2021-2818. JP is supported by the Simons Foundation through the It from Qubit: Simons Collaboration on Quantum Fields, Gravity, and Information.}

\clearpage
\begin{singlespace}
\printbibliography[heading=subbibliography]
\end{singlespace}

\chapter{ENTROPY CONES AND ENTANGLEMENT EVOLUTION FOR DICKE STATES}\label{Chapter5}

\textit{The contents of this chapter were originally published in Physical Review A \cite{Munizzi:2023ihc}.}

\textit{The $N$-qubit Dicke states $\ket{D^N_k}$, of Hamming-weight $k$, are a class of entangled states which play an important role in quantum algorithm optimization. We present a general calculation of entanglement entropy in Dicke states, which we use to describe the $\ket{D^N_k}$ entropy cone. We demonstrate that all $\ket{D^N_k}$ entropy vectors emerge symmetrized, and use this to define a min-cut protocol on star graphs which realizes $\ket{D^N_k}$ entropy vectors. We identify the stabilizer group for all $\ket{D^N_k}$, under the action of the $N$-qubit Pauli group and two-qubit Clifford group, which we use to construct $\ket{D^N_k}$ reachability graphs. We use these reachability graphs to analyze and bound evolution of $\ket{D^N_k}$ entropy vectors in Clifford circuits.}

\section{Introduction}

There are numerous ways to classify sets of quantum states. For states in a factorizable Hilbert space $\Hil = \bigotimes^N_i \Hil_i$, one such classification is given by considering the entropy vector of each state \cite{Bao2015}. The entropy vector of a pure state $\ket{\psi} \in \Hil$ is constructed as the ordered set of all $2^N-1$ von Neumann entropies, computed by tracing out all tensor product states $\ket{\psi_i} \in \Hil_i$. While every $\ket{\psi} \in \Hil$ can be assigned an entropy vector, specifying an entropy vector does not uniquely determine a state. Instead, entropy vectors describe an equivalence relation on states in $\Hil$, assigning each state to a class based on its entanglement structure.

The space of allowed entropy vectors for a specific class of states defines the entropy cone for that class \cite{Schnitzer:2022exe,Linden2013,Bao2020,Bao2020a}. Since a particular entanglement structure often accompanies other interesting state characteristics, identifying specific entropy cones which contain the entropy vectors for different state classifications has received significant research interest. One notable case is that of holographic states, those quantum states with a smooth classical dual geometry via the AdS/CFT correspondence \cite{Maldacena:1997re,Witten:1998qj}, which possess subsystem entanglement that obeys the Ryu-Takayanagi formula \cite{Ryu:2006bv,Hubeny:2007xt}. The entropy vectors of holographic states are likewise confined to a convex polyhedral subspace of the ambient vector space, known as the holographic entropy cone. In general, having an entropy vector that is contained within a particular entropy cone is a necessary, though not sufficient, condition for a state to belong to the class described by that cone.

The $N$-qubit Dicke states are a particular class of entangled quantum state which have received notable recognition for application in quantum algorithm development and precision measurement. Dicke states have found significant use as the initial state for Quantum Approximate Optimization Algorithm (QAOA) implementation, a quantum algorithm designed to approximate combinatorial-optimization solutions \cite{farhi2014quantum}. More recently, techniques have been established to deterministically prepare $N$-qubit Dicke states using circuits of depth $\mathcal{O}(n)$ \cite{Baertschi2019}. Much of the NISQ-era utility of Dicke states for quantum computation is due to their unique entanglement properties. Certain highly-entangled Dicke states may be projected, via measurement, onto distinct lower-qubit states that cannot be locally-transformed into each other.

An alternative, and useful, classification on $\Hil$ considers how state sets transform under the action of a group. Consider a group $G \in L(\Hil)$, which transforms states $\ket{\psi} \in \Hil$. Certain elements of $G$ may act trivially on some $\ket{\psi}$, mapping the state to itself. Such elements define the stabilizer group for $\ket{\psi}$ under the action of $G$, which we denote $\Stab_G(\ket{\psi})$. The most generic quantum states are stabilized by only the identity operator $\mathbb{1}$ in a given $G$. Special sets of states, however, have larger stabilizer groups with respect to certain sets of operators. A well-known example is the set of stabilizer states, a class of classically-simulable quantum states, which are stabilized by a maximal number, specifically $2^N$, of Pauli group elements \cite{aaronson2004improved,Gottesman1997,Gottesman1998,Bravyi2004,knill2004faulttolerant,Keeler2022}.

Acting on $\ket{\psi}$ with every $g \in G$ defines the orbit $G \cdot \ket{\psi}$, which describes the trajectory of $\ket{\psi}$ through $\Hil$. For every finitely-generated $G$ the orbit is discrete, lending itself to a natural graph-theoretic description. We construct the graph which corresponds to $G \cdot \ket{\psi}$, known as a reachability graph, by assigning vertices to states in the orbit, and edges to represent generators of $G$. Viewed through the lens of quantum computation, each path in a reachability graph defines a quantum circuit and the graph itself encodes the state's evolution under all possible circuits composed of the generating gate set. States with isomorphic reachability graphs are congruent as they share an isomorphic stabilizer group under some chosen group action \cite{Keeler:2023xcx}.

We may also wish to track the evolution of specific state properties under the action of a group $G$. Analogous to the subgroups which stabilize a state, there exist elements of $G$ which leave a particular state property invariant \cite{Keeler:2023shl}, e.g. all local gates preserve entanglement structure. If our goal is to understand the dynamics of a chosen state parameter, we can restrict consideration to the subset of $G$ which non-trivially evolves that parameter. In a reachability graph representation, this restriction corresponds to eliminating vertices and edges from the graph, leaving only the graph paths, i.e. the quantum circuits, which modify the parameter under study. This modification on reachability graphs allows us to establish bounds on the dynamics of a chosen state property under circuits composed of a certain gate set.

In this paper we explore entanglement structure in Dicke states and the manner in which that entanglement can evolve through quantum circuits. In Section \ref{EntropyConeSection}, we leverage the symmetric structure of Dicke states to explicitly compute all entanglement entropies which arise in $N$-qubit Dicke systems. We use our calculation of subsystem entanglement to generate all possible Dicke state entropy vectors, which accordingly describe the $N$-qubit Dicke state entropy cone. We demonstrate inclusion, and exclusion, of Dicke state entropy vectors relative to other known entropy cones, and use our calculation to reproduce the entropy vectors for $W$ states found in \cite{Schnitzer:2022exe}. Additionally, we propose a min-cut model on weighted star graphs which realizes the symmetrized entropies of Dicke states.

In Section \ref{StabilizerSection}, we define the stabilizer group for all Dicke states under the action of the Pauli and Clifford groups. We use this set of stabilizers to construct reachability graphs which illustrate the orbit of Dicke states under circuits composed of Pauli and Clifford gates. As we are interested to understand the dynamics of Dicke state entropy vectors in Clifford circuits, we restrict to a subset of Clifford gates consisting of Hadamard and CNOT acting on two qubits. By analyzing reachability graphs built from these gates, we are able to observe bounds on entropy vector evolution directly from the reachability graphs themselves \cite{Keeler:2023shl}.

In forthcoming work, we apply the stabilizing operations identified in Section \ref{StabilizerSection} to construct error-correcting codes for logical Dicke states. One reason for using Dicke states in error-correcting codes is an increased resistance to information loss upon single-qubit thermalization. Since tracing out a single qubit from certain Dicke states recovers the same state, now defined on one lower qubit number, the redundancy of such states may yield comparative advantages in codes. 

We also consider the potential of highly-entangled Dicke states for magic distillation protocols. In this context, one benefit of Dicke states is found in their ease of preparation and consequential preference as initial states for computation. Another promising feature of using Dicke states for magic distillation relies on the significant amount of non-local magic which contained in these states. In both applications, and perhaps others not considered in this paper, we believe this exploration into Dicke state entanglement will prove useful.

\section{Review: Dicke States, Entropy Cones, and the Stabilizer Formalism}

We offer a short review of relevant background material used throughout this work. Comprehensive discussions exist for the different topics covered here, which we invite the curious reader to consult. Many significant papers discuss the structure and properties of Dicke states, as well as their utility for realizable quantum computation, of which we recommend \cite{PhysRev.93.99,Baertschi2019,nepomechie2023qudit,PhysRevA.80.052302,Stockton_2004,PhysRevLett.103.020503,PhysRevA.95.013845}. For additional details on entropy vectors and entropy cone construction, we suggest \cite{Bao2015,Hayden2013,Schnitzer:2022exe,Avis2023,Linden2013,Fadel2021}. Finally, the group-theoretic constructs presented in this section are discussed extensively in \cite{Keeler2022,aaronson2004improved,Gottesman1997,Gottesman1998,Veitch2013},and more formally in the text \cite{Alperin1995}.

\subsection{Dicke States}

The $N$-qubit Dicke states $\ket{D^N_k}$ compose an interesting class of states which can be efficiently prepared using a polynomial number of gates, despite having a larger-than-polynomial number $\binom{N}{k}$ of excitations \cite{PhysRev.93.99,Baertschi2019}. This property affords significant resource conservation compared to arbitrary state preparation, which relies on an application of $\mathcal{O}(2^N)$ gates. For this reason, Dicke states often find preference as initial states for quantum optimization algorithms, and have even been successfully implemented in experiment \cite{PhysRevA.95.013845,PhysRevLett.103.020503,Stockton_2004,PhysRevA.80.052302}. Furthermore, the highly-entangled structure of certain Dicke states can be used to project out non-locally transformable states upon measurement, such as the $GHZ$ and $W$ states, with very little computational overhead.

We construct each $N$-qubit Dicke state $\ket{D^N_k}$ as the equal superposition over all $N$-qubit states $\ket{b}$, where $b$ is a bit-string of fixed Hamming-weight $h(b) = k$. Explicitly,
\begin{equation}\label{DickeStateDefinition}
    \ket{D^N_k} \equiv \binom{N}{k}^{-1/2} \sum_{b \in \{0,1\}^n, \hspace{.1cm} h(b) = k} \ket{b}.
\end{equation}
Specific examples of Dicke states include,
\begin{equation}
    \begin{split}
        \ket{D^3_1} &= \frac{1}{\sqrt{3}} \left(\ket{100} + \ket{010} + \ket{001} \right),\\
        \ket{D^4_2} &= \frac{1}{\sqrt{6}} \left(\ket{1100} + \ket{1010} + \ket{1001} + \ket{0110} + \ket{0101} + \ket{0011} \right).
    \end{split}
\end{equation}

Dicke states of the form $\ket{D^N_1}$, those with Hamming-weight $k=1$, are exactly the $N$-qubit $W$ states $\ket{W_N}$, defined
\begin{equation}
    \ket{W_N} \equiv \frac{1}{\sqrt{N}} \left(\ket{100...00} + \ket{010...00} + ... + \ket{000...01} \right).
\end{equation}
Similarly, Dicke states of Hamming-weight $k=N$ are the $N$-qubit measurement basis state
\begin{equation}
    \ket{D^N_N} \equiv \ket{111...1} = \ket{1}^{\otimes N}.
\end{equation}

\subsection{Entropy Vectors and Entropy Cones}

We compute the entanglement entropy of a state $\rho_{\psi}$ as the von Neumann entropy
\begin{equation}\label{vonNeumannEntropy}
   S_{\psi} \equiv -\Tr \rho_{\psi} \ln \rho_{\psi}.
\end{equation}
When $\rho_{\psi}$ represents a pure state, i.e. when $\rho_{\psi} \equiv \ket{\psi}\bra{\psi}$, the property $\rho_{\psi}^2 = \rho_{\psi}$ yields total entropy $S_{\psi} = 0$. When information is measured in \textit{dits}, as with a state $\ket{\psi} \in \Hil^d$, the entropy in Eq.\ \eqref{vonNeumannEntropy} is computed using $\log_d$.

Even for an overall pure state, non-zero entanglement can exist when considering complementary subsystems of $\ket{\psi}$. For a state $\ket{\psi}$, we can consider an $\ell$-party subsystem which we denote $I$. The entanglement entropy between $I$ and its $(N-\ell)$-party complement system $\bar{I}$ is then computed
\begin{equation}
   S_{I} = -\Tr \rho_{I} \ln \rho_{I}.
\end{equation}
The object $\rho_{I}$ is the reduced density matrix of subsystem $I$, computed by tracing out its compliment $\bar{I}$.

For an $N$-party state $\ket{\psi}$, there are $2^N-1$ different subsystems we can consider. Computing $S_I$ for each subsystem $I$, and ordering the resulting set, defines the entropy vector $\Vec{S}$ for $\ket{\psi}$. For example, the entropy vector for some $3$-party pure state would have the form,
\begin{equation}\label{EntropyVectorExample}
    \Vec{S} = (S_A, S_B, S_O; S_{AB}, S_{AO}, S_{BC}; S_{ABO}).
\end{equation}
where we use a semicolon to distinguish entropies for regions of different sizes $|I|$. The final party is often labeled with $O$, as it acts a purifier for the remainder of the system.

If the overall state $\ket{\psi}$ is pure, we have the additional constraint $S_{I} = S_{\bar{I}}$, which comes from the fact that $S_{\psi} = 0$. This condition allows the entropy vector for $\ket{\psi}$ to be expressed using only $2^{N-1}-1$ entropies. Accordingly, the vector in Eq.\ \eqref{EntropyVectorExample} can be described in the reduced form
\begin{equation}\label{EntropyVectorExample2}
    \Vec{S} = (S_A, S_B; S_{AB}).
\end{equation}
We use this reduced entropy vector presentation throughout Section \ref{StabilizerSection}.

Subsystem entropies $S_I$ for multi-partite quantum states are required to obey certain entropy inequalities \cite{Bao2015,Pippenger,Hayden2013}, which can also be used to classify that state. For example, all quantum states are subadditive, meaning $S_{I} + S_{J} \geq S_{IJ}$ for all disjoint subsystems $I$ and $J$. Other entropy inequalities are more strict, and are not necessarily satisfied by generic quantum states, e.g. the monogamy of mutual information (MMI) \cite{Hayden2013} which states
\begin{equation}\label{MMIInequality}
    S_{IJ} + S_{IK} + S_{JK} \geq S_I + S_J + S_K + S_{IJK}, 
\end{equation}
for disjoint subsystems $I,\,J,$ and $K$. The MMI inequality is satisfied by all holographic states, states with a smooth classical geometric dual through the AdS/CFT correspondence.

A linear entropy inequality, such as that in Eq.\ \eqref{MMIInequality}, defines a hyperplane in some $2^N-1$ dimensional entropy-vector space, bisecting the space and placing entropy vectors which satisfy the inequality on one side, and those which fail the inequality on the other. Entropy vectors which saturate an inequality reside on the hyperplane itself. The set of linear inequalities satisfied by a class of quantum states, bounds a convex polyhedral cone in the entropy vector space, known as an entropy cone \cite{Bao2015}. Entropy vectors which correspond to a certain class of quantum states, e.g. holographic states or stabilizer states, must have an entropy vector which lies in the convex hull%
\footnote{We again highlight that this condition is necessary, but not sufficient, for identifying states corresponding to a particular class.} %
of the corresponding entropy cone. Alternatively, entropy cones can be specified by identifying all extremal rays, the entropy vectors which saturate two inequalities and lie at the intersection of two hyperplanes, or by directly identifying all possible entropy vectors for the class of states, as performed in Section \ref{EntropyConeSection}.

Entropy cones for various classes of states are well-understood for low party number. However as system size increases, so too does the number of necessary inequalities, and complexity of each inequality, needed to characterize each entropy cone. To navigate this complexity increase, much effort has turned towards studying more fundamental properties of entropy cones. The symmetrized entropy cone prescription \cite{Czech2021,Fadel2021} focuses the extremal properties of a cone's structure under a symmetry projection. Symmetrized entropies are defined as in Eq.\ \eqref{EntropyVectorExample}, with the addition of a normalization factor based on the cardinality of the subsystem. For subsystems $I$, we have
\begin{equation}\label{SymmetrizedEntropyVector}
    \Tilde{S}_k \equiv \left[ \binom{N+1}{k} \right]^{-1} \sum_{I \in \{I\}_k} S_I,
\end{equation}
where the sum is computed over all subsystems $I$, of and $N$-party states, with fixed cardinality $k = |I|$. As an example, computing the symmetrized entropy of all single-party subsystems for a $4$-party state, we have
\begin{equation}\label{SymmetrizedEntropyVector}
    \Tilde{S}_1 = \frac{1}{4}\left( S_A + S_B + S_C + S_O\right).
\end{equation}

Mathematical graphs also offer a useful description of entanglement in multi-partite quantum systems, particularly with regards to holographic systems \cite{Bao2015}. The entropy vectors of holographic states can be realized as a min-cut protocol on weighted undirected graphs, where edge cuts in the graph correspond to traversing minimal-length geodesics in the dual geometry. Broader classes of states require more generic graph descriptions, including hypergraphs \cite{Bao2020a,Bao2020} or topological links \cite{Bao2022a}. Symmetrized entropy vectors can likewise be realized using a min-cut prescription on weighted star graphs \cite{Czech2021,Fadel2021}. In Section \ref{StarGraphSection}, we extend this star graph proposal to describe the structure of Dicke state entropy vectors.

\subsection{Stabilizer Formalism and Reachability Graphs}

An essential set of gates in quantum computing is the set of Pauli gates, defined in a unitary matrix representation as
\begin{equation}\label{PauliMatrices}
    \mathbb{1}\equiv \begin{bmatrix}1&0\\0&1\end{bmatrix}, \,\, 
    \sigma_X\equiv \begin{bmatrix}0&1\\1&0\end{bmatrix}, \,\,
    \sigma_Y\equiv \begin{bmatrix}0&-i\\i&0\end{bmatrix}, \,\,
    \sigma_Z\equiv \begin{bmatrix}1&0\\0&-1\end{bmatrix}.
\end{equation}
In a fixed measurement basis $\{\ket{0},\,\ket{1}\}$, the matrices in Eq.\ \eqref{PauliMatrices} act as operators on a Hilbert space $\mathbb{C}^2$. The set $\{\sigma_X,\,\sigma_Y,\,\sigma_Z\}$ generates the $16$-element Pauli group under multiplication, denoted $\Pi_1$.

We can extend the matrix representation of Pauli gates to arbitrary qubit number by composing sets of Pauli strings. Each Pauli string describes a set of local actions performed on specified qubits in an $N$-qubit system. Every $N$-qubit Pauli string can be defined%
\footnote{This Pauli string representation as the $N$-fold tensor product of $2 \times 2$ matrices requires two conditions: first, we assume the Hilbert space factorizes into a product of $N$ qubits ($N$ copies of $\Hil^2$), and secondly, we ascribe an ordering to the set of qubits which will serve as an indexing system.} %
as a tensor product over $2 \times 2$ matrices. For example, the action of $\sigma_X$ on the $k^{th}$ qubit of an $N$-qubit system can be written
\begin{equation}\label{PauliStringExample}
    \sigma^k_X \equiv \mathbb{1}^1\otimes\ldots\otimes \mathbb{1}^{k-1} \otimes \sigma_X \otimes \mathbb{1}^{k+1} \otimes \ldots \otimes \mathbb{1}^N.
\end{equation}
Eq.\ \eqref{PauliStringExample} is an example of a weight-$1$ Pauli string, where the weight of a string denotes the number of non-identity operations in the tensor product. The $N$-qubit Pauli group $\Pi_N$ is generated by the set of all weight-$1$ Pauli strings.

Having constructed $\Pi_N$, we could further consider operations which map $\Pi_N$ to itself. The $N$-qubit Clifford group $\mathcal{C}_N$ is the set of unitaries which normalizes the Pauli group, i.e. $\mathcal{C}_N$ maps elements of $\Pi_N$ to elements of $\Pi_N$ via conjugation. In the single-qubit case, we can define $\mathcal{C}_1$ as the group generated by the Hadamard \cite{sylvester1867lx,hadamard1893resolution} and phase gates, defined as the matrices
\begin{equation}
    H\equiv \frac{1}{\sqrt{2}}\begin{bmatrix}1&1\\1&-1\end{bmatrix}, \quad P\equiv \begin{bmatrix}1&0\\0&i\end{bmatrix}.
\end{equation}

Just as with Pauli gates we can generalize to an $N$-qubit description by composing strings of Clifford operators, where we use a subscript to indicate the qubit being acted on, e.g.
\begin{equation}\label{HadamardString}
    H_k \equiv \mathbb{1}^1\otimes\ldots\otimes \mathbb{1}^{k-1} \otimes H \otimes \mathbb{1}^{k+1} \otimes \ldots \otimes \mathbb{1}^N.
\end{equation}

Unlike the Pauli group, however, the group $\mathcal{C}_N$ is not generated by only weight-$1$ Clifford strings. For $N>1$, constructing $\mathcal{C}_N$ requires the addition of the bi-local $CNOT$ gate, defined
\begin{equation}
    C_{i,j} = \begin{bmatrix}
            1 & 0 & 0 & 0\\
            0 & 1 & 0 & 0\\
	    0 & 0 & 0 & 1\\
	    0 & 0 & 1 & 0
            \end{bmatrix}.
\end{equation}
The $CNOT$ gate acts on two qubits in an $N$-qubit system by first evaluating the state of the $i^{th}$ qubit, the control bit, then performing a $NOT$ operation on the $j^{th}$ qubit, the target bit, if the $i^{th}$ qubit is found in the state $\ket{1}$. It is important to note that $C_{i,j} \neq C_{j,i}$. We may now define the group $\mathcal{C}_N$ as
\begin{equation}
    \mathcal{C}_N \equiv \langle H_1,\,...,\,H_N,\,P_1,\,...,\,P_N,\,C_{1,2},\,C_{2,1},\,...,\,C_{N-1,N},\,C_{N,N-1} \rangle.
\end{equation}

Given a Hilbert space $\Hil$ and group $G \subset L(\Hil)$, an element of $G$ is said to stabilize $\ket{\psi} \in \Hil$ if it acts trivially on $\ket{\psi}$. The set of all $g \in G$ that stabilize $\ket{\psi}$, defined
\begin{equation}
    \Stab_{G}(\ket{\psi}) \equiv \{g\in G\: | \:g\ket{\psi}=\ket{\psi}\},
\end{equation}
makes up the stabilizer subgroup of $\ket{\psi}$ under the action of $G$. Otherwise stated, $\ket{\psi}$ is a $+1$ eigenvector of each $g \in \Stab_{G}(\ket{\psi})$.

For $G$ a finite group, Lagrange's theorem \cite{Alperin1995} ensures a partition of $|G|$ for any subgroup $H \leq G$, explicitly
\begin{equation}\label{LagrangeTheorem}
    |G| = \left[G:H\right] \cdot |H|,
\end{equation}
with $[G:H]$ the index of $H$ in $G$. Furthermore when $G$ acts on a set $X$, the Orbit-Stabilizer theorem \cite{Alperin1995} gives the orbit of $x \in X$ under the action of $G$ as
\begin{equation}\label{OrbitStabilizerTheorem}
    |G \cdot x| = \left[G:\Stab_G(x)\right] = \frac{|G|}{|\Stab_G(x)|}.
\end{equation}
When $G$ acts on a Hilbert space $\Hil$, we can use Eqs. \eqref{LagrangeTheorem} and \eqref{OrbitStabilizerTheorem} to construct orbits of $\ket{\psi} \in \Hil$ under the group action \cite{Keeler:2023xcx}.

When considering the action of $\Pi_N$ on $\Hil$, states which are stabilized by a $2^N$-element subgroup of $\Pi_N$ are known as stabilizer states \cite{aaronson2004improved,Gottesman1997,garcia2017geometry,Keeler2022}. Stabilizer states play a critical role in near-term realization of quantum computing, as they represent the set of quantum systems which can be efficiently classically simulated \cite{Gottesman1998}. One way to construct the set of $N$-qubit stabilizer states is by acting on a state in the measurement basis $\{\ket{0},\ket{1}\}^N$ with the Clifford group $\mathcal{C}_N$. The set of $N$-qubit stabilizer states $\mathcal{S}_N$ is exactly the orbit of each state in $\{\ket{0},\ket{1}\}^N$, under the action of $\mathcal{C}_N$. The number of $N$-qubit stabilizer states generated by the orbit $\mathcal{S}_N$ is derived in \cite{doi:10.1063/1.4818950}, and has order
\begin{equation}\label{StabilizerSetSize}
    \left|\mathcal{S}_N\right| = 2^n \prod_{k=0}^{n-1} (2^{n-k}+1).
\end{equation}

The process of constructing state orbits under group action naturally admits a graph-theoretic description \cite{aaronson2004improved,Keeler:2023xcx}. For some $G$ acting on $\ket{\psi} \in \Hil$, we can assign a vertex to each state in the orbit $[G \cdot \ket{\psi}]$, and an edge to each generator of $G$. This graph is known as the reachability graph for $\ket{\psi}$, and maps the evolution of $\ket{\psi}$ through $\Hil$ under the action of $G$. Figure \ref{FullC1Graph} depicts the reachability graph for $\ket{0}$ under the single-qubit Clifford group $\mathcal{C}_1$, with vertices representing the $6$ single-qubit stabilizer states.
    \begin{figure}[h]
        \centering
        \includegraphics[width=10cm]{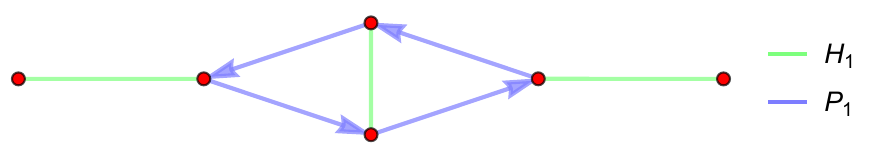}
        \caption{Orbit of $\ket{0}$ under the action of $\mathcal{C}_1$, depicted as a reachability graph. Graph vertices represent the $6$ single-qubit stabilizer states, and edges correspond to $\mathcal{C}_1$ generators. For generators which are self-inverse, e.g. the Hadamard gate, we use undirected edges.}
    \label{FullC1Graph}
    \end{figure}

It is often useful to consider only the action of a subgroup $H \leq G$ on $\Hil$. Focusing on the action of $H$ highlights specific features of a state's orbit, and can better-exhibit the evolution of certain state properties through the orbit. When composing the orbit of a state $\ket{\psi}$ under the action of some $H \leq G$, the term restricted graph is sometimes used to discuss the emergent reachability graph \cite{Keeler2022}.

\section{The Entropy Cone for $\ket{D^N_k}$}\label{EntropyConeSection}

In this section, we describe the entropy cone for $N$-qubit Dicke states by explicitly building all $\ket{D^N_k}$ entropy vectors for qubit number $N$ and Hamming-weight $k$. We highlight the symmetric properties of $\ket{D^N_k}$ entropy vectors, and note the relative containment of the $\ket{D^N_k}$ entropy cone in other known entropy cones. We demonstrate that our construction for Dicke states reproduces the $W$ state entropy cone found in \cite{Schnitzer:2022exe}. Additionally, we give a realization of $\ket{D^N_k}$ entropy vectors as a min-cut prescription on weighted star graphs. In later sections we analyze the evolution of the $\ket{D^N_k}$ entropy vectors defined here, under the action of Clifford circuits.

\subsection{Dicke State Entropy Vectors}

The symmetric structure of Dicke states $\ket{D^N_k}$ enables a direct calculation of subsystem entanglement entropy from the non-zero diagonal elements of the density matrix \cite{Witten:2018zva}. For an $N$-party pure state $\ket{D^N_k}$ of Hamming-weight $k$, the entanglement entropy of an $\ell$-party subsystem is computed
\begin{equation}\label{DickeStateEntropies}
    S_{\ell}\left(\ket{D^N_k} \right) \equiv - \binom{N}{k}^{-1} \sum_{i=0}^{min(\ell,k)} \binom{\ell}{i}\binom{N-\ell}{k-i}\ln\left[\binom{N}{k}^{-1}\binom{\ell}{i}\binom{N-\ell}{k-i}\right].
\end{equation}
We can directly verify that $S_{\ell} = S_{N-\ell}$ and $S_N = 0$, from Eq.\ \eqref{DickeStateEntropies}. Furthermore, we highlight the property that $S_{\ell}\left(\ket{D^N_k}\right)$ depends only on the cardinality of a chosen subsystem.

The calculation of $S_{\ell}$ admits simplifications for specific values of $\ell$ and $k$. For states $\ket{D^N_k}$ with $\ell \geq k$, Eq.\ \eqref{DickeStateEntropies} becomes
\begin{equation}\label{lIndependentExpression}
        S_{\ell} = \ln\left[\binom{N}{k} \right] -\binom{N}{k}^{-1}\sum_{i=0}^{k}\binom{\ell}{i}\binom{N-\ell}{k-i}\ln\left[\binom{\ell}{i}\binom{N-\ell}{k-i}\right],
\end{equation}
where we note the $\ell$-independence of the first term. A derivation of \eqref{lIndependentExpression} is given in Appendix \ref{lIndependentSimplification}. A similar decoupling exists for $\ell < k$, and is shown in Eq.\ \eqref{lLessThankSimplification}.

For $k=1$, states $\ket{D^N_k}$ are the subset of $N$-qubit $W$ states, $\ket{W_N} \equiv \ket{D^N_1}$. For an $\ell$-party subsystem of $\ket{D^N_1}$, the expression for $S_{\ell}$ in Eq.\ \eqref{DickeStateEntropies} gives
\begin{equation}\label{WStateEntropies}
    S_{\ell} \left(\ket{D^N_1} \right) = \frac{\ell}{N} \ln\left[\frac{N}{\ell}\right] + \frac{(N-\ell)}{N} \ln \left[\frac{N}{N-\ell}\right],
\end{equation}
in agreement with the calculations given in \cite{Schnitzer:2022exe}.

The ordered set of all $2^N-1$ subsystem entropies for a state $\ket{D^N_k}$, computed according to Eq.\ \eqref{DickeStateEntropies}, compose the entropy vector $\vec{S}\left(\ket{D^N_k}\right)$. Since each $S_{\ell}$ depends only on $|\ell|$, all Dicke state entropy vectors share the form
\begin{equation}\label{DickeStateEntropyVector}
    \vec{S}\left(\ket{D^N_k}\right) \equiv \Bigl(\underbrace{S_1,\,...,\,S_1}_{\binom{N}{1}};\,\underbrace{S_2,\,...,\,S_2}_{\binom{N}{2}};\,...;\,\underbrace{S_{N-1},\,...,\,S_{N-1}}_{\binom{N}{N-1}};\,0\Bigl).
\end{equation}
The entropy vectors in Eq.\ \eqref{DickeStateEntropyVector} are manifestly symmetrized \cite{Czech2021}, leaving $\vec{S}\left(\ket{D^N_k}\right)$ invariant up to exchange of subsystems of equal size $|\ell|$.

All $N$-qubit Dicke state entropy vectors can be calculated using Eqs. \eqref{DickeStateEntropies} and \eqref{DickeStateEntropyVector}. Collectively, these equations describe the $N$-qubit Dicke state entropy cone, defined for each $N$ as the convex hull of all $\vec{S}\left(\ket{D^N_k}\right)$. Since each $\vec{S}\left(\ket{D^N_k}\right)$ is symmetrized, all Dicke state entropy vectors automatically satisfy the symmetrized quantum entropy cone (SQEC) inequalities \cite{Fadel2021}, 
\begin{equation}
    -S_{\ell-1} + 2S_{\ell} -S_{\ell+1} \geq 0, \qquad \forall \,1 \leq \ell \leq \lceil N/2 \rceil,
\end{equation}
which verify symmetric instances of subadditivity and strong-subadditivity. The $N$-qubit Dicke state entropy cone is therefore contained within the SQEC for all $N$. 

The monogamy of mutual information (MMI) inequality, given in Eq.\ \eqref{MMIInequality}, defines a subset of facets which bound the holographic entropy cone. For $\ket{D^N_k}$, MMI is saturated%
\footnote{MMI is trivially saturated by entropy vectors of the unentangled Dicke states $\ket{D^N_N} = \ket{1}^{\otimes N}$.} %
when $N = 3$, and violated otherwise. Similarly, $\ket{D^N_k}$ entropy vectors violate requisite inequalities for the symmetrized holographic entropy cone (SHEC), when $N > 3$, namely
\begin{equation}
    -\ell(\ell+1)S_{\ell-1}+2(\ell-1)(\ell+1)S_{\ell}-\ell(\ell-1)S_{\ell+1} \geq 0, \qquad \forall \, \ell \in [2,n/2].
\end{equation}
Consequently, portions of the $N$-qubit Dicke state entropy cone lie outside the holographic and symmetrized holographic entropy cones.

The entropy cone of stabilizer states is completely characterized up through $4$ parties ($N=5$ qubits including the purifier). It is also known that states $\ket{D^N_k}$ are not stabilizer states%
\footnote{We again note the exception for $\ket{D^N_N}$ which is trivially a stabilizer state.} %
for all $N \geq 3$. Nevertheless, we observe the following for Dicke state systems of $N \leq 5$ qubits:
\begin{observation}\label{DickeStateVectorContainment}
    The Dicke state entropy cone, for $N \leq 5$, is completely contained within the convex hull of the stabilizer entropy cone.
\end{observation}
Extending Observation \ref{DickeStateVectorContainment} to a general conjecture for all $N$ would require further knowledge of higher-party stabilizer entropy cones.

Acting with Clifford circuits on states $\ket{D^N_k}$ generates additional entropy vectors beyond those given in Eq.\ \eqref{DickeStateEntropyVector}. These Clifford group orbits of Dicke states are discussed in Section \ref{StabilizerSection}, as are the resulting entropy vectors reached under corresponding Clifford circuits. Here we note the following observation for entropy vectors generated by $2$-qubit Clifford action on $\ket{D^N_k}$:
\begin{observation}\label{CliffordOrbitDickeStateInclusion}
    All entropy vectors generated by $2$-qubit Clifford action on $\ket{D^N_k}$, for $N \leq 5$, are contained in the convex hull of the stabilizer entropy cone.
\end{observation}
While we expect Observation \ref{CliffordOrbitDickeStateInclusion} to hold for all $\ket{D^N_k}$, as well as for arbitrary Clifford circuits, we do not make an attempt towards a conjecture in this work.

We have given an explicit calculation of all $\ell$-party entanglement entropies in Dicke states $\ket{D^N_k}$, for arbitrary system size $N$ and Hamming-weight $k$. We used this result to construct all Dicke state entropy vectors $\vec{S}\left(\ket{D^N_k}\right)$, and showed that our results reproduce previous entropy vector calculations for $W$ states for $k=1$. We present the set of all $\vec{S}\left(\ket{D^N_k}\right)$, for a fixed qubit number $N$, as the $N$-qubit Dicke state entropy cone. We have highlighted that, since $\ket{D^N_k}$ entropy vectors emerge symmetrized, the $\ket{D^N_k}$ entropy cone is contained within the SQEC. At $N \geq 4$ we observed that $\vec{S}\left(\ket{D^N_k}\right)$ violates holographic inequalities, e.g. MMI, and lies outside both the holographic and symmetrized holographic entropy cones. In the next section, we use our $\ket{D^N_k}$ entropy vector construction to define a min-cut protocol which realizes Dicke state entropy vectors using weighted star graphs.

\subsection{A Graph Model for $\ket{D^N_k}$ Entropies}\label{StarGraphSection}

We now outline a protocol which compute Dicke state entropies, as given in Eq.\ \eqref{DickeStateEntropies}, as a sum over minimum-weight edge cuts on star graphs. Initial descriptions using star graphs to represent average entropies were presented in \cite{Czech:2021rxe,Fadel2021}, and later extended to include the possibility of negative edge weights in \cite{Harper:2022sky,Schnitzer:2022exe}. We demonstrate an explicit example of this star graph construction for $\ket{D^N_1}$ entropy vectors, and describe how to recursively generalize the model for $k>1$.

To construct our representation of entanglement entropy, we consider a graph $G = (V,E)$, with vertex set $V$ partitioned into subsets of internal vertices $V_{Int.} \subseteq V$, and external vertices $V_{Ext.}\subseteq V$. For an $N$-party $\ket{\psi}$ with purifier, each disjoint subsystem $\ell$ is assigned a vertex $v_{\ell} \in V_{Ext.}$, where $|V_{Ext.}| = N+1$. The entropy $S_{\ell}$ is then computed as the total weight of a min-cut on $G$ which separates $v_\ell$ from its complement subsystem $v_{\ell^C} \subset V_{Ext.}$.

For Dicke states $\ket{D^N_k}$ we represent each $S_{\ell}$ using a star graph with $N$ edges of unit weight, and one edge of weight $w \leq 0$. One novel feature of these graphs is that $w$ may take on select negative values, subject to the required inequalities%
\footnote{We often consider a tuple of non-negative entropies which lives in a totally non-negative sector of some $2^{N}-1$ vector space. If instead one considers a perfect tensor decomposition, negative entropies are permitted as long as they sum to a positive value in the entropy basis and satisfy the required inequalities of a chosen entropy cone. For further detail we recommend \cite{Harper:2022sky}.} %
of a particular entropy cone, which ultimately sum to non-negative entropies. Since $\ell$-party $\ket{D^N_k}$ entropies depend only on the cardinality $|\ell|$, we compute $S_{\ell}$ as the min-cut
\begin{equation}\label{SergioFadelEq}
    S_{\ell} = \min \{|\ell|,N-1 -|\ell|+w \},
\end{equation}
following \cite{Schnitzer:2022exe,Fadel2021}. Figure \ref{BasicStarGraph} gives an example of a star graph which realizes $S_{\ell}$ for a state $\ket{D^3_k}$.
	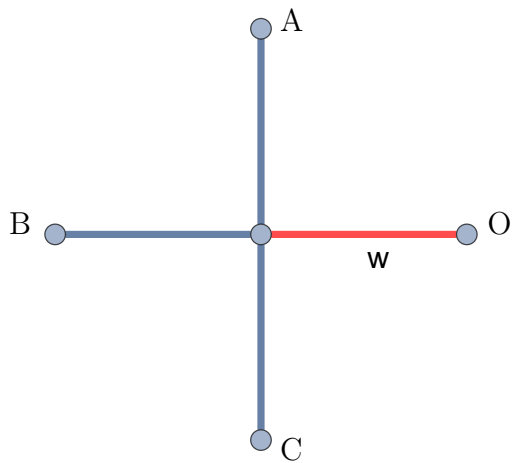
\begin{figure}[h]
		\begin{center}
		\begin{overpic}[width=6cm]{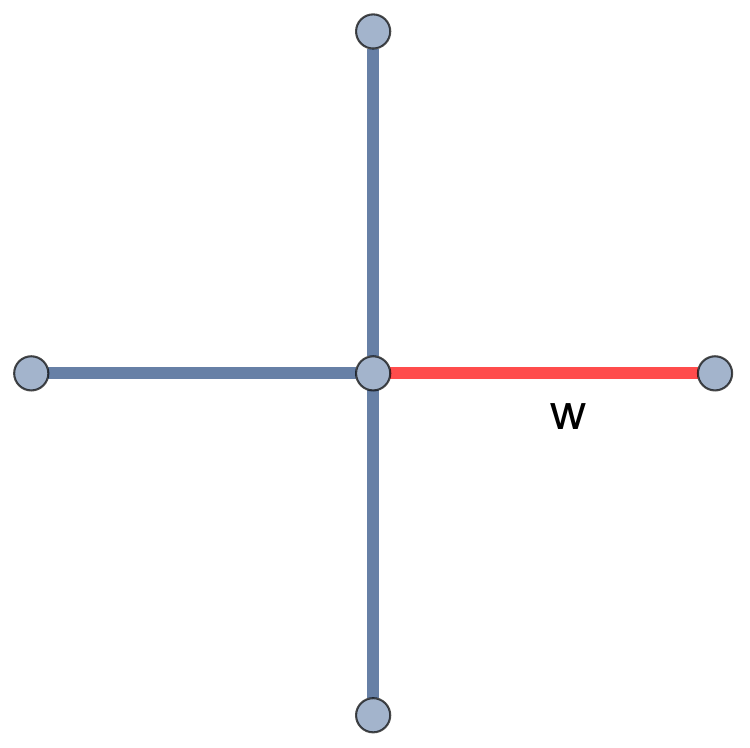}
		\put (54,95){A}
        \put (-6,50){B}
        \put (54,0){C}
        \put (100,50){O}
        \end{overpic}
        \caption{Example of a $4$-legged star graph, with $3$ legs of unit weight and one leg of weight $w \leq 0$, which realizes the entropies of $\ket{D^3_k}$. The weight $w$ is negative in this graph, and defined as a function of $k$.}
		\label{BasicStarGraph}
	\end{center}
	\end{figure}
The value of $w$ is defined in terms of $k$, as shown in Eqs. \eqref{StarGraphMin}--\eqref{W2Bounds}.

Since $\ket{D^N_k}$ entropies in Eq.\ \eqref{DickeStateEntropies} obey the symmetry $S_{\ell} = S_{N-\ell}$, we can define $\Tilde{S}_\ell$ to be the symmetrized variable over all $\ell$-party entanglement entropy
\begin{equation}\label{StarGraphMin}
    \Tilde{S}_\ell = \binom{N}{\ell}^{-1}\left[\binom{N-1}{\ell}S_\ell + \binom{N-1}{N-\ell}S_{N-\ell}  \right].
\end{equation}
As shown in Eq.\ \eqref{lIndependentExpression}, each $S_{\ell}$ is computed as a sum over $\ell+1$ terms when $\ell < k$, or $k+1$ terms when $\ell \geq k$. Accordingly, each $\Tilde{S}_\ell$ in Eq.\ \eqref{StarGraphMin} is realized as a sum over $\ell+1$ (or $k+1$) star graphs, for $1 \leq \ell \leq \lceil N/2 \rceil$. This sum over graphs for each $S_{\ell}$ generalizes the previous constructions in \cite{Schnitzer:2022exe} to all $\ket{D^N_k}$.

To demonstrate this min-cut model, we construct an explicit representation of $S_{\ell}$ for states $\ket{D^N_1}$. Applying Eq.\ \eqref{SergioFadelEq} to Eq.\ \eqref{StarGraphMin} we have
\begin{equation}\label{WStateStarEquation}
    \Tilde{S}_{\ell} = \frac{1}{N} \left[(N-\ell)\min\{\ell,N-1-\ell+w_1\} + \ell \min \{N-\ell,w_2+\ell-1 \} \right].
\end{equation}
Figure \ref{TwoStarGraphs} shows an example pair of graphs, for the state $\ket{D^4_1}$, whose sum over min-cuts realizes Eq.\ \eqref{WStateStarEquation}. In both graphs, the negatively-weighted edge connects to the external vertex for the purifier $O$.
    \begin{figure}[h]
        \centering
        \includegraphics[width=11cm]{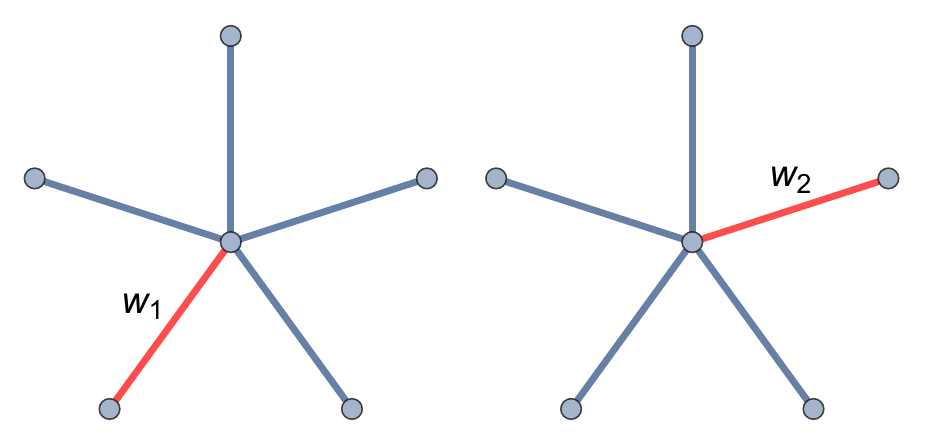}
        \caption{Pair of star graphs whose min-cut sum calculates $\Tilde{S}_{\ell}$, as in Eq.\ \eqref{WStateStarEquation}, for $\ket{D^5_1}$. The values $w_1$, $w_2$ are both negative and set by inserting Eq.\ \eqref{DickeStateEntropies} into Eq.\ \eqref{SergioFadelEq}.}
    \label{TwoStarGraphs}
    \end{figure}

Beginning with the first term in Eq.\ \eqref{WStateStarEquation}, which we denote $(\Tilde{S}_{\ell})_1$, we have
\begin{equation}
    (\Tilde{S}_{\ell})_1 = \frac{1}{N}(N-\ell)\min \{ \ell,N-1-\ell+w_1\}.
\end{equation}
From Eq.\ \eqref{DickeStateEntropies} we require
\begin{equation}
    \min \{ \ell,N-1-\ell+w_1\} = \ln\left[\frac{N}{N-\ell}\right],
\end{equation}
which we solve for $w_1$ to give the bound
\begin{equation}\label{w1Solve}
    w_1 = \ell + \ln\left[\frac{N}{N-\ell}\right] - (N-1).
\end{equation}
The weight $w_1$ in Eq.\ \eqref{w1Solve} takes on negative values for
\begin{equation}\label{W1Bounds}
    \ell < (N-1) - \ln\left[\frac{N}{N-\ell}\right].
\end{equation}

Evaluating the second term $(\Tilde{S}_{\ell})_2$ in Eq.\ \eqref{StarGraphMin}, we have
\begin{equation}\label{STilde2}
    (\Tilde{S}_{\ell})_2 = \frac{\ell}{N}\min\{N-\ell,w_2+\ell-1\}.
\end{equation}
We solve Eq.\ \eqref{STilde2} to find the weight
\begin{equation}
    w_2 = \ln\left[\frac{N}{\ell}\right]-\ell+1,
\end{equation}
which is negative while
\begin{equation}\label{W2Bounds}
    \ell > 1+ \ln\left[\frac{N}{\ell} \right].
\end{equation}

The procedure in Eqs. \ref{WStateStarEquation}--\ref{W2Bounds} can be applied for all $\ket{D^N_1}$, with the resulting symmetrized entropies $\Tilde{S}_{\ell}$ described as a min-cut protocol on a pair of weighted star graphs analogous to those in Figure \ref{TwoStarGraphs}. Each graph possesses a single edge of negative weight, and the values of each weight can be determined as in Eqs. \eqref{W1Bounds} and \eqref{W2Bounds}. We now describe how to generalize this model to arbitrary $k$, by inserting sequences of star graphs to evaluate each $S_{\ell}$.

We can naturally extend the protocol described in Eqs. \ref{WStateStarEquation}--\ref{W2Bounds} to all $\ket{D^N_k}$ with $k\geq1$. For any $k >1$, each of the terms $(\Tilde{S}_\ell)_1$ and $(\Tilde{S}_\ell)_2$ in Eq.\ \eqref{StarGraphMin} are computed as a sum over min-cuts on $\ell +1$ star graphs for $\ell \geq k$, or $k+1$ star graphs for $\ell <k$. For example, consider the symmetrized entropies of $\ket{D^5_2}$,
\begin{equation}\label{5QubitEntropyRecursion}
    \begin{split}
        \Tilde{S}_1 &= (\Tilde{S}_1)_1 + (\Tilde{S}_1)_2 = 2\binom{5}{1}^{-1}\left(\frac{3}{5}\ln\left[\frac{5}{3} \right] + \frac{2}{5}\ln\left[\frac{5}{2} \right]\right) = \Tilde{S}_4,\\
        \Tilde{S}_2 &= (\Tilde{S}_2)_1 + (\Tilde{S}_2)_2 = 2\binom{5}{2}^{-1}\left(\frac{3}{5}\ln\left[\frac{5}{3} \right] + \frac{3}{10}\ln\left[\frac{10}{3} \right] + \frac{1}{10}\ln\left[\frac{10}{1} \right]\right)  = \Tilde{S}_3,\\
        \Tilde{S}_5 &= 0\\
    \end{split}
\end{equation}
The quantity $\Tilde{S}_1$ in Eq.\ \eqref{5QubitEntropyRecursion} admits a star graph representation exactly as described in Eqs. \eqref{StarGraphMin}--\eqref{W2Bounds}. Meanwhile, $\Tilde{S}_2$ is given by a sum over three star graphs, each having a single edge of negative weight.

We have given a min-cut protocol on weighted star graphs which realizes the symmetrized entropies $\Tilde{S}_{\ell}$ of all Dicke states $\ket{D^N_k}$. We gave a direct example showing graph realizations of $S_{\ell}$ for states $\ket{D^N_1}$, in agreement with results demonstrated in \cite{Schnitzer:2022exe}. We generalized this technique to $k>1$ by computing each term in  $\Tilde{S}_{\ell}$ as a sum over $\ell + 1$ star graphs, each having a single edge of negative weight. In the next section we explore group stabilizers for Dicke states under action of the Pauli and Clifford groups. We likewise analyze the orbits of Dicke states under these groups, as well as the dynamics of $\ket{D^N_k}$ entropy vectors under Clifford circuits. We illustrate $\ket{D^N_k}$ orbits as reachability graphs, using the methods given in \cite{Keeler2022,Keeler:2023xcx,Keeler:2023shl}.

\section{Stabilizers and Orbits of $\ket{D^N_k}$}\label{StabilizerSection}

In this section we construct the stabilizer subgroups for all Dicke states $\ket{D^N_k}$ under action of the Pauli and Clifford groups. We use each subgroup to construct the reachability graph for all $\ket{D^N_k}$, under the action of Pauli and Clifford group elements \cite{Keeler:2023xcx}. We highlight differences between the reachability graph structures observed for Dicke states, and those seen among the $N$-qubit stabilizer states. We later remark on the utility of the $\ket{D^N_k}$ stabilizer groups identified for error-correcting codes. We analyze $\ket{D^N_k}$ reachability graphs, with vertices colored to indicate the entropy vector, and determine the evolution of $\ket{D^N_k}$ entropy vectors under a restricted subgroup of Clifford operators. Further, we establish bounds on how much each $\ket{D^N_k}$ entropy vector can change under the select set of gates. The Mathematica data and packages used to generate all graphs is publicly available \cite{githubStab, githubCayley}.

\subsection{Pauli Group Orbits}\label{PauliStabilizers}

We first consider the action of the Pauli group%
\footnote{In the case of $\Pi_N$, as well as with $\mathcal{C}_N$, we first mod out each group by elements which act as a global phase on the group. For $\mathcal{C}_N$ this global phase element is $\omega \equiv (H_iP_i)^3$, which has the property $\omega^8 = \mathbb{1}$. Likewise for $\Pi_N$, this global phase is $\omega^2$. For details, see Section 5.1 of \cite{Keeler:2023xcx}.} %
$\Pi_N$ on the set of $N$-qubit Dicke states. While all quantum states are trivially stabilized by $\mathbb{1} \in \Pi_N$, Dicke states admit larger stabilizer subgroups in $\Pi_N$, which make them useful for stabilizer code construction. Every $\ket{D^N_k}$ is stabilized by, at least, $2$ elements of $\Pi_N$, but some are stabilized by more. 

In addition to $\mathbb{1}$, all Dicke states $\ket{D^{N}_k}$ are stabilized by the $\Pi_N$ element 
\begin{equation}\label{AllDickeStabilizer}
\begin{split}
    &\bigotimes_{i=1}^N \sigma_Z^i, \quad \textnormal{for k even},\\
    -&\bigotimes_{i=1}^N \sigma_Z^i, \quad \textnormal{for k odd}.\\
\end{split}
\end{equation}
The operator in Eq.\ \eqref{AllDickeStabilizer} acts as a $\sigma_Z$ on every qubit of an $N$-qubit system, with an additional $-1$ phase for $k$ odd. For example, Dicke states $\ket{D^{3}_1}$ and $\ket{D^{5}_2}$ have respective stabilizer subgroups given by
\begin{equation}
\begin{split}
    \Stab_{\Pi_3}(\ket{D^{3}_1}) &= \{\mathbb{1}, -\sigma_Z^1\sigma_Z^2\sigma_Z^3\},\\
    \Stab_{\Pi_5}(\ket{D^{5}_2}) &= \{\mathbb{1}, \sigma_Z^1\sigma_Z^2\sigma_Z^3\sigma_Z^4\sigma_Z^5\}.\\
\end{split}
\end{equation}

The stabilizer subgroup containing $\mathbb{1}$ and Eq.\ \eqref{AllDickeStabilizer} quotients $\Pi_N$ into a group of order $2^{2n-1}$. Figure \ref{D31PauliGraph} illustrates the reachability graph for $\ket{D^{3}_1}$ under the action of $\Pi_3$. 
    \begin{figure}[h]
        \centering
        \includegraphics[width=9cm]{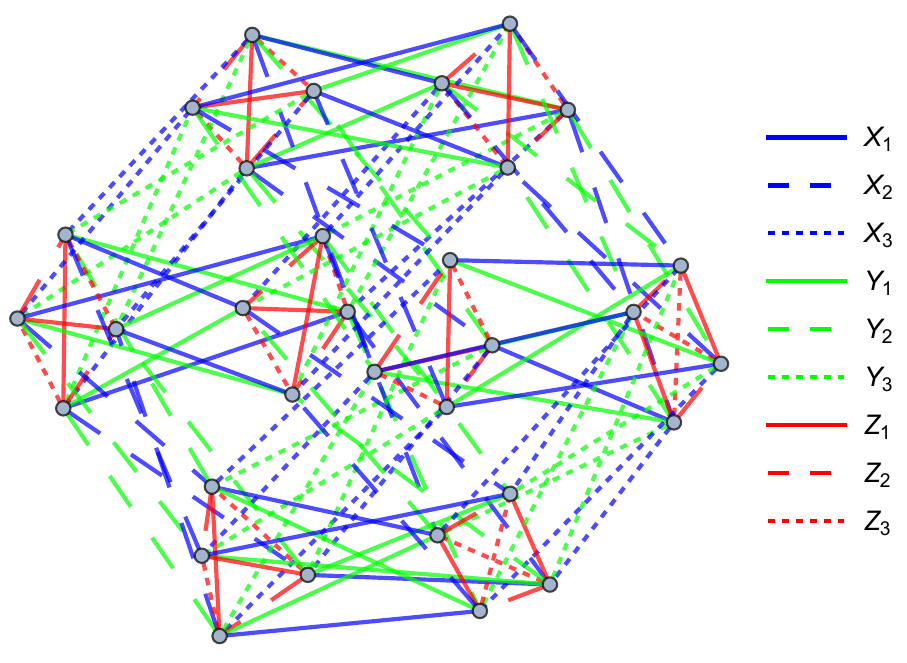}
        \caption{Orbit of $\ket{D^3_1}$ under the $3$-qubit Pauli group $\Pi_3$, which contains $32$ vertices. In general, states stabilized by only $\mathbb{1}$ and Eq.\ \eqref{AllDickeStabilizer} will have a Pauli orbit of length $2^{2n-1}$.}
    \label{D31PauliGraph}
    \end{figure}

The Dicke state $\ket{D^{N}_N}$ is a stabilizer state, specifically $\ket{D^{N}_N} = \ket{1}^{\otimes N}$, as are all $\ket{D^{N}_k}$ for $N \leq 2$. Accordingly, $\ket{D^{N}_N}$ is stabilized by a $2^N$-element subgroup of $\Pi_N$. In addition to $\mathbb{1}$ and the operator in Eq.\ \eqref{AllDickeStabilizer}, $\ket{D^{N}_N}$ is stabilized by the action of $-\sigma_Z$ on any single qubit, as well as $\sigma_Z$ on any qubit pair for $N \geq 2$. Written explicitly, the set
\begin{equation}\label{StabilizerNN}
    \Stab_{\Pi_N}(\ket{D^{N}_N})\supseteq\{-\sigma_Z^i,\, \sigma_Z^i\sigma_Z^j\}, \quad \forall \, i,j \in \{1,N\}.
\end{equation}

As an example, the stabilizer subgroup of $\ket{D^{3}_3}$ consists of the $6$ operations
\begin{equation}
    \Stab_{\Pi_3}(\ket{D^{3}_3}) = \{\mathbb{1},\, -\sigma_Z^1,\,-\sigma_Z^2,\,-\sigma_Z^3, \,\sigma_Z^1\sigma_Z^2, \,\sigma_Z^1\sigma_Z^3,\,\sigma_Z^2\sigma_Z^3,\, -\sigma_Z^1\sigma_Z^2\sigma_Z^3\}.
\end{equation}
The corresponding Pauli orbit for $\ket{D^{3}_3}$ is shown in Figure \ref{D33PauliGraph}, where we note that all edges of the graph simultaneously represent the actions of $\sigma_X^i$ and $\sigma_Y^i$, as they act identically on $\ket{D^{3}_3}$ up to global phase.
    \begin{figure}[h]
        \centering
        \includegraphics[width=9cm]{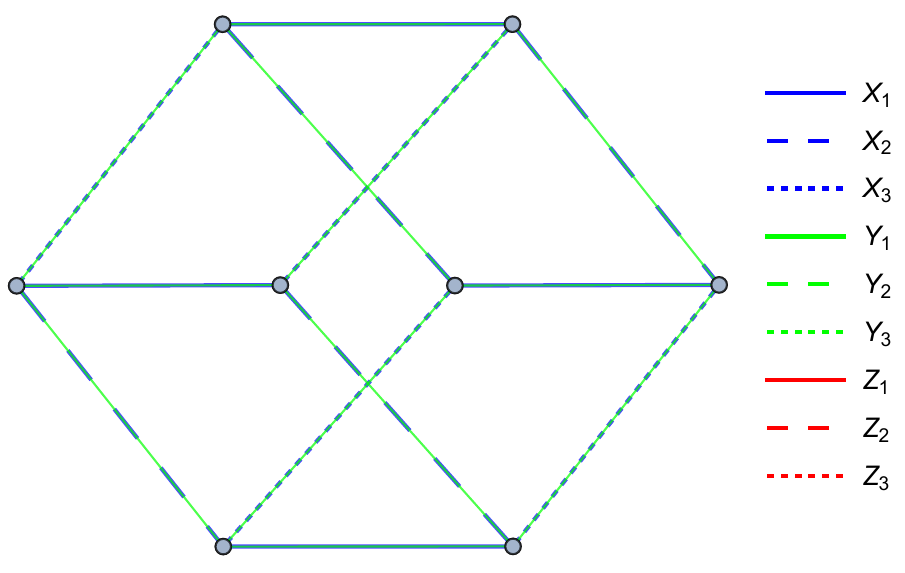}
        \caption{Reachability graph of state $\ket{D^3_3}$ under the action of $\Pi_3$. This graph contains $8$ vertices, with gates $\sigma_X^i$ and $\sigma_Y^i$ acting the same on $\ket{D^3_3}$. States of the form $\ket{D^N_N}$ are stabilizer states and are stabilized by $2^N$ elements of $\Pi_N$.}
    \label{D33PauliGraph}
    \end{figure}

Finally, we consider Dicke states $\ket{D^{N}_k}$ where $N=2k$. States $\ket{D^{2k}_k}$ are stabilized by the simultaneous action of $\sigma_X$ and $\sigma_Y$ on every qubit. For states $\ket{D^{2k}_k}$, we have
\begin{equation}\label{Stabilizer2kk}
\Stab_{\Pi_N}(\ket{D^{2k}_k}) \supset \left\{
    \bigotimes_{i=1}^{2k} \sigma_{X}^i,\, \bigotimes_{i=1}^{2k} \sigma_{Y}^i \right\},
\end{equation}
as well as $\mathbb{1}$ and Eq.\ \eqref{AllDickeStabilizer}. The two additional stabilizers in Eq.\ \eqref{Stabilizer2kk} result in $4$-element stabilizer subgroup for $\ket{D^{2k}_k}$ under the action of $\Pi_N$. 

For the example state $\ket{D^{4}_2}$, its stabilizer subgroup under $\Pi_4$ can be written
\begin{equation}
\Stab_{\Pi_4}(\ket{D^{4}_2}) =\{\mathbb{1},\,\sigma_X^1\sigma_X^2\sigma_X^3\sigma_X^4,\,\sigma_Y^1\sigma_Y^2\sigma_Y^3\sigma_Y^4,\,\sigma_Z^1\sigma_Z^2\sigma_Z^3\sigma_Z^4\}.
\end{equation}
Since $\ket{D^{4}_2}$ is stabilized by a $4$-element subgroup of $\Pi_4$, its orbit under $\Pi_4$, depicted in Figure \ref{D42PauliGraph} of Appendix \ref{lIndependentSimplification}, reaches $64$ states.

We have given the stabilizer subgroup for all $\ket{D^{N}_k}$ under action of the $N$-qubit Pauli group. We used the stabilizer subgroup to generate a reachability graph for $\ket{D^{N}_k}$, which represents each state's orbit under $\Pi_N$. In the following section we extend our analysis to consider action of the $N$-qubit Clifford group $\mathcal{C}_N$, as well as $\mathcal{C}_N$ subgroups. We use the reachability graphs of $\ket{D^{N}_k}$ to analyze entanglement structures observed in Dicke state orbits.

\subsection{Clifford Group Orbits and Entanglement Evolution}

Dicke states $\ket{D^N_k}$ are not stabilizer states for $N \geq 3$ and $N \neq k$. However, interestingly, states $\ket{D^N_k}$ are stabilized by more Clifford group elements than just $\mathbb{1}$. In this section, we extend our study of $\ket{D^N_k}$ orbits by considering the action of the two-qubit Clifford group $\mathcal{C}_2$. Since entanglement modification via Clifford gates occurs through bi-local action, this restriction to $\mathcal{C}_2$ is sufficient for exploring the evolution of Dicke state entropy vectors under Clifford circuits. We construct the stabilizer subgroup for each $\ket{D^N_k}$ under the action of $\mathcal{C}_2$, and compute the size of each orbit.

We also present reachability graphs for $\ket{D^N_k}$ under the action of the $\mathcal{C}_2$ subgroup $(HC)_{1,2} \equiv \langle H_1,\, H_2,\, C_{1,2},\, C_{2,1} \rangle$, with vertices colored by entropy vector as in \cite{Keeler2022,Keeler:2023xcx}. Since the $P_1$ and $P_2$ gates cannot modify a state's entropy vector, the subgroup $(HC)_{1,2}$ contains all non-trivial entropy vector dynamics. Furthermore, graph representations of $(HC)_{1,2}$ orbits are easier to parse than $\mathcal{C}_2$ orbits, as they contain a factor of $10$ less vertices. We use $(HC)_{1,2}$ orbits of Dicke states to give a bound on the number of times $\ket{D^N_k}$ entropy vectors can change under this gate set. A more general bound on entropy vector dynamics using quotient graphs is derived in \cite{Keeler:2023shl}.

All one and two-qubit Dicke states, $\ket{D^1_1}$ and $\ket{D^2_k}$, are also stabilizer states. Likewise, every $\ket{D^N_N}$ is a stabilizer state as well. Accordingly, states $\ket{D^2_k}$ and $\ket{D^N_N}$ are stabilized by a $192$-element subgroup of $\mathcal{C}_2$, and their reachability graph is exactly the two-qubit stabilizer state graph shown in Figure \ref{FullC2Graph} of Appendix \ref{lIndependentSimplification}. 

When we restrict to the action of $(HC)_{1,2} = \langle H_1,\,H_2\,C_{1,2}\, C_{2,1}\rangle$, the orbit of $\ket{D^2_k}$ and $\ket{D^N_N}$ contains $24$ states. Figure \ref{HCGraphD21} illustrates this orbit, showing the reachability graph of $\ket{D^N_N}$ under the action of $(HC)_{1,2}$. While the state $\ket{D^N_N}$ is unentangled, elements of the $(HC)_{1,2}$ subgroup are capable of generating instances of $GHZ$-type entanglement throughout the orbit.
	\begin{figure}[h]
		\begin{center}
		\begin{overpic}[width=13cm]{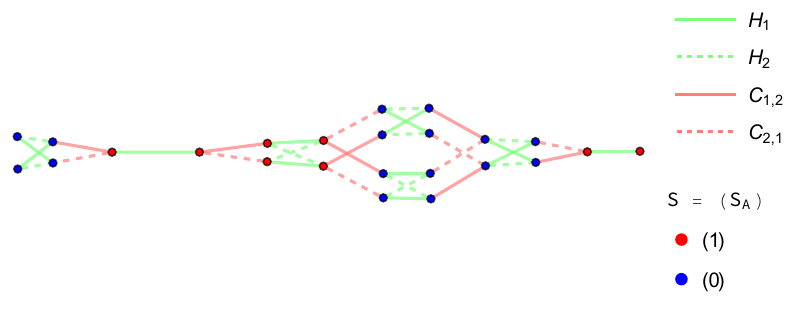}
		\put (61.5,23) {\footnotesize{$\swarrow \ket{D^N_N}$}}
        \put (13.5,16) {\footnotesize{$\uparrow \ket{GHZ}_N$}}
        \end{overpic}
        \caption{Orbit of $\ket{D^N_N}$ under $\langle H_1,\,H_2,\,C_1,\,C_2\rangle$ subgroup. This reachability graph has $24$ vertices and $2$ entanglement possibilities, unentangled and maximally-entangled. Since $\ket{D^N_N} = \ket{1}^{\otimes N}$, this reachability graph is shared by a subset of the $N$-qubit stabilizer states.}
		\label{HCGraphD21}
	\end{center}
	\end{figure}

Dicke states of the form $\ket{D^N_1}$, which define the set of $N$-qubit $W$ states, as well as states $\ket{D^N_{N-1}}$, are stabilized by $4$ elements of $\mathcal{C}_2$. Specifically, the states $\ket{D^N_1}$ and $\ket{D^N_{N-1}}$ have stabilizer subgroup
\begin{equation}\label{DN1StabSubgroup}
\begin{split}
        \Stab_{\mathcal{C}_2}(\ket{D^N_1}) &= \{\mathbb{1},\,H_2C_{1,2}H_2,\,C_{1,2}C_{2,1}C_{1,2},\,H_2C_{1,2}H_2C_{1,2}C_{2,1}C_{1,2}\},\\
        & = \Stab_{\mathcal{C}_2}(\ket{D^N_{N-1}}).
\end{split}
\end{equation}

The stabilizer group in Eq.\ \eqref{DN1StabSubgroup} yields an orbit of $2880$ states for $\ket{D^N_1}$ and $\ket{D^N_{N-1}}$, under the action of $\mathcal{C}_2$. 

Restricting group action to $(HC)_{1,2}$, the orbits of all $\ket{D^N_{1}}$ and $\ket{D^N_{N-1}}$, for $N \geq 3$, consist of $288$ states. Figure \ref{HCGraphD31} depicts the $\ket{D^N_{1}}$ reachability graph under $(HC)_{1,2}$, shown for the example state $\ket{D^3_{1}}$. While this reachability graph has $288$ vertices, it is not isomorphic to the $288$-vertex graph observed for stabilizer states under the action of $(HC)_{1,2}$, presented in \cite{Keeler2022}.
	\begin{figure}[h]
		\begin{center}
		\begin{overpic}[width=13cm]{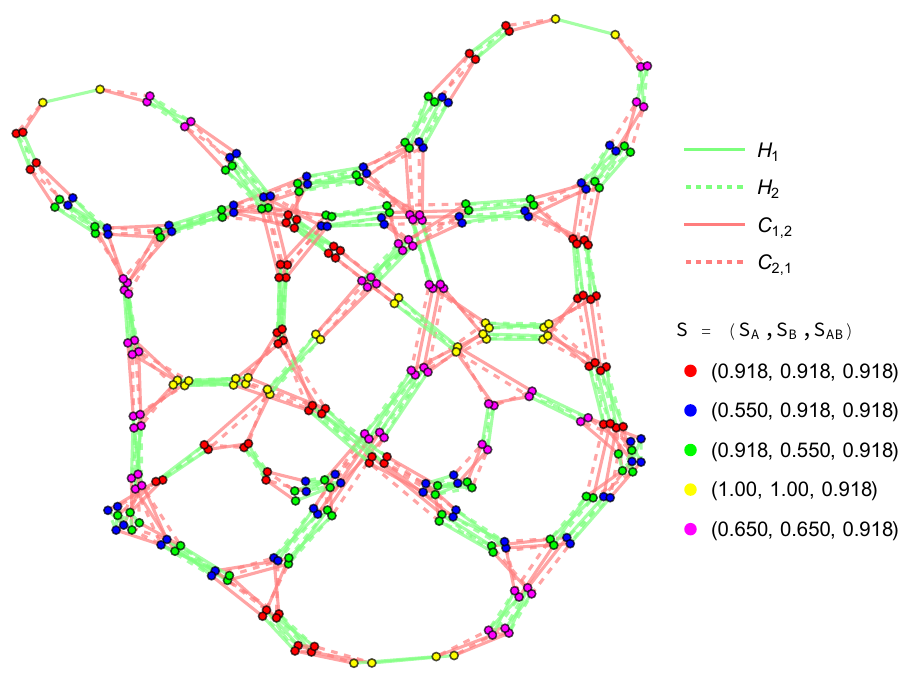}
		\put (53,67.9) {\footnotesize{$\leftarrow \ket{D^3_1}$}}
        \end{overpic}
        \caption{Orbit of $\ket{D^3_1}$ under $\langle H_1,\,H_2,\,C_1,\,C_2\rangle$ action. The graph has $288$ vertices, and contains $5$ different entropy vectors. We especially note the topological distinction of this graph, compared to the $288$-vertex stabilizer state graph. Numerical approximations for entanglement entropies are shown in the figure, with exact values in Table \ref{tab:EntropyVectorTable1}.}
		\label{HCGraphD31}
	\end{center}
	\end{figure}

The orbit of $\ket{D^N_{1}}$ under $(HC)_{1,2}$, for all $N \geq 3$, contains $5$ different entropy vectors. As described in \cite{Keeler:2023shl}, there are maximally $5$ unique entropy vectors that can be generated for states $\ket{D^N_{1}}$, using all $(HC)_{1,2}$ circuits in Figure \ref{HCGraphD31}. While the number of entropy vectors in graphs like Figure \ref{HCGraphD31} cannot increase beyond $5$, the number of different entanglement entropies comprising those entropy vectors, denoted $|s_N|$, continues to grow with increasing qubit number $N$. In Figure \ref{HCGraphD31}, the entropy vectors in the orbit of $\ket{D^3_1}$ are built of $4$ distinct entanglement entropies, shown to the right of the figure. For arbitrary $N$-qubit states $\ket{D^N_1}$, we conjecture the following:
\begin{conjecture}\label{EntanglementCardinality}
For $N \geq 2$, the number of unique entanglement entropies which comprise all entropy vectors in the $(HC)_{1,2}$ orbit of $\ket{D^N_1}$ increases as
\begin{equation}
    |s_N| = \lfloor \frac{5N-7}{2} \rfloor. 
\end{equation}
\end{conjecture}
The state $\ket{D^1_1}$ is pure and has zero entanglement entropy. The number of unique entanglement entropies encountered in the $(HC)_{1,2}$ orbit of $\ket{D^N_k}$ are depicted in Figure \ref{DickeEntropiesQubitNum}, for $N \leq 10$ qubits.

All remaining Dicke states $\ket{D^{N}_k}$, with $1 < k < N-1$, are stabilized by a $2$-element subgroup of $\mathcal{C}_2$. The stabilizer subgroup for such $\ket{D^{N}_k}$ states is given by
\begin{equation}\label{AllOtherDStabilizer}
\Stab_{\mathcal{C}_2}\left(\ket{D^N_k}\right) = \{\mathbb{1},\, C_{1,2}C_{2,1}C_{1,2},\}, \quad \forall \, 1 < k < N-1,
\end{equation}
Consequently, the $\mathcal{C}_2$ orbit of states stabilized by Eq.\ \eqref{AllOtherDStabilizer} reaches $5760$ states.

The action of $(HC)_{1,2}$ on $\ket{D^{N}_k}$, for $1 < k < N-1$, generates an orbit of $576$ states. This $576$-element orbit under $(HC)_{1,2}$ is particularly interesting as it differs in size from any stabilizer state orbit under $(HC)_{1,2}$ action \cite{Keeler2022,Keeler:2023xcx}. Stated alternatively, the stabilizer subgroup in Eq.\ \eqref{AllOtherDStabilizer} is not shared by any stabilizer state at any qubit number. As a result, reachability graphs with $576$ vertices, like that in Figure \ref{HCGraphD42}, are never witnessed for stabilizer states.
	\begin{figure}[h]
		\begin{center}
		\begin{overpic}[width=15cm]{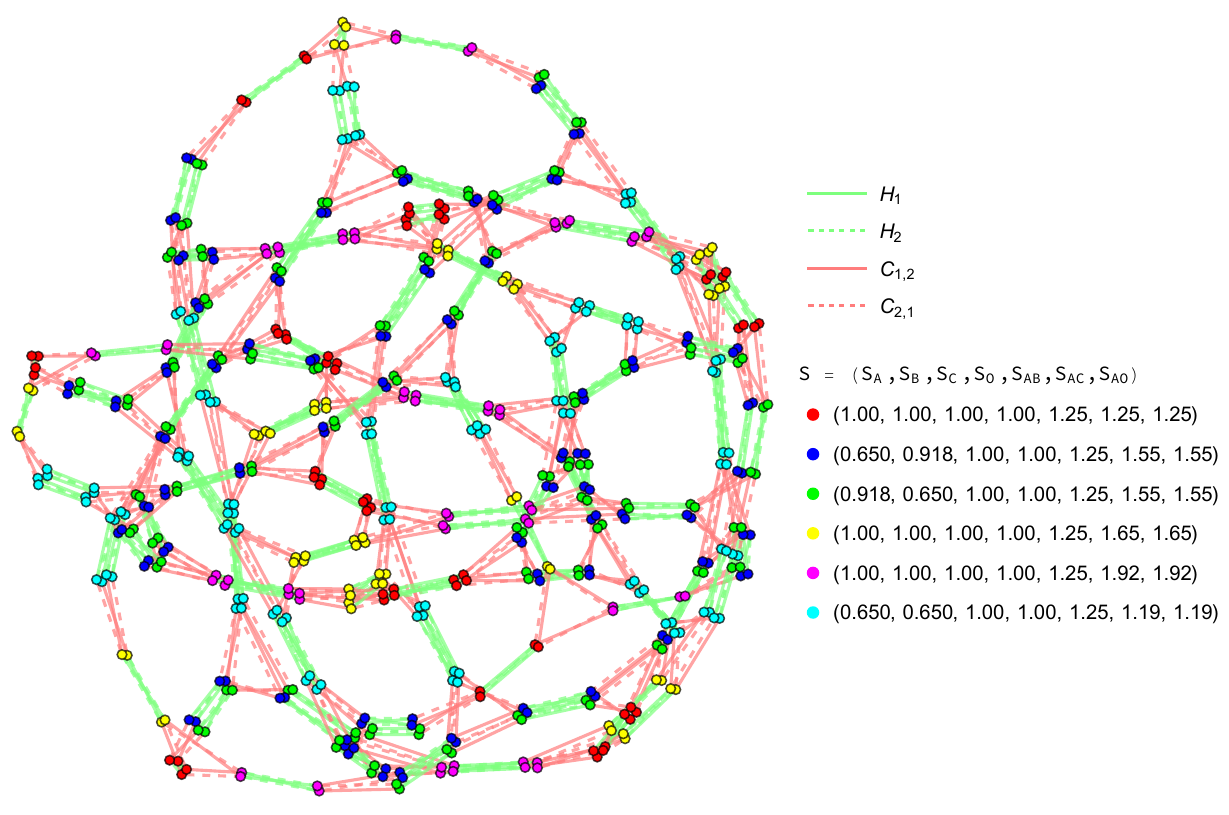}
		\put (19.6,55.5) {\footnotesize{$\uparrow \ket{D^4_2}$}}
        \end{overpic}
        \caption{Reachability graph showing the orbit of $\ket{D^4_2}$ under $(HC)_{1,2}$. This reachability graph has $576$ vertices, a vertex count never observed among stabilizer states, and contains $6$ different entropy vector possibilities. We provide numerical approximations for the entropy vector component in the figure, with exact values given in Table \ref{tab:EntropyVectorTable2}.}
		\label{HCGraphD42}
	\end{center}
	\end{figure}

For $k > 1$, the number of unique entanglement entropies that make up entropy vectors in the $(HC)_{1,2}$ orbit of $\ket{D^N_k}$ increases with system size. Figure \ref{DickeEntropiesQubitNum} illustrates the relationship between cardinality $|s_N|$ and qubit number $N$, for $\ket{D^N_k}$ up to $N=10$ qubits.
    \begin{figure}[h]
        \centering
        \includegraphics[width=14cm]{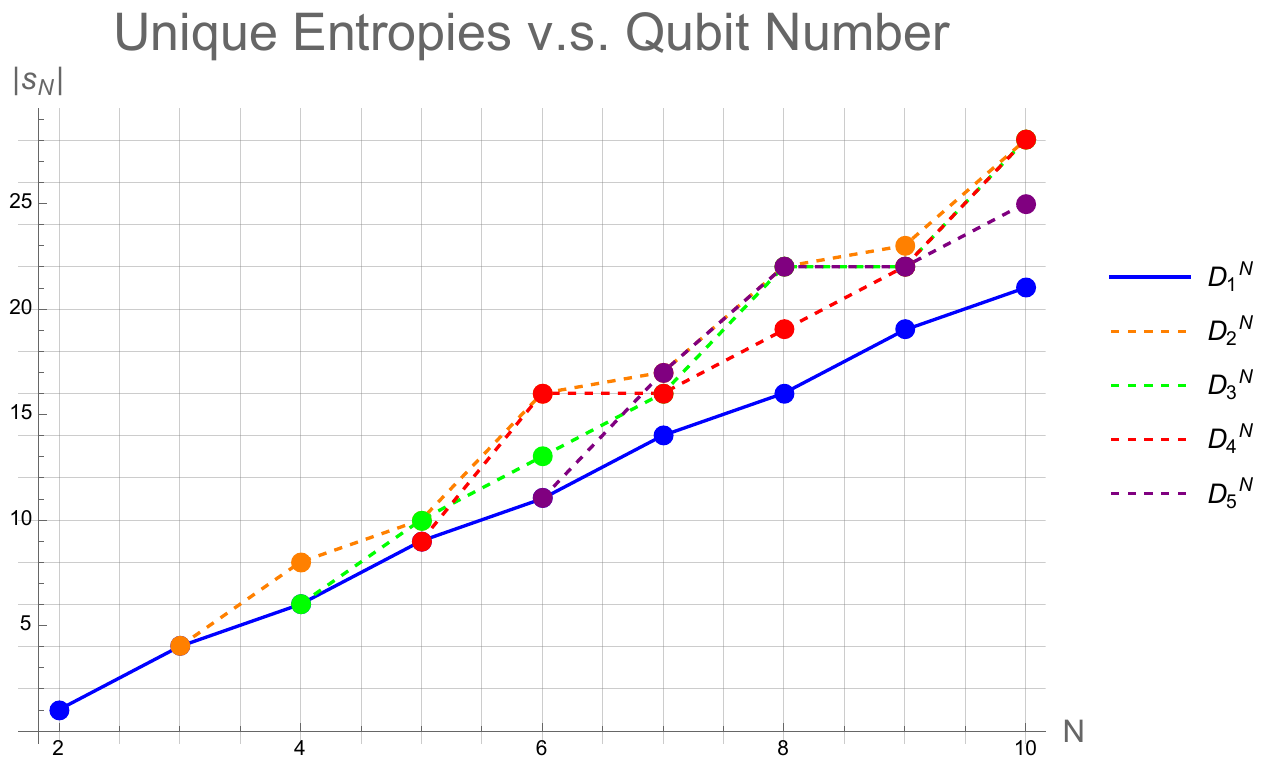}
        \caption{The number of unique entanglement entropies $|s_N|$ comprising all entropy vectors in the $(HC)_{1,2}$ orbit of $\ket{D^N_k}$. We plot this entanglement entropy cardinality against increasing qubit number $N$, for $N \leq 10$ and $1 \leq k \leq 5$. The solid blue line depicts the special case of $\ket{D^N_1}$ described in Conjecture \ref{EntanglementCardinality}.}
    \label{DickeEntropiesQubitNum}
    \end{figure}

Reachability graphs like that in Figure \ref{HCGraphD42} admit $6$ unique entropy vectors throughout the orbit. The bounds proposed in \cite{Keeler:2023shl} limit graphs isomorphic to Figure \ref{HCGraphD42} to having, at most, $9$ different entropy vectors. However, since entanglement dynamics additionally depends on the state being evolved through the quantum circuit, the symmetries of Dicke states constrain the number of entropy vectors in these graphs to $6$. As with the orbits of $\ket{D^N_1}$, while the overall number of entropy vectors in the reachability graph is fixed, for all $N$ and $k$, the number of distinct entanglement entropies which make up those vectors continues to increase for larger and larger $N$, and varies for different values of $k$.

We have identified the stabilizer subgroups for all Dicke states $\ket{D^N_k}$, under the action of the $N$-qubit Pauli group $\Pi_N$, as well as the two-qubit Clifford group $\mathcal{C}_2$. We demonstrated that there exist three distinct stabilizer subgroups for states $\ket{D^N_k}$, depending on the values of $N$ and $k$. States $\ket{D^N_N} = \ket{1}^{\otimes N}$ belong to the set of $N$-qubit stabilizer states, and share the corresponding stabilizer groups \cite{Keeler2022,Keeler:2023xcx}. States $\ket{D^N_1}$ and $\ket{D^N_{N-1}}$ share a stabilizer subgroup in $\Pi_N$ and $\mathcal{C}_2$, as do all $\ket{D^N_k}$ with $1 < k < N-1$. We illustrated the orbit of $\ket{D^N_k}$, under the action of $\Pi_N$ using reachability graphs.

In order to understand the evolution of Dicke state entropy vectors, we likewise constructed reachability graphs for all $\ket{D^N_k}$ under the action of the $\mathcal{C}_2$ subgroup $(HC)_{1,2} = \langle H_1,\,H_2,\,C_1,\,C_2\rangle$. Since entanglement modification in Clifford circuits occurs through bi-local action, restricting to this subgroup enabled us to place constraints on the dynamics of $\ket{D^N_k}$ entropy vectors under Clifford gates. We found the number of entropy vectors in each $(HC)_{1,2}$ orbit to be constant, with $5$ entropy vectors possible on graphs of $288$ vertices, and $6$ entropy vectors on graphs with $576$ vertices. While the number of entropy vectors on these graphs is fixed, the number of distinct entropies continued to increase.  

\section{Discussion}

In this work we constructed the entropy cone for $N$-qubit Dicke states $\ket{D^N_k}$ by calculating all entropy vectors for arbitrary values of $N$ and $k$. We first defined a function to compute the entanglement entropy $S_{\ell}$ of any $\ell$-party subsystem of $\ket{D^N_k}$. We demonstrated that $\ket{D^N_k}$ entropy vectors are manifestly symmetric, with $S_{\ell}$ only dependent on the size of subsystem $\ell$, and therefore lie within the convex hull of the SQEC. We likewise find that $\ket{D^N_k}$ entropy vectors are contained within the Stabilizer entropy cone as far as it is characterized, up to $N=5$. Dicke state entropy vectors do not, however, satisfy the necessary HEC or SHEC conditions for all $N \geq 3$, where $k \neq N$. We verified that our calculation accurately reproduces all vectors of the $N$-qubit $W$ state entropy cone \cite{Schnitzer:2022exe}, since $\ket{W_N} = \ket{D^N_1}$. 

We additionally define a prescription which realizes average entropies $\Tilde{S}_{\ell}$, for $\ket{D^N_k}$, as a min-cut protocol on weighted star graphs. Entropies $S_{\ell}$, as in Eq.\ \ref{DickeStateEntropies}, are computed as the minimum-weight edge cut on $\ell+1$ star graphs of $N+1$ legs each. Every star graph has a single edge of weight $w \leq 0$, with the precise value of $w$ constrained by the values of $N$ and $\ell$. The sum of two set of star graphs, one representing $S_{\ell}$ and one $S_{N-\ell}$, defines $\Tilde{S}_{\ell}$ for all $\ket{D^N_k}$. This graph representation of $\ket{D^N_k}$ entropies builds upon other min-cut protocols for symmetrized entropies \cite{Czech:2021rxe,Fadel2021,Schnitzer:2022exe}, and is an interesting direction for future study.

We studied the orbits of $\ket{D^N_k}$ under action of $N$-qubit Pauli group $\Pi_N$ and the $2$-qubit Clifford group $\mathcal{C}_2$. Interestingly, Dicke states form a set of non-stabilizer states which are stabilized by more Clifford elements than just $\mathbb{1}$. We identified the stabilizer subgroup for every $\ket{D^N_k}$, under the action of both groups, which we used to generate $\ket{D^N_k}$ reachability graphs \cite{Keeler:2023xcx}. Each $\ket{D^N_k}$ reachability graph depicts the state's orbit under the action of a chosen group, 
and generates all states which can be reached through circuits built of the generating gates. Since $\ket{D^N_k}$ is often initialized as the starting state for many quantum algorithms, 
the reachability graphs in Section \ref{StabilizerSection} provide a map through the Hilbert space for algorithms that begin with $\ket{D^N_k}$. Furthermore, since construction of reachability graphs is not limited to the Clifford group \cite{Munizzi_2022}, it would be interesting to explore Dicke state orbits under a universal set of gates. 

Reachability graphs can also be used to bound entanglement evolution under a chosen set of gates, by examining how many times the entropy can be changed by circuits in the graph \cite{Keeler:2023shl}. Motivated to explore the dynamics of $\ket{D^N_k}$ entropy vectors under $\mathcal{C}_2$, we focused on the subgroup $(HC)_{1,2} = \langle H_1,\,H_2,\,C_{1,2},\,C_{2,1} \rangle$ since entanglement in Clifford circuits occurs, at most, through the bi-local CNOT gate. We found that the number of entropy vectors on each $\ket{D^N_k}$ reachability graph is constant, with $5$ entropy vectors on the $288$-vertex graphs, like that in Figure \ref{HCGraphD31}, and $6$ entropy vectors on the $576$-vertex graphs, like Figure \ref{HCGraphD42}. While the number of entropy vectors is fixed, the number of distinct entanglement structures which compose each vector continues to increase for larger and larger systems.

We expect our analysis of the $N$-qubit Dicke state entropy cone and $\ket{D^N_k}$ orbits to generalize for qudit Dicke states. Circuits for deterministically preparing arbitrary qudit Dicke states are known, and many recursive generalizations Dicke state properties have also been demonstrated. We expect the  $\ket{D^N_k}$ entropy vectors presented in this paper to generalize similarly, with preliminary efforts towards $W$ state entropy cone generalization given in \cite{Schnitzer:2022exe}. The Pauli and Clifford groups can likewise be extended to arbitrary Hilbert space dimension \cite{Jagannathan:2010sb,Hostens_2005}, which would allow us to extend our orbit model and consider the evolution of entanglement evolution for higher-dimensional Dicke systems. 

The Dicke state stabilizers presented in this work find immediate application in stabilizer code construction. Given a scheme for encoding logical Dicke states and a suitable choice of measurement, we can construct an error-correcting channel directly using the stabilizers in Section \ref{StabilizerSection}. For specific Dicke states, such as $\ket{D^N_k}$, the entanglement structure renders the state robust to single-qubit loss, particularly at large $N$. We expect this characteristic to offer significant error-correction advantages when using Dicke state encoding for noisy processing. In future work, we explore this proposal and construct a class of Dicke stabilizer codes, evaluating their performance when compared to existing schemes.

Finally, the entanglement structure of certain Dicke states makes $\ket{D^N_k}$ an interesting candidate for magic distillation protocols \cite{Bravyi2004}. Coupled with the ease of preparing $\ket{D^N_k}$, Dicke states can enable improved protocols for distilling magic with minimal overhead. The states $\ket{D^5_1}$, $\ket{D^5_2}$, and $\ket{D^5_4}$ specifically possess a significant amount non-local magic \cite{Bao2022a}, though ultimately experience error rates slightly above the fault-tolerant Bravyi-Kitaev threshold. These error rates can be improved however, by passing $\ket{D^5_k}$ through a short sequence of gates. This further motivates an understanding $\ket{D^N_k}$ orbits under universal gate sets, which can provide circuits to improve the utility of Dicke states in distillation schemes.

\textit{HJS wishes to thank Matt Headrick for alerting him to the possible interest overlaps of the Arizona State University group with his. WM wishes to thank Adam Burchardt, ChunJun Cao, Jonathan Harper, Cynthia Keeler, and Jason Pollack for helpful discussions. WM is supported by the U.S. Department of Energy under grant number DE-SC0019470 and by the Heising-Simons Foundation ``Observational Signatures of Quantum Gravity'' collaboration grant 2021-2818. Dicke states were brought to our attention by Rafael Nepomechie in \cite{nepomechie2023qudit}, and in private communication to HJS.}

\clearpage
\begin{singlespace}
\printbibliography[heading=subbibliography]
\end{singlespace}

\chapter{BOUNDING ENTANGLEMENT ENTROPY WITH CONTRACTED GRAPHS}\label{Chapter6}

\textit{Following on our previous work \cite{Keeler2022,Keeler:2023xcx} studying the orbits of quantum states under Clifford circuits via `reachability graphs', we introduce `contracted graphs' whose vertices represent classes of quantum states with the same entropy vector.  These contracted graphs represent the double cosets of the Clifford group, where the left cosets are built from the stabilizer subgroup of the starting state and the right cosets are built from the entropy-preserving operators. We study contracted graphs for stabilizer states, as well as W states and Dicke states, discussing how the diameter of a state's contracted graph constrains the `entropic diversity' of its $2$-qubit Clifford orbit. We derive an upper bound on the number of entropy vectors that can be generated using any $n$-qubit Clifford circuit, for any quantum state. We speculate on the holographic implications for the relative proximity of gravitational duals of states within the same Clifford orbit. Although we concentrate on how entropy evolves under the Clifford group, our double-coset formalism, and thus the contracted graph picture, is extendable to generic gate sets and generic state properties.}

\section{Introduction}

One primary goal of quantum computation is to outperform classical computers: that is, for certain tasks, to take a classical input and compute a classical output more rapidly, or efficiently, than any known classical algorithm. (In recent years, this goal has been achieved or brought within reach for certain sets of problems \cite{arute2019quantum,zhong2020quantum}.) Intuitively, quantum computers can only do better on these tasks because they're doing something intrinsically \emph{quantum}: if they weren't, they couldn't outperform the classical method. Formalizing this intuitive result is an object of ongoing research: precisely what feature of a particular quantum algorithm allows it to gain an advantage?
\begin{figure}
  \begin{minipage}{.4\textwidth}
    \centering
    \includegraphics[width=6cm,height=6cm]{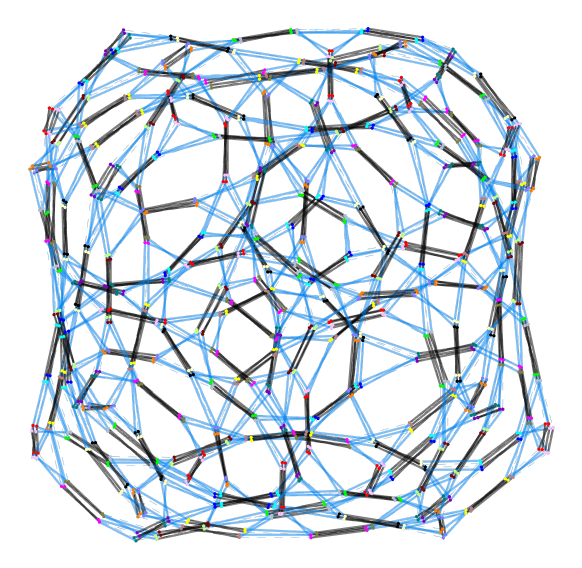}
  \end{minipage}
  \begin{minipage}{.15\textwidth}
    \huge{\begin{myequation}
     {\rightarrow}
    \end{myequation}}
  \end{minipage}
    \begin{minipage}{.4\textwidth}
    \centering
    \includegraphics[width=6cm,height=6cm]{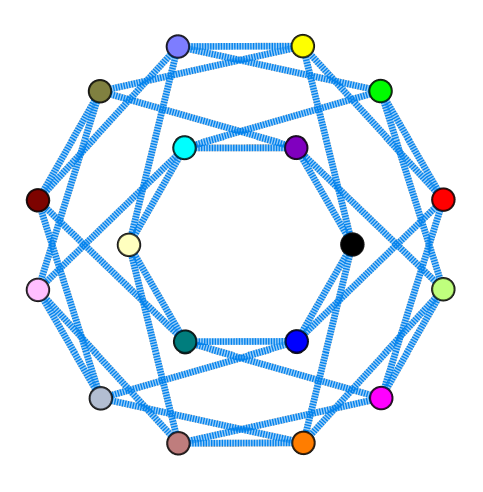}
  \end{minipage}
  \begin{minipage}{.4\textwidth}
  \vspace{.4cm}
    \begin{myequation}
     G/Stab_{G}(\ket{\psi})
    \end{myequation}
  \end{minipage}
  \begin{minipage}{.2\textwidth}
    \begin{myequation}
     \quad
    \end{myequation}
  \end{minipage}
    \begin{minipage}{.4\textwidth}
    \begin{myequation}
     H \backslash G/Stab_{G}(\ket{\psi})
    \end{myequation}
  \end{minipage}
\caption{A reachability graph and its reduction to a contracted graph. In this example, discussed in more detail in Figure \ref{G1152WithContractedGraph}, $G$ is the subgroup of the two-qubit Clifford group generated by Hadamard and $CNOT$ gates and $H$ is the set of operations which leave entropy vectors unchanged.}
\label{TitleFigContractedGraph}
\end{figure}

Setting aside not-even-wrong explanations like ``quantum computers act on each term in a superposition simultaneously," the folk wisdom is that the source of quantum advantage has something to do with interference, superposition, and entanglement. This appealing picture is challenged by the famous result that Clifford circuits, which are generated by the one-qubit Hadamard and phase gates and the two-qubit $CNOT$ gate, can be efficiently classically simulated \cite{Gottesman1998,aaronson2004improved}. That is, even though Clifford circuits can, via $CNOT$ gate applications, produce entanglement, they \emph{can't} give quantum speedups. Evidently, if some kind of entanglement is the key to quantum advantage, the type produced by Clifford gates doesn't suffice.

In order to understand the evolution of entanglement as a state is evolved through a quantum circuit, it's useful to track the \emph{entropy vector}, which characterizes the entanglement entropy of every subsystem of the state. In a recent series of papers, we have investigated how the entropy vector changes under the restricted action of Clifford gates acting on the first two qubits of a state. We first obtained \cite{Keeler2022} the \emph{reachability graphs}, colored by entropy vector, which show how stabilizer states evolve under the action of the two-qubit Clifford group $\mathcal{C}_2$ and its subgroups. In our second paper \cite{Keeler:2023xcx}, having better understood the underlying group-theoretic structures from which the reachability graphs are attained, we were able to find a representation of $\mathcal{C}_2$ as generated by the Clifford gates, as well as explore the reachability graphs produced from initial non-stabilizer states.

Although reachability graphs are useful for directly showing the action of explicit circuits and explicit states, they fail to fully illuminate the paths by which the entropy vector can change. The problem, in short, is that some circuits, even when they contain $CNOT$ gates, fail to change the entropy. For example, one defining relation of $\mathcal{C}_2$ is \cite{Keeler:2023xcx}
\begin{equation}
    \left(CNOT_{1,2}P_2 \right)^4 = P_1^2.
\end{equation}
Hence the structure of reachability graphs by themselves can only loosely bound how the entropy vector might change.

In this paper, we accordingly pass to a more concise graphical representation, the \emph{contracted graphs}, whose vertices represent not single states but classes of states with the same entropy vector. We show how to construct these graphs from the \emph{double cosets} of the Clifford group $\mathcal{C}_2$ and its cosets. An example of this procedure is shown in Figure \ref{TitleFigContractedGraph}. Our protocol for constructing contracted graphs is easily generalized to groups beyond the Clifford group and state properties beyond the entropy vector, and might be of use for other applications.

The remainder of this paper is organized as follows. In Section \ref{sec:review}, we review the Clifford group and stabilizer formalism, as well as the group-theoretic concepts of cosets and double cosets. We also recall the objects used in our previous papers: Cayley graphs, reachability graphs, and entropy vectors. In Section \ref{sec:building}, we give a general procedure for constructing the contracted graphs which retain information about entropy-changing operations in a group. In Section \ref{ContractedGraphsSection}, we apply this procedure to $\mathcal{C}_2$ and its subgroup $\HC$. For each of the reachability graphs in our previous papers, we obtain the resulting contracted graph, and show how these combine together under the action of the full Clifford group. In Section \ref{sec:diversity} we consider the diameter and entropic diversity of the reachability graphs, and discuss implications for the available transformations on a dual geometry via holography. In Section \ref{sec:discussion} we conclude and discuss future work. Appendix \ref{EntropyVectorTables} collects additional details of our computations.

\section{Review}\label{sec:review}

\subsection{Clifford Group and Stabilizer Formalism}

The Pauli matrices are a set of unitary and Hermitian operators, defined in the computational basis $\{\ket{0},\ket{1}\}$ as
\begin{equation}\label{PauliMatrices}
    \mathbb{1}\equiv\begin{bmatrix}1&0\\0&1\end{bmatrix}, \,\, \sigma_X\equiv\begin{bmatrix}0&1\\1&0\end{bmatrix}, \,\,
    \sigma_Y\equiv\begin{bmatrix}0&-i\\i&0\end{bmatrix}, \,\,
    \sigma_Z\equiv\begin{bmatrix}1&0\\0&-1\end{bmatrix}.
\end{equation}
The multiplicative matrix group generated by $\sigma_X,\,\sigma_Y,$ and $\sigma_Z$ is known as the single-qubit Pauli group $\Pi_1$, which we write
\begin{equation}
    \Pi_1 \equiv \langle \sigma_X,\,\sigma_Y,\,\sigma_Z \rangle.
\end{equation}
When $\Pi_1$ acts on a Hilbert space $\Hil \equiv \mathbb{C}^2$, in the fixed basis spanned by $\{\ket{0},\,\ket{1}\}$, it generates the algebra of all linear operations on $\Hil$.

The Clifford group is likewise a multiplicative matrix group, generated by the Hadamard, phase, and CNOT operations:
\begin{equation}
    H\equiv \frac{1}{\sqrt{2}}\begin{bmatrix}1&1\\1&-1\end{bmatrix}, \quad P\equiv \begin{bmatrix}1&0\\0&i\end{bmatrix}, \quad     C_{i,j} \equiv \begin{bmatrix}
            1 & 0 & 0 & 0\\
            0 & 1 & 0 & 0\\
	    0 & 0 & 0 & 1\\
	    0 & 0 & 1 & 0
            \end{bmatrix}.
\end{equation}
The CNOT gate is a bi-local operation which, depending on the state of one qubit, the control bit, may act with a $\sigma_X$ operation on a second qubit, the target bit. For the gate $C_{i,j}$, the first subscript index denotes the control bit and the second subscript the target bit. We define the single qubit Clifford group $\mathcal{C}_1$ as the group $\langle H,\,P \rangle$. Elements of $\mathcal{C}_1$ act as automorphisms on $\Pi_1$ under conjugation; hence $\mathcal{C}_1$ is the normalizer of $\Pi_1$ in $L(\Hil)$.

When considering the action of the Pauli and Clifford groups on multi-qubit systems, we compose strings of operators which act collectively on an $n$-qubit state. For an element of $\Pi_1$ which acts locally on the $k^{th}$ qubit in an $n$-qubit system, for example, we write
\begin{equation}\label{PauliString}
    I^1\otimes\ldots\otimes I^{k-1} \otimes \sigma_X^k \otimes I^{k+1} \otimes \ldots \otimes I^n.
\end{equation}
Eq. \eqref{PauliString} is referred to as a Pauli string, where the weight of each string counts the number non-identity insertions. The multiplicative group generated by all Pauli strings of weight $1$ is the $n$-qubit Pauli group $\Pi_n$.

We similarly can extend the action of $\mathcal{C}_1$ to multiple qubits, now incorporating $C_{i,j}$ into the generating set. Composing Clifford strings analogously to Eq. \eqref{PauliString}, we define the $n$-qubit Clifford group $\mathcal{C}_n$ as\footnote{This is not the minimal generating set for the Clifford group, since some Clifford gates can be written in terms of the others, c.f.\ Eqs.\ (4.5, 4.6) of \cite{Keeler:2023xcx}. A more minimal definition of the Clifford group is $C_n\equiv\langle\{H_i,P_1,C_{i,j}\}\rangle$ where $i\in\{1\ldots n\}$, $j>i$.}
\begin{equation}
    \mathcal{C}_n \equiv \langle H_1,...,\,H_n,\,P_1,...,P_n,\,C_{1,2},\,C_{2,1},...,\,C_{n-1,n},\,C_{n,n-1}\rangle.
\end{equation}
When indicating the action of some local gate, Hadamard or phase, the gate subscript denotes which qubit the gate acts on, e.g.\ $H_1$ for the action of Hadamard on the first qubit of an $n$-qubit system.

Beginning with any $n$-qubit computational basis state, e.g.\ $\ket{0}^{\otimes n}$, the group $\mathcal{C}_n$ is sufficient to generate the full set of $n$-qubit stabilizer states. As we noted in the introduction, stabilizer states are notable in quantum computing as a set of quantum systems which can be efficiently simulated with classical computing \cite{aaronson2004improved,Gottesman1997}. Additionally, stabilizer states comprise the elements of $\Hil$ which are left invariant under a $2^n$ element subgroup of $\Pi_n$. Since the group $\mathcal{C}_n$ is finite, the set of $n$-qubit stabilizer states $S_n$ is also finite \cite{doi:10.1063/1.4818950} and has order given by
\begin{equation}\label{StabilizerSetSize}
    \left|S_n\right| = 2^n \prod_{k=0}^{n-1} (2^{n-k}+1).
\end{equation}

\subsection{Cosets and Double Cosets}

Throughout this paper we support our graph models with parallel group-theoretic arguments. Many of our explanations make substantial use of coset and double coset constructions, which we review here. We also take this opportunity to set notation and establish language that will be used throughout the remainder of the paper.

Let $G$ be a group and $K \leq G$ an arbitrary subgroup. The set of all left cosets of $K$ in $G$ are constructed as
\begin{equation}\label{LeftCosetDefintion}
 g\cdot K, \quad  \forall g \in G.
\end{equation}
Each left coset built in Eq.\ \eqref{LeftCosetDefintion} is an equivalence set of elements $[g_i]$, which are equivalent under $K$ group action on the right:
\begin{equation}\label{LeftCosetEquivalenceRelation}
    g_i \sim g_j \Longleftrightarrow \exists \,k \in K : g_i = g_j k.
\end{equation}
Any two cosets $[g_i]$ in $g\cdot K$ must be either equal or disjoint, and every $g \in G$ must be found in, at most, one equivalence class. As a result, the set of all $[g_i]$ gives a complete decomposition of $G$.

Eqs. \eqref{LeftCosetDefintion} and \eqref{LeftCosetEquivalenceRelation}, as well as the accompanying explanations, apply analogously when generating all right cosets $H \cdot g$, for arbitrary $H \leq G$. We build all right cosets by computing $H \cdot g$, for every $g \in G$, where each equivalence class $[g_i]$ is now determined by left subgroup action
\begin{equation}
    g_i \sim g_j \Longleftrightarrow \exists \,h \in H : g_i = hg_j.
\end{equation}
When $H \leq G$ is normal in $G$, the left and right cosets are equal, and both $H \cdot g$ and $g \cdot H$ form a group under the same binary operation which defines $G$.

Two subgroups $H,K \leq G$ can be used to construct double cosets of $G$. We build each $(H,K)$ double coset by acting on $g \in G$ on the right by subgroup $K$, and on the left by $H$, explicitly
\begin{equation}\label{DoubleCosetDefinition}
 H\cdot g \cdot K, \quad  \forall g \in G,\, h\in H,\, k\in K.
\end{equation}
The double coset space built using Eq.\ \eqref{DoubleCosetDefinition} is denoted $H \backslash G / K$, and is defined by the equivalence relation
\begin{equation}
    g_i \sim g_j \Longleftrightarrow \exists \,h \in H,\, k \in K : g_i = hg_jk.
\end{equation}

In order to utilize the above coset constructions in this paper, we invoke several foundational group theory concepts (see e.g.\ \cite{Alperin1995}). First, for a finite group $G$, the order of any subgroup $K \leq G$ partitions the order of $G$ by Lagrange's theorem
\begin{equation}\label{LagrangeTheorem}
    \frac{|G|}{|K|} = [G:K], \quad \forall K\leq G,
\end{equation}
where $[G:K] \in \mathbb{N}$ is the number of left (or right) cosets of $K$ in $G$. When acting with $G$ on a set $X$, the orbit-stabilizer theorem fixes the size of each orbit $[G \cdot x]$ to be
\begin{equation}\label{OrbitStabilizerTheorem}
    [G \cdot x] = [G:K] = \frac{|G|}{|K|}, \quad \forall x\leq X,
\end{equation}
where $K \leq G$ is the set of elements which map an $x \in X$ to itself.

We can likewise use Eq.\ \eqref{LagrangeTheorem} with Eq.\ \eqref{OrbitStabilizerTheorem} to compute the order of a double coset space, i.e.\ the orbit of all left (or right) cosets under left (or right) subgroup action. For finite $G$ and subgroups $H,K \leq G$, the order%
\footnote{Note that a direct application of Lagrange's theorem to the order of a double coset space is false, i.e.\ the order of a double coset space of $G$ does not necessarily divide $|G|$.} %
of $H \backslash G / K$ is computed as
\begin{equation}\label{DoubleCosetOrder}
     |H \backslash G / K| = \frac{1}{|H||K|} \sum_{(h,k) \in H \times K} |G^{(h,k)}|, 
\end{equation}
where $G^{(h,k)}$ is the set of equivalence classes $[g_i]$ under Eq.\ \eqref{DoubleCosetDefinition}. The sum in Eq.\ \eqref{DoubleCosetOrder} is taken over all ordered pairs $(h,k)$ of $h \in H$ and $k \in K$.

\subsection{Cayley Graphs and Reachability Graphs}\label{CayleyGraphReview}

A \textit{Cayley graph} encodes in graphical form the structure of a group. For a group $G$ and a chosen set of generators, we construct the Cayley graph of $G$ by assigning a vertex for every $g \in G$, and an edge%
\footnote{Formally, each edge in a Cayley graph is directed. However, for improved legibility, we will often represent group generators which are their own inverse using undirected edges.} %
for every generator of $G$. When $G$ corresponds to a set of quantum operators acting on a Hilbert space, paths in the Cayley graph represent quantum circuits that can be composed using the generating gate set. Different paths which start and end on the same pair of vertices indicate sequences of operators whose action on any quantum state is identical. Loops in a Cayley graph represent operations equivalent to the identity.

For a group $G \subset L(\Hil)$, we define the stabilizer subgroup $\Stab_G(\ket{\psi})$ of some $\ket{\psi} \in \Hil$ as the subset of elements $g \in G$ which leave $\ket{\psi}$ unchanged,
\begin{equation}
    \Stab_G(\ket{\psi}) \equiv \{g \in G \,|\, g\ket{\psi} = \psi \}.
\end{equation}
In other words, the subgroup $\Stab_G(\ket{\psi})$ consists of all $g \in G$ for which $\ket{\psi}$ is a $+1$ eigenvector.

Reachability graphs can be obtained more generally as quotients of Cayley graphs \cite{Keeler:2023xcx,githubStab,githubCayley}. To perform this procedure, we first identify a group $G \in L(\Hil)$ to act on a Hilbert space $\Hil$, and a generating set for $G$. We then first quotient $G$ by any subgroup of elements which act as an overall phase on the group. For $\mathcal{C}_n$, this is the subgroup $\langle \omega \rangle$, where
\begin{equation}\label{OmegaQuotient}
    \omega \equiv \left(H_iP_i\right)^3 = e^{i\pi/4}\mathbb{1}.
\end{equation}
Once we have removed overall phase and constructed the quotient group%
\footnote{Since $\langle \omega \rangle < \mathcal{C}_n$ is normal, modding by $\langle \omega \rangle$ builds a proper quotient. It therefore does not matter whether we apply $\langle \omega \rangle$ on the left or right of $G$ when building cosets, nor does it affect any subsequent double coset construction. Accordingly, when it will not cause confusion we continue to use $G$ to refer to the group modded by global phase, rather than using an alternative notation.} %
$\bar{G} = G/\langle \omega \rangle$, we identify a state $\ket{\psi} \in \Hil$. Selecting  $\ket{\psi}$ immediately defines the stabilizer subgroup $\Stab_{\bar{G}}(\ket{\psi})$. We then construct the left coset space $\bar{G}/\Stab_{\bar{G}}(\ket{\psi})$ whose elements are
\begin{equation}\label{StabGroupQuotient}
    g\cdot\Stab_{\bar{G}}(\ket{\psi}) \quad \forall g \in \bar{G}.
\end{equation}

To graphically represent this procedure, we begin with a graph $\Gamma \equiv (V,E)$, which we quotient by first defining a partition on the vertices $V$. This partition induces the equivalence relation $u \sim v$ iff $u$ and $v$ lie in the same subset of the partition, defined for any $u,v \in V$. In this way, each vertex in the quotient graph represents one subset of the partition, and two vertices in the quotient graph are considered adjacent if any two elements of their respective subsets are adjacent in $\Gamma$.

While the graphs in this paper often represent groups, constructing a graph quotient is not equivalent to quotienting a group. Building a group quotient requires modding by a normal subgroup, which ensures that the left and right coset spaces of the chosen subgroup are equal, preserving the original group action in quotient group. We do not impose such a requirement when building graph quotients in this paper, even when our graphs illustrate the relation between groups of operators. We distinguish graph quotients from group quotients wherever potential confusion could occur.

\subsection{Entropy Vectors and Entropy Cones}

For a state $\ket{\psi} \in \Hil$, and some specified factorization for $\Hil$, we can compute the von Neumann entropy of the associated density matrix:
\begin{equation}\label{vonNeumannEntropy}
   S_{\psi} \equiv -\Tr \rho_{\psi} \ln \rho_{\psi},
\end{equation}
where $\rho_{\psi} \equiv \ket{\psi}\bra{\psi}$. For $\ket{\psi}$ a pure state, the property $\rho_{\psi}^2 = \rho_{\psi}$ implies $S_{\psi} = 0$. Throughout this paper, we measure information in \textit{bits}, and entropies in Eq.\ \eqref{vonNeumannEntropy} are computed with $\log_2$.

For a multi-partite pure state $\ket{\psi}$, we can still observe non-zero entanglement entropy among complementary subsystems of $\ket{\psi}$. Let $\ket{\psi}$ be some $n$-party pure state, and let $I$ denote an $\ell$-party subsystem of $\ket{\psi}$. We can compute the entanglement entropy between $I$ and its $(n-\ell)$-party complement, $\bar{I}$, using
\begin{equation}\label{SubsystemEntropyDefinition}
   S_{I} = -\Tr \rho_{I} \ln \rho_{I}.
\end{equation}
The object $\rho_{I}$ in Eq.\ \eqref{SubsystemEntropyDefinition} indicates the reduced density matrix of subsystem $I$, which is computed by tracing out the complement subsystem $\bar{I}$.

In general, there are $2^n-1$ possible subsystem entropies we can compute for any $n$-qubit pure state $\ket{\psi}$. Computing each $S_{I}$, using Eq.\ \eqref{SubsystemEntropyDefinition}, and arranging all entropies into an ordered tuple defines the entropy vector $\Vec{S}\left(\ket{\psi} \right)$. As an example, consider the $4$-qubit pure state $\ket{\psi}$, where $\Vec{S}\left(\ket{\psi} \right)$ is defined
\small{
\begin{equation}\label{EntropyVectorExample}
    \Vec{S} = (S_A,S_B,S_C,S_O; S_{AB},S_{AC},S_{AO},S_{BC},S_{BO},S_{CO};S_{ABC},S_{ABO},S_{ACO},S_{BCO}; S_{ABCO}).
\end{equation}}\normalsize
In Eq.\ \eqref{EntropyVectorExample} we use a semicolon to separate entropy components for subregions of distinct cardinality $|I|$. Additionally, for an $n$-qubit state it is customary to denote the $n^{th}$ subsystem using $O$, as this region acts as a purifier for the other $n-1$ parties. 

For an $n$-party system, each entropy vector contains $2^n-1$ components, with the first $n$ components representing single-qubit subsystems. We list entropy vector components in lexicographic order: with the first region denoted $A$, the second region denoted $B$, and so forth. Unlike what is sometimes found in the literature, we use $O$ to represent a smaller bipartition, instead of the one which does not contain the purifier. For example, in Eq.\ \eqref{EntropyVectorExample} we declare $O$ a single-party subsystem which purifies $ABC$, and write $S_O$ in place of $S_{ABC}$ among the single-party entries of the entropy vector.

When $\ket{\psi}$ is a pure state, the condition $S_{\psi} = 0$ implies an additional equivalence between entropies of complement subsystems
\begin{equation}\label{SubregionComplement}
    S_{I} = S_{\bar{I}}.
\end{equation}
Using Eq.\ \eqref{SubregionComplement} we can write $\Vec{S}\left(\ket{\psi} \right)$, for a pure state $\ket{\psi}$, using only $2^{N-1}-1$ entropies. For example, the entropy vector in Eq.\ \eqref{EntropyVectorExample} simplifies to the form
\begin{equation}\label{ReducedEntropyVectorExample}
    \Vec{S} = (S_A,S_B,S_C,S_O; S_{AB},S_{AC},S_{AO}).
\end{equation}
Since we are always considering pure states in this paper, all entropy vectors are written using the reduced notation in Eq.\ \eqref{ReducedEntropyVectorExample}.

\section{Building Contracted Graphs}\label{sec:building}

We now define a procedure to quotient reachability graphs by operations which preserve some specified property of a quantum system. In this paper we focus on the evolution of entanglement entropy under the action of the Clifford group; however, this prescription is sufficiently general to study any state property%
\footnote{In this work, the term \textit{state property} refers to anything computable from knowledge of the state, along with some additional information such as a specified factorization of the Hilbert space. We do not restrict analysis to properties which are observables; in fact, the main property discussed in this paper, the entropy vector, is not itself an observable.\label{fn:state_property}} %
under the action of any finitely-generated group.

We build a \textit{contracted graph} by identifying vertices in a reachability graph which are connected by entropy-preserving circuits. In this way, a contracted graph details the evolution of a state's entropy vector under the chosen gate set. The number of vertices in a contracted graph gives a strict upper bound on the number of different entanglement vector values reachable via circuits constructed using the chosen gate set. We will later use contracted graphs to derive an upper bound on entropy vector variation in Clifford circuits.

We now give an algorithm for generating contracted graphs.

\begin{enumerate}
    \item We first select a group $G$, and a generating set for $G$, as well as a property of our quantum system we wish to study under the action of $G$.

    \item  We next build the Cayley graph for $G$ by assigning a vertex for every $g \in G$, and a directed edge for each generator action on an element $g \in G$. We quotient $G$, and its Cayley graph, by any subgroup which acts as a global phase on the group, such as in Eq.\ \eqref{OmegaQuotient}.
    
    \item Next, we construct the reachability graph for some $\ket{\psi}$ under the action of $G$, as detailed in Subsection \ref{CayleyGraphReview}, which we denote%
    \footnote{A more precise notation for such reachability graphs would be $\mathcal{R}\left(\textnormal{Stab}_G\left(\ket{\psi}\right)\right)$, however we choose $\mathcal{R}_G\left(\ket{\psi}\right)$ instead for brevity.} %
    $\mathcal{R}_G\left(\ket{\psi}\right)$.
    We determine the stabilizer subgroup $\Stab_G\left(\ket{\psi}\right)$ for $\ket{\psi}$, and generate the left coset space $G/ \Stab_{G}\left( \ket{\psi} \right)$ using the equivalence relation
    \begin{equation}\label{StabilizerEquivalence}
        g_i \sim g_j \Longleftrightarrow \exists \,s \in \Stab_G\left(\ket{\psi}\right) : g_i = g_js.
    \end{equation}

    We glue together vertices in the Cayley graph of $G$ that correspond to elements which share an equivalence class $[g_i]$ in $G/ \Stab_{G}\left( \ket{\psi} \right)$. This graph quotient yields $\mathcal{R}_G\left(\ket{\psi}\right)$.
    
    \item We now identify the subgroup $H \leq G$ of elements that leave the entropy vector of any state invariant. The subgroup $H$ defines the equivalence relation
    \begin{equation}\label{EntropyEquivalence}
        g_i \sim g_j \Longleftrightarrow \exists \,h \in H : g_i = hg_j.
    \end{equation}
    For any $G$, the group $H$ will at least contain all $g \in G$ which act as local gates on a single qubit, since local action cannot modify entanglement. However, $H$ may also contain additional circuits which do not change the entropy vector.
    
    \item Finally, we build all double cosets $H \backslash G / \Stab_{G}\left( \ket{\psi} \right)$. We identify all vertices in $\mathcal{R}_G\left(\ket{\psi}\right)$ which share an equivalence class in $H \backslash G / \Stab_{G}\left( \ket{\psi} \right)$, and subsequently quotient $\mathcal{R}_G\left(\ket{\psi}\right)$ to give the final contracted graph.
\end{enumerate}

We generate reachability graphs by building left cosets $G/ \Stab_{G}\left( \ket{\psi} \right)$, defined by an equivalence up to right subgroup action by $\Stab_{G}\left( \ket{\psi} \right)$ as in Eq.\ \eqref{StabilizerEquivalence}. Since $\Stab_{G}\left( \ket{\psi} \right)$ acts trivially on $\ket{\psi}$, appending any $s \in \Stab_{G}\left( \ket{\psi} \right)$ to the right of any  $g \in G$ does not change how $g$ transforms the state $\ket{\psi}$. Conversely, we build a contracted graph by generating right cosets $G \backslash H$, with equivalence defined up to left subgroup action as shown in Eq.\ \eqref{EntropyEquivalence}. Every element of $H$ preserves a state's entropy vector, therefore acting on the left of $g\ket{\psi}$ by any $h \in H$ does not change the measurement of the full state entropy vector, for every $g \in G$.

Recall that there are two interpretations of a reachability graph. By identifying a state $\ket{\psi}$ and group $G$ of operators acting on that state, $\mathcal{R}_G\left(\ket{\psi}\right)$ represents the orbit of $\ket{\psi}$ under the action of $G$. In this state-orbit interpretation, vertices of $\mathcal{R}_G\left(\ket{\psi}\right)$ represent states reached in the orbit of $\ket{\psi}$. For simplicity, we choose this state-orbit interpretation in this explanatory section. A more general interpretation of reachability graphs exists which defines $\mathcal{R}_G\left(\ket{\psi}\right)$ as a quotient space of the Cayley graph of the abstract group $G$. In this interpretation, vertices represent equivalence classes of $g \in G$ defined by the left coset $g \cdot \Stab_G\left(\ket{\psi}\right)$.

\paragraph{Example:}
For clarity, we now work through an explicit example. Consider the subgroup of the two-qubit Clifford group%
\footnote{For additional detail on this Clifford subgroup see Section 4.2 of \cite{Keeler:2023xcx}, where all group elements are derived using two-qubit Clifford group relations.} %
generated by the $P_2$ and $CNOT_{1,2}$ gates,
\begin{equation}
    G \equiv \langle P_2,\, CNOT_{1,2}\rangle.
\end{equation}
The group $\langle P_2,\, CNOT_{1,2}\rangle$ consists of $32$ elements, specifically
\begin{equation}
   \langle P_2,\, CNOT_{1,2}\rangle = \{p,\,pCp,\,C\bar{p}Cp\},
\end{equation}
where we introduce the notations
\begin{equation}
    p \in \{\mathbb{1},\,P_2,\,P_2^2,\,P_2^3\}, \qquad \bar{p} \in \{\,P_2,\,P_2^2,\,P_2^3\}.
\end{equation}

We select the state $\ket{\psi} = \left(\ket{00} + 2\ket{01} + 4\ket{10}+ 3\ket{11}\right)/\sqrt{30}$, which we choose for its particular entropic properties that we will discuss at the end of the section. We construct the reachability graph $\mathcal{R}_G\left(\ket{\psi}\right)$ for $\ket{\psi}$, shown in the left panel of Figure \ref{P2C12CayleyWithContractedGraph}. The only element of $G$ which leaves $\ket{\psi}$ invariant is $\mathbb{1}$ in $G$, therefore
\begin{equation}\label{IdentityStabGroup}
    \Stab_G(\ket{\psi}) = \{ \mathbb{1}\}. 
\end{equation}
Since the stabilizer group in Eq.\ \eqref{IdentityStabGroup} consists of just the identity, and is therefore a normal subgroup, the group $\Stab_G(\ket{\psi})$ quotients $G$ and the reachability graph $\mathcal{R}_G\left(\ket{\psi}\right)$ is exactly the $32$-vertex Cayley graph. In the more general case, $\mathcal{R}_G\left(\ket{\psi}\right)$ would not necessarily represent a group quotient, but would represent a left coset space.
    \begin{figure}[h]
        \centering
        \includegraphics[width=13cm]{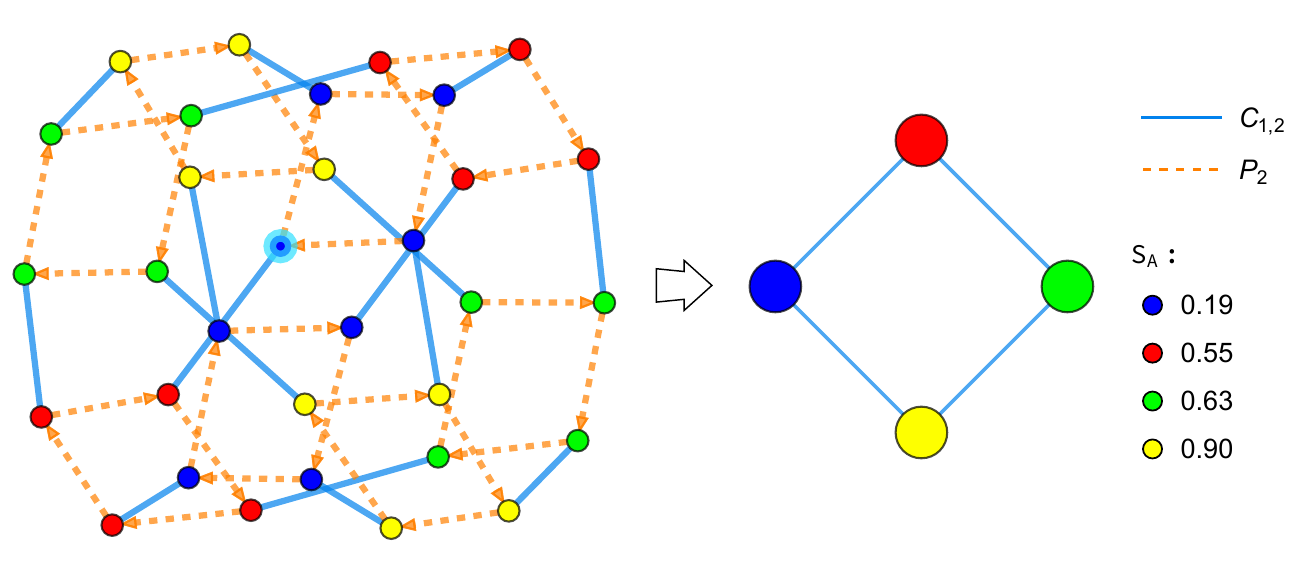}
        \caption{Reachability graph (left) of $\ket{\psi} = \left(\ket{00} + 2\ket{01} + 4\ket{10}+ 3\ket{11}\right)/\sqrt{30}$, highlighted in cyan, under action of $\langle P_2,\ CNOT_{1,2}\rangle$, and its associated contracted graph (right). The contracted graph has $4$ vertices and $4$ edges connecting any two vertices, indicating the entropy vector can maximally change $4$ times under any circuit built of $P_2$ and $CNOT_{1,2}$. The $4$ entropy vector possibilities, defined by Eq.\ \eqref{EntropyVectorExample}, are given in the legend.}
        \label{P2C12CayleyWithContractedGraph}
    \end{figure}

We construct the contracted graph of $\mathcal{R}_G\left(\ket{\psi}\right)$ by identifying the elements of $G$ which cannot modify the entropy vector of $\ket{\psi}$. Since the gate $P_2$ acts locally on a single qubit, it can never modify entanglement. Accordingly, we initially contract $\mathcal{R}_G\left(\ket{\psi}\right)$ by gluing together all vertices connected by a $P_2$ edge, represented by the orange dashed lines. Additionally, as we recognized in \cite{Keeler:2023xcx},
\begin{equation}
    \left(C_{1,2}P_2 \right)^4 = P_1^2.
\end{equation}
 Hence all vertices connected by the circuit $\left(C_{1,2}P_2 \right)^4$ must be identified together as well, since $P_1$ likewise does not change a state's entropy vector. The right panel of Figure \ref{P2C12CayleyWithContractedGraph} shows the final contracted graph of $\mathcal{R}_G\left(\ket{\psi}\right)$, which contains $4$ vertices. In this particular example, the contracted graph represents the right coset space of the quotient group $G/\Stab_G(\ket{\psi})$. In general, however, the contracted graph will represent the double coset space $H \backslash G/\Stab_G(\ket{\psi})$, where $G/\Stab_G(\ket{\psi})$ need not be a quotient group.

It is important to note that edges in a contracted graph do not represent any one particular $C_{i,j}$ operation. Instead, every edge bearing a CNOT coloration represents sequences of operations which, at least, include a $C_{i,j}$ gate and are capable of modifying the entropy vector of a sufficiently-general state. In this way, the edges of a contracted graph bound the number of times the entropy vector of a system can change. Since the process of building a contracted graph removes all group elements which leave entanglement entropy unchanged, we are left with a graph structure that represents the orbit of an entropy vector under the group action.

The number of vertices in a contracted graph give an upper bound on the number of distinct entropy vectors which can be generated in particular reachability graph. For example, the contracted graph in Figure \ref{P2C12CayleyWithContractedGraph} contains $4$ vertices, indicating the maximum number of entropy vectors that can be achieved by acting on $\ket{\psi}$ with $\langle P_2,\, CNOT_{1,2}\rangle$.  The number of vertices in a contracted graph is fixed by the overall group structure of $G$, as well as the group structure of $\Stab_G$; however, the different ways in which those vertices can be colored according to entanglement structure is set by the choice of state. The fact that this contracted graph contains a unique entropy vector for each of its $4$ vertices, i.e.\ the contracted graph is maximally-colored, is the reason we chose $\ket{\psi}$ as we did. While the number of vertices in a contracted graph gives an upper bound on entropic diversity in reachability graphs, there can be multiple entropic colorings of the same graph, depending on factors such as qubit number or the specific state. 

We have defined a procedure for building contracted graphs from the reachability graph of arbitrary state $\ket{\psi}$. When considering a group $G$ which acts on a Hilbert space, we build the reachability graph of $\ket{\psi}$ by decomposing $G$ into left cosets $G/\Stab_G(\ket{\psi})$, with elements equivalent up to action by $\Stab_G(\ket{\psi})$. We build the contracted graph of the $\ket{\psi}$ reachability graph by building the double coset space $H \backslash G /\Stab_G(\ket{\psi})$, for a subgroup $H \leq G$ of elements which preserve a state's entropy vector. 

We have demonstrated how contracted graphs illustrate the evolution of entanglement entropy under the action of some quantum gate set. The number of vertices in a contracted graph gives an upper bound on the maximal number of times an entropy vector can change under the chosen set of gates. We have chosen in this paper to construct contracted graphs from reachability graphs in order to analyze the evolution of state entropy vectors; however, the contraction procedure can be applied directly to Cayley graphs as well. 

In the next section we use the techniques defined above to build contracted graphs for all stabilizer state reachability graphs studied in \cite{Keeler2022,Keeler:2023xcx}, establishing upper bounds on the variation of entanglement entropy in stabilizer state systems. We also extend our analysis beyond stabilizer states, deriving upper bounds on the evolution of entanglement entropy for any quantum state under the action of the Clifford group.

\section{Contracted Clifford Reachability Graphs}\label{ContractedGraphsSection}

In this section, we build contracted graphs to illustrate entropy vector evolution in stabilizer and non-stabilizer state reachability graphs. We begin by first considering stabilizer state reachability graphs under the action of the $\mathcal{C}_2$ subgroup $\HC \equiv \langle H_1,\,H_2,\,C_{1,2},\,C_{2,1} \rangle$, as studied in \cite{Keeler2022,Keeler:2023xcx}. We demonstrate how the contracted version of each $\HC$ reachability graph explains the bounds on entanglement variation observed in our earlier work \cite{Keeler2022}. We then extend our analysis to consider the full action of $\mathcal{C}_2$ on stabilizer states, showing how $\mathcal{C}_2$ contracted graphs constrain the evolution of entanglement entropy in stabilizer systems under any $2$-qubit Clifford circuit. 

We extend our study beyond the stabilizer states to the set of $n$-qubit Dicke states, a class of non-stabilizer quantum states possessing non-trivial stabilizer group under Clifford action \cite{Munizzi:2023ihc}. We construct $\HC$ and $\mathcal{C}_2$ reachability and contracted graphs for all Dicke states, establishing constraints on entropy vector evolution for such states. Finally we move toward complete generality, deriving an upper bound for the number of entropy vectors that can be realized by any $n$-qubit Clifford circuit, acting on an arbitrary quantum state.

\subsection{Contracted Graphs of $g_{24}$ and $g_{36}$}

The complete set of $n$-qubit stabilizer states can be generated by acting with $\mathcal{C}_n$ on the state $\ket{0}^{\otimes n}$. However, since we are motivated to better understand the evolution of entropy vectors in stabilizer systems, we restrict analysis to $\mathcal{C}_2$ and its subgroups, since all entanglement modification in Clifford circuits occurs through bi-local operations. Acting with $\mathcal{C}_2$ on $\ket{0}^{\otimes n}$, for $n >1$, generates an orbit of $60$ states.

First, we consider the class of states with stabilizer subgroup%
\footnote{A comprehensive derivation of all stabilizer subgroups, for stabilizer states under the action of $\HC$, is given in Section 5.3 of \cite{Keeler:2023xcx}.} %
isomorphic to $\mathcal{S}_{HC}(\ket{0}^{\otimes n})\equiv\Stab_{\HC}(\ket{0}^{\otimes n})$, under the action of $\HC$. The state $\ket{0}^{\otimes n}$, and any other state with stabilizer group isomorphic to $\mathcal{S}_{HC}(\ket{0}^{\otimes n})$, has an orbit of $24$ states under $\HC$.

\subsubsection{$\left(HC\right)_{1,2}$ Contracted Graphs of $g_{24}$ and $g_{36}$}

The stabilizer subgroup $\mathcal{S}_{HC}(\ket{0}^{\otimes n})$ contains $48$ elements. As a result, generating all left cosets of the $1152$-element group $\HC$ by $\mathcal{S}_{HC}(\ket{0}^{\otimes n})$ builds a coset space of $1152/48 = 24$ equivalence classes. The corresponding reachability graph of $\ket{0}^{\otimes n}$ under $\HC$ contains $24$ vertices, which we appropriately term $g_{24}$. The left panel of Figure \ref{G24WithContractedGraph} shows the graph $g_{24}$, which is shared by all states with stabilizer group isomorphic to $\mathcal{S}_{HC}(\ket{0}^{\otimes n})$.  

To build the associated contracted graph we quotient $g_{24}$ by all elements of $\HC$ which do not modify the entropy vector. One immediate $\HC$ subgroup which cannot modify entanglement entropy is $\langle H_1,\,H_2 \rangle$, which describes all circuits composed of Hadamard gates acting on two qubits. Additionally, as we recognized in \cite{Keeler:2023xcx}, the relation
\begin{equation}\label{HCLocalAction}
    \left(C_{i,j}H_j \right)^4 = P_i^2,
\end{equation}
demonstrates that certain sequences of Hadamard and CNOT gates are actually equivalent to phase operations. We therefore need to also identify all vertices connected by the circuits in Eq.\ \eqref{HCLocalAction}, since phase operations cannot change entanglement. After identifying all vertices connected by entropy-preserving edges, the reachability graph $g_{24}$ contracts to a graph with $2$ vertices, shown on the right of Figure \ref{G24WithContractedGraph}.
    \begin{figure}[h]
        \centering
        \includegraphics[width=15cm]{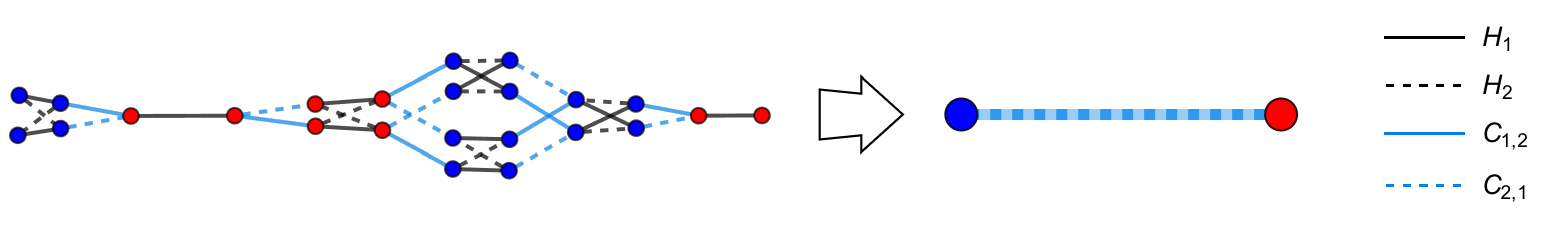}
        \caption{Reachability graph $g_{24}$ (left) and its contracted graph (right). Any state with stabilizer group isomorphic to $\mathcal{S}_{HC}(\ket{0}^{\otimes n})$ will have reachability graph $g_{24}$ under $\HC$. The $g_{24}$ contracted graph has $2$ vertices, indicating the maximum number of unique entropy vectors that can exist in any $g_{24}$ graph. Each edge in the contracted graph represents a set of entanglement-modifying circuits, each containing at least one CNOT gate.}
        \label{G24WithContractedGraph}
    \end{figure}

The contracted graph of $g_{24}$ contains $2$ vertices, and is shown in the right panel of Figure \ref{G24WithContractedGraph}. These $2$ vertices represent the $2$ possible entropy vectors that can be reached by all circuits in any $g_{24}$ graph, regardless of qubit number. All states represented by blue vertices in $g_{24}$ are connected by some circuit composed of $H_1,\,H_2,\,P_1^2,$ and $P_2^2$, and are therefore identified to a single blue vertex in the contracted graph. Likewise, all red vertices in $g_{24}$ are identified to a single red vertex in the contracted graph. For the specific case of $\ket{0}^{\otimes n}$, the two entropy vectors in $g_{24}$ correspond to completely unentangled states, or states which share an EPR pair among two qubits. 

As a group-theoretic object, the vertices of a contracted graph represent the equivalence classes of a double coset space, as defined in Eq.\ \eqref{DoubleCosetDefinition}. For the group $\HC$ acting on $\Hil$, the subgroup
\begin{equation}
    (HP^2)_{1,2} \equiv \langle H_1,\,H_2,\,P_1^2,\,P_2^2 \rangle
\end{equation}
can never modify the entropy vector of any state. Accordingly, the $2$ vertices of the contracted graph in Figure \ref{G24WithContractedGraph} indicate the $2$ distinct equivalence classes in the double coset space $(HP^2)_{1,2} \backslash \HC / \mathcal{S}_{HC}(\ket{0}^{\otimes n})$.

Acting with the gates $H_1$ followed by $P_1$ on the state $\ket{0}^{\otimes n}$, that is
\begin{equation}
    \ket{\phi} = P_1H_1\ket{0}^{\otimes n},
\end{equation}
yields a state $\ket{\phi}$ with stabilizer group $\mathcal{S}_{HC}(\ket{\phi})$, consisting of $32$ elements, which is not isomorphic to $\mathcal{S}_{HC}(\ket{0}^{\otimes n})$. Consequently the state $\ket{\phi}$, as well as any other state with stabilizer group isomorphic to $\mathcal{S}_{HC}(\ket{\phi})$, is not found on any $g_{24}$ graph. Instead, each state stabilized by $\mathcal{S}_{HC}(\ket{\phi})$ resides on a reachability graph of $36$ vertices, which we term $g_{36}$, shown on the left of Figure \ref{G36WithContractedGraph}. In general, any state which is the product of a $2$-qubit stabilizer state and a generic $(n-2)$-qubit state will either have reachability graph $g_{24}$ or $g_{36}$.
    \begin{figure}[h]
        \centering
        \includegraphics[width=15cm]{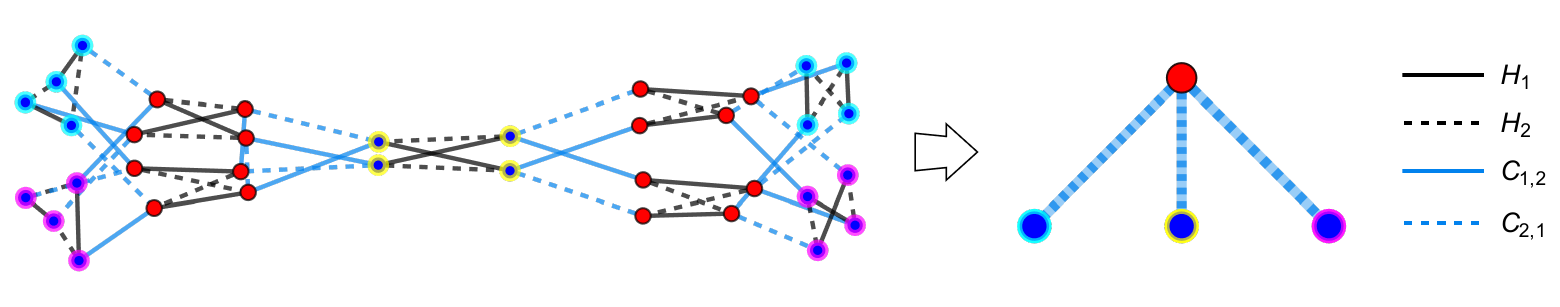}
        \caption{Reachability graph $g_{36}$ (left) and its contracted graph (right). The $g_{36}$ contracted graph contains $4$ vertices, but only ever realizes $2$ entropy vectors among those vertices. Different sets of blue vertices, highlighted in cyan, yellow, and magenta, identify respectively to the three blue vertices in the contracted graph. All red vertices in $g_{36}$ identify to a single red vertex in the contracted graph. Non-trivial entropy-preserving circuits, e.g.\ $(C_{i,j}H_j)^4$ from Eq.\ \eqref{HCLocalAction}, map vertices on opposite sides of $g_{36}$ to each other.}
        \label{G36WithContractedGraph}
    \end{figure}

The contracted graph of $g_{36}$, shown in the right panel of Figure \ref{G36WithContractedGraph}, contains $4$ vertices. All red vertices in $g_{36}$ identify to the same red vertex in the contracted graph. There are three distinct sets of blue vertices in $g_{36}$, highlighted with colors cyan, yellow, and magenta in Figure \ref{G36WithContractedGraph}, which identify to the three blue vertices in the contracted graph. All vertices highlighted by the same color in $g_{36}$ are connected by circuits which preserve the entropy vector.

The vertices of the $g_{36}$ contracted graph in Figure \ref{G36WithContractedGraph} represent the $4$ unique equivalence classes of the double coset space $(HP^2)_{1,2} \backslash \HC / \mathcal{S}_{HC}(\ket{\phi})$. Examining the vertex identifications in Figure  \ref{G36WithContractedGraph}, we again observe that the contraction map is not a quotient map on the original group. Vertex sets of different cardinalities in $g_{36}$ are identified together under this graph contraction, which cannot occur in a formal group quotient.

While the $g_{36}$ contracted graph contains four vertices, these vertices only ever realize two different entropy vector possibilities. Specifically, the two entropy vectors found on any $g_{36}$ graph are exactly the same as those found on the $g_{24}$ graph in Figure \ref{G24WithContractedGraph}. As we will show below, graph $g_{24}$ attaches to $g_{36}$ when we add phase gates back to our generating set. This connection of the $g_{24}$ and $g_{36}$ reachability graphs by local operations constrains the number of distinct entropy vectors that can be found on either graph.

\subsubsection{$\mathcal{C}_2$ Contracted Graphs of $g_{24}$ and $g_{36}$}

We now analyze the full action of $\mathcal{C}_2$ on states in a $g_{24}$ or $g_{36}$ reachability graph under $\HC$. Acting with $\mathcal{C}_2$ on any such state generates a reachability graph of $60$ vertices, which can be seen in Figure \ref{PhaseConnectedG24:G36}. This $60$-vertex reachability graph consists of a single copy of $g_{24}$ (top), attached to a single copy of $g_{36}$ (bottom) by sets of $P_1$ and $P_2$ edges.
    \begin{figure}[h]
        \centering
        \includegraphics[width=13cm]{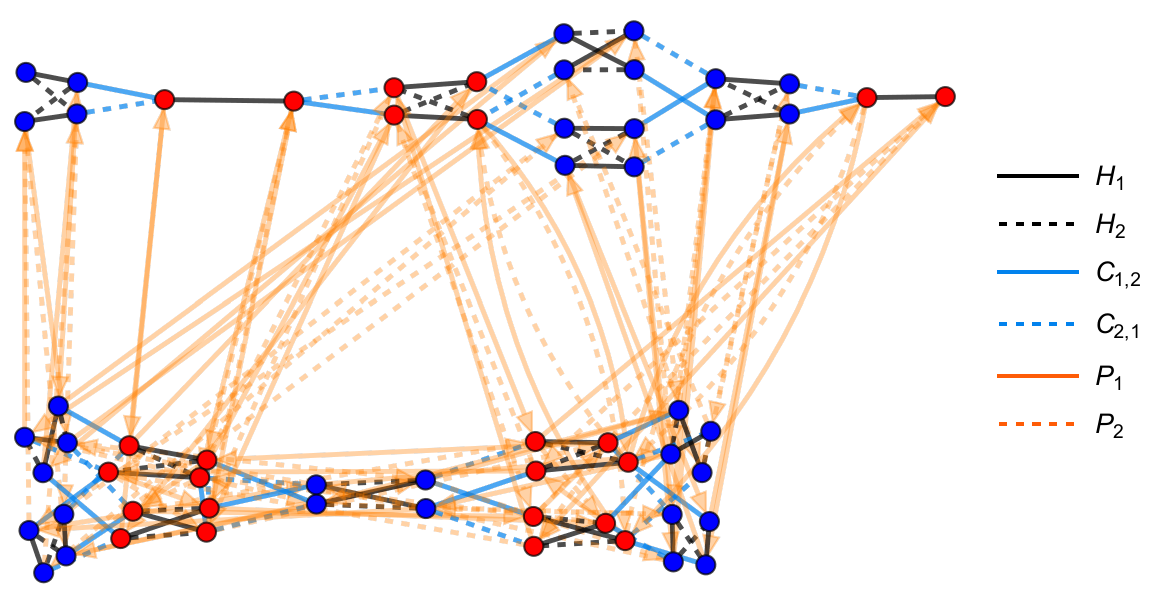}
        \caption{Reachability graph for all states with $\mathcal{S}_{HC}(\ket{0}^{\otimes n})$ under the action of $\mathcal{C}_2$. This $60$-vertex reachability graph is the attachment of $g_{24}$ (Figure \ref{G24WithContractedGraph}) to $g_{36}$ (Figure \ref{G36WithContractedGraph}) by $P_1$ and $P_2$ gates. This reachability graph is likewise shared by all stabilizer product states.}
        \label{PhaseConnectedG24:G36}
    \end{figure}

Following the $P_1$ and $P_2$ edges in Figure \ref{PhaseConnectedG24:G36}, we can observe how vertices of a certain color connect to other vertices of the same color. Blue vertices in $g_{24}$ always connect to blue vertices in $g_{36}$, as is true for red vertices. Red vertices in $g_{36}$ may connect to other red vertices in $g_{36}$, or to red vertices in $g_{24}$. The three distinct batches of blue vertices in $g_{36}$, highlighted in Figure \ref{G36WithContractedGraph}, connect to each other via sequences of $H_1,\,H_2,\,P_1,$ and $P_2$, all of which leave the entropy vector unchanged. We can also directly observe circuits such as $\left(C_{1,2}H_2 \right)^4$, as in Eq.\ \eqref{HCLocalAction}, and verify that this sequence is indeed equivalent to the entropy-preserving $P_1^2$ operation.

As before, we contract the $\mathcal{C}_2$ reachability graph in Figure \ref{PhaseConnectedG24:G36} by identifying vertices connected by entropy-preserving circuits. When performing this contraction on the full $\mathcal{C}_2$ graph we do not rely on any special operator relations, e.g.\ Eq.\ \eqref{HCLocalAction}, since we are identifying vertices connected by all $2$-qubit local operations, i.e.\ all operations built of $H_1,\,H_2,\,P_1,$ and $P_2$. The contracted graph of the $\mathcal{C}_2$ reachability graph in Figure \ref{PhaseConnectedG24:G36} is shown in the right panel of Figure \ref{PhaseConnectedG24:g36ContractedGraph}. The $2$ vertices in this contracted graph represent the $2$ equivalence classes in $(HP^2)_{1,2} \backslash \mathcal{C}_2 / \mathcal{S}_{\mathcal{C}_2}(\ket{0}^{\otimes n})$.
    \begin{figure}[h]
        \centering
        \includegraphics[width=13cm]{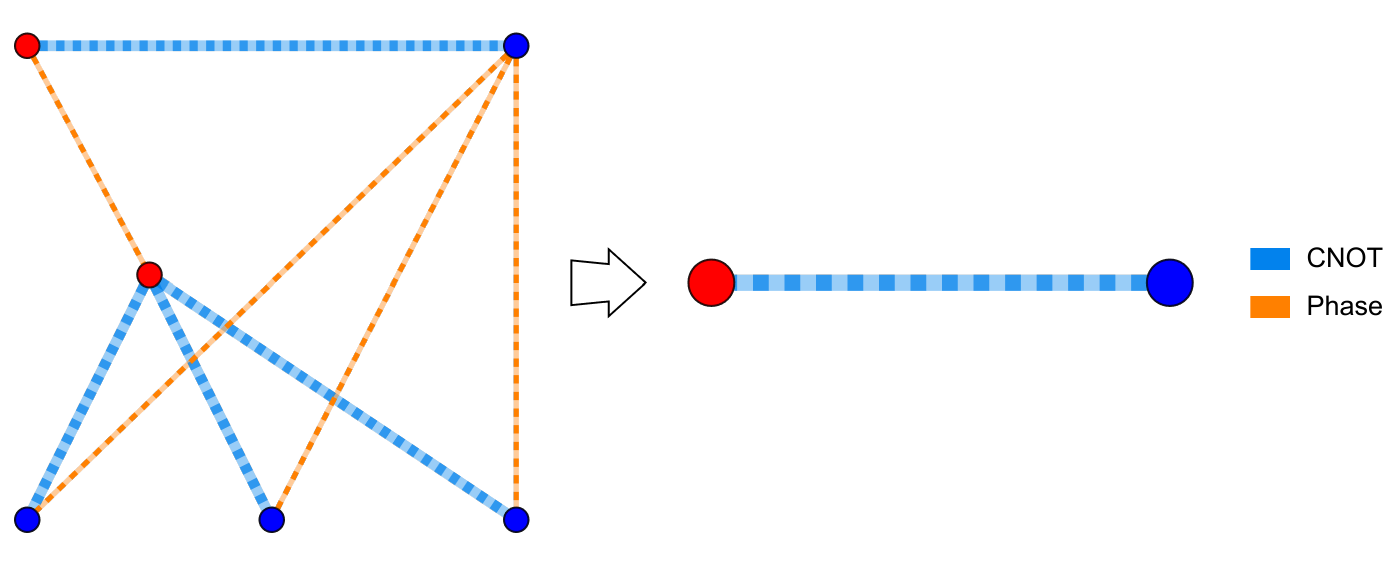}
        \caption{Contracted graph (right) of $\mathcal{C}_2$ reachability graph in Figure \ref{PhaseConnectedG24:G36}. The left panel shows the contracted graphs of $g_{24}$ (top) and $g_{36}$ (bottom), connected by $P_1$ and $P_2$ circuits. Identifying vertices connected by phase edges quotients the left graph to the $2$-vertex contracted graph on the right. The $2$ vertices of this contracted graph represent the $2$ unique entropy vectors that can be found in the reachability graph in Figure \ref{PhaseConnectedG24:G36}.}
        \label{PhaseConnectedG24:g36ContractedGraph}
    \end{figure}

Figure \ref{PhaseConnectedG24:G36} depicts how sets of phase gates connect reachability graphs $g_{24}$ and $g_{36}$. Similarly, the left panel of Figure \ref{PhaseConnectedG24:g36ContractedGraph} shows how the respective contracted graphs of $g_{24}$ and $g_{36}$ are connected by sets of phase edges. The right panel of Figure \ref{PhaseConnectedG24:g36ContractedGraph} gives the final contracted graph after quotienting the $\mathcal{C}_2$ reachability graph in Figure \ref{PhaseConnectedG24:G36} by all entropy-preserving edges. The contracted graph has $2$ vertices, corresponding to the $2$ possible entropy vectors that can be found on any $\mathcal{C}_2$ reachability graph of the form shown in Figure \ref{PhaseConnectedG24:G36}. Furthermore, the $2$ vertices in the contracted explain why both graphs $g_{24}$ and $g_{36}$ individually only ever realize $2$ entropy vector colors among their vertices.

We examined the action of $\HC$ and $\mathcal{C}_2$ on $n$-qubit states with stabilizer group isomorphic to $\mathcal{S}_{HC}(\ket{0}^{\otimes n})$ and $\mathcal{S}_{HC}(\ket{0}^{P_1H_1\otimes n})$. We generated the reachability graphs for all states with both stabilizer groups, and quotiented each reachability graph by entropy-preserving operations to build the associated contracted graphs. The number of vertices in each contracted graph gave an upper bound on the number of different entropy vectors found in each reachability graph. Similarly, the edges in each contracted graph indicated the ways an entropy vector can change under all circuits comprising the reachability graph. We will now consider the reachability graphs of $n > 2$ qubit stabilizer states, where more-complicated entanglement structures can arise.

\subsection{Contracted Graphs of $g_{144}$ and $g_{288}$}

When we consider the action of $\HC$ and $\mathcal{C}_2$ on systems of $n > 2$ qubits, new reachability graph structures appear \cite{Keeler2022}. Additionally at $n > 2$ qubits, we observe new entanglement possibilities as well as new entropy vector colorings for reachability graphs. In this subsection, we define two new sets of stabilizer states which arise at $n=3$ qubits, defined by their stabilizer subgroup under $\HC$ action. We build all reachability graphs and contracted graphs for these two families of states, and determine the bounds on entropy vector evolution in their respective reachability graphs. We then consider the full action of $\mathcal{C}_2$ on these classes of states, and again build all reachability and contracted graphs.

At three qubits, acting with $\HC$ on certain stabilizer states produces an additional two reachability graphs beyond $g_{24}$ and $g_{36}$ discussed in the previous subsection. One new graph which arises at three qubits contains $144$ vertices, shown on the left of Figure \ref{G144WithContractedGraph}, and corresponds to states which are stabilized by $8$ elements in $\HC$. One example of a state with $g_{144}$ reachability graph is the $3$-qubit GHZ state $\ket{GHZ}_3 \equiv \ket{000} + \ket{111}$. The graph $g_{144}$ is shared by all states with a stabilizer subgroup isomorphic to $\mathcal{S}_{HC}(\ket{GHZ}_3)$. For reasons we will explain in a moment, Figure \ref{G144WithContractedGraph} depicts the specific reachability graph for the $6$-qubit state defined in Eq.\ \eqref{SixQubit144State}.
    \begin{figure}[h]
        \centering
        \includegraphics[width=14.5cm]{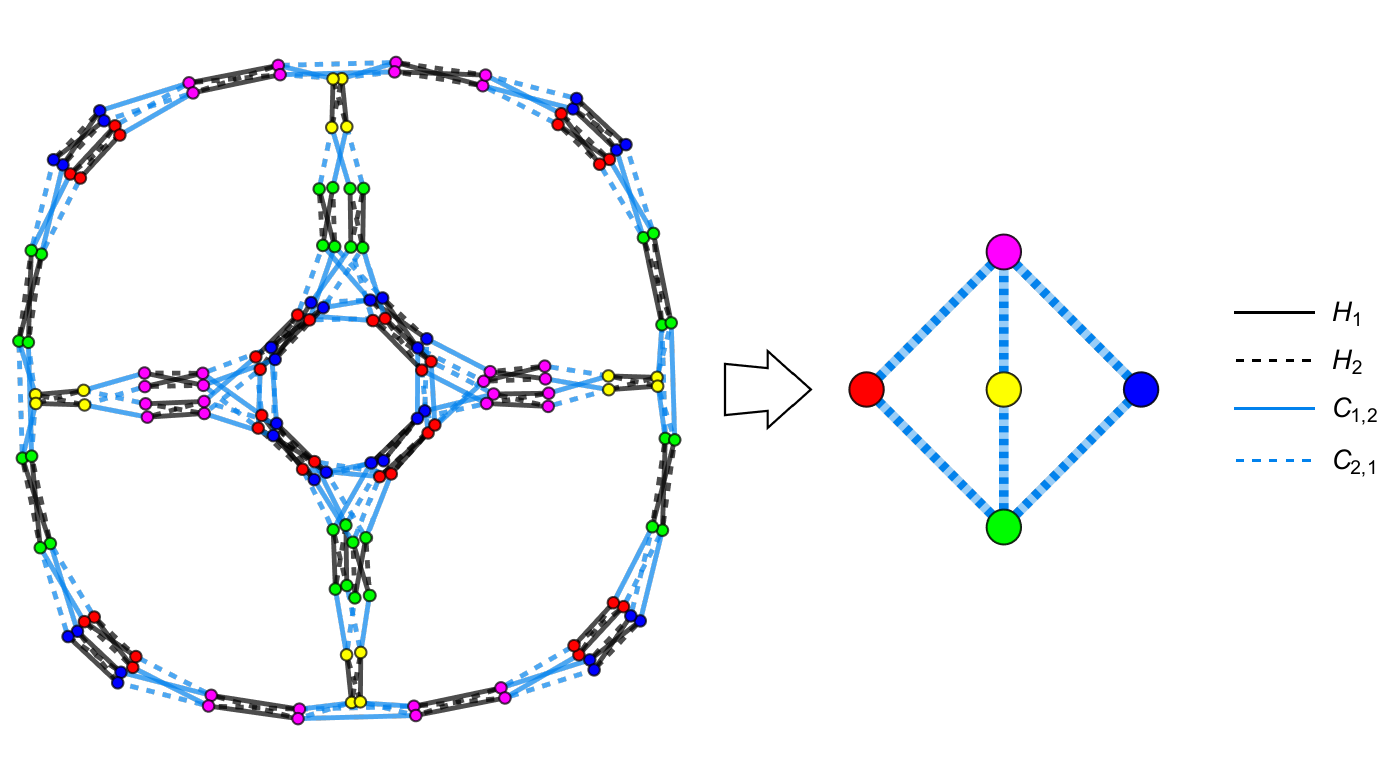}
        \caption{Reachability graph $g_{144}$ (left), and its associated contracted graph (right). The contracted graph contains $5$ vertices, corresponding to the $5$ unique entropy vectors that can be found on $g_{144}$. We depict a $g_{144}$ graph for the $6$-qubit state defined in Eq.\ \eqref{SixQubit144State}, as it contains the maximal number of $5$ entropy vectors among its vertices. Again we observe that certain circuits, e.g.\ Eq.\ \eqref{HCLocalAction}, do not modify entanglement and map vertices of the same color together. The specific entropy vectors shown are given in Table \ref{tab:g144g288EntropyVectorTable}.}
        \label{G144WithContractedGraph}
    \end{figure}

The contracted graph of $g_{144}$, shown on the right of Figure \ref{G144WithContractedGraph}, contains $5$ vertices. These $5$ vertices represent the $5$ unique entropy vectors that can be found on any $g_{144}$ reachability graph. While the graph $g_{144}$ is first observed among $3$-qubit systems, we do not find a maximal coloring of $g_{144}$, i.e.\ a copy of $g_{144}$ with $5$ different entropy vectors, until $6$ qubits. The specific graph shown in Figure \ref{G144WithContractedGraph} corresponds to the orbit of the $6$-qubit state defined in Eq.\ \eqref{SixQubit144State}, which we choose precisely because its $g_{144}$ graph displays the maximum allowable entropic diversity. The specific entropy vectors corresponding to the colors seen in Figure \ref{G144WithContractedGraph} can be found in Table \ref{tab:g144g288EntropyVectorTable} of Appendix \ref{EntropyVectorTables}.

Also beginning at three qubits, we witness a stabilizer state reachability graph with $288$ vertices, which we denote $g_{288}$. States with reachability graph $g_{288}$ are stabilized by $4$ elements of $\HC$, specifically by a subgroup isomorphic to
\begin{equation}\label{g288StabGroup}
    \{\mathbb{1},\, H_2(C_{1,2}H_1)^4,\, (C_{1,2}H_1)^4H_2,\, \left((C_{1,2}H_1)^3C_{1,2}H_2\right)^2\}.
\end{equation}
The left panel of Figure \ref{G288WithContractedGraph} depicts a $g_{288}$ reachability graph, specifically for a $6$-qubit state stabilized by the group in Eq.\ \eqref{g288StabGroup}.
    \begin{figure}[h]
        \centering
        \includegraphics[width=15.1cm]{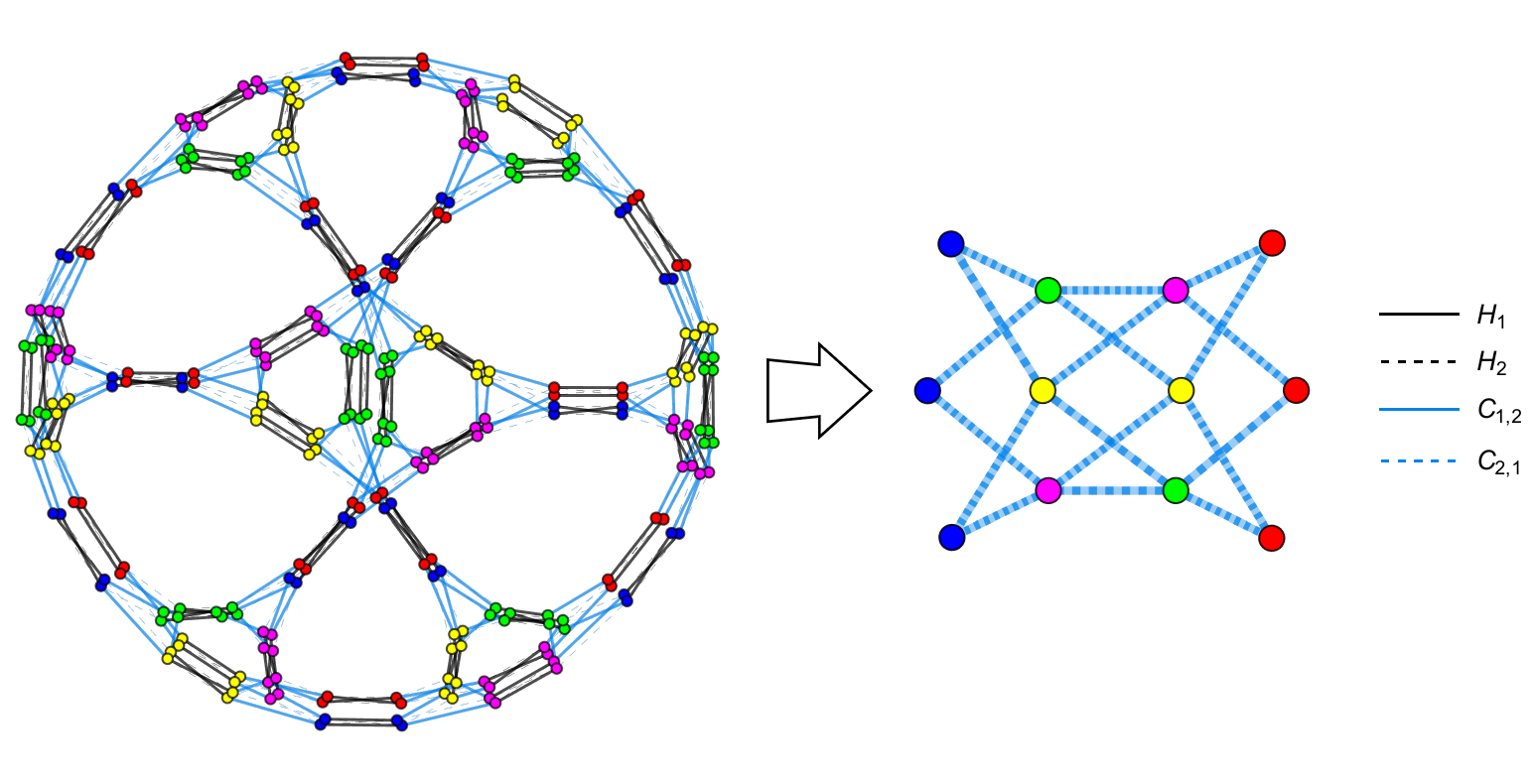}
        \caption{Reachability graph $g_{288}$, and its contracted graph, for $6$-qubit state stabilized by Eq.\ \eqref{g288StabGroup}. While the $g_{288}$ contracted graph has $12$ vertices, we only ever witness $5$ entropy vectors among those vertices. The specific entropy vectors depicted are the same as those in Figure \ref{G144WithContractedGraph}, and can be found in Table \ref{tab:g144g288EntropyVectorTable}.}
        \label{G288WithContractedGraph}
    \end{figure}

The $g_{288}$ contracted graph shown in the right panel of Figure \ref{G288WithContractedGraph} contains $12$ vertices, which provides a weak upper bound on the number of entropy vectors that can be found on any $g_{288}$ graph. However, for reasons we will soon explain, the $12$ vertices of this contracted graph are only ever colored by $5$ different entropy vectors. The specific $5$ entropy vectors shown in Figure \ref{G288WithContractedGraph} are exactly those seen in Figure \ref{G144WithContractedGraph}, and are defined in Table \ref{tab:g144g288EntropyVectorTable}. Similar to the case of $g_{144}$ in Figure \ref{G144WithContractedGraph}, the graph $g_{288}$ is first observed among $3$-qubit systems, but only witnesses a maximal coloring beginning at $n \geq 6$ qubits. 

We now consider the full action of $\mathcal{C}_2$ on states with a $g_{144}$ or $g_{288}$ reachability graph, returning $P_1$ and $P_2$ to our generating set. Every state in a $g_{144}$ and $g_{288}$ reachability graph under $\HC$ is stabilized by $15$ elements of the full group $\mathcal{C}_2$. The orbit of all such states under $\mathcal{C}_2$ therefore contains $768$ states, and the associated $768$-vertex reachability graph is shown in Figure \ref{HhCcPhaseOverlay}. The orange edges in the reachability graph, which correspond to $P_1$ and $P_2$ gates, illustrate specifically how three different copies of $g_{144}$ attach to a single copy of $g_{288}$ under phase operations.
    \begin{figure}[h]
        \centering
        \includegraphics[width=14cm]{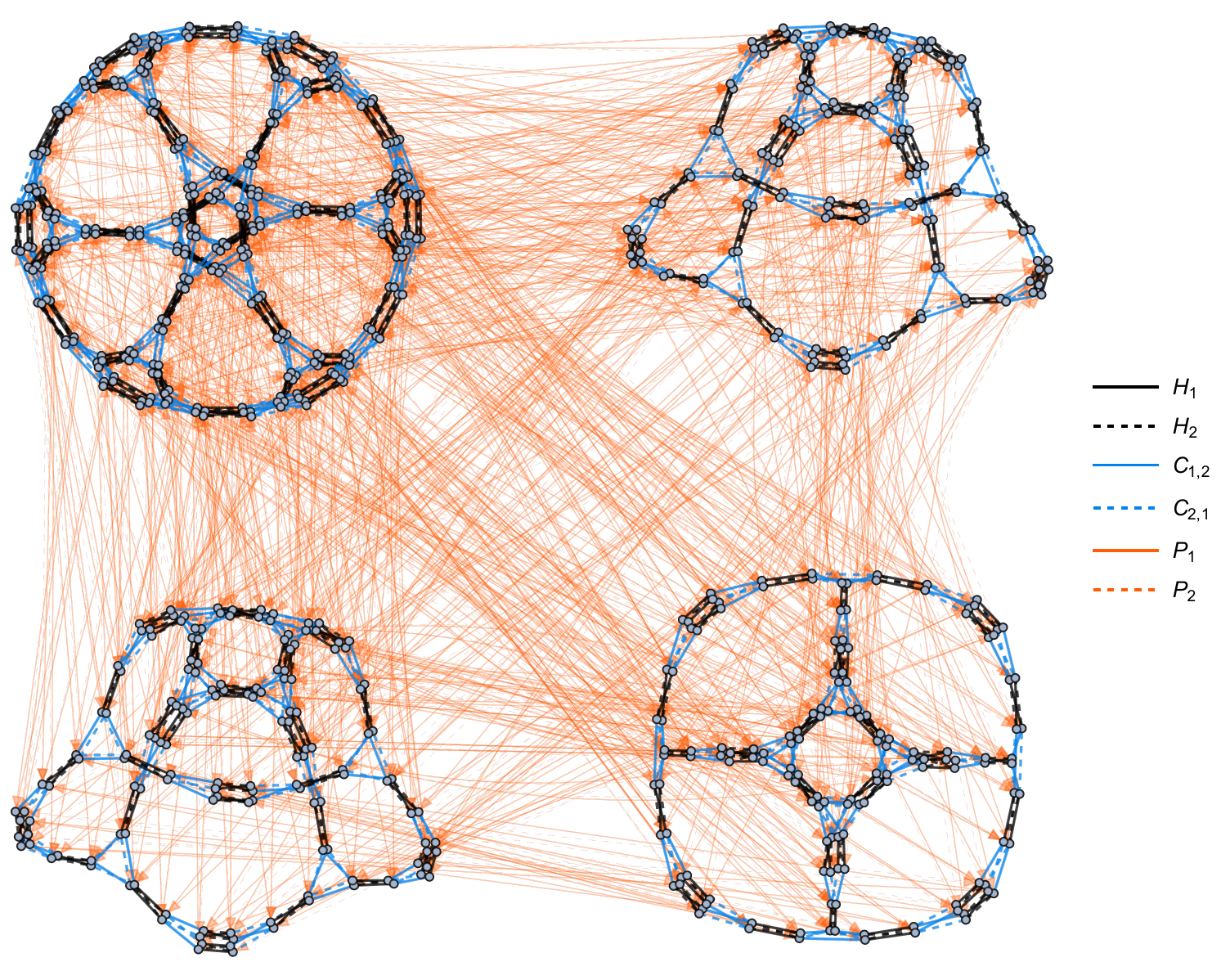}
        \caption{Reachability graph for states in $g_{144}$ and $g_{288}$ graphs, under the full action of $\mathcal{C}_2$. This $768$-vertex graph is composed of $3$ copies of $g_{144}$ and a single $g_{288}$. The graph connectivity constrains the diversity of entropy vectors which can be found on any single $g_{144}$ and $g_{288}$ graph. For clarity we choose not to color vertices by their entropy vector here.}
        \label{HhCcPhaseOverlay}
    \end{figure}

The contracted graphs for each $g_{144}$ and $g_{288}$ in Figure \ref{HhCcPhaseOverlay} are compiled in the left panel of Figure \ref{PhaseConnectedg144_288ContractedGraph}. Each of the three copies of $g_{144}$ contracts to a $5$-vertex graph that is isomorphic to Figure \ref{G144WithContractedGraph}, while the single copy of $g_{288}$ contracts to the $12$-vertex graph seen in Figure \ref{G288WithContractedGraph}. These four contracted graphs attach to each other under phase operations, adding connections which do not change a state's entropy vector. The final contracted graph of Figure \ref{HhCcPhaseOverlay} is shown on the right of Figure \ref{PhaseConnectedg144_288ContractedGraph}, and only has $5$ vertices.

The full $\mathcal{C}_2$ contracted graph in Figure \ref{PhaseConnectedg144_288ContractedGraph} is almost identical to the $g_{144}$ contracted graph in Figure \ref{G144WithContractedGraph}, but with an additional edge connecting two of the vertices. Since every $g_{288}$ attaches to $3$ copies of $g_{144}$ by phase gates, which do not modify entanglement, the maximum number of entropy vectors on any $g_{288}$ is bounded by the entropic coloring of each $g_{144}$ it connects to. This connectivity explains why we only observe at most $5$ entropy vectors on any $g_{288}$ graph, as can be seen in Figure \ref{G288WithContractedGraph}.
    \begin{figure}[h]
        \centering
        \includegraphics[width=14cm]{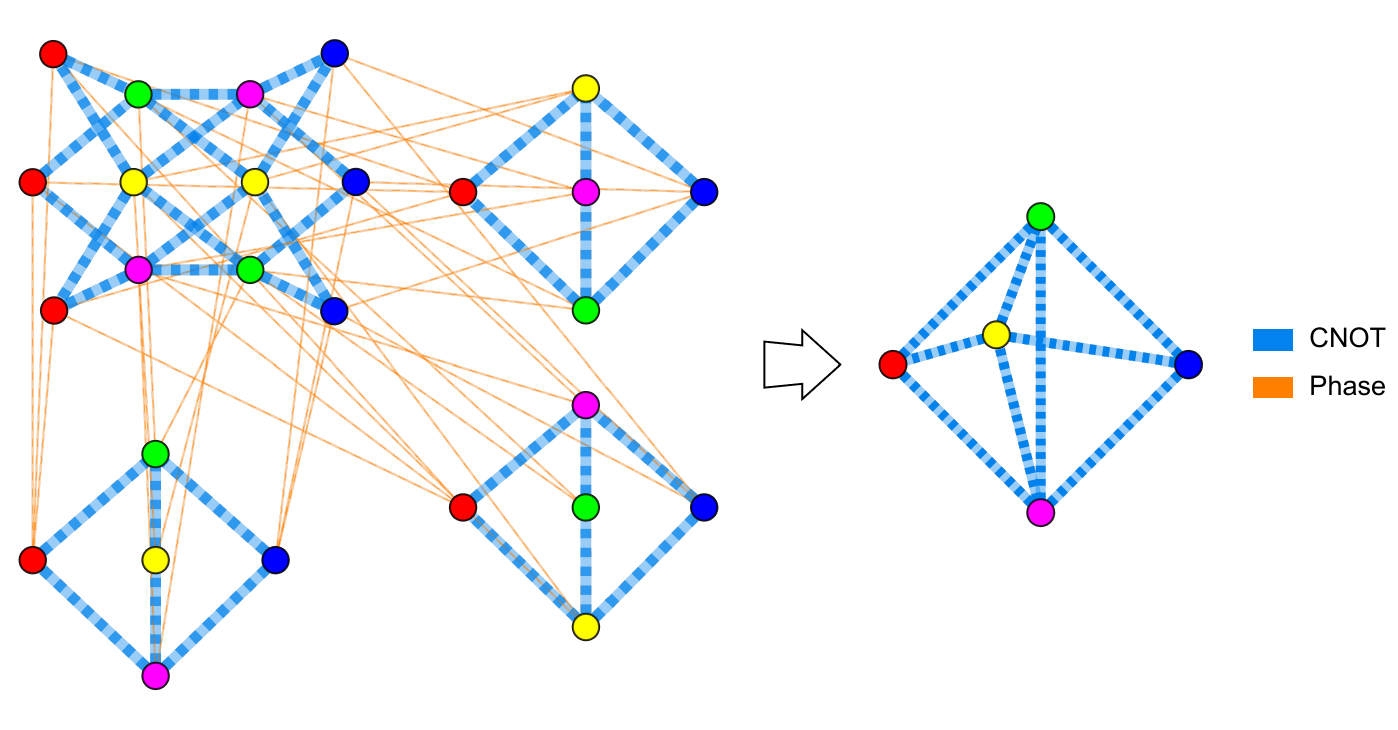}
        \caption{Contracted graph of $\mathcal{C}_2$ reachability graph from Figure \ref{HhCcPhaseOverlay}. The left panel depicts the individual contracted graphs of the $3$ $g_{144}$ graphs attached to a single $g_{288}$ graph. The right panel shows the final contracted graph, with $5$ vertices, and explains why we only ever find $g_{288}$ and $g_{144}$ graphs with $5$ different entropy vectors (given in Table \ref{tab:g144g288EntropyVectorTable}).}
        \label{PhaseConnectedg144_288ContractedGraph}
    \end{figure}

Figure \ref{PhaseConnectedg144_288ContractedGraph} depicts a symmetry between red and blue vertices which corresponds to an equivalence of these two entropy vectors under an exchange of the first two qubits. We likewise observe a symmetry between green, yellow, and magenta vertices, reflecting the three ways to divide the $4$-qubit subsystem $CDEO$ into two groups of two qubits each. For each $g_{144}$ contracted graph in Figure \ref{PhaseConnectedg144_288ContractedGraph}, the middle vertex corresponds to the entropy vector that occurs the fewest number of times, specifically $16$ times, in each respective $g_{144}$ reachability graph. We again observe that the contraction procedure generates a double coset space, rather than a group quotient, since the resulting equivalence classes have different cardinalities.

In this subsection we built contracted graphs for the stabilizer state reachability graphs $g_{144}$ and $g_{288}$, corresponding to states which are stabilized by $4$ and $8$ elements of $\HC$ respectively. We showed how the contracted graph for $g_{144}$, with $5$ vertices, and the contracted graph for $g_{288}$, with $12$ vertices, both witness a maximum of $5$ different entropy vectors. This constraint on the number of different entropy vectors, perhaps surprising in the case of $g_{288}$, can be understood by considering the full action of $\mathcal{C}_2$, which attaches three copies of $g_{144}$ to $g_{288}$ by phase operations. The number of entropy vectors found on any $g_{288}$ reachability graph is bounded by the number of entropy vectors found on each of the $g_{144}$ graphs to which it attaches, since $P_1$ and $P_2$ cannot modify entanglement. In the next subsection we consider the action of $\HC$ and $\mathcal{C}_2$ on generic quantum states, which allows us to extend our analysis beyond the stabilizer states. 

\subsection{Contracted Graphs of $g_{1152}$ and Full $\mathcal{C}_2$}

We now study the generic $\HC$ reachability graph for any quantum state stabilized by only the identity in $\HC$. For stabilizer state systems, this final $\HC$ reachability graph structure arises at $n \geq 4$ qubits. The reachability graph, which we term $g_{1152}$, contains $1152$ vertices and is shown on the left of Figure \ref{G1152WithContractedGraph}.
    \begin{figure}[h]
        \centering
        \includegraphics[width=15cm]{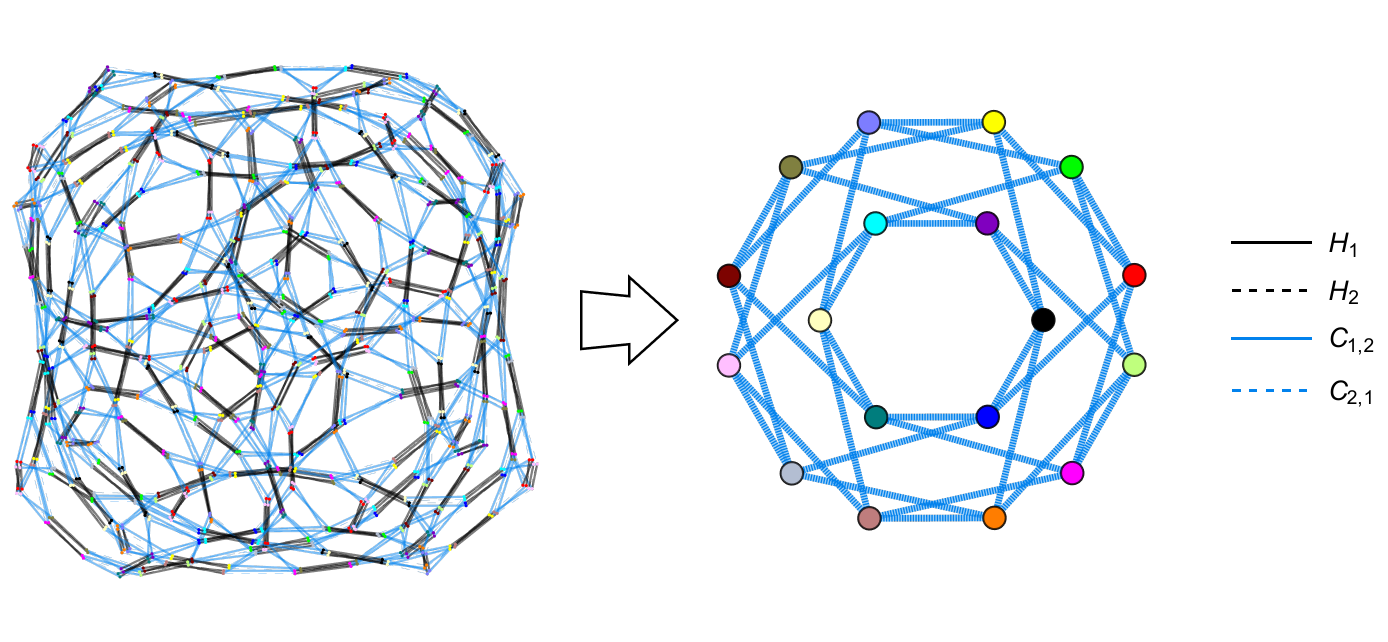}
        \caption{Reachability graph $g_{1152}$ (left) and its contracted graph (right). The graph $g_{1152}$ is shared by all stabilizer states stabilized by only $\mathbb{1} \in \HC$, as well as generic quantum states. In this Figure, we illustrate an example $g_{1152}$ for the $8$-qubit state in Eq. \eqref{EightQubitState}, where the contracted graph achieves a maximal coloring of $18$ different entropy vectors (given in Figure \ref{EightQubitEntropyVectors}).}
        \label{G1152WithContractedGraph}
    \end{figure}

The contracted graph of $g_{1152}$, shown in the right panel of Figure \ref{G1152WithContractedGraph}, contains $18$ vertices. These $18$ vertices indicate the maximum number of unique entropy vectors that can be generated for any quantum state using only operations in $\HC$. The $g_{1152}$ contracted graph is symmetric, and achieves a maximal coloring at $8$ qubits. The specific instance of $g_{1152}$ in Figure \ref{G1152WithContractedGraph} corresponds to the $8$-qubit state given in Eq. \eqref{EightQubitState}, for which the entropy vectors are given in Table \ref{EightQubitEntropyVectors}.

The full $2$-qubit Clifford group $\mathcal{C}_2$ is composed of $11520$ elements. A generic quantum state will only be stabilized by $\mathbb{1} \in \mathcal{C}_2$, and therefore has an orbit of $11520$ states under $\mathcal{C}_2$ action. Every state in an $11520$-vertex reachability graph under $\mathcal{C}_2$ will trivially lie in a $g_{1152}$ graph under $\HC$, however, the converse%
\footnote{Since any state in an $11520$-vertex graph under $\mathcal{C}_2$ is stabilized by only the identity, each state will likewise be stabilized by only the identity in $\HC$. There exist states, however, which are stabilized by only the identity in $\HC$, but can be transformed under phase gates into states stabilized by more than one element of $\HC$. Subsection \ref{DickeStateSubsection} discusses two classes of Dicke states which demonstrate this counterexample. \label{C2Footnote}} %
is not always true. We display the full $\mathcal{C}_2$ reachability graph, in a compressed format, to the left of Figure \ref{FullC2WithContractedGraph}. Each vertex in the left panel of Figure \ref{FullC2WithContractedGraph} represents a distinct copy of $g_{1152}$ from Figure \ref{G1152WithContractedGraph}. Each of the $10$ copies of $g_{1152}$ attaches to every other $g_{1152}$ via $P_1$ and $P_2$ gates.
    \begin{figure}[h]
        \centering
        \includegraphics[width=15cm]{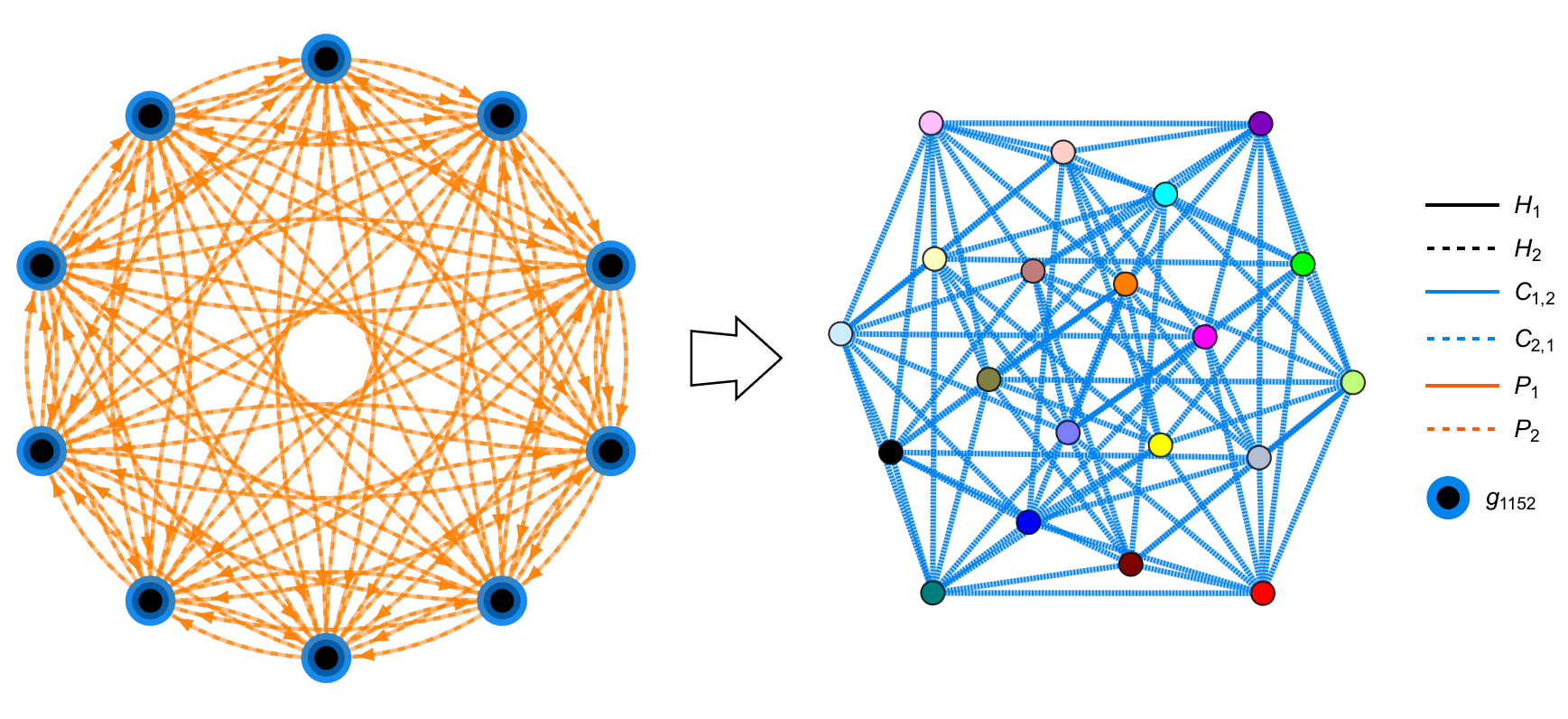}
        \caption{The full $\mathcal{C}_2$ reachability graph (left) with $11520$ vertices. We present this reachability graph as a collection of attached $g_{1152}$ graphs, illustrating how $\HC$ reachability graphs connect via $P_1$ and $P_2$ gates. We also remove all loops in the $\mathcal{C}_2$ reachability graph, i.e.\ all phase edges which map a copy of $g_{1152}$ to itself. The contracted graph of the $\mathcal{C}_2$ reachability graph is given to the right, and has $20$ vertices. These $20$ vertices give an upper bound on the number of distinct entropy vectors that can be reached by applying any sequence of $2$-qubit operations on any quantum state.}
        \label{FullC2WithContractedGraph}
    \end{figure}

The contracted graph%
\footnote{Since all states in the $11520$-vertex reachability graph are stabilized by only $\mathbb{1} \in \mathcal{C}_2$, and since $\langle \mathbb{1} \rangle$ is normal in $\mathcal{C}_2$, the object $\mathcal{C}_2/\langle \mathbb{1} \rangle$ defines a formal group quotient on $\mathcal{C}_2$. Consequently, the contracted graph to the right of Figure \ref{FullC2WithContractedGraph} actually represents the right coset space $\mathcal{C}_2 \backslash \langle H_1,\,H_2,\,P_1,\,P_2 \rangle$, as opposed to a double coset space.} %
of the $11520$-vertex $\mathcal{C}_2$ reachability graph contains $20$ vertices, and is shown on the right of Figure \ref{FullC2WithContractedGraph}. This contracted graph is complete and symmetric, and the $20$ entropy vectors shown in Figure \ref{FullC2WithContractedGraph} are given in Table \ref{EightQubitEntropyVectors}. Since we are considering the full action of $\mathcal{C}_2$, the $20$ vertices in this contracted graph constrain the number of entropy vectors that can be generated by any $2$-qubit Clifford circuit. Otherwise stated, given a generic quantum state with arbitrary entanglement structure, any unitary composed of $2$-qubit Clifford gates can maximally achieve $20$ distinct entropy vectors.

In the remainder of the section we extend our discussion beyond stabilizer states, examining contracted graphs for non-stabilizer Dicke states under $\HC$ and $\mathcal{C}_2$ action. We also derive a general upper bound for the number of entropy vectors that can be achieved under any $n$-qubit Clifford circuit, for arbitrary $n$.

\subsection{Non-Stabilizer State Contracted Graphs}\label{DickeStateSubsection}

\cite{Keeler:2023xcx,Munizzi:2023ihc} showed that certain non-stabilizer states can have non-trivial stabilizer subgroups, i.e.\ they are stabilized by more than just the identity, under the action of $\mathcal{C}_n$. One class of states in particular, the set of $n$-qubit Dicke states \cite{nepomechie2023qudit}, always admits a non-trivial $\mathcal{C}_n$ stabilizer group. In this subsection, we discuss all $\HC$ and $\mathcal{C}_2$ reachability graphs for Dicke states and construct their associated contracted graphs. We use the contracted graphs to bound the number of possible entropy vectors that can be generated in Dicke state systems under Clifford group action \cite{Schnitzer:2022exe,Munizzi:2023ihc}.

Each $n$-qubit Dicke state $\ket{D^n_k}$ is defined as an equal-weight superposition over all $n$-qubit states of a fixed Hamming weight. Using the $n$-qubit states $\{\ket{b}\}$, where $b$ denotes some binary string of length $2^n$, we construct $\ket{D^n_k}$ as the state
\begin{equation}
    \ket{D^n_k} \equiv \binom{n}{k}^{-1/2} \sum_{b \in \{0,1\}^n,\,h(b)=k} \ket{b},
\end{equation}
where $h(b) = k$ denotes the fixed Hamming weight of $b$. Some examples of Dicke states include
\begin{equation}
\begin{split}
        \ket{D^2_1} &= \frac{1}{\sqrt{2}}\left(\ket{01} + \ket{10}\right),\\
        \ket{D^4_2} &= \frac{1}{\sqrt{6}}\left(\ket{1100} + \ket{1010} + \ket{1001} + \ket{0110} + \ket{0101} + \ket{0011}\right).
\end{split}
\end{equation}
Dicke states of the form $\ket{D^n_1}$ are exactly the non-biseparable $n$-qubit $W$-states, while $\ket{D^n_n}$ are the computational basis states $\ket{1}^{\otimes n}$.

For $n \geq 3$ qubits, the state $\ket{D^n_1}$ is not a stabilizer state. Regardless, each $\ket{D^n_k}$ is stabilized by a subset of $\mathcal{C}_n$ that contains more than just the identity. When considering the action of $\mathcal{C}_2$ on $\ket{D^n_k}$, states of the form $\ket{D^n_1}$ and $\ket{D^n_{n-1}}$ share one particular set of stabilizers, while those of the form $\ket{D^n_k}$ with $1 < k < n-1$ share another. We discuss both cases below. 

Dicke states of the form $\ket{D^n_1}$ and $\ket{D^n_{n-1}}$ are not stabilizer states for all $n \geq 3$. However, both $\ket{D^n_1}$ and $\ket{D^n_{n-1}}$ are stabilized by a $4$-element subgroup%
\footnote{There is a more compact representation of this stabilizer group using $CZ$ gates (see also \cite{Latour:2022gsf}), which can be written $\mathcal{S}_{HC}(\ket{D^n_1}) = \{\mathbb{1},\,CZ_{1,2},\,C_{1,2}C_{2,1}C_{1,2},\,CZ_{1,2}C_{1,2}C_{2,1}C_{1,2}\}$.} %
of $\mathcal{C}_2$, specifically
\begin{equation}\label{WStateStabGroup}
\begin{split}
    \mathcal{S}_{HC}(\ket{D^n_1}) &= \{\mathbb{1},\,H_2C_{1,2}H_2,\,C_{1,2}C_{2,1}C_{1,2},\,H_2C_{1,2}H_2C_{1,2}C_{2,1}C_{1,2}\},\\
    &= \mathcal{S}_{HC}(\ket{D^n_{n-1}}).
\end{split}
\end{equation}
Furthermore, we note that the subgroup in Eq.\ \eqref{WStateStabGroup} is contained in $\HC$. Therefore the left coset space $\HC / \mathcal{S}_{HC}(\ket{D^n_1})$ contains $288$ elements.

The reachability graph for all $\ket{D^n_1}$ and $\ket{D^n_{n-1}}$, which we denote $g_{288^*}$, has $288$ vertices, as dictated by the order of $\mathcal{S}_{HC}(\ket{D^n_1})$ in Eq.\ \eqref{WStateStabGroup}. While the graph $g_{288^*}$ has the same number of vertices as the $g_{288}$ graph for stabilizer states, shown in Figure \ref{G288WithContractedGraph}, its topology is distinct from $g_{288}$ and the two graphs are not isomorphic. Graphs with the topology of $g_{288^*}$ are never observed among stabilizer states, and provide an example of non-stabilizer states that are stabilized by more than just the identity in $\mathcal{C}_2$. The left panel of Figure \ref{WStateG288WithContractedGraph} depicts an example of $g_{288^*}$, specifically for the state $\ket{D^3_1}$.
    \begin{figure}[h]
        \centering
        \includegraphics[width=15cm]{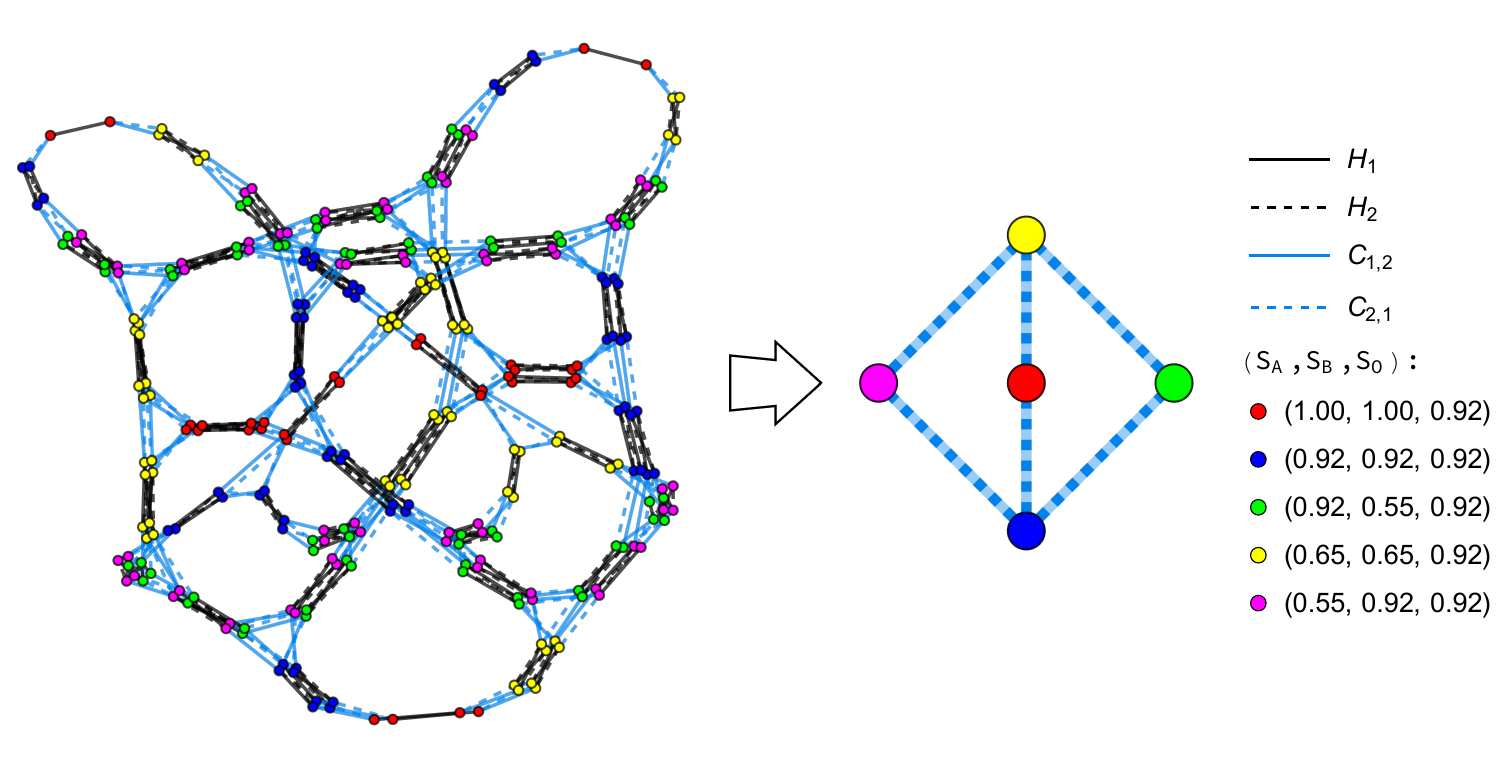}
        \caption{Reachability graph $g_{288^*}$ (left) for $\ket{D^3_1}$ under the action of $\HC$. The graph $g_{288^*}$ has different topology than the $g_{288}$ graph for stabilizer states. The $g_{288^*}$ contracted graph (right) has $5$ vertices, and is isomorphic to the stabilizer state contracted graph of $g_{144}$ from Figure \ref{G144WithContractedGraph}. The exact, rather than numerical, values of the $5$ entropy vectors given in the legend are shown in Table \ref{tab:WStateEntropyVectorTable}.}
        \label{WStateG288WithContractedGraph}
    \end{figure}

The contracted graph of $g_{288^*}$ has $5$ vertices, and is shown on the right of Figure \ref{WStateG288WithContractedGraph}. While the reachability graph $g_{288}$ for stabilizer states has a contracted graph of $12$ vertices, the distinct connectivity of $g_{288^*}$ yields a smaller contracted graph. Interestingly, the $g_{288^*}$ contracted graph is isomorphic to the $g_{144}$ contracted graph seen in Figure \ref{G144WithContractedGraph}. There are $5$ possible entropy vectors found on any $g_{288^*}$, and the graph achieves a maximal coloring beginning at $3$ qubits.

The orbit of $\ket{D^n_1}$ and $\ket{D^n_{n-1}}$ under the full group $\mathcal{C}_2$ reaches $2880$ states, generating a reachability graph of $2880$ vertices. The left panel of Figure \ref{PhaseConnectedWG288ContractedGraph} illustrates this $2880$-vertex reachability graph for the state $\ket{D^3_1}$, which is comprised of several attached copies of $\HC$ reachability graphs. For clarity, we allow each vertex of the $2880$-vertex reachability graph to represent graphs $g_{288^*},\,g_{576}$ (introduced later in Figure \ref{G576WithContractedGraph}), and $g_{1152}$, focusing on the connectivity between different $\HC$ orbits under $P_1$ and $P_2$ operations.
    \begin{figure}[h]
        \centering
        \includegraphics[width=15cm]{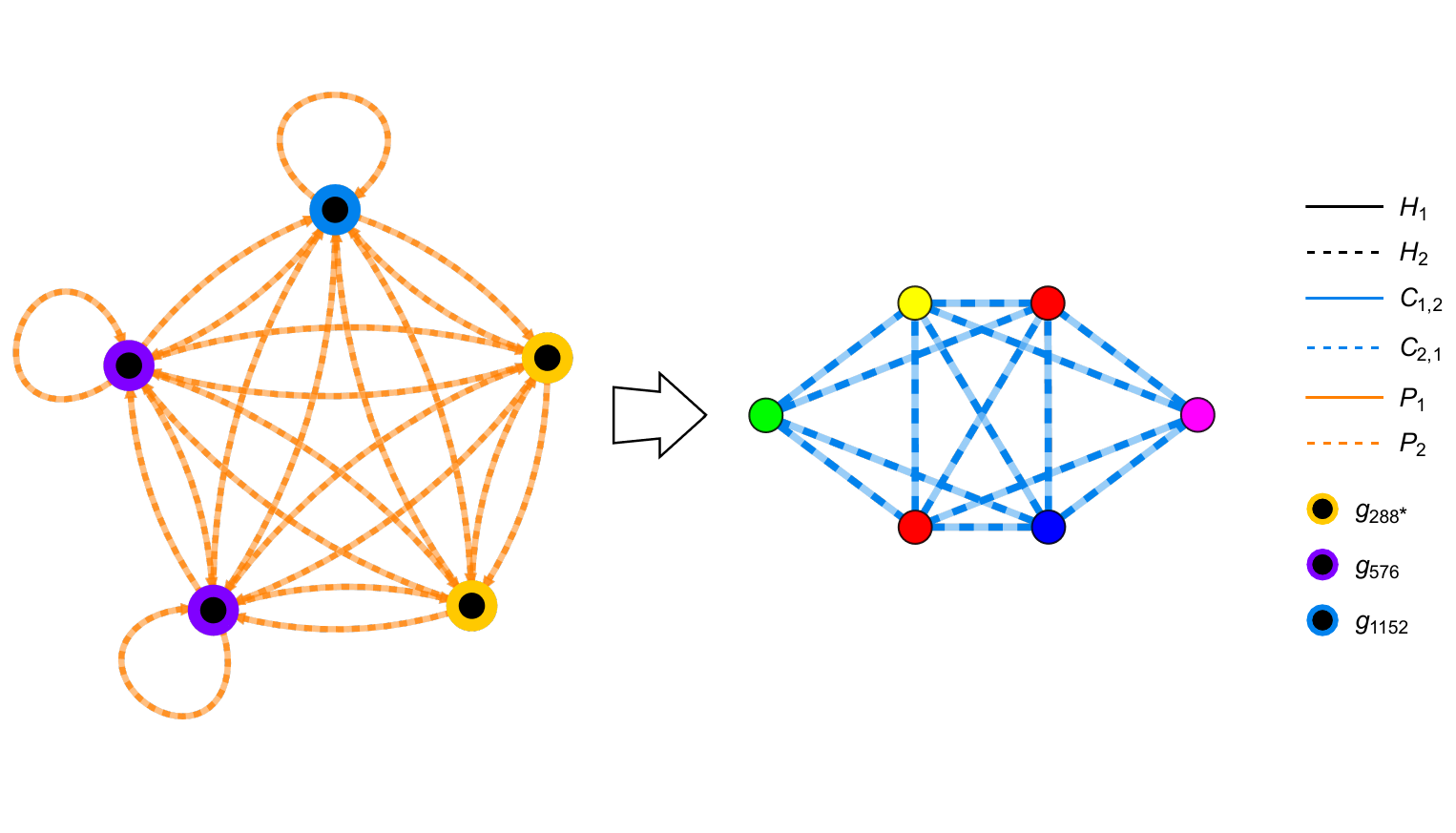}
        \caption{Reachability graph (left) of $\ket{D^3_1}$ under the full action of $\mathcal{C}_2$, containing $2880$ vertices. We illustrate this reachability graph with vertices representing graphs $g_{288^*}, g_{576},$ and $g_{1152}$ to illustrate the connectivity of certain $\HC$ reachability graphs under phase gates. The right panel of the Figure depicts the associated contracted for the $\mathcal{C}_2$ reachability graph, which contains $6$ vertices.}
        \label{PhaseConnectedWG288ContractedGraph}
    \end{figure}

The $\mathcal{C}_2$ reachability graph in Figure \ref{PhaseConnectedWG288ContractedGraph} is built of $2$ attached copies of $g_{288^*}$, $2$ copies of $g_{576}$, and a single $g_{1152}$. Every state in this $2880$-vertex reachability graph is stabilized by $4$ elements of $\mathcal{C}_2$. Certain states, such as $\ket{D^n_1}$ and $\ket{D^n_{n-1}}$, are stabilized by a $4$-element subgroup of $\mathcal{C}_2$ which is also completely contained within $\HC$, as shown in Eq.\ \eqref{WStateStabGroup}. However, other states are stabilized by $4$ elements of $\mathcal{C}_2$, but by only $2$ elements in $\HC$ (see Footnote \ref{C2Footnote}). Accordingly, such states are found in one of the $g_{576}$ graphs in Figure \ref{PhaseConnectedWG288ContractedGraph}. Still other states are stabilized by $4$ elements of $\mathcal{C}_2$, but only by the identity in $\HC$, and reside in the single copy of $g_{1152}$ in Figure \ref{PhaseConnectedWG288ContractedGraph}.

The $\mathcal{C}_2$ reachability graph of $\ket{D^3_1}$ contracts to a $6$-vertex graph, seen to the right of Figure \ref{PhaseConnectedWG288ContractedGraph}, after identifying vertices connected by entropy-preserving circuits. While the contracted graph in Figure \ref{PhaseConnectedWG288ContractedGraph} has $6$ vertices, we only ever observe $5$ different entropy vectors among those vertices. We address this point further in the discussion. The $5$ entropy vectors of the $\ket{D^3_1}$ contracted graph are listed in Table \ref{tab:WStateEntropyVectorTable}.

All remaining Dicke states, those of the form $\ket{D^{n}_k}$ with $1 < k < n-1$, are stabilized by only $2$ elements in $\mathcal{C}_2$. For any $\ket{D^{n}_k}$ of this form, its stabilizer subgroup under $\mathcal{C}_2$ action is given by
\begin{equation}\label{AllOtherDStabilizer}
\mathcal{S}_{\mathcal{C}_2}\left(\ket{D^n_k}\right) = \{\mathbb{1},\, C_{1,2}C_{2,1}C_{1,2}\}, \quad \forall \, 1 < k < n-1.
\end{equation}
We again note that the stabilizer group in Eq.\ \eqref{AllOtherDStabilizer} is also contained completely within $\HC$, and therefore the left coset space $\HC / \mathcal{S}_{\mathcal{C}_2}\left(\ket{D^n_k}\right)$ consists of $576$ elements. 

The reachability graph for $\ket{D^{n}_k}$ under $\HC$, which we denote $g_{576}$, has $576$ vertices. The left panel of Figure \ref{G576WithContractedGraph} depicts $g_{576}$, specifically for the state $\ket{D^4_2}$. Reachability graphs with $576$ vertices, under $\HC$ action, are never observed for stabilizer states. Again, as with $g_{288^*}$, the graph $g_{576}$ corresponds to non-stabilizer states which are non-trivially stabilized by $\mathcal{C}_n$.
    \begin{figure}[h]
        \centering
        \includegraphics[width=15cm]{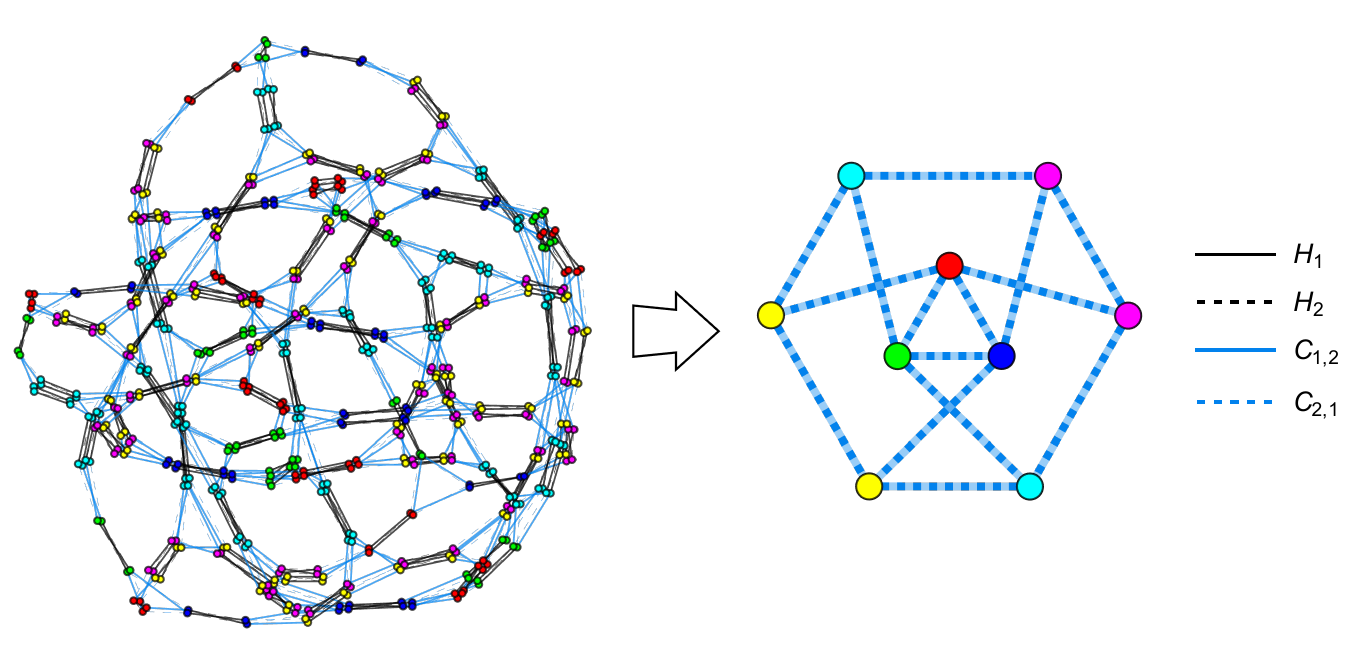}
        \caption{The $g_{576}$ reachability graph (left) for $\ket{D^4_2}$ under $\HC$ action. Graphs of $576$ vertices are never observed among stabilizer states under $\HC$ action. The graph $g_{576}$ contracts to a graph of $9$ vertices under entropy-preserving operations, with $6$ different entropy vectors among those vertices. The $6$ entropy vectors found in this contracted graph are given in Table \ref{tab:D42EntropyVectors}.}
        \label{G576WithContractedGraph}
    \end{figure}

After identifying vertices in $g_{576}$ connected by entropy-preserving operations, we are left with a contracted graph of $9$ vertices shown on the right of Figure \ref{G576WithContractedGraph}. These $9$ vertices are colored by $6$ different entropy vectors, with maximal coloring beginning at $4$ qubits. Among the $6$ entropy vectors in this contracted graph, there are symmetries shared among cyan, magenta, and yellow vectors, and separately among red, blue, and green vectors. The specific $6$ entropy vectors for the $\ket{D^4_2}$ contracted graph are given in Table \ref{tab:D42EntropyVectors}.

Acting with the full group $\mathcal{C}_2$ on $\ket{D^{n}_k}$, for $1 < k < n-1$, generates an orbit of $5760$ states. The 
$\mathcal{C}_2$ reachability graph of $\ket{D^{n}_k}$ therefore has $5760$ vertices, and is depicted in the left panel of Figure \ref{PhaseConnectedG576ContractedGraph} for the case of $\ket{D^4_2}$. As before, we depict the full $5760$-vertex reachability graph as $7$ attached copies of different $\HC$ reachability graphs $g_{576}$ and $g_{1152}$.
    \begin{figure}[h]
        \centering
        \includegraphics[width=15cm]{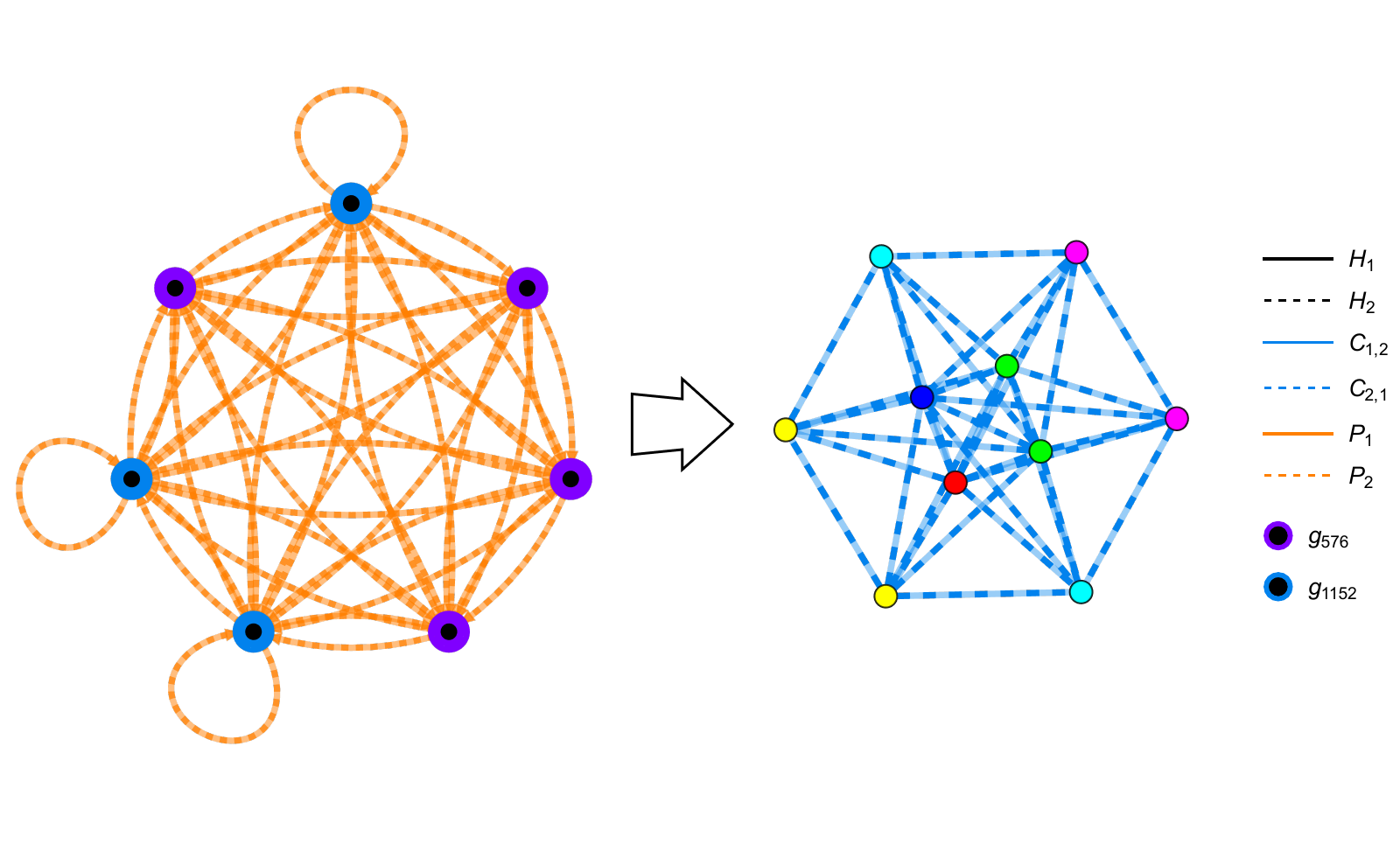}
        \caption{Reachability graph of $\ket{D^4_2}$ under $\mathcal{C}_2$ (left), and its associated contracted graph (right). We display the $5760$-vertex reachability graph as a network of $\HC$ graphs $g_{576}$ and $g_{1152}$, connected by $P_1$ and $P_2$ gates. The contracted graph contains $10$ vertices, but we only ever observe $6$ entropy vectors due to how the $g_{576}$ and $g_{1152}$ copies connect under phase action. The $6$ different entropy vectors shown are given in Table \ref{tab:D42EntropyVectors}.}
        \label{PhaseConnectedG576ContractedGraph}
    \end{figure}

The $5760$-vertex reachability graph in Figure \ref{PhaseConnectedG576ContractedGraph} consists of $4$ copies of $g_{576}$ and $3$ copies of $g_{1152}$, all connected via $P_1$ and $P_2$ operations. While every state in the full $5760$-vertex reachability graph is stabilized by $2$ elements of $\mathcal{C}_2$, some states have a stabilizer group completely contained within $\HC$. States stabilized by $2$ elements of $\HC$ are found in one of the $4$ copies of $g_{576}$ in Figure \ref{PhaseConnectedG576ContractedGraph}. Alternatively, states which are stabilized by $2$ elements of $\mathcal{C}_2$, but only the identity in $\HC$, are found in one of the $3$ copies of $g_{1152}$.

If we identify vertices connected by entropy-preserving operations in the $\mathcal{C}_2$ reachability graph of $\ket{D^4_2}$, we are left with a contracted graph containing $10$ vertices shown to the right of Figure \ref{PhaseConnectedG576ContractedGraph}. While this contracted graph has $10$ vertices, we only ever observe $6$ different entropy vectors among those $10$ vertices. We again return to this point in the discussion. The contracted graph in Figure \ref{PhaseConnectedG576ContractedGraph} also reflects the symmetry among magenta, cyan, and yellow vertices observed in Figure \ref{G576WithContractedGraph}. These $6$ entropy vectors which can be generated from $\ket{D^4_2}$ under $\mathcal{C}_2$ are given in Table \ref{tab:D42EntropyVectors}.

In this subsection we extended our analysis beyond the stabilizer states, building contracted graphs for non-stabilizer Dicke states under the action of $\HC$ and $\mathcal{C}_2$. States $\ket{D^n_k}$, for $k \neq n$, are particularly interesting at $n \geq 3$ qubits as they comprise a class of non-stabilizer states that are non-trivially stabilized by elements of $\mathcal{C}_n$. We constructed the two possible reachability graphs for $\ket{D^n_k}$, one for states $\ket{D^n_1}$ and $\ket{D^n_{n-1}}$, and the other for all $\ket{D^n_k}$ with $1 < k < n-1$. We described how each Dicke state reachability graph under $\mathcal{C}_2$ corresponds to a connection of $\HC$ reachability graphs $g_{288^*},\, g_{576},$ and $g_{1152}$ under $P_1$ and $P_2$ operations.

We built the contracted graphs for each $\ket{D^n_k}$ reachability graph, both under $\HC$ and $\mathcal{C}_2$ action. We illustrated that states $\ket{D^n_1}$ and $\ket{D^n_{n-1}}$ can realize $5$ different entropy vectors under $\mathcal{C}_2$. Alternatively, states of the form $\ket{D^n_k}$ with $1 < k < n-1$ can achieve $6$ different entropy vectors under $\mathcal{C}_2$. In the next subsection we completely generalize to an argument for $\mathcal{C}_n$ action on arbitrary quantum states. We use our construction up to this point to bound the entropy vector possibilities that can be achieved for any state under $n$-qubit Clifford action.

\subsection{Entanglement in $n$-Qubit Clifford Circuits}

We now use our results to present an upper bound on entropy vector evolution in Clifford circuits, for arbitrary qubit number. We begin by determining the subset of $\mathcal{C}_n$ operations which cannot modify the entanglement entropy of any state. We then build a contracted graph by identifying the vertices in the $\mathcal{C}_n$ Cayley graph that are connected by entropy-preserving circuits.

Local actions, i.e.\ all operations which act only on a single qubit in some $n$-qubit system, will always preserve a state's entropy vector. When considering action by the Clifford group $\mathcal{C}_n$, the subgroup of all local actions is exactly the group generated by $n$-qubit Hadamard and phase gates, which we denote $(HP)_n$. We build $(HP)_n$ as the direct product \cite{Keeler:2023xcx}
\begin{equation}
    (HP)_n \equiv \prod_{i=1}^n \langle H_i,\,P_i \rangle.
\end{equation}
Since $(HP)_n$ is a direct product, and $|\langle H_i,\,P_i \rangle| = 24$, the order of $|(HP)_n|$ is just $24^n$. The order of the $n$-qubit Clifford group is likewise known \cite{Walter2016}. We can compute $|\mathcal{C}_n|$ as
\begin{equation}
    |\mathcal{C}_n| = 2^{n^2+2n}\prod_{j=1}^n(4^j-1).
\end{equation}

Generating the right coset space $\mathcal{C}_n \backslash (HP)_n$ identifies all elements in $\mathcal{C}_n$ equivalent up to local gate operations. Invoking Lagrange's theorem (Eq.\ \eqref{LagrangeTheorem}) allows us to compute the size of $\mathcal{C}_n \backslash (HP)_n$ as
\begin{equation}\label{NonLocalGroupOrder}
    \frac{|\mathcal{C}_n|}{|(HP)_n|} = \frac{2^{n^2-n}}{3^n}\prod_{j=1}^n(4^j-1).
\end{equation}

It is important to note that $(HP)_n$ is not a normal subgroup of $\mathcal{C}_n$, which we can immediately verify by considering any Hadamard operation $H_j \in (HP)_n$. The element
\begin{equation}
    C_{i,j}H_jC_{i,j}^{-1} \notin \langle H_i,\,P_i,\,H_j,\,P_j \rangle,
\end{equation}
which violates the necessity that any normal subgroup be invariant under group conjugation. Accordingly, $(HP)_n$ does not generate a quotient of $\mathcal{C}_n$.

The coset space $\mathcal{C}_n \backslash (HP)_n$ partitions $\mathcal{C}_n$ into sets of Clifford circuits which are equivalent up to local action. Consequently, Eq.\ \eqref{NonLocalGroupOrder} provides an upper bound on the number of entropy vectors that can possibly be generated under any $n$-qubit Clifford circuit, for any arbitrary quantum state. This upper bound is equivalently captured by directly building a contracted graph from the $\mathcal{C}_n$ Cayley graph, and counting the number of vertices. The right panel of Figure \ref{FullC2WithContractedGraph} illustrates the $20$-vertex contracted graph of the $\mathcal{C}_2$ Cayley%
\footnote{Formally, the left panel of Figure \ref{FullC2WithContractedGraph} depicts the reachability graph for some set of states, rather than the Cayley graph of $\mathcal{C}_2$. However, since the particular class of states is stabilized by only the identity in $\mathcal{C}_2$, the reachability graph in the left panel of Figure \ref{FullC2WithContractedGraph} is exactly the phase-modded $\mathcal{C}_2$ Cayley graph.} %
graph. Table \ref{tab:EntropicDiversity} gives the explicit number of entropy vectors that can be achieved using $n \leq 5$ qubit Clifford circuits.
\begin{table}[h]
    \centering
    \begin{tabular}{|c||c|}
    \hline
    $n$ & $|\mathcal{C}_n|/|(HP)_n|$\\
    \hline
    \hline
    $1$ & $1$ \\
    \hline
    $2$ & $20$ \\
    \hline
    $3$ & $6720$ \\
    \hline
    $4$ & $36556800$ \\
    \hline
    $5$ & $3191262412800$ \\
    \hline
    \end{tabular}
\caption{Maximum number of entropy vectors that can be generated using elements of the $n$-qubit Clifford group, for $n \leq 5$.}
\label{tab:EntropicDiversity}
\end{table}

In Eq.\ \eqref{NonLocalGroupOrder} we count the right cosets of $\mathcal{C}_n$ by the subgroup of entropy-preserving operations. This upper bound equivalently constrains the number of entropy vectors which can be realized by a generic quantum state, stabilized by only $\mathbb{1} \in \mathcal{C}_n$, under any Clifford circuit. However, we can tighten this bound for states which are non-trivially stabilized by some subset of $\mathcal{C}_n$. For a state $\ket{\psi}$ with stabilizer group $\mathcal{S}_{\mathcal{C}_n}(\ket{\psi})$, the number of achievable entropy vectors is bounded by the size of the double coset space $(HP)_n \backslash \mathcal{C}_n/\mathcal{S}_{\mathcal{C}_n}(\ket{\psi})$. As by Eq.\ \eqref{DoubleCosetOrder}, the size of $(HP)_n \backslash \mathcal{C}_n/\mathcal{S}_{\mathcal{C}_n}(\ket{\psi})$ is
\begin{equation}\label{GeneralContractedGraphOrder}
    |(HP)_n \backslash \mathcal{C}_n/\mathcal{S}_{\mathcal{C}_n}(\ket{\psi})| = \frac{1}{|(HP)_n||\mathcal{S}_{\mathcal{C}_n}(\ket{\psi})|} \sum_{(h,s) \in (HP)_n \times \mathcal{S}_{\mathcal{C}_n}(\ket{\psi})} |\mathcal{C}_n^{(h,s)}|, 
\end{equation}
where $\mathcal{C}_n^{(h,s)}$ is defined by Eq.\ \eqref{DoubleCosetDefinition}.

Applying Eq.\ \eqref{GeneralContractedGraphOrder} when $\ket{\psi}$ is a stabilizer state dramatically reduces the number of possible entropy vectors that can be reached under $\mathcal{C}_n$. Specifically, when restricting to group action by $\HC$, Eq.\ \eqref{GeneralContractedGraphOrder} computes the vertex count for each of the five contracted graphs shown in Figures \ref{G24WithContractedGraph} -- \ref{FullC2WithContractedGraph}.

In this subsection we provided an upper bound on the number of entropy vectors that can be generated by any Clifford circuit, at arbitrary qubit number. For a generic quantum state, we showed that the number of possible entropy vectors is bounded by the size of the right coset space $\mathcal{C}_n \backslash (HP)_n$. Alternatively, for states stabilized by additional elements in $\mathcal{C}_n$, the number of possible entropy vectors is bounded by the size of the double coset space $(HP)_n \backslash \mathcal{C}_n/\mathcal{S}_{\mathcal{C}_n}(\ket{\psi})$.

\section{From Entropic Diversity to Holographic Interpretation}\label{sec:diversity}

The contracted graphs in Section \ref{ContractedGraphsSection} illustrate the diversity of entropy vectors on $\HC$ and $\mathcal{C}_2$ reachability graphs. We now analyze this entropic diversity as we move towards a holographic interpretation of our results. We begin by considering the maximum number of different entropy vectors that can be found on each of the $\HC$ and $\mathcal{C}_2$ graphs studied in the above section, as well as the minimum number of qubits needed to realize that maximal diversity. We explore the implications of entropic diversity and graph diameter as constraining the transformations of a geometric gravitational dual in holography. We then present the number of $\HC$ subgraphs, including isomorphic subgraphs with different entropic diversities, as we increase qubit number. We remark how our contracted graphs encode information about entropy vector evolution through entropy space.

\subsection{Clifford Gates in Holography}

The AdS/CFT conjecture \cite{Maldacena:1997re} is a bulk/boundary duality which relates gravitational objects in an asymptotically hyperbolic spacetime, evaluated at some fixed timeslice $\Sigma$, with computable properties of a quantum-mechanical system on the boundary of that spacetime $\partial \Sigma$. For a special class of quantum states known as holographic states, the Ryu-Takayanagi formula relates all components of the state's entropy vector to areas of extremal surfaces in the dual gravity theory \cite{Ryu:2006bv,Faulkner:2013ana}. In this way, a description of the spacetime geometry in $\Sigma$ is inherited from knowledge of the entanglement structure on $\partial \Sigma$. For this relation to hold, holographic states are required to have an entropy vector structure which satisfies a set of holographic entropy inequalities \cite{Bao2015,HernandezCuenca2019}. One holographic inequality, the monogamy of mutual information (MMI) \cite{Hayden2013}, reads
\begin{equation}\label{MMI}
    S_{AB} + S_{AC} + S_{BC} \geq S_{A} + S_{B} + S_{C} + S_{ABC},
\end{equation}
and must be satisfied for all%
\footnote{It is important to note that each $A,B,C \subseteq \partial \Sigma$ may separately correspond to the disjoint union of multiple qubits in the $n$-party boundary theory. Accordingly, the MMI inequality in Eq.\ \eqref{MMI} must hold for disjoint triples $\{A,B,C\}$, as well as those of the form $\{AB,C,DE\}$ or $\{ABC,DE,F\}$, and so on. Furthermore, holographic states must saturate or satisfy MMI for all permutations among any chosen $A,B,C \subseteq \partial \Sigma$.} %
$A,B,C \subseteq \partial \Sigma$. While MMI constitutes only one of many holographic entropy inequalities, it arises at four qubits, while all other holographic inequalities require more parties.

Understanding the entropy-vector dynamics of a state in $\partial \Sigma$ gives insight into bulk geometric transformations in $\Sigma$. When a local operator acts on $\ket{\psi}$ and modifies its entropy vector to another vector within the holographic entropy cone, geodesics in the dual spacetime geometry are likewise modified in accordance with the RT formula. Consequently, analyzing how a group of operators transforms the entropy vector of a state can reveal how gate action on $\partial \Sigma$ alters geometries in $\Sigma$. When a sequence of Clifford gates causes the state to violate holographic inequalities, the geometry may be only a semi-classical approximation.

The distance between vertices on reachability graphs encodes a natural notion of circuit complexity. Entropy vectors which populate the same reachability graph, e.g.\ under $\HC$ or $\mathcal{C}_2$, may be considered close in the sense that a limited number of gate applications is required to transform a state with one entropy vector into some state with another. The gravitational dual geometries of states with ``nearby'' entropy vectors may be considered close in a similar sense, since a small number of manipulations are needed to transform one dual geometry into each other. 

Some $n$-qubit stabilizer states have entropy vectors which violate the holographic entropy inequalities, beginning at $n=4$. Since stabilizer entanglement is generated by bi-local gates, $2$-qubit Clifford operations are sufficient to generate all stabilizer entropy vectors in an $n$-qubit system. We can therefore explore the transition from holographic entropy vectors to non-holographic stabilizer entropy vectors by observing entropy vector evolution under $\mathcal{C}_2$. In the following subsections we discuss how entropic diversity on $\HC$ and $\mathcal{C}_2$ reachability graphs can inform us about states which are geometrically close, and not so close, in the dual gravitational theory.

\subsection{Maximal Entropic Diversity for Stabilizer States}

Each $\HC$ and $\mathcal{C}_2$ reachability graph describes the full orbit of some state $\ket{\psi} \in \Hil$ under the action of $\HC$ or $\mathcal{C}_2$ respectively. While we can construct reachability graphs for an arbitrary $n$-qubit quantum state, including states with arbitrary entanglement structure, the set of possible entropy vectors that can be reached under $\HC$ and $\mathcal{C}_2$ remains bounded at the operator level. For a given reachability graph, we refer to the maximum number of possible entropy vectors that can be generated in that graph as the maximal entropic diversity of the graph.

Table \ref{tab:MaximalHCColoringTable} gives each stabilizer state $\HC$ reachability graph, and the maximal entropic diversity determined by its contracted graph. For certain subgraphs, such as $g_{144},\,g_{288},$ and $g_{1152}$, the number of qubits needed to realize the maximal entropic diversity is higher than the number of qubits at which each graph first appears.
\begin{table}[h]
    \centering
    \begin{tabular}{|c||c|c|c|}
    \hline
    $\HC$ Graph & \multicolumn{1}{|p{2.6cm}|}{\centering Max Entropic \\ Diversity} & \multicolumn{1}{|p{2.7cm}|}{\centering Stab. Qubit\\ Num. Appears} & \multicolumn{1}{|p{3.5cm}|}{\centering Stab. Qubit Num.\\ Max Diversity}\\
    \hline
    \hline
    $g_{24}$ & $2$  & $2$ & $2$ \\
    \hline
    $g_{36}$ & $2$  & $2$ & $2$ \\
    \hline
    $g_{144}$ & $5$  & $3$ & $6$ \\
    \hline
    $g_{288}$ & $5$  & $3$ & $6$ \\
    \hline
    $g_{1152}$ & $18$  & $4$ & $7$ or $8$ \\
    \hline
    \end{tabular}
\caption{Stabilizer state $\HC$ graphs listed alongside their maximal entropic diversities, set by contracted graphs. We give the qubit number when each graph is first observed for stabilizer states, and the minimum qubit number needed to realize the maximal entropic diversity for stabilizer states. We have found $g_{1152}$ graphs with maximal diversity for $8$-qubit stabilizer states, but have not completely ruled out a maximally diverse $g_{1152}$ graph at $7$ qubits since an exhaustive search is computationally difficult.}
\label{tab:MaximalHCColoringTable}
\end{table}

The entropy vectors on $g_{24}$ and $g_{36}$ correspond to maximal and minimal $2$-qubit entanglement, and can therefore be achieved by entangling only $2$ qubits in an $n$-party system. These two entropy vectors are close in the sense that they are connected by a single $C_{1,2}$ action.  Since this single gate acts on only $2$ out of the $n$ qubits, we expect states with these entropy vectors to admit close dual (possibly semi-classical) geometries. Analogously, altering only small segments of the boundary of a holographic state will affect its geometry only inside the entanglement wedge of the union of these segments. 

For larger reachability graphs, the graph diameter upper bounds the $\HC$ gate distance, and thus the geometric closeness, of the included entropy vectors. In particular, $g_{1152}$ is the $\HC$ reachability graph for generic quantum states, and its maximal entropic diversity gives an upper bound on the number of distinct entropy vectors, and thus the number of distinct semi-classical geometries, reachable under $\HC$ action. 

We additionally compile the entropic diversity data for all stabilizer state $\mathcal{C}_2$ reachability graphs. As shown throughout Section \ref{ContractedGraphsSection}, every $\mathcal{C}_2$ graph is a complex of $\HC$ subgraphs attached by $P_1$ and $P_2$ edges. Table \ref{tab:MaximalC2ColoringTable} lists the different $\mathcal{C}_2$ complexes, and the maximal entropic diversity of each.
\begin{table}[h]
    \centering
    \begin{tabular}{|c||c|c|c|}
    \hline
    $\mathcal{C}_2$ Graph & \multicolumn{1}{|p{2.6cm}|}{\centering Max Entropic \\ Diversity} & \multicolumn{1}{|p{2.7cm}|}{\centering Stab. Qubit\\ Num. Appears} & \multicolumn{1}{|p{3.5cm}|}{\centering Stab. Qubit Num.\\ Max Diversity}\\
    \hline
    \hline
    $g_{24} + g_{36}$ & $2$  & $2$ & $2$\\
    \hline
    $3 \cdot g_{144}+g_{288}$ & $5$  & $3$ & $6$\\
    \hline
    $10 \cdot g_{1152}$ & $20$ & $4$ & $7$ or $8$\\
    \hline
    \end{tabular}
\caption{Each stabilizer state $\mathcal{C}_2$ graph, built of attached $\HC$ subgraphs. Each graph is listed alongside its maximal entropic diversity, set by its contracted graph. We give the first time each graph appears as a stabilizer state orbit, and the first time each graph achieves maximal entropic diversity for stabilizer states.}
\label{tab:MaximalC2ColoringTable}
\end{table}

The addition of $P_1$ and $P_2$ enables two more entropy vectors to be reached by states in a $g_{1152}$ subgraph. Although this section has so far concentrated on the stabilizer states, the $10 \cdot g_{1152}$ $\mathcal{C}_2$ complex is actually the generic reachability graph for arbitrary quantum states, which are not stabilized by any non-identity element of a given two-qubit Clifford group. Accordingly, the $20$ entropy vectors in this complex constrain the possible unique entropy vectors that can be generated by starting with a generic quantum state and acting with $2$-qubit Clifford operations.

In this subsection we provided Tables \ref{tab:MaximalHCColoringTable}--\ref{tab:MaximalC2ColoringTable} which detailed the maximal entropic diversity of each stabilizer state $\HC$ and $\mathcal{C}_2$ reachability graph. Additionally, we provided the minimal system size needed to realize each maximal entropic diversity in a stabilizer state orbit. Note that for other quantum states with the same reachability graphs, maximal entropic diversity could be achieved at lower qubit numbers. We speculated that the maximal entropic diversity of reachability graphs constrains the available transformations, and that the graph diameter constrains the dissimilarity, of the dual geometries that can be generated from $\HC$, or $\mathcal{C}_2$, action on the boundary state. In the next subsection we analyze the number, and diversity, of stabilizer state reachability graphs as the number of qubits in the system increases.

\subsection{$\mathcal{C}_2$ Subgraph Count by Qubit Number}

The number of times each stabilizer state $\mathcal{C}_2$ reachability graph in Section \ref{ContractedGraphsSection} occurs in the set of $n$-qubit stabilizer states increases with every qubit added to the system. Furthermore, as we increase qubit number we observe different entropic diversities which are possible on $\mathcal{C}_2$ reachability graphs. Table \ref{tab:SubgraphCountTable} gives a count for each variety of stabilizer state $\mathcal{C}_2$ graph, with increasing qubit number, for $n \leq 5$ qubits.
\begin{table}[h]
    \centering
    \begin{tabular}{|c||c|c|c|}
     \hline
     & \multicolumn{3}{|c|}{$\mathcal{C}_2$ Graph} \\
    \hline
    Qubit \# & $g_{24}/g_{36}$ & $g_{144}/g_{288}$ & $g_{1152}$\\
    \hline
    \hline
    2 & 1 (2) & 0  & 0\\
    \hline
    3 & 6 (2) & 1 (3) & 0\\
    \hline
    4 & 60 (2) & 12 (3), 18 (4) & 1 (2), 9 (4)\\
    \hline
    5 & 1080 (2) & 180 (3), 1080 (4) & 18 (2), 216 (4), 486 (6), 540 (7)\\
    \hline
    \end{tabular}
\caption{Distribution of stabilizer state $\mathcal{C}_2$ reachability graphs, and their different entropic diversities, for $n \leq 5$ qubits. The first number in each cell gives the number of occurrences for each $\mathcal{C}_2$ subgraph, while the number in parentheses gives the entropic diversity of each subgraph variation.}
\label{tab:SubgraphCountTable}
\end{table}

The overall count of each $\mathcal{C}_2$ subgraph increases as the size of the system grows. Graph $g_{1152}$ however, shown in the final column of Table \ref{tab:SubgraphCountTable}, has an occurrence count which increases the fastest with qubit number. As expected, when the system size grows large the percentage of states stabilized by any non-identity $2$-qubit Clifford subgroup decreases. 

Subgraphs $g_{144}/g_{288}$ can have an entropic diversity of $3,\,4,$ or $5$, while states in a $g_{1152}$ $\mathcal{C}_2$ complex can reach up to $20$ different entropy vectors. As qubit number increases the number of entanglement possibilities grows, yielding more complex entropy vectors. Entropy vectors with sufficient complexity will change the maximal number of allowed times under $\mathcal{C}_2$ action. We therefore expect the number of $g_{144}/g_{288}$ graphs with $5$ entropy vectors, and $g_{1152}$ $\mathcal{C}_2$ graphs with $20$ entropy vectors, to dominate the subgraph occurrence count in the large system limit. For larger subgraphs, e.g.\ those composed of $g_{1152}$ subgraphs, understanding the precise distribution of entropic diversity for arbitrary qubit number presents a challenging problem, which we leave for future work. We now conclude this section with a discussion of Dicke state entropic diversity in $\HC$ and $\mathcal{C}_2$ reachability graphs.

\subsection{Maximum Entropic Diversity for Dicke States}

We now analyze the entropic diversity of the Dicke state $\ket{D^n_k}$ reachability graphs in Section \ref{DickeStateSubsection}. Subgraphs $g_{288^*}$ and $g_{576}$ correspond to the two possible $\ket{D^n_k}$ orbits under $\HC$ action, shown in Figures \ref{WStateG288WithContractedGraph}--\ref{G576WithContractedGraph}. Under the full action of $\mathcal{C}_2$, $P_1$ and $P_2$ edges attach copies of $g_{288^*}, g_{576},$ and $g_{1152}$ together, creating the graph complexes seen in Figures \ref{PhaseConnectedWG288ContractedGraph}--\ref{PhaseConnectedG576ContractedGraph}. In Table \ref{tab:MaximalDiversityTableDickeStates} we present the maximal entropic diversity of each $\ket{D^n_k}$ $\HC$ and $\mathcal{C}_2$ reachability graph, as determined by their contracted graphs.
\begin{table}[h]
    \centering
    \begin{tabular}{|c||c|c|c|}
    \hline
    Graph & \multicolumn{1}{|p{3cm}|}{\centering Max Entropic \\ Diversity} & \multicolumn{1}{|p{3cm}|}{\centering First Appears \\ for $\ket{D^n_k}$} & \multicolumn{1}{|p{3cm}|}{\centering Max Diversity\\ for $\ket{D^n_k}$}\\
    \hline
    \hline
    $g_{288^*}$ & $5$  & $3$ & $5$ \\
    \hline
    $2 \cdot g_{288^*} + 2 \cdot g_{576} + g_{1152}$ & $6$  & $3$ & $5$\\
    \hline
    $g_{576}$ & $9$ & $4$ & $6$ \\
    \hline
    $4 \cdot g_{576} + 3 \cdot g_{1152}$ & $10$ & $4$ & $6$\\
    \hline
    \end{tabular}
\caption{All $\HC$ reachability graphs (rows $1$ and $3$) and $\mathcal{C}_2$ reachability graphs (rows $2$ and $4$) for Dicke states. We give the maximal entropic diversity of each graph, as set by the contracted graph, as well as the first time the graph appears for Dicke states and the largest entropic diversity achieved among $\ket{D^n_k}$ states. For $\mathcal{C}_2$ graphs in particular, we never observe a $\ket{D^n_k}$ orbit that achieves the maximum number of allowed entropy vectors.}
\label{tab:MaximalDiversityTableDickeStates}
\end{table}

Both $\mathcal{C}_2$ reachability graphs in Table \ref{tab:MaximalDiversityTableDickeStates} do not achieve their maximal entropic diversities as orbits of Dicke states. We expect that a state with sufficiently general entanglement structure, which also shares one of these reachability graphs%
\footnote{Recall that the reachability graphs in Table \ref{tab:MaximalDiversityTableDickeStates} are shared by all states with stabilizer group given by Eqs. \eqref{WStateStabGroup} or \eqref{AllOtherDStabilizer}, and are not restricted to $\ket{D^n_k}$ orbits. Since the entropy vector is a state property, the state structure determines entropy vector complexity and therefore how much an entropy vector can change under some group action.}%
, would realize the maximum allowed number of distinct entropy vectors, though we have not shown this explicitly. In Section \ref{sec:discussion} we speculate on the highly symmetric structure of $\ket{D^n_k}$ entropy vectors as a potential cause for the maximal diversity not being achieved in such graphs.

In this section we analyzed the entropic diversity of reachability graphs studied throughout Section \ref{ContractedGraphsSection}. We detailed each reachability graph achieves its maximal entropic diversity, and speculate implications for the geometric interpretations of state entropy vectors in a dual gravity theory. We demonstrated how certain $\HC$ and $\mathcal{C}_2$ subgraphs appear more frequently with increasing qubit number, as well as how different entropic variations of each subgraph are distributed when the system size grows large. We addressed the notable case of Dicke state reachability graphs, which do not achieve their maximal entropic diversity as orbits of $\ket{D^n_k}$. We will now conclude this work with an overview of our results and some ideas for future research.

\section{Discussion and Future Work}\label{sec:discussion}

In this work we presented a procedure for quotienting a reachability graph to a contracted graph, which allowed us to analyze and bound entropy vector evolution under group action on a Hilbert space. We first constructed a reachability graph, built as a quotient of the group Cayley graph \cite{Keeler:2023xcx}, for a family of states defined by their stabilizer subgroup under the chosen group action. As a group-theoretic object, the vertex set of a reachability graph is the left coset space generated by the stabilizer subgroup for the family of states. We then further quotiented this reachability graph by identifying all vertices connected by edges that preserve the entropy vector of a state. This second graph quotient corresponds to the right coset space generated by the subgroup of elements which leave an entropy vector invariant. The resultant object, after both graph quotients, is a contracted graph. This contracted graph represents the double coset space built of group elements which simultaneously stabilize a family of states, and do not modify an entropy vector.

A contracted graph encodes the evolution of a state entropy vector under group action. Specifically, the number of vertices in a contracted graph strictly bounds the maximal number of distinct entropy vectors that can be found on a reachability graph. The edges of a contracted graph detail the possible changes an entropy vector can undergo through circuits composed of the group generating set. We built contracted graphs for all stabilizer states under the action of $\HC$ and $\mathcal{C}_2$, and demonstrated how the vertex count of each explains the reachability graph entropy distributions observed in our previous work \cite{Keeler2022,Keeler:2023xcx}.

Although we did derive a general upper bound on the number of different entropy vectors that can be reached using any $n$-qubit Clifford circuit starting from an arbitrary quantum state, much of our work focused on $\mathcal{C}_2$ contracted graphs.  However, we could use the same techniques to extend our analysis to $\mathcal{C}_n$, for $n \geq 3$, increasing our generating gate set for additional qubits. In fact, a presentation for $\mathcal{C}_n$ is proposed in \cite{Selinger2013}, using Clifford relations up through $3$ qubits. Understanding precisely how contracted graphs scale with qubit number might offer tighter constraints on achievable entropy vectors in $\mathcal{C}_n$ circuits, and enable us to study more general entropy vector transformations. In AdS/CFT, we only expect systems with arbitrarily large numbers of qubits to be dual to smooth classical qubits.  Consequently, an improved understanding of large-qubit-number contracted graph behavior would strengthen the connection to previous holographic entropy cone work, and could even yield insights for spacetime reconstruction efforts.

While our work in this paper has focused on Clifford circuits, the contracted graph protocol can be applied equally to circuits composed of alternative gate sets (for example, generators of crystal-like subgroups of $SU(N)$ such as $\mathbb{BT}$ \cite{Gustafson:2022xdt}). When the chosen gate set generates a finite group of operators, the associated Cayley graph will be finite, as will any graph quotients. For all such cases, a contracted graph analysis follows exactly as in Section \ref{ContractedGraphsSection}, and can be used accordingly to bound entropy vector evolution in different circuit architectures. By exploring different circuit constructions, we can precisely tune our analysis to focus on operations which may be preferred for specific experiments, e.g. arbitrary rotation gate circuits, constructions which replace multiple CNOT gates with Toffoli gates, and architectures that deliberately avoid gates which are noisy to implement.

Alternatively, if the chosen gate set is finite, but generates an infinite set of operators, we can impose a cutoff at some arbitrary fixed circuit depth. This cutoff truncates the associated Cayley graph, and enables an extension of our methods toward the study of universal quantum circuits up to finite circuit depth. Even without an imposed cutoff, we could use our graph analysis to establish bounds on the rate of entanglement entropy per gate application. This description is reminiscent of the notion of entanglement ``velocity'' in universal quantum circuits \cite{Couch:2019zni,Munizzi_2022}.

Although we were originally interested in entropy vector evolution under some chosen gate set, our techniques are sufficiently general to study the evolution of any state property (see footnote \ref{fn:state_property}). Of immediate interest, for example, is the amount of distillable quantum magic present in a state \cite{Bravyi2004,Bao:2022mkc}, and how this particular measure of non-stabilizerness changes throughout a quantum circuit. Since magic is preserved up to Clifford group action, one subgroup which leaves the amount of magic in a state invariant is exactly the set $\mathcal{C}_n$. 

In Section \ref{sec:diversity}, we analyzed the maximal entropic diversity of reachability graphs. A reachability graph has maximal entropic diversity when it realizes the maximum number of possible entropy vectors permitted by its contracted graph. We analyzed at which qubit number each $\HC$ and $\mathcal{C}_2$ reachability graph achieves maximal entropic diversity for stabilizer states, and remarked on the growth of entropic diversity with increasing qubit number. 

Since contracted graphs are defined at the operator level, we are also able to extend our analysis to non-stabilizer states.  In this paper, we generated all contracted graphs under $\HC$ and $\mathcal{C}_2$ for $n$-qubit Dicke states, a class of non-stabilizer states heavily utilized in optimization algorithms \cite{Cerezo:2020jpv,Niroula:2022wvn}. For these states, we derived an upper bound on the number of different entropy vectors that can exist in Dicke state $\HC$ and $\mathcal{C}_2$ reachability graphs. Interestingly, we have not observed $\mathcal{C}_2$ graphs achieve a maximal entropic diversity for Dicke states (see Figures \ref{PhaseConnectedWG288ContractedGraph}--\ref{PhaseConnectedG576ContractedGraph}). The contracted graphs of $g_{288^*}$ and $g_{576}$ permit $6$ and $10$ unique entropy vectors respectively, but we have only ever witnessed $5$ and $9$ entropy vectors for Dicke states with these graphs. We suspect the reason no Dicke state orbit attains its permitted maximal entropy diversity is due to additional $\mathcal{C}_2$ elements which stabilize specifically the highly symmetric entropy vectors of Dicke states \cite{Schnitzer:2022exe,Munizzi:2023ihc}.

In the body of this work, we connected our analysis of entropic diversity to the holographic framework, where entropy vectors admit a description as geometric objects in a dual gravity theory. We used our entropic diversity results to speculate about constraints on geometric transformations in the dual gravity theory, for states which are holographic or near-holographic. We interpret a contracted graph as a coarse-grained map of an entropy vector's trajectory, through entropy space, under a set of quantum gates. Thus, contracted graphs provide information about moving in entropy space, and thereby moving between different entropy cones. 

In future work, we plan to study precisely which Clifford operations move a holographic entropy vector out of, and back into, the holographic entropy cone. Furthermore, we will explore Clifford circuits that transition a stabilizer entropy vector from satisfying, to saturating, to failing holographic entropy conditions, particularly including the monogamy of mutual information (MMI).  We plan to concentrate on MMI since every explicit stabilizer state we have checked either satisfies all holographic inequalities, or violates at least one MMI condition. While \emph{a priori} we have no reason to expect that all stabilizer states which are not holographic necessarily violate MMI in particular, in practice we observe this to be the case empirically for $n \leq 6$ qubits.

\textit{The authors thank ChunJun Cao, Zohreh Davoudi, Temple He, Sergio Hernandez-Cuenca, Bharath Sambasivam, Howard Schnitzer, Aaron Szasz, and Claire Zukowski for helpful discussions.  CAK and WM are supported by the U.S. Department of Energy under grant number DE-SC0019470 and by the Heising-Simons Foundation ``Observational Signatures of Quantum Gravity'' collaboration grant 2021-2818. JP is supported by the Simons Foundation through the It from Qubit: Simons Collaboration on Quantum Fields, Gravity, and Information.}

\clearpage
\begin{singlespace}
\printbibliography[heading=subbibliography]
\end{singlespace}

\chapter{GRAVITATIONAL BACK-REACTION IS THE HOLOGRAPHIC DUAL OF MAGIC}\label{Chapter7}

\textit{We study interplay between magic and entanglement in quantum many-body systems. We show that non-local magic which is supported by the quantum correlations is lower bounded by the flatness of entanglement spectrum and upper bounded by the amount of entanglement in the system. We then argue that a smoothed version of non-local magic bounds the hardness of classical simulations for incompressible states. In conformal field theories, we conjecture that the non-local magic should scale linearly with entanglement entropy but sublinearly when an approximation of the state is allowed. We support the conjectures using both analytical arguments based on unitary distillation and numerical data from an Ising CFT. If the CFT has a holographic dual, then we prove that the non-local magic vanishes if and only if there is no gravitational back-reaction. Furthermore, we show that non-local magic approximately equals the rate of change of minimal surface area in response to the change of the tension of cosmic branes in the bulk.}

\section{Introduction}
\label{sec:intro}
Entanglement is an important quantum resource and an integral part of our understanding of quantum many-body physics and quantum gravity, such as topological order \cite{kitaev_topological_2006,levin_detecting_2006,hamma_bipartite_2005}, non-equilibrium dynamics \cite{hosur_chaos_2016,vonkeyserlingk_operator_2018,nahum_operator_2018,skinner_measurementinduced_2019}
, spacetime \cite{VanRaamsdonk:2010pw}, and black holes \cite{Almheiri_2013,Maldacena_2013}.  In the Anti-de~Sitter/Conformal Field Theory (AdS/CFT) correspondence \cite{Maldacena_1999,Witten:1998qj}, entanglement in the CFT is important for emerging spacetime geometry \cite{Czech_Lampros,Czech_2014,Czech_2017,Radon,Bao:2019bib} in the dual gravity theory, e.g. via the Ryu-Takayanagi formula\cite{Ryu_2006,Lewkowycz_2013,Faulkner_2013,Hubeny_2007}. Surprisingly, this connection between geometry and entanglement holds not only for holographic CFTs, but also for more general quantum many-body systems like tensor network toy models, which have been enormously successful in reproducing an analogous Ryu-Takayanagi formula\cite{HarlowRT}, the emergent bulk geometry, and subregion operator reconstruction through quantum error correction\cite{Pastawski_2015,Hayden_2016,Yang_2016,Harris_2018,ABSC,Steinberg_2023}. This is a profound development as it suggests the lessons from holography may also apply beyond the confines of AdS\cite{Jacobson_2016,Cao_2017,Cao_2018}.

However, the entanglement patterns in the tensor network models alone do not capture the full quantum landscape spanned by holography. Despite many recent advances\cite{dong2023holographic,akers2024background,cheng2022random,ABSC,HMERA,Bao:2019}, it is still unclear how gravity can emerge in such models. In particular, neither the holographic stabilizer codes\cite{Pastawski_2015} nor the random tensor networks\cite{Hayden_2016} can fully capture the CFT entanglement spectrum and gravitational back-reaction. Stabilizer tensor networks also fail to capture power-law correlations, robust multi-partite entanglement, and non-trivial area operators\cite{Akers:2019gcv,Hayden:2021gno,nogo}. From a resource-theoretic perspective, what are these tensor network models missing compared to the low energy states in holographic theories? We show in this work that the answer is magic\cite{bravyi_universal_2005,veitch_resource_2014,stabrenyi,bu_stabilizer_2023}, or more precisely, \emph{non-local magic}.

Quantumness comes in two layers: entanglement gives the power of building correlations stronger than classical and violates Bell's inequalities while quantum advantage characterizes the hardness of simulating quantum systems on a classical computer. The latter is distinct from entanglement  --- a task involving a highly entangled system is not always hard to simulate classically as it can be achieved purely using Clifford operations that are classically simulable. 
This notion of classical hardness that constitutes the second layer of quantumness is intimately connected to the amount of non-stabilizerness, also known as magic, in the system. Although magic alone cannot generate the intricate patterns of complexity that are crucial for the  complex behavior in a quantum wave-function, when used in conjunction with Clifford operations, non-stabilizerness ~\cite{bravyi_universal_2005} is both necessary and sufficient in realizing (fault-tolerant) universal quantum computation. Therefore, it is the remaining piece needed for quantum advantage and for simulating holographic conformal field theories.

In addition to being an important resource for fault-tolerant quantum computation\cite{bravyi_universal_2005,veitch_resource_2014} and quantum simulation, pioneering work has established magic as an important ingredient for characterizing quantum many-body systems\cite{white_conformal_2021,sarkar_characterization_2020,liu_manybody_2022,tarabunga_manybody_2023}, such as dynamics \cite{stabrenyi,chaosbymagic,sewell_mana_2022,rattacaso_stabilizer_2023}, quantum phases\cite{leone2023phase,Niroula:2023meg}, quantum circuits~\cite{leone_quantum_2021,oliviero_transitions_2021,bejan2023dynamical},and randomness\cite{vairogs2024extracting}. In the context of holography, \cite{white_conformal_2021,magicising,tarabunga_manybody_2023} showed that magic is abundant in CFTs and is therefore expected to play an important role for reproducing the correct CFT entanglement spectrum, for generating power-law correlations, for building non-trivial area operators in holographic codes, and for reproducing the correct multipartite entanglement in holographic geometries\cite{Hayden:2021gno}\footnote{Although this is not noted by the authors explicitly, it is clear that holographic states require $O(1/{G}_N)$ tripartite entanglement but cannot be predominantly GHZ-type\cite{Nezami:2016zni}. }.  

There are also many questions surrounding the role played by magic. Empirically, the amount of non-stabilizerness or non-Gaussianity\cite{veitch_resource_2014,campbell_catalysis_2011,leone_nonstabilizerness_2023,Hebenstreit_2019,saxena_quantifying_2022,bu_stabilizer_2023,bu2023discrete,weedbrook_gaussian_2012} present in a quantum process appears to correlate with the hardness of classical simulations\cite{zhang2024unconditional}, e.g. in stabilizer and matchgate simulations\cite{aaronson_improved_2004,gottesman1997stabilizer,Jozsa_2008,hebenstreit_computational_2020,bravyi_improved_2016,bravyi_trading_2016,bravyi_simulation_2019} as well as in Monte Carlo sampling\cite{magicMC}.
However, its precise connection with complexity is yet unclear. While it is proposed \cite{white_conformal_2021,nogo} that the replication of the CFT entanglement spectrum and emergent gravity in AdS/CFT requires magic, the specific mechanism through which magic accomplishes this also remains uncertain.
Furthermore, although the amount of magic present in a system can be illuminating all by itself, it is becoming clear the distribution of magic is equally, if not more, important for understanding  non-equilibrium dynamics and entanglement spectrum\cite{flatness}. For example, the amount of magic is generally expected to scale volumetrically with the number of qubits in quantum many-body systems. The tensor product of nonstabilizer states, CFT ground states, and Haar random states all have a high magic density and volume law magic scaling, and yet their physical properties and their usefulness for quantum computation are completely different. Therefore, a more profound understanding of the interplay between  entanglement and magic will shed new light on the structure of quantum matter, quantum information, and gravity.

In this work, we report multiple advances in respond to the above queries. We define non-local magic and offer compelling evidence for how it is connected to the hardness in classically simulating incompressible states. We provide rigorous bounds as well as computable estimates for non-local magic in any quantum system and show that it is lower bounded by the anti-flatness of the entanglement spectrum and upper bounded by various functions of the R\'enyi entropies. When applied to CFTs, we propose a straightforward relationship between magic, entropy, and anti-flatness. For theories with holographic dual, we show that the non-local magic controls the amount of gravitational back-reaction in response to stress energy, and thus critical for the emergence of gravity.

\section{Main results}
In this section, we explain the main results of this paper and lay down informally the setup and strategy of this work. Then, in the following sections, we derive them rigorously. The main goal of this paper is to show that  the non-local magic is responsible for the non-flat entanglement spectrum in a CFT and for the back-reaction in AdS through the AdS-CFT dictionary. 

Since the seminal work of Ryu and Takayanagi~\cite{Ryu_2006}, a number of entries have been added to the AdS-CFT dictionary where one can connect quantum information theoretic quantities on the boundary to geometric quantities in the bulk. Notably, the correspondence can be used to find the holographic dual to functions of the spectrum of a reduced density operator $\psi_A$ in the conformal field theory\cite{Dong:2018lsk}, where $A$ is a subsystem of the CFT. 
The strategy of this work is to find a holographic dual of magic in a state $\psi$  by connecting it to the spectrum of its reduced density operator $\psi_A$.

At the first sight, this may seem like an impossible task. There are two reasons: first of all, the way magic relates to  spectral properties is complicated. The reason is that, as long as  magic {\em only} quantifies the distillability of non-Clifford resources,  convex combinations through a probability distribution $p_i$ cannot create resources. 
However, we may think of magic in a way such that these probabilities are resourceful. In this case, there is a different resource theory of magic that we call $\stab_0$. This is the resource theory of magic established by the null set of stabilizer entropy\cite{stab0-inprep}. We will first therefore first develop this theory by employing as monotones both the trace distance $M_{dist}$ and relative entropy of resource $M_R$. They will both be useful later to establish our results.

The second reason why mapping magic in a spectral quantity is problematic is due to the fact that 
 magic is generally a property of the full state $\psi$, which has trivial spectrum if pure. What has a non-trivial spectrum is the reduced state  $\psi_A$ to a subregion $A$, because of entanglement. Then the question becomes: how can the spectrum of $\psi_A$ give us information on the magic of the full parent state $\psi$? 
 
The answer comes from the remarkable fact that the magic of a state $\psi$ is related to the average deviation from the flat spectrum of the spectrum of the reduced density operator $\psi_A$ through the Clifford orbit \cite{flatness,Keeler:2022ajf}. The Clifford orbit preserves the magic, but entangles the system \cite{Keeler:2023xcx,Munizzi:2023ihc,Keeler:2023shl}, therefore populating the spectrum of the reduced density operator. In fact, there is no need to take this average, as long as the spectrum of the subsystem density operator possesses an entropy obeying volume law. In this case, its flatness (or lack thereof) is enough to probe the magic of the full state. 

Unfortunately, this is not good enough to explore CFT as these states are not hosting volume law for entanglement. In order to exploit the flatness-magic correspondence for a theory that generally has an area-law scaling of entanglement, we must focus on the boundary $\partial A$ between the subregions $A$ and its complement $B$ where most of the entanglement is being mediated. On the Hilbert subspace supported on $\partial A$, the density operator $\psi_{\partial A}$ is well populated and as a consequence, we can compute its magic through the spectrum. This gives rise to the notion of non-local magic. 

The main results of this work are grouped in two parts: (i) Quantum information-theoretic results that rigorously define non-local magic 
 for both the magic measures defined above, namely the trace distance of non-local magic $M_{dist}^{(NL)}$ and the relative entropy of non-local magic $M_{R}^{(NL)}$ 
 and relate them to spectral quantities. In particular, it will play an important role in the notion of {\em anti-flatness} $\mathcal F$\cite{flatness}, that is, a measure of how much the spectrum of a density operator is far from a flat distribution; and (ii) the application of these tools to AdS/CFT by first making precise the relation between entanglement, non-local magic and spectral flatness in a CFT. Then for holographic CFTs, we show that the holographic dual of gravitational back-reaction is indeed non-local magic. 
 
\subsection{Quantum information-theoretic results}

The first result is that, given the bipartition $AB$, for a subsystem $A$ of a quantum state $\psi_{AB}$,  $M_{dist}^{(NL)}$ is lower bounded by the anti-flatness $\mathcal F(\psi_A)$ and upper bounded by the entanglement:
\begin{equation}
		\mathcal{F}(\psi_{ A})/8\le M_{dist}^{(NL)}(\psi_{AB})\le  \sqrt{1-e^{-S_{max}(A)}+e^{S_\infty(A)}\left(1-\frac{e^{\log d \lfloor S_{max}(A)/\log d \rfloor}}{e^{S_{max}(A)}}\right)}
	\end{equation}
where $\lfloor\cdot\rfloor$ is the floor function, $S_{max}(A):=\log \rank{\psi_A}$, $S_{\infty}(A)=\lambda_{\text{max}}(\psi_A)$. For this, we assume the total Hilbert space is a tensor product of qubits (or qudits) with uniform local dimension $d$. 

Second, by using the non-local magic measured by relative entropy $M^{(NL)}_{RS}$, one can find another relationship between magic in a  quantum state $\psi$ and its entanglement: 
\begin{equation}\label{eq:relstab0}
        S_{max}(A)-S(A) \leq  M^{(NL)}_{RS}(\psi_{AB})\leq \log d \lceil   S_{max}(A)/\log d  \rceil,
    \end{equation}
where $\lceil\cdot\rceil$ is the ceiling function. The lower and upper bounds in the above equation are essentially tight for weakly-entangled states. 
\cref{eq:relstab0} has also the advantage of allowing one to find good estimates for $M^{(NL)}_{RS}(\psi_{AB})$ in terms of the Schmidt coefficients of $\psi_{AB}$ (see \cref{estimate_prop}). This is important because non-local magic is otherwise very difficult to calculate.  Moreover, the relative entropy of magic allows us to define the smoothed (non-local) magic as 
\begin{equation}\label{smoothednlmagic}
\sma(\psi_{A\bar{A}}) := \min_{\Vert\chi-\psi_{A\bar{A}\Vert}<\epsilon} M_{RS}^{(NL)} (\chi).
\end{equation}
For a pure state therefore, we obtain the bounds 
\begin{equation}
     S_{max}^{\epsilon}(A)-(1-\epsilon)^{-1}S(A)\leq \sma{\psi_{AB}}\leq \log d \lceil   S^{\epsilon}_{max(A)}/\log d  \rceil,
     \label{eqn:smoothinequality0}
\end{equation}
with the smoothed maximal entropy is defined as $S_{max}^{\epsilon}(A):=\min_{\Vert\chi-\psi_A\Vert<\epsilon} \ln(\rank(\chi))$. Since the lower bound quantifies the compressibility of a state, we show that incompressible states with low entanglement, but high non-local magic, can still be difficult to classically simulate.

Finally for a system of qubits, we use the spectral information $\{\lambda_i\}$ of state $\psi$, along with the magic measure known as stabilizer 2-R\'enyi entropy $\mathcal{M}_2$, to estimate the non-local magic $\mathcal{M}_2(\{\lambda_i\})$. This calculation yields a tighter upper bound on non-local magic than \cref{eq:relstab0}, stating
\begin{equation}
\mathcal{M}_2^{NL}(\psi_{AB})\leq \min\{2S_2(A),4(S_{max}(A)-S_{1/2}(A))\},
\end{equation}
where $S_{n}(A)$ are the R\'enyi-$n$ entropies of $\psi_A$ and $S_{max}$ is the logarithm of Schmidt rank, which is taken to be an integer power of two. 

For each of the above measures, we show that non-local magic vanishes if and only if the entanglement spectrum is flat, see Lemma \ref{lemmaNL}.

\subsection{AdS/CFT results}

We now state the main holographic result of this work. One can use non-local magic to derive an RT-like formula for gravitational back-reaction, defined as the susceptibility of a backreacted surface area $\mathcal{A}$ with respect to the insertion of a cosmic brane with tension $\mathcal{T}$. The first step is connecting back-reaction to spectral quantities. We obtain, 
%
\ba\label{eq:nFoverpurity}
\left.\frac{\partial \mathcal{A}}{\partial \mathcal{T}}\right\vert_{\mathcal{T}=0} \approx -\Big(\frac{4G}{\mathrm{Pur}(\psi_A)}\Big)^2\mathcal{F}(\psi_A),
\ea
where the approximation holds when $S_2(A)-S_3(\psi_A)<1/2$, i.e., for the small subregion $A$ or in the near-flat limit. 
Together with the relation between anti-flatness and non local magic, Theorem \ref{th:magicdist},
we find  
\ba M_{dist}^{(NL)}(\psi_{AB})\ge \frac{1}{8}\Big(\frac{\mathrm{Pur}(\psi_A)}{4G}\Big)^2\left\vert\frac{\partial \mathcal{A}}{\partial \mathcal{T}}\right\vert_{\mathcal{T}=0} \ge \frac{1}{8}\Big(\frac{e^{-\mathcal A/4G}}{4G}\Big)^2\left\vert\frac{\partial \mathcal{A}}{\partial \mathcal{T}}\right\vert_{\mathcal{T}=0} \propto \frac{1}{8}\left|\frac{\partial e^{-2\mathcal A/4G}}{\partial \mathcal{T}} \right|_{\mathcal{T}=0} \ea
the left-hand side is the magic in the CFT side, the right end side of the above equation is a measure of the back-reaction in AdS. As we prove in  \cref{section:brane}, the above equation also implies that back-reaction is non-zero only if non-local magic is non vanishing.  

Further exploiting the structure of entanglement in CFT, (see \cref{eqn:NLMflatness},) we can also obtain a simpler relation that holds more generally without the $S_2(A)-S_3(\psi_A)<1/2$ constraint.
%
\ba\label{branemagicNL0}
\left\vert\frac{\partial \mathcal{A}}{\partial \mathcal{T}}\right\vert_{\mathcal{T}=0} 
\approx \frac{(4G)^2}{\kappa} \mathcal{M}_2^{NL}(\psi_{AB})
\ea
which shows a more direct relation between gravitational back-reaction and non-local magic based on the stabilizer 2-R\'enyi entropy for some constant $\kappa$.

Now for more general CFTs that need not have holographic duals, the above relations continue to hold with suitable substitutions of $\mathcal{T}\rightarrow (n-1)/4Gn$ and $\mathcal{A}/4G\rightarrow \tilde{S}_n$ where $\tilde{S}_n$ is a function of R\'enyi entropy defined by \cite{Dong:2018lsk}. We present compelling evidence that an additive anti-flatness measure is proportional to the amount of non-local magic in the system. We show with analytical arguments and numerical results that the exact non-local magic in the CFT scales as $S(A)$ whereas the smoothed non-local magic scales as $\sqrt{S(A)}$. 
We then conjecture that such relations hold for general CFTs and apply this conjecture to evaluate magic for selected examples in holographic CFT using \cref{eqn:smoothinequality0}. Specifically, we do so for the static thermofield double state, and for non-equilibrium dynamics after local and global quantum quenches. We also examine the magic evolution in a time-evolved wormhole geometry described by a thermal field double state.
\section{Non-local Magic}

\subsection{Magic measures}\label{sec:magicmeasure}
In this section, we introduce several measures of magic that will be central to supporting the claims in this manuscript. In order to properly establish a magic state resource theory, it is essential that we define an initial null set for such a resource theory. To achieve this purpose, we introduce three null sets, which we label as $\pstab$,  $\stab_0$, and $\stab$. Then we derive the free operations on such sets. 

Additionally, we must introduce several useful concepts: the Pauli group, the Clifford group, and the set of stabilizer quantum states. Consider the Hilbert space of single qudit $\mathcal{H} = \mathbb{C}^d $, on which we define the following Pauli operators
\begin{equation}
	X\ket{i}=\ket{i+1} \quad Z\ket{j}=\omega^j\ket{j}, 
	\label{eq:XZ}
\end{equation}
where $\omega\equiv \exp(2 i \pi/d)$. The selection of operators in \eqref{eq:XZ} likewise defines the qudit computational basis $\{\ket{i}\}_i^d$.

The Pauli group $\tilde{\mathcal{P}}$ is defined as follows
\begin{equation}
	\tilde{\mathcal{P}}\equiv \langle \tilde{\omega}\bbbone, X,Z \rangle 
	\label{eq:Pauligroup}
\end{equation}
where $\langle \cdot \rangle$ labels the set generated by $\{\tilde{\omega}\bbbone, X,Z\}$, and $\tilde{\omega}=\omega$ for $d$ odd, and $\tilde{\omega}=\exp[i\pi /d] $ for $d$ even. When the number of qudits is $n$, the Pauli group $\tilde{\mathcal{P}}_n$ is defined as the $n$-fold tensor product of the single qudit Pauli group $\tilde{\mathcal{P}}$. 

The Clifford group $\mathcal{C}(d^n)$ is defined as the normalizer of the Pauli group, meaning that for any $U\in\mathcal{C}(d^n)$ we have $U^{\dagger}\tilde{\mathcal{P}}_n U\equiv\tilde{\mathcal{P}}_n$. The group $\mathcal{C}(d^n)$ is a multiplicative matrix group. For qubits, $d=2$ it  can be generated by the Hadamard, phase, and Controlled-Z quantum gates
\begin{equation}\label{CliffordGates}
   \operatorname{H}\equiv \frac{1}{\sqrt{2}}\begin{bmatrix}1&1\\1&-1\end{bmatrix}, \quad \operatorname{P}\equiv \begin{bmatrix}1&0\\0&i\end{bmatrix}, \quad     \operatorname{CZ} \equiv \begin{bmatrix}
            1 & 0 & 0 & 0\\
            0 & 1 & 0 & 0\\
	    0 & 0 & 1 & 0\\
	    0 & 0 & 0 & -1
            \end{bmatrix}.
\end{equation}
For general $d$ the generators are~\cite{jafarzadeh_randomized_2020} the controlled-$Z$ $\operatorname{CZ}$, the quantum Fourier Transform $\operatorname{F}$ and the phase gate $\operatorname{P}$, whose action of the $d$-computational basis is
\begin{equation}\label{cliffordqutitgates}
    \operatorname{CZ}\ket{ii^{\prime}}:=\omega^{ii^{\prime}}\ket{ii^{\prime}} \quad \operatorname{F}\ket{i}:=\frac{1}{\sqrt{d}}\sum_{i\in\mathbb{Z}_d}\omega^{ii^{\prime}}\ket{i^{\prime}} \quad \operatorname{P}\ket{i}:=\omega^{s(s+\phi_d)/2}\ket{s}
\end{equation}
where $\phi_d=1$ if $d$ is odd, $0$ otherwise. 
Notably, circuits composed of the Clifford gates in \eqref{CliffordGates} can be efficiently simulated on a classical computer \cite{gottesman_heisenberg_1998,aaronson_improved_2004}.

At this point, one can define the notion of stabilizer states for pure states. We first say that a pure state $\ket{\phi}$ is stabilized by $P\in \tilde{\mathcal{P}}_n$ if
$P\ket{\phi}=\ket{\phi}$. Then we define the pure stabilizer states as the set
\ba
\pstab^{(n)}:=\{ \ket{\phi}\bra{\phi}=\frac{1}{|G|}\sum_{P\in G} P | G\subset \tilde{\mathcal{P}}_n, \; G \,\mbox{abelian} \}
\ea
with the cardinality of $G$ is $|G|=d^n$ and $G$ is a group of commuting Pauli operators. Notice that $\pstab^{(n)}$ is the orbit through the Clifford group of any computational basis state for $n$ qudits, i.e., $\pstab^{(n)}= \{ C\ket{i_1\ldots i_n}| C\in\mathcal{C}(d^n)\}$. The notion of pure stabilizer states conveys the fact of a set of resources that is closed under Clifford operations.

For mixed states, the most primitive notion of stabilizer states is that of~\cite{stabrenyi} $\stab_0$, defined as the set of states $\sigma=\frac{1}{d^n}\sum{P\in G}$, where $G$ is a group of commuting Pauli operators (see~\cite{nielsen_quantum_2000}). In \cite{stabrenyi}, $\stab_0^{(n)}$ is introduced as the set of states for which the stabilizer entropy (SE) is zero and  SE is a good monotone for  $\pstab^{(n)}$.
From a more foundational perspective, $\stab_0^{(n)}$ is the set of states that can be purified in $\pstab^{(n)}$ and they can only yield trivial probability distributions, see\cite{stab0-inprep}

When one allows for general probabilities distributions we obtain the convex hull of $\pstab^{(n)}$, namely
 $\stab^{(n)}:=\{\sigma|\sigma=\sum_i p_i \ket{\phi_i}\bra{\phi_i},~ |\phi_i\rangle\in \mathrm{PSTAB}^{(n)}\}$.  Note that $\stab_0^{(n)}\subset\stab^{(n)}$. 
 
The next step in the definition of our measures of magic is to define the free operations of $\stab^{(n)}$ and $\stab_0^{(n)}$. For $\stab^{(n)}$ the free operations are given in~\cite{veitch_resource_2014}, and we list them here for the sake of completeness: 
\begin{enumerate}
	\item Clifford unitaries. $\rho\rightarrow U\rho U^{\dagger}$ with $U\in \mathcal{C}(d^n)$.
	\item Composition with stabilizer states, $\rho \rightarrow \rho \otimes \sigma $ with $\sigma $ a stabilizer state.
	\item Computational basis measurement on the first qudit, $\rho \rightarrow (\st{i}\otimes \bbbone_{n-1}) \rho (\st{i}\otimes \bbbone_{n-1} )/\Tr(\rho \st{i}\otimes \bbbone_{n-1})$ with probability $\Tr(\rho \st{i}\otimes \bbbone_{n-1})$
	\item Partial trace of the first qudit, $\rho \rightarrow \Tr_{1}(\rho)$
	\item The above operations conditioned on the outcomes of measurements or classical randomness. 
\end{enumerate}
It is straightforward to show that operations $1.-4.$ also apply to $\stab_0^{(n)}$ (see~\cref{app:stab0invariance}). However, it's important to note that stabilizer operations conditioned on measurements or classical randomness do not belong to the set of free operations for $\stab_0^{(n)}$. This is an important feature of the $\stab_0^{(n)}$ resource theory as it counts non-flat probabilities as resources. It is the key element to use deviation from flatness as the resource that connects magic in CFT to geometry in AdS.

Given the notion of null sets and free operations, one can then proceed to introduce suitable measures of magic. Let us start by defining the trace distance of magic:
\begin{definition}[Trace distance of magic$_0$]\label{def:mod_tr_dist_mag}
	The trace distance of magic$_0$ of a state $\psi$ is given by:
	\begin{equation}
	M_{\text{dist}}(\psi):=\min_{\sigma\in\stab_0^{(n)}}\frac{1}{2}\left\|\psi-\sigma  \right\|_{1}
  \end{equation} 
\end{definition}
\begin{proposition}
	The trace distance of magic satisfies the following properties: 
 \begin{enumerate}
\item Faithfulness: $M_{\text dist}(\rho)=0$ if and only if $\rho$ is a stabilizer state. 
\item Monotonicity: for all completely positive trace-preserving channels $\xi$ preserving $\stab_0^{(n)}$,  $M_{\text dist}(\xi(\rho))\le M_{\text dist}(\rho)$
\item Subadditivity: $M_{\text dist}(\rho_1\otimes \rho_2)\le M_{\text dist}(\rho_1)+ M_{\text dist}(\rho_2)$
\end{enumerate}
\end{proposition}
\begin{proof} 
	\begin{enumerate}
		\item By definition $M_{dist}(\psi)=0$ if and only $\psi\in\stab_0^{(n)}$, and so $\psi$ is a stabilizer state.
  \item The monotonicity descends from the monotonicity of the trace distance under trace-preserving CP maps. Because given a map $\xi:\stab_0^{(n)}\mapsto\stab_0^{(n^\prime)}$ we have
  \begin{align}
    M_{\text dist}(\xi(\rho))&=\min_{\sigma\in\stab_0^{(n^\prime)}}\frac{1}{2}\norm{\xi(\rho)-\sigma}_1\\&=\min_{\sigma\in\stab_0^{(n^\prime)}}\frac{1}{2}\norm{\xi(\rho-\sigma)}_1\\
    &\le\min_{\sigma\in\xi(\stab_0^{(n)})}\frac{1}{2}\norm{\xi(\rho-\sigma)}_1\\
    &\le\min_{\sigma\in\stab_0^{(n)}}\frac{1}{2}\norm{\rho-\sigma}_1=M_{\text dist}(\rho)
  \end{align}
where we used that $\stab_0^{(n^\prime)}\subseteq \xi(\stab_0^{(n)})$, the proof of the last statement is straightforward. One must observe that since $\xi$ is expressed in terms of stabilizer operations, the only operations that reduce the dimension are partial traces. Therefore, it is evident that since states in $\stab_0^{(n)}$ are mapped to stabilizer states in $\stab_0^{(n^\prime)}$ after a partial trace, the statement must hold true because there are more states whose partial trace returns the same state. 
\item Subadditivity:
\begin{align}
    M_{\text{dist}}(\rho_L)= M_{\text dist}(\rho_1\otimes \rho_2)&=\frac{1}{2}\min_{\sigma\in\stab_0^{(n)}}\norm{\rho_1\otimes\rho_2-\sigma}_1
\\&=\frac{1}{2}\min_{\sigma\in\stab_0^{(n)}}\|\rho_1\otimes\rho_2-\sigma_1\otimes\sigma_2+\sigma_1\otimes\sigma_2-\sigma\|_1
\\&\leq\frac{1}{2}\|\rho_1\otimes\rho_2-\sigma_1\otimes\sigma_2\|_1+\frac{1}{2}\min_{\sigma\in\stab_0^{(n)}}\|\sigma_1\otimes \sigma_2-\sigma\|_1
\\&\leq\frac{1}{2}\|\rho_1\otimes\rho_2+\rho_1\otimes\sigma_2 -\rho_1\otimes\sigma_2 -\sigma_1\otimes\sigma_2\|_1
\\&\leq\frac{1}{2}\|\rho_1\|_1\|\rho_2-\sigma_2\|_1+\frac{1}{2}\|\sigma_2\|_1\|\rho_1-\sigma_1\|_1
\\&\leq\frac{1}{2}\|\rho_2-\sigma_2\|_1+\frac{1}{2}\|\rho_1-\sigma_1\|_1
\end{align}
where we used that $\sigma_1,\sigma_2$ are two stabilizer states, then  $\min_{\sigma\in\stab_0^{(n)}}\norm{\sigma_1\otimes\sigma_2-\sigma}=0$, and the tightest bound is obtained by minimizing over $\sigma_1$ and $\sigma_2$ proving the statement.
	\end{enumerate}
\end{proof}
 One can also define an entropic quantity the  Relative stabilizer entropy of magic:
\begin{definition}[Relative Stabilizer Entropy of Magic]\label{def:rel_stab_magic}
The relative stabilizer entropy of magic of $\rho$ is given by
\begin{align}
    M_{RS}(\rho)= \min_{\sigma \in \stab_0^{(n)}}S(\rho||\sigma)  
 \label{eqn:rel_stab_magic}
\end{align}
\end{definition}
\begin{proposition}
    The relative stabilizer entropy is a magic monotone, i.e., 1. it is zero iff $\rho\in \stab_0^{(n)}$, 2. is invariant under Clifford conjugation, 3. is non-increasing on average under stabilizer measurement, 4. is non-increasing under partial trace and 5. is invariant under stabilizer composition.
\end{proposition}

\begin{proof}
The proof is similar to~\cite[Appendix A]{veitch_resource_2014}, where the only difference is the definition of $\stab^{(n)}$. Here we recount for completeness.
\begin{enumerate}
    \item Note that $S(\rho||\sigma)\geq0$ where equality is attained iff $\rho=\sigma$. Hence it only vanishes when $\rho\in \stab_0^{(n)}$, which by our definition is a stabilizer state. 
	
	\item Recall that $\rm STAB_0$ is invariant under Cliffords, therefore for $U\in\mathcal{C}(d^n)$ $$M_{RS}(U\rho U^{\dagger}) = \min_{\sigma\in \stab_0^{(n)}} S(U\rho U^{\dagger}||\sigma)=\min_{\sigma\in \stab_0^{(n)}} S(\rho ||U^{\dagger}\sigma U)=\min_{\sigma\in \stab_0^{(n)}}S(\rho||\sigma).$$ 
	
	\item The action of partial stabilizer measurements of the form $V_i=I\otimes |i\rangle\langle i|$ for some Pauli basis state $|i\rangle$ on $\stab_0$ returns a stabilizer state up to normalization. Using that $p_i=\Tr[\rho V_i], q_i=\Tr[\sigma V_i]$ and $\rho_i=V_i\rho V_i^{\dagger}, \sigma_i=V_i\sigma V_i^{\dagger}$, we can reuse the proof from \cite{veitch_resource_2014} and note that $$\sum_i p_i S\left(\left.\frac{\rho_i}{p_i}\right\Vert\frac{\sigma_i}{q_i}\right)\leq S(\rho||\sigma).$$ The rest follows because $\sigma_i/q_i$ is again a stabilizer state.
    \item By Lieb and Ruskai \cite{lieb_ruskai}, it is shown that quantum relative entropy is non-increasing under partial trace, i.e., $S(\Tr_B(\rho_{AB})||\Tr_B(\sigma_{AB}))\leq S(\rho||\sigma)$. 
    \item It is known that for any state $\tau$, $S(\rho\otimes \tau||\sigma\otimes \tau)=S(\rho||\sigma)$, hence the desired result follows when we take $\tau\in \stab_0^{(n)}$.
\end{enumerate}
\end{proof}
\subsection{(Anti-)Flatness}\label{sec:antiflatness}
Flatness is the property of a quantum state that describes how close its spectrum is to a flat spectrum. From the operational point of view, the flatness of a state describes how flat is the classical probability distribution over a basis of pure states in which we can decompose it. Of course, this does not imply that this state will return a flat probability distribution for the measurements in any other basis. As an example of flat states, both the completely mixed state and pure states possess flat spectrum. Another notable example \cite{flammia_topological_2009} are the ground states of string-net Hamiltonians, e.g. the toric code and its generalizations. 

Flat states are the free states for the resource theory of flatness. We thus define the null set as
\be
\operatorname{FLAT}^{(n)}:=\left\{\sigma\in\mathcal{H}\,|\, \sigma^2= \frac{\sigma}{\rank{\sigma}}\right\}
\ee
Let us now define the following measure of anti-flatness, that is, how far is a spectrum from the flat one. Of course, this quantity must measure the resource defined by $\operatorname{FLAT}^{(n)}$.
\begin{definition}\label{flatnessdist}
 We define the anti-flatness of $\psi_A$ as \cite{flatness}
\begin{equation}\label{flatnessdef}
  \mathcal{F}(\psi_A)=\Tr(\psi_A^3)-\Tr^2(\psi_A^2)
\end{equation}
This quantity is very natural as it can be defined classically as the variance of a probability distribution $p(x)$ according to the probability distribution itself. More concretely, if one defines $\langle x\rangle_p := \sum_x xp(x)$, and one defines $\Delta p^2 := \langle (p-\langle p\rangle_p)^2\rangle_p$, then one has 
\ba
\mathcal{F}(\psi_A)= \Delta \lambda^2
\ea 
with $\{\lambda\}\equiv\mbox{spec} [\psi_A]$.
Of course, this quantity is zero on the flat states, that is,
 $\mathcal{F}(\sigma)=0$ for $\sigma\in\operatorname{FLAT}^{(n)}$ as it is immediate to verify. 
\end{definition}

There is a profound connection between anti-flatness and magic. It connects magic, which is a property of the full state, to bipartite entanglement, and thus to the spectrum of a reduced density operator. In particular, it has been shown that\cite{flatness}, given a pure state $\psi_{AB}$ in a bipartite Hilbert space $\mathcal H = \mathcal H_A\otimes\mathcal H_B$, its linearized stabilizer entropy $M_{lin}$ is the average anti-flatness of $\psi_A$ on the Clifford orbit,  that is,
\begin{equation}\label{flatnessdth1}
\langle\mathcal F(\psi_A^C)\rangle_C =f(d_A,d_B) M_{lin} (\psi_{AB})
\end{equation}
where $\psi^C_A = \Tr_B \psi_{AB}^C \equiv \Tr_B (C\psi_{AB} C^{\dag})$. It is also true that anti-flatness shows typicality. Later, we will use this property to connect magic to spectral properties. The main message of \cref{flatnessdth1} is that, as long as the state $\psi$ is very entangled, and therefore $\psi_A$ is full rank, one can use the spectral quantity $\mathcal F(\psi_A)$ to probe magic. Note that - by definition - every density matrix is full rank on its support. 
This will come in handy in the next section.

 It is possible to define another monotone for the resource theory of flatness through the quantum relative entropy,
\be
\mathcal{F}_R(\rho)=\min_{\sigma \in \mathrm{FLAT}^{(n)}}S(\rho\Vert\sigma).
\ee
One can prove the following proposition
\begin{proposition}\label{prop:qrf}
Given a state $\rho\in\mathcal H$, it holds that
\begin{equation}
    \mathcal F_R(\rho)= S_{max}(\rho)-S(\rho)
\end{equation}  
\end{proposition}
See \cref{qrf} for a proof. 
Note that $\mathrm{FLAT}^{(n)} \supset \stab_0^{(n)}$ where $\stab_0^{(n)}$ is the set of states with zero stabilizer R\'enyi entropy, hence $\min_{\sigma\in \stab_0^{(n)}}S(\rho||\sigma)\geq F_R(\rho)$, therefore the flatness lower bounds the total subregion magic for any state. 
The same would not be true if $\stab^{(n)}$ is the usual stabilizer polytope, because it overlaps with $\mathrm{FLAT}^{(n)}$ but is not a subset as one can take a classical mixture of it such that the eigenvalues of $\rho$ are not equal (or zero).

Finally, let us define yet another flatness that will be natural for holography. Recall from \cite{Dong1} that a variant of the R\'enyi entropy is given by,

\begin{equation}
\tilde{S}_n(\rho)=n^2\partial_n\left(\frac{n-1}{n}S_n(\rho)\right)
    =-n^2\partial_n(\frac{\log\Tr(\rho^n)}{n}).
\end{equation}

If we rewrite $\Tr(\rho^n)$ in terms of the spectrum $\{\lambda_k\}$ of $\rho$, it becomes

\begin{equation}
    \tilde{S}_n(\rho)=-n^2\partial_n(\frac{\log(\sum_{k}\lambda_k^n)}{n})=\log(\sum_{k}\lambda_k^n)-n\frac{\sum_k\lambda_k^n\log\lambda_k}{\sum_k \lambda_k^n}.
\end{equation}

Now we take the derivative of this expression and obtain another definition of anti-flatness. In fact, this quantity is known as the \textit{Capacity of Entanglement}, which has been explored in the context of condensed matter system \cite{PhysRevLett.105.080501,PhysRevB.83.115322} and in quantum gravity \cite{PhysRevD.99.066012,Nakaguchi:2016zqi,Bueno:2022jbl,Zurek:2022xzl}.  

\begin{proposition}\label{def:braneflatness}
    $\partial_n\tilde{S}_n (\rho)$ is a measure of anti-flatness in that $\partial_n\tilde{S}_n (\rho)=0$ if and only if $\rho$ has a flat spectrum.
\end{proposition}
\begin{proof}
Expanding the definition using the set of eigenvalues of $\rho$.
    \begin{align}\label{eqn:holononflat}
    \partial_n\tilde{S}_n(\rho)=&-n\frac{(\sum_k\lambda_k^n\log^2\lambda_k)(\sum_l\lambda_l^n)-(\sum_k\lambda_k^n\log\lambda_k)^2}{(\sum_k\lambda_k^n)^2}\\
    =&-n\frac{(\sum_{kl}\lambda_k^n\lambda_l^n\log^2\lambda_k)-(\sum_{kl}\lambda_k^n\lambda_l^n\log\lambda_k\log\lambda_l)}{(\sum_k\lambda_k^n)^2}\\
    =&-n\frac{\sum_{(kl)}\lambda_k^n\lambda_l^n(\log^2\lambda_k+\log^2\lambda_l-2\log\lambda_k\log\lambda_l)}{(\sum_k\lambda_k^n)^2}\\
    =&-n\frac{\sum_{(kl)}\lambda_k^n\lambda_l^n\log^2\frac{\lambda_k}{\lambda_l}}{(\sum_k\lambda_k^n)^2},
\end{align}
where $\sum_{(kl)}$ denotes sum over each pair of distinct indices $k\neq l$. Note that each term in the numerator is non-negative. Therefore $\partial_n\tilde{S}_n=0$ if and only if $\log{\frac{\lambda_k}{\lambda_l}}=0$, which is equivalent to $\lambda_k=\lambda_i$ for all $k,l$. 
\end{proof}

This anti-flatness (\cref{eqn:holononflat}) can be connected to (\cref{flatnessdef})
by first noticing that the anti-flatness $\mathcal{F}(\rho)$ corresponds to the variance of $\rho$. The proof is straightforward
\ba
\mathcal{F}(\rho)&=\tr(\rho^3)-\tr^2(\rho^2) &=\tr(\rho \, \rho^2) -\tr^2(\rho \, \rho)
=\langle\rho^2\rangle_\rho - \langle \rho \rangle_\rho^2 = \operatorname{Var}_\rho(\rho)
\ea
Let us connect this definition with the derivative at $n=1$. 
Let $\rho\equiv\sum_k\lambda_k \ket{\lambda_k}\bra{\lambda_k}$.
Note that the following relation can also be written as a variance, by defining $p_k=\frac{\lambda_k^n}{\sum_k\lambda_k^n}$, it is easy to observe that $\sum p_k=1$ and we can define the state 
\ba
\Xi:=\sum_k p_k \ket{\lambda_k}\bra{\lambda_k}
\ea
 and so 
\ba
\partial_n\tilde{S}_n(\rho)&=&-n \sum_{kl}p_k p_l (\log^2 \lambda_k-\log\lambda_k\log\lambda_l) \\
&=&-n\langle\log^2\rho\rangle_{\Xi}+n\langle\log\rho\rangle_{\Xi}^2=-n\operatorname{Var}_\Xi(\log\rho)
\ea
Let us compute it for $n=1$
\ba \label{sigmalogrho}
\left.{\partial_n \tilde{S}_n}(\rho)\right\vert_{n=1}&=& -\sum_{kl}\lambda_k\lambda_l \log\lambda_k(\log\frac{\lambda_k}{\lambda_l})\\ 
&=&-\sum_k \lambda_k \log^2 \lambda_k + \sum_{kl} \lambda_k\lambda_l\log\lambda_k\log\lambda_l \\
&=&-\tr(\rho\log^2\rho)+\tr^2(\rho\log\rho)\\
&=&-\langle\log^2\rho\rangle_\rho+\langle\log\rho\rangle_\rho^2=-\operatorname{Var}_{\rho}(\log\rho)
\ea
Interestingly, when $n=1$, $\Xi$ coincides with $\rho$. Seeing $\log\rho$ as a function of $\rho$, the variances between the two quantities are connected. We make use of standard techniques of error propagation to get the relationship between $\operatorname{Var}_\rho(\rho)$ and $\operatorname{Var}_\rho(\log(\rho))$. 
\be \label{eqn:approxrenyid}
\operatorname{Var}_\rho(\log(\rho))\approx\frac{\operatorname{Var}_\rho(\rho)}{\langle \rho \rangle_\rho^2}=\frac{\operatorname{Var}_\rho(\rho)}{\operatorname{Pur}(\rho)^2}=\frac{\mathcal{F}(\rho)}{\operatorname{Pur}^2(\rho)},
\ee 
The approximation is valid when $S_0(\rho)-S_2(\rho)<\log 2$. Therefore, the two measures coincide in the near-flat or weak entanglement regime. 

In fact,~\cref{eqn:holononflat} has a convenient rewriting as the variance of the modular Hamiltonian. Given a state $\rho\equiv \sum_k \lambda_k\ket{\lambda_k}\bra{\lambda_k}$, its eigenvalues can be written as $\lambda_k:=\exp(-\beta E_k)$ where $\beta$ is an effective temperature. Note that since $\sum_k \lambda_k=1$ then $Z[\beta]=1$. This also defines the (entanglement) Hamiltonian
\ba
H=\sum_k E_k\ketbra{\lambda_k}{\lambda_k}
\ea
From~(\cref{sigmalogrho}), we get 
\ba
\left.{\partial_n \tilde{S}_n}(\rho)\right\vert_{n=1}&=&-\sum_{k}\exp(-\beta E_k) (-\beta E_k)^2 + \left(\sum_k \exp(-\beta E_k)(-\beta E_k)\right)^2\\
 &=&-\beta^2\left[\langle H^2 \rangle_\beta -\langle H\rangle_\beta^2 \right]
\ea
also, with simple algebra, one obtains that 
\ba
\left.{\partial_n \tilde{S}_n}(\rho)\right\vert_{n=1}&=&
-\beta^2\left[\langle H^2 \rangle_\beta -\langle H\rangle_\beta^2 \right]= -\beta^2\langle (E_k-E_l)^2\rangle_{kl} 
\ea
In other words, the modified R\'enyi entropy is proportional - by inverse temperature - to the fluctuations of the Hamiltonian $H$ which in turn is the average gap squared in the energies $E_k$. This also shows why the derivative is connected to the anti-flatness of the state. 

This result can be extended to any $n$, to do this, note first that $\Xi=\exp(-n\beta H)Z^{-1}[n\beta]$. Then with some algebra, we obtain
\ba
{\partial_n \tilde{S}_n}(\rho)=-n^3\beta^2 (\langle H^2 \rangle_{n\beta}-\langle H\rangle_{n \beta }^2).
\ea

\subsection{Non-local Magic, Entropy, and anti-Flatness}\label{sec:magicbounds}
In this section, we are going to introduce the concept of \emph{non-local magic}, and how it relates to both entanglement and anti-flatness. 

\begin{definition}[Multi-partite non-local magic]
Given $M$ a measure of magic and $\psi_{A_1\dots A_n}\equiv\st{\psi _{A_1\dots A_n}}$ a pure state, we define as $n$-partite non-local magic
\begin{equation}
		M^{(n-\rm NL)} (\psi_{A_1\dots A_n}) := \min_{U=\otimes_{i=1}^n U_{A_i}} M(U \psi_{A_1\dots A_n} U^{\dagger}).
  \label{eq:nlmagic}
\end{equation}
\end{definition}
As we exclusively discuss the case of bipartite non-local magic when $n=2$ for the rest of this work, we set $A=A_1, B=A_2$ and simply refer to $M^{(NL)}=M^{(2-NL)}$ as non-local magic for convenience. 

Intuitively, non-local magic is the non-stabilizerness that lives in the correlation between $A$ and $B$ because $U_A\otimes U_B$ removes all ``local'' magic in $A$ or $B$ separately. This is distinct from other notions of long-range magic \cite{white_conformal_2021,Bao_2022,tarabunga2023critical}.  Note that $A,B$ themselves can be multi-qubit systems, so $U_A,U_B$ need not be single qubit unitaries.

In this work, we will use as measures of magic $M_{\text{dist}}$ and the two relative entropies of magic $M_{R},\ M_{SR}$. 

\subsubsection{Non-local  magic and flatness}
Let us start with a general relation valid for \textit{any} measure of anti-flatness and \textit{any} measure of non-local magic. 

\begin{lemma}\label{lemmaNL}
A pure quantum state $\ket{\psi}$ possesses no non-local magic, that is,  $M^{NL}(\ket{\psi})=0$, iff $\ket{\psi}$ is unitarily locally equivalent to a state $\ket{\psi'} = U_A\otimes U_B\ket{\psi}$ with flat reduced density matrix $\psi_A'\equiv\tr_B\st{\psi'}$ with integer R\'enyi entropies\footnote{In this work, information is measured using bits. Accordingly, entropies are computed using $\log_d$.}. In formulae,
\be
M_{NL}(\ket{\psi})=0\iff 
F(\psi_A)=0 \wedge \rank(\psi_A)=d^{r_A},\,\,r_A\in\mathbb{N} 
\ee
\begin{proof}
     Let us start from the left-to-right implication. We employ the fact that any faithful measure of magic $M(\ket{\psi})$ vanishes on the free states.
For any such measure, its non-local counterpart with respect to the bipartition $A|B$ is $M^{NL}(\ket{\psi}):= \min_{U_A\otimes U_B}M(U_A\otimes U_B\ket{\psi})$. Given $M^{(NL)}(\ket{\psi})=0$, then we know that there exist a bi-local unitary $U_A\otimes U_B$ such that $\ket{\psi^{\prime}}\equiv U_A\otimes U_B\ket{\psi}\in \stab_0^{(n)}$. Since $\stab_0^{(n)}$ is closed under partial trace, see \cref{sec:magicmeasure}, then $\psi_A'\in\stab_0^{(n)}$. We know that $\psi_A'\in\mathrm{FLAT}^{(n)}$. Moreover, being $\psi_A'\in\stab_0^{(n)}$ we know that $\rank(\psi_A)=d^{r_A}$ with $r_A\in\mathbb{N}$. Let us now show that also the converse is true.  Consider a flat state $\ket{\psi}$, that is, a state such that its reduced density matrix $\psi_A\equiv\tr_B\ketbra{\psi}{\psi}=\frac{1}{d^{r_A}}\sum_{i}\ketbra{\phi_i}{\phi_i}_A$ where the sum run on $d^{r_A}$ many rank-one projectors $\ketbra{\phi_i}{\phi_i}_A$. Note that we exploited the fact that $S_{\alpha}(A)=r_A\in\mathbb{N}$ for every $\alpha\in[0,\infty)$. Via the Schmidt decomposition, we can write the state as $\ket{\psi}=\sum_{i}\frac{1}{\sqrt{d^{r_A}}}\ket{\phi_i}_A\otimes \ket{\psi_i}_B$. Without loss of generality, we choose now $|A|<|B|$.  We further know that $\langle \phi_i|\phi_j\rangle=\langle \psi_i|\psi_j\rangle=\delta_{ij}$. Now  choose $U_A$ (resp. $U_B$) such that $U_{A}\ket{\phi_i}_A=\ket{i}_A$ (resp. $U_{B}\ket{\psi_i}_B=\ket{i}_B$) for $\ket{i}_A$ (resp. $\ket{i}_B$) being the computational basis on $A$ (resp. $B$). We obtain 
\ba\label{magicflatness}
U_A\otimes U_B\ket{\psi}=\sum_{i}\frac{1}{\sqrt{d^{r_A}}}\ket{i}_A\otimes \ket{i}_B\equiv \ket{EPR}_{A\bar{A}}\otimes \ket{j}_{B\setminus \bar{A}}
\ea
where $\ket{EPR}_{A\bar{A}}$ is a EPR pair between the full $A$ and \textit{any} subsystem $\bar{A}\subset B$ such that $|A|=|\bar{A}|$, while $\ket{j}$ is a computational basis state on $B\setminus \bar{A}$.    Since  $\ket{EPR}_{A\bar{A}}\otimes \ket{j}_{B\setminus \bar{A}}$ is a stabilizer state, we obtain
\ba
0=M(U_A\otimes U_B\ket{\psi})\ge \min_{U_A\otimes U_B}M(U_A\otimes U_B\ket{\psi})=M^{NL}(\ket{\psi})\ge 0
\ea
  \end{proof}
\end{lemma}
Notice that a vanishing non-local magic is a sufficient  condition for anti-flatness to be zero. However, there are possibly states with non-integer R\'enyi entropies that can possess some non-local magic without being guaranteed that anti-flatness is non-vanishing. With the additional condition of integer R\'enyi entropy, also the other implication holds, that is, a flat state implies vanishing non local magic for any sensible measure of non local magic.

We now show lower and upper bounds to the non-local magic based on the trace distance of $M_{dist}$ defined in~\cref{def:mod_tr_dist_mag}. As we saw previously, anti-flatness connects magic and entanglement. 
We have:
\begin{theorem}\label{th:magicdist} Let $\psi_{AB}$ be a pure state  in a bipartite Hilbert space $\mathcal H = \mathcal H_A\otimes\mathcal H_B$,
then
	\begin{equation}
		\mathcal{F}(\psi_{ A})/8\le M_{dist}^{(NL)}(\psi_{AB})\le  \sqrt{1-e^{S_{max}(A)}+e^{S_{\infty}(A)}\left(1-\frac{e^{\log (d) \lfloor S_{max}(A)/\log d \rfloor}}{e^{S_{max}(A)}}\right)}
	\end{equation}
\end{theorem}
where $\lfloor \cdot \rfloor$, is the floor function. The proof can be found in the~\cref{app:proofthmd}.

 As we shall see in \cref{section:brane}, the \cref{lemmaNL} and \cref{th:magicdist} will have important consequences for the relationship between the non-local magic in the CFT side and gravity in AdS.

\subsubsection{Non-local stabilizer relative entropies of magic}

In this section, 
we show that similar to the trace distance of magic, the relative stabilizer entropy also has a tight connection with flatness and entanglement.

\begin{theorem}\label{th:relstab}
Let $\psi_{AB}$ be a pure state, then 
    \begin{equation}\label{eq:relstab}
        S_{max}(A)-S(A)=\mathcal F_R(\psi_A) \leq \min_{U_A}M_{RS}(U_A\psi_A U_A^{\dagger})\leq M^{(NL)}_{RS}(\psi_{AB})\leq\log d \lceil   S_{max}(A)/\log d  \rceil.
    \end{equation}
\end{theorem}
Here $S(A)=S(\rho_A)$ and $S_{max}(A)=S_{max}(\rho_A)$. The proof can be found in~\cref{proofth2}.
Let us briefly comment on the tightness of the bound. It is clear that when $|\psi\rangle_{AB}$ has a dominant Schmidt coefficient and many small trailing singular values, then the bound is essentially tight. A case in point is $\sqrt{1-\epsilon}|00\rangle+\sqrt{\epsilon}|11\rangle$. However, the upper bound is quite loose for states with near-flat spectrum, e.g. $\epsilon=1/2$. This is an artifact of choosing the maximally mixed state as a reference even though other stabilizer states clearly yield a lower distance.

A similar upper bound can be obtained with the usual relative entropy measure of magic. 
\begin{proposition}[Entanglement upper bounds NL magic]\label{prop:RelativeEnt}
    Suppose $\rho_{AB}$ is pure, and \begin{align}
        M^{(NL)}_R(\rho_{AB})=\min_{U_A\otimes U_B} M_R((U_A\otimes U_B)\rho_{AB}(U_A\otimes U_B)^{\dagger}),
    \end{align} then $M_R^{(NL)}(\rho_{AB})\leq S(A)=S(B)$, where $S(A)$ is the von Neumann entropy of subsystem $A$.
\end{proposition}
The proof is given in~\cref{proofprop4}. This upper bound suffers from the same drawbacks as (\cref{eq:relstab}) for states that are maximally entangled.

\subsubsection{Magic estimates}\label{section:estimate}
As minimization can be difficult for the relative entropy measure, let's also derive a tighter upper bound based on a computable measure of magic, that is, the stabilizer R\'enyi entropy. 
To do so, we can pick a good estimate that is reasonably close to the minimum. Suppose the entanglement spectrum of the state under the same bipartition $AB$ is $\{\lambda_i\}$, construct a state 
\begin{equation}
    |\psi'\rangle_{AB}=\sum_{i=0}^{2^n-1}\sqrt{\lambda_i} |s_i\rangle|s_i\rangle,
\end{equation}
where $\{|s_i\rangle\}$ are eigenstates of a stabilizer group $\mathcal{S}=\{S_1,S_2, \cdots, S_n\}$ such that for any $S_k$ in $\mathcal{S}$, $S_k\ket{s_i}=\pm\ket{s_i}$. Because the entanglement spectrum is invariant under local unitary $U_A\otimes U_B$, $|\psi'\rangle$ is a reasonable construction such that the reduced density matrix on both $A$ and $B$ are within the \textit{stabilizer polytope}, and hence have vanishing local magic by the relative entropy measure $M_R$.    Note that other choices of the Schmidt basis may yield lower overall magic on $AB$, therefore $M(|\psi'\rangle)$ provides an upper bound of non-local magic.

We now present an estimate of $M(|\psi'\rangle)$ using the Stabilizer R\'enyi Entropy measure.



\begin{proposition}
The non-local stabilizer R\'enyi entropy estimate for a state with entanglement spectrum $\{\lambda_i\}$ is 
    \begin{equation}\label{estimate_prop}
        \mathcal{M}_2 (\{\lambda_i\})=\mathcal{M}_2(\sum_{i=0}^{2^n-1}\sqrt{\lambda_i}|s_i\rangle|s_i\rangle), \qquad \lambda_i\geq\lambda_j,\ \ \text{for} \ i<j.
    \end{equation}
\end{proposition}
Note that this non-local magic estimate doesn't depend on the choice of stabilizer group $\mathcal{S}$. However, the ordering of eigenvalues does affect its magnitude. Remarkably, one can obtain an exact expression for $\mathcal{M}_2(\{\lambda_i\})$. {A similar expression has also been obtained by \cite{RK_wavefcn} but in a different context. With additional ancillae, it is identical to the one below after applying a global Clifford unitary. }

 \begin{theorem}\label{thm:nlSRE}
     The non-local stabilizer Rényi entropy estimate is 
     \begin{equation}\label{eq:analyticalM}
    \mathcal{M}_2(\{\lambda_i\})=-\log\left(\sum_{i_1,i_2,i_3,i_4=0}^{2^n-1}\sqrt{\lambda_{i_1}\lambda_{i_2}\lambda_{i_3}\lambda_{i_4}\lambda_{i_3\wedge i_2\wedge i_1}\lambda_{i_4\wedge i_2\wedge i_1}\lambda_{i_1\wedge i_3\wedge i_4}\lambda_{i_2\wedge i_3\wedge i_4}}\right),
     \end{equation}
where $\wedge$ denotes the bitwise XOR operation. This expression depends on the ordering of the eigenvalues and reaches its minimum when the eigenvalues are in the descending order, that is, 
$\lambda_i\geq\lambda_j$ for  $i<j$.
 \end{theorem}

In \cref{sec:numerics} we present numerical results of $\mathcal{M}_2(\{\lambda_i\})$ for finite-sized physical system. It is helpful to see that the estimate constitutes a non-local magic upper bound. 

\begin{corollary}\label{th:srebound}
Let $\{\lambda_i\}$ be the Schmidt values for $|\psi\rangle_{AB}$ when bipartitioning the system into $A$ and $B$. The non-local stabilizer Rényi entropy  is upper bounded by 
\begin{equation}
\mathcal{M}_2^{NL}(|\psi\rangle_{AB})\leq\mathcal {M}_{2}(\{\lambda_i\})\leq \min\{2S_2(A),4(S_{\text{max}}(A)-S_{1/2}(A))\}
\label{eqn:avgnlsre}
\end{equation}
where $S_{max}(A)=n\log 2$, $S_{\alpha}(A) = S_{\alpha}(\rho_A)$ with $\rho_A=\Tr_{B}[|\psi\rangle\langle\psi|]$. 
\end{corollary}

See~\cref{app:estimate} for the proof. Based on this result, we discuss two regimes. One is when the spectrum is almost flat. In this regime, the bipartite non-local magic is upper bounded by,
\begin{equation}
    \mathcal{M}_2(\{\lambda_i\})\leq 4(S_{\text{max}}(A)-S_{1/2}(A)).
\end{equation}
This has the interpretation as anti-flatness. 
Although the measure of magic is different, we see that this gives a much tighter bound compared to (\cref{th:relstab}) in the near-flat regime.



\begin{remark}\label{rmk:1}
    Haar random states have small bipartite non-local magic. 
\end{remark}
We see that $\mathcal{M}_2\sim S_0-S_{1/2}$ whereas the lower bound from relative stabilizer entropy measure in (\cref{th:relstab}) is $S_0-S_1$, both are bounded by a constant for Haar random states\cite{avgrenyi} --- for small $\alpha$, $\dim A\ll \dim B = m$, $S_0-S_{\alpha}\leq \frac{\alpha}{2} +O(1/m^2)$.
This is somewhat surprising because Haar random states are magic rich and have non-trivial total magic~\cite{white2020mana,Liu_2022}. However, the magic sustained by their bipartite entanglement is small even though local magic in any subregion $A$ with $|A|\gg |B|$ can be large.  

Another limit is when $S_0(A)\gg S_{1/2}(A)$, which applies for quantum field theory. In this regime the magic is approximated by the second R\'enyi entropy, 
\begin{equation}
\mathcal{M}_2(\{\lambda_i\})\leq 2S_2(A). 
\end{equation}
As we shall see in~\cref{sec:CFT}, this is consistent with our MERA intuition for conformal field theories.


\subsubsection{Smoothed magic}
The concept of magic and its bound, as discussed earlier, are applicable to systems with finite dimensions. However, in quantum field theory, the Hilbert space has an infinite dimension. In this case, the bounds given by max entropy in~\cref{th:relstab} can easily be divergent. To produce a non-trivial bound, it is imperative to introduce the `smoothed magic', defined as
\begin{equation}
    M_{RS}^{\epsilon}(\rho):= \min_{\Vert\chi-\rho\Vert<\epsilon}M_{RS}(\chi),
\end{equation}
as well as the `smoothed non-local magic', defined as
\begin{equation}
\sma{\rho_{AB}}:= \min_{\Vert\chi_{AB}-\rho_{AB}\Vert<\epsilon}M_{RS}^{(NL)}(\chi_{AB}).
\end{equation}

For this, a smoothed version of~\cref{th:relstab} holds.

\begin{theorem}\label{th:smoothed}
Let $\rho_{AB}$ be a pure state, then 
\begin{equation}
     S_{max}^{\epsilon}(\rho_A)-(1-\epsilon)^{-1}S(\rho_A)\leq \sma{\rho_{AB}}\leq \log d \lceil   S_{max}(A)^{\epsilon}/\log d  \rceil.
     \label{eqn:smoothinequality}
\end{equation}
\end{theorem}
where the smoothed maximal entropy is defined as
\begin{equation}
    S_{max}^{\epsilon}(\rho)=\min_{\Vert\chi-\rho\Vert<\epsilon} \ln(\rank(\chi)). 
\end{equation}

The proof can be found in~\cref{proofth3}. As stated by~\cref{th:smoothed}, the magic is bounded from below by the difference between the smoothed maximal entropy and the entanglement entropy which is finite for conformal field theories.


Before we discuss CFTs, let's examine the physical meaning of the lower bound, which is the difference between smoothed max entropy and the von Neumann entropy. In addition to the anti-flatness of the entanglement spectrum, this quantifies the \textit{compressibility} of a state~\cite{Akers_2021}. Consider a bipartition of the state followed by a Schmidt decomposition. It is compressible if we can still well approximate it after truncating the less significant singular values, as one is wont to do in DMRG. Here we can show that this compressibility gap which lower bounds smoothed non-local magic also quantifies the classical hardness in simulations. 

Let us build up the following argument by recalling that there are states such as random stabilizer states that have high entanglement but are classically easy to simulate. Since magic and entanglement capture two orthogonal perspectives of quantumness, are there quantum states with low entanglement but high magic that are classically hard to simulate? Na\"ively, a state with high magic will have high stabilizer rank, which is hard in the stabilizer simulation. On the other hand, the system will be classically hard using the tensor network method if it has high bond dimensions. However, a folk theorem in tensor network suggests that the small entanglement would permit one to capture the state with a tensor network whose bond dimension only needs scale as $O(e^S)$ where $S$ is the von Neumann entropy of each subsystem. Therefore, it seems that as long as the entanglement is small, there should be a classically easy description. However, one needs to be careful in applying this lore as it is known that there exist states with low entanglement but classically complex\cite{ge_area_2016}. 

More precisely, consider an exact MPS description of a state with low entanglement such that for any subsystem $A$, $S(A)\ll \log \rank(\rho_A)$ where we have taken the bond dimension $\chi$ to be sufficiently large to reproduce the state exactly. One would be tempted to truncate the singular values and only keep $O(e^S)$ as suggested by the folk theorem. However, we note that this truncation is only justified if there exists $\sigma_A$ with  $||\sigma_A-\rho_A||<\epsilon$ such that  $$\Delta S_{\epsilon}(A)=S_{\rm max}^{\epsilon}(A)-S(A)=\log \rank(\sigma_A)-S(A)$$ is small compared to $S(A)$. In other words, the state is (perfectly) compressible. Such is indeed true for conformal field theory ground states, where $\Delta S_{\epsilon}\sim \sqrt{S\log(1/\epsilon)}$. However, this is not true in general. For example, consider a state $\ket{\psi}=\frac{1}{\mathcal{N}}\sum_{i=1}^r\frac{1}{\sqrt{i}}\ket{i}_{A}\ket{i}_{B}$. The smoothed max entropy $S_{max}^{\epsilon}=\log{r}-\epsilon$, while entanglement entropy is nearly half of it, $S\approx \frac{1}{2}\log{r}$. In holography, \cite{Akers_2021} argued that certain state mixtures, such as that of a thermal and pure state, can lead to an arbitrarily large $\Delta S_{\epsilon}(A)$. 

Therefore, high incompressibility on the one hand forces high tensor network bond dimension, and on the other necessitates high non-local magic from \cref{th:smoothed}. This implies that such states will be classically hard to simulate and sharpens a general empirical observation that relates magic to classical complexity. Furthermore, if $S\ll\Delta S_{\epsilon}\approx S_{\rm max}^{\epsilon}$, then both the lower and upper bounds are approximately saturated. In this case, the smoothed non-local magic provides a quantitative measure for the classical hardness of simulating such states. Treating magic as roughly as the log of stabilizer rank and bond dimension, one would expect that classical resource of order $O(\exp(M_{RS}^{(NL,\epsilon)}))$ will be needed. It then follows that \textit{simulating such incompressible states is classical hard} using not only the tensor network method but also the stabilizer and the Monte Carlo method\cite{magicMC} by having large magic\footnote{A careful treatment of this problem should include other formulations of non-magical processes like Gaussian states, matchgates with \cite{bu_stabilizer_2023}.}.

\section{Magic in conformal field theories}\label{sec:CFT}
Having seen a quantitative connection between anti-flatness in entanglement spectrum and non-local magic, we examine these relations in the context of CFTs.

\subsection{Geometric Interpretation through tensor networks}\label{section:MERA}

To figure out (1) how much non-local magic there is in a CFT and (2) how such magic connected with the anti-flatness of the entanglement spectrum, it is instructive to first look at an intuitive picture from tensor networks.
For CFTs with small central charges, MERAs have been shown to be good approximations of CFT ground states $|\psi\rangle_{AB}$. By extension, it also holds for products of CFTs with small central charge. Let us assume that the tensor network structure remains valid for arbitrary degree of accuracy, perhaps at the cost of increasing the bond dimension, which is supported by empirical observations. Using this as a heuristic, we deduce that local unitary deformations $U_A\otimes U_B$ unitarily ``distills'' an entangled state\footnote{For simplicity, we will refer to such a process as distillation from now on. However, one should note that it is distinct from the usual entanglement distillation of perfect Bell pairs unless otherwise specified.} between $A$ and $B$ with log Schmidt rank that is upper bounded by the number of edge cuts (green triangle~\cref{fig:mera}). As such cuts scale linearly with the size of the RT surface, i.e. the boundary of the triangle in the bulk, the log of Schmidt rank must be bounded by the number of edge cuts which scale the same way as entanglement entropy in this case. This implies that the non-local Magic in CFTs should scale linearly with the area of the Ryu-Takayanagi surface.

In fact, we can almost identify the optimal distilled state that has the same Schmidt rank but removes the unnecessary zero eigenvalues by just acting mostly unitaries and disentanglers. Let the blue rectangles at the bottom layer be the CFT ground state but at a more coarse-grained scale. As the ground state is an IR fixed point, we can simply use it as an input in the MERA to generate the more fine-grained state on the top layer. We can decompose the IR state by Schmidt decomposition, and the Schmidt rank is upper bounded by the bond dimension (here the bond is represented as 3 edges on each side of the blue rectangle on the bottom assuming the worst case volume law upper bound in the central region). By acting disentanglers and isometries in B followed by global unitaries on the subsystems represented by the blue rectangles on two sides of the bottom layer in the the IR ground state, we ``pushed'' the subregion $B$ on the top layer to the red boundary by acting $U_B$, which now lives on $\partial B$. Similarly, acting $U_A$ by running unitaries and isometries in $A$, we remove the bulk dof and push $A$ to $\partial A$, marked by the orange lines. The qubits on $\partial A$ and $\partial B$ are entangled and their entanglement spectrum is unchanged since we only applied unitaries $U_A\otimes U_B$. 

Let $|\partial A|,|\partial B|$ be the number of edges in $\partial A,\partial B$. The distilled state $|\chi\rangle_{AB}$ is not optimal as $ \log rank(\rho_A)\leq |\partial A|<|\partial B|$, where we would have hoped that $|\partial A|=|\partial B|=\log rank(\rho_A)$, but this is close enough as $|\partial A|$ and $|\partial B|$ both scale as $\sim \log |A|$ as the blue region that contributed to suboptimality in the edge cuts is only constant (AdS) radius away from the true minimal surface. The number of edge cuts on the bottom layer is always bounded as the width of the MERA past causal cone is bounded. This means that $|\partial A|+const = |\partial B|$ where the constant depends on the network discretization. For binary MERA it stabilizes at 4 to 6 sites. 

As a consequence, after the removal of local magic in each wedge, the remaining magic is tied up into the interface between $A$ and $B$ marked by the region shaded in blue. Since the amount of magic generically scales linearly with the number of tensors, for a contiguous subregion $A$, the size of the interface region scales as $\log |A|$, which is proportional to the size of the RT surface up to subleading corrections. Note that while it may be possible to lower the size of this interface region further by local unitary transformations, the number of sites it involves must be lower bounded by the minimum number of edges connecting $A$ and $B$, which is given by $|\partial A|$. Heuristically, consider a case where all the bipartite entanglement between $A$ and $B$ have been ``distilled'' into imperfect Bell pairs connecting the two complementary regions. Then for any additive measure of magic, the non-local magic should scale linearly with the number of such imperfect Bell states, which is again proportional to the length of the minimal surface.

\begin{figure}
    \centering
    \includegraphics[width=\linewidth]{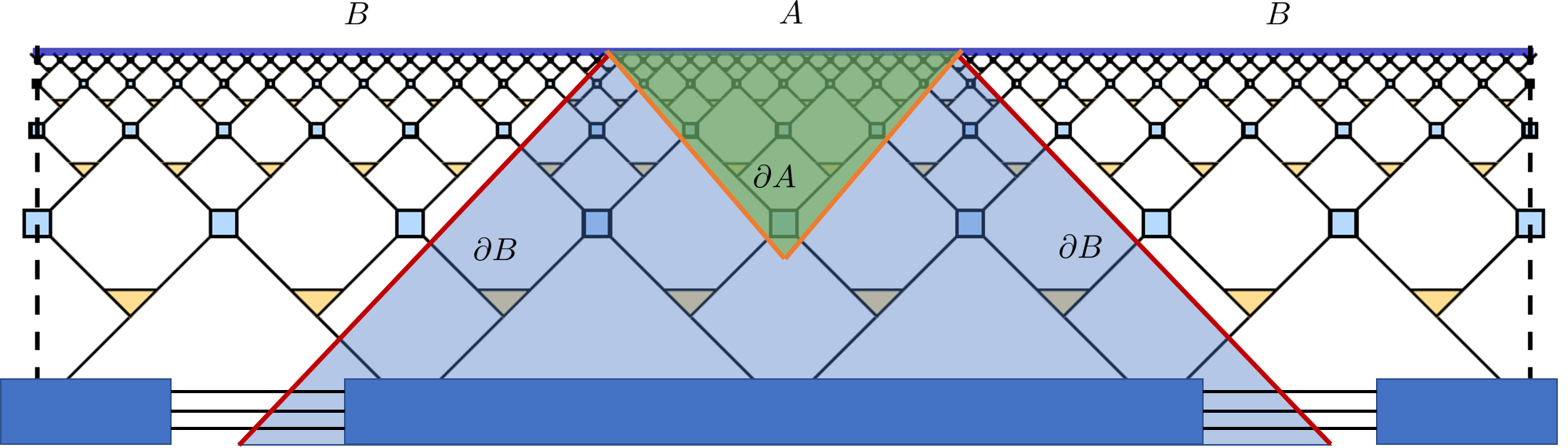}
    \caption{Green: past causal domain of dependence of $A$, Union of blue and green: past causal cone of $A$. Time runs upwards.}
    \label{fig:mera}
\end{figure}

More precisely, we observe that the tensor network of the interface region is a matrix product state (MPS) (\cref{figD1b}) by removing the local unitaries. The remaining structure contributes to the non-local magic is shown in \cref{figD1a}. Each matrix in the chain consists of two isometries and one disentangler. 

\begin{equation}\label{eq:MPS}
\ket{\chi}_{AB}=M_1^{(s_1r_1)}M_2^{(s_2r_2)}\cdots M_n^{(s_nr_n)}\ket{s_1s_2\cdots s_n}_{\partial A}\ket{r_1r_2\cdots r_n}_{\partial B}.
\end{equation}

We expect the magic of this state to scale linearly with the number of matrices, namely the size of the light-cone, $\min\{|\partial A|, |\partial B|\}$. Indeed, we verify that magic scales as volume of the MPS, which is $\sim \log|A|$.

\tikzset{pics/matrix/.style args={#1#2#3}{
    code={
       \def\s{0.4}
    \draw[fill=orange] (-#1,0)--(-#1/2,-\s)--(-#1*1.5,-\s)--cycle;
   \draw[fill=orange] (#1,0)--(#1/2,-\s)--(#1*1.5,-\s)--cycle;
  \draw[fill=green] (-#1/2,-2*\s)--(#1/2,-2*\s)--(#1/2,-3*\s)--(-#1/2,-3*\s)--cycle;
  \draw[thick] (-#1/2,-\s)--(-#1/2,-2*\s);
  \draw[thick] (#1/2,-\s)--(#1/2,-2*\s);
  \draw[thick] (-#1*1.5,-\s)--(-#1*1.5,-2*\s);
  \draw[thick] (#1*1.5,-\s)--(#1*1.5,-2*\s);
  \draw[thick] (#1,0)--(#1,\s);
    \draw[thick] (-#1,0)--(-#1,\s);
    \draw[thick] (#1/2,-3*\s)--(#1/2,-4*\s);
    \draw[thick] (-#1/2,-3*\s)--(-#1/2,-4*\s);
    \draw (-#1*1.5,-2*\s) node[below]{#2};
    \draw (#1*1.5,-2*\s) node[below]{#3};
    }
},
    pics/square/.style args={#1#2}{
        code={
            \draw[fill=yellow] (-0.3,0) rectangle (0.3,-0.6);
            \draw[thick] (0,0)--(0,0.5);
            \draw[thick] (0,-0.6)--(0,-1.1);
            \draw[thick] (-0.3,-0.3)--(-1.1,-0.3);
            \draw[thick] (0.3,-0.3)--(1.1,-0.3);
            \draw (-1.1,-0.3) node[left] {#1};
            \draw (1.1,-0.3) node[right] {#2};
        }
    }
}

\begin{figure}[ht]
    \centering
    \begin{subfigure}[b]{0.4\textwidth}
        \begin{tikzpicture}
        \def\s{0.4}
        \pic at (0,0) {matrix={2}{$i_4$}{$j_4$}};
        \pic at (0,-4*\s) {matrix={1}{$i_3$}{$j_3$}};
        \pic at (0,-8*\s) {matrix={0.5}{$i_2$}{$j_2$}};
        \draw (-0.5/2,-12*\s) node[below]{$i_1$};
        \draw (0.5/2,-12*\s) node[below]{$j_1$};
    \end{tikzpicture}
    \caption{}
    \label{figD1a}
    \end{subfigure}
    \hfill
    \begin{subfigure}[b]{0.4\textwidth}
        \centering 
        \begin{tikzpicture}
          \pic at (0,0) {square={$i_1$}{$j_1$}};
          \pic at (0,1.1) {square={$i_2$}{$j_2$}};
          \pic at (0,2.2) {square={$i_3$}{$j_3$}};
          \pic at (0,3.3) {square={$i_4$}{$j_4$}};
        \end{tikzpicture}
        \caption{}
        \label{figD1b}
    \end{subfigure}
    \caption{The MERA tensor network with local unitaries removed produces a tensor network (a) that contributes to non-local magic. It can be written as an MPS (b) for which its stabilizer R\'enyi entropy can be computed numerically.}
    \label{fig:MPS}
\end{figure}
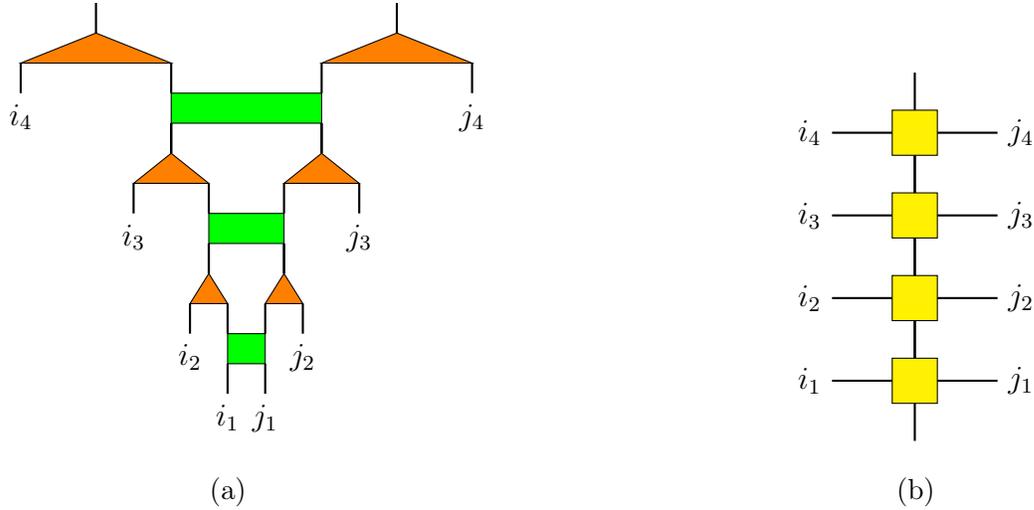

For the numerics, we pick a random realization of the disentangler and isometry and use them for each each layer, in accordance of the scaling invariance. Then we present two estimations of the non-local magic of this MPS state. The first estimation we calculate the lower bound of the stabilizer relative entropy, given in \eqref{eq:relstab}. We present the result in~\cref{figD2a}. Both max entropy and the von Neumann entropy scale linearly with the number of matrices, and thus linearly with respect to the RT surface area and the entanglement entropy $S(A)$ of the boundary theory of subregion $A$. 
In the second estimation we calculate the entanglement spectrum of this state, denoting the set of eigenvalues as $\{\lambda_i\}$. Then we construct a state with the same entanglement spectrum and compute its non-local magic estimate using~\ref{estimate_prop}. 

\begin{figure}[H]
    \centering
    \begin{subfigure}[b]{0.4\textwidth}
    \scalebox{0.75}{
    \begin{tikzpicture}
        \begin{axis}[
            xlabel={$|\partial A|$}, 
            ylabel={$S$}, 
            tick align=inside, 
            legend style={at={(0.2,0.9)},anchor=north}, 
            grid={both}
        ]
        \addplot[
            color=orange,
            mark=*, 
            smooth 
        ] coordinates {
            (3,2.079)
            (4,2.773)
            (5,3.466)
            (6,4.159)
            (7,4.852)
            (8,5.545)
            (9,6.238)
            (10,6.931)
            (11,7.624)
            (12,8.317)
        };
        \addlegendentry{\(S_{\text{max}}\)}
        \addplot[
            color=blue,
            mark=*, 
            smooth 
        ] coordinates {
            (3, 0.67)
             (4, 0.767)
             (5, 0.883)
             (6, 0.965)
             (7, 1.057)
             (8, 1.152)
             (9, 1.243)
             (10, 1.339)
             (11, 1.434)
             (12, 1.528)
        };
        \addlegendentry{\(S_{max}-S\)}
        \end{axis}
    \end{tikzpicture}
    }
    \caption{}
    \label{figD2a}
\end{subfigure}
\hfill
\begin{subfigure}[b]{0.4\textwidth}
    \centering
    \scalebox{0.75}{
    \begin{tikzpicture}
        \begin{axis}[
            xlabel={$|\partial A|$}, 
            ylabel={$\mathcal{M}_2$}, 
            tick align=inside, 
            legend style={at={(0.2,0.9)},anchor=north}, 
            grid={both}
        ]
        \addplot[
            color=blue,
            mark=*, 
            smooth 
        ] coordinates {
            (3,0.752299)
            (4,1.04593)
            (5,1.50105)
            (6,2.13376)
            (7,2.603420)
            (8,3.086418)
            (9,3.481793)
        };
        \addlegendentry{\(\mathcal{M}_{2}\)}
        \end{axis}
    \end{tikzpicture}
    }
    \caption{}
    \label{figD2b}
\end{subfigure}
\caption{(a) Maximal entropy and von Neumann entropy of the MPS as a function of the number of sites in the state. (b) Stabilizer R\'enyi entropy of the state with small local magic.}
\end{figure}
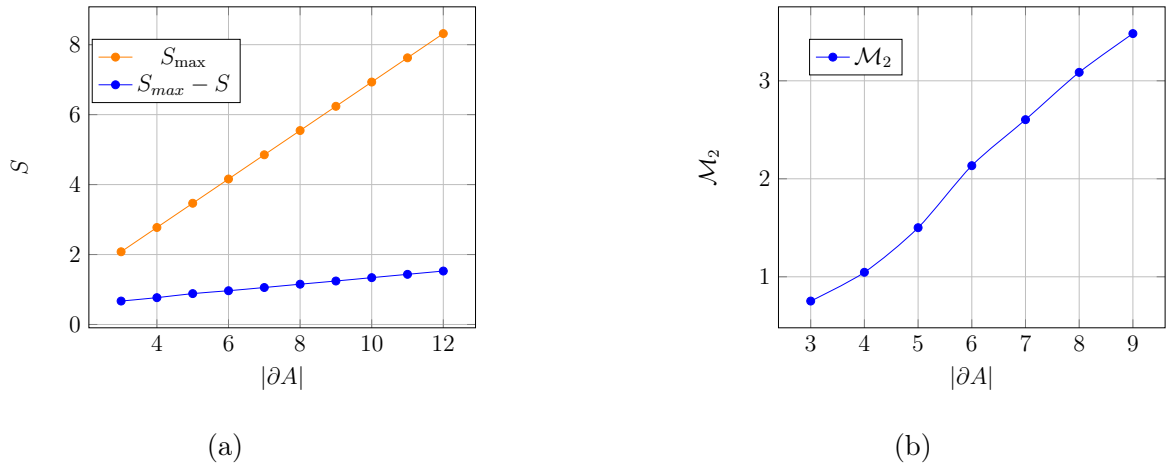


The above intuition is also apparent when we think of the holographic QECC perspective of AdS/CFT where it is given by a code that corrects erasures approximately. In this case, complementary \cite{HarlowRT} approximate erasure correction promises the existence of recover unitaries supported on each subregion, such that

\begin{equation}
    U_A U_{A^c}|\tilde{\psi}\rangle_{AA^c} U_A^{\dagger} U_{A^c}^{\dagger} \approx |\psi\rangle |\chi\rangle,
\end{equation}
where $|\psi\rangle$ captures the bulk encoded information while $|\chi\rangle$ is the entanglement mediating erasure correction\footnote{We note that this heuristic argument is only expected to hold approximately in the leading order $N$ for holographic CFTs as that is when they function as approximate erasure correction codes with recovery errors suppressed by $1/N$. }. This is the case for certain of holographic QECC toy models, such as instances of approximate holographic Bacon-Shor codes\cite{ABSC}, (see e.g. Fig. 41 or generally when the skewing is small,) and \cite{Hayden_2016} when imperfectly entangled pairs are used in place of maximally entangled states when building the tensor network. The latter is known to be able to produce the correct single-interval CFT entanglement entropy but fails at the multi-interval level.

The states $|\chi\rangle$ 
now play the role of the interface tensor in MERA. To leading order, the R\'enyi entropies associated with $|\chi\rangle$ again scale as the area of the extremal surface where it is explicitly given by the number of entangled states across the bulk cut. Therefore, by (\cref{eqn:avgnlsre}), its magic as measured by the stabilizer R\'enyi entropy should scale as the area of the RT surface for any additive magic measure as one simply has to count the number of such approximate Bell pairs.
Strictly speaking, this again yields an upper bound as we do not optimize over all basis choices.

We also expect this linear dependence between non-local magic and entanglement to extend to non-critical states with translational invariance. For example, it is well-known that the truncated MERA (or more simply an MPS) can describe ground state of gapped phases where entanglement can increase slightly as the system grows, but plateaus at sufficiently large $|A|$\footnote{See \cite{Molina-Vilaplana:2012rmg} for example.}.  The non-local magic again resides on the edges connecting $A$ and $B$. From the bond counting argument, we again arrives at $\mathcal{M}^{NL}\sim S(A)$ provided $\mathcal{M}$ is an additive measure of magic.

Having established that the non-local magic should scale as the entropy, let's now examine how it should be connected with the anti-flatness of the entanglement spectrum. Consider again the distilled state $|\chi\rangle$ which we represent as an MPS shared between $A$ and $B$ with local magic removed. Recall that since $\chi_A=\tr_B[|\chi\rangle\langle\chi|] = U_A\rho_A U_A^{\dagger}$, their entanglement spectrum and anti-flatness are identical, i.e., $\mathcal{F}(\chi_A)=\mathcal{F}(\rho_A)$. 

For simplicity, let's approximate the MPS as entangled states $|\phi\rangle^{\otimes n}$ where each state $|\phi\rangle$ can be thought of as imperfect entangled pairs\footnote{For concreteness, one can think of them as imperfect Bell pairs. More generally, they do not have to be qubits, but a pair of qudits that are not maximally entangled. }. These states have volume law entanglement across $A$ and $B$.  We expect this to be a reasonable approximation because MPS with constant bond dimension limits the amount of correlation to be short-ranged, making them close to the tensor products which one can think of as a mean field approximation. We further support this claim with numerical evidence in Appendix~\ref{app:MPS}. 

It is known from \cite{flatness} that for a typical state $|\phi\rangle_{ab}$ with stabilizer linear entropy $M_{\rm lin}(|\phi\rangle)$ chosen from its Clifford orbit $\{\Gamma_{ab}|\phi\rangle, \forall \Gamma\in \mathcal{C}_2\}$, the anti-flatness of the entanglement spectrum of $\phi=\Tr_{b}[|\phi\rangle\langle\phi|]$ when cutting the state in half is given by
\begin{equation}
    \mathcal{F}(\phi) = c(d,d_a) M_{\rm lin}(\phi),
    \label{eqn:Cliffordtypical}
\end{equation}
where $c(d_a,d)=\frac{(d^2-d_a^2)(d_a^2-1)}{(d^2-1)(d+2)}$. Note that the second stabilizer R\'enyi entropy is related to the stabilizer linear entropy $\mathcal{M}_2=-\log (1-M_{\rm lin})$. Here $M_{\rm lin}\leq 1-2(d+1)^{-1}$ with $d$ being the dimension of the Hilbert space of $|\phi\rangle$ and $d_a=\sqrt{d}$ the Hilbert space dimension of subsystem $a$. Applying \eqref{eqn:Cliffordtypical} to each pair, we would have 
\begin{equation}
\begin{split}
    \mathcal{M}_2(|\phi\rangle) \approx &\frac{\mathcal{F}(\phi_a)}{c(d,d_a)}\\
    \approx &- \frac{\pur(\phi_a)^2}{c(d,d_a)}\left.\frac{\partial\tilde{S}_m(\phi_a)}{\partial m}\right\vert_{m=1},
\end{split}
\end{equation}
where we have applied the approximation \eqref{eqn:approxrenyid} to rewrite the R.H.S. in terms of additive anti-flatness measure. 
Based on the assumption of distillation $|\chi\rangle\approx |\phi\rangle^{\otimes n}$ (See~\cref{app:MPS} for discussion) and additivity of $\mathcal{M}_2$ we conclude that
\begin{equation}\label{eqn:minsurfNLM}
\begin{split}
 \mathcal{M}_2(\ket{\chi})&\approx n\frac{\pur(\phi_a)^2}{c(d,d_a)}|\partial_m\tilde{S}_m(\phi_a)||_{m=1}\\
 &=\kappa |\partial_m\tilde{S}_m(\chi_A)||_{m=1}.
\end{split}
\end{equation}
where $\kappa = \pur(\phi_a)^2/c(d,d_a)$  is some coefficient that depends on the details of $|\phi\rangle$. Note that $\partial_m\tilde{S}_m$ is negative in our convention.

Since we argued that $\mathcal{M}_2(|\chi\rangle) \approx \mathcal{M}_2^{NL}(|\psi\rangle_{AB})$, 
\begin{equation}
    \mathcal{M}_2^{NL}(|\psi\rangle_{AB}) \approx\kappa |\partial_m\tilde{S}_m(\chi_A)||_{m=1}.
    \label{eqn:NLMflatness}
\end{equation}
for a CFT ground state. Therefore, if one uses the computable stabilizer R\'enyi entropy $\mathcal{M}_2$, we predict that the non-local magic scales linearly with both entanglement entropy and the additive anti-flatness $\partial_m\tilde{S}_m(\chi_A)|_{m=1}$ of the entanglement spectrum across $A$ and $B$. 
In \cref{sec:numerics}, we numerically verify that this is indeed the case for an Ising CFT.

\begin{remark}
    Notably, our reasoning for the \emph{area-law} scaling of exact magic, i.e. $\mathcal{M}^{(NL)}\sim S(A)$, and the non-flatness relation (\ref{eqn:NLMflatness}) does not rely on any particular properties of the CFT. Indeed, as long as one can concentrate the magic from $A$ to $\partial A$ using the kind of unitary distillation procedure, this area law would also hold for gapped system with entanglement area law. The anti-flatness relation is also similar to the area law scaling of entanglement spread\cite{Anshu_2022}.
\end{remark}

\subsection{Non-local magic in Ising model}\label{sec:numerics}
 
In this section we provide numerical computations to support our prior conjectures. We begin with the $1 + 1D$ transverse field Ising model, with Hamiltonian given by
\ba\label{IsingHamiltonian}
H_{\rm Ising}=-\cos(\theta)\sum_i Z_i Z_{i+1}-\sin(\theta)\sum_i X_i.
\ea
We particularly consider \eqref{IsingHamiltonian} near its critical point, when $\theta = \pi/4$.

This model is described by an Ising CFT in the thermodynamic limit at criticality, that is when $\theta=\frac{\pi}{4}$. For our analysis, we perform exact diagonalization to determine the ground state of a 26-site spin chain with periodic boundary condition. Subsequently, the state is partitioned into two contiguous segments: $A$ and $\bar{A}$. To numerically estimate the non-local magic related to this bipartition, we use the Stabilizer Rényi Entropy measure $\mathcal{M}_2(\{\lambda_i\})$, as defined in~\cref{section:estimate}. Importantly, this measure relies solely on the entanglement spectrum, which we obtain through Singular Value Decomposition (SVD).

At the critical point, we compute the  $\mathcal{M}_2(\{\lambda_i\})$ measure while progressively increasing the size of the subsystem $|A|$. The plot of $\mathcal{M}_2$ is present in \cref{fig:magic_CFTa}. 

\begin{figure}
    \centering
    \begin{subfigure}[b]{0.4\textwidth}
    \scalebox{0.75}{
    \begin{tikzpicture}
        \begin{axis}[
            xlabel={$|A|$}, 
            ylabel={$\mathcal{M}_2$}, 
            tick align=inside, 
            title={$\mathcal{M}_2$ v.s. $|A|$},
            legend style={at={(0.2,0.9)},anchor=north}, 
            grid={both}
        ]
        \addplot[
            color=blue,
            mark=*, 
            smooth 
        ] coordinates {
            (1, 0.2759145827825882)
             (2, 0.28228738186896263)
             (3, 0.29336691658661357)
             (4, 0.30462444579473225)
             (5, 0.3146928627439401)
             (6, 0.32331311085491654)
             (7, 0.3305184956722919)
             (8, 0.3364093242384839)
             (9, 0.34108791009708955)
             (10, 0.3446403128445661)
             (11, 0.3471323166256153)
             (12, 0.34860965777546515)
             (13, 0.3490991613456804)
        };
        \end{axis}
    \end{tikzpicture}
    }
    \caption{}
   \label{fig:magic_CFTa}
\end{subfigure}
\hfill
\begin{subfigure}[b]{0.4\textwidth}
    \centering
    \scalebox{0.75}{
    \begin{tikzpicture}
        \begin{axis}[
            xlabel={$S$}, 
            ylabel={$\mathcal{M}_2$}, 
            tick align=inside, 
            title={$\mathcal{M}_2$ v.s. $S$ },
            legend style={at={(0.2,0.9)},anchor=north}, 
            grid={both}
        ]
        \addplot[
            color=orange,
            mark=*, 
            smooth 
        ] coordinates {
            (0.47365496896115133, 0.2759145827825882)
             (0.5918813439084066, 0.28228738186896263)
             (0.6578352071417765, 0.29336691658661357)
             (0.7030228571069249, 0.30462444579473225)
             (0.7365380200480853, 0.3146928627439401)
             (0.7623514890669422, 0.32331311085491654)
             (0.782554333578475, 0.3305184956722919)
             (0.7983723766168666, 0.3364093242384839)
             (0.8105729284602428, 0.34108791009708955)
             (0.8196538174273882, 0.3446403128445661)
             (0.8259401351911264, 0.3471323166256153)
             (0.8296366915149734, 0.34860965777546515)
             (0.8308567814379734, 0.3490991613456804)
        };
        \end{axis}
    \end{tikzpicture}
    }
    \caption{}
    \label{fig:magic_CFTb}
\end{subfigure}
\hfill
\begin{subfigure}[b]{0.4\textwidth}
    \centering
    \scalebox{0.75}{
    \begin{tikzpicture}
        \begin{axis}[
            xlabel={$\partial_n\tilde{S}_n|_{n=1}$}, 
            ylabel={$\mathcal{M}_2$}, 
            tick align=inside, 
            title={$\mathcal{M}_2$ v.s. $\partial_n\tilde{S}_n|_{n=1}$ },
            legend style={at={(0.2,0.9)},anchor=north}, 
            grid={both}
        ]
        \addplot[
            color=black,
            mark=*, 
            smooth 
        ] coordinates {
            (0.38844044873524697, 0.2759145827825882)
             (0.4080570661929262, 0.28228738186896263)
             (0.4258068254534285, 0.29336691658661357)
             (0.44083017326285073, 0.30462444579473225)
             (0.4533650556446711, 0.3146928627439401)
             (0.46375639975735045, 0.32331311085491654)
             (0.4722992639562782, 0.3305184956722919)
             (0.47922104982083263, 0.3364093242384839)
             (0.4846911425107271, 0.34108791009708955)
             (0.4888330874178818, 0.3446403128445661)
             (0.4917343044615911, 0.3471323166256153)
             (0.493452921810052, 0.34860965777546515)
             (0.49402218734097664, 0.3490991613456804)
        };
        \end{axis}
    \end{tikzpicture}
    }
    \caption{}
    \label{fig:magic_CFTc}
\end{subfigure}
\label{fig:magic_CFT}
\caption{(a) Plot of non-local Stabilizer R\'enyi Entropy $\mathcal{M}_2$ v.s. subsystem size |A|;  (b) Plot of  $\mathcal{M}_2$ v.s. Entropy $S$. (c) Plot of  $\mathcal{M}_2$ v.s. the anti-flatness based on entanglement capacity. Model is at critical point, with 26 lattice sites.}
\end{figure}
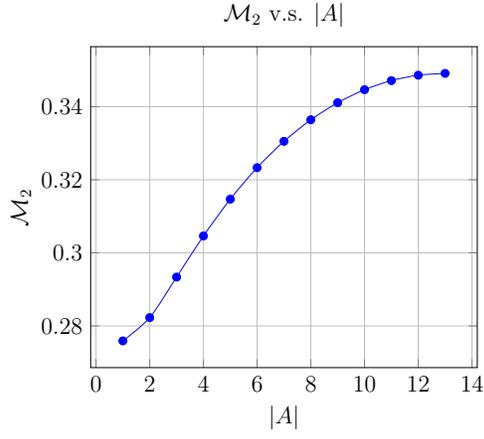
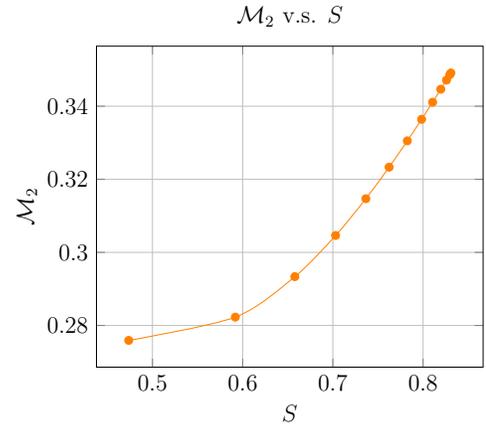
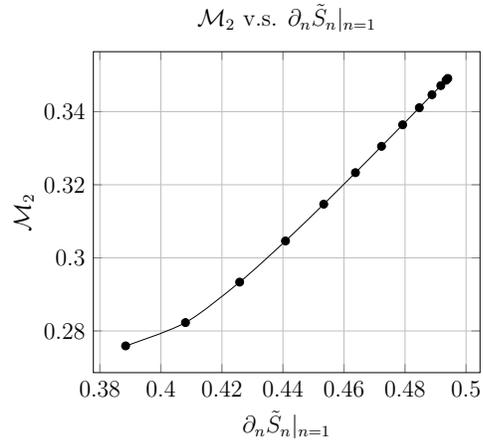

In \cref{fig:magic_CFTb}, we observe that the non-local magic scales similarly to entropy when we increase the size of the subregion $|A|$, particularly beyond 3 qubits. This indicates that the non-local magic in the CFT scales logarithmically with $|A|$, in agreement with our analysis presented in the MERA framework in~\cref{section:MERA}. Additionally, \cref{fig:magic_CFTc} demonstrates the proportional relationship between non-local magic and anti-flatness, supporting the estimation in   \eqref{eqn:NLMflatness}. 

A similar analysis is applied to study the model away from the critical point, as illustrated in~\cref{fig:offcritical}.  We define the parameter $g=\theta-\frac{\pi}{4}$, where $g$ quantifies the deviation from criticality. In this regime, we observe that the non-local magic reaches a plateau at a certain point, mirroring the behavior observed in entropy. 

\begin{figure}
    \centering
    \begin{tikzpicture}
        \begin{axis}[
            xlabel={|A|}, 
            ylabel={$\mathcal{M}_2$}, 
            tick align=inside, 
            title={$\mathcal{M}_2$ v.s. $|A|$},
            legend style={at={(0.2,1.2)},anchor=north}, 
            grid={both}
        ]
        \addplot[
            color=blue,
            mark=*, 
            smooth 
        ] coordinates {
            (1, 0.28407617773515886)
              (2, 0.34397774537887427)
              (3, 0.38751332633761937)
              (4, 0.4190970491310785)
              (5, 0.4421015965056568)
              (6, 0.45896702985427884)
              (7, 0.4713795896484691)
              (8, 0.4804991524617711)
              (9, 0.48712558430690295)
              (10, 0.4918089362542348)
              (11, 0.4949206945520738)
              (12, 0.49669922942297173)
              (13, 0.49727781169858104)
        };
        \addlegendentry{\(g=0.05\)}
        \addplot[
            color=orange,
            mark=*, 
            smooth 
        ] coordinates {
            (1, 0.2622746594000573)
              (2, 0.32291548379551155)
              (3, 0.35835308837480184)
              (4, 0.37928016995541014)
              (5, 0.3917771326269049)
              (6, 0.3993340830642571)
              (7, 0.4039561283143258)
              (8, 0.4068079872130534)
              (9, 0.4085740871637637)
              (10, 0.40965959874648833)
              (11, 0.41030367773714116)
              (12, 0.41064345911210953)
              (13, 0.41074948178076964)
        };
        \addlegendentry{\(g=0.1\)}
        \addplot[
            color=green,
            mark=*, 
            smooth 
        ] coordinates {
            (1, 0.21230120003858938)
              (2, 0.2542448918104177)
              (3, 0.27134898072257646)
              (4, 0.2784874066332622)
              (5, 0.28152527488291956)
              (6, 0.2828413759502512)
              (7, 0.28342048270374093)
              (8, 0.28367879657447115)
              (9, 0.2837954188472098)
              (10, 0.2838485595687859)
              (11, 0.28387271597683295)
              (12, 0.2838830717667294)
              (13, 0.28388594897480607)
        };
        \addlegendentry{\(g=0.15\)}
        \addplot[
            color=red,
            mark=*, 
            smooth 
        ] coordinates {
            (1, 0.18148373123245282)
              (2, 0.21215783087646567)
              (3, 0.22223636904171473)
              (4, 0.22565129476412513)
              (5, 0.2268364651797759)
              (6, 0.22725618915490875)
              (7, 0.22740733176914482)
              (8, 0.22746252974583842)
              (9, 0.22748294302369282)
              (10, 0.22749058018764026)
              (11, 0.22749345708531876)
              (12, 0.2274945078106672)
              (13, 0.2274947725650576)
        };
        \addlegendentry{\(g=0.2\)}
        \end{axis}
    \end{tikzpicture}
    \caption{ Plot of non-local Stabilizer R\'enyi Entropy $\mathcal{M}_2$ v.s. subsystem size |A|. The model parameter $g=\theta-\frac{\pi}{4}$ is adjusted to position the model away from its critical point.}
    \label{fig:offcritical}
\end{figure}
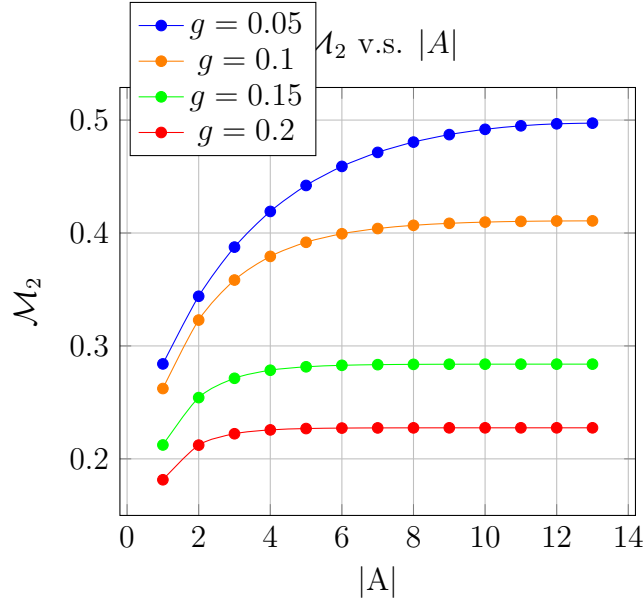

In our final analysis, we keep the size of the subregion $|A|$ constant and track the changes in non-local magic as the model approaches and passes through the critical point. As depicted in \cref{fig:transition}, a distinct peak in non-local magic is observed. Notably, this peak shifts closer to the critical point ($g=0$) and becomes increasingly sharp as the total system size ($n$) is enlarged. These observations suggest the potential presence of a phase transition in the non-local magic measure. \cref{fig:comparison} presents a comparison of non-local magic and anti-flatness against the model parameter $g$, revealing a consistent trend as $g$ changes. 

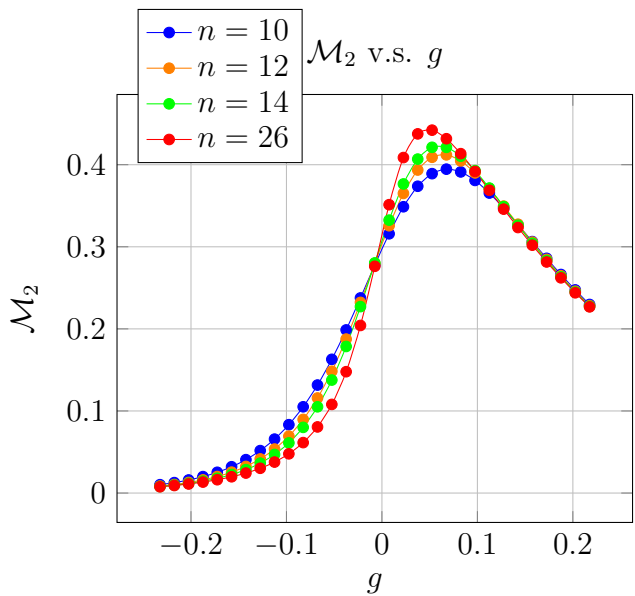
\begin{figure}
    \centering
    \begin{tikzpicture}
        \begin{axis}[
            xlabel={$g$}, 
            ylabel={$\mathcal{M}_2$}, 
            title={$\mathcal{M}_2$ v.s. $g$ },
            tick align=inside, 
            legend style={at={(0.2,1.2)},anchor=north}, 
            grid={both}
        ]
        \addplot[
            color=blue,
            mark=*, 
            smooth 
        ] coordinates {
            (-0.2324768366025517, 0.010179027166607601)
             (-0.2174768366025518, 0.012700677518416869)
             (-0.20247683660255178, 0.015914929866794172)
             (-0.18747683660255177, 0.020024212927263287)
             (-0.17247683660255175, 0.02528791586737733)
             (-0.15747683660255174, 0.032034550069390655)
             (-0.14247683660255173, 0.04067294569286466)
             (-0.1274768366025517, 0.05169871325524065)
             (-0.11247683660255181, 0.06568911429733332)
             (-0.09747683660255169, 0.08327519508063366)
             (-0.08247683660255178, 0.10507537532282413)
             (-0.06747683660255177, 0.13157250926463754)
             (-0.05247683660255176, 0.16292287590693638)
             (-0.037476836602551744, 0.19870953869174304)
             (-0.02247683660255173, 0.2376998361713844)
             (-0.0074768366025518285, 0.27772582130124907)
             (0.007523163397448296, 0.31583345964325205)
             (0.022523163397448198, 0.348777714878471)
             (0.03752316339744832, 0.37375895572844275)
             (0.052523163397448225, 0.3891051689553718)
             (0.06752316339744824, 0.3945830900350399)
             (0.08252316339744825, 0.39122799220737287)
             (0.09752316339744815, 0.38084991622703784)
             (0.11252316339744828, 0.36548868544182306)
             (0.12752316339744818, 0.3470148330281814)
             (0.1425231633974483, 0.32693104318822336)
             (0.1575231633974482, 0.3063316700703259)
             (0.17252316339744822, 0.2859505784494)
             (0.18752316339744823, 0.26624121060662875)
             (0.20252316339744825, 0.247456486000428)
             (0.21752316339744826, 0.22971483114393784)
        };
        \addlegendentry{\(n=10\)}
        \addplot[
            color=orange,
            mark=*, 
            smooth 
        ] coordinates {
            (-0.2324768366025517, 0.008465978246463156)
             (-0.2174768366025518, 0.010407193006500729)
             (-0.20247683660255178, 0.012875327118831839)
             (-0.18747683660255177, 0.016039921641527576)
             (-0.17247683660255175, 0.0201294554803197)
             (-0.15747683660255174, 0.025450769720905188)
             (-0.14247683660255173, 0.03241275236797553)
             (-0.1274768366025517, 0.04155263056466259)
             (-0.11247683660255181, 0.053559580056458625)
             (-0.09747683660255169, 0.06928326208862791)
             (-0.08247683660255178, 0.08970284247116354)
             (-0.06747683660255177, 0.11581559301769835)
             (-0.05247683660255176, 0.14839133812423738)
             (-0.037476836602551744, 0.1875538019511966)
             (-0.02247683660255173, 0.23223365361545356)
             (-0.0074768366025518285, 0.27971406882100586)
             (0.007523163397448296, 0.3256706714548758)
             (0.022523163397448198, 0.3650407584940913)
             (0.03752316339744832, 0.39354680139542364)
             (0.052523163397448225, 0.4090658305660494)
             (0.06752316339744824, 0.41201429206413853)
             (0.08252316339744825, 0.40466397043774915)
             (0.09752316339744815, 0.39001853899237493)
             (0.11252316339744828, 0.37090936249006595)
             (0.12752316339744818, 0.3495543491942982)
             (0.1425231633974483, 0.32748593479147553)
             (0.1575231633974482, 0.30566298236122613)
             (0.17252316339744822, 0.2846304790946891)
             (0.18752316339744823, 0.2646616263994785)
             (0.20252316339744825, 0.24586326932071415)
             (0.21752316339744826, 0.22824671738457047)
        };
        \addlegendentry{\(n=12\)}
        \addplot[
            color=green,
            mark=*, 
            smooth 
        ] coordinates {
            (-0.2324768366025517, 0.007933413987301907)
             (-0.2174768366025518, 0.009643632352382655)
             (-0.20247683660255178, 0.011791929928751497)
             (-0.18747683660255177, 0.014519609408522049)
             (-0.17247683660255175, 0.018021681370604066)
             (-0.15747683660255174, 0.02256839312394047)
             (-0.14247683660255173, 0.028534509077290136)
             (-0.1274768366025517, 0.03643744218320802)
             (-0.11247683660255181, 0.04698286638780856)
             (-0.09747683660255169, 0.06110965787415222)
             (-0.08247683660255178, 0.08001058854894082)
             (-0.06747683660255177, 0.10507569246146076)
             (-0.05247683660255176, 0.13766195830438846)
             (-0.037476836602551744, 0.17856395871997527)
             (-0.02247683660255173, 0.22713655856876372)
             (-0.0074768366025518285, 0.28033989294322587)
             (0.007523163397448296, 0.33248637767314865)
             (0.022523163397448198, 0.37655598026422266)
             (0.03752316339744832, 0.40684951518439183)
             (0.052523163397448225, 0.4212065461968845)
             (0.06752316339744824, 0.4211640128304521)
             (0.08252316339744825, 0.41036777854896495)
             (0.09752316339744815, 0.39276294067321815)
             (0.11252316339744828, 0.37155015442930295)
             (0.12752316339744818, 0.34892560405388956)
             (0.1425231633974483, 0.3262306498393502)
             (0.1575231633974482, 0.3042075712158339)
             (0.17252316339744822, 0.283222806200963)
             (0.18752316339744823, 0.2634239992990375)
             (0.20252316339744825, 0.24483906960710441)
             (0.21752316339744826, 0.22743469159561436)
        };
        \addlegendentry{\(n=14\)}
        \addplot[
            color=red,
            mark=*, 
            smooth 
        ] coordinates {
            (-0.2324768366025517, 0.007637717074741366)
             (-0.2174768366025518, 0.009174338805064993)
             (-0.20247683660255178, 0.011052959394144994)
             (-0.18747683660255177, 0.013365130617364369)
             (-0.17247683660255175, 0.01623280968600742)
             (-0.15747683660255174, 0.01982129786435423)
             (-0.14247683660255173, 0.024359537181350557)
             (-0.1274768366025517, 0.03017316361630484)
             (-0.11247683660255181, 0.037740180797763534)
             (-0.09747683660255169, 0.047786889149565595)
             (-0.08247683660255178, 0.061453076176694446)
             (-0.06747683660255177, 0.08056229243796553)
             (-0.05247683660255176, 0.1079877179114993)
             (-0.037476836602551744, 0.14782590237321785)
             (-0.02247683660255173, 0.20418188817273544)
             (-0.0074768366025518285, 0.2762491075509873)
             (0.007523163397448296, 0.351294063986569)
             (0.022523163397448198, 0.4087064418576)
             (0.03752316339744832, 0.437567104391568)
             (0.052523163397448225, 0.4421015965056568)
             (0.06752316339744824, 0.4316509592969976)
             (0.08252316339744825, 0.41346368643960796)
             (0.09752316339744815, 0.3917771326269049)
             (0.11252316339744828, 0.36885259515848734)
             (0.12752316339744818, 0.3458776892433747)
             (0.1425231633974483, 0.32347725558175994)
             (0.1575231633974482, 0.3019745954831864)
             (0.17252316339744822, 0.28152527488291956)
             (0.18752316339744823, 0.262188711010384)
             (0.20252316339744825, 0.24396854935148515)
             (0.21752316339744826, 0.2268364651797759)
        };
        \addlegendentry{\(n=26\)}
        \end{axis}
    \end{tikzpicture}
    \caption{ Plot of non-local Stabilizer R\'enyi Entropy $\mathcal{M}_2$ v.s. parameter $g=\theta-\frac{\pi}{4}$, at $|A|=5$ with increasing total spins $n$. }
    \label{fig:transition}
\end{figure}

\begin{figure}[ht]
    \centering
    \begin{subfigure}[b]{0.4\textwidth}
        \scalebox{0.8}{
        \begin{tikzpicture}
            \begin{axis}[
            xlabel={$g$}, 
            ylabel={$\mathcal{M}_2$}, 
            tick align=inside, 
            legend style={at={(0.2,1.2)},anchor=north}, 
            grid={both}
        ]
        \addplot[
            color=blue,
            mark=*, 
            smooth 
        ] coordinates {
            (-0.2324768366025517, 0.007145304321954218)
             (-0.2174768366025518, 0.008487960605390173)
             (-0.20247683660255178, 0.010101485878992943)
             (-0.18747683660255177, 0.012053517023108064)
             (-0.17247683660255175, 0.014435402927989302)
             (-0.15747683660255174, 0.017374563446235013)
             (-0.14247683660255173, 0.02105582840800511)
             (-0.1274768366025517, 0.02575961508929479)
             (-0.11247683660255181, 0.03193201207252851)
             (-0.09747683660255169, 0.04031495948275232)
             (-0.08247683660255178, 0.05218588101359276)
             (-0.06747683660255177, 0.06977904131169534)
             (-0.05247683660255176, 0.09693602789547702)
             (-0.037476836602551744, 0.13974677692867024)
             (-0.02247683660255173, 0.20574693233055968)
             (-0.0074768366025518285, 0.29755202131258474)
             (0.007523163397448296, 0.3990203245055871)
             (0.022523163397448198, 0.47501942013468546)
             (0.03752316339744832, 0.5045213886126098)
             (0.052523163397448225, 0.49727781169858104)
             (0.06752316339744824, 0.4722337165709027)
             (0.08252316339744825, 0.44161526289047964)
             (0.09752316339744815, 0.41074948178076964)
             (0.11252316339744828, 0.3814622520861778)
             (0.12752316339744818, 0.354203227070393)
             (0.1425231633974483, 0.3289563794953849)
             (0.1575231633974482, 0.30557374117264996)
             (0.17252316339744822, 0.28388594897480607)
             (0.18752316339744823, 0.2637345106763872)
             (0.20252316339744825, 0.24497865343083947)
             (0.21752316339744826, 0.2274947725650576)
        };
        \end{axis}
        \end{tikzpicture}
        }
        \caption{}
    \end{subfigure}
    \hfill
    \begin{subfigure}[b]{0.4\textwidth}
       \scalebox{0.8}{
        \begin{tikzpicture}
            \begin{axis}[
            xlabel={$g$}, 
            ylabel={$\frac{\mathcal{F}(\rho)}{\mathrm{Pur}(\rho)^2}$}, 
            tick align=inside, 
            legend style={at={(0.2,1.2)},anchor=north}, 
            grid={both}
        ]
        \addplot[
            color=green,
            mark=*, 
            smooth 
        ] coordinates {
            (-0.2324768366025517, 0.0017906550155215046)
             (-0.2174768366025518, 0.002128430358248331)
             (-0.20247683660255178, 0.002535271954815998)
             (-0.18747683660255177, 0.0030292921303933847)
             (-0.17247683660255175, 0.003635877087823563)
             (-0.15747683660255174, 0.004392447420408408)
             (-0.14247683660255173, 0.005357636636801908)
             (-0.1274768366025517, 0.006629761987891159)
             (-0.11247683660255181, 0.008385239763439722)
             (-0.09747683660255169, 0.010960203971579534)
             (-0.08247683660255178, 0.015024513163008176)
             (-0.06747683660255177, 0.021939845122460465)
             (-0.05247683660255176, 0.0343972961772237)
             (-0.037476836602551744, 0.05700884023110594)
             (-0.02247683660255173, 0.09447793318209577)
             (-0.0074768366025518285, 0.1426692425225171)
             (0.007523163397448296, 0.18137932190436926)
             (0.022523163397448198, 0.1939601185649182)
             (0.03752316339744832, 0.18574544217475258)
             (0.052523163397448225, 0.16953391029563253)
             (0.06752316339744824, 0.15267413044338285)
             (0.08252316339744825, 0.1375274912018791)
             (0.09752316339744815, 0.12439700438451141)
             (0.11252316339744828, 0.11302272776847992)
             (0.12752316339744818, 0.10307817688079396)
             (0.1425231633974483, 0.09429054631087523)
             (0.1575231633974482, 0.08645135784048381)
             (0.17252316339744822, 0.0794032334518217)
             (0.18752316339744823, 0.07302573383562753)
             (0.20252316339744825, 0.06722470772402801)
             (0.21752316339744826, 0.061924963221174356)
        };
        \end{axis}
        \end{tikzpicture}
        }
        \caption{}
    \end{subfigure}
     \vskip\baselineskip
     \begin{subfigure}[b]{0.4\textwidth}
        \scalebox{0.8}{
         \begin{tikzpicture}
            \begin{axis}[
            xlabel={$g$}, 
            ylabel={$|\partial_n \tilde{S}_n||_{n=1}$}, 
            tick align=inside, 
            legend style={at={(0.2,1.2)},anchor=north}, 
            grid={both}
        ]
        \addplot[
            color=black,
            mark=*, 
            smooth 
        ] coordinates {
            (-0.2324768366025517, 0.005280910267068737)
             (-0.2174768366025518, 0.006510627625492099)
             (-0.20247683660255178, 0.00803596046963862)
             (-0.18747683660255177, 0.009939502333841627)
             (-0.17247683660255175, 0.012334930862737499)
             (-0.15747683660255174, 0.015384434278946342)
             (-0.14247683660255173, 0.019329900398019965)
             (-0.1274768366025517, 0.02455095011158324)
             (-0.11247683660255181, 0.031675969945498386)
             (-0.09747683660255169, 0.04179745283457983)
             (-0.08247683660255178, 0.056886945731048945)
             (-0.06747683660255177, 0.08056106196909345)
             (-0.05247683660255176, 0.11931511666420035)
             (-0.037476836602551744, 0.18365927393584316)
             (-0.02247683660255173, 0.2854831656147101)
             (-0.0074768366025518285, 0.4227007543631155)
             (0.007523163397448296, 0.5574088985771578)
             (0.022523163397448198, 0.6380384975874119)
             (0.03752316339744832, 0.6538151915333436)
             (0.052523163397448225, 0.629612387418702)
             (0.06752316339744824, 0.5898810310801145)
             (0.08252316339744825, 0.5471976045328252)
             (0.09752316339744815, 0.5062319411694376)
             (0.11252316339744828, 0.4682415484170781)
             (0.12752316339744818, 0.43330194413223755)
             (0.1425231633974483, 0.4011593535907676)
             (0.1575231633974482, 0.37151086480571227)
             (0.17252316339744822, 0.3440817036848223)
             (0.18752316339744823, 0.3186383619697557)
             (0.20252316339744825, 0.29498445375401916)
             (0.21752316339744826, 0.27295368302457707)
        };
        \end{axis}
         \end{tikzpicture}
    }
         \caption{}
     \end{subfigure}
     \hfill
     \begin{subfigure}[b]{0.4\textwidth}
        \scalebox{0.8}{
         \begin{tikzpicture}
         \begin{axis}[
            xlabel={$g$}, 
            ylabel={$S$}, 
            tick align=inside, 
            legend style={at={(0.2,1.2)},anchor=north}, 
            grid={both}
        ]
        \addplot[
            color=red,
            mark=*, 
            smooth 
        ] coordinates {
            (-0.2324768366025517, 0.7075195993413678)
             (-0.2174768366025518, 0.7098622966766157)
             (-0.20247683660255178, 0.7126093020806886)
             (-0.18747683660255177, 0.7158476059617306)
             (-0.17247683660255175, 0.7196900786617996)
             (-0.15747683660255174, 0.7242859044010616)
             (-0.14247683660255173, 0.7298358887171571)
             (-0.1274768366025517, 0.7366142225423622)
             (-0.11247683660255181, 0.7449965966544939)
             (-0.09747683660255169, 0.7554868402188267)
             (-0.08247683660255178, 0.7687086181565803)
             (-0.06747683660255177, 0.7852568413515487)
             (-0.05247683660255176, 0.8051400532946233)
             (-0.037476836602551744, 0.8263312358056276)
             (-0.02247683660255173, 0.8423208036318784)
             (-0.0074768366025518285, 0.8413255742297105)
             (0.007523163397448296, 0.8131684153975448)
             (0.022523163397448198, 0.7604474831216399)
             (0.03752316339744832, 0.6968134434919527)
             (0.052523163397448225, 0.634554879455417)
             (0.06752316339744824, 0.5791007799488848)
             (0.08252316339744825, 0.5311801664511724)
             (0.09752316339744815, 0.48976840976322106)
             (0.11252316339744828, 0.45356850203610183)
             (0.12752316339744818, 0.4214828918603129)
             (0.1425231633974483, 0.39268085746758397)
             (0.1575231633974482, 0.366551789226895)
             (0.17252316339744822, 0.34264488394343834)
             (0.18752316339744823, 0.32062101358587114)
             (0.20252316339744825, 0.30021898829499505)
             (0.21752316339744826, 0.2812328944584596)
        };
        \end{axis}
         \end{tikzpicture}
         }
         \caption{}
     \end{subfigure}
    \caption{(a) Plot of non-local magic $\mathcal{M}_2$ v.s. $g$. (b) Plot of anti-flatness v.s. $g$. (c) Plot of anti-flatness $|{\partial_n\tilde{S}_n}|$ v.s. $g$. (d) Plot of entropy $S$ v.s. $g$. All of these plots are based on data for fixed subregion size $|A|=13$. } 
    \label{fig:comparison}
\end{figure}
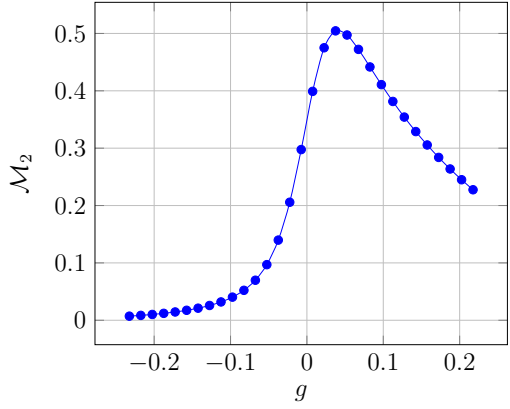
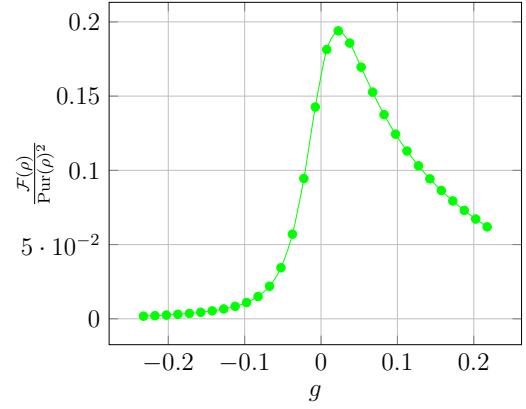
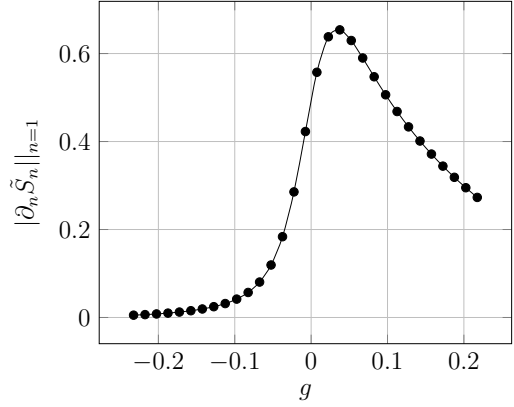
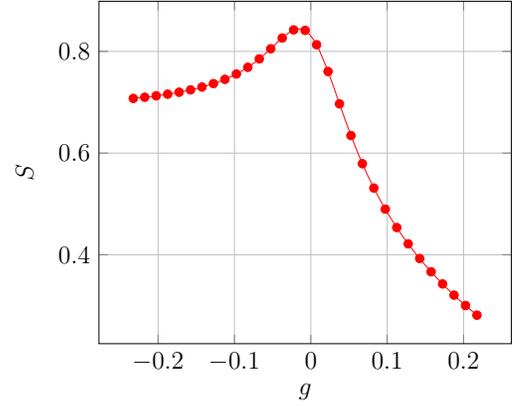

However, it is important to point out that non-local magic is not simply the entanglement entropy despite their similarity in this example. For instance, the ratio between non-local magic and entanglement depends on $g$. \cref{fig:NonLocalMagicSurface} gives a complete picture of $\mathcal{M}_2/S$ for a $14$-qubit Ising chain, as we vary both the parameter $g$ and the subsystem cardinality $|A|$. We observe that $\mathcal{M}_2/S$ maximizes for angles slightly above the critical point ($g = 0$) due to finite size effect, in agreement with \cref{fig:transition,fig:comparison}. 
\begin{figure}
    
\begin{center}
		\begin{overpic}[width=10cm]{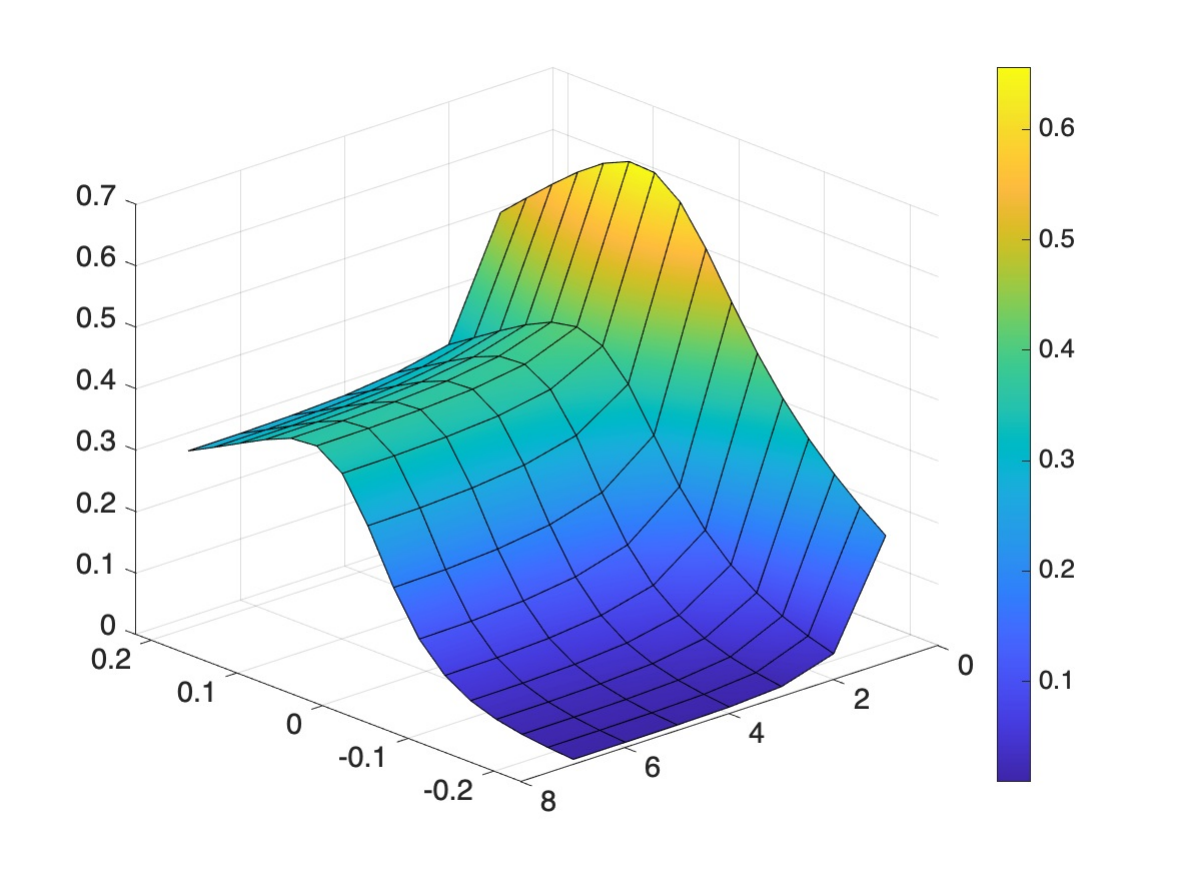}
		\put (-1,40) {$\frac{\mathcal{M}_2}{S}$}
		\put (20,10) {$g$}
		\put (66,7) {$|A|$}
        \end{overpic}
        \caption{Surface illustrating the ratio of non-local magic to entanglement entropy in $n=14$ Ising CFT. We plot $\mathcal{M}_2/S$ as a function of parameter $g = \theta - \pi/4$ and subsystem size $|A|$. The value $\mathcal{M}_2/S$ reaches a maximum just above criticality ($g = 0$), before decreasing and ultimately plateauing.}
        \label{fig:NonLocalMagicSurface}
\end{center}
\end{figure}

The plateau in  \cref{fig:NonLocalMagicSurface} suggests a linear scaling between $M_2$ and $S$, as subsystem $|A|$ grows large. As we see that the linear behavior is already apparent at $n=14$. Recall from the tensor network picture, the linear scaling between non-local magic and entanglement entropy is expected, however, the density of non-local magic can vary depending on the shape of the spectrum. This is reflected in the figure as the asymptotic proportionality constant between $M_2$ and $S$ depends on $\theta$. 

Another instance where non-local magic distinguishes itself from entanglement can be found in the context of symmetry breaking. For $g<0$, the Ising model enters the symmetry-breaking phase in the thermodynamic limit where the non-local magic further displays a transition. We refer interested readers to~\cref{app:symmbreak} for details of this discussion.

\subsection{Smoothed Magic from Entropic Bounds}
However, beyond tensor network and finite-size numerics, we recognize that many of the entropic quantities we have examined so far are generally infinite in conformal field theories and need regularization. It makes more sense to look at smoothed magic, which can be bounded by smoothed max-entropies. On the one hand, it generally leads to finite quantities. On the other hand, for any reasonable simulation of a CFT, it is far more relevant to produce approximations of a target state up to a small precision parameter $\epsilon$ instead of the exact state defined by the theory.

In \cite{Bao:2019}, it was shown that under the assumption that the R\'enyi entropies satisfy $S_{n} = \frac{s_n}{G_N}$, the smoothed maximal entropy is directly proportional to the following expression:

\begin{equation}
    S_{max}^{\epsilon}=S+\sqrt{\log{\frac{1}{\epsilon}}S}+O(c^0), 
\end{equation}
where $S$ denotes the von Neumann entropy of the state. A similar expression is obtained by \cite{maxminentropy} using the explicit spectrum for a 1+1D CFT by Calabrese and Lefevre\cite{CFTspec}. This entropy is proportional to the central charge $c$ of holographic CFT, which is assumed to be  large. The leading-order correction to this expression is at $O(1)$, making it negligible relative to the primary term. 

With this in mind, we can estimate the lower bound for magic as follows:
\begin{equation}
    M^{(NL,\epsilon)}_{RS}(\rho_{AA^c})\geq S_{max}^{\epsilon}(A)-S(A)=\sqrt{S(A)\log\frac{1}{\epsilon}}+O(\epsilon c). 
\end{equation}

We assume the parameter $\epsilon$ to fall within the range $e^{-c}\ll\epsilon\ll c^{-1}$. 

Recall that for a given bipartition $A$ and $A^c$ in a holographic CFT, the von Neumann entropy of subregion $A$ to leading order is equal to the area $\mathcal{A}$ of the extremal surface anchored to the entangling boundary $\partial A$ divided by $4G_N$ according to the Ryu-Takayanagi formula\cite{Ryu_2006}. Thus, we can formally represent the lower bound of non-local magic as:
\begin{equation}
    M^{(NL,\epsilon)}_{RS}(\rho_{AA^c})\geq \sqrt{\log{\frac{1}{\epsilon}}}\sqrt{\frac{\mathcal{A}}{4G_N}},
\end{equation}
where $G_N$ denotes the bulk gravitational constant, which is related to the central charge of the CFT through the equation $c=\frac{3R}{2G_N}$ for 1+1 d CFT, and $c\sim \frac{R^{d-1}}{G_N}$ for general dimensions. $R$ is the AdS radius.

\subsubsection*{Exact and Smoothed Magic in CFTs}
Having obtained a lower bound, we now examine the smoothed magic upper bound. Let's pause for a moment and make an interesting observation about exact versus smoothed magic. Consider $n$ copies of $|\psi\rangle=a|00\rangle+b|11\rangle$ which is not maximally entangled. For any additive magic measure, the total magic $M\sim n$. The same can be deduced from the entropy bounds as both the lower and upper bounds pick up a constant multiple of $n$ compared to that of a single copy. 

However, smoothed entropies are not additive. If we allow for approximations, then it is known that \cite{Hayden_2003} for any state $|\psi\rangle$ there exist local unitaries $U_A\otimes U_B$ such that 
\begin{equation}
    F(U_A\otimes U_B |\psi\rangle^{\otimes n} , |\Phi^+\rangle^{S_n-O(\sqrt{n})} \otimes |\chi\rangle)\geq 1-\epsilon
    \label{eqn:distill}
\end{equation}
for some $\epsilon$, where $F(\sigma,\rho)= (\Tr[\sqrt{\sigma^{1/2}\rho\sigma^{1/2}}])^2$ is the Uhlmann fidelity and $|\chi\rangle$ is a state that's entangling $O(\sqrt{n})$ qubits. Because the perfect Bell pairs $|\Phi^+\rangle$ contain zero magic, the smoothed non-local magic of such a system must be upper bounded by $O(\sqrt{n})$ with implicit $\epsilon$ dependence. From this, we can derive a tighter upper bound of $O(\sqrt{n})\sim O(\sqrt{S})$. This agrees with the lower bound up to constant factors. Hence assuming the distillation argument, the smoothed non-local magic $M_{RS}^{(NL,\epsilon)}(\rho_{AA^c})\sim O(\sqrt{S(A)})$. This is contrasted with magic scaling without smoothing, which has shown to scale linearly with $S(A)$ in tensor networks and small size numerics without smoothing.



A similar argument can be applied to CFTs by taking an $n$-fold tensor product. Let $|\psi\rangle_{AB}$ now be a CFT ground state with some fixed bipartition. Under such an n-fold tensor product, $c\rightarrow nc$ and the magic lower bound scales as $O(\sqrt{c})\rightarrow O(\sqrt{cn})$ where we take identical bipartitions $A,B$ for all copies of the CFT. Although the magic scaling is $O(n)$ according to the smoothed max entropy upper bound, by \cref{eqn:distill}, a tighter bound from Bell pair counting yield $O(\sqrt{n})$ scaling again. On the surface, an $n-$fold copy of CFTs should have n fold increase of the non-local magic if the measure is additive, however, we see that smoothing in fact always brings about a quadratic reduction in the amount of required magic in producing an approximation of the target state. 

It is natural to ask whether the square root scaling of smoothed magic persists for $SU(N)$ gauge theories like holographic CFTs in the large $N$ limit and when the upper/lower bounds in~\cref{th:smoothed} are tight. Here we conjecture that the lower bound (\cref{eqn:smoothinequality}) is essentially saturated by smoothed magic whereas the non-smoothed magic can scale linearly with the R\'enyi entropies $S_n(A)$. In other words, the upper bounds (\cref{prop:RelativeEnt}) and (\cref{eqn:smoothinequality}) are approximately saturated up to constant multiplicative factors.






\begin{conjecture}
\label{conj:cft}
Let $|\psi\rangle_{AB}$ be a low energy state of any conformal field theory. Assuming a UV cut off to render entropies finite, let $S(A)$ be the von Neumann entropy of the state on a contiguous subregion $A$. For any additive measure of magic,

\begin{enumerate}[label=(\alph*)]
    \item the smoothed non-local magic evaluated at any fixed precision $\epsilon$ is of $O(\sqrt{S(A)})$.
\item If the exact non-local magic is well-defined, then it scales as $O(S(A))$. \label{conjb}
\end{enumerate}

\end{conjecture}

A simple reasoning is as follows. Suppose the bipartite entanglement across $AB$ are distillable such that for each Planck area of the RT surface, we can obtain a Bell-like state $|\chi\rangle_{AB}$ which need not be maximally entangled; suppose these states are near identical by the conformal symmetries of the CFT ground state, then we must have $O(S(A))$ copies of such states. Following a distillation like~\cref{eqn:distill}, we obtain at most $O(\sqrt{S})$ states that are imperfectly entangled, in which non-local magic can reside. Note that if no smoothing is allowed, and the magic measure is additive, then the $O(S)$ number of entangled pairs simply contain $O(S)$ amount of magic, consistent with our MERA intuitions and CFT numerics. 

This conjecture, if true, has a wider implication for quantum simulations of conformal field theories. Although our na\"ive expectation is that the non-local magic should increase as the volume of the minimal surface, as indicated by holographic tensor networks, the magic needed to produce a good approximation allows a quadratic reduction. In terms of non-Clifford resources, it implies that an practical preparation of a CFT ground state may permit a quadratic reduction of $T$ gates compared to na\"ive expectations with moderate scaling with increasing precision $\epsilon$. 
However, the actual state preparation has to take into account local magic, which is volume law, and multipartite non-local magic, which is not covered by our bipartite analysis. Therefore, although a state isospectral to $\rho_A$ may consume less non-Clifford resource, we make no claim as to how it alters the total resource scaling for the preparation of $\rho_A$.


\subsubsection*{Anti-flatness and smoothed magic}
Now we comment on a key relation between smoothed magic and entanglement in the CFT. It was suggested in \cite{white_conformal_2021} that magic non-locally distributed would be needed to reproduce the anti-flatness of the CFT entanglement spectrum. We have seen a version of it for exact magic in \cref{sec:CFT}. We can also verify this relation precisely for smoothed magic --- the spectral anti-flatness $\mathcal{F}_R(\rho_A)$ is proportional to the amount of smoothed non-local magic $M_{RS}^{NL}(\rho_{AB})$ to leading order. However, the scaling with entropy is different.
 
\begin{proposition}\label{prop:flatsmoothNLM}
    For any bipartition $A$ and $A^c$ of the CFT ground state, the anti-flatness of the CFT entanglement spectrum necessitates the existence of smoothed non-local magic of at least $$O \left(\sqrt{S(A) \log (1/ \epsilon)} \right)$$. If the distillation argument holds, then 
    \begin{equation}
        \mathcal{F}_R(\rho_A) \sim M_{RS}^{(NL,\epsilon)}(\rho_{AA^c})= O(\sqrt{S(A)}).
        \label{eqn:smoothentBound}
    \end{equation}
\end{proposition}

\section{Holographic Magic and Gravity} \label{sec:gravity}

Heuristically, anti-flatness of the entanglement spectrum is critical in emerging gravity. Various approaches for (entanglement) entropic derivations of the Einstein's equations make use of entanglement first law in both AdS/CFT, e.g.\cite{Blanco_2013,Faulkner_2014,swingle2014universality,Czech_2017}, and beyond~\cite{Jacobson_2016,Cao_2017,Cao_2018}. 
This simple relation connects the stress energy by way of modular Hamiltonian $H_A=-\log \rho_A$. Under a perturbation $\rho_A\rightarrow \rho_A+\delta \rho$ such that $\delta S \equiv S(\rho_A+\delta \rho) - S(\rho_A)$ and $\delta \langle H_A\rangle \equiv \Tr[H_A \delta{\rho}]$, then to linear order  $\delta S  = \delta \langle H_A\rangle$. As entropy is linked to the area of an extremal surface and $H_A$ can be linked to functions of the stress energy tensor in quantum field theories, $\delta\langle H_A\rangle$ is connected to perturbation in stress energy caused by the perturbation $\delta\rho$ while $\delta S$ can be linked to the area and hence metric perturbation. The combination of these relations produce the Hamiltonian constraint, where a covariantized version leads to the (linearized) Einstein's equations. It is clear that if the spectrum was flat, i.e. the system has zero non-local magic and the modular Hamiltonian is proportional to the identity, then no state perturbation can ever incur entropy and therefore metric perturbations, let alone Einstein gravity. Therefore, it is natural to link non-local magic to the emergence of gravity by way of entanglement spectrum.

In this section, we examine non-local magic in CFTs with dual gravity theories. Although it is speculated that non-local magic should play an important role in the dual theory \cite{white_conformal_2021,nogo}, the precise relation has not been made clear. We now provide a holographic dual of non-local magic: non-local magic in the CFT is backreaction in the bulk.

\subsection{Brane tension and magic}\label{section:brane}
Let's make a more precise statement from the point of view of R\'enyi entropies. Recall that  the R\'enyi entropies  in holographic CFTs are computed by the replica geometries which insert a conical singularity that correspond to cosmic branes at various tensions \cite{Dong:2018lsk,Dong1}. Therefore, anti-flatness in the entanglement spectrum can be naturally interpreted as the difference between minimal surfaces areas in different backreacted geometries caused by the addition of some stress energy in the form of a cosmic brane with tension $\mathcal{T}$. 

More precisely, the derivative of brane area is related to anti-flatness (\cref{def:braneflatness}),
\ba
\frac{\partial_n A_n}{4G}=\partial_n\tilde{S}_n.
\ea

The brane tension $\mathcal{T}$ is related to $n$ by 
\ba
\mathcal{T}_n=\frac{n-1}{4n G}
\ea

Hence for $n=1$, or tension $\mathcal{T}=0$, we have that 
$4G\partial_n A_n|_{n=1} = \partial \mathcal{A}/ \partial\mathcal{T}|_{\mathcal{T}=0}$.
Applying (\cref{eqn:NLMflatness}) we arrive at a linear relation between $\partial\mathcal{A}/\partial\mathcal{T}\sim \mathcal{M}_2(|\phi\rangle)$, specifically
\ba\label{branemagicNL}
\left\vert\frac{\partial \mathcal{A}}{\partial \mathcal{T}}\right\vert_{\mathcal{T}=0} = (4G)^2|\partial_n\tilde{S}_n||_{n=1}\approx \frac{(4G)^2}{\kappa} \mathcal{M}_2^{NL}(|\psi\rangle_{AB})
\ea
which then provides an estimate for the non-local magic $M_{\rm dist}^{(NL)}$ across the bipartition from~\cref{th:magicdist}. Note that the bipartition is arbitrary and each subregion $A$ need not be connected.


That is, non-local magic controls the level of geometric change in response to adding mass energy in the bulk, where the zero magic limit indeed recovers the trivial response function in stabilizer holographic tensor networks. As we showed earlier in \cref{lemmaNL}, anti-flatness is  zero if and only if the non-local magic vanishes. Then, through \cref{branemagicNL},  there is no back-reaction in the zero magic limit. This is consistent with results from \cite{nogo}.


\begin{remark}
    Recall the flatness problem of the entanglement spectrum is also present in random tensor networks even though they are not stabilizer codes. This is because non-local magic is also low for Haar random states (\cref{rmk:1}), even though they are not stabilizer codes. Therefore the same type of gravitational backreaction is also ``turned off'' in \cite{Hayden_2016}.
\end{remark}

A more rigorous bound relating non-local magic and the R\'enyi entropy derivatives $\partial_n \mathcal{A}$  can also be proven. 

\begin{proposition}\label{pp:branebound}
Assuming the distillation argument where $U_A\otimes U_B\ket{\psi}_{AB} \approx \otimes_i\ket{\phi_i}_{a_ib_i}$ for the state with local magic removed, then the non-local stabilizer R\'enyi entropy for a CFT under bipartition $AB$ is bounded by
    \begin{equation}
    \frac{1}{2}\left\vert\frac{\partial_n\mathcal{A}_n|_{n=2}}{4G}(\ket{\psi}_{AB})\right\vert\leq\mathcal{M}_2(\ket{\psi}_{AB})\leq\left\vert\frac{\partial_n\mathcal{A}_n|_{n=1}}{4G}(\ket{\psi}_{AB})\right\vert
\end{equation}
\end{proposition}

See proof in~\cref{app:branebound} and justification of the distillation assumption for CFT in~\cref{app:MPS}. 
We elaborate the regime of validity for various magic bounds and anti-flatness relations in \cref{app:bound}.

\subsection{Magic in Holographic CFT}
Note that magic in quantum many-body systems is generally difficult to compute as the cost can grow exponentially with the system size\cite{white_conformal_2021,PhysRevA.106.042426}. This scaling is much improved for measures like stabilizer R\'enyi entropy where the non-linear function of the state can be computed using MPS\cite{haug_efficient_2023,tarabunga_manybody_2023,tarabunga2024nonstabilizerness} or enumerator-based tensor networks \cite{cao2023quantum}. However, the computation remains costly at high bond dimensions and for other measures.  
On the other hand, the bounds of magic from~\cref{sec:magicbounds} offer an entropic perspective into this otherwise hard-to-compute quantity by leveraging existing results. 

We now study non-local magic in CFTs in light of the general relations derived in~\cref{sec:magicbounds}.
Using the holographic dictionary and applying \cref{conj:cft}, we can predict the behaviour of non-local magic in CFTs that are otherwise difficult to compute. Although the following examples essentially amounts to putting square roots on known holographic entanglement entropies, it is instructive to review their behaviours and analyze their implications for magic and, by extension, classical complexity and quantum resource needed for state preparation. At the same time, holographic calculations enable us to study magic dynamics under quantum quenches, for which existing results have been sparse and size-limited \cite{Sewell} due to prohibitive computational costs.

\subsubsection*{Static Configurations}
We now apply (\ref{eqn:smoothentBound}) to estimate the smoothed non-local magic in the CFT state. To illustrate, consider the thermal state $\rho_{AA^c}$ of a (1+1)d CFT which is purified by $B$, e.g. in a thermal field double state.
\begin{equation}
    |TFD\rangle \propto \sum_n \exp(-\beta E_n/2) |E_n\rangle_{AA^c} |E_n\rangle_{B}
\end{equation}

Bipartitioning the system into $A$ and $A^c\cup B$, the behavior of the non-local magic is given by:
\begin{equation}
   \sma{|TFD\rangle_{A{A}^cB}}\sim \sqrt{\frac{c}{3}\log\left(\frac{\beta}{\pi\delta_{UV}}\sinh\left(\frac{\pi l}{\beta}\right)\right)},
\end{equation}
where $l$ is size of subregion $A$. The magic increases logarithmically with the subregion size 
$l$ for $l\ll \beta$. However, when the size surpasses the thermal correlation length, represented by $\beta=\frac{1}{T}$, it becomes proportional to $\sqrt{l}$. A similar result holds for a small subsystem $A$ of a pure state $|\psi\rangle_{AA^c}$ with that thermalizes under ETH such that $A$ has fixed temperature $T=1/\beta$.

Now instead consider the bipartition of the system in to $AA^c$ and $B$. It is known that for holographic CFTs, the system undergoes a confinement-deconfinement phase transition which corresponds to the Hawking-Page transition in the bulk at a critical temperature $T_c$\footnote{This is a simplified account of the transition, which for different theories there can be different phases as one dial up the temperature\cite{Aharony_2004}.}. 

It is known that 
\begin{equation}
    S(B)=S(AA^c) \sim \begin{cases}
        O(N^0)\quad T<T_c\\
        O(N^2)\quad T>T_c
    \end{cases}
\end{equation}
In the same way, we predict a magic phase transition where $M_{RS}^{(NL,\epsilon)}/N$ is discontinuous across $T_c$ in the $N\rightarrow \infty$ limit.

\subsubsection*{Local quench}
In the following sections, we consider several time-dependent scenarios and analyze their implications on the system dynamics.

For our first scenario, let's examine a CFT ground state that's been perturbed by a smeared local operator $O_{\alpha}(x,0)$ at $t=0$.  This is then subjected to time evolution governed by the CFT Hamiltonian. We can express the state as:
\begin{equation}
    |\psi(t)\rangle=\mathcal{N}e^{-iHt}e^{-\delta H}O_{\alpha}(x,t)|\Omega\rangle. 
\end{equation}

In the corresponding bulk dual, this equates to introducing an in-falling particle with mass $m$ into the initially vacuum AdS spacetime. The energy-momentum tensor for this scenario can be characterized as:  
\begin{equation}
    T_{uu}=\frac{mR\alpha^2}{8\pi (u^2+\alpha^2)^2}. 
\end{equation}

Here, $\alpha$ denotes the size of the smeared operator. As 
$\alpha$ approaches 0, this converges to a delta function in $u$. The subsequent effect on the bulk spacetime is encapsulated by a shock-wave geometry, as illustrated below in \figref{figholo1}.

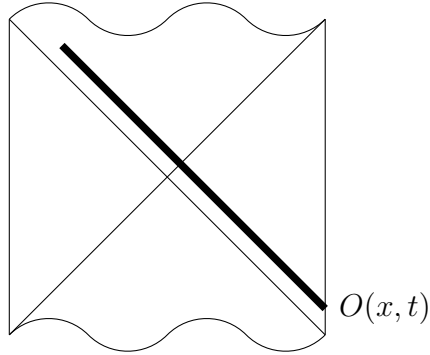
\begin{figure}[H]
\center
\begin{tikzpicture}[scale=0.7]
\draw (-3,-3)--(-3,3);
\draw (3,-3)--(3,3);
 \draw (-3,3) to[out=45,in=135] (-1.5,3) to[out=-45,in=-135] (0,3) to[out=45,in=135] (1.5,3) to[out=-45,in=-135] (3,3);
\draw (-3,-3) to[out=45,in=135] (-1.5,-3) to[out=-45,in=-135] (0,-3) to[out=45,in=135] (1.5,-3) to[out=-45,in=-135] (3,-3);
\draw (-3,3)--(3,-3);
\draw (-3,-3)--(3,3);
\draw[line width=1mm] (3,-2.5)node[right]{$O(x,t)$} --(-2,2.5);
\end{tikzpicture}
\caption{Penrose diagram depicting a shock wave in global coordinates. }
\label{figholo1}
\end{figure}

We aim to investigate the non-local magic of subsystem 
$A$ in relation to $A^c$. These subsystems are separated by the boundary $\partial A=\partial A^c$, a 
$d-2$ sphere of radius $l$. By solving the Einstein equation, \cite{Nozaki:2013wia} derived the leading-order change in entanglement entropy due to the injected energy. Specifically, for a (1+1)d holographic CFT, this change is expressed as:
\begin{equation}
    \Delta S(t)=\frac{2mRl\alpha+mR(l^2-\alpha^2-t^2)\arctan(\frac{2\alpha l}{t^2+\alpha^2-l^2})}{8l\alpha}+O\left((mR)^2\right).
\end{equation}

We can then employ the lower bound to estimate the growth of the non-local magic as follows:
\begin{equation}
\begin{split}
    M^{(NL)}_{RS}(|\psi(t)\rangle)\sim& \sqrt{S(0)+\Delta S(t)}\\
    \approx& \sqrt{S(0)}+\frac{1}{2}\frac{\Delta S(t)}{\sqrt{S(0)}}. 
\end{split}
\end{equation}

In the early-time regime,  $t\ll \sqrt{l^2-\alpha^2}$, the magic exhibits quadratic growth with time, independent of the spacetime dimension. This can be expressed as: 
\begin{equation}
    \Delta M^{(NL)}_{RS}(t)\sim \kappa_d\frac{mR}{\sqrt{S_0}}(\frac{\alpha l}{l^2-\alpha^2})^2\frac{t^2}{l^2-\alpha^2}+O(\frac{t^4}{(l^2-\alpha^2)^2}).
\end{equation}

At $t=\sqrt{l^2-\alpha^2}$, the magic reaches its peak value of $\Delta M^{(NL)}_{RS}=\kappa_d\frac{mR}{\sqrt{S_0}}$, after which it declines to zero. In the long-term regime, it decays following a power-law pattern: 
\begin{equation}
    \Delta M^{(NL)}_{RS}(t)\sim \frac{mR}{\sqrt{S_0}}\left(\frac{\alpha l}{t^2}\right)^d\left(1+O(\frac{l^2-\alpha^2}{t^2})\right).
\end{equation}

For the (1+1)d CFT, another intriguing scenario arises when subsystem $A$ encompasses half of the space, signifying $l\rightarrow \infty$. In this context, there exists a range in which the magic grows logarithmically with $t$ \cite{Caputa:2014vaa}, specifically when $l\ll t\ll D^{1/mR}\alpha$,
\begin{equation}
    \Delta M^{(NL)}_{RS}(t)\sim \frac{mR}{\sqrt{S_0}}\log\frac{t}{\alpha},
\end{equation}
where $D$ is quantum dimension of the quench operator $O$. The value reaches a constant late-time limit of 
$\Delta M^{(NL)}_{RS}=\frac{\log D}{\sqrt{S_0}}$. This logarithmic growth can only be observed in system with large central charge due to the otherwise small value of $D$.
Note that holographic methods are at a distinct advantage here because magic dynamics for large systems over long periods of time is numerically intractable using existing methods.

\subsubsection*{Global Quench}
We also explore the global quench scenario wherein the perturbation isn't confined to a localized region but influences the entire CFT state. Within the bulk dual, this corresponds to a spherically symmetric in-falling mass shell.

\begin{figure}[H]
\center
\begin{tikzpicture}
\draw (1,-5)--(1,1);
\draw (-1,1) to[out=45,in=135] (-0.5,1) to[out=-45,in=-135] (0,1) to[out=45,in=135] (0.5,1) to[out=-45,in=-135] (1,1);
\draw[line width=1mm] (1,-1)node[right]{$\int dx O(x,t)$} --(-1,1);
\draw (-1,1)--(-3,-1)--(1,-5);
\end{tikzpicture}
\caption{Vaidya geometry. Right side boundary denotes the asymptotic boundary of the AdS-Vaidya spacetime. Outside the mass shell is the black hole geometry. Inside the mass shell is  Vacuum AdS in Poincaré patch.}
\label{figholo2}
\end{figure}
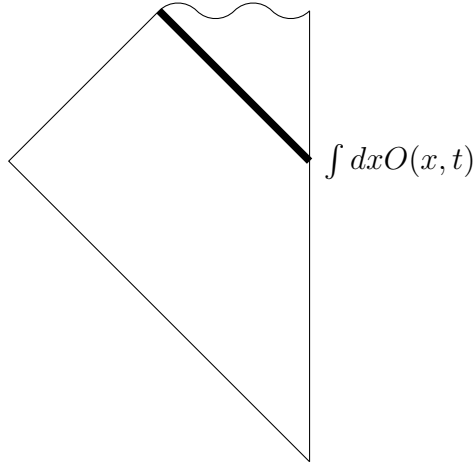

The geometry impacted by the mass shell is characterized by the Vaidya metric. This is essentially the integration of pure AdS with an AdS-Schwarzschild black hole, aligned along the mass shell, as illustrated in~\cref{figholo2}.

The shell's descent into the bulk parallels the boundary CFT's thermalization following the global perturbation. The state transitions from the ground state and progressively thermalizes to a certain finite temperature. The entanglement entropy of subregion $A$ serves as a quantitative measure, increasing during this process. Correspondingly, in the bulk perspective, this entropy surge is represented by the expanding area of the minimal surface anchored to the boundary of $A$. 

In a (1+1)-dimensional CFT, it's feasible to precisely solve for the minimal surface \cite{Balasubramanian:2011ur}. The entropy at 
$t=0$ is equivalent to the CFT ground state entropy, given by $S(0)=\frac{c}{3}\log{\frac{l}{\delta_{UV}}}$. This aligns with the length of the geodesic fully contained within the pure AdS. Following the onset of the quench, the geodesic begins to intersect with the in-falling mass shell, causing its length to increase over time. Initially, this growth is quadratic with respect to $t$,
\begin{equation}
    \mathcal{L}(t)=2\log{\frac{l}{\delta_{UV}}}+2\frac{\pi^2 t^2}{\beta^2}+O(t^3). 
\end{equation}

As thermalization progresses, the geodesic's intersection with the mass shell delves deeper into the bulk. Once the subregion completes its thermalization at time $t=\frac{l}{2}$, the geodesic no longer intersects the in-falling shell, stabilizing its length to an equilibrium value,
\begin{equation}
    \mathcal{L}(t>\frac{l}{2})=2\log{\frac{ \beta}{\pi\delta_{UV}}\sinh{\frac{\pi l}{\beta}}}. 
\end{equation}

We also detail the behavior of the geodesic length in the late stages, prior to reaching full thermalization, as outlined below:
\begin{equation}
    \mathcal{L}(t\lesssim\frac{l}{2})=2\log{\left(\frac{ \beta}{\pi\delta_{UV}}\sinh{\frac{\pi l}{\beta}}\right)}-\frac{2}{3}\sqrt{2\tanh{\frac{\pi l}{\beta}}}\left(\frac{l}{2}-t\right)^{\frac{3}{2}}+O\left(\left(\frac{l}{2}-t\right)^{2}\right).
\end{equation}

Based on the aforementioned results, the evolution of the smoothed non-local magic for a subregion in a 2d CFT can be characterized as follows: it increases according to, 
\begin{equation}
    \sma t \sim\sqrt{c\abs{\log\epsilon}}\left(\sqrt{S_0/c}+\frac{\pi^2 t^2}{6\sqrt{S_0/c}}+O(t^3)\right),
\end{equation}
during the initial stages, and as,
\begin{equation}
    \sma{t}\sim\sqrt{c\abs{\log\epsilon}}\left(\sqrt{S_T/c}-\frac{1}{18}\frac{\sqrt{2\tanh{\frac{\pi l}{\beta}}}}{\sqrt{S_T/c}}\left(\frac{l}{2}-t\right)^{\frac{3}{2}}+O\left(\left(\frac{l}{2}-t\right)^{2}\right)\right),
\end{equation}
during the latter phases when the subregion is nearing full thermalization. This can be contrasted with the dynamics of total subsystem magic under thermalization\cite{Sewell} which decays after a quick initial rise.

\subsubsection*{Wormhole}
Lastly, we examine a thermalization process involving two copies of CFT states. This dynamic process corresponds to the evolution of an expanding wormhole in the bulk dual. 

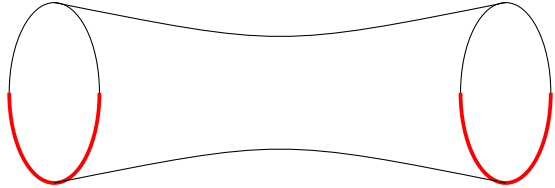
\begin{figure}[H]
    \centering
    \begin{tikzpicture}[scale=0.6]
        \def\s{10}
        \draw (0,0) ellipse (1cm and 2cm);
        \draw (\s,0) ellipse (1cm and 2cm);
        \draw[red,line width=0.5mm] (-1,0) arc(180:360:1cm and 2cm) ;
        \draw[red,line width=0.5mm] (\s-1,0) arc(180:360:1cm and 2cm);
        \draw (0,2) ..controls (\s/2,1).. (\s,2);
        \draw (0,-2).. controls (\s/2,-1).. (\s,-2);
    \end{tikzpicture}
    \caption{Wormhole geometry}
    \label{figholo3}
\end{figure}

Let us revisit the thermal-field-double (TFD), 
\begin{equation}
    \ket{\text{TFD}}=\frac{1}{\sqrt{Z(\beta)}}\sum_n e^{-(\frac{\beta}{2}+2it) E_n}\ket{E_n}_L\ket{E_n}_R
\end{equation}

We designate our region of interest to encompass a section from both the left and right CFT states (illustrated in~\cref{figholo3}). The entanglement entropy of this composite region is probed by the extremal surface spanning the wormhole, connecting the left and right segments.

In this setup, we assume symmetry when exchanging the two CFT sides. Specifically, we mandate that the subregion $A$ on one side mirrors its counterpart on the other side. See red region in~\cref{figholo3}. Given this symmetry, the extremal surface occupies a plane defined by constant transverse spatial coordinates and is characterized solely by the relationship between time and the radial direction.

At the boundary time $t=0$, the area of extremal surface is given by
\begin{equation}
    \mathcal{A}(0)=\frac{\beta r_{\infty}}{\pi}V_{d-2}.
\end{equation}
where $r_{\infty}$ is the UV cutoff of radial coordinates. This extremal area is proportional to the volume of subregion boundary $\partial A$, reminiscent of the area law entanglement observed in gapped systems.  As time progresses, the extremal surface accrues additional contributions from regions beyond the horizon. As highlighted in \cite{Hartman:2013qma}, this contribution exhibits a straightforward linear relationship with the boundary time, as illustrated below:
\begin{equation}
    \mathcal{A}(t)= \frac{4\pi t}{\beta}\alpha_d V_{d-2},   \qquad \text{for $t\gg \beta$}.
\end{equation}

The linear growth eventually ceases when the extremal surface traversing the wormhole is surpassed by another, more minimal configuration. A different set of competing extremal surfaces, anchored to the same entangling boundary but bypassing the wormhole, emerges. These surfaces are essentially combinations of the extremal surfaces corresponding to subregions within each individual thermal CFT. Their area is given by
\begin{equation}
    \mathcal{A}(\infty)-\mathcal{A}(0)=\frac{2\pi}{\beta}V_{d-1}.
\end{equation}

The transition of dominant extremal surface occurs around $t\sim R$, which corresponds to the size of the subregion under consideration. Consequently, we anticipate the non-local magic in this TFD state to scale as follows:
\begin{equation}
\begin{split}
\sma{t}&\sim\sqrt{\abs{\log{\epsilon}}}\sqrt{S_0+\frac{4\pi t}{\beta}\alpha_dV_{d-2}}, \qquad \text{for $\beta\ll t<R$}\\
&\sim \sqrt{\abs{\log{\epsilon}}S_T},  \qquad \text{for $t\geq R$}.
\end{split}
\end{equation}

\section{Discussion}
In this work, we explored the question: what dual boundary quantity enables gravitational back-reaction in the bulk? The celebrated formula of Ryu and Takayanagi provides a fundamental observation of the AdS-CFT conjecture by showing that areas in AdS correspond to entanglement entropies in the CFT. 
In the greater context of spacetime and gravity emerging from quantum information, we ask: If entanglement builds geometry, then what builds gravity?
In this work we show that the strength of gravitational back-reaction is connected to (non-local) magic in CFT. In other words, gravity is magical! Accordingly, both defining properties of quantumness admit holographic counterparts in AdS. 

To obtain this result, we 
studied the interplay between non-local magic and entanglement. We show that for any quantum state in a finite dimensional Hilbert space, this form of non-stabilizerness that can only live in the bipartite correlations is lower bounded by the anti-flatness of the entanglement spectrum and upper bounded by the amount of entanglement in the system as defined by R\'enyi entropies. We then apply these results to CFTs and conclude that both the exact and smoothed non-local magic is proportional to various notions of anti-flatness. However, they scale differently with entropy --- the exact non-local magic scales linearly with the von Neumann entropy of a CFT subregion while the smoothed magic only scales as the square root. Numerically we verify that non-local magic is sensitive to quantum phase transition in a way that is different from entanglement. We also examined its behaviour under symmetry breaking. 

Finally, in the context of holographic CFTs, we derive a quantitative relation between non-local magic and the level of gravitational back-reaction. Using the bulk gravity theory, smoothed non-local magic in the CFT can also be estimated holographically. As non-stabilizerness in quantum systems are generically hard to compute, our work also provides an important estimate on the practical level and constrain magic distributions using existing data and well-founded methods like tensor networks and DMRG.

There are several directions that are of interest for future work. The key constraints for non-local magic here are given in terms of inequalities. Part of the reason for bounds instead of a precise equality is that non-local magic requires extremization while the computation of magic itself is already non-trivial. However, given the universal behaviour of non-local magic across multiple distinct measures of anti-flatness, there is reason to believe that a unifying statement or even a precise equality exists between entanglement spectral properties and magic. In the particular case of quantum field theory, it is also crucial to generalize our observations to definitions of magic that is native to the infinite dimensional system, e.g. non-Gaussianity, as well as other different measures of spectral anti-flatness. 

Approaching non-local magic from a different perspective, we can start with state $\rho_A$ from the usual stabilizer polytope and construct a purified state $\psi_{AB}$. One can also define a non-local magic as the minimal magic among all possible purifications.
In the same vein of connecting magic with entanglement, we ask whether it is possible to define instead magical entanglement, i.e., the entanglement that cannot be removed by any Clifford operation\footnote{We thank Kaifeng Bu for this suggestion.}. In this case, one can easily show from our entropy bounds that magical entanglement is an upper bound of non-local magic. However, it is yet unknown whether the two definitions are equivalent. Finally, recall that non-local magic can be generalized to systems with multi-partite entanglement. This will be crucial in understanding the behaviour of e.g. Haar random states, random tensor networks, and holographic states. As the type of multi-partite entanglement is quite constrained for stabilizer states, non-local magic may be crucial in the classification of multi-partite entanglement.

For CFTs specifically, several of our results rely on the assumption that the bipartite entanglement across a subregion $A$ and its complement $B$ in a pure state can be approximately converted into a tensor product of entangled pairs through unitaries that only act on the respective subregions. Although this assumption is well-supported by numerics and well-motivated by holographic tensor networks models, it is unclear the extent to which this holds for a single copy of (holographic) CFT in general. This assumption may also admit further modification in the case where $A$ consists of multiple disjoint regions. It is important that we understand the regime of validity for such assumptions and pave the way for proving \cref{conj:cft} and extending the generality of \cref{prop:flatsmoothNLM}. 

Just like various types of entanglement can admit different holographic interpretations, a similar situation may hold for magic. Although we take a first step towards addressing the open question of what is the holographic dual of magic, much remains unknown. For instance, the connection we identify with anti-flatness signals a link between non-local magic and gravitational back-reaction. However, because we lack a systematic understanding of how the bulk duals should deform under a sequence of boundary theories that have increasing flat spectrum, the physical meaning of how the removal of magic turns off backreaction is unclear. Additionally, it is possible that other gravitational phenomena generating backreaction have a different magical origin on the boundary. If that is the case, we eventually wish to distinguish them from the consequence of bipartite non-local magic in the boundary theory.

Although it has been suggested that boundary states with flat entanglement spectrum are dual to peculiar bulk states of fixed areas\cite{Akers:2018fow,Dong_2019}, exactly how these bulk states should be interpreted holographically remains to be understood. To this end, a more precise relation between magic and emergent gravity \cite{Faulkner_2013} in the bulk, one which does not rely on the distillation assumptions used in this work, is highly desirable. Furthermore, a connection between magic and a local function of curvature generated by more physical forms of stress energy instead of an extended conical singularity such as a cosmic brane may provide a more natural link with the Einstein's equations or the Hamiltonian constraint. A more comprehensive understanding of holographic magic through the lens of dynamics such as quantum chaos\cite{chaosbymagic} and (classical) complexity can also provide another unique perspective that is not captured by our current work.

More broadly, this work calls for several important lines of investigation as we move towards establishing non-local magic as a key metric for characterizing quantum many-body systems. For instance,  the tensor product of random single-qubit states, the ground states of physical quantum many-body systems, and the Haar random states all have volume law magic scaling. Purely from the point of view of entanglement entropy, they can also be mimicked by stabilizer states. However, their non-local magic behaves very differently. Thus it provides a distinct indicator for the properties of the underlying quantum systems that are invisible to entanglement entropy or total non-stabilizerness alone. It would also be intriguing to study the role of non-local magic in quantum phase transition, in symmetry breaking, and in non-equilibrium systems.





\textit{We would like to thank Chris Akers, Vijay Balasubramanian, Ning Bao, Ed Barnes, Kaifeng Bu, Xi Dong, Sophia Economou, Monica Kang, Cynthia Keeler, Nick Mayhall, Jason Pollack, Howard Schnitzer, Brian Swingle, Christopher White, Tianci Zhou for helpful comments, resource, and discussions. We are especially grateful to Christopher White and Daniele Iannotti for identifying the mistakes in the earlier version of this manuscript. C.C. and A.H. would like to thank the organizers of the Quantum Information and Quantum Matter Conference at NYU Abu Dhabi during which this work was first conceived.
 C.C.\ acknowledges the support by the Air Force Office of Scientific Research (FA9550-19-1-0360), the National Science Foundation (PHY-1733907), and the Commonwealth Cyber Initiative. The Institute for Quantum Information and Matter is an NSF Physics Frontiers Center.
AH acknowledges support from PNRR MUR project PE0000023-NQSTI and
PNRR MUR project CN $00000013$ -ICSC. W.M. is supported by the U.S. Department of Energy under grant number DE-SC0019470, and by the Heising-Simons Foundation ``Observational Signatures of Quantum Gravity'' collaboration grant 2021-2818. S.F.E.O. acknowledges support from PNRR MUR project PE0000023-NQSTI. L.L. is funded through the Munich Quantum Valley project (MQV-K8) by Bayerisches Staatsministerium für Wissenschaft und Kunst and DFG (CRC 183).}

\clearpage
\begin{singlespace}
\printbibliography[heading=subbibliography]
\end{singlespace}

\chapter{CONCLUSION}\label{Chapter8}

We began by discussing OPE spectral densities of conformal field theories in the large exchange dimension limit. Utilizing the modular bootstrap technique, we determine Virasoro vacuum contributions to the OPE spectrum in the lightcone limit. We then extend beyond this limit, towards $h/\mathcal{C} \sim 1$, by proposing an ansatz solution to the crossing equations in the lightcone limit. Despite lacking a closed form for the Virasoro blocks in the large spin limit, our kernel ansatz method is sufficient to numerically explore the density of OPE coefficients to leading-order.

Although a gravitational interpretation of the lightcone modular bootstrap was not directly explored in this work, we anticipate that our results recover semi-classical gravity in $AdS_3$, even for finite values of central charge. The four-point functions we analyzed may be equivalently represented by bulk Witten diagrams, enabling a gravitational parallel for the calculations herein. In this vein, our results provide a useful model uncovering gravitational descriptions for coarse-grained CFTs, or sets thereof. For non-minimal CFTs in particular, it would be interesting to explore constraints on such theories which retain a gravitational dual throughout the coarse-graining process.

We next transition towards an exploration of stabilizer states, and ultimately an investigation of quantum group action on state entropy vectors. We explicitly generate sets of stabilizer states and compute their associated entropy vectors, demonstrating the connectivity of these states using reachability graphs. For fixed qubit number, our reachability graphs offer a complete description of entanglement entropy distribution and state proximity under the Clifford group in the Hilbert space. We further introduce a procedure for building restricted graphs, where only group action under certain Clifford subgroups is considered, demonstrating how higher-qubit graphs are composed of attached lower-qubit subgraphs. This restricted graph analysis significantly extends the tractability of direct analysis to higher qubit number, and even provides a general understanding of certain reachability graphs at arbitrary qubit number.

Our graph protocol offers a novel way of studying state parameter evolution under group action on the Hilbert space. In this work we focus on the evolution of entanglement entropy, observing the possible entropy vector transformations under Clifford group action. Beginning at $4$ qubits we observe stabilizer states with entropy vectors that violate known holographic inequalities, specifically the monogamy of mutual information. Accordingly, the $n \geq 4$ stabilizer state restricted graphs enable us to directly identify quantum circuits which evolve a state with a holographic entropy vector to one which violates MMI. Furthermore, since we restrict to stabilizer states and Clifford circuits, this construction easy to reproduce using near-term quantum computers.

Extending beyond the scope of stabilizer states we derive and analyze the entanglement structure of $n$-qubit Dicke states, constructing the Dicke state entropy cone. We further define a min-cut protocol on weight directed graphs which reproduces Dicke state entropy vectors, and verify that such vectors violate both the holographic and symmetrized holographic entropy conditions. We explore the evolution of Dicke state entropy vectors under Clifford group action, identifying the stabilizing operations for all Dicke states and building their associated reachability graphs. Finally, we use our graph model to demonstrate strict bounds on Dicke state entropy vector evolution under Clifford group action.

The symmetric entanglement structure of Dicke states affords a particular robustness to single-qubit loss. This characteristic renders large-qubit Dicke states an excellent option for logical encoding, as well as a preferred initial state for many quantum optimization algorithms. Further studying into the Dicke state entropy cone, and more importantly the evolution of Dicke state entropy vectors under larger sets of quantum operations, may reveal additional utility for near-term quantum algorithms. In addition, select Dicke states possess a significant amount of distillable magic. The ease of preparation and magic content of Dicke states suggests such states as a natural candidate for preliminary experimental magic distillation schemes.  

Moving towards a more general description of Clifford group action on arbitrary quantum states, we introduce a protocol to construct reachability graphs and restricted graphs as Cayley graph quotients. In this abstraction, identifiable relations among different quantum circuits are equally valid when applied to arbitrary quantum states. Furthermore generalized properties of parameter evolution, e.g. classes of circuits which must preserve the entropy vector of any state, can likewise be studied. This quotient protocol is not limited to the Clifford group, but can be identically applied to any discrete and finite set of operators on a Hilbert space. Upon relaxing the conditions of discreteness and finiteness, a corollary analysis for infinite or continuous gate sets may be performed to establish maximal bounds or rates of parameter evolution under quantum operators.

One natural extension of this work would consider Cayley graph quotients for systems of qudits, or perhaps an analogous treatment for generalized Clifford algebras. Alternatively we could consider intermediate magic state injection to otherwise Clifford circuits, thereby enabling controlled non-Clifford behavior in the circuit. Such an analysis would particularly benefit near-term quantum computing architecture, as non-Clifford evolution is expected to be severely limited. Finally we could consider alternative coupling maps which render certain gate actions easier (or harder) to perform than others. In a graph description, this could correspond to promoting our Cayley graph quotients to weighted graphs, where the difficulty of implementing a certain gate fixes the corresponding edge weight.

Applying our quotient graph protocol to study parameter evolution for arbitrary quantum states, we introduce contracted graphs. Defined as a graph quotient on reachability graphs, under operations which leave a certain parameter invariant, these contracted graphs enable us to derive explicit bounds on parameter evolution under a chosen group. More rigorously defined, upon choosing a group of operators to act on a Hilbert space and some computable state parameter to study, a contracted graph represents the coset space of the group with equivalence classes that leave the chosen property invariant. Furthermore by identifying the reachability graph for a family of states as well as some chosen state property, the contracted graph indicates the double coset space of group elements which fix both the family of states and the property of interest.

In this work we specifically focus on the evolution of entropy vectors, and derive a strict upper bound on the number of possible entanglement entropy configurations which can be achieved under the action of a chosen group of operators. This maximal number of entropy configurations is set by the vertex count of the associated contracted graph, and applies generally to any initial state acted on by the group. If the group is finite, a maximal number of achievable entanglement configurations can be directly identified. If however the group is infinite, a fixed number of configurations can be derived up to arbitrary circuit depth, or a rate of entanglement evolution can be established as the evolution becomes continuous.

The contracted graph techniques presented are sufficiently general to study the evolution of any computable state parameter under any chosen gate set. Natural extensions of this work include exploring the dynamics of additional state properties, such as different magic measures or stabilizer rank. Furthermore an in-depth investigation into holography-violating circuits could yield new insight into entropy vector transformation under various sets of unitaries. Primarily we wish to focus on Clifford circuits which evolve a holographic entropy vector into one that violates MMI, particularly since, thus far, all observed non-holographic stabilizer states violate (at least) MMI.

Finally we extend a more general discussion away from stabilizer states and entanglement entropy towards the holographic implications of quantum magic. We demonstrate that the presence of non-local magic in a boundary CFT yields gravitational back-reaction in the bulk. We show that non-local magic is proportional to different measures of anti-flatness, and how this anti-flatness of the state results in back-reaction on the cosmic branes that represent Renyi entropies of boundary subregions. Additionally, we numerically verify our conjectures and discuss the sensitivity of non-local magic to quantum phase transitions in the boundary CFT.

We could consider extrapolating our results to explore the entanglement in magic states which cannot be removed by Clifford operations. Such a entanglement structure contributes to the presence of non-local magic, and may admit some novel holographic interpretation. While we show that the existence of non-local magic in a boundary CFT generates bulk back-reaction, it remains unknown precisely how much of that gravitational back-reaction can be attributed to this boundary origin. Perhaps most-desired is a refined connection between boundary magic and bulk local curvature by more tangible sources. Relating canonical stress-energy spacetime deformations in AdS to magic content in the boundary CFT would provoke a holographic connection between magic and Einstein's equations.  

Throughout this work we traverse both sides of the AdS/CFT correspondence. We discuss different prospects for exploring gravitational phenomena using properties of quantum systems and implications derived from quantum information theory. We highlight the fascinating overlap of entanglement and magic with geometry and gravitational fluctuations, made possible through the mathematical underpinnings of holography. In anticipation of realizable quantum computing, we remark on the practical utility of our results for improving quantum algorithms and a general understand of quantum computing architecture. While much work has been done towards each of these efforts, much more remains to be understood. This dissertation constitutes only a small collection of efforts to probe the mysteries of quantum gravity.

\clearpage

\begin{singlespace}
\addcontentsline{toc}{part}{REFERENCES}
\newrefsection
\saverefsection
\newrefcontext[sorting=nty]
\nocite{\ReferencesList}
\centering{Bibliography}
\printbibliography[heading=none]
\end{singlespace}

\addcontentsline{toc}{part}{APPENDIX}
\renewcommand{\cftlabel}{APPENDIX}
\renewcommand{\chaptername}{APPENDIX}
\appendix
\chapter{ADDITIONAL GRAPHS}\label{Two-Qubit Graphs}
\newpage
\noindent

Contained below are some graphs not featured in the body of this paper. Futher graphs are available in the repository \cite{github}.

 The 2-qubit complete reachability graph and some 2-qubit restricted graphs were discussed in section \ref{sec:two}. Here we compile additional restricted graphs to further illustrate the relation between states under action $C_2$ subgroup action. Figure \ref{TwoQubitP1P2} restricts to only phase operations, while Figure \ref{TwoQubitCNOT12CNOT21} shows the restricted graph for only CNOT operation.

Figure \ref{TwoQubitP1P2} shows the 2-qubit graph restricted to only phase operations. The graph consists of $5$ distinct and disconnected substructures.

\begin{figure}[h]
\begin{center}
\includegraphics[scale=0.82]{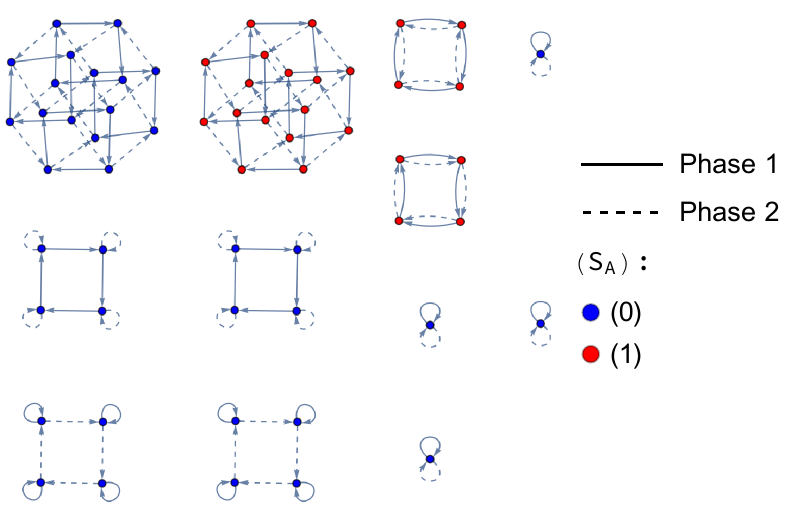}
\caption{The 2-qubit $P_1,P_2$ restricted graph contains $5$ unique substructures. The $4$ isolated points are states on which both $P_1$ and $P_2$ act trivially. The box-like structures come in two varieties. The box of unentangled states has a trivial loop at each corner, while the boxes of entangled states witness degenerate action instead. There exist two largest structures of states (top-left) on which both phase gates act non-trivially and non-degenerately.}
\label{TwoQubitP1P2}
\end{center}
\end{figure}

Figure \ref{TwoQubitCNOT12CNOT21} shows all interactions between 2-qubit states under only CNOT operations. The CNOT gate can modify entropy structure, therefore we witness the first occurrences of states with different entropy vectors lying in the same substructures. Alternating action of $CNOT_{1,2}$ and $CNOT_{2,1}$ has a maximum cycle of $6$, seen in the hexagonal structures top-left.

\begin{figure}[h]
\begin{center}
\includegraphics[scale=0.5]{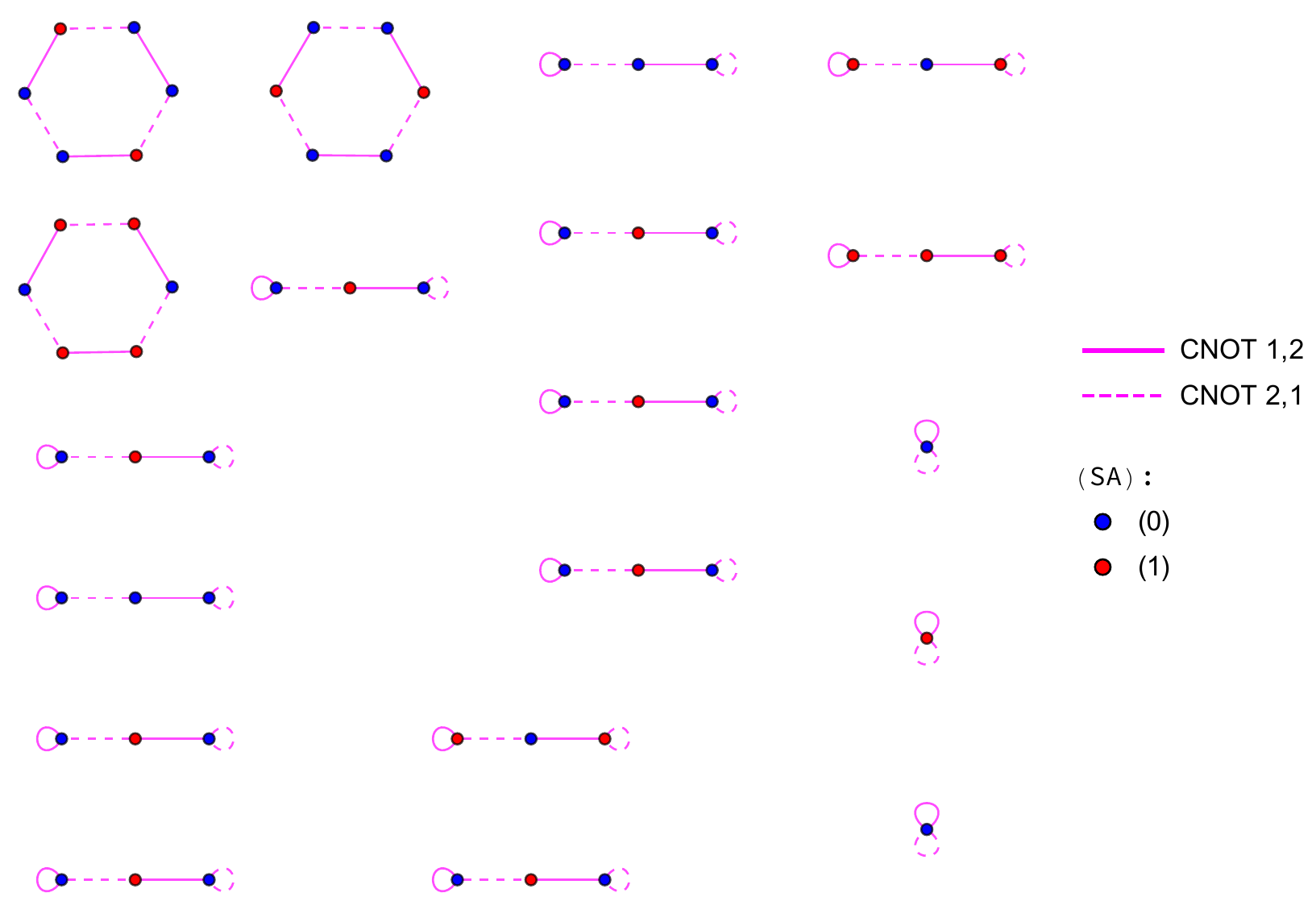}
\caption{The subgroup generated by $CNOT_{1,2}$ and $CNOT_{2,1}$ partitions the set of 2-qubit states into $3$ graph substructures. The isolated points are states on which both CNOT gates act trivially. The linear triplets are built of one state, on which both CNOT gates act non-trivially, connected to two states on which opposing one CNOT gate acts trivially. These triplets can occur with states of similar or different entropic structure. The largest structure is a hexagon which illustrates the maximum cycle of the subgroup generated by $CNOT_{1,2}$ and $CNOT_{2,1}$.}
\label{TwoQubitCNOT12CNOT21}
\end{center}
\end{figure}

\newpage

 The 3-qubit reachability diagrams were discussed in section \ref{sec:three} with a focus on the $H_1,H_2,CNOT_{1,2},CNOT_{2,1}$ restricted graph (Figure \ref{ThreeQubitH1H2CNOT12CNOT21}). Here we provide additional graphs of potential interest, including the complete reachability graph on three qubits (Figure \ref{ThreeQubitCompleteGraph}). Figure \ref{ThreeQubitP1P2P3Subgraphs} displays the action of all phase gates on three qubits, generalizing the cycles and structures seen at two qubits (Figure \ref{TwoQubitP1P2}) to three qubits.

Figure \ref{ThreeQubitCompleteGraph} displays the full reachability graph for three qubits (trivial loops removed). This graph contains all non-trivial information about 3-qubit interaction under operations of the Clifford group ($C_3$).

\begin{figure}[h]
\begin{center}
\includegraphics[scale=0.6]{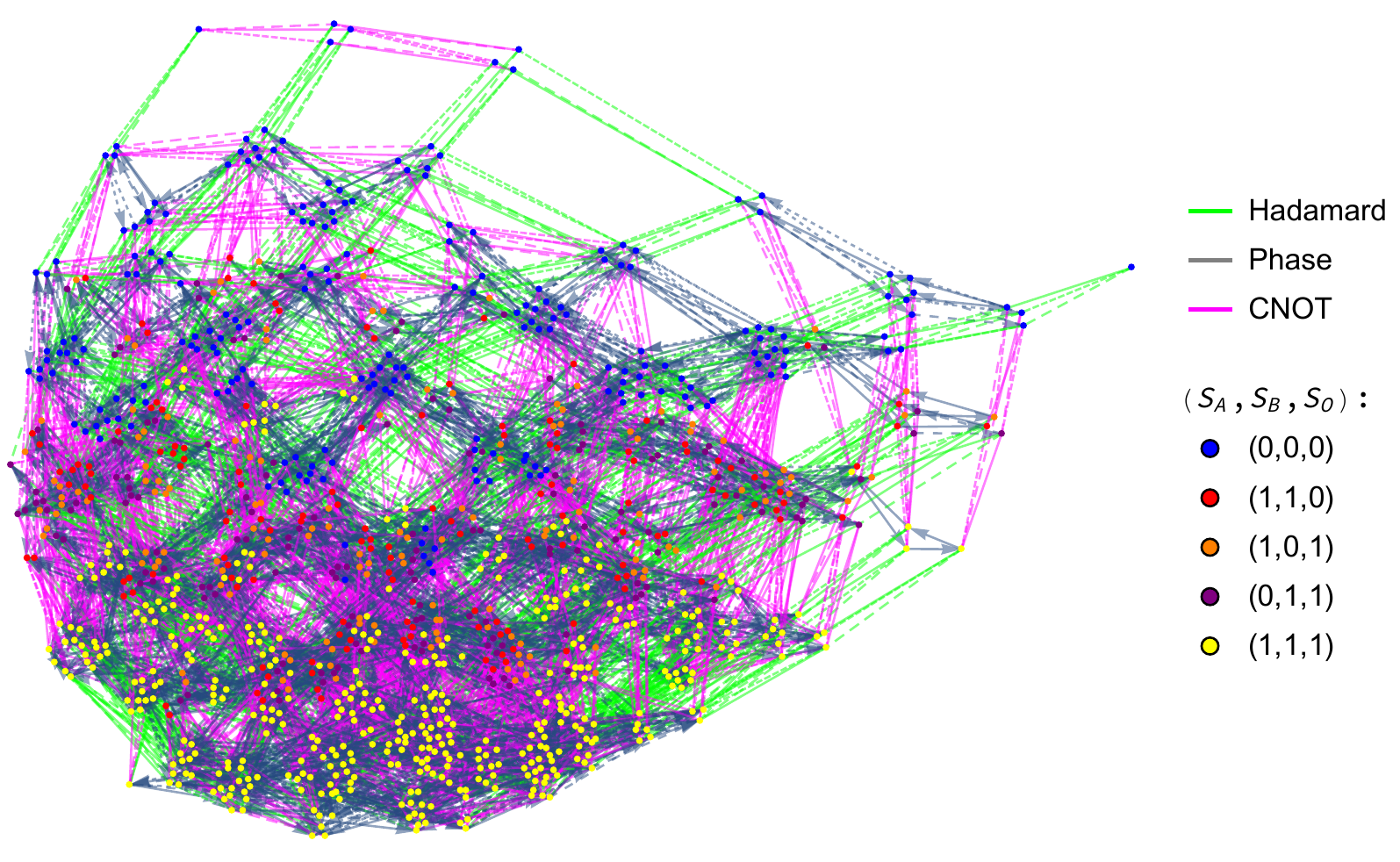}
\caption{Complete reachability diagram on three qubits with trivial loops removed. The Hadamard and phase gates act individually on all three qubits, while the CNOT gate acts on any pair of qubits. Line texture indicates the particular action, e.g. a solid line for $H_1$ and medium dashed line for $H_2$.}
\label{ThreeQubitCompleteGraph}
\end{center}
\end{figure}

\newpage

In Figure \ref{ThreeQubitP1P2P3Subgraphs} is the 3-qubit graph restricted to only phase gates. The addition of $P_3$ to the generating set allows for longer cycles, resulting in more complex structures than were witness at lower qubit number.

\begin{figure}[h]
\begin{center}
\includegraphics[scale=0.58]{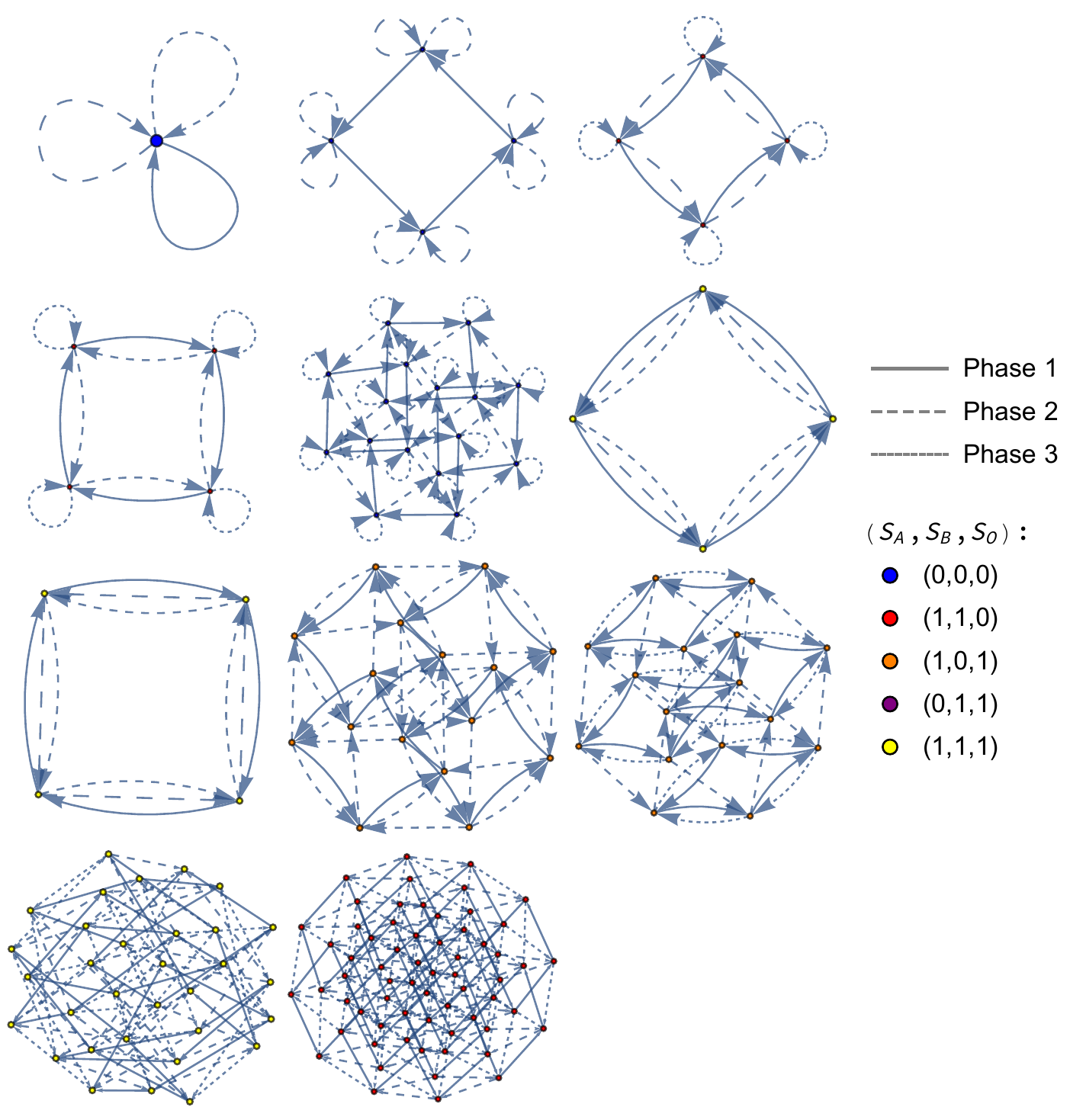}
\caption{There are $11$ unique subgraphs whose copies build the 3-qubit $P_1,P_2,P_3$ restricted graph. These subgraphs segregate according to which sequences of $P_1,P_2,$ and $P_3$ are equivalent.}
\label{ThreeQubitP1P2P3Subgraphs}
\end{center}
\end{figure}

\newpage

Figure \ref{ThreeQubitH1H2CNOT12CNOT21} shows the 3-qubit Three-qubit restricted graph displaying only Hadamard and CNOT operations on the first two qubits of in the system. The graph has repeated copies of structures found at two qubits, as well as the addition of two new subgraphs $g_{144}$ and $g_{288}$.

\begin{figure}[h]
\begin{center}
\includegraphics[scale=0.54]{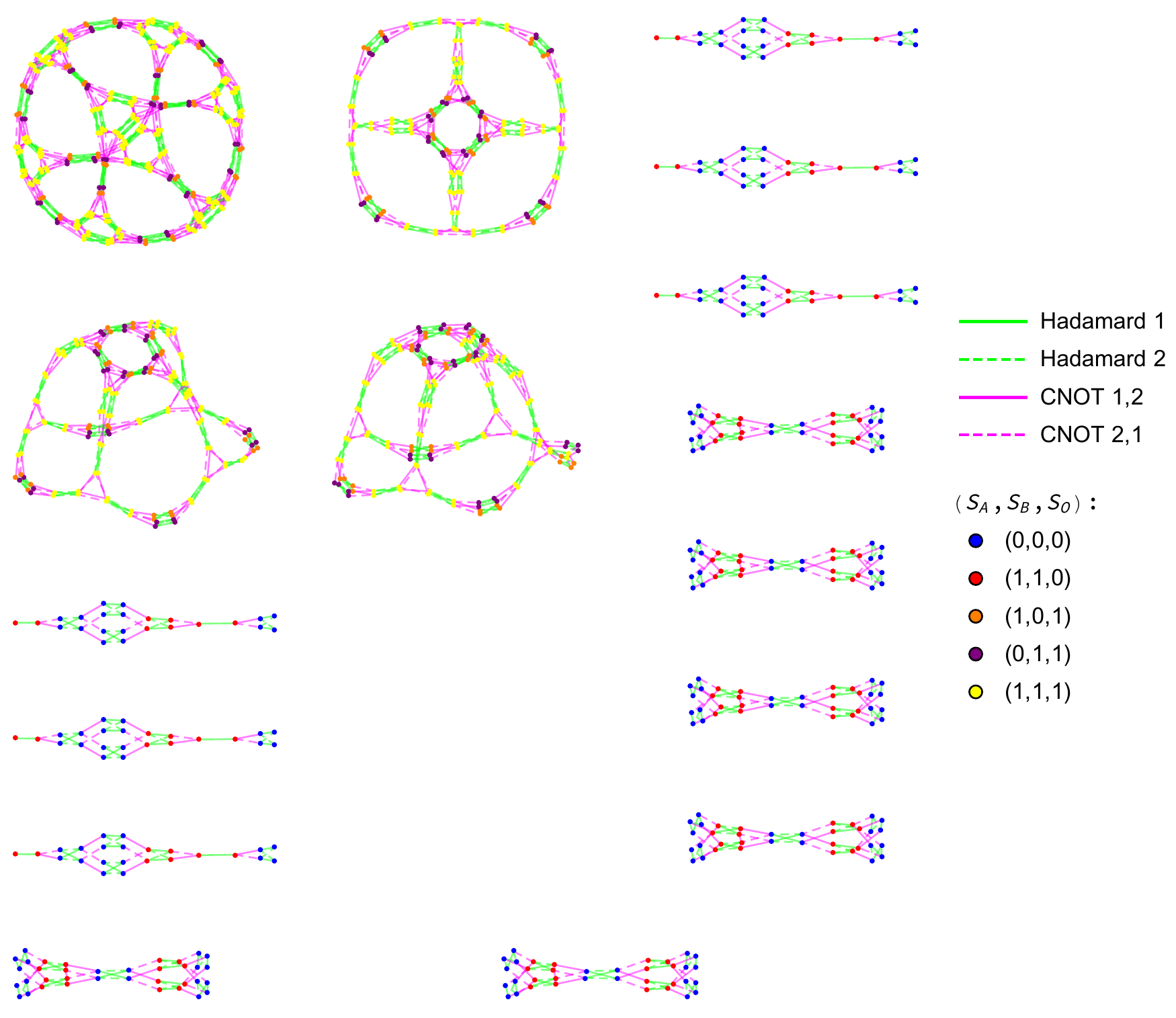}
\caption{The graph contains $6$ copies of $g_{24}$ and $g_{36}$, $3$ copies of $g_{144}$, and a single copy of $g_{288}$. The $3$ $g_{144}$ subgraphs are isomorphic, but were generated with slightly different layouts by the software used to build this graph.}
\label{ThreeQubitH1H2CNOT12CNOT21}
\end{center}
\end{figure}

Below we present the complete set of entropy vectors for the 5-qubit stabilizer set. There are $2423520$ stabilizer states at five qubits, with $93$ different entropic arrangements. There are $16$ of these $93$ entropy vectors which violate the monogamy of mutual information (Equation \eqref{eq:MonogamyofMutualInformation}), and therefore correspond to non-holographic states. 
\begin{table}[h]
    \centering
    \small{
    \begin{tabular}{|c||c|c|} 
 \hline
Holographic & Entropy Vector & Subgraph\\
\hline
\hline 
Yes & $(0,0,1,0,1,0,1,0,1,1,0,1,1,0,1)$ & $g_{24}$, $g_{36}$\\
 \hline
 Yes & $(1, 1, 1, 0, 1, 0, 2, 1, 2, 2, 1, 2, 1, 0, 1)$ & $g_{24}$, $g_{36}$\\
\hline
 \hline
 Yes & $(0, 0, 1, 1, 0, 0, 1, 1, 0, 1, 1, 0, 0, 1, 1)$ & $g_{24}$, $g_{36}$\\
 \hline
 Yes & $(1, 1, 1, 1, 0, 0, 2, 2, 1, 2, 2, 1, 0, 1, 1)$ & $g_{24}$, $g_{36}$\\
\hline
 \hline
 Yes & $(0, 0, 0, 1, 1, 0, 0, 1, 1, 0, 1, 1, 1, 1, 0)$ & $g_{24}$, $g_{36}$\\
 \hline
 Yes & $(1, 1, 0, 1, 1, 0, 1, 2, 2, 1, 2, 2, 1, 1, 0)$ & $g_{24}$, $g_{36}$\\
\hline
 \hline
 Yes & $(0, 0, 0, 0, 0, 0, 0, 0, 0, 0, 0, 0, 0,0, 0)$ & $g_{24}$, $g_{36}$\\
 \hline
 Yes & $(1, 1, 0, 0, 0, 0, 1, 1, 1, 1, 1, 1, 0,0, 0)$ & $g_{24}$, $g_{36}$\\
\hline
 \hline
 Yes & $(0, 0, 1, 1, 1, 0, 1, 1, 1, 1, 1, 1, 1,1, 1)$ & $g_{24}$, $g_{36}$\\
 \hline
 Yes & $(1, 1, 1, 1, 1, 0, 2, 2, 2, 2, 2, 2, 1,1, 1)$ & $g_{24}$, $g_{36}$\\
 \hline
 \hline
 Yes & $(0, 1, 1, 1, 1, 1, 1, 1, 1, 0, 2, 2, 2, 2, 0)$ & $g_{144}, g_{288}$\\
 \hline
 Yes & $ (1, 0, 1, 1, 1, 1, 0, 2, 2, 1, 1, 1, 2, 2, 0),$ & $g_{144}, g_{288}$\\
  \hline
 Yes & $(1, 1, 1, 1, 1, 1, 1, 2, 2, 1, 2, 2, 2, 2, 0)$ & $g_{144}, g_{288}$\\
  \hline
  \hline
 Yes & $(0, 1, 1, 1, 1, 1, 1, 1, 1, 2, 2, 0, 0, 2, 2)$ & $g_{144}, g_{288}$\\
 \hline
 Yes & $(1, 0, 1, 1, 1, 1, 2, 2, 0, 1, 1, 1, 0, 2, 2)$ & $g_{144}, g_{288}$\\
  \hline
 Yes & $(1, 1, 1, 1, 1, 1, 2, 2, 1, 2, 2, 1, 0, 2, 2)$ & $g_{144}, g_{288}$\\
\hline
\hline
 Yes & $(0, 1, 0, 1, 0, 1, 0, 1, 0, 1, 0, 1, 1, 0, 1)$ & $g_{144}, g_{288}$\\
 \hline
 Yes & $(1, 0, 0, 1, 0, 1, 1, 0, 1, 0, 1, 0, 1, 0, 1)$ & $g_{144}, g_{288}$\\
  \hline
 Yes & $(1, 1, 0, 1, 0, 1, 1, 1, 1, 1, 1, 1, 1, 0, 1)$ & $g_{144}, g_{288}$\\
 \hline
\hline
 Yes & $(0, 1, 0, 0, 1, 1, 0, 0, 1, 1, 1, 0, 0, 1, 1)$ & $g_{144}, g_{288}$\\
 \hline
 Yes & $(1, 0, 0, 0, 1, 1, 1, 1, 0, 0, 0, 1, 0, 1, 1)$ & $g_{144}, g_{288}$\\
  \hline
 Yes & $(1, 1, 0, 0, 1, 1, 1, 1, 1, 1, 1, 1, 0, 1, 1)$ & $g_{144}, g_{288}$\\
\hline
\hline
 Yes & $(0, 1, 1, 0, 0, 1, 1, 0, 0, 0, 1, 1, 1, 1, 0)$ & $g_{144}, g_{288}$\\
 \hline
 Yes & $(1, 0, 1, 0, 0, 1, 0, 1, 1, 1, 0, 0, 1, 1, 0)$ & $g_{144}, g_{288}$\\
  \hline
 Yes & $ (1, 1, 1, 0, 0, 1, 1, 1, 1, 1, 1, 1, 1, 1, 0)$ & $g_{144}, g_{288}$\\
\hline
\end{tabular}}
    \caption{At five qubits, there are $93$ stabilizer state entropy vectors (listed here and on the next two pages). Of these, $16$ correspond to non-holographic states. All non-holographic states are located on subgraphs with $4$ different entropy vectors.}
    \label{tab:5QubitEntropiesPart1}
\end{table}

\newpage

\begin{table}[h]
    \centering
    \small{
    \begin{tabular}{|c||c|c|c|} 
 \hline
Holographic & Entropy Vector & Subgraph\\
\hline
\hline
 Yes & $(0, 1, 1, 1, 1, 1, 1, 1, 1, 2, 0, 2, 2, 0, 2)$ & $g_{144}, g_{288}$\\
 \hline
 Yes & $(1, 0, 1, 1, 1, 1, 2, 0, 2, 1, 1, 1, 2, 0, 2)$ & $g_{144}, g_{288}$\\
  \hline
 Yes & $(1, 1, 1, 1, 1, 1, 2, 1, 2, 2, 1, 2, 2, 0, 2)$ & $g_{144}, g_{288}$\\
\hline
\hline
  \color{red} No & $(0, 1, 1, 1, 1, 1, 1, 1, 1, 1, 1, 1, 1, 1, 1)$ & $g_{144}, g_{288}$\\
 \hline
  \color{red} No & $ (1, 0, 1, 1, 1, 1, 1, 1, 1, 1, 1, 1, 1, 1, 1)$ & $g_{144}, g_{288}$\\
  \hline
  \color{red} No & $(1, 1, 1, 1, 1, 1, 1, 1, 1, 1, 1, 1, 1, 1, 1)$ & $g_{144}, g_{288}$\\
   \hline
  \color{red} No & $(1, 1, 1, 1, 1, 1, 2, 2, 2, 2, 2, 2, 1, 1, 1)$ & $g_{144}, g_{288}$\\
  \hline
\hline
  Yes & $(0, 1, 1, 1, 1, 1, 1, 1, 1, 2, 2, 1, 1, 2, 2)$ & $g_{144}, g_{288}$\\
 \hline
 Yes & $(1, 0, 1, 1, 1, 1, 2, 2, 1, 1, 1, 1, 1, 2, 2)$ & $g_{144}, g_{288}$\\
  \hline
  \color{red} No & $(1, 1, 1, 1, 1, 1, 2, 2, 1, 2, 2, 1, 1, 2, 2)$ & $g_{144}, g_{288}$\\
   \hline
 Yes & $(1, 1, 1, 1, 1, 1, 2, 2, 2, 2, 2, 2, 1, 2, 2)$ & $g_{144}, g_{288}$\\
\hline
\hline
Yes & $(0, 1, 1, 1, 0, 1, 1, 1, 0, 1, 1, 1, 1, 1, 1)$ & $g_{144}, g_{288}$\\
 \hline
 Yes & $(1, 0, 1, 1, 0, 1, 1, 1, 1, 1, 1, 0, 1, 1, 1)$ & $g_{144}, g_{288}$\\
  \hline
  \color{red} No & $(1, 1, 1, 1, 0, 1, 1, 1, 1, 1, 1, 1, 1, 1, 1)$ & $g_{144}, g_{288}$\\
   \hline
 Yes & $ (1, 1, 1, 1, 0, 1, 2, 2, 1, 2, 2, 1, 1, 1, 1)$ & $g_{144}, g_{288}$\\
\hline
\hline
  Yes & $(0, 1, 0, 1, 1, 1, 0, 1, 1, 1, 1, 1, 1, 1, 1)$ & $g_{144}, g_{288}$\\
 \hline
 Yes & $(1, 0, 0, 1, 1, 1, 1, 1, 1, 0, 1, 1, 1, 1, 1)$ & $g_{144}, g_{288}$\\
  \hline
  \color{red} No & $ (1, 1, 0, 1, 1, 1, 1, 1, 1, 1, 1, 1, 1, 1, 1)$ & $g_{144}, g_{288}$\\
   \hline
 Yes & $(1, 1, 0, 1, 1, 1, 1, 2, 2, 1, 2, 2, 1, 1, 1)$ & $g_{144}, g_{288}$\\
\hline
\hline
  Yes & $(0, 1, 1, 1, 1, 1, 1, 1, 1, 1, 2, 2, 2, 2, 1)$ & $g_{144}, g_{288}$\\
 \hline
 Yes & $ (1, 0, 1, 1, 1, 1, 1, 2, 2, 1, 1, 1, 2, 2, 1)$ & $g_{144}, g_{288}$\\
  \hline
  \color{red} No & $(1, 1, 1, 1, 1, 1, 1, 2, 2, 1, 2, 2, 2, 2, 1)$ & $g_{144}, g_{288}$\\
   \hline
 Yes & $(1, 1, 1, 1, 1, 1, 2, 2, 2, 2, 2, 2, 2, 2, 1)$ & $g_{144}, g_{288}$\\
\hline
\hline
  Yes & $(0, 1, 1, 0, 1, 1, 1, 0, 1, 1, 1, 1, 1, 1, 1)$ & $g_{144}, g_{288}$\\
 \hline
 Yes & $(1, 0, 1, 0, 1, 1, 1, 1, 1, 1, 0, 1, 1, 1, 1)$ & $g_{144}, g_{288}$\\
  \hline
  \color{red} No & $(1, 1, 1, 0, 1, 1, 1, 1, 1, 1, 1, 1, 1, 1, 1)$ & $g_{144}, g_{288}$\\
   \hline
 Yes & $(1, 1, 1, 0, 1, 1, 2, 1, 2, 2, 1, 2, 1, 1, 1)$ & $g_{144}, g_{288}$\\
\hline
\hline
  Yes & $(0, 1, 1, 1, 1, 1, 1, 1, 1, 2, 1, 2, 2, 1, 2)$ & $g_{144}, g_{288}$\\
 \hline
 Yes & $(1, 0, 1, 1, 1, 1, 2, 1, 2, 1, 1, 1, 2, 1, 2)$ & $g_{144}, g_{288}$\\
  \hline
  \color{red} No & $(1, 1, 1, 1, 1, 1, 2, 1, 2, 2, 1, 2, 2, 1, 2)$ & $g_{144}, g_{288}$\\
   \hline
 Yes & $(1, 1, 1, 1, 1, 1, 2, 2, 2, 2, 2, 2, 2, 1, 2)$ & $g_{144}, g_{288}$\\
\hline
\hline
   Yes & $(1, 1, 1, 0, 1, 2, 2, 1, 1, 1, 1, 2, 1, 2, 1)$ & $g_{1152}$ (2,4)\\
  \hline
   Yes & $(1, 1, 1, 0, 1, 2, 1, 1, 2, 2, 1, 1, 1, 2, 1)$ & $g_{1152}$ (2,4)\\
  \hline
   Yes & $(1, 1, 1, 0, 1, 2, 0, 1, 2, 2, 1, 0, 1, 2, 1)$ & $g_{1152}$ (4)\\
 \hline
   Yes & $ (1, 1, 1, 0, 1, 2, 2, 1, 0, 0, 1, 2, 1, 2, 1)$ & $g_{1152}$ (4)\\
 \hline
\end{tabular}}
    \label{tab:5QubitEntropiesPart2}
\end{table}

\newpage

\begin{table}[h]
    \centering
    \small{
    \begin{tabular}{|c||c|c|c|} 
 \hline
Holographic & Entropy Vector & Subgraph\\
\hline
  \hline
   Yes & $(1, 1, 0, 1, 1, 2, 1, 2, 1, 1, 1, 2, 1, 1, 2)$ & $g_{1152}$ (2,4)\\
  \hline
   Yes & $(1, 1, 0, 1, 1, 2, 1, 1, 2, 1, 2, 1, 1, 1, 2)$ & $g_{1152}$ (2,4)\\
  \hline
   Yes & $(1, 1, 0, 1, 1, 2, 1, 0, 2, 1, 2, 0, 1, 1, 2)$ & $g_{1152}$ (4)\\
 \hline
   Yes & $ (1, 1, 0, 1, 1, 2, 1, 2, 0, 1, 0, 2, 1, 1, 2)$ & $g_{1152}$ (4)\\
 \hline
  \hline
   Yes & $(1, 1, 1, 1, 0, 2, 2, 1, 1, 1, 2, 1, 2, 1, 1)$ & $g_{1152}$ (2,4)\\
  \hline
   Yes & $(1, 1, 1, 1, 0, 2, 1, 2, 1, 2, 1, 1, 2, 1, 1)$ & $g_{1152}$ (2,4)\\
  \hline
   Yes & $(1, 1, 1, 1, 0, 2, 0, 2, 1, 2, 0, 1, 2, 1, 1)$ & $g_{1152}$ (4)\\
 \hline
   Yes & $(1, 1, 1, 1, 0, 2, 2, 0, 1, 0, 2, 1, 2, 1, 1)$ & $g_{1152}$ (4)\\
 \hline
\hline
 \color{red} No & $(1, 1, 1, 1, 1, 2, 1, 2, 2, 2, 1, 1, 2, 2, 1)$ & $g_{1152}$ (4,6)\\
 \hline
  \color{red} No & $ (1, 1, 1, 1, 1, 2, 2, 1, 1, 1, 2, 2, 2, 2, 1)$ & $g_{1152}$ (4,6)\\
 \hline
 Yes & $(1, 1, 1, 1, 1, 2, 1, 2, 2, 2, 2, 2, 2, 2, 1)$ & $g_{1152}$ (4,6)\\
  \hline
 Yes & $(1, 1, 1, 1, 1, 2, 2, 2, 2, 1, 2, 2, 2, 2, 1)$ & $g_{1152}$ (4,6)\\
  \hline
   Yes & $(1, 1, 1, 1, 1, 2, 0, 2, 2, 2, 1, 1, 2, 2, 1)$ & $g_{1152}$ (6)\\
  \hline
     Yes & $(1, 1, 1, 1, 1, 2, 2, 1, 1, 0, 2, 2, 2, 2, 1)$ & $g_{1152}$ (6)\\
\hline
\hline
      \color{red} No & $(1, 1, 1, 1, 1, 2, 1, 1, 2, 2, 2, 1, 1, 2, 2)$ & $g_{1152}$ (4,6)\\
\hline
  \color{red} No & $(1, 1, 1, 1, 1, 2, 2, 2, 1, 1, 1, 2, 1, 2, 2)$ & $g_{1152}$ (4,6)\\
\hline
   Yes & $(1, 1, 1, 1, 1, 2, 2, 2, 1, 2, 2, 2, 1, 2, 2)$ & $g_{1152}$ (4,6)\\
   \hline
 Yes & $(1, 1, 1, 1, 1, 2, 2, 2, 2, 2, 2, 1, 1, 2, 2)$ & $g_{1152}$ (4,6)\\
  \hline
     Yes & $(1, 1, 1, 1, 1, 2, 2, 2, 0, 1, 1, 2, 1, 2, 2)$ & $g_{1152}$ (6)\\
  \hline
     Yes & $(1, 1, 1, 1, 1, 2, 1, 1, 2, 2, 2, 0, 1, 2, 2)$ & $g_{1152}$ (6)\\
\hline
\hline
  \color{red} No & $(1, 1, 1, 1, 1, 2, 1, 2, 1, 2, 1, 2, 2, 1, 2)$ & $g_{1152}$ (4,6)\\
 \hline
  \color{red} No & $(1, 1, 1, 1, 1, 2, 2, 1, 2, 1, 2, 1, 2, 1, 2)$ & $g_{1152}$ (4,6)\\
  \hline
   Yes & $(1, 1, 1, 1, 1, 2, 2, 1, 2, 2, 2, 2, 2, 1, 2)$ & $g_{1152}$ (4,6)\\
 \hline
 Yes & $ (1, 1, 1, 1, 1, 2, 2, 2, 2, 2, 1, 2, 2, 1, 2)$ & $g_{1152}$ (4,6)\\
  \hline
     Yes & $(1, 1, 1, 1, 1, 2, 1, 2, 1, 2, 0, 2, 2, 1, 2)$ & $g_{1152}$ (6)\\
  \hline
  Yes & $(1, 1, 1, 1, 1, 2, 2, 0, 2, 1, 2, 1, 2, 1, 2)$ & $g_{1152}$ (6)\\
  \hline
  \hline
 Yes & $(1, 1, 1, 1, 1, 2, 1, 2, 2, 2, 1, 2, 2, 2, 2)$ & $g_{1152}$ (7)\\
\hline
Yes & $(1, 1, 1, 1, 1, 2, 1, 2, 2, 2, 2, 1, 2, 2, 2)$ & $g_{1152}$ (7)\\
\hline
Yes & $(1, 1, 1, 1, 1, 2, 2, 1, 2, 1, 2, 2, 2, 2, 2)$ & $g_{1152}$ (7)\\
\hline
Yes & $ (1, 1, 1, 1, 1, 2, 2, 1, 2, 2, 2, 1, 2, 2, 2)$ & $g_{1152}$ (7)\\
\hline
Yes & $(1, 1, 1, 1, 1, 2, 2, 2, 1, 1, 2, 2, 2, 2, 2)$ & $g_{1152}$ (7)\\
\hline
Yes & $(1, 1, 1, 1, 1, 2, 2, 2, 1, 2, 1, 2, 2, 2, 2)$ & $g_{1152}$ (7)\\
\hline
Yes & $(1, 1, 1, 1, 1, 2, 2, 2, 2, 2, 2, 2, 2, 2, 2)$ & $g_{1152}$ (7)\\
\hline
\end{tabular}}
    \label{tab:5QubitEntropiesPart3}
\end{table}
\chapter{TABLES AND RELATIONS}\label{ExtendedTableAppendix}
\newpage
\noindent

Table \ref{tab:OrbitLengthCliffordSubgroupNoRelations} with relations to generate each subgroup.

\begin{table}[h]
    \centering
    \scriptsize
    \begin{tabular}{|c|c|c|c|c|p{0.2\textwidth}|}
    \hline
    Generators & Order & Diam. & Fact. & Diam.* & Relation\\
    \hline
    \hline
    \{$H_1$\} & $2^\dagger$ & 1 & - & - & \ref{HSquared}\\
    \hline
    \{$C_{1,2}$\} & $2^\dagger$ & 1 & - & - & \ref{CSquared}\\
    \hline
    \{$P_1$\} & $4$ & 3 & - & - & \ref{PFourth}\\
    \hline
    \{$H_1,H_2$\} & $4$ & 2 & - & - & \ref{HSquared}, \ref{HhComm}\\
    \hline
    \{$C_{1,2},C_{2,1}$\} & $6$ & 3 & - & - & \ref{CSquared}, \ref{CcComm}\\
    \hline
    \{$H_1,P_2\}$ & $8^\dagger$ & 4 & - & - & \ref{HSquared}, \ref{PFourth}, \ref{HpComm}\\
    \hline
    \{$P_1,C_{1,2}\}$ & $8^\dagger$ & 4 & - & - & \ref{PFourth}, \ref{CSquared}, \ref{PCComm}\\
    \hline
    \{$P_1,P_2$\} & $16$ & 6 & - & - & \ref{PFourth}, \ref{PpComm}\\
    \hline
    \{$H_1,C_{2,1}\}$ & $16^\dagger$ & 8 & - & - & \ref{HSquared}, \ref{CSquared}, \ref{ChFourth}\\
    \hline
    \{$H_1,C_{1,2}\}$ & $16^\dagger$ & 8 & - & - & \ref{HSquared}, \ref{CSquared},\ref{CNOTTransform},\newline \ref{ChFourth}\\
    \hline
    \{$H_1,P_2,C_{2,1}$\} & $32$ & 6 & - & - & \ref{HSquared}, \ref{PFourth}, \ref{CSquared},\newline \ref{HpComm}, \ref{PCComm}, \ref{ChFourth}\\
    \hline
    \{$P_1,C_{2,1}$\} & $32$ & 8 & - & - &  \ref{PFourth}, \ref{CSquared}, \ref{CpFourth}\\
    \hline
    \{$P_1,P_2,C_{2,1}$\} & $64$ & 7 & - & - &  \ref{PFourth}, \ref{CSquared}, \ref{PpComm},\newline \ref{CpFourth}, \ref{PCComm}\\
    \hline
    \{$P_1,C_{2,1},C_{1,2}$\} & $192$ & 11 & - & - & \ref{PFourth}, \ref{CSquared}, \ref{PpComm},\newline\ref{HpComm}--\ref{CNOTTransform},\newline  \ref{CcComm}--\ref{ChFourth} \\
    \hline
    \{$H_1,P_1$\} & $192$ & 16 & 8 & 6 & \ref{HSquared}--\ref{HPComm}\\
    \hline
    \{$H_1,H_2,P_1$\} & $384$ & 17 & 8 & 7 & \ref{HSquared}--\ref{HPComm}, \ref{HhComm},\newline \ref{HpComm}\\
    \hline
    \{$P_1,P_2,H_1$\} & $768$ & 19 & 8 & 9 & \ref{HSquared}--\ref{HPComm}, \ref{HhComm},\newline\ref{HpComm}\\
    \hline
    \{$H_1,C_{2,1},C_{1,2}\}$ & $2304^*$ & 26 & 2 & 15 & \ref{HSquared}--\ref{HPComm}, \ref{CSquared},\newline \ref{CcComm}, \ref{ChFourth} \\
    \hline
    \{$H_1,H_2,C_{1,2}\}$ & $2304^*$ & 27 & 2 & 17 & \ref{HSquared}--\ref{HPComm}, \ref{CSquared},\newline \ref{CcComm}, \ref{ChFourth}\\
    \hline
    \{$H_1,H_2,C_{1,2},C_{2,1}\}$ & $2304^*$ & 25 & 2 & 15 & \ref{HSquared}--\ref{HPComm}, \ref{CSquared},\newline \ref{CcComm}, \ref{ChFourth}\\
    \hline
    \{$H_1,P_1,C_{2,1}\}$ & $3072^*$ & 19 & 8 & 9 & \ref{HSquared}--\ref{HPComm}, \ref{CSquared},\newline \ref{PpComm}, \ref{HpComm}, \ref{FourGenRelation},\newline \ref{CpFourth}, \ref{PCComm} \\
    \hline
    \{$H_1,P_1,C_{1,2}$\} & $3072$ & 19 & 8 & 11 & \ref{HSquared}--\ref{HPComm}, \ref{PpComm},\newline \ref{CpCpRelation}, \ref{CHpSquared}\\
    \hline
    \{$H_1,P_1,P_2,C_{2,1}\}$ & $3072^*$ & 19 & 8 & 9 & \ref{PFourth}, \ref{HPComm}, \ref{CSquared},\newline \ref{PpComm},\ref{HpComm}--\ref{CNOTTransform},\newline \ref{CcComm}--\ref{ChFourth}\\
    \hline
    \{$H_1,H_2,P_1,P_2$\} & $4608$ & 17 & 8 & 12 & \ref{HSquared}--\ref{HPComm}, \ref{PpComm},\newline\ref{HpComm}\\
    \hline
    \{$H_1,P_2,C_{1,2}$\} & $9216$ & 24 & 8 & 13 & \ref{HSquared}, \ref{PFourth}, \ref{CSquared},\newline \ref{PpComm}, \ref{HpComm}--\ref{CNOTTransform},\newline \ref{PCComm}, \ref{ChFourth}\\
    \hline
    \{$H_1,H_2,P_1,C_{2,1}\}$ & $92160^*$ & 21 & 8 & 13 & \ref{HPComm}, \ref{PpComm}--\ref{CcComm}\\
    \hline
    \{$H_1,H_2,P_1,C_{1,2}\}$ & $92160^*$ & 21 & 8 & 16 & \ref{HPComm}, \ref{PpComm}--\ref{CcComm}\\
    \hline
    \{$H_1,P_1,P_2,C_{1,2}\}$ & $92160^*$ & 21 & 8 & 14 & \ref{HPComm}, \ref{PpComm}--\ref{CcComm}\\
    \hline
    All & $92160^*$ & 19 & 8 & 11 & \ref{HPComm}, \ref{PpComm}--\ref{CcComm}\\
    \hline
    \end{tabular}
\caption{$\mathcal{C}^2$ subgroups built by restricting generating set. Asterisk indicates subgroups with same elements, dagger indicates subgroups with isomorphic Cayley graphs. Relations to present each subgroup given in rightmost column.}
\label{tab:OrbitLengthCliffordSubgroup}
\end{table}

\newpage

In this Appendix, we provide a derivation of additional relations useful for subgroup construction, but not explicitly included in our presentation. Each of the relations below can be derived from Eqs. \eqref{HSquared}--\eqref{ChFourth}.

Using relations \ref{PFourth}, \ref{PpComm}, \ref{HpComm}, \ref{PCComm}, \ref{ChFourth}, \ref{FourGenRelation}, and \ref{CNOTTransform}, we can construct the useful identity,
\begin{equation}\label{CpCpRelationDerivation}
    C_{i,j}P_jC_{i,j}P_j = P_jC_{i,j}P_jC_{i,j}.
\end{equation}
A derivation of this identity is provided below for $C_{2,1}$ and $P_1$. For simplicity, let $H = H_1, h = H_2, C = C_{1,2},c = C_{2,1}, P=P_1,$ and $p=P_2$. We then derive,
    \begin{align*}
      HcH &= HcH, \nonumber \\
      HcHP &= HcHP, \nonumber \\
      hChP &= HcHP, \nonumber \\
      PhCh &= HcHP, \nonumber \\
      PHcH &= HcHP, \nonumber \\
      cPHcH &= cHcHP, \nonumber \\
      cPHcH &= p^2HcHcP, \nonumber \\
      p^2cPHcH &= HcHcP, \nonumber \\
      pcPHcpH &= HcHcP, \nonumber \\
      pcPccHcpH &= HcHcP, \nonumber \\  
      pcPcPcP^3 &= HcHcP, \nonumber \\  
      pcPcPcP^2 &= HcHc, \nonumber \\  
      cPcPcP^2p &= HcHc, \nonumber \\ 
      cPcPcP^2cHcpH &= \mathbb{1}, \nonumber \\
      ^* \qquad cPcPcP^3cP^3 &= \mathbb{1}, \nonumber \\
      cPcP &= PcPc, \nonumber
    \end{align*}
    where \ref{FourGenRelation} was used to show $cP^3cP^3 = cP^2cHcpH$, and acquire the starred line from the one above it.

A useful identity, derivable from relations \ref{PFourth}, \ref{CNOTTransform} and, \ref{ChFourth}, is the following,
\begin{equation}\label{CHPSquaredIdentity}
    (C_{i,j}H_iP_j^2)^2 = (P_j^2H_iC_{i,j})^2.
\end{equation}
A derivation of \eqref{CHPSquaredIdentity} is given below using $H_1, C_{1,2}$ and $P_2$. For simplicity, we again let $H = H_1, h = H_2, C = C_{1,2},c = C_{2,1},$ and $p=P_2$. We have,
    \begin{align*}
      p^4 = (p^2)^2 = (Ch)^8 &= \mathbb{1}, \nonumber \\
      ChChChChChChChCh &= \mathbb{1}, \nonumber \\
      CHcHCHcHCHcHCHcH &= \mathbb{1}, \nonumber \\
      (CH(cH)^4)^4 &= \mathbb{1}, \nonumber \\
      (CHp^2)^4 &= \mathbb{1}, \nonumber \\
      (C_{1,2}H_1P_2^2)^2 &= (P_2^2H_1C_{1,2})^2, \nonumber
    \end{align*}

A useful identity for constructing $\langle H_i, H_j, C_{i,j}\rangle$ is the following,
\begin{equation}\label{HTransformIdentity}
    C_{j,i}C_{i,j}C_{j,i}H_iC_{j,i}C_{i,j}C_{j,i} = H_j.
\end{equation}

A useful identity for constraining entropy in $\langle H_i, H_j, C_{i,j}\rangle$ reachability graphs is $(C_{i,j}H_j)^4 = P_i^2$. This relation can be proved using other relations in our presentation, beginning with the fact that $C_{i,j}$ and $P_i$ commute,
\begin{equation}\label{ChFourthEqualsPSquared}
\begin{split}
    C_{i,j}P_i&=P_iC_{i,j}, \\
    C_{i,j}P_iH_j&=P_iC_{i,j}H_j, \\
    C_{i,j}H_jP_i&=P_iC_{i,j}H_j, \\
    C_{i,j}H_jC_{i,j}P_iC_{i,j}&=C_{i,j}P_iH_j, \\
    C_{i,j}H_jC_{i,j}P_i&=C_{i,j}P_iH_jC_{i,j}, \\
    P_jC_{i,j}P_j^3H_j&=H_jP_j^3C_{i,j}P_j, \\
    C_{i,j}P_j^3H_jP_j^3&=P_j^3H_jP_j^3C_{i,j}, \\
    H_j&=P_jC_{i,j}P_j^3H_jP_j^3C_{i,j}P_j, \\
    H_j&=P_jC_{i,j}P_j^3C_{i,j}P_iH_jC_{i,j}, \\
C_{i,j}H_jC_{i,j}H_j&=C_{i,j}H_jC_{i,j}P_jC_{i,j}P_j^3H_jH_jC_{i,j}P_iH_jC_{i,j}, \\
C_{i,j}H_jC_{i,j}H_j&=P_iH_jC_{i,j}P_iH_jC_{i,j}, \\
(C_{i,j}H_j)^4 &= P_i^2.\\
\end{split}
\end{equation}

The details of building $\langle H_1,\,P_2,\,C_{1,2} \rangle$ are given below.
        \begin{enumerate}
        \item For words containing $0$ $H_1$ operations, we have the $32$ elements $b \in \langle C_{1,2},\,P_2 \rangle = \{p,\, pC_{1,2}p,\, C_{1,2}\overline{p}C_{1,2}p\}$, with $p$ defined as in Eq. \eqref{PhaseGroup}.
        \item Words containing $1$ $H_1$ have the form $bH_1b$. Since $H_1$ and $P_2$ commute, we can push all $P_2$ operations to the right until they reach a $C_{1,2}$. This action corresponds to multiplying all $0$ $H_1$ words on the left by $H_1,\, pC_{1,2}H_1,$ and $C_{1,2}\overline{p}C_{1,2}H_1$, giving the set $\{H_1b,\, pC_{1,2}H_1b,\, C_{1,2}\overline{p}C_{1,2}H_1b\}$, which has $32 + (4 \times 32) + (3 \times 32) = 256$ elements. 
         \item To generate words with $2$ $H_1$ operations, we left-multiply all $1$ $H_1$ words that do not begin with $H_1$, by the set $\{H_1, pC_{1,2}H_1,C_{1,2}\overline{p}C_{1,2}H_1\}$ (since otherwise we would collapse resulting $H_1$ pair). This procedure generates $1792$ elements, with $512$ duplicates such as
         \begin{equation}
                H_1C_{1,2}H_1 = P_2^2C_{1,2}H_1P_2^2C_{1,2}H_1P_2^2C_{1,2}P_2^2,
         \end{equation}
         as well as, 
         \begin{equation}
                C_{1,2}H_1C_{1,2}H_1 = C_{1,2}P_2^2C_{1,2}H_1P_2^2C_{1,2}H_1P_2^2C_{1,2}P_2^2,
         \end{equation}
         both reducible by Eq. \eqref{CHpSquared}. After removing duplicates, there are $1280$ unique $2$ $H_1$ words added to our set.

         \item For words containing $3$ $H_1$ operations, we again left-multiply all $2$ $H_1$ words which \underline{do not} begin with $H_1$, by $\{H_1, pC_{1,2}H_1,$ and $C_{1,2}\overline{p}C_{1,2}H_1\}$, generating $12544$ elements. 

         Of these $12544$ elements, $9152$ are duplicates of other $3$ $H_1$ words, e.g.
         \begin{equation}
                    H_1C_{1,2}H_1C_{1,2}H_1 = H_1C_{1,2}P_2^2C_{1,2}H_1P_2^2C_{1,2}H_1P_2^2C_{1,2}P_2^2,
         \end{equation}
         described by Eq. \eqref{CHpSquared}. An additional $256$ are duplicates of $1$ $H_1$ words, e.g.
         \begin{equation}
                H_1C_{1,2}H_1P_2^2C_{1,2}H_1 = P_2^2C_{1,2}H_1P_2^2C_{1,2}P_2^2,
         \end{equation}
         also described by Eq. \eqref{CHpSquared}, leaving only $3136$ new contributions to the subgroup.
        
         \item Words with $4$ $H_1$ operations are likewise built by left-multiplying all $3$ $H_1$ words that \underline{do not} begin with $H_1$, by $\{H_1, pC_{1,2}H_1,$ and $C_{1,2}\overline{p}C_{1,2}H_1\}$. This process generates $87808$ elements, $83200$ of which are duplicates of other $4$ $H_1$ words, e.g.
                \begin{equation}
                    H_1C_{1,2}H_1C_{1,2}H_1C_{1,2}H_1 = C_{1,2}H_1C_{1,2}H_1C_{1,2}H_1C_{1,2}H_1C_{1,2},
                \end{equation}
         which can be reduced by $(H_1C_{1,2})^8 = \mathbb{1}$. Another $1280$ are duplicates of $2$ $H_1$ words, such as
         \begin{equation}
             H_1C_{1,2}H_1C_{1,2}H_1P_2^2C_{1,2}H_1 = H_1C_{1,2}P_2^2C_{1,2}H_1P_2^2C_{1,2}P_2^2,
         \end{equation}
         described by Eqs. \eqref{CHpSquared} and \eqref{CpCpRelation}. A final $32$ elements are duplicates of $0$ $H_1$ words, e.g.
         \begin{equation}
             H_1C_{1,2}H_1P_2^2C_{1,2}H_1P_2^2C_{1,2}H_1 = P_2^2C_{1,2}P_2^2,
         \end{equation}
         explained using Eq. \eqref{CHpSquared}. Removing duplicates adds $3296$ new $4$ $H_1$ words to our subgroup.
        
         \item Finally we construct words with $5$ $H_1$ operations, multiplying all $4$ $H_1$ words that \underline{do not} begin with $H_1$ by $H_1, pC_{1,2}H_1,$ and $C_{1,2}\overline{p}C_{1,2}H_1$, and generating $614656$ words. Of these $614656$ words, $609792$ are duplicates of other $5$ $H_1$ words, e.g.
         \begin{equation}
         \begin{split}
         H_1P_2C_{1,2}&H_1P_2^3C_{1,2}H_1P_2^3C_{1,2}H_1P_2C_{1,2}H_1P_2C_{1,2}P_2^3\\
        &= C_{1,2}H_1C_{1,2}H_1C_{1,2}H_1C_{1,2}H_1C_{1,2}H_1,
         \end{split}
         \end{equation}
         described by Eq. \eqref{CHpSquared} using relations \ref{PFourth}, \ref{CSquared}, and \ref{CpFourth}. Another $3392$ are duplicates of $3$ $H_1$ words, e.g.
            \begin{equation}
                    C_{1,2}H_1C_{1,2}H_1C_{1,2}H_1C_{1,2} = H_1C_{1,2}H_1C_{1,2}H_1C_{1,2}H_1C_{1,2}H_1,
            \end{equation}
         described by $(H_1C_{1,2})^8 = \mathbb{1}$, and a final $256$ are duplicates of $1$ $H_1$ words, e.g.
                \begin{equation}
                    H_1 = H_1C_{1,2}H_1C_{1,2}H_1P_2^2C_{1,2}H_1P_2^2C_{1,2}H_1C_{1,2}P_2^2C_{1,2}P_2^2,
                \end{equation}
         by Eqs. \eqref{CHpSquared} and \eqref{CpCpRelation}. Upon removing duplicates, there are $1216$ unique $5$ $H_1$ words added to our subgroup, giving an order of $9216$.
        \end{enumerate}

The Magma Computer Algebra system \cite{MR1484478} was used to validate many of the subgroup constructions in this paper. This program, which can be downloaded to a local device or accessed via a browser at \href{http://magma.maths.usyd.edu.au/magma/}{http://magma.maths.usyd.edu.au/magma/}, offers a powerful suite of group-theoretic calculators. We have included below a few examples of code input and output that demonstrate some functionality, and provide the interested reader a coarse template for using the software.

We formally construct subgroups of $\mathcal{C}_1$ and $\mathcal{C}_2$ by taking quotient groups of the free group generated by a select set of operations. For example, the group $\langle H_i,\,P_i \rangle$ can be built as
\inBox{$F\langle H,P \rangle := \textnormal{FreeGroup}(2);$\\
$G\langle x, y\rangle, \textnormal{phi} := \textnormal{quo}\langle F|H^2 = P^4 = 1, (H*P)^3 = (P*H)^3\rangle;$\\
$G$}
\outBox{Finitely presented group $G$ on $2$ generators\\
Relations\\
\begin{equation*}
    \begin{split}
    x^2 &= \textnormal{Id}(G)\\
    y^4& = \textnormal{Id}(G)\\
    (x * y)^3 &= (y * x)^3\\
    \end{split}
\end{equation*}}

Further calculations can be used to yield desired group properties such as order, number of defining generators, checks for abelianess, cyclicity, and many more. Some simple examples are included below.
\inBox{$F\langle P,C \rangle := \textnormal{FreeGroup}(2);$\\
$G\langle x, y\rangle, \textnormal{phi} := \textnormal{quo}\langle F|C^2 = P^4 = 1, C*P*C=P\rangle;$\\
Order(G)}
\outBox{$8$}
\inBox{$F\langle H,c \rangle := \textnormal{FreeGroup}(2);$\\
$G\langle x, y\rangle, \textnormal{phi} := \textnormal{quo}\langle F|c^2 = H^2 = (H*c)^8 = 1;$\\
Parent(G)}
\outBox{Power Structure of GrpFP}

To generate the complete two-qubit Clifford group $\mathcal{C}_2$, we include all generators and the relations of our presentation. We can subsequently simplify the presentation with the following code.
\inBox{$F<H,\,h,\,C,\,c,\,P,\,p> := \textnormal{FreeGroup}(6);$\\
$G<u,\,v,\,w,\,x,\,y,\,z>, \textnormal{phi} := \textnormal{quo}\langle F|  P^4 = p^4 = H^2 = h^2 = C^2 = c^2 = (H*C*p^2)^4 = (c*C)^3  =(c*H)^8 = 1,\,
H*h*H = h,\, H*p*H = p,\, h*P*h = P,\, P^3*p*P = p,\, C*P*C = P,\, c*p*c = p,\, (H*P)^3 = (P*H)^3,\, (p*h)^3 = (h*p)^3,\,  P*H*P = (c*p*h)^3,\, C*h*C*p*C*(p^3)*h = P,\, c*C*c*H*c*C*c = h,\, H*h*c*H*h = C,\, c*C*p^2*C*c = P^2,\,  p*C*p*C = C*p*C*p,\, P^2 = (C*h)^4,\, P^2 = (C*p)^4>;$\\
Order(G);\\
Simplify(G)}
\outBox{$92160$\\
Finitely presented group on 4 generators\\
Generators as words in group G\\
\begin{equation*}
    \begin{split}
    \$.1 &= u\\
    \$.2 &= v\\
    \$.3 &= w\\
    \$.4 &= z\\
    \end{split}
\end{equation*}
Relations\\
  \$.3$\wedge$2 = Id(\$)\\
  \$.1$\wedge$2 = Id(\$)\\
  \$.2$\wedge$2 = Id(\$)\\
  (\$.1 * \$.2)$\wedge$2 = Id(\$)\\
  \$.4$\wedge$4 = Id(\$)\\
  \$.1 * \$.4 * \$.1 * \$.4$\wedge$-1 = Id(\$)\\
   \$.2 * \$.4$\wedge$-1 * \$.3 * \$.4 * \$.2 * \$.4 * \$.3 * \$.4$\wedge$-1 = Id(\$)\\
   \$.2 * \$.3 * \$.2 * \$.4$\wedge$-1 * \$.2 * \$.3 * \$.2 * \$.4 = Id(\$)\\
   \$.4 * \$.3 * \$.4 * \$.3 * \$.4$\wedge$-1 * \$.3 * \$.4$\wedge$-1 * \$.3 = Id(\$)\\
   (\$.4$\wedge$-1 * \$.3 * \$.4$\wedge$2 * \$.3 * \$.4$\wedge$-1)$\wedge$2 = Id(\$)\\
   \$.2 * \$.4 * \$.3 * \$.4$\wedge$-1 * \$.2 * \$.3 * \$.2 * \$.4$\wedge$-1 * \$.3 * \$.4 * \$.2* \$.3 = Id(\$)\\
 \$.3 * \$.4$\wedge$-1 * \$.3 * \$.2 * \$.3 * \$.2 * \$.3 * \$.4 * \$.3 * \$.2 * \$.3 * \$.2 = Id(\$)\\
 \$.4 * \$.2 * \$.4 * \$.2 * \$.4 * \$.2 * \$.4$\wedge$-1 * \$.2 * \$.4$\wedge$-1 * \$.2 * \$.4$\wedge$-1 * \$.2 = Id(\$)\\
  \$.2 * \$.1 * \$.3 * \$.2 * \$.1 * \$.4 * \$.2 * \$.1 * \$.3 * \$.2 * \$.1 * \$.4$\wedge$-1 =Id(\$)\\
  \$.1 * \$.3 * \$.4$\wedge$2 * \$.1 * \$.3 * \$.4$\wedge$-2 * \$.1 * \$.3 * \$.4$\wedge$-2 * \$.1 * \$.3 *
    \$.4$\wedge$-2 = Id(\$)\\
    \$.3 * \$.2 * \$.1 * \$.3 * \$.2 * \$.1 * \$.3 * \$.1 * \$.3 * \$.2 * \$.1 * \$.3 * \$.2
    * \$.1 * \$.3 * \$.2 = Id(\$)\\
    (\$.3 * \$.1)$\wedge$8 = Id(\$)\\
    \$.4$\wedge$-1 * \$.1 * \$.3 * \$.2 * \$.1 * \$.3 * \$.4$\wedge$-2 * \$.3 * \$.4 * \$.1 * \$.2 * \$.3*\\ \$.1 * \$.2 * \$.3 * \$.4$\wedge$-1 * \$.2 * \$.4 * \$.3 = Id(\$)\\
    \$.2 * \$.3 * \$.2 * \$.4$\wedge$-1 * \$.3 * \$.1 * \$.3 * \$.4 * \$.3 * \$.2 * \$.3 * \$.2*\\
    \$.4$\wedge$-1 * \$.1 * \$.2 * \$.3 * \$.4$\wedge$-1 * \$.2\\ * \$.3 * \$.4$\wedge$-1 * \$.2 * \$.3 * \$.1 *\$.3 = Id(\$)\\
    \$.3 * \$.4 * \$.3 * \$.2 * \$.3 * \$.1 * \$.4 * \$.3 * \$.4$\wedge$-1 * \$.2 * \$.3 * \$.2\\
    *\$.3 * \$.1 * \$.4 * \$.3 * \$.4$\wedge$-1 * \$.2 * \$.3 * \$.2 * \$.3 * \$.1 * \$.3 * \$.2
    *\$.3 * \$.4$\wedge$-1 * \$.3 * \$.4 * \$.1 * \$.2 * \$.3 * \$.2 * \$.3 * \$.4$\wedge$-1 * \$.3 * \$.4* \$.1 * \$.2 * \$.3 * \$.2 * \$.3 * \$.4$\wedge$-1 * \$.3 * \$.1 = Id(\$)}

In this section we include several additional graphs, not included in the main text. Each figure below offers further visualization for relations \eqref{HSquared}--\eqref{ChFourth}.

Figure \ref{H1P2CayleyGraph} shows the Cayley graph for $\langle H_1,\,P_2 \rangle$, containing $8$ vertices. Since $H_1$ and $P_2$ commute, the group $\langle H_1,\,P_2 \rangle$ is the direct product of $\langle H_1 \rangle \times \langle P_2 \rangle$, which is manifest in the Cayley graph structure.
    \begin{figure}[h]
        \centering        
        \includegraphics[width=8.5cm]{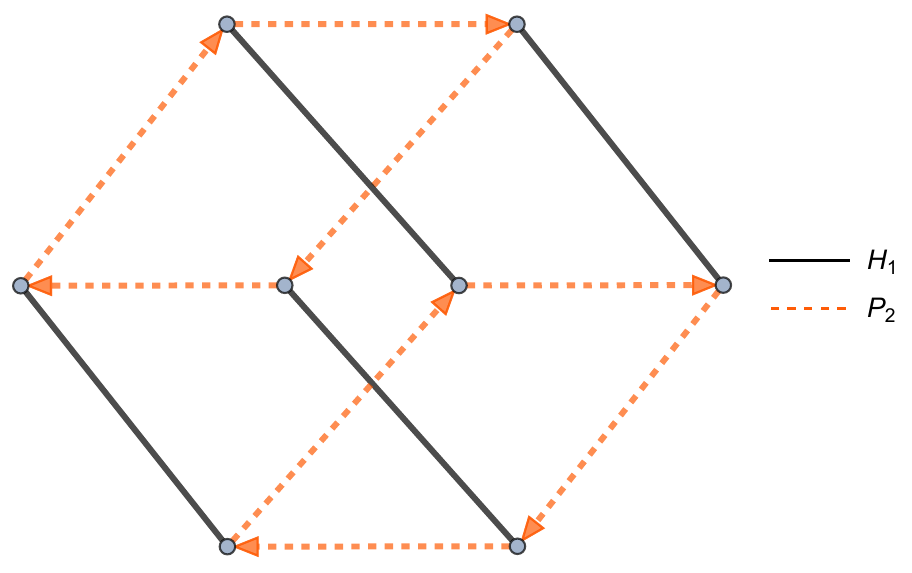}
        \caption{Cayley graph for $\langle H_1,\,P_2 \rangle$, the direct product $\langle H_1 \rangle \times \langle P_2 \rangle$, with $8$ vertices. Individual group structures for $\langle H_1 \rangle$ and $\langle P_2 \rangle$ are easily verified.}
        \label{H1P2CayleyGraph}
    \end{figure}

Figure \ref{P1P2CayleyGraph} shows the Cayley graph of $\langle P_1,\,P_2 \rangle$. The graph has $16$ vertices, and corresponds to the direct product $\langle P_1 \rangle \times \langle P_2 \rangle$.
    \begin{figure}[h]
        \centering        
        \includegraphics[width=8.8cm]{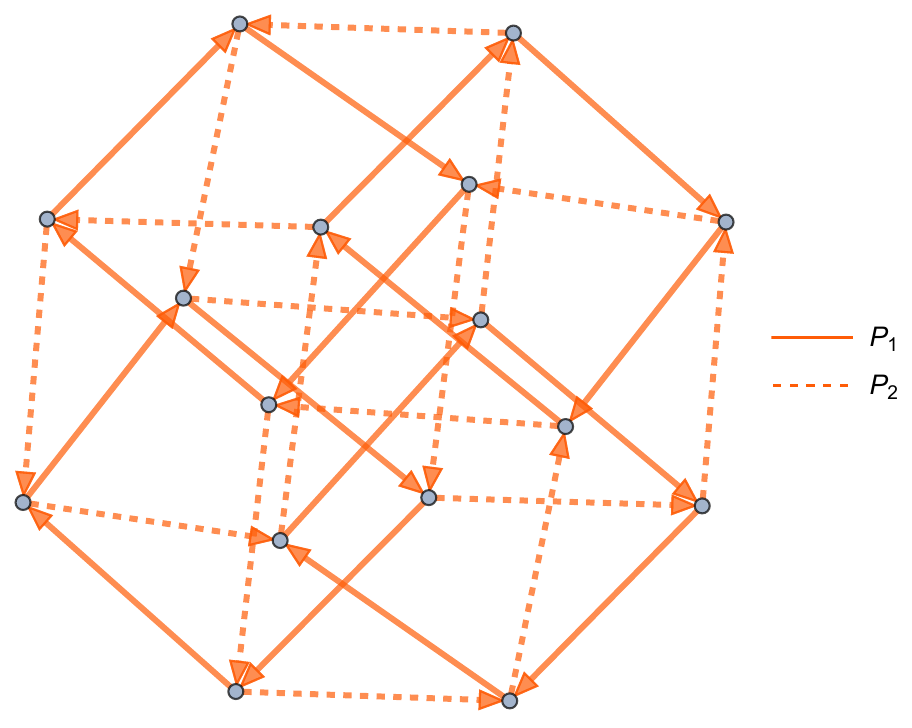}
        \caption{Cayley graph of $\langle P_1,\,P_2 \rangle$, the direct product $\langle P_1 \rangle \times \langle P_2 \rangle$, containing $16$ vertices.}
        \label{P1P2CayleyGraph}
    \end{figure}
\newpage
The group $\langle H_1,\,P_2,\,C_{1,2} \rangle$ is represented by the Cayley graph in Figure \ref{H1P2C12CayleyGraph}. The graph contains $32$ vertices, and the relation $(H_1C_{1,2})^4 = P_2^2$ can be directly visualized.
    \begin{figure}[h]
        \centering        
        \includegraphics[width=9cm]{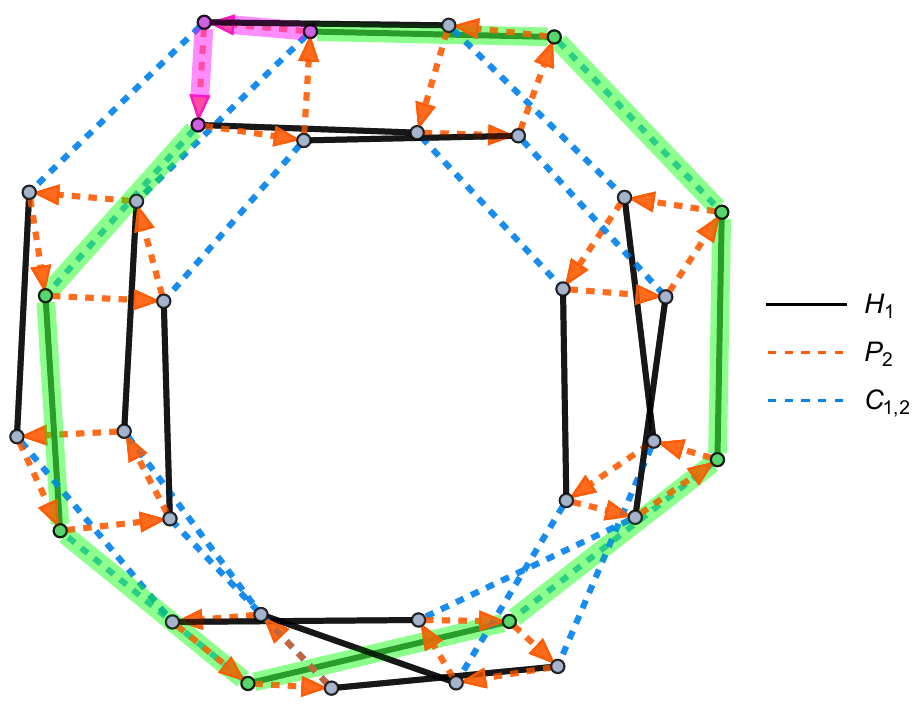}
        \caption{Cayley graph of $\langle H_1,\,P_2,\,C_{1,2} \rangle$, with $32$ vertices, where we note the non-trivial relation $(H_1C_{1,2})^4 = P_2^2$. Sequence $(H_1C_{1,2})^4$ is highlighted in green, and $P_2^2$ in magenta.}
        \label{H1P2C12CayleyGraph}
    \end{figure}

Figure \ref{P1P2H1CayleyGraph} gives the Cayley graph for $\langle P_1,\,P_2,\,H_1 \rangle$. While the graph is large, the symmetric structure of stacked $P_1$ and $P_2$ boxes connected by $H_1$ can be observed.
    \begin{figure}[h]
        \centering        
        \includegraphics[width=9.5cm]{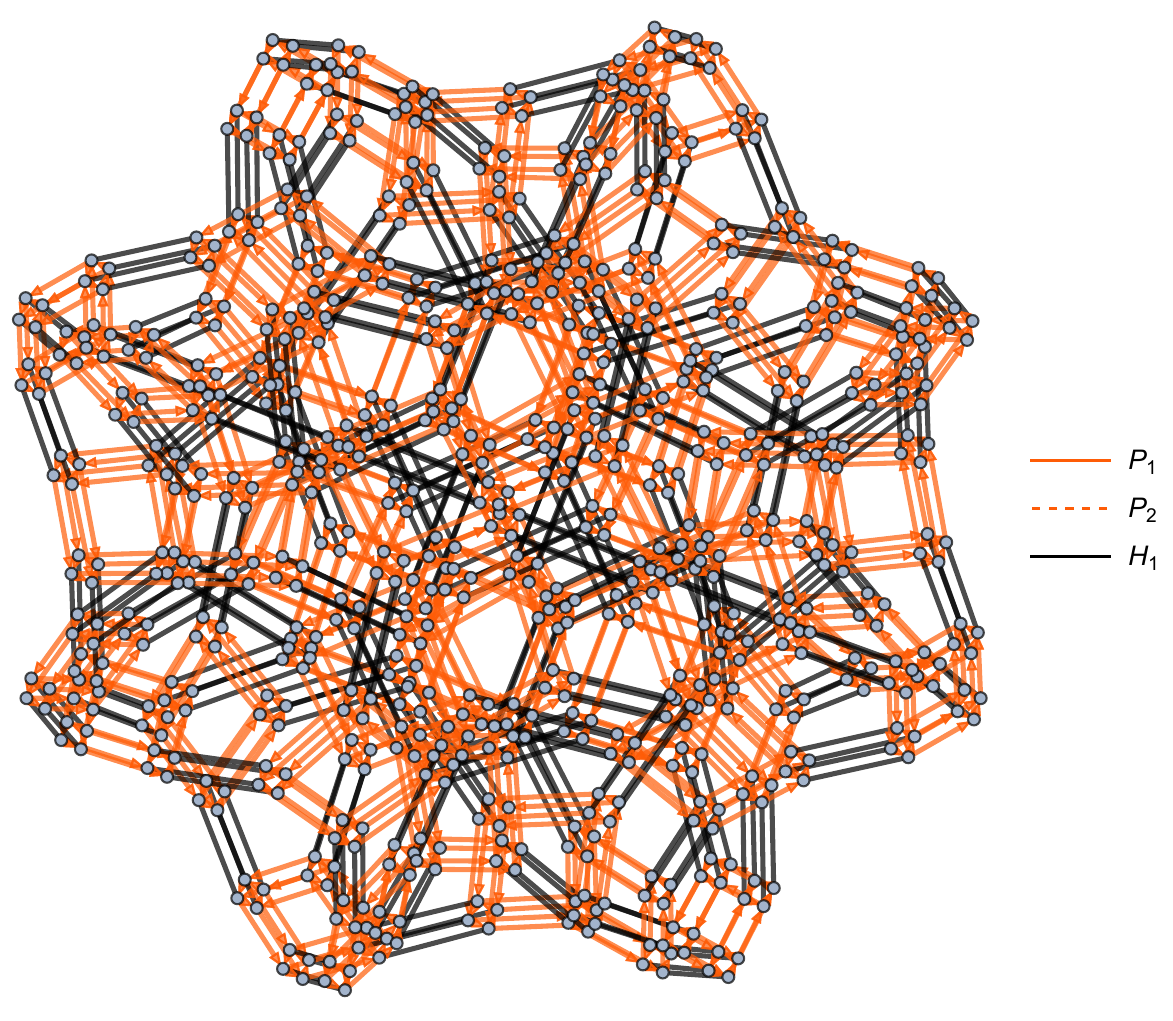}
        \caption{Cayley graph of $\langle P_1,\,P_2,\,H_1 \rangle$ with $768$ vertices. The graph contains directed $P_1$ and $P_2$ boxes connected by $H_1$ edges.}
        \label{P1P2H1CayleyGraph}
    \end{figure}
\chapter{SIMPLIFICATION FOR FIXED PARAMETERS}\label{lIndependentSimplification}
\newpage
\noindent

In this Appendix we give a few reductions for the $\ket{D^N_k}$ entropies in Eq.\ \eqref{DickeStateEntropies}. Expanding $\ln$ in Eq.\ \eqref{DickeStateEntropies} gives,
\begin{equation}
\begin{split}
    S_{\ell} = - \binom{N}{k}^{-1}\sum_{i=0}^{min(\ell,k)}& \binom{\ell}{i}\binom{N-\ell}{k-i}\ln\left[\binom{N}{k}^{-1}\right]\\
    &-  \binom{N}{k}^{-1}\sum_{i=0}^{min(\ell,k)}\binom{\ell}{i}\binom{N-\ell}{k-i}\ln\left[\binom{\ell}{i}\binom{N-\ell}{k-i}\right].\\
\end{split}
\end{equation}
We invoke the Chu-Vandermonde identity, which reads
\begin{equation}
    \sum_{i=0}^n \binom{r}{i}\binom{s}{n-i} = \binom{r+s}{n},
\end{equation}
to simplify the first sum as
\begin{equation}\label{DecoupledDickeEntropy}
\begin{split}
    S_{\ell} = \binom{N}{k}^{-1}& \binom{N}{\min(\ell,k)}\ln\left[\binom{N}{k} \right] \\
    &  -\binom{N}{k}^{-1}\sum_{i=0}^{min(\ell,k)}\binom{\ell}{i}\binom{N-\ell}{k-i}\ln\left[\binom{\ell}{i}\binom{N-\ell}{k-i}\right].\\
\end{split}
\end{equation}
For the case when $\ell \geq k$, Eq.\ \eqref{DecoupledDickeEntropy} further simplifies to,
\begin{equation}
        S_{\ell} = \ln\left[\binom{N}{k} \right] -\binom{N}{k}^{-1}\sum_{i=0}^{k}\binom{\ell}{i}\binom{N-\ell}{k-i}\ln\left[\binom{\ell}{i}\binom{N-\ell}{k-i}\right],
\end{equation}
where we note the $\ell$-independence of the first term. 

Likewise for $\ell < k$ we have the reduction
\begin{equation}\label{lLessThankSimplification}
        S_{\ell} = \frac{k!(N-k)!}{\ell!(N-\ell)!}\ln\left[\binom{N}{k} \right] -\binom{N}{k}^{-1}\sum_{i=0}^{\ell}\binom{\ell}{i}\binom{N-\ell}{k-i}\ln\left[\binom{\ell}{i}\binom{N-\ell}{k-i}\right].
\end{equation}

\newpage

The reachability graph for $\ket{D^N_N}$, and accordingly for all stabilizer states, under the action of $\mathcal{C}_2$ is shown in Figure \ref{FullC2Graph}.
    \begin{figure}[h]
        \centering
        \includegraphics[width=11cm]{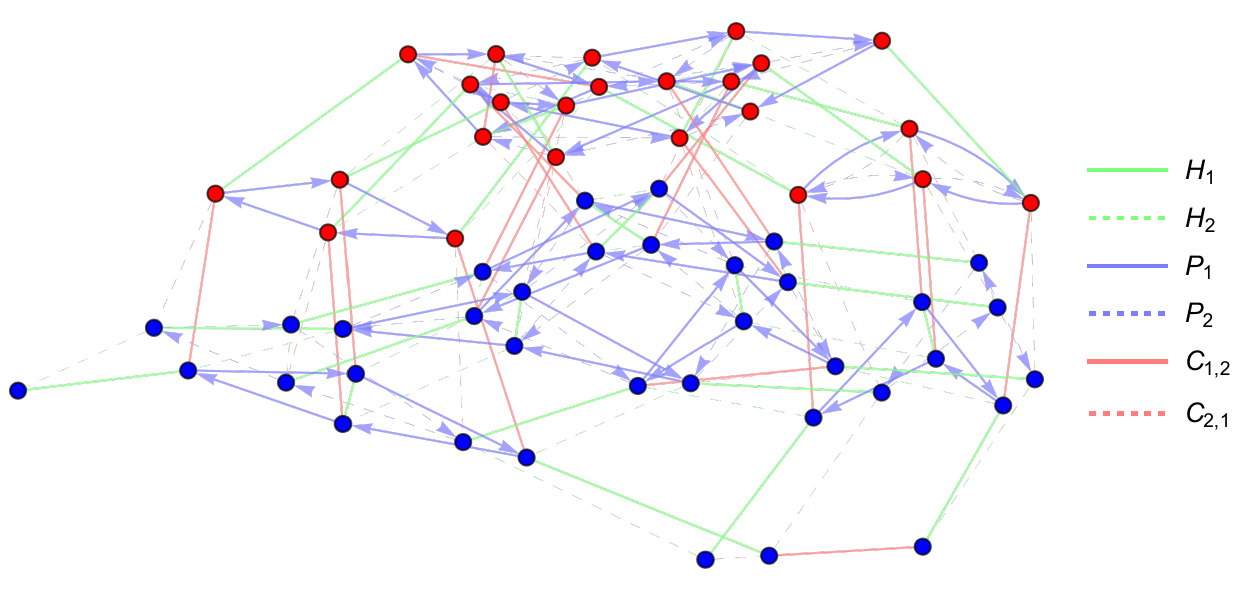}
        \caption{Orbit of all $\ket{D^N_N}$ under the action of the $2$-qubit Clifford group $\mathcal{C}_2$. The state $\ket{D^N_N}$ is a stabilizer state, and therefore its reachability graph is isomorphic to that of all $2$-qubit stabilizer states.}
    \label{FullC2Graph}
    \end{figure}

Figure \ref{FullC3Graph} gives the reachability graph for $\ket{D^N_N}$ under the action of $\mathcal{C}_3$.
    \begin{figure}[h]
        \centering
        \includegraphics[width=15cm]{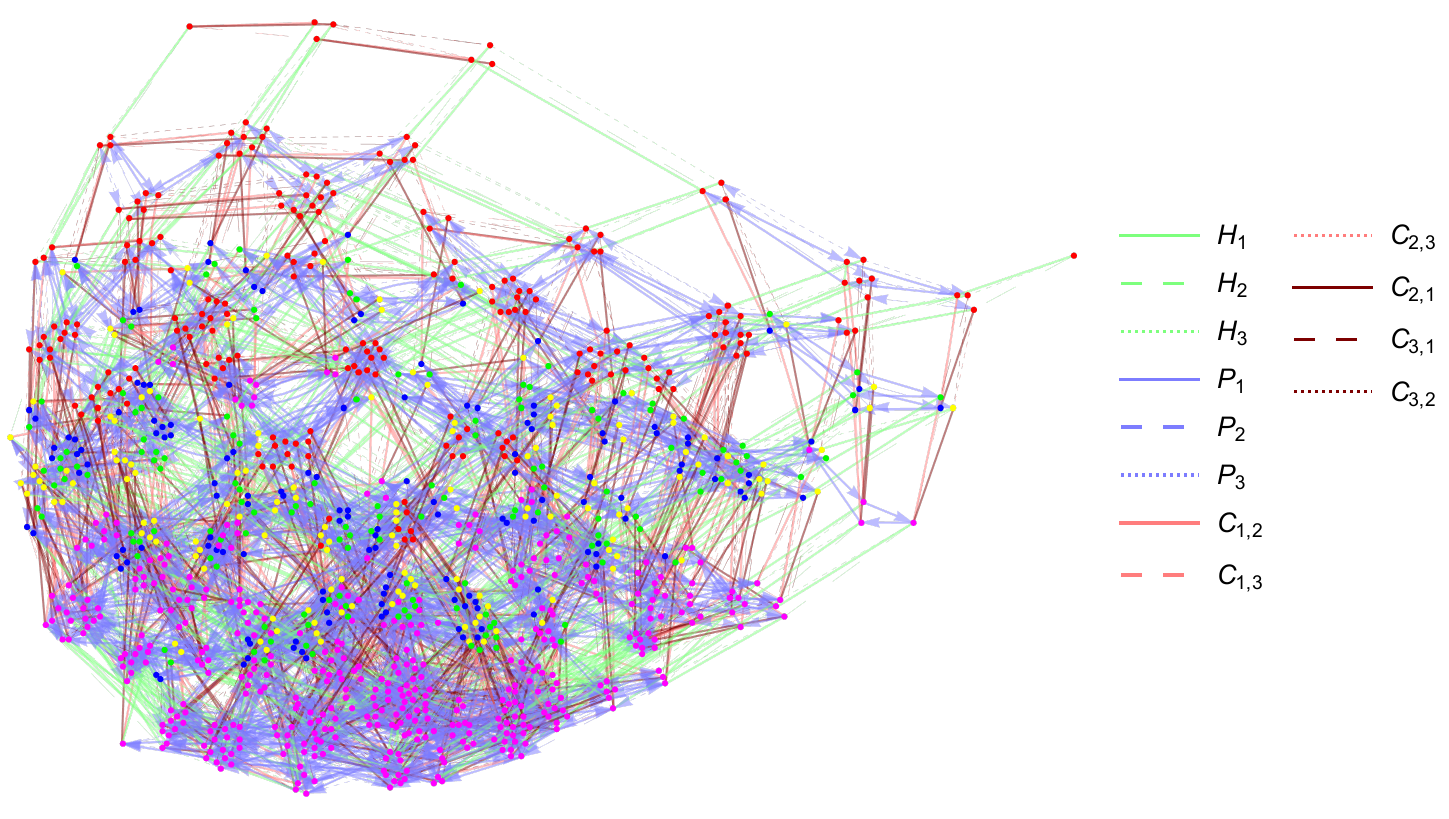}
        \caption{Orbit of $\ket{D^N_N}$ under the action of the $3$-qubit Clifford group $\mathcal{C}_3$. Since $\ket{D^N_N}$ is a stabilizer state, this reachability graph is isomorphic to the orbit shared by all $3$-qubit stabilizer states.}
    \label{FullC3Graph}
    \end{figure}

Figure \ref{D42PauliGraph} depicts the orbit of state $\ket{D^{4}_2}$ under action of the $4$-qubit Pauli group $\Pi_4$. Since $\ket{D^{4}_2}$ is stabilized by a $4$-element subgroup of $\Pi_4$, its reachability graph contains $64$ vertices.
    \begin{figure}[h]
        \centering
        \includegraphics[width=12cm]{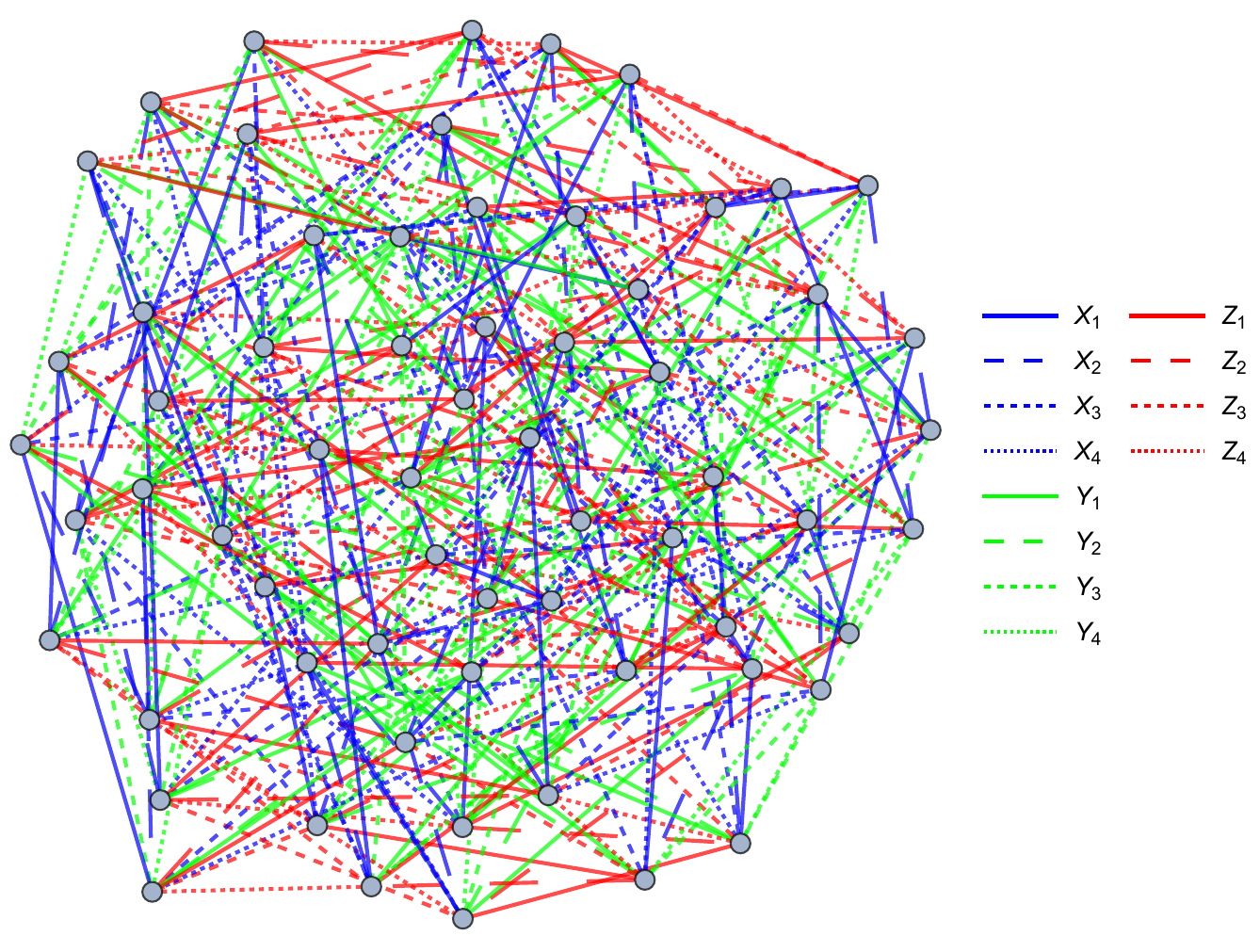}
        \caption{Orbit of $\ket{D^4_2}$ under the action of $\Pi_4$. This reachability graph contains $64$ vertices as $\ket{D^4_2}$ is only stabilized by $4$ elements of $\Pi_4$.}
    \label{D42PauliGraph}
    \end{figure}

Below we include additional examples of $\ket{D^N_1}$ orbits under the action of $(HC)_{1,2} \equiv \langle H_1,\,H_2\,C_{1,2}\, C_{2,1}\rangle$. Figure \ref{HCGraphD41} shows the orbit for $\ket{D^4_1}$ under $(HC)_{1,2}$, which contains $288$ states and $4$ different entropy vectors. This set of $5$ entropy vectors is built of $6$ different entanglement entropies.
    \begin{figure}[h]
        \centering
        \includegraphics[width=14cm]{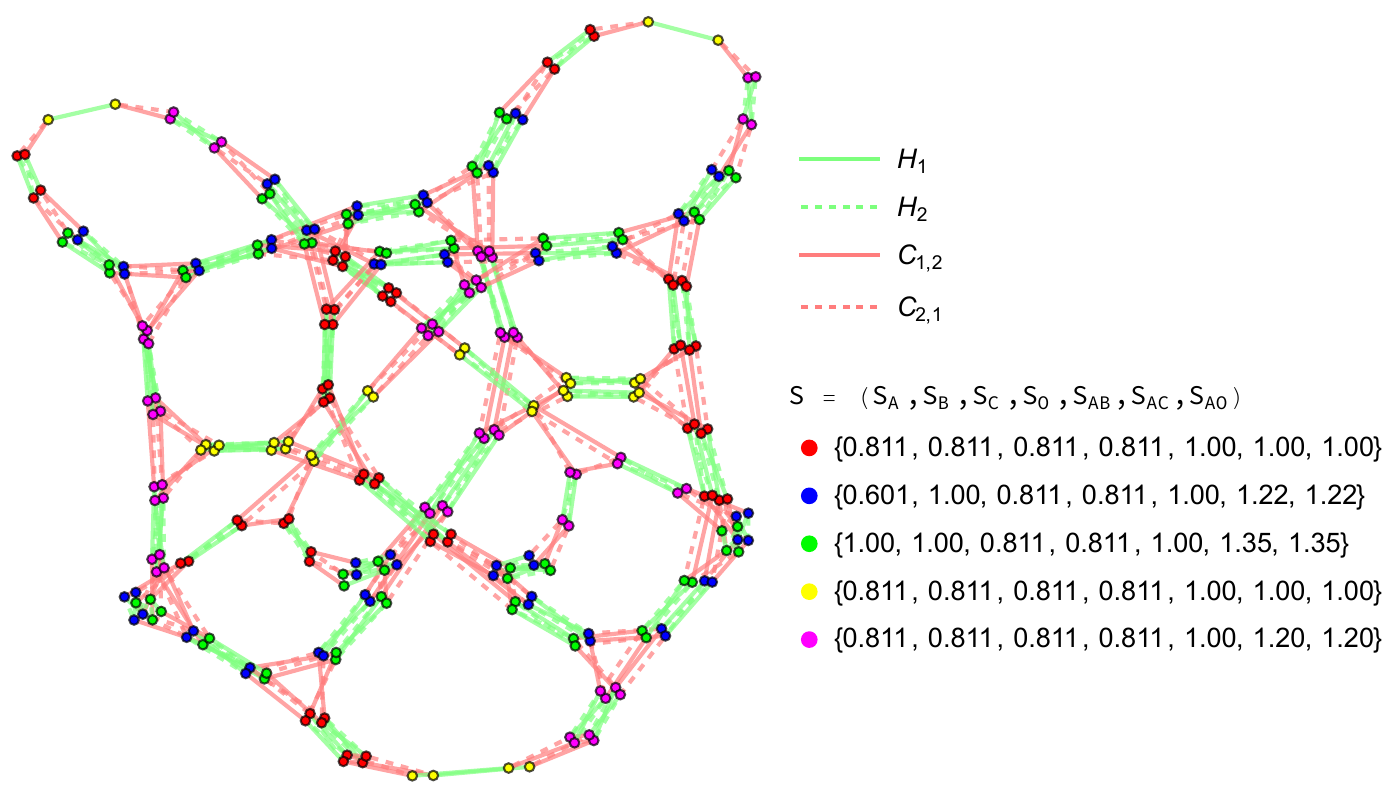}
        \caption{Orbit of $\ket{D^4_1}$, and stabilizer state, under the action of $(HC)_{1,2}$.}
    \label{HCGraphD41}
    \end{figure}

Figure \ref{HCGraphD51} gives the orbit of $\ket{D^5_1}$ under $(HC)_{1,2}$. This orbit likewise has $5$ different entropy vectors, which are composed of $9$ different entanglement entropies.
    \begin{figure}[h]
        \centering
        \includegraphics[width=15cm]{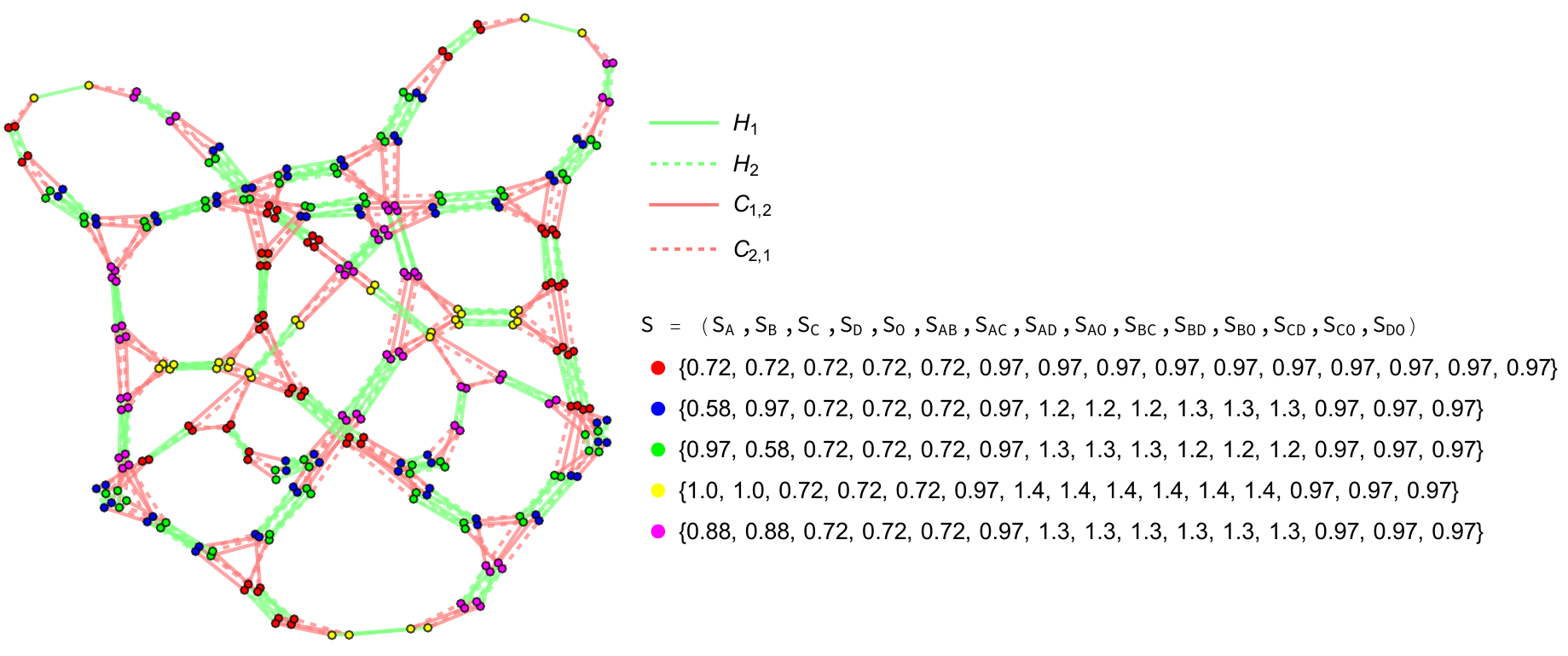}
        \caption{Orbit of $\ket{D^5_1}$, and stabilizer state, under the action of $(HC)_{1,2}$.}
    \label{HCGraphD51}
    \end{figure}

In this Appendix we give exact entropy vectors seen in Figures \ref{HCGraphD31} and \ref{HCGraphD42}. Table \ref{tab:EntropyVectorTable1} gives each entropy vector from Figure \ref{HCGraphD31}. There are $4$ entanglement entropies observed in the orbit of $\ket{D^3_1}$ under the action of $\langle H_1,\,H_2\,C_{1,2}\, C_{2,1}\rangle$, which we define as variables in Eq.\ \eqref{Entropies31} for presentation clarity.
\begin{equation}\label{Entropies31}
    \begin{split}
        s_0&\equiv 1,\\
        s_1&\equiv \frac{2}{3}\log_2\left[\frac{3}{2}\right] +\frac{1}{3}\log_2\left[3\right],\\
        s_2&\equiv \frac{5}{6}\log_2\left[\frac{6}{5}\right] +\frac{1}{6}\log_2\left[6\right],\\
        s_3&\equiv \frac{3-\sqrt{5}}{6}\log_2\left[\frac{6}{3-\sqrt{5}}\right] +\frac{3+\sqrt{5}}{6}\log_2\left[\frac{6}{3+\sqrt{5}}\right],
    \end{split}
\end{equation}

\newpage

The four entropies in Eq.\ \eqref{Entropies31} build the entropy vectors in Table \ref{tab:EntropyVectorTable1}.
\begin{table}[h]
    \centering
    \begin{tabular}{|c||c|}
    \hline
    Label & Entropy Vector\\
    \hline
    \hline
    \fcolorbox{black}{red}{\rule{0pt}{6pt}\rule{6pt}{0pt}} & $(s_1,\,s_1,\,s_1)$\\
    \hline
    \fcolorbox{black}{blue}{\rule{0pt}{6pt}\rule{6pt}{0pt}} & $(s_3,\,s_1,\,s_1)$\\
    \hline
    \fcolorbox{black}{green}{\rule{0pt}{6pt}\rule{6pt}{0pt}} & $(s_1,\,s_3,\,s_1)$\\
    \hline
    \fcolorbox{black}{yellow}{\rule{0pt}{6pt}\rule{6pt}{0pt}} & $(s_0,\,s_0,\,s_1)$\\
    \hline
    \fcolorbox{black}{magenta}{\rule{0pt}{6pt}\rule{6pt}{0pt}} & $(s_2,\,s_2,\,s_1)$\\
    \hline
    \end{tabular}
\caption{The $5$ entropy vectors found in the $(HC)_{1,2}$ orbit of $\ket{D^3_1}$, shown in Figure \ref{HCGraphD31}. For brevity, we introduce the variables in Eq.\ \eqref{Entropies31} to present these entropy vectors.}
\label{tab:EntropyVectorTable1}
\end{table}

Similarly for the orbit of $\ket{D^4_2}$ under $\langle H_1,\,H_2\,C_{1,2}\, C_{2,1}\rangle$ action, there are $5$ entanglement entropies observe. We likewise define the variables,
\begin{equation}\label{Entropies42}
    \begin{split}
        s_0&\equiv \frac{5}{6}\log_2\left[\frac{12}{5}\right] +\frac{1}{6}\log_2\left[12\right],\\
        s_1&\equiv \frac{3-\sqrt{5}}{6}\log_2\left[\frac{12}{3-\sqrt{5}}\right] +\frac{3+\sqrt{5}}{6}\log_2\left[\frac{12}{3+\sqrt{5}}\right],\\
        s_2&\equiv \frac{2}{3}\log_2\left[\frac{3}{2}\right] +\frac{1}{3}\log_2\left[6\right],\\
        s_3&\equiv \frac{3-2\sqrt{2}}{6}\log_2\left[\frac{12}{3-2\sqrt{2}}\right] +\frac{3+2\sqrt{2}}{6}\log_2\left[\frac{12}{3+2\sqrt{2}}\right],\\
        s_4&\equiv 1,\\
        s_5&\equiv \frac{2}{3}\log_2\left[\frac{3}{2}\right] +\frac{1}{3}\log_2\left[3\right],\\
        s_6&\equiv \frac{5}{6}\log_2\left[\frac{6}{5}\right] +\frac{1}{6}\log_2\left[6\right].
    \end{split}
\end{equation}

\begin{table}[h]
    \centering
    \begin{tabular}{|c||c|}
    \hline
    Label & Entropy Vector\\
    \hline
    \hline
    \fcolorbox{black}{red}{\rule{0pt}{6pt}\rule{6pt}{0pt}} & $\left(s_4,\,s_4,\,s_4,\,s_4,\,s_2,\,s_2,\,s_2 \right)$\\
    \hline
    \fcolorbox{black}{blue}{\rule{0pt}{6pt}\rule{6pt}{0pt}} & $\left(s_6,\,s_5,\,s_4,\,s_4,\,s_2,\,s_1,\,s_1 \right)$\\
    \hline
    \fcolorbox{black}{green}{\rule{0pt}{6pt}\rule{6pt}{0pt}} & $\left(s_5,\,s_6,\,s_4,\,s_4,\,s_2,\,s_1,\,s_1 \right)$\\
    \hline
    \fcolorbox{black}{yellow}{\rule{0pt}{6pt}\rule{6pt}{0pt}} & $\left(s_4,\,s_4,\,s_4,\,s_4,\,s_2,\,s_0,\,s_0 \right)$\\
    \hline
    \fcolorbox{black}{magenta}{\rule{0pt}{6pt}\rule{6pt}{0pt}} & $\left(s_4,\,s_4,\,s_4,\,s_4,\,s_2,\,s_2,\,s_2 \right)$\\
    \hline
    \fcolorbox{black}{cyan}{\rule{0pt}{6pt}\rule{6pt}{0pt}} & $\left(s_6,\,s_6,\,s_4,\,s_4,\,s_2,\,s_3,\,s_3 \right)$\\
    \hline
    \end{tabular}
\caption{The $6$ entropy vectors contained in the orbit of $\ket{D^4_2}$ under the action of $(HC)_{1,2}$, illustrated in Figure \ref{HCGraphD42}. We introduce variables in Eq.\ \eqref{Entropies42} to display the entropy vectors in the table.}
\label{tab:EntropyVectorTable2}
\end{table}
\chapter{TABLES OF ENTROPY VECTORS}\label{EntropyVectorTables}
\newpage
\noindent

Below we include sets of entropy vectors referenced throughout the paper. The states used to generate each entropy vector set are likewise given in bit-address notation. A bit-address is the ordered set of coefficients multiplying each basis ket of an n-qubit system, e.g. the bit-address $(1,0,0,1,0,0,i,i)$ indicates the state $\ket{000}+\ket{011}+i\ket{110}+i\ket{111}$. We order index qubits within each ket from right to left, i.e. the rightmost digit corresponds to the first qubit of the system, while the leftmost digit represents the $n^{\text{th}}$ qubit of an $n$-qubit system.

Reachability graphs $g_{144}$ and $g_{288}$, shown in Figures \ref{G144WithContractedGraph}--\ref{PhaseConnectedg144_288ContractedGraph}, can be generated by the action of $\HC$ or $\mathcal{C}_2$ on the $6$-qubit state in Eq. \eqref{SixQubit144State}.
\begin{equation}\label{SixQubit144State}
\begin{split}
\frac{1}{8}(1, &-1, 1, 1, -1, 1, 1, 1, 1, -1, 1, 1, 1, -1, -1, -1, 1, -1, -1, -1, -1, 1, -1, -1, -1, 1, 1, \\
& 1, -1, 1, -1, -1, -1, 1, 1, 1, 1, -1, 1, 1, -1, 1, 1, 1, -1, 1, 1, -1, -1, -1, 1,-1, 1, 1, 1,\\
& \qquad  1, -1, 1, -1, -1, -1, 1, 1, 1)\\
\end{split}
\end{equation}

There are $5$ distinct entropy vectors that can be reached in the orbit of Eq. \eqref{SixQubit144State} under $\HC$ and $\mathcal{C}_2$, given in Table \ref{tab:g144g288EntropyVectorTable}. The colors in the table correspond to the vertex colors in Figures \ref{G144WithContractedGraph}--\ref{PhaseConnectedg144_288ContractedGraph}.
\begin{table}[h]
    \centering
    \begin{tabular}{|c||c|}
    \hline
    Label & Entropy Vector\\
    \hline
    \hline
    \fcolorbox{black}{red}{\rule{0pt}{6pt}\rule{6pt}{0pt}} & $\left(1,0,1,1,1,1,1,2,2,2,2,1,1,1,1,2,2,2,2,2,2,2,2,2,2,2,2,2,2,2,2 \right)$\\
    \hline
    \fcolorbox{black}{blue}{\rule{0pt}{6pt}\rule{6pt}{0pt}} & $\left(0,1,1,1,1,1,1,1,1,1,1,2,2,2,2,2,2,2,2,2,2,2,2,2,2,2,2,2,2,2,2 \right)$\\
    \hline
    \fcolorbox{black}{green}{\rule{0pt}{6pt}\rule{6pt}{0pt}} & $\left(1,1,1,1,1,1,1,2,2,2,2,2,2,2,2,2,2,2,2,2,2,2,2,2,2,2,3,3,3,3,2\right)$\\
    \hline
    \fcolorbox{black}{yellow}{\rule{0pt}{6pt}\rule{6pt}{0pt}} & $\left(1,1,1,1,1,1,1,2,2,2,2,2,2,2,2,2,2,2,2,2,2,2,2,2,2,3,3,2,2,3,3 \right)$\\
    \hline
    \fcolorbox{black}{magenta}{\rule{0pt}{6pt}\rule{6pt}{0pt}} & $\left(1,1,1,1,1,1,1,2,2,2,2,2,2,2,2,2,2,2,2,2,2,2,2,2,2,3,2,3,3,2,3\right)$\\
    \hline
    \end{tabular}
\caption{Table of the $5$ entropy vectors found on $g_{144}$ and $g_{288}$ reachability graphs in Figures \ref{G144WithContractedGraph}--\ref{PhaseConnectedg144_288ContractedGraph}. Colors in the leftmost column correspond to the vertex colors of these figures.}
\label{tab:g144g288EntropyVectorTable}
\end{table}

To construct the reachability graphs shown in Figure \ref{G1152WithContractedGraph}--\ref{FullC2WithContractedGraph}, we consider the orbit of the $8$-qubit state in Eq. \eqref{EightQubitState} under the action of $\HC$ and $\mathcal{C}_2$.
\begin{equation}\label{EightQubitState}
\begin{split}
\frac{1}{\sqrt{32}}(0, &0, 0, 0, 0, 0, 0, 0, 0, 0, 0, 0, 0, 0, 0, 0, 0, 0, 0, 0, 0, 0, 0, 0, 0, 0, 0,0, 0, 0, 0, 0,0, 0, 0, 0, \\
& \quad 0, 1, 0, -i, 0, -1, 0, -i, 0, 0, 0, 0, 0, 0, 0, 0, i, 0, -1, 0, -i, 0,-1, 0, 0, 0, 0,0, 0, 0,\\
& \quad \quad   0, 0,  0, 0, 0, 0, 0, 0, 0, 0, 0, 0, 0, 0,0, 0, 0, 0, 0,0, 0, 0, 0, 0, 0, 0, 0, 0, 0, 0, 1, 0, -i, 0,  \\
& \quad \quad  0, 0, 0, 0, 0, 0, 0,0, -1, 0, -i, 0,0, i, 0, -1, 0,0, 0, 0, 0, 0,0, 0, 0, -i, 0,-1, 0, 0, 0,   \\
& \quad \quad \quad  0, 0,   0, 0, 0, 0, 0,0, 0, 0, 0, 0,0, 0, 0, 0, 0, 0, 0, 0, 0, 0,0, 0, 0, 0, 0, 0, 0, 0, -i, 0,\\
& \quad \quad \quad \quad  -1, 0, 0,0, 0, 0, 0, 0,0, 0, i, 0, -1,-1, 0, -i, 0, 0, 0, 0, 0, 0, 0, 0, 0, 1, 0, -i, 0,  \\
& \quad \quad \quad \quad 0, 0, 0, 0, 0,0, 0, 0, 0, 0, 0, 0, 0, 0, 0,0, 0, 0, 0, 0, 0, 0, 0, 0, 0, 0, 0,0, 0, 0, 0, 0, \\
& \quad \quad \quad \quad 0, 0, 0, 0, i,0,  1, 0, -i, 0,1, 0, 0, 0, 0, 0, 0,0, 0, 0, 0, 1,0, i, 0, -1, 0,i, 0, 0, 0, 0) \\
\end{split}
\end{equation}

The entropy vectors generated along the $\HC$ and $\mathcal{C}_2$ orbits of Eq. \eqref{EightQubitState} are given in Figure \ref{EightQubitEntropyVectors}. The color preceding each entropy vector corresponds to the vertex coloring in Figures \ref{G1152WithContractedGraph}--\ref{FullC2WithContractedGraph}.
    \begin{figure}[h]
        \centering
        \includegraphics[width=\textwidth]{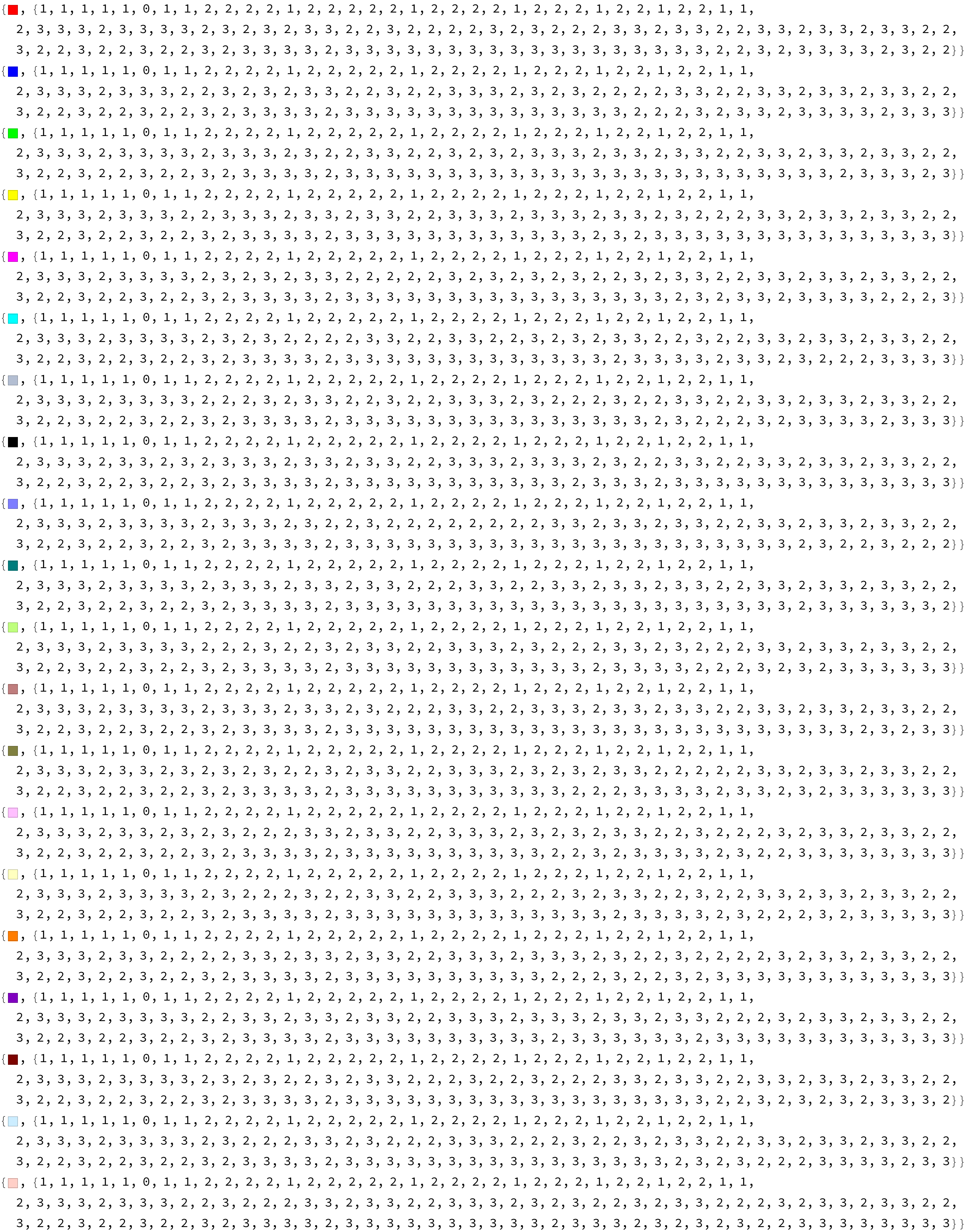}
        \caption{All $8$-qubit entropy vectors reached in the orbit of Eq. \ref{EightQubitState} under the action of $\mathcal{C}_2$. Of these $20$ entropy vectors, $18$ can be generated with $\HC$ alone.}
        \label{EightQubitEntropyVectors}
    \end{figure}

The orbit of $\ket{D^3_1}$ under $\HC$ and $\mathcal{C}_2$ reaches $5$ entropy vectors, built of $4$ different entangle entropy values. We define these $4$ unique entropy values in Eq.\ \eqref{Entropies31}.
\begin{equation}\label{Entropies31}
    \begin{split}
        s_0&\equiv 1,\\
        s_1&\equiv \frac{2}{3}\log_2\left[\frac{3}{2}\right] +\frac{1}{3}\log_2\left[3\right],\\
        s_2&\equiv \frac{5}{6}\log_2\left[\frac{6}{5}\right] +\frac{1}{6}\log_2\left[6\right],\\
        s_3&\equiv \frac{3-\sqrt{5}}{6}\log_2\left[\frac{6}{3-\sqrt{5}}\right] +\frac{3+\sqrt{5}}{6}\log_2\left[\frac{6}{3+\sqrt{5}}\right],
    \end{split}
\end{equation}

The specific entropy vectors encountered in the $\HC$ and $\mathcal{C}_2$ orbit of $\ket{D^3_1}$ are given in Table \ref{tab:WStateEntropyVectorTable}. Each entropy vector is built from the entanglement entropies given in Eq.\ \ref{Entropies31}. Numerical approximations for each entropy vector were provided in Figure \ref{WStateG288WithContractedGraph} when each first appeared.

\newpage

\begin{table}[h]
    \centering
    \begin{tabular}{|c||c|}
    \hline
    Label & Entropy Vector\\
    \hline
    \hline
    \fcolorbox{black}{red}{\rule{0pt}{6pt}\rule{6pt}{0pt}} & $(s_1,\,s_1,\,s_1)$\\
    \hline
    \fcolorbox{black}{blue}{\rule{0pt}{6pt}\rule{6pt}{0pt}} & $(s_3,\,s_1,\,s_1)$\\
    \hline
    \fcolorbox{black}{green}{\rule{0pt}{6pt}\rule{6pt}{0pt}} & $(s_1,\,s_3,\,s_1)$\\
    \hline
    \fcolorbox{black}{yellow}{\rule{0pt}{6pt}\rule{6pt}{0pt}} & $(s_0,\,s_0,\,s_1)$\\
    \hline
    \fcolorbox{black}{magenta}{\rule{0pt}{6pt}\rule{6pt}{0pt}} & $(s_2,\,s_2,\,s_1)$\\
    \hline
    \end{tabular}
\caption{Table showing the $5$ entropy vectors seen in Figures \ref{WStateG288WithContractedGraph} and \ref{PhaseConnectedWG288ContractedGraph}, reached in the orbit of $\ket{D^3_1}$ under $\HC$ and $\mathcal{C}_2$. For clarity, we introduce variables in Eq.\ \eqref{Entropies31} to succintly present each entropy vector.}
\label{tab:WStateEntropyVectorTable}
\end{table}

\newpage

Similarly for the orbit of $\ket{D^4_2}$ under $\HC$ and $\mathcal{C}_2$, we observe $6$ different entropy vectors. Following the notation of \cite{Munizzi:2023ihc}, we give these $6$ entropy vectors in terms of their $5$ distinct entanglement entropy components, which we list in Eq. \eqref{Entropies42}.
\begin{equation}\label{Entropies42}
    \begin{split}
        s_0&\equiv \frac{5}{6}\log_2\left[\frac{12}{5}\right] +\frac{1}{6}\log_2\left[12\right],\\
        s_1&\equiv \frac{3-\sqrt{5}}{6}\log_2\left[\frac{12}{3-\sqrt{5}}\right] +\frac{3+\sqrt{5}}{6}\log_2\left[\frac{12}{3+\sqrt{5}}\right],\\
        s_2&\equiv \frac{2}{3}\log_2\left[\frac{3}{2}\right] +\frac{1}{3}\log_2\left[6\right],\\
        s_3&\equiv \frac{3-2\sqrt{2}}{6}\log_2\left[\frac{12}{3-2\sqrt{2}}\right] +\frac{3+2\sqrt{2}}{6}\log_2\left[\frac{12}{3+2\sqrt{2}}\right],\\
        s_4&\equiv 1,\\
        s_5&\equiv \frac{2}{3}\log_2\left[\frac{3}{2}\right] +\frac{1}{3}\log_2\left[3\right],\\
        s_6&\equiv \frac{5}{6}\log_2\left[\frac{6}{5}\right] +\frac{1}{6}\log_2\left[6\right].
    \end{split}
\end{equation}

The $5$ entropies in Eq. \eqref{Entropies42} build the $6$ entropy vectors in Table \ref{tab:D42EntropyVectors}.
\begin{table}[h]
    \centering
    \begin{tabular}{|c||c|}
    \hline
    Label & Entropy Vector\\
    \hline
    \hline
    \fcolorbox{black}{red}{\rule{0pt}{6pt}\rule{6pt}{0pt}} & $\left(s_4,\,s_4,\,s_4,\,s_4,\,s_2,\,s_2,\,s_2 \right)$\\
    \hline
    \fcolorbox{black}{blue}{\rule{0pt}{6pt}\rule{6pt}{0pt}} & $\left(s_6,\,s_5,\,s_4,\,s_4,\,s_2,\,s_1,\,s_1 \right)$\\
    \hline
    \fcolorbox{black}{green}{\rule{0pt}{6pt}\rule{6pt}{0pt}} & $\left(s_5,\,s_6,\,s_4,\,s_4,\,s_2,\,s_1,\,s_1 \right)$\\
    \hline
    \fcolorbox{black}{yellow}{\rule{0pt}{6pt}\rule{6pt}{0pt}} & $\left(s_4,\,s_4,\,s_4,\,s_4,\,s_2,\,s_0,\,s_0 \right)$\\
    \hline
    \fcolorbox{black}{magenta}{\rule{0pt}{6pt}\rule{6pt}{0pt}} & $\left(s_4,\,s_4,\,s_4,\,s_4,\,s_2,\,s_2,\,s_2 \right)$\\
    \hline
    \fcolorbox{black}{cyan}{\rule{0pt}{6pt}\rule{6pt}{0pt}} & $\left(s_6,\,s_6,\,s_4,\,s_4,\,s_2,\,s_3,\,s_3 \right)$\\
    \hline
    \end{tabular}
\caption{The $6$ entropy vectors in the orbit of $\ket{D^4_2}$ under $\HC$ and $\mathcal{C}_2$. The vectors appears in Figures \ref{G576WithContractedGraph} and \ref{PhaseConnectedG576ContractedGraph}, and are built using the variables in Eq.\ \eqref{Entropies42}.}
\label{tab:D42EntropyVectors}
\end{table}
\chapter{PROOF OF CLIFFORD INVARIANCE}\label{app:stab0invariance}
\newpage
\noindent

\section{Invariance of $\stab_0$}\label{app:stab0invariance}
In this section, we prove that $\stab_0$ is invariant under the following operations
\begin{enumerate}
	\item Clifford unitaries. $\rho\rightarrow U\rho U^{\dagger}$ with $U\in \mathcal{C}(d^n)$.
	\item Composition with stabilizer states, $\rho \rightarrow \rho \otimes \sigma $ with $\sigma $ a stabilizer state.
 	\item Partial trace of the first qudit, $\rho \rightarrow \Tr_{1}(\rho)$
	\item Computational basis measurement on the first qudit, $\rho \rightarrow (\st{i}\otimes \bbbone_{n-1}) \rho (\st{i}\otimes \bbbone_{n-1} )/\Tr(\rho \st{i}\otimes \bbbone_{n-1})$ with probability $\Tr(\rho \st{i}\otimes \bbbone_{n-1})$

\end{enumerate}
\begin{proposition}{Clifford Invariance.}
Given $\sigma\in\stab_0$ and $C\in\mathcal{C}(d^n)$, then $C\sigma C^{\dagger}\in\stab_0$
\begin{proof}
    \begin{equation}
    C\sigma C^\dagger = \frac{1}{d^n}\sum_{P\in G} C PC^\dagger = \frac{1}{d}\sum_{\tilde{P}\in\tilde{G}} \tilde{P}
\end{equation}
the latter is an element of $\stab_0$  since it is the equal-weighted sum of Pauli operators of a commuting set. This is since $C:P\mapsto \tilde{P}\in\tilde{\mathcal{P}}$ and the action of a unitary on a subgroup $G$ does not modify the commutation relations.
\end{proof}
\end{proposition}
\begin{proposition}
    Given $\rho\in\stab_0$ and $\tau\in\stab_0$ then $\rho\otimes\tau\in\stab_0$
\end{proposition}
\begin{proof}
    \begin{equation}
        \rho\otimes\sigma = \frac{1}{d^{2n}}\sum_{P\in G_1, Q\in G_2}P\otimes Q=\frac{1}{d^2}\sum_{P\otimes Q\in G_1\times G_2} P \otimes Q. 
    \end{equation}
where the latter is an element of $\stab_0$ since the tensor product of Pauli operators is still a Pauli operator and the Cartesian product of a group is still a group, and since the tensor product does not affect the commutation relations of the $G_1$ or $G_2$, then $G_1\times G_2$ is a commuting group and so $\rho\otimes\sigma\in\stab_0$.
\end{proof}
\begin{proposition}
    Given a state $\rho\in\stab_0$ then $\Tr_1{\rho}\in\stab_0$
\end{proposition}
\begin{equation}
    \Tr_1(\rho)=\frac{1}{d^n}\sum_{P\in G}\Tr(P_1)P_{2\ldots n}=\frac{1}{d^{n-1}}\sum_{P_{2\ldots n}\in\Tr_1(G)}P_{2\ldots n}
\end{equation}
where $P_1$ labels the Pauli operator on the first qudit of $P$. It is easy to observe that the only elements whose partial trace is different from $0$ are the ones with $P_1=\bbbone$. These elements that were in $G$ are still commuting Pauli operator in the traced group $\Tr_1{G}$. 
\begin{proposition}
    Given a state $\rho$ and $\{|i\rangle\}$ the 1-qudit computational basis, then $(\st{i}\otimes \bbbone_{n-1}) \rho (\st{i}\otimes \bbbone_{n-1} )/\Tr(\rho \st{i}\otimes \bbbone_{n-1})\in\stab_0$
\begin{proof}
For the sake of simplicity, let us consider the case for a multi-qubit system and $i=0$, it can be easily generalized for $i\neq 0$ and qudits.
    \begin{align}
    \frac{(\st{0}\otimes \bbbone_{n-1}) \rho (\st{0}\otimes \bbbone_{n-1} )}{\Tr(\rho \st{0}\otimes \bbbone_{n-1})}&=\frac{\sum_{P\in G}\Tr(\st{0}P_1) \st{0}\otimes P_{2\ldots n}}{\sum_{P\in G}\Tr(\st{0}P_1)\Tr(P_{2\ldots n})}\\
    &=\frac{1}{2^{n-1}}\frac{\sum_{P\in G}\Tr(P_1\st{0})\st{0}\otimes P_{2\ldots n}}{\sum_{P_1\in\Tr_{2\ldots n}G} \Tr(P_1\st{0})}\\
    &=\frac{1}{2^{n-1}}\frac{\sum_{P\in G|P_1\in\{\bbbone,Z\}}\st{0}\otimes P_{2\ldots n}}{\sum_{P_1\in{\bbbone,Z}\cap\Tr_{2\ldots n}G}}\\
    &=\st{0}\otimes\frac{1}{2^{n-1}}\frac{\sum_{P\in G|P_1\in\{\bbbone,Z\}}P_{2\ldots n}}{\sum_{P_1\in{\bbbone,Z}\cap\Tr_{2\ldots n}G}}\\
    \end{align}
    note that $\sum_{P_1\in{\bbbone,Z}\cap\Tr_{2\ldots n}G}$ can be either $1$ or $2$, due to the terms $\bbbone_{n}$ $Z\bbbone_{n-1}$. While on the numerator the only terms surviving have on the first qubit $\bbbone$ or $Z$. Now it is not difficult to see that for $G$ to be a commuting group if $\sum_{P_1\in{\bbbone,Z}\cap\Tr_{2\ldots n}G}=1$ then there will be no multiplying factor to the numerator, while in the other case, there will be a $2$ since each non-zero $P_{2\ldots n}$ has to repeat twice. Then it is not difficult to see that one has a stabilizer state, because $P_{2\ldots n} $ is still summing on a commuting Pauli subgroup. 
\end{proof}

\end{proposition}
\subsection{Prof of~\cref{prop:qrf}}\label{qrf}
\begin{proof}
   Let us start by expanding the relative entropy, we have 
   \begin{align}
   \mathcal{F}_R(\rho)=-\min_{\sigma\in \mathrm{FLAT}^{(n)}}\Tr[\rho\log\sigma]-S(\rho).
   \end{align}
   Since the elements of FLAT are all proportional to projection operators,  it is possible to choose $\sigma$ such that $\sigma=\bbbone_{ r}/r\oplus 0_{d- r}$ is diagonal in the same basis as $\rho$ where $r\equiv \rank(\rho)$, $\bbbone_r$ is the identity on a subspace of dimension $r$, and $0_{d-r}$ is the zero-matrix of dimension $(d-r)\times (d-r)$. Hence the first term becomes $\log r\,\Tr\rho = \log r=S_{max}(\rho)$ and $\mathcal F_R(\rho)\leq S_{max}(\rho)-S(\rho)$. 

    To show that the minimum is attained for a rank $r$ density operator $\rho$ when $\mathcal F_R(\rho)= S_{max}(\rho)-S(\rho)$, suppose on the contrary that there exists $\sigma=\Pi_k/k$, where $\Pi_k$ is a projection operator of rank $k$ such that the first term is less than $\log r$ and making $\mathcal{F}_R(\rho)<S_{\rm max}(\rho)-S(\rho)$. Let $M=\rho\log\sigma$; in the diagonal basis of $\sigma$, where we denote the diagonal element by $\lambda_i$, one has  
    
  \begin{align}
  \Tr M=\sum_i M_{ii} = \sum_i \rho_{ij}\delta_{ji}\log\lambda_i = \sum_i \rho_{ii}\log\lambda_i.
  \end{align}
Let us note that when $i>k$, $\log\lambda_{i>k}=-\infty$. Then in order for the trace of $M=\rho\log\sigma$ to be finite, we need $\rho_{ii}=0$ for any $\lambda_i\ne 0$ to be $0$.
    On the other hand, we know that if the diagonal of a positive semi-definite matrix has a zero on the diagonal, then the corresponding rows and columns must be all $0$s. Then it implies that up to rearranging the rows and columns for the sake of clarity, 
    \begin{align}
    \rho = \begin{pmatrix}
        A & 0\\
        0 & 0
    \end{pmatrix}
    \end{align}
    where $A$ is a $k\times k$ block matrix of rank at most $k$. Therefore, $-\Tr M = \log k\, \Tr[AI] = \log k<\log r$ by assumption, we must have $k<r$. Because $\dim A\geq \rank(A)$, it follows that $\rank(A)=\rank(\rho)\leq k<r$, which is a contradiction. 
\end{proof}
\section{Proof of~\cref{th:magicdist}}\label{app:proofthmd}
In this section, we prove~\cref{th:magicdist}. Let us start from the upper bound. Being defined through two minima, we can arbitrarily choose a state $\psi$ and a stabilizer $\sigma$ to upper bound $M_{\text{dist}}^{(NL)}$. Consider the state $\ket{\psi_{AB}}$, whose Schmidt decomposition can be written as $\ket{\psi_{AB}}=\sum_{i}^{D}\lambda_i\ket{\lambda_i^{A}\lambda_i^{B}}$, where $\lambda_i$ are the Schmidt coefficients and $D$ its Schmidt rank. Due to the minimization over $U=U_{A}\otimes U_B$, the basis $\ket{\lambda_{i}^{A/B}}$ can be brought in the computational basis (or another complete stabilizer basis)
\be 
U\ket{\lambda_{i}^{A/B}}=\ket{s_i^{A/B}}
\ee
Then, we choose $\ket{\sigma}=\sum_{i}^{d^{\lfloor \log_d D \rfloor}}\frac{1}{d^{\lfloor \log_d D \rfloor/2}} \ket{s_i^As_i^B}$, where $\lfloor \cdot \rfloor$ labels the floor function. Let us then compute the upper-bound to $M_{\text{dist}}^{(NL)}(\psi_{AB})$. 
\begin{align}
M_{\text{dist}}^{(NL)}(\psi_{AB})&\le \frac{1}{2}\norm{\sum_{ij}^{D}\lambda_i\lambda_j\ketbra{s_i^A s_i^B}{s_j^A s_j^B}-\sum_{ij}^{d^{\lfloor \log D \rfloor}}d^{-\lfloor \log_d D \rfloor}\ketbra{s_i^A s_i^B}{s_j^A s_j^B}}\\ 
&=\sqrt{1-\sum_{ij}^{D}\sum_{kl}^{d^{\lfloor \log D \rfloor}}d^{-\lfloor \log D \rfloor}\lambda_{i}\lambda_{j}\braket{s_l^A s_l^B}{s_i^A s_i^B}\braket{s_j^A s_j^B}{s_k^As_k^B}}
\end{align}
where we first used the Fuchs-Van der Graaf inequality, where the equality comes by $\psi_{AB}$ and $\sigma$ being pure states, and then rewritten the states in their Schmidt decomposition. Now without loss of generality, since $d^{\lfloor \log_d D \rfloor}\le D$, due to our degrees of freedom in the choice of $\sigma$ and $\psi_{AB}$ we can order the basis states such that only the first $d^{\lfloor \log_d D \rfloor}$ have nonzero overlap, and so it follows: 
\begin{align}
M_{\text{dist}}^{(NL)}(\psi_{AB})&\le \sqrt{1-\sum_{ij}^{d^{\lfloor \log_d D \rfloor}}d^{-\lfloor \log_d D \rfloor}\lambda_{i}\lambda_{j}}\\ 
&\le \sqrt{1-\sum_{ij}^{d^{\lfloor \log_d D \rfloor}}D^{-1}\lambda_{i}\lambda_{j}}\\
&\le \sqrt{1-\sum_{i}^{d^{\lfloor \log_d D \rfloor}}D^{-1}\lambda_{i}^2}\\ 
&=\sqrt{1-\frac{1}{D}+\sum_{i=d^{\lfloor \log D \rfloor}}^{D}D^{-1}\lambda_{i}^2}\\ 
&=\sqrt{1-\frac{1}{D}+\lambda_{\max}^{2}
\left(1-\frac{d^{\lfloor \log_d D \rfloor}}{D}\right)}\\ 
&=\sqrt{1-e^{S_{max}(A)}+e^{S_{\infty}(A)}\left(1-\frac{e^{\log d \lfloor S_{max}(A)/\log d \rfloor}}{e^{S_{max}(A)}}\right)}
\end{align}
where we first utilized the inequality $D^{-1}\leq d^{-\log_d D}$, and since $\lambda_i>0$ by definition, we can upper bound $M_{\text{dist}}^{(NL)}(\psi_{AB})$ by simply considering the diagonal terms. Next, we employed $\sum_{i}^{D}\lambda_i^2=1$ to rewrite our inequality, and finally, we utilized $S_{\max}(A)=\log{D}$ and $S_{\infty}=\log \lambda_{\text{max}}^2$.

Let us now focus on the lower bound. To prove it let us first provide a bound between $\mathcal{F}(\psi)$ and $ M_{dist}(\psi)$.
\begin{lemma}\label{lemmaflatmdist}
Let $\psi$ be a state then its flatness $\mathcal{F}(\psi)$ is upper bounded by $M_{\text{dist}}$ as follows
\begin{equation}
    \mathcal{F}(\psi)\le 8M_{\text{dist}}(\psi).
\end{equation}
\begin{proof}
Starting from the flatness one can add a zero term to it; take a flat state $\sigma\in\stab_0$
\begin{equation}
  \mathcal{F}(\psi)=\mathcal{F}(\psi)-\mathcal{F}(\sigma).
  \label{eq:flatsigma}
\end{equation}
We can then bound the flatness as follows:
\begin{align}
  \mathcal{F}(\psi)&=\Tr(\psi^3-\sigma^3)-\Tr\left((\psi^2)^{\otimes 2}-(\sigma^2)^{\otimes 2}\right)\\
                     &=\left| \Tr(\psi^3-\sigma^3) \right| + \left| \Tr\left((\psi^2)^{\otimes 2}-(\sigma^2)^{\otimes 2}\right) \right| \\
                     &\le \left| \Tr(\psi^3-\sigma^3) \right|+ 2 \left|\Tr\left((\psi^2)-(\sigma^2)\right) \right|\\
                     &\le 1- (1-T)^3 +2 - 2(1-T)^2\le T^3+ 7T\le 8T
\end{align}
where $T=1/2 \left\| \psi -\sigma \right\|_1 $. In the second line we made use of the triangular inequality, in the third line, we used the following inequality
\begin{align}
    |\Tr(\psi^2)\Tr(\psi^2)-\Tr(\sigma^2)\Tr(\sigma^2)|&\le|\Tr(\psi^2)\left(\Tr(\psi^2)-\Tr(\sigma^2)\right)|+|\left(\Tr(\psi^2)-\Tr(\sigma^2)\right)\Tr(\sigma^2)|\\&\le 2|\Tr(\psi^2)-\Tr(\sigma^2)|
\end{align}
while in the fourth line we used \cite[ Lemma 1.2]{chen_sharp_2016} and then $T^3\le T$, since $0\le T\le 1$. 
By minimizing over $\sigma\in\stab$ we prove the lower bound with $M_{dist}(\psi)$. 
\end{proof}
\end{lemma}
Using Lemma~\ref{lemmaflatmdist}, we can thus write,  
\begin{equation}
    \mathcal{F}(\psi_A)\le 8 \min_{U_A}M_{\text{dist}}(U_A\psi_A U_A^{\dag})\label{cscsc}
\end{equation}
where $\psi_A=\tr_B{\psi_{AB}}$ and we used that $\mathcal{F}(\psi_A)$  is invariant under the action of global unitaries. 
Now, let us show that $ \min_{U_A}M_{\text{dist}}(U_A\psi_A U_A^{\dag})\le M_{\text{dist}}^{(NL)}(\psi_{AB})$. First recall that given $\psi_A=\Tr_B(\psi_{AB}) $ due to the monotonicity of $M_{\text{dist}}$ one has $M_{\text{dist}}(\psi_A)\le M_{\text{dist}}(\psi_{AB})$. Now let us prove the statement by contradiction. First, let $U_A$ be the unitary attaining the minimum in Eq.~\eqref{cscsc}. Let us suppose that there exists a bipartite unitary $U\equiv V_A\otimes V_B$ obeying
\begin{equation}
    M_{\text{dist}}(U_A \psi_A U_A^{\dag})> M_{\text{dist}}(U \psi_{AB} U^{\dag})
\end{equation}
Then we have the following chain of inequalities
\begin{equation}
    M_{\text{dist}}(U_A \psi_A U_A^{\dag})> M_{\text{dist}}(U \psi_{AB} U^{\dag})\ge  M_{\text{dist}}(\tr_B U \psi_{AB} U^{\dag}) =  M_{\text{dist}}(V_A \psi_{A} V_A^{\dag}) 
\end{equation}
and this is a contradiction to the statement that $U_A$ attains the minimum. Therefore, one obtains that $ \min_{U_A}M_{\text{dist}}(U_A\psi_A U_A^{\dag})\le M_{\text{dist}}^{(NL)}(\psi_{AB})$. This result combined with ~\cref{lemmaflatmdist} concludes the proof.

\section{Stabilizer relative entropies}\label{app:nlentropies}
\subsection{Proof of~\cref{th:relstab}}\label{proofth2}
Let us start by proving the upper-bound to $M_{RS}^{(NL)}$. We choose $\sigma_{AB}=\bbbone_{d^{\lceil \log_d  D  \rceil}}/d^{\lceil\log_d  D  \rceil}\oplus 0_{n-\lceil \log_d D  \rceil}$ where $D$ is the Schmidt rank of $\rho_A$.  Then expanding the relative entropy expansion one obtains the following bound
\begin{align}
M_{RS}^{(NL)}(\psi_{AB})&\le-\Tr[\psi_{AB} \log\sigma_{AB}] = \Tr[\sum_{i,j=1}^D\lambda_i\lambda_j |s_i\rangle\langle s_j| \sum_{k=1}^{d^{\lceil\log_d  D  \rceil}} |s_k\rangle\langle s_k|\log d \lceil \log_d D  \rceil] \\&= \lceil \log_d  D  \rceil \log d\,   \Tr[\psi_{AB}]=\log d \lceil   S_{max}(A)/\log d  \rceil,
\end{align}
%
Concluding the proof for the upper bound. Shifting our focus on the lower bound instead, let us note that for any $\rho$, $M_{RS}(\rho)\geq \mathcal{F}(\rho)$. This is a simple consequence of $\stab_0^{(n)}\subset \mathrm{FLAT}^{(n)}$. Because $\mathcal{F}(\rho)$ is isospectral under any unitary conjugation, it must follow that $M_{RS}(U \rho U^{\dagger})\geq F(U\rho U^{\dagger})=F(\rho)$. Therefore, 
\begin{align}
\mathcal F_R(\rho_A)\leq \min_{U_A} M_{RS}(U_A\rho_A U_A^{\dagger}).
\end{align}

    On the other hand, for any $\rho_A$, from monotonicity it follows that $M_{RS}(\rho_{AB})\geq M_{RS}(\rho_A)$ where $\rho_A=\Tr_B[\rho_{AB}]$. Therefore, we must have \begin{equation}
       \min_{U_A} M_{RS}(U_A\rho_A U_A^{\dagger})\leq \min_{U=V_A\otimes V_B}M_{RS}(U\rho_{AB}U^{\dagger}) \equiv M_{RS}^{(NL)}(\rho_{AB}).
    \end{equation}
    We can see that this is true from a proof by contradiction. Suppose there exists some $U_{A},U$ which attains the respective minima but has $$M_{RS}(U_A\rho_A U_A^{\dagger})> M_{RS}(U\rho_{AB} U^{\dagger}),$$ then from monotonicity, we must have $$M_{RS}(U_A\rho_A U_A^{\dagger})>M_{RS}(U\rho_{AB}U^{\dagger})\geq M_{RS}(\Tr_B [U\rho_{AB} U^{\dagger}]) = M_{RS}(V_A\rho_A V_A^{\dagger})$$ for some local unitary $V_A$ which yields a lower distance than $U_A$. Since we assumed that $U_A$ attains the minimum, this violates our assumption, concluding the proof for the lower bound. 

    \subsection{Proof of~\cref{prop:RelativeEnt}}\label{proofprop4}

To bound the non-local magic of a pure state $\rho_{AB}$, consider a pure state $\psi_{AB}= U \rho_{AB}U^{\dagger}$ where $U=U_A\otimes U_B$ and $\psi_A$, $\psi_B$ are isospectral (up to truncation of 0 eigenvalues) to that of a subsystem $\rho_A,\rho_B$. Suppose they are states where we have removed the local magic such that both $\psi_A,\psi_B$ are diagonal in the computational (or another complete stabilizer basis). Again, this can be done by first rewriting the state $\rho_{AB}$ in the Schmidt basis, which is orthonormal. Then we replace the Schmidt basis with an orthonormal stabilizer basis to get $\psi_{AB}$. Since the mixture of stabilizer states is in the convex hull of stabilizer states, each $\psi_A,\psi_B$ must have zero local magic. Note that there are also other basis choices such that the basis state need not be a stabilizer, such states can also be in the convex hull of the stabilizer group as long as they are not pure states. 

By definition, $M^{(NL)}_{R}(\psi_{AB})\leq M_{R}(\psi_{AB})$ because we have chosen a particular instance of the local unitary $U_A\otimes U_B$ on the right hand side whereas the left hand side is minimized over all possible instances.  Now we evaluate the relative entropy of magic $M_R(\psi_{AB})=-S(\psi_{AB})-\min_{\sigma \in \mathrm{STAB}}\Tr[\psi_{AB}\log \sigma_{AB}]$. Since $\psi_{AB}$ is pure, $S(\psi_{AB})=0$. If $\sigma_{AB}$ is pure, then the relative entropy is either $0$ when $\sigma=\psi$ or $\infty$ for any other $\sigma$ that's mixed.

We pick a stabilizer state $\sigma_{AB} = \sum_i\lambda_i^2 |s_i\rangle\langle s_i|_{AB}$ where $\lambda_i$ are the Schmidt coefficients of $\psi_{AB} = \sum_i\lambda_i |s_i\rangle_{AB}$ where $|s_i\rangle$ are the stabilizer basis we chose. Then for the second term, we write

\begin{align}
   M_R(\psi_{AB})&= -\Tr[\psi_{AB}\log\sigma_{AB}] \\
   &= -\Tr[\sum_{ij}\lambda_i\lambda_j |s_i\rangle\langle s_j|\sum_k \log(\lambda_k^2)|s_k\rangle\langle s_k|]\\
    &=-\sum_{i,j,k}\delta_{ij}\delta_{ik} \lambda_i\lambda_j \log(\lambda_k^2) \\
    &=-\sum_k p_k \log p_k
\end{align}
where we have set $\lambda_k^2=p_k$ because each Schmidt coefficient is real. $\delta_{ij}$ are Kronecker deltas because we have chosen the basis $\{|s_k\rangle\}$ to be orthonormal. Note that $\sum_k p_k=1$.
In this case, the second term is nothing but $S(A)=S(B)$ which is the von Neumann entropy of a subsystem. 

Since we have chosen a particular stabilizer state $\sigma_{AB}$, this serves as an upper bound of the relative entropy of magic. Hence 
\begin{align}
M^{(NL)}_{R}(\rho_{AB})\leq M^{(NL)}_{R}(\psi_{AB})\leq S(A)=S(B).
\end{align}

\subsection{Proof of~\cref{th:smoothed}}\label{proofth3}

    Let $\rho_{AB}^{\epsilon}$ represent the state that minimizes the non-local magic.  Therefore $\sma{\rho_{AB}}=M_{RS}^{(NL)}(\rho^{\epsilon}_{AB})$.  Drawing from \cref{th:relstab}, we understand that: 
\begin{equation}
\begin{split}
    \sma{\rho_{AB}}\geq & S_{max}(\rho^{\epsilon}_A)-S(\rho^{\epsilon}_A)\geq \min_{\Vert\chi-\rho_A\Vert<\epsilon}\left(S_{max}(\chi)-S(\chi)\right).
\end{split}
\end{equation}

On the right-hand side, our goal is to identify a state $\chi$ within the $\epsilon$-ball of $\rho_A$ that minimizes the difference between $S_{max}(\chi)$ and $S(\chi)$. Interestingly, the state that minimizes this difference also reduces $S_{max}(\chi)$ to its lowest value $S_{max}^{\epsilon}$.  To illustrate, denote $\chi_A^{\epsilon}$ as the state that minimizes $S_{max}$ within the $\epsilon$-ball. Then consider increasing $S_{max}$ by modifying one eigenvalue of $\chi_A^{\epsilon}$ from zero to $\delta$.  This adjustment results in an increase  $\Delta S_{max}=e^{-S_{max}}$, while the change in entropy is capped at  $\Delta S\leq \delta\abs{\log\delta}$. Such a modification invariably elevates the entropy gap, i.e. $\Delta (S_{max}-S) \geq e^{-S_{max}}-\delta\abs{\log\delta}>0$, since $\delta$ can be arbitrarily small. 

To evaluate the von Neumann entropy of the state  $\chi_A^{\epsilon}$, as a modification from $\rho_A$ by dropping some eigenvalues whose total contribution to the trace is smaller than $\epsilon$. Let’s denote their contribution to the von Neumann entropy as  $S_{\epsilon}$. Then the entropy of the new state $\chi_A^{\epsilon}$ is given by $S(\chi_A^{\epsilon})= \frac{S(\rho_A)-S_{\epsilon}}{1-\epsilon}\leq \frac{S(\rho_A)}{1-\epsilon}$. Therefore, we get the following inequality:

\begin{equation}
    \begin{split}
         \sma{\rho_{AB}}\geq & S_{max}(\chi_A^{\epsilon})-S(\chi_A^{\epsilon})\\ 
        \geq & S_{max}^{\epsilon}(\rho_A)-(1-\epsilon)^{-1}S(\rho_A).
    \end{split}
\end{equation}

Regarding the upper bound, since $\chi_A^{\epsilon}$ minimizes the maximal entropy, it satisfies the following condition:
\begin{equation}
    S_{max}^{\epsilon}(\rho_A)=S_{max}(\chi_A^{\epsilon}).
\end{equation}

While this condition specifies the spectrum of $\chi_A^{\epsilon}$, we retain the flexibility to select a purification  $\chi_{AB}^{\epsilon}$, ensuring its deviation from $\rho_{AB}$ remains within an $\epsilon$ bound. Consequently, the process of minimizing the non-local magic leads us to the following inequality: 
\begin{equation}
\begin{split}
    \sma{\rho_{AB}}\leq & M_{RS}^{(NL)}(\chi_{AB}^{\epsilon})\\
    \leq &(\log d) \lceil \log_d\rank{\chi_{AB}^{\epsilon}} \rceil =\log d \lceil   S_{max}^{\epsilon}(A)/\log d  \rceil.
\end{split}
\end{equation}
where the second step is a result from~\cref{th:relstab}. 


\section{Estimate by Stabilizer-R\'enyi-entropy}\label{app:estimate}

\subsection{Proof of~\cref{thm:nlSRE}}

In this section, we provide an estimation of the second Stabilizer-R\'enyi-entropy measure of the non-local magic. It is defined in \cite{stabrenyi} as the second R\'enyi-entropy of a probability distribution, $p_a=\frac{1}{d}|\langle\psi|P_a|\psi\rangle|^2$, over all the Pauli-string basis $P_a$.
\begin{equation}
    \mathcal{M}_2(\ket{\psi}):=-\log(\sum_a p_a^2)-\log{d}.
\end{equation}

Given the entanglement spectrum $\{\lambda_i\}$, we construct a state $|\psi'\rangle$ with small local magic,   
\begin{equation}
    |\psi'\rangle_{AB} = \sum_{i=0}^{r-1}\sqrt{\lambda_i}|s_i\rangle_A|s_i\rangle_B. 
\end{equation}
where the rank $r$ is taken to be $2^n$ for integer $n$. The Pauli operators on the Hilbert space $\mathcal{H}_{AB}=\mathcal{H}_A\otimes\mathcal{H}_B$ can be  factorized as product of Pauli operators on $\mathcal{H}_A$ and $\mathcal{H}_B$ respectively, $P^{ab}=P^a\otimes P^b$. We denote their matrix elements as $P^{a,b}_{ij}:=\bra{s_i}P^{a,b}\ket{s_j}$, and compute the magic measure $\mathcal{M}_2$ as follows, 
\begin{equation}
\begin{split}
    &\mathcal{M}_2(\ket{\psi'})=-\log\left(\sum_{a=1}^{r^2}\sum_{b=1}^{r^2}\left|\sum_{i,j=0}^{r-1}\sqrt{\lambda_i}\sqrt{\lambda_j}P^a_{ij}P^b_{ij}\right|^4\right).
\end{split}
\end{equation}

The result is complicated and depends on specific choice of the basis $|s_i\rangle$'s. We simplify the analysis by assuming that the orthonormal basis $|s_i\rangle$'s are common eigenstates of a stabilizer group $\mathcal{S}=\{S_1,S_2,\cdots,S_n\}$.  This condition allows us to write the Pauli matrices $P^a_{ij}$ in computational basis. Substituting the matrix representation  of Pauli operators, we find that 

\begin{equation}\label{M2Estimate}
\begin{split}
    \mathcal{M}_2=&-\log\left(\sqrt{\lambda_{i_1}}\sqrt{\lambda_{i_2}}\sqrt{\lambda_{i_3}}\sqrt{\lambda_{i_4}}\sqrt{\lambda_{i_5}}\sqrt{\lambda_{i_6}}\sqrt{\lambda_{i_7}}\sqrt{\lambda_{i_8}}(\sum_a P^a_{i_1i_2}P^a_{i_3i_4}P^a_{i_5i_6}P^a_{i_7i_8})^2\right)\\
    =&-\log\left(\sum_{i_1,i_2,i_3,i_4=0}^{r-1}\sqrt{\lambda_{i_1}\lambda_{i_2}\lambda_{i_3}\lambda_{i_4}\lambda_{i_3\wedge i_2\wedge i_1}\lambda_{i_4\wedge i_2\wedge i_1}\lambda_{i_1\wedge i_3\wedge i_4}\lambda_{i_2\wedge i_3\wedge i_4}}\right).
\end{split}
\end{equation}
where $\wedge$ denotes the bitwise XOR operation. This expression depends on order of eigenvalues and takes minimum when the eigenvalues are ordered, $\lambda_i>\lambda_j$ for  $i<j$. If we take all the eigenvalues to be the same, then each term in the summation is equal to $\frac{1}{r^4}$. The number of terms is $r^4$ since we are summing over four indices. The argument is equal to 1 in this case. Therefore, the non-local Stabilizer R\'enyi entropy vanishes when the spectrum is flat. 

\cref{M2Comparison} gives a comparison of the direct SRE calculation against the estimation given by  \eqref{M2Estimate}. As can be observed in the plots, the approximation in  \eqref{M2Estimate} is correct up to numerical imprecision.
\begin{figure}[h]
\begin{center}
    \begin{subfigure}[b]{0.49\textwidth}
        \includegraphics[width=8cm]{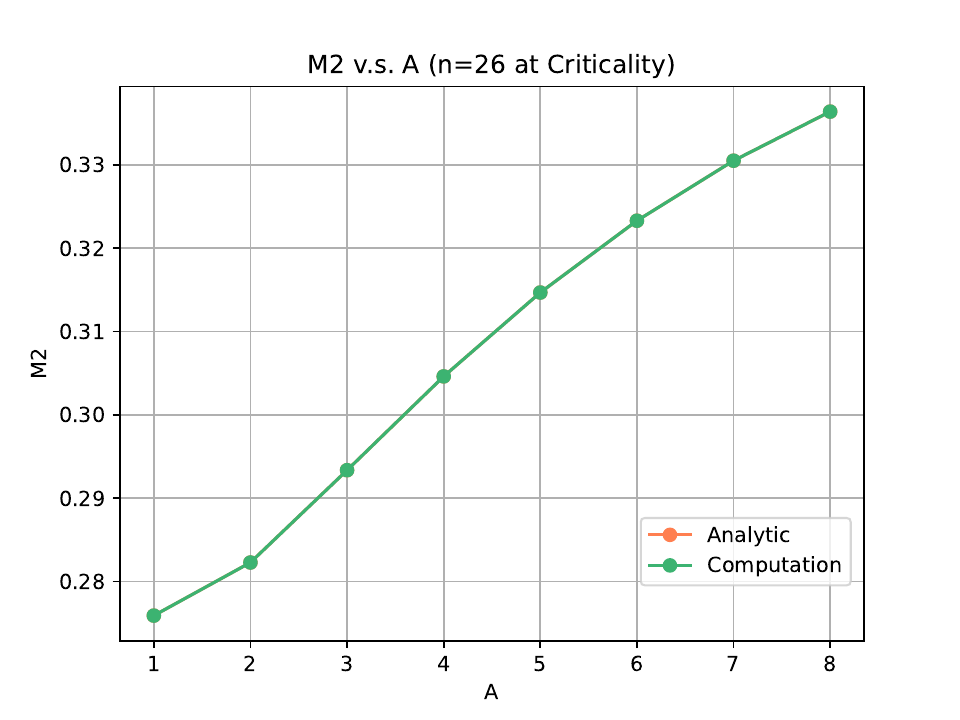}
        \caption{$\mathcal{M}_2$ computed from stabilizer R\'enyi entropy, compared to the estimation in \eqref{M2Estimate}.}
    \end{subfigure}
    \hfill
    \begin{subfigure}{0.49\textwidth}
        \includegraphics[width=8cm]{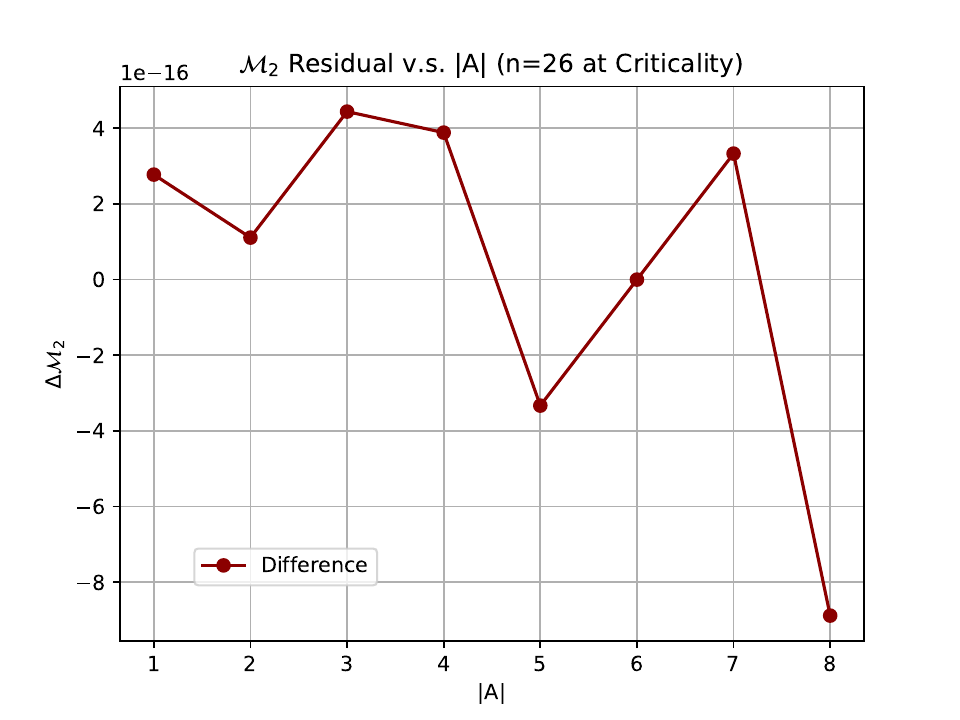}
    \caption{Residual from computational and analytic calculations of $\mathcal{M}_2$, accurate to one part in $10^{16}$.}
    \end{subfigure}
    \caption{}
    \label{M2Comparison}
\end{center}
\end{figure}
%


We can derive an upper bound for $\mathcal{M}_2$, by averaging over the permutations of eigenvalues, this gives us the expression,
\begin{equation}\label{M2UpperBound}
\begin{split}
    \mathcal{M}_2\leq\overline{\mathcal{M}}_2=-\log\left(\sum_{i=0}^{r-1}\lambda_i^4+7\sum_{0\leq i\neq j\leq r-1}\lambda_i^2\lambda_j^2+\frac{7}{r-3}\sum_{0\leq i\neq j\neq k\neq l\leq r-1}\lambda_i\lambda_j\lambda_k\lambda_l+\right.\\
    \left. \frac{\sum_{0\leq i_1\neq i_2\neq \cdots\neq i_8\leq r-1}\prod_{a=1}^8\sqrt{\lambda_{i_a}}}{(r-3)(r-5)(r-6)(r-7)}\right).
\end{split}
\end{equation}

In \eqref{M2UpperBound}, the sum inside the logarithm is taken over products of distinct eigenvalues. Computing this sum explicitly, and expressing the result in terms of different R\'enyi entropies $S_{\alpha}$, we obtain
\begin{equation}\label{eq:aveM}
\begin{split}
    \overline{\mathcal{M}}_2=-\log &\left(7e^{-2S_2}-6e^{-3S_4}+7e^{-S_0}(1-6e^{-S_2}+8e^{-2S_3}+3e^{-2S_2})+\right. \\
    &\left. e^{-4S_0}(e^{4S_{1/2}}+105e^{-3S_4}-420e^{S_{1/2}}+\cdots)\right)\\
    = -\log&\left(7e^{-2S_2}-6e^{-3S_4}+e^{4S_{1/2}-4S_0}\right)+O(e^{-S_{1/2}}).
\end{split}
\end{equation}
The averaged magic  $\overline{\mathcal{M}}_2$ is a complicated combination of R\'enyi entropies, ranging from $S_{1/2}$ to $S_4$. However, in the large Hilbert dimension limit, where $S_{1/2}\gg 1$, the averaged magic $\overline{\mathcal{M}}_2$ simplifies to the final expression in \eqref{eq:aveM}. It provides a straightforward estimate of $\mathcal{M}_2$ based on a few Rényi entropy terms. 

To establish a rigorous bound for $\mathcal{M}_2$, we start with \eqref{M2UpperBound}, leading to:
\begin{equation}\label{eq:bound1}
\begin{split}
    \overline{\mathcal{M}}_2&\leq -\log\left(\sum_{i=0}^{r-1}\lambda_i^4+7\sum_{0\leq i\neq j\leq r-1}\lambda_i^2\lambda_j^2\right)\\
    &\leq -\log\left(\left(\sum_{i}\lambda_i^2\right)^2\right)=2S_2.
\end{split}
\end{equation}

This holds for all spectrum distributions. Rewriting \eqref{M2UpperBound} in terms of the entropy difference $\delta_{1/2}=S_0-S_{1/2}$, we obtain the following expansion; 
\begin{equation}
\begin{split}
     \overline{\mathcal{M}}_2=&-\log\left( e^{-4\delta_{1/2}}+e^{-S_0}(7+21e^{-4\delta_{1/2}}-28e^{-3\delta_{1/2}})+O(e^{-2S_0})\right)\\
     \leq& 4\delta_{1/2}.
\end{split}
\end{equation}

Note that the coefficient associated with $e^{-S_0}$ in the expansion remains non-negative for any value of  $\delta_{1/2}$ and vanishes when $\delta_{1/2}=0$. Verifying that these coefficients are non-negative for every order of  $e^{-S_0}$ supports the inequality. Finally, combining this with the previously established bound finishes our proof that:
\begin{equation}
    \mathcal{M}_2(\{\lambda_i\})\leq \overline{\mathcal{M}}_2\leq \min\{2S_2,4(S_0-S_{1/2})\}.
\end{equation}

\subsection{Proof of~\cref{pp:branebound}}\label{app:branebound}

Let $\ket{\phi}$ denotes an entangled pair of qubits, with the entanglement spectrum given by $\{\lambda,1-\lambda\}$. We show that the non-local stabilizer R\'enyi entropy $\mathcal{M}_2(\lambda)$ of $\ket{\phi}$ is bounded by the non-flatness $\partial_n\tilde{S}_n$. 

From \eqref{eq:analyticalM}, we find that $\mathcal{M}_2(\lambda)$ is equal to, 
\begin{equation}
    \mathcal{M}_2(\lambda)=-\log\left(1-4\lambda+20\lambda^2-32\lambda^3+16\lambda^4\right).
\end{equation}

By definition  \eqref{eqn:holononflat}, the non-flatness is 
\begin{equation}
    -\partial_n\tilde{S}_n=n\frac{\lambda^n(1-\lambda)^n\left(\log\frac{\lambda}{1-\lambda}\right)^2}{\left(\lambda^n+(1-\lambda)^n\right)^2}.
\end{equation}

Both functions are zero at $\lambda=0,\ \frac{1}{2},\ 1$. So let's make a Taylor expansion around these value. Around $\lambda=\frac{1}{2}$, we have that 

\begin{equation}
    \begin{split}
    \mathcal{M}_2(\lambda)&=4(\lambda-1/2)^2-8(\lambda-1/2)^4+O((\lambda-1/2)^5)\\
    -\partial_n\tilde{S}_n\vert_{n=1}&=4(\lambda-1/2)^2-\frac{16}{3}(\lambda-1/2)^4+O((\lambda-1/2)^5)\\
    -\frac{1}{2}\partial_n\tilde{S}_n\vert_{n=2}&=4(\lambda-1/2)^2-\frac{160}{3}(\lambda-1/2)^4+O((\lambda-1/2)^5).
\end{split}
\end{equation}

Therefore for $\lambda$ close to $1/2$, the following inequality holds:
\begin{equation}
   -\frac{1}{2}\partial_n\tilde{S}_n\vert_{n=2} \leq \mathcal{M}_2(\lambda)\leq -\partial_n\tilde{S}_n\vert_{n=1}.
\end{equation}

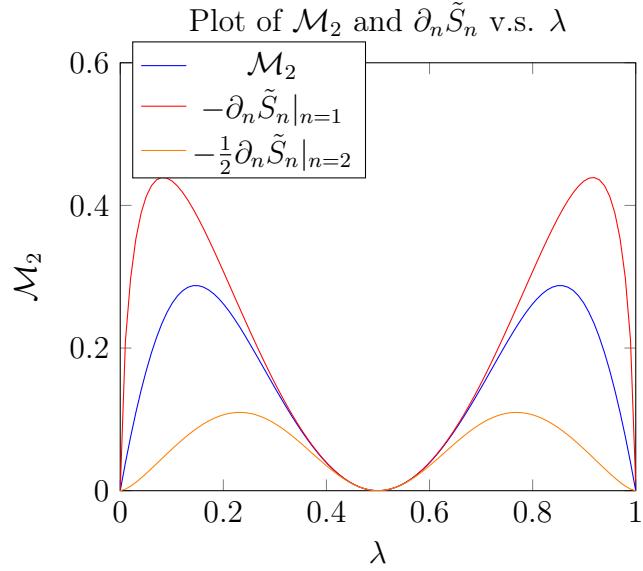
\begin{figure}
    \centering
    \begin{tikzpicture}
    \begin{axis}[
        title={Plot of $\mathcal{M}_2$ and $\partial_n\tilde{S}_n$ v.s. $\lambda$},
        xlabel={$\lambda$},
        ylabel={$\mathcal{M}_2$},
        xmin=0, xmax=1,
        ymin=0, ymax=0.6,
        legend style={at={(0.25,1.05)},anchor=north},
        grid style=dashed,
    ]
    
    \addplot[
        domain=0:1, 
        samples=100, 
        color=blue,
    ]{-ln(1-4*x+20*x^2-32*x^3+16*x^4)};
    \addlegendentry{$\mathcal{M}_2$}
    
    \addplot[
        domain=0:1, 
        samples=100, 
        color=red,
    ]{x*(1-x)*(ln(x/(1-x)))^2};
    \addlegendentry{$-\partial_n\tilde{S}_n\vert_{n=1}$}

    \addplot[
        domain=0:1, 
        samples=100, 
        color=orange,
    ]{x^2*(1-x)^2*(ln(x/(1-x)))^2/(x^2+(1-x)^2)^2};
    \addlegendentry{$-\frac{1}{2}\partial_n\tilde{S}_n\vert_{n=2}$}
    
    \end{axis}
    \end{tikzpicture}
    \caption{The non-local stabilizer R\'enyi entropy $\mathcal{M}_2$ is bounded by the anti-flatness $\partial_n\tilde{S}_n$ }
    \label{fig:singlequbit}
\end{figure}

Similarly, one can show that this inequality holds for $\lambda$ close to $0$ and $1$, where the functions are,

\begin{equation}
\begin{split}
    \mathcal{M}_2(\lambda)&=4\lambda-12\lambda^2+O(\lambda^3)\\
    -\partial_n\tilde{S}_n\vert_{n=1}&=\lambda\log^2 \lambda+(2\log \lambda-\log^2 \lambda)\lambda^2+O(\lambda^3)\\
    -\frac{1}{2}\partial_n\tilde{S}_n\vert_{n=2}&=\lambda^2\log^2 \lambda+2(\log \lambda+\log^2 \lambda)\lambda^3+O(\lambda^4).
\end{split}
\end{equation}

For other value of $\lambda$, we justify this inequality by the plot in \cref{fig:singlequbit}.

Both the Stabilizer R\'enyi entropy and  anti-flatness are additive. Therefore for state $\ket{\psi}$ that can be distilled into product of entangled pairs $U_A\otimes U_B\ket{\psi}_{AB}=\otimes_{i=1}^k\ket{\phi}_{a_ib_i}$, we have,
\begin{equation}\label{eq:branebound}
    \frac{1}{2}\left\vert\frac{\partial_n\mathcal{A}_n|_{n=2}}{4G}(\ket{\psi}_{AB})\right\vert\leq\mathcal{M}_2(\ket{\psi}_{AB})\leq\left\vert\frac{\partial_n\mathcal{A}_n|_{n=1}}{4G}(\ket{\psi}_{AB})\right\vert.
\end{equation}

\subsection{Distillation of Matrix Product State}\label{app:MPS}
We further elaborate our discussions from~\cref{section:MERA}. Building on the MERA representation of CFT, we transform the state on the boundary of the past light-cone, $\partial A$, into a Matrix Product State (MPS) using local unitaries, as defined in \eqref{eq:MPS} and illustrated in \cref{fig:MPS}. In this section, we further contend that this MPS state can approximately be distilled into a tensor product of entangled pairs:
\begin{equation}
    \ket{\chi}_{AB}\approx U_A\otimes U_B\left(\otimes_{i=1}^k \ket{\phi_i}_{a_ib_i}\right).
\end{equation}
where $k$ is the size of MPS state. It's clear that this approximation does not hold in general due to the disparity in the number of free parameters between the most general entanglement spectrum (contains $2^{k-1}$ parameters) and that of the tensor product of entangled pairs ($k$ parameters). However, for translationally invariant MPS states characterized by short correlation lengths, this approximation is valid.

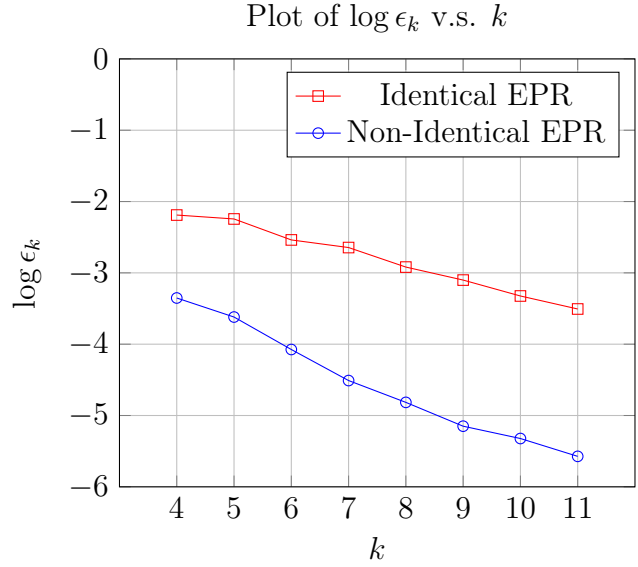
\begin{figure}
    \centering
    \begin{tikzpicture}
    \begin{axis}[
        title={Plot of $\log\epsilon_k$ v.s. $k$},
        xlabel={$k$},
        ylabel={$\log{\epsilon_k}$},
        xmin=3, xmax=12,
        ymin=-6, ymax=0,
        xtick={4,5,6,7,8,9,10,11},
        ytick={-6,-5,-4,-3,-2,-1,0},
        legend pos=north east,
        grid={both},
    ]
    
    \addplot[
        color=red,
        mark=square,
        ]
        coordinates {
        (4, -2.1892564076870427)
        (5, -2.24431618487007)
        (6, -2.5383074265151158)
        (7, -2.645075401940822)
        (8, -2.9187712324178627)
        (9, -3.101092789211817)
        (10, -3.3242363405260273)
        (11, -3.506557897319982)
        };
        \addlegendentry{Identical EPR}
    
    \addplot[
        color=blue,
        mark=o,
        ]
        coordinates {
        (4, -3.3524072174927233)
        (5, -3.619353391465326)
        (6, -4.074541934925921)
        (7, -4.509860006183766)
        (8, -4.815891217303744)
        (9, -5.149897361429764)
        (10, -5.322610059117081)
        (11, -5.572754212249797)
        };
        \addlegendentry{Non-Identical EPR}
    
    \end{axis}
    \end{tikzpicture}
    \caption{Scaling of error with system size.}
    \label{fig:error}
\end{figure}

\begin{figure}
    \centering
    \begin{subfigure}[b]{0.4\textwidth}
           \scalebox{0.75}{
        \begin{tikzpicture}
        \begin{axis}[
            title={Plot of $\mathcal{M}_2$ and $\partial_n\tilde{S}_n$ v.s. $k$},
            xlabel={$k$},
            ylabel={$\mathcal{M}_2$},
            legend pos=north west,
            grid style=dashed,
        ]

        \addplot[
            color=red,
            mark=*,
            ]
            coordinates {
            (2, 0.5284476869093326)
            (3, 0.8796597152606531)
            (4, 1.1786243687213358)
            (5, 1.5067020734480323)
            (6, 1.8209208299973332)
            (7, 2.149153547543458)
            (8, 2.4804444445207428)
            };
            \addlegendentry{$\partial_n \tilde{S}_n\vert_{n=1}$}

            \addplot[
            color=blue,
            mark=square,
            ]
            coordinates {
            (2, 0.2420997928053625)
            (3, 0.5296291745870216)
            (4, 0.6341573657488669)
            (5, 0.8412419311548289)
            (6, 0.9369987551747572)
            (7, 1.0234907616064521)
            (8, 1.1120380681237716)
            };
            \addlegendentry{$M_2$}
        
        \addplot[
            color=orange,
            mark=triangle,
            ]
            coordinates {
            (2, 0.044403073081369274)
            (3, 0.13215785067590613)
            (4, 0.2031747844453768)
            (5, 0.28858671321308377)
            (6, 0.3708788831847463)
            (7, 0.4552538910164904)
            (8, 0.5382145816013133)
            };
            \addlegendentry{$\frac{1}{2}\partial_n \tilde{S}_n\vert_{n=2}$}
        \end{axis}
        \end{tikzpicture}
           }
        \caption{}
    \end{subfigure}
    \hfill
    \begin{subfigure}[b]{0.4\textwidth}
        \centering
        \scalebox{0.75}{
        \begin{tikzpicture}
            \begin{axis}[
                title={$\mathcal{M}_2$ and $\partial_nS_n$ for random samples},
                xlabel={Samples},
                ylabel={$M_2$},
                xmin=0, xmax=13,
                ymin=0, 
                ymax=4, 
                xtick={1,2,...,12},
                ytick={0,1,...,4}, 
                legend pos=north west,
                grid style=dashed,
                scatter/classes={
                    a={mark=*,red},
                    b={mark=square,blue},
                    c={mark=triangle,orange}
                }
            ]
            
            \addplot[scatter,only marks,scatter src=explicit symbolic]
                coordinates {
                (1,1.6313001268796445)[a]
                (2,2.0960085432431907)[a]
                (3,2.8261617899455858)[a]
                (4,2.21701795524457)[a]
                (5,2.740026056362322)[a]
                (6,1.5321860145108637)[a]
                (7,1.7622913869001764)[a]
                (8,2.055347541498999)[a]
                (9,1.2638202464353943)[a]
                (10,2.3765282087280477)[a]
                (11,2.2131910133599826)[a]
                (12,1.7362921659012711)[a]
                };
                \addlegendentry{$\partial_n \tilde{S}_n\vert_{n=1}$}

            \addplot[scatter,only marks,scatter src=explicit symbolic]
                coordinates {
                (1,0.9368407085714446)[b]
                (2,1.0243860301630408)[b]
                (3,1.1735551691638577)[b]
                (4,1.0407553593874668)[b]
                (5,1.2056303937669703)[b]
                (6,0.9159318631290625)[b]
                (7,0.9496265212961039)[b]
                (8,1.0030008626048081)[b]
                (9,0.7647095619498042)[b]
                (10,1.0945398160220088)[b]
                (11,1.0546323613045863)[b]
                (12,0.8853456721851094)[b]
                };
                \addlegendentry{$M_2$}
            

            \addplot[scatter,only marks,scatter src=explicit symbolic]
                coordinates {
                (1,0.5567483087609976)[c]
                (2,0.54726892014064)[c]
                (3,0.44819896725805314)[c]
                (4,0.5455860132789314)[c]
                (5,0.597299095385845)[c]
                (6,0.6400977858963219)[c]
                (7,0.5802832928523738)[c]
                (8,0.5899668363746525)[c]
                (9,0.40702575591041695)[c]
                (10,0.5677481440401517)[c]
                (11,0.6851334147660781)[c]
                (12,0.39684244389082163)[c]
                };
                \addlegendentry{$\frac{1}{2}\partial_n \tilde{S}_n\vert_{n=2}$}
            \end{axis}
            \end{tikzpicture}
        }
        \caption{}
    \end{subfigure}
    \caption{(a) Scaling of $\mathcal{M}_2$ and $\partial_n\tilde{S}_n$ with state size $k$ for a particular sample of random matrix in MPS. (b) For randomly generated samples of MPS states with a fixed size $k=7$, $\mathcal{M}_2$ is bounded by anti-flatness $\partial_n\tilde{S}_n$.  }
    \label{fig:mpsmagic}
\end{figure}
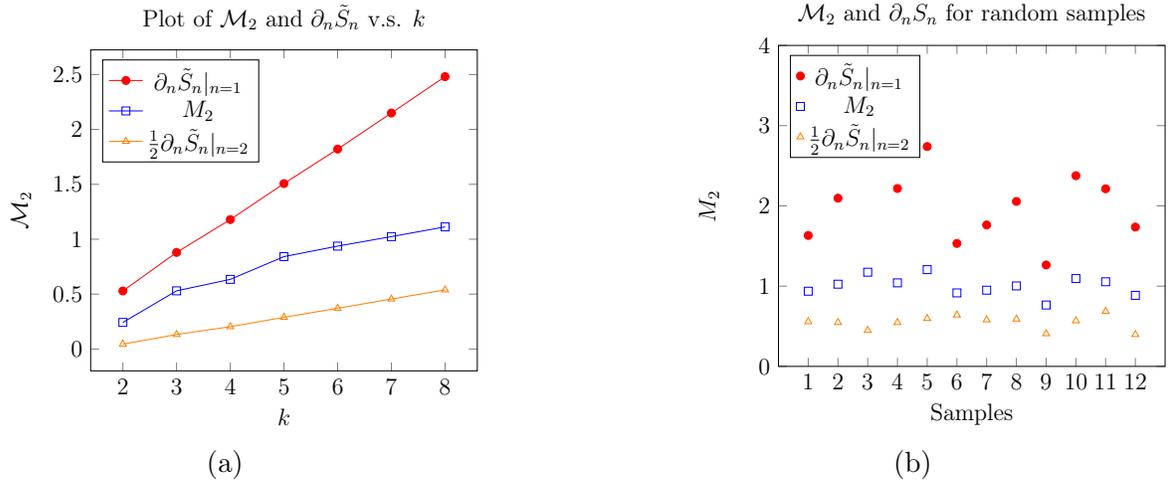

To substantiate this approximation, we simulate several MPS states using $k$ number of identical random matrices to construct the reduced state $\rho_A=\tr_B(\ket{\chi}\bra{\chi})$ and evaluate its entanglement spectrum. We then approximate this spectrum by fitting it to the tensor product of individual entangled pair spectra:
\begin{equation}
    \min_{\{\lambda_i\}}\left\vert\mathrm{Spec}(\rho_A)- \bigotimes_{i=1}^k \begin{pmatrix}
        \lambda_i & 0\\
        0 & 1-\lambda_i
    \end{pmatrix}\right\vert=\epsilon_k
\end{equation}
where $\epsilon_k$ quantifies the approximation error. Our numerical analysis up to $k=11$ reveals an exponential decrease in $\epsilon_k$ with increasing $k$. We present two distinct scenarios in \cref{fig:error}:  In the first scenario, we require all EPR pairs in the tensor product to be identical, yielding an error trend of $\epsilon_k\sim 0.1 \times 1.2^{-k}$. In the second scenario, we relax this constraint, allowing for variability among the EPR pairs, which results in a more pronounced error reduction, following $\epsilon_k\sim 0.05 \times 1.4^{-k}$.

With the distillation assumption justified we expect the inequality~\eqref{eq:branebound} to be true for general MPS state and therefore for a CFT. We plot the magic and the R\'enyi entropy (dual to brane area) for a set of randomly generated samples of MPS states in \cref{fig:mpsmagic} and verify the validity of the bound \eqref{eq:branebound}.

\section{Validity of various bound for magic}\label{app:bound}
In the main text, we introduced several approximations for non-local magic, noting its proportional relationship to anti-flatness in certain regimes and its closeness to entropy in others. This section delineates the conditions under which these approximations hold true.
\subsection*{Near flat limit}
 We begin by examining the approximation between non-local magic and anti-flatness, specifically:
\begin{equation}
    \mathcal{M}_2(|\psi\rangle_{AB})\approx \frac{\mathcal{F}(\rho_A)}{\pur ^2(\rho_A)},
\end{equation}

which is applicable primarily in the near-flat limit of the entanglement spectrum. This is because the left-hand side (LHS) is additive and scales linearly with $n$, while the right-hand side (RHS) can be expressed as:
\begin{equation}
    \frac{\mathcal{F}(\rho_A)}{\pur^2(\rho_A)}=\frac{\tr(\rho_A^3)-\tr(\rho_A^2)^2}{\tr(\rho_A^2)^2}=e^{2(S_2(A)-S_3(A))}-1,
\end{equation}
which becomes additive only at the linear order of the Taylor expansion in the entropy difference. Hence, the condition  $S_2(A)-S_3(A)<\frac{1}{2}$ must be met, indicating an almost flat spectrum or a small system size.

Additionally, this regime aligns with where the two anti-flatness measures defined previously converge, particularly when:
\begin{equation}
    \langle(\delta\log\rho)^2\rangle_{\rho}\approx \frac{\langle(\delta\rho)^2\rangle_{\rho}}{\langle\rho\rangle_{\rho}^2}
\end{equation}
under the condition $\sum_i\delta\lambda_i^2\ll \sum_i\bar{\lambda}_i^2$. In the flat limit 
$\sum_{i}\bar{\lambda}_i^2\sim e^{-S_0}$. and $\sum_i\delta\lambda_i^2\sim e^{-S_2}-e^{-S_0}$. So the approximation requires $e^{S_0-S_2}-1<1$, that is $S_0-S_2\sim O(1)$. This also corresponds to near-flat regime or small system size. 

\subsection*{Far from flat limit}
In contrast, for quantum states with a far-from-flat entanglement spectrum, where the entropy differences across Rényi indices are comparable to the entropy itself, the scenario changes. Referring to Theorem \ref{thm:nlSRE}, the upper bound for the second Stabilizer Rényi entropy measure of non-local magic is:

\begin{equation}
    \mathcal{M}_2^{NL}(|\psi\rangle_{AB})\leq\mathcal{M}_2(\{\lambda_i\})\leq\min\{2S_2(A),4(S_0(A)-S_{1/2}(A))\},
\end{equation} 
indicating a transitional crossover around $S_0(A)-S_{1/2}(A)\sim \frac{1}{2}S_2(A)$. Beyond this point, non-local magic transitions from being proportional to anti-flatness to being proportional to entropy. Our numerical analyses within the Ising model confirm this transition: in the disordered phase and at critical points, non-local magic correlates with entropy $S$ both when varying the model parameter and the subsystem size. However, in the symmetry-breaking phase (refer to \cref{app:symmbreak}), it deviates and becomes anti-correlated with entropy, as shown in \figref{fig:MvsA}.

\section{Supplemental results for Ising Model}
\subsection{Symmetry breaking phase}\label{app:symmbreak}

 In the $g<0$ regime, the Ising model enters the symmetry-breaking phase in the thermodynamic limit. However, our analysis is conducted on a finite-size lattice, where the ground state remains symmetric to spin flipping. Heuristically, we can think of this ground state being approximated by something similar to the GHZ state:

 \begin{equation}\label{eq:symG}
     \ket{G}_{sym}\approx\frac{1}{\sqrt{2}}(\ket{00\cdots0}+\ket{11\cdots1}).
 \end{equation}

 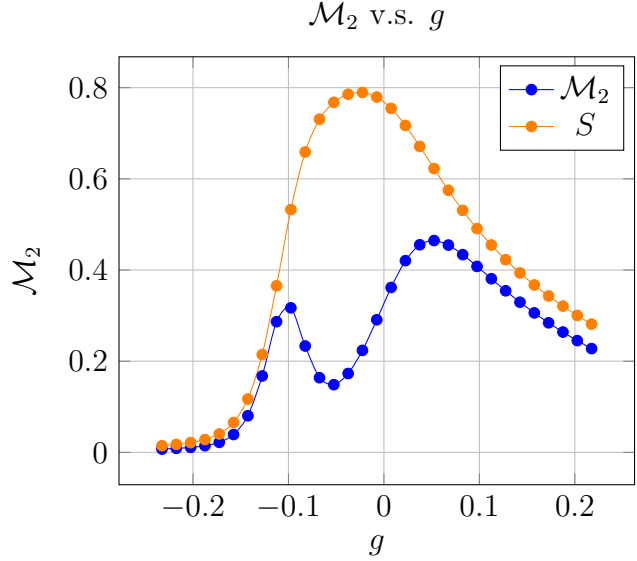
\begin{figure}
    \centering
    \begin{tikzpicture}
        \begin{axis}[
        xlabel={$g$},
        ylabel={$\mathcal{M}_2$},
        title={$\mathcal{M}_2$ v.s. $g$},
        grid={both}
    ]
    \addplot[color=blue,smooth,mark=*] coordinates {
        (-0.2324768366025517, 0.007223787236453515)
        (-0.2174768366025518, 0.008748730695846136)
        (-0.20247683660255178, 0.010926710694541172)
        (-0.18747683660255177, 0.01458468304470301)
        (-0.17247683660255175, 0.021989759114994827)
        (-0.15747683660255174, 0.03917264971677005)
        (-0.14247683660255173, 0.08035700169072464)
        (-0.1274768366025517, 0.16764680998652184)
        (-0.11247683660255181, 0.2867334739089084)
        (-0.09747683660255169, 0.3170054821476055)
        (-0.08247683660255178, 0.233293650708921)
        (-0.06747683660255177, 0.16375202920327683)
        (-0.05247683660255176, 0.14849097486137633)
        (-0.037476836602551744, 0.1730143829660695)
        (-0.02247683660255173, 0.2237686039807405)
        (-0.0074768366025518285, 0.2908969759535281)
        (0.007523163397448296, 0.3617439850129772)
        (0.022523163397448198, 0.42049487003948655)
        (0.03752316339744832, 0.45535876881380544)
        (0.052523163397448225, 0.4646698074563559)
        (0.06752316339744824, 0.4547522488064993)
        (0.08252316339744825, 0.43376417359303615)
        (0.09752316339744815, 0.40791779310947524)
        (0.11252316339744828, 0.38087234049479696)
        (0.12752316339744818, 0.3544456797479232)
        (0.1425231633974483, 0.32940432557212895)
        (0.1575231633974482, 0.30598887725715823)
        (0.17252316339744822, 0.28420249005417625)
        (0.18752316339744823, 0.2639536107301466)
        (0.20252316339744825, 0.24512150642745928)
        (0.21752316339744826, 0.22758407535400627)
    };    
    \addlegendentry{$\mathcal{M}_2$}
    \addplot[color=orange,mark=*,smooth] coordinates {
        (-0.2324768366025517, 0.01461178390556179)
        (-0.2174768366025518, 0.017416796403754247)
        (-0.20247683660255178, 0.0214297514406032)
        (-0.18747683660255177, 0.028014722238020845)
        (-0.17247683660255175, 0.04036083757912194)
        (-0.15747683660255174, 0.06549872277150799)
        (-0.14247683660255173, 0.11707699227960004)
        (-0.1274768366025517, 0.21456086879545458)
        (-0.11247683660255181, 0.3653251204975156)
        (-0.09747683660255169, 0.5325930140867744)
        (-0.08247683660255178, 0.6589732880629487)
        (-0.06747683660255177, 0.7310116119917656)
        (-0.05247683660255176, 0.7679080628052476)
        (-0.037476836602551744, 0.7855963300861846)
        (-0.02247683660255173, 0.7895360222061816)
        (-0.0074768366025518285, 0.7794869557462782)
        (0.007523163397448296, 0.75469489883342)
        (0.022523163397448198, 0.7170592299429965)
        (0.03752316339744832, 0.6712356927065845)
        (0.052523163397448225, 0.6225427177920325)
        (0.06752316339744824, 0.5750215914538588)
        (0.08252316339744825, 0.5308729056604539)
        (0.09752316339744815, 0.49086029922811086)
        (0.11252316339744828, 0.45492904778342713)
        (0.12752316339744818, 0.422669944398889)
        (0.1425231633974483, 0.3935800090837694)
        (0.1575231633974482, 0.3671815763943674)
        (0.17252316339744822, 0.3430640729680953)
        (0.18752316339744823, 0.32088993615140493)
        (0.20252316339744825, 0.3003866225713615)
        (0.21752316339744826, 0.28133492698065415)
    };
    \addlegendentry{$S$}
    \end{axis}
    \end{tikzpicture}
    \caption{Plot of $\mathcal{M}_2$ v.s. $g$, at $b=10^{-4}$ and $|A|=9$. }
    \label{fig:valley}
\end{figure}

\begin{figure}
    \centering
    \begin{subfigure}[b]{0.4\textwidth}
    \scalebox{0.75}{
    \begin{tikzpicture}
    \begin{axis}[
    title={$\mathcal{M}_2$ v.s. $|A|$},
    xlabel={$|A|$},
    ylabel={$\mathcal{M}_2$},
    legend style={at={(0.9,1.2)},anchor=north},
    grid=both,
    minor tick num=1,
]

\addplot coordinates {
    (1, 0.2689501026534531)
    (2, 0.24634756602268076)
    (3, 0.23742530624036945)
    (4, 0.2341664183399669)
    (5, 0.23313755401323172)
    (6, 0.23297365188896818)
    (7, 0.233092371527223)
    (8, 0.23323569080131712)
    (9, 0.233293650708921)
};
\addlegendentry{$g=-0.09$}

\addplot coordinates {
    (1, 0.24384767003162575)
    (2, 0.19876334120774347)
    (3, 0.1795866756827645)
    (4, 0.17091478801827334)
    (5, 0.16686418543136147)
    (6, 0.16499089767821182)
    (7, 0.16416477899486343)
    (8, 0.16383662079327974)
    (9, 0.16375202920327683)
};
\addlegendentry{$g=-0.08$}

\addplot coordinates {
    (1, 0.23839783569239717)
    (2, 0.18955310403593092)
    (3, 0.16809913606613852)
    (4, 0.15791302724907258)
    (5, 0.15286686130812485)
    (6, 0.15036294406760037)
    (7, 0.14916089721150316)
    (8, 0.14863627875794155)
    (9, 0.14849097486137633)
};
\addlegendentry{$g=-0.06$}

\addplot coordinates {
    (1, 0.24587523917318302)
    (2, 0.20549824561143945)
    (3, 0.18798522386416197)
    (4, 0.17986867395697714)
    (5, 0.17599510113906647)
    (6, 0.17418589317601396)
    (7, 0.17339509465078115)
    (8, 0.17309029627445577)
    (9, 0.1730143829660695)
};
\addlegendentry{$g=-0.04$}

\addplot coordinates {
    (1, 0.25826482031576664)
    (2, 0.23387389829982486)
    (3, 0.22518745256634734)
    (4, 0.2224731577645139)
    (5, 0.22206669463778722)
    (6, 0.22250574062437792)
    (7, 0.22312413998687008)
    (8, 0.22359615874074287)
    (9, 0.2237686039807405)
};
\addlegendentry{$g=-0.02$}

\end{axis}
    \end{tikzpicture}
    }
    \caption{}
    \label{fig:MvsAa}
    \end{subfigure}
    \hfill
    \begin{subfigure}[b]{0.4\textwidth}
    \scalebox{0.75}{
        \begin{tikzpicture}
        \begin{axis}[
    title={$S$ v.s. $|A|$},
    xlabel={$|A|$},
    ylabel={$S$},
    legend style={at={(0.85,0.44)}, anchor=north},
    grid=both,
    minor tick num=1,
]

\addplot coordinates {
    (1, 0.4882591491810362)
    (2, 0.5732661692916436)
    (3, 0.6110127980169322)
    (4, 0.6316082158129683)
    (5, 0.643830287493819)
    (6, 0.6513098074317467)
    (7, 0.6558042836506663)
    (8, 0.6582136549531553)
    (9, 0.6589732880629487)
};
\addlegendentry{$g=-0.09$}

\addplot coordinates {
    (1, 0.5251676197278689)
    (2, 0.6218593298388037)
    (3, 0.667228502555712)
    (4, 0.6933234776238961)
    (5, 0.7095691714738573)
    (6, 0.7199287515551955)
    (7, 0.7263602411156028)
    (8, 0.7298868121301564)
    (9, 0.7310116119917656)
};
\addlegendentry{$g=-0.08$}

\addplot coordinates {
    (1, 0.5315614343935706)
    (2, 0.6363808878074139)
    (3, 0.6881586158913452)
    (4, 0.7193994273394969)
    (5, 0.7396787045457957)
    (6, 0.7530676361495676)
    (7, 0.7616056520333332)
    (8, 0.7663733432572958)
    (9, 0.7679080628052476)
};
\addlegendentry{$g=-0.06$}

\addplot coordinates {
    (1, 0.5226859724979404)
    (2, 0.6334301601541248)
    (3, 0.6906057714327959)
    (4, 0.7265118349936612)
    (5, 0.7506177919230987)
    (6, 0.7669719090196478)
    (7, 0.7776168549095465)
    (8, 0.7836431182868497)
    (9, 0.7855963300861846)
};
\addlegendentry{$g=-0.04$}

\addplot coordinates {
    (1, 0.5059548186719761)
    (2, 0.6207548641952213)
    (3, 0.6820164453523454)
    (4, 0.7216323203848032)
    (5, 0.748869817389466)
    (6, 0.7676981101667162)
    (7, 0.7801247556788512)
    (8, 0.7872244076630824)
    (9, 0.7895360222061816)
};
\addlegendentry{$g=-0.02$}

\end{axis}
    \end{tikzpicture}
    }
    \caption{}
    \label{fig:MvsAb}
    \end{subfigure}
    \caption{(a)Non-local magic $\mathcal{M}_2$ v.s. $|A|$. (b) Entropy $S$ v.s. $|A|$, at $b=10^{-5}$ and $g>-0.1$. }
    \label{fig:MvsA}
\end{figure}
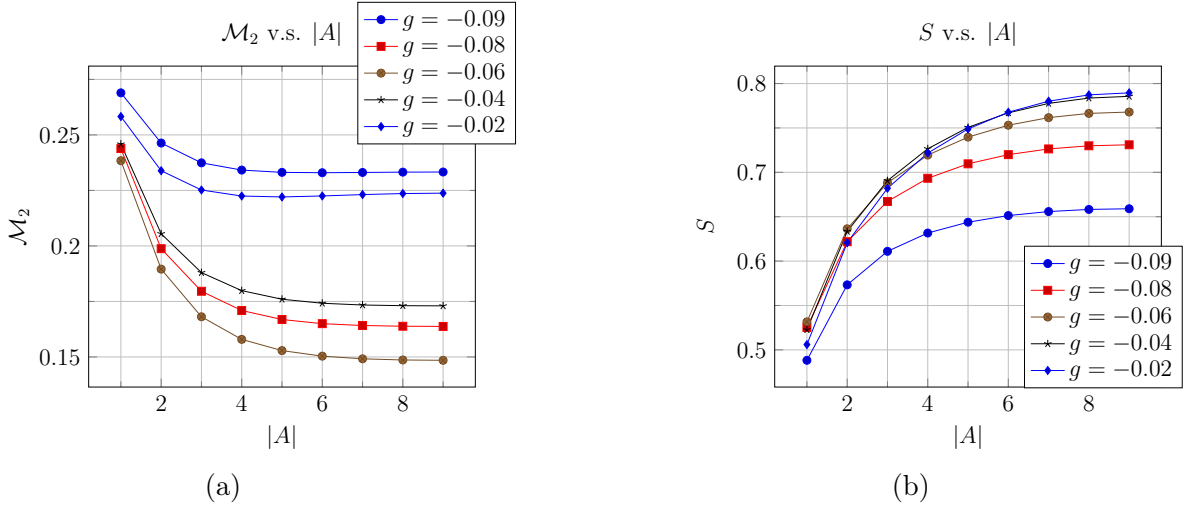

To approximate the true ground state achievable in the thermodynamic limit within our finite lattice model, we introduce a small bias field in the
$z$-direction: 
\ba
H=H_{\rm Ising}(g)+b\sum_i Z_i. 
\ea

As the bias $b$ increases,  the ground state transitions towards one of the two symmetry-broken states:
\ba
    &\ket{G}_{\uparrow}&=\ket{\uparrow\uparrow\cdots \uparrow}\\
    &\ket{G}_{\downarrow}&=\ket{\downarrow\downarrow\cdots \downarrow}.
\ea

Exploring how non-local magic $\mathcal{M}_2$ behaves as we adjust different parameters led to some fascinating results that are particularly noticeable when a non-zero bias field is applied. As shown in \cref{fig:valley}, a distinctive ``valley'' emerges in the  $\mathcal{M}_2$ plot within the $g<0$ regime.  We juxtapose entropy and non-local magic in our plots to underscore their divergent behaviors and the unique information conveyed by non-local magic.

This valley can be understood as arising from the competition between two types of ground states. Within the valley, the system's ground state approximates the symmetric GHZ-like state $\ket{G}_{sym}$, as defined in~\eqref{eq:symG}. In this region, non-local magic values are minimized because the reduced density matrix of  $\ket{G}_{sym}$ resembles that of a maximally mixed single qubit state, leading to a flat spectrum and, consequently, lower $\mathcal{M}_2$ follows from \cref{th:srebound}. Additionally, we observe diminished $\mathcal{M}_2$ values in regions far from the critical point, where $|g|$ is sufficiently large, as indicated by the plateau beyond $g<-0.2$ in \cref{fig:valley}.  Here, the ground state transitions to a symmetry-broken state $\ket{G}_{\uparrow/\downarrow}$, which lacks non-local magic due to its tensor product structure.

Despite the low non-local magic values associated with both  $\ket{G}_{sym}$ and $\ket{G}_{\uparrow/\downarrow}$, the transition between these states has to past through a regime of non-trivial non-local magic. This occurs because continuous parameter changes cannot be approximated by discrete Clifford transformations, resulting in a notable increase in non-local magic.   The $\mathcal{M}_2$ measure captures this as a pronounced peak, delineating the transition between the two ground states near $g\sim -0.1$ in \cref{fig:valley}.  

An additional noteworthy aspect of non-local magic inside the valley is its counterintuitive decrease with increasing subregion size $|A|$, as depicted in  \cref{fig:MvsAa}.  This phenomenon is unique to the valley. In contrast, entropy consistently increases with $|A|$. This unusual trend in $\mathcal{M}_2$ is also linked to the proximity to the symmetric state $\ket{G}_{sym}$, which results in an almost flat entanglement spectrum within the valley.  Consequently, $\mathcal{M}_2$ aligns more closely with the entropy differential $S_0-S$  rather than the entropy itself, as discussed in \cref{section:estimate}, offering an explanation for the inverse relationship observed between $S$ and $\mathcal{M}_2$ in this region.  

It's also important to note that the competition between  $\ket{G}_{sym}$ and $\ket{G}_{\uparrow/\downarrow}$  is a manifestation of finite-size effects. As demonstrated in  \cref{fig:size},   the valley tends to diminish with increasing lattice size $n$. Specifically, when we set $b=10^{-4}$ (see \cref{fig:sizea}), the peak of non-local magic shifts closer to $g=0$ with larger lattice sizes. Similarly, with $g=-0.11$ (see \cref{fig:sizeb}), the peak moves towards $b=0$ as the lattice size expands. This suggests that the parameter space favoring the symmetric state narrows in both dimensions with increasing lattice size.


\begin{figure}
    \centering
    \begin{subfigure}[b]{0.4\textwidth}
        \scalebox{0.75}{
        \begin{tikzpicture}
            \begin{axis}[
        xlabel={$g$}, 
        ylabel={$\mathcal{M}_2$},
        title={$\mathcal{M}_2$ v.s. $g$},
        legend style={at={(0.2,0.98)},anchor=north},
        grid={both}
        ]
    \addplot coordinates {
        (-0.2324768366025517, 0.007223787236453515)
        (-0.2174768366025518, 0.008748730695846136)
        (-0.20247683660255178, 0.010926710694541172)
        (-0.18747683660255177, 0.01458468304470301)
        (-0.17247683660255175, 0.021989759114994827)
        (-0.15747683660255174, 0.03917264971677005)
        (-0.14247683660255173, 0.08035700169072464)
        (-0.1274768366025517, 0.16764680998652184)
        (-0.11247683660255181, 0.2867334739089084)
        (-0.09747683660255169, 0.3170054821476055)
        (-0.08247683660255178, 0.233293650708921)
        (-0.06747683660255177, 0.16375202920327683)
        (-0.05247683660255176, 0.14849097486137633)
        (-0.037476836602551744, 0.1730143829660695)
        (-0.02247683660255173, 0.2237686039807405)
        (-0.0074768366025518285, 0.2908969759535281)
        (0.007523163397448296, 0.3617439850129772)
        (0.022523163397448198, 0.42049487003948655)
        (0.03752316339744832, 0.45535876881380544)
        (0.052523163397448225, 0.4646698074563559)
        (0.06752316339744824, 0.4547522488064993)
        (0.08252316339744825, 0.43376417359303615)
        (0.09752316339744815, 0.40791779310947524)
        (0.11252316339744828, 0.38087234049479696)
        (0.12752316339744818, 0.3544456797479232)
        (0.1425231633974483, 0.32940432557212895)
        (0.1575231633974482, 0.30598887725715823)
        (0.17252316339744822, 0.28420249005417625)
        (0.18752316339744823, 0.2639536107301466)
        (0.20252316339744825, 0.24512150642745928)
        (0.21752316339744826, 0.22758407535400627)
    };
    \addlegendentry{$n=18$}
    \addplot coordinates {
        (-0.2324768366025517, 0.007138323500646095)
        (-0.2174768366025518, 0.008478137420016225)
        (-0.20247683660255178, 0.010087538081067002)
        (-0.18747683660255177, 0.012036072476847123)
        (-0.17247683660255175, 0.014429732020143615)
        (-0.15747683660255174, 0.01747103174356087)
        (-0.14247683660255173, 0.021682284953683427)
        (-0.1274768366025517, 0.02875667639269592)
        (-0.11247683660255181, 0.04468008032379104)
        (-0.09747683660255169, 0.08927084067997512)
        (-0.08247683660255178, 0.20474773189791878)
        (-0.06747683660255177, 0.343935624133068)
        (-0.05247683660255176, 0.29458702358256)
        (-0.037476836602551744, 0.21639555782368375)
        (-0.02247683660255173, 0.2288576762236713)
        (-0.0074768366025518285, 0.29780205322798725)
        (0.007523163397448296, 0.3845718675314944)
        (0.022523163397448198, 0.4537579542896862)
        (0.03752316339744832, 0.4856460489390489)
        (0.052523163397448225, 0.4846369370554632)
        (0.06752316339744824, 0.4649908048348451)
        (0.08252316339744825, 0.43779072468821545)
        (0.09752316339744815, 0.4088117511954463)
        (0.11252316339744828, 0.38049969574410225)
        (0.12752316339744818, 0.35372882542412704)
        (0.1425231633974483, 0.3287227825545529)
        (0.1575231633974482, 0.30545833795560945)
        (0.17252316339744822, 0.2838286043166396)
        (0.18752316339744823, 0.26370581053604863)
        (0.20252316339744825, 0.2449641774289902)
        (0.21752316339744826, 0.22748741229171957)
    };
    \addlegendentry{$n=24$}
    \addplot coordinates {
        (-0.2324768366025517, 0.007138227912831164)
        (-0.2174768366025518, 0.008477694356282846)
        (-0.20247683660255178, 0.010085541655881018)
        (-0.18747683660255177, 0.012027317911429997)
        (-0.17247683660255175, 0.014392344398262721)
        (-0.15747683660255174, 0.01731547711723815)
        (-0.14247683660255173, 0.02105222614873403)
        (-0.1274768366025517, 0.026281070467837425)
        (-0.11247683660255181, 0.03536996261958408)
        (-0.09747683660255169, 0.05742683034275903)
        (-0.08247683660255178, 0.12352601687391683)
        (-0.06747683660255177, 0.2779640900307487)
        (-0.05247683660255176, 0.3525557793174458)
        (-0.037476836602551744, 0.25496987082659067)
        (-0.02247683660255173, 0.2367824080235811)
        (-0.0074768366025518285, 0.29989716620332707)
        (0.007523163397448296, 0.38971078542149157)
        (0.022523163397448198, 0.4607632643053709)
        (0.03752316339744832, 0.49084971817892564)
        (0.052523163397448225, 0.4871000147650576)
        (0.06752316339744824, 0.4656894698078343)
        (0.08252316339744825, 0.43774367409216186)
        (0.09752316339744815, 0.40856786113387106)
        (0.11252316339744828, 0.38026941416269877)
        (0.12752316339744818, 0.35356398349312906)
        (0.1425231633974483, 0.32861871489030337)
        (0.1575231633974482, 0.30539735413686264)
        (0.17252316339744822, 0.2837946474498429)
        (0.18752316339744823, 0.2636876197956447)
        (0.20252316339744825, 0.2449547347721947)
        (0.21752316339744826, 0.22748264215185654)
    };
    \addlegendentry{$n=26$}
    \end{axis}
        \end{tikzpicture}
        }
        \caption{}
        \label{fig:sizea}
    \end{subfigure}
    \hfill
    \begin{subfigure}[b]{0.4\textwidth}
        \scalebox{0.75}{
        \begin{tikzpicture}
\begin{axis}[
    title={$\mathcal{M}_2$ v.s. $\log b$},
    xlabel={$\log b$},
    ylabel={$\mathcal{M}_2$},
    xmin=-7, xmax=-2,
    ymin=0, ymax=0.4,
    legend pos=north east,
    grid={both},
    every axis plot/.append style={thick}
]

\addplot[
    color=orange,
    mark=*,
    ]
    coordinates {
    (-8, 0.035634345760970211646198095199)
    (-7, 0.03631833084924441)
    (-6, 0.03650516400246957)
    (-5.0, 0.05465961086097081)
    (-4.8, 0.08045085735915408)
    (-4.599999999999999, 0.13598974833468974)
    (-4.3999999999999995, 0.23023981318059344)
    (-4.2, 0.31365841803612904)
    (-4.199970640755866, 0.31366418738740776)
    (-4.179973576680279, 0.31694128531956933)
    (-4.159976512604692, 0.318873262956139)
    (-4.139979448529106, 0.31941765657095933)
    (-4.119982384453518, 0.3185608183049997)
    (-4.0999853203779315, 0.31631858715807953)
    (-4.079988256302346, 0.3127355487012461)
    (-4.059991192226759, 0.30788294744476635)
    (-4.039994128151172, 0.3018554297073823)
    (-4.019997064075586, 0.294766884738027)
    (-3.999999999999999, 0.2867457070623583)
    (-3.7999999999999994, 0.18484455356621748)
    (-3.599999999999999, 0.1041265640094637)
    (-3.3999999999999995, 0.062134224406794954)
    (-3.1999999999999997, 0.04349412266078626)
    (-2.9999999999999996, 0.03563434576097021)
    (-2, 0.02670719766200395)
    };
    \addlegendentry{n=18}

\addplot[
    color=blue,
    mark=square*,
    ]
    coordinates {
    (-5.999999999999999, 0.039754416471600426)
    (-5.799999999999999, 0.05011609228992938)
    (-5.6, 0.07469590800766968)
    (-5.4, 0.12813141910521825)
    (-5.199999999999999, 0.22107426054724227)
    (-5.0, 0.30964631505179396)
    (-4.979997073344623, 0.3138833533429611)
    (-4.959994146689245, 0.31681435786145246)
    (-4.9399912200338685, 0.31837814663203196)
    (-4.919988293378491, 0.3185412810675044)
    (-4.899985366723115, 0.3172996623181483)
    (-4.879982440067738, 0.3146787332854024)
    (-4.85997951341236, 0.31073226831681294)
    (-4.8399765867569835, 0.3055398559595673)
    (-4.819973660101606, 0.29920328721229517)
    (-4.8, 0.2918535976034281)
    (-4.79997073344623, 0.29184214223277216)
    (-4.599999999999999, 0.19165102401443654)
    (-4.3999999999999995, 0.10838340169117754)
    (-4.2, 0.06431979738533082)
    (-3.999999999999999, 0.04468314071476549)
    (-3, 0.030687732641793277)
    (-2, 0.026683527440224727)
    };
    \addlegendentry{n=24}
\end{axis}
\end{tikzpicture}
        }
        \caption{}
        \label{fig:sizeb}
    \end{subfigure}
    \caption{(a) Plot of $\mathcal{M}_2$ v.s. $g$, at $b=10^{-4}$.  The lattice size takes $n=18$, $n=24$ and $n=26$; (b) Plot of $\mathcal{M}_2$ v.s. $\log b$, at $g=-0.11$, with lattice size taking $n=18$ and $n=24$ }
    \label{fig:size}
\end{figure}
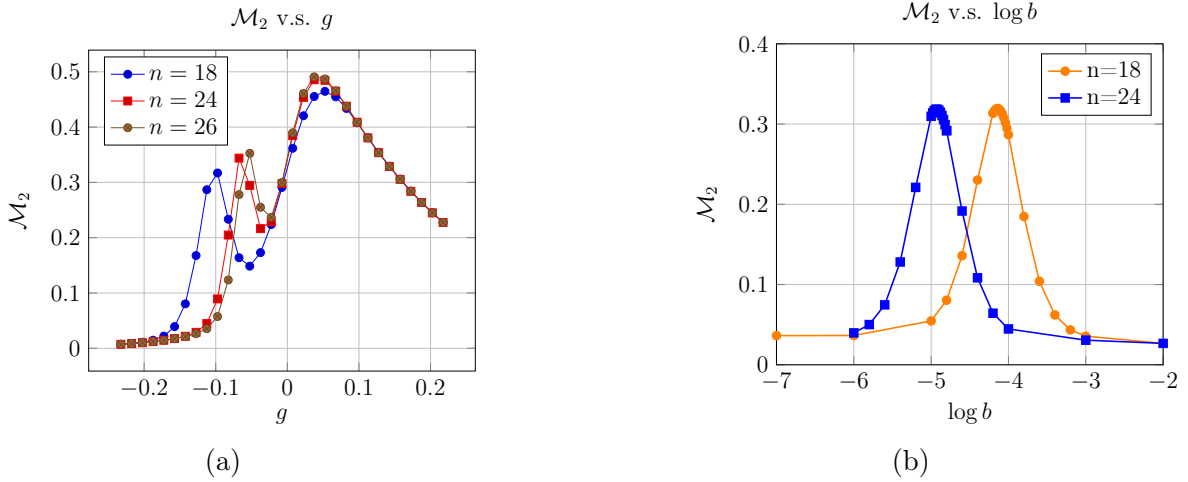



\begin{figure}
    \centering
    \includegraphics[scale=0.5]{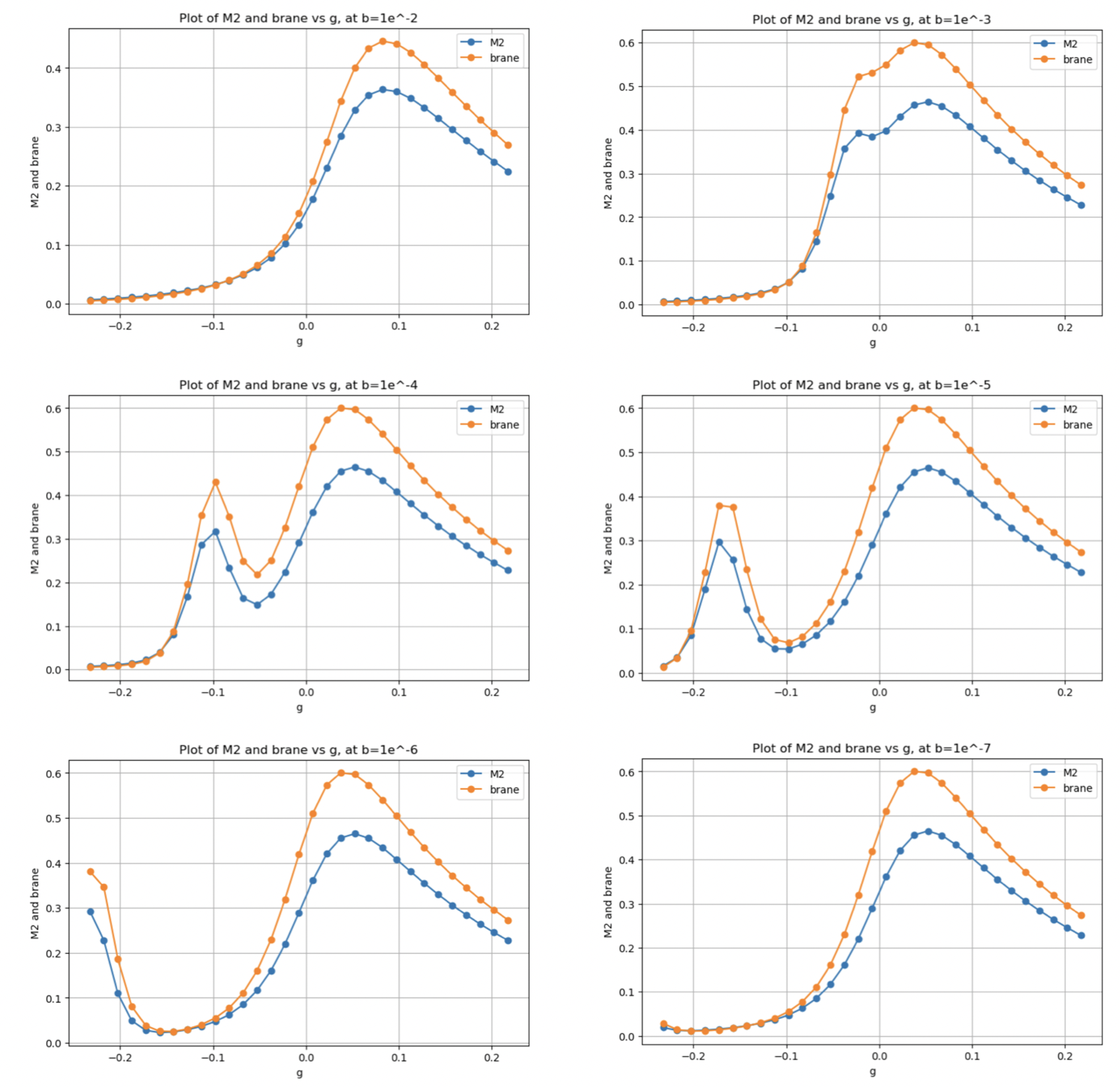}
    \caption{Comparison between $\mathcal{M}_2$ and $|\partial_n\tilde{S}_n|$ (labeled as brane) at various magnetic field $b$ and model parameter $g$. The small peak separates and is pushed to the left as $b$ decreases.}
    \label{fig:Mvsbrane}
\end{figure}

Expanding our analysis, \cref{fig:Mvsbrane} explores the non-local magic across a broader range of the bias field $b$. We find that beyond $b>0.01$ the valley disappears, and the  $g<0$ phase transitions to being governed by the symmetry-broken ground state $\ket{G}_{\uparrow/\downarrow}$. As $b$ decreases towards zero, the peak is pushed to the left where the valley widens, signifying the growing significance of the symmetric ground state $\ket{G}_{sym}$, which becomes dominant for all $g<0$ in the absence of $b$.



\figref{fig:NonLocalMagicSurface2} depicts the $\mathcal{M}_2$ surface as a function of subregion size $|A|$ and critical angle $\theta$. As we decrease the magnitude of the bias field, the symmetry-breaking peak is pushed towards lower and lower $\theta$ values.
\begin{figure}[H]
\begin{center}
	\includegraphics[width=11cm]{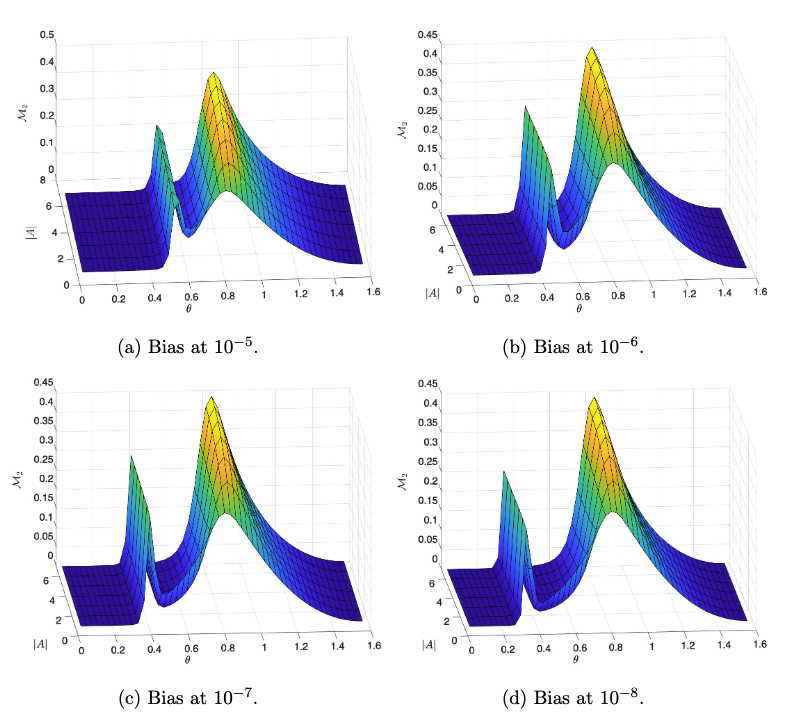}
    \caption{Non-local magic $\mathcal{M}_2$ surface, as a function of critical angle $\theta$ and subregion size $|A|$, for different bias offset fields. The bias magnetic field decreases, the peak indicating a symmetry-breaking effect in the system is pushed further away from criticality.}
    \label{fig:NonLocalMagicSurface2}
\end{center}
\end{figure}

\chapter{PERMISSION FOR REPLICATION}

\newpage
The contents of Chapters 2--7 comprise previously published work. I thank the coauthors of these papers for their permission to replicate the corresponding work here.

\end{document}